%% file: HDR.tex
\documentclass[11pt,twoside,a4paper]{book}

\usepackage{HDR}

\usepackage{fancyhdr} %%% POUR MODIFIER LE STYLE DE PAGE
\input{xy}
\xyoption{all}

\begin{document}
	
	\frontmatter %Use lowercase Roman numerals for page numbers

	\title{%\includegraphics[width=5cm]{UHA.jpeg}\hspace{5cm}
	\includegraphics[width=10cm]{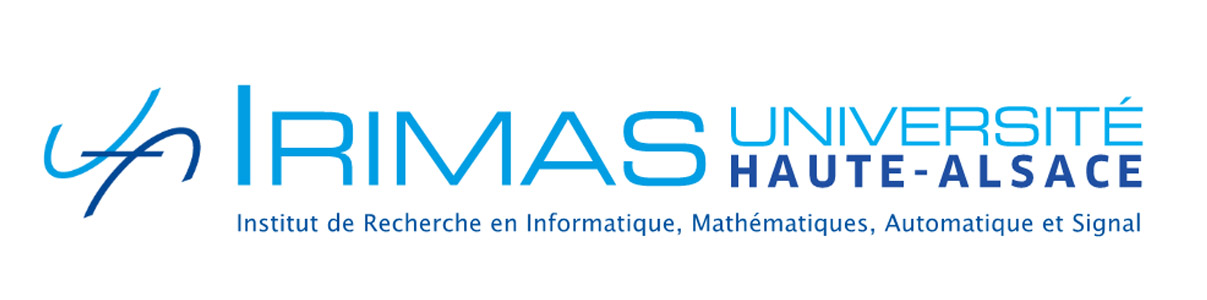}  \\
% 	~\vspace{0.5cm}
	{\Large UNIVERSITE DE HAUTE ALSACE} ~\\
	{\large UNIVERSITE DE STRASBOURG}
	~\vspace{1.5cm}\\
	{\bf\Huge Habilitation à Diriger des Recherches} \\
	\vspace{0.7cm}
	{\bf\normalsize ECOLE DOCTORALE MSII (ED 269)} \\
	{\normalsize Discipline: Mathématiques} \\
	\vspace{0.7cm}
	{\normalsize Présentée et soutenue publiquement \\ par \\
	Pierre J. CLAVIER-BAUER} \\
	\vspace{1.2cm}
	\hrulefill \\
	{\Huge TRAPs, Generalisations of MZVs, Locality and Resurgence for Quantum Field Theories} \\
	\hrulefill \\
	\vspace{1cm}
	{\normalsize
	\begin{tabular}{cl}
     Jury: & Prof. Dominique MANCHON, CNRS et Université Clermont-Auvergne (Rapporteur) \\[1ex]
     & Prof. Kasia REJZNER, University of York (Rapportrice) \\[1ex]
	                    & Prof. Leila SCHNEPS, Sorbonne Université, (Rapportrice)  \\[1ex]
	                    & Prof. Frédéric CHAPOTON, CNRS et Université de Strasbourg (Examinateur) \\[1ex]
	                    & Prof. Vladimir DOTSENKO, Université de Strasbourg (Examinateur) \\[1ex]
	                    & Prof. Martin BORDEMANN, Université de Haute Alsace (Garant) \\[1ex]
	                    & Dr. Frédéric MENOUS, Université de Paris-Saclay (Invité)
    \end{tabular}
% 	\begin{minipage}[t]{0.9\linewidth}
% 	\begin{itemize}
% 	 \item[] Prof. Dominique MANCHON, CNRS et Université Clermont-Auvergne (Rapporteur)
% 	                    \item[] Prof. Kasia REJZNER, University of York (Rapportrice)
% 	                    \item[] Prof. Leila SCHNEPS, Sorbonne Université, (Rapportrice)
% 	                    \item[] Prof. Frédéric CHAPOTON, CNRS et Université de Strasbourg (Examinateur)
% 	                    \item[] Prof. Vladimir DOTSENKO, Université de Strasbourg (Examinateur)
% 	                    \item[] Prof. Martin BORDEMANN, Université de Haute Alsace (Garant)
% 	                    \item[] Dr. Frédéric MENOUS, Université de Paris-Saclay (Invité)
%     \end{itemize}
%     \end{minipage}
 }}
	
	\author{
% 	Pierre J. Clavier${}^{1}$\\
% 		~\\
% 		{\small \it $^1$ Université de Haute Alsace, IRIMAS,}
% 		{\small \it 12 rue des Frères Lumière,}\\
% 		{\small \it 68 093 MULHOUSE Cedex, France}\\ 
}
	
	\date{}
	
	\maketitle

% 	Quotes:
% \epigraph{Alas, our fraility is the cause, not we, for such we are made, for such we be.}{Shakespeare}
% ``Alas, our fraility is the cause, not we, for such we are made, for such we be.'' -- Shakespeare.
% 
% ``A quoi bon vous tracasser pour si peu, allez donc faire un somme en attendant.'' -- Calaferte.

~\thispagestyle{empty}
\begin{center}
 \emph{A la mémoire de Yolande Clavier.}
\end{center}

\newpage

~\thispagestyle{empty}
\vspace{2cm}
\begin{epigraphs}
 \qitem{Alas, our \ty{frailty} is the cause, not we, \\ For such we are made, for such we be.}{Shakespeare\\\emph{\ty{Twelfth} night}, II, 2, Viola}
 \vspace{1cm}
 \qitem{A quoi bon vous tracasser pour si peu, allez donc faire un somme en attendant.}{Calaferte\\ \emph{Septentrion}}
\end{epigraphs}

    \tableofcontents

%% MAIN TEXT:

    \input{INT}

    \renewcommand{\chaptermark}[1]{%
\markboth{\MakeUppercase{%
\chaptername\ \thechapter.%
\ #1}}{}}
    
    \input{PROP}

%     Décommenter après avoir enlever le header (et enddocument) du document PROP.tex.
    
    \input{MZV}

    \input{LOC}

    \input{RES}

    %% END MAIN TEXT
    
%     \setcounter{chapter}{0}
%     
%     \input{APPA}

% \includepdf[pages={1}]{cv_eng_24.pdf}

% %% TO INCLUDE CV:
% 
% \newpage
% \phantomsection
% \addcontentsline{toc}{chapter}{Appendix: Curriculum Vitae}
% 
% % \includepdf[pages={2-}]{cv_eng_24.pdf}
% 
% \includepdf[pages={-}]{cv_eng_24.pdf}
% 
% %% END CV

% \includepdf[
% pages=-,
% pagecommand={
% \thispagestyle{empty}
% \stepcounter{insertpages}
% \ifnum\value{insertpages}=1\addcontentsline{toc}{section}{Curriculum Vitae}\fi
% }
% ]{cv_eng_24.pdf}
% \setcounter{insertpages}{0}

 \bibliographystyle{alpha}
 \addcontentsline{toc}{chapter}{Bibliography}
\bibliography{HDR_biblio}
 
\end{document}

%% file: INT.tex
% 
% 
% \documentclass[11pt,twoside,a4paper]{book}
% 
% \usepackage{HDR}
% 
% %% SI ON LAISSE LES DEUX LIGNES SUIVANTES DANS LE STY, ELLES NE MARCHENT PAS... (probablement à cause de \input).
% \input{xy}
% \xyoption{all}
% 
% 
% \begin{document}
% 
% % % % 
% % % % % % COMMENTER CE QUI EST AVANT CA ET ENDDOCUMENT POUR COMPILER HDR.tex.

\newpage

\section*{Remerciements}
    
    \addcontentsline{toc}{section}{Remerciements}
    
    \renewcommand{\chaptermark}[1]{%
\markboth{#1}{}}
\chaptermark{Remerciements}
    
Je voudrais exprimer mes remerciements les plus sincères envers les membres du jury, et en particulier les trois rapporteur.ice.s qui ont pris le temps de lire ce très long manuscrit. Leurs retours, ainsi que ceux de mon garant ont grandement amélioré la qualité de ce travail et je leurs en suis extrêmement reconnaissant. Il va sans dire que toute erreur restante est uniquement de mon fait.

J'éprouve également une profonde gratitude envers tout les membres de l'équipe d'algèbre-géométrie, et plus généralement du département de mathématiques, qui créent un environnement de travail plaisant et stimulant. Ce point s'étends également à l'ensemble de l'IRIMAS, dont les membres, et surtout la direction, m'ont encouragés à passer mon HDR. J'inclue aussi ici l'ensemble des membres du MATHEOR dont la quête perpetuelle de l'excellence mathématique est une source d'inspiration inépuisable. Un salut tout particulier à Douglas, qui a accepter de me faire confiance pour faire sa thèse sous ma co-tutelle.

Ecrire des remerciements est une bonne manière de réaliser à quel point on doit beaucoup à énormément de personnes. Dans mon cas, il en est deux sans lesquelles je ne serais pas en train de faire de la recherche aujourd'hui. Il s'agit en premier lieu de Marc Bellon, qui fut mon directeur de thèse, et de Sylvie Paycha, sous la supervision de qui j'ai effectué mes post-doctorats. Sachez que je suis immensément conscient et reconnaissant des chances que vous m'avez données. Et merci à tous les autres, que nous ayons collaboré, discuté ou simplement passé un moment agréable ensemble. Je vous dois beaucoup, et sachez que j'essaye, modestement, de transmettre ce qui m'a été offert.

Enfin, un très grand merci à ma famille, à mes parents Jean-Philippe et Anouk, à mes soeurs Noémie et Suzanne, pour être les personnes merveilleuses que vous êtes. Un grand merci à Wilfried et Waltraud: \emph{Es bedeutet mir viel, dass ihr heute in Mulhouse seid. Merci!} Et j'embrasse au passage toute la famille éloignée, de A comme André à Y comme Yolande et Yui, qui serait trop longue à énumérer.

Il y a des bons jours, il y a des mauvais jours. Mais aucun n'est vraiment mauvais si il se termine avec vous. Merci Franziska et L.

% The acknowledgements (in the TZVs paper) have to be included somewhere there.

\mainmatter % Now Use Arabic numerals for page numbers

\chapter*{General Introduction} 
    
%     CHANGER REF PAPIER BROWN!!!!! XXXXXX
    
    \addcontentsline{toc}{chapter}{General Introduction}
    
    \chaptermark{GENERAL INTRODUCTION}

%    THERE ARE STILL XXXXX IN CHAP 2, check everywhere!!!  

I start with some considerations on what this thesis is and how I came to \ty{write} it. These discussions are not strictly scientific and the readers not interested in them should start with Section \nameref{sec:contents} below.

\section*{Foreword}

\addcontentsline{toc}{section}{Foreword}

I would like to start the thesis with a few words regarding what it actually is. This requires a short presentation of the French system for those not familiar with it. I apologise to the readers that are and advise them to skip the next paragraph.

I am currently maître de conférences, a status that more or less corresponds to lecturer in the Anglo-Saxon academic \ty{world}. If I eventually wish to apply to positions of professeur des universités (a.k.a. full professor) I first need to have a diploma that essentially allows me to supervise PhD students. This thesis was written in order to obtain this diploma, which is called habilitation à diriger des recherches.

The goal of this thesis is two-fold. First it is to provide a summary of my research after my PhD. Of course it does not aim at being exhaustive, but it nonetheless covers what I deemed to be my most important results. The second goal is to try to present the directions future works could point towards and the goals I might hope to achieve with them. 

% Therefore I do not aim at presenting new results here.
It might be that some results are written in a slightly new form (for example Theorem \ref{thm:flattening} is the concatenation of two separate results), while some other were known but just not written in published papers (e.g. Theorem \ref{thm:loc_Hopf_prop_forests}). Nonetheless, this text is essentially not a research text. I end each of the chapter of the thesis with a presentation of open questions. 
% However the most original pieces of this thesis are the last sections of each chapters, where I attend to present open questions.
Some of them are written in the form of conjectures, and at least two of them (Conjectures \ref{conj:TRAP_reg} and \ref{conj:asympt_bound}) did not appear in previous work. Some results that are direct consequences of these conjectures are also new, e.g. Proposition \ref{prop:rapport_conj}. Furthermore, some other open questions are written in a less precise form and therefore are not called conjectures. When possible, I try to explain how I imagine one could tackle these questions.

% ~\\
% 
% Explain what the thesis is (and french system). Point out the method I like to use (algebra to solve analytical problem). Statement about no new results.

% Dire comment ces sujets sont tous reliés à la TQC?

% New results (more or less new, often known but not written, or just not written like that):
% \begin{itemize}
%  \item the one in chapter 2 (equivalence with shuffles) \ref{thm:flattening}
%  \item \ref{thm:loc_Hopf_prop_forests} (just not written)
%  \item Conjecture \ref{conj:TRAP_reg}
%  \item Conjecture  \ref{conj:asympt_bound}
% \end{itemize}

\section*{The long march to the HDR}

\addcontentsline{toc}{section}{The long march to the HDR}

When people ask me what I do, I say that I am a mathematical physicist, or just a mathematician if they ask in \ty{French}. In general people make a couple of derogatory remarks, stop asking about my work and often to talk to me altogether. I am great at parties. But you, dear reader, I can imagine your reaction to fit my plans for this section, and there \ty{is} not much you can do about it. So I'd like to believe that you would have asked me what I mean by that.

Well, thanks for asking. I am a mathematician, but the questions I like to investigate are typically related to Physics. Sometimes, admittedly in a rather unclear way. I like being a mathematical physicist. The rigor of mathematical work suits me better than the intuitive (some would say hand-wavy) aspects of physics, however theoretical. And I like that my work might still relate to an exterior truth. 
I hope to stay long in this area of mathematics where one can still see, even faintly, some Physics. \\ 

So, if I like rigour, why mathematical physics rather than some more traditional domain of mathematics? Well, this seems like a natural question, but it turns out that
I have a bachelor, master and PhD in Theoretical Physics. Thus it might be more relevant to explain to the curious readers how I ended up writing an HDR in Mathematics, even one with little pieces of Physics sprinkled here and there.
% About why they are reading it though, I can offer no answer.

As a kid, I wanted to do something with ``astro'' in it, thus the Physics studies. This lasted until a seminar I attended during an internship I had to do at the very end of my bachelor. In this excellent seminar, a young researcher explained to us his method to find exoplanets, and how he successfully used it to discover one. It was pretty exciting. At some point during his presentation he showed us his raw data, which took the form of thick, rather random-looking wavy lines. At the end of his presentation, someone attending the conference asked which spikes of these wavy lines were the signal of the exoplanet he found. The speaker answer something like ``oh, none of them! This is only noise, the signal on this graph would be about one thousandth of a pixel''. I was already sure I would be rather bad at processing signals, and that I would hate to \ty{attempt} to do it, so I decided here and there to not pursue a career in astro-something. To this day, I am extremely thankful to the speaker, the asker and all the people that made it possible for me to attend this talk.

Therefore I decided to look towards theoretical physics, which long discussions with my close friend E. Dumonceau convinced me to be more to my liking. During my master degree I had \ty{remarkable} lectures in QFT and it seemed like an exciting topic. I also had the opportunity to dip a couple of toes in some aspects of string theory and was less convinced. With this naive point of view I had the great \ty{luck} to meet Marc Bellon, who accepted to be my PhD advisor. With him and for three years I studied QFT, and in particular Schwinger-Dyson equations.

Many things happened during these three years, good and bad. But two that are relevant here are that Marc encouraged me to \ty{attend} to a weekly mathematical physics seminar held in Paris VII (Diderot) in the mathematics group. There I met many \ty{remarkable} mathematicians and physicists and also got a (much) better understanding of what mathematical research is. Secondly, toward the end of my PhD Marc and I realised that the most \ty{interesting} questions we \ty{encountered} in our line of research were of rather mathematical nature. These questions \ty{were} typically of the form ``which singularities solutions of Schwinger-Dyson Equations can have and where can they be located in the complex plane?''.

A few months before my PhD thesis was due, I started looking for a post-doc position, with very limited success. Growing rather worried that I would end up a jobless doctor, I wrote to many colleagues I had meet at the \ty{aforementioned} mathematical physics seminar. My second great \ty{luck} is that Sylvie Paycha gave me a positive answer. I could work in her group for four years, plus one where I was in Peter Friz's group in TU Berlin as well as in Sylvie's group in Potsdam University. \\

At some point during these five year the tipping moment came when I became a mathematician more than a physicist. I cannot say when it was precisely, probably towards the beginning. I must say that it was not always easy: it sometimes felt like doing a second PhD, but with a very insufficient background. It was only possible due to the patience, kindness and willingness of Sylvie and other colleagues, in particular of the Analysis group of Potsdam University, to share their knowledge that this \ty{endeavour} was possible. I would like to express here once again my warmest thanks to all of them.

It might seem surprising that the first mathematics group I worked in was an analysis group and that I am now in an algebra and geometry group. First and foremost, our group also includes \ty{researchers} (beside myself) with interests in mathematical physics. I also want to say that at least some aspects of analysis speak to the remainder of my physicist's soul. One can do long and technical computations in analysis that would shame no physicist. But most importantly working with analysts made me realise I like to tackle the challenges they offer using algebraic methods. For example, it is \ty{not such an} easy task to determine when arborified zeta values converge. I used the universal property of rooted forests to define them and control their behaviour. I really like this method of defining an analytical object using an algebraic property, typically a universal property in an appropriate category. I aim at using it further, ideally for Feynman rules. I think I can -- very humbly -- say that this is my style of doing mathematics and that I could only find it because I had the chance to work in an analysis group.

In any cases, I started to work in math, but still on questions related to Physics. The first question I looked at with co-authors Sylvie Paycha, Li Guo and Bin Zhang, was to \ty{develop} a framework adapted to multivariate renormalisation schemes. This led to other questions, that may be more or less mathematical. For example, I am interested by number theory, a very mathematical topic indeed, but that I learned first because some of the numbers I study appear in evaluations of Feynman graphs and in the amplitudes of some string theories. I guess I could say at this point of this presentation that it is a full loop: I went from Physics to discuss my current relation to Physics. It could be a form of closure and be  satisfactory but I don't think it is correct. Today is not about closure since it is not an end, merely a step. I am not worried about the destination where this step and the next ones will lead me. After all I am just here for the ride.

% Here an explanation of how my research get to exist: first goal astrophysics, then theoretical physics, then mathematical physics, then mathematics (say it is chance?)

\section*{Contents and main results} \label{sec:contents}

\addcontentsline{toc}{section}{Contents and main results}

\subsection*{General overview}

This HDR contains four chapters and one appendix, which is simply a curriculum vitae and will not be talked about further here. Each of the chapters is intended to be rather self-contained. However I will briefly present the content of the thesis here in the hope to ease the reading of this admittedly long text. In particular, I will aim at pointing out the main results of each chapters. This results will not be stated precisely and sometimes not rigorously but rather in a way that gives the readers the main ideas behind these works, the logic that articulates the various chapters. I also \ty{try} to give a feeling for the motivations behind each of these research projects. \\

I must also say that these four chapters might seem rather unrelated so I would like to specify here the common logic \ty{underlining} the whole text. The common thread is {\bf Quantum Field Theory} (QFT). Very roughly speaking, QFT consists of the physical models describing small particles going fast, as well as the mathematical tools and structures behind these models. Thus it aims at providing a framework to unify special relativity with quantum mechanics.

The most famous aspect of QFT are the so-called Feynman graphs which encode possible interactions in the \ty{perturbative} approach of QFT. One then needs to evaluate these graphs using {\bf Feynman rules}. However these rules are more like a recipe: put an integral here, a propagator there... They are not given as a map from a set of \ty{Feynman} graphs to some analytical space. It is actually not clear in general that they are well-defined as a map. Typically, one does not know in which space Feynman rules take their values. The first chapter aims at presenting a method to answer this question. \ty{It provides} a category in which graphs have a universal property that could \ty{enable us} to define Feynman rules. Our strategy is to introduce the category of {\bf TRAces and Permutations} (TRAPs), in which graphs have a universal property.

Whatever their status, when one applies Feynman rules to non trivial graphs, one typically finds divergent expressions. These then need to be renormalised which is often, somewhat improperly, described as removing divergences. This process is now fairly well-understood thanks in particular to the celebrated Connes-Kreimer Hopf algebra. Various renormalisation schemes exist in the \ty{literature} and the goal of the third chapter is to present a new one, which has the particularity of using {\bf multiple regularisation parameters}. It requires to \ty{define} a class of partial structures named {\bf locality structures} which are \ty{tailored} to encode the desired properties of the renormalised quantities.

So, once one has applied the Feynman rules, and renormalised the results, \ty{which} type of numbers does one find? Well, one finds in particular evaluations of Riemann's zeta function at integers greater than one, and their generalisation: Multiple Zeta Values (MZVs). Many open questions still stand regarding these numbers. In the second chapter, we investigate some {\bf generalisations of MZVs}, the algebraic identities they obey, and their relations to MZVs. Some of these generalisations also appear in some physical theories, e.g. conical zeta values appeared in some string theoretical computations.

At this point, one might ask why the second and third chapters are not exchanged. It would make more sense according to the previous discussion, but this is not the logic I followed. I chose to first talk about the most current research directions I am following. This is because the original physical question have sometimes (and in particular for the second and third chapters) slipped in \ty{the back} of my mind and are not  my main motivation \ty{anymore}.

This being said, is the story of QFT over once one understands the number-theoretical content of Feynman graphs? The answer to this question is a resounding no. Many aspects of QFT, even mathematical aspects, were not at all mentioned in the discussion above. In particular, I only described perturbative aspects of QFTs, but these are far from being the only ones. One way to obtain {\bf non-\ty{perturbative} results} for QFTs is through the analysis of their Schwinger-Dyson Equations (SDEs) and their Renormalisation Group Equations (RGEs). While the famous RGEs encode how the theory behave under changes of scales, the less known SDEs are, in some sense, equations of motion of QFTs. In the fourth chapter I argue that \'Ecalle's  {\bf resurgence theory} offers tools to reach these non-perturbative sectors of QFTs. In particular I discuss how it provides a mechanism to build two-point functions that are not just formal \ty{perturbative} series but true functions analytic in some complex domains. \\

As already mentioned, this thesis is based in particular on nine research articles. Let me list them all together here, regrouped by the chapters they contribute to:
\begin{itemize}
 \item Chapter 1: \cite{ClFoPa20} and \cite{ClFoPa21}.
 \item Chapter 2: \cite{Cl20} and \cite{ClPe23}.
 \item Chapter 3: \cite{CGPZ1}, \cite{CGPZ2} and \cite{CFLP22}.
 \item Chapter 4: \cite{Cl21} and \cite{BC19}.
\end{itemize}
Some results of other articles are of course also used. These articles are quoted below.

\subsection*{Chapter 1: PROPs and TRAPs for QFT}

This first chapter is based on the unpublished preprint \cite{ClFoPa20} and the paper \cite{ClFoPa21}, both written with coauthors Lo\"ic Foissy and Sylvie Paycha. The work presented in this chapter aims at offering a way to rigorously build Feynman rules for QFTs.\\

PROPs is an \ty{abbreviation} for PROducts and Permutations. This structure can be defined in a very algebraic way, as \emph{a strict symmetric monoidal  category enriched over $\Vect_\K$ whose objects are indexed by natural numbers $\{[n]\}$ and whose monoidal product is given by the addition:
$[n]\otimes[m]=[n+m]$} (Definition \ref{defn:PROP_cat}). The first chapter starts with a very quick overview of some aspects of category theory; the ones that allow to understand this definition. Then working out the various aspects of this definition allows us to reformulate it in a much less compact, but somewhat more easily tractable form, the one of Definition \ref{def:prop}. 

Roughly speaking, a PROP is a collection of vector spaces indexed by two positive integers. One of these can be seen as the number of ``\ty{inputs}'' and the other as the number of ``outputs''. Then one has a horizontal concatenation (where the numbers of inputs and outputs are added) and a vertical concatenation, where the inputs of the second element are branched to the outputs of the first element. This can be done when the number of inputs and outputs coincide. Furthermore the symmetric group acts on the inputs and the outputs.

This more pedestrian definition is used to build examples of PROPs. The first example (Definition-Proposition \ref{defi:Hom_V}) is the PROP $\Hom_V$ of linear morphisms between tensor products of a finite dimensional vector space $V$. To generalise this to infinite dimensional vector spaces, we may put some more structure on these vector spaces. This is done in the second example: Theorem \ref{thm:Hom_V_generalised}, which is the PROP $\Hom_V^c$ of continuous linear maps between completed tensor powers of a Fréchet nuclear space $V$. 

We then introduce \emph{generalised graphs} (Definition \ref{def:graph}). Roughly speaking, a generalised graph consists of vertices, oriented internal edges between two vertices, ingoing edges (with no origin but that arrive on a vertex), outgoing edges (with no end but that start on a vertex) and ingoing-outgoing edges, that just ``go through'' without being attached at either \ty{end} to a vertex. Furthermore, the inputs of the graph (given by ingoing edges and ingoing-outgoing edges) are labelled, as well as the outputs of the graph (given by outgoing edges and ingoing-outgoing edges). We then show in Theorem \ref{theo:ProP_graph} that the set of generalised graphs can be endowed with a PROP structure.

The next step is to introduce \emph{indecomposable graphs} (Definition \ref{def:indecomposable}). Roughly speaking, a generalised graph is indecomposable if it has at least one vertex, no ingoing-outgoing edges and can not be written in a non-trivial way as the horizontal or vertical concatenation of two subgraphs. \cy{Let $\Gri$ be} the set of indecomposable generalised graphs. The first important result of this chapter is then
\begin{theo*} (Theorem \ref{thm:freeness_Gr})
 The PROP of generalised graphs is the free PROP over indecomposable graphs.
\end{theo*}
Notice that this is a rather ``folklore'' result in that it is often mentioned in passing but, to the best of our knowledge, no proof was given before \cite{ClFoPa20}. This is possibly due to the technical issues of the proof but most probably to the fact that this result turns out to not have many consequences. Indeed, our motivation to study free PROPs was to use their universal property to build a Feynman map, i.e. a map $F_\tau:\calG_\tau\longrightarrow X$ with $\calG_\tau$ the set of Feynman graphs of the Quantum Field Theory $\tau$ and $X$ a suitable analytical space. Therefore, to use this universal property of generalised graphs in the category of PROPs to build such a map one needs to show that the target algebra $X$ carries a PROP structure and that there exists a morphism $f:\Gri\longrightarrow X$ with the right algebraic properties. The universal property of generalised graphs then gives a lift of $f$ to all generalised graphs from which we obtain the desired map $F_\tau$.

The first of these tasks is rather formidable but essentially unavoidable. However the second one is not less difficult, and seems rather unnecessary. Indeed, the practice of computing Feynman graphs tells us that $F_\tau$ is hard to build in particular on indecomposable graphs. Therefore this approach does not \ty{much reduce} the difficulty of \ty{the problem.}

This observation has two immediate consequences. First I do not speak of the possible generalisations of this first result such as similar universal properties for planar graphs or for decorated generalised graphs. The second is the inevitable consequence that PROPs are not the right objects to build Feynman rules. This is why we introduced \cy{what we thought was a new structure: TRAPs. It turned out that unitary TRAPs are wheeled PROPs (see \cite{Merkulov2006,JY15}). In this thesis I omit the proof that the categories of TRAPs and wheeled PROPs are isomorphic, see \cite[Theorem 5.3.1]{ClFoPa21} for this result. Therefore, TRAPs can be seen as non-unitary wheeled PROPs.} \\

% isomorphic to the category \\

TRAces and Permutations (TRAPs) are defined in \ref{def:Trap}. Again, roughly speaking, a TRAP is a \ty{family} of vector spaces indexed by two non-negative integers which one can still interpret as numbers of inputs and outputs. As for PROPs, a TRAP has a ``horizontal concatenation'' where the numbers of inputs and outputs are added. \ty{Furthermore}, a TRAP comes with ``partial traces'', where one output of an element is branched to the input of the same element. 
%Notice that TRAPs are closely related to wheeled PROPs (see \cite{Merkulov2006,JY15} for introductions to wheeled PROPs
a%nd \cite[Theorem 5.3.1]{ClFoPa21} for the relation between them and TRAPs).

This rather pedestrian definition allows us to find example of TRAPs. First (Proposition \ref{prop:HomV}) the TRAP $\Hom_V$ of linear maps between tensor powers of a finite dimensional vector space $V$. Then (Proposition \ref{prop:Homfr}) the TRAP $\Hom_V^{fr}$ of linear maps of finite rank between tensor powers of a vector space $V$. Still mimicking our work for PROPs, we show that the set of continuous linear maps between tensor powers of a nuclear \ty{Fr\'echet} space $V$ also admits  a TRAP structure. Our last example of this type (Theorem \ref{theo:Kinfty}) is the TRAP 
 $\mathcal{K}_{M}^\infty$ of generalised smooth kernels on the smooth finite dimensional orientable closed manifold $M$.
 
 In order to obtain a freeness result in the category of TRAPs we introduce in Definition \ref{def:solar_graphs} \emph{corolla oriented} generalised graphs. These are the generalised graphs that were previously introduced with the further structure that the sets of ingoing and outgoing edges of each of the vertices are totally ordered. We also introduce \emph{solar} graphs as graphs that have no ingoing-outgoing edges. We further define a \emph{decorated graph} (corolla ordered or not) in Definition \ref{defidecorations} to be a generalised graph $G$ with a decoration map $d_G:V(G)\longrightarrow X=(X(k,l))_{k,l\geqslant 0}$ such that if $v\in V(G)$ has $k$ inputs and $l$ outputs, then $d_G(v)\in X(k,l)$; i.e. the decoration respects the number of inputs and outputs of the vertices.
We have then our second important result of this first chapter
\begin{theo*} (Theorem \ref{thm:freetraps})
 Solar corolla oriented generalised graphs decorated by $X$ are the free TRAP generated by $X$.
\end{theo*}
In particular (Remark \ref{rk:PhiPGrX}), if the collection of sets $X$ has a TRAP structure, then we have a canonical morphism from solar corolla oriented generalised graphs decorated by $X$ to $X$.

Our next result is Definition-Proposition \ref{propverticalconcatenation}. We show that from a (unitary) TRAP, one can always build a PROP. In particular from the TRAP structures $\Hom_V$ and $\Hom_V^c$, one can rebuilt the PROP structures on the same sets and denoted by the same symbols.

We then show that iterating the partial traces of a TRAP gives an operation that can be called generalised trace since it is defined on any TRAP. The generalised trace is introduced in Definition \ref{defi:generalised_Traces} and its properties are given in Proposition \ref{prop:generalised_traces}.

The universal property of the TRAP of solar corolla oriented graphs allows us to define the generalised amplitude associated to a TRAP (Definition \ref{def:generalised-P-convolution}). We show in Proposition \ref{prop:gen_conv_vert_conc} that the generalised amplitude respects the horizontal concatenation of TRAPs but also the vertical concatenation of the PROP structure obtained from the TRAP structure. As an example of these various construction we work them all out explicitly in Theorem \ref{thm:generalised_convolution} for the TRAP $\mathcal{K}_M^\infty$ of generalised smooth kernels. In this Theorem, we state that the vertical concatenation  of two kernels corresponds to their generalised  convolution; that the associativity property of this generalised convolution 
     amounts to  the Fubini property for the corresponding multiple integrals; that the generalised trace  of a generalised kernel is given by its integral along the small diagonal and that the $\mathcal{K}_M^\infty$-amplitude is compatible with the horizontal and vertical concatenations. This Theorem \ref{thm:generalised_convolution} serves as a proof of concept of the type of results one can obtain with the TRAP technology we \ty{developed}. This chapter does not dwell deeper \ty{on} the applications of these results but instead ends on an ambitious possible application of the theory of TRAPs to perturbative QFT. \\
     
     The last section of this first chapter is devoted to this conjectural application. First, we set out to precisely state the problem we wish to tackle. Based on \cite{Ri91} we understand that a multivariate regularised Feynman rules for a scalar QFT over $\R^d$ should be a map 
     \begin{equation*}
        A_G:\vec z=(z_1,\cdots,z_E)\longmapsto\left(A_G(\vec z):\calD(\R^d)^{\otimes(k+l)}\longrightarrow\C\right).
     \end{equation*}
     The $z$s are needed because the naive application of Feynman rules to many important QFTs typically gives divergent expressions. These need to be regularised (and later renormalised) and the $z$s do just this: they are regularisation parameters.
     
     Of course this map $A_G$ should have some regularity properties. These properties are given by the theory of \emph{meromorphic germs} (Definition \ref{def:mero_germs_distrib}) with values in a space of distributions. Our first conjecture, numbered \ref{conjecture1}, is that the Feynman rules of a scalar QFT on $\R^d$, when applied to a Feynman graph of the same theory, gives a multivariate meromorphic germ with 
     values in distributions. Notice that I do not claim that this conjecture is entirely new. If I do not recall having read it in this precise form, it may nonetheless exist, and in any case is very much part of the ``folklore knowledge'' of the QFT community. In particular, a version of this conjecture for scalar QFT over a closed Riemannian manifold has been proved in \cite{DZ17}. Notice also that this conjecture could be made more precise by demanding that the germs are of a given order. 
     
     Moving on, I show how one can take an inductive limit on the number of variables for these distribution-valued meromorphic germs. On the spaces obtained after this limit one can conjecturally define a horizontal concatenation and partial traces. The symmetric group acts naturally on these objects, so we can state our main conjecture:
     \begin{conj*} (Conjecture \ref{conj:TRAP_reg})
      The inductive limit of the spaces of meromorphic germs with linear poles and values on distributions \ty{carries} a TRAP structure.
     \end{conj*}
     I expect that the main difficulty with proving this result is showing that these spaces are stable under the action of the partial trace maps. However recent analytical results of \cite{DPS22} about spaces related to these give me some hope that this conjecture could actually be tackled.
     
     I conclude this section by showing that this conjecture gives a canonical way to build Feynman rules (Definition \ref{def:TRAP_Feynman_rules}) and that these Feynman rules have the expected analytical and algebraic properties (Proposition \ref{prop:rapport_conj}).
     
     \subsection*{Chapter 2: Generalisations of Multi Zeta Values}
     
     As its title suggests, the second chapter of this thesis deals with generalisations of \ty{Multi Zeta} Values (MZVs), and presents their algebraic and number-theoretic properties. It is based on the papers \cite{Cl20} and \cite{ClPe23} (the \ty{latter} written with Dorian Perrot). Some methods and results this work is based on also appeared in the earlier work \cite{CGPZ2}. My initial motivation \ty{for studying} these generalisations of MZVs was that some difficult open questions regarding MZVs could be easier to solve for these generalisations.\\
     
     This second chapter starts with various definitions from the combinatorics of words and the statement of some classical results of the theory of MZVs. The definitions playing the main role for us are the ones of the \emph{shuffle product} (Definition \ref{def:shuffle}) and of the $\lambda$-shuffle products (Definition \ref{def:stuffle}), with $\lambda$ a real number. For $\lambda=1$ one finds the classical \emph{stuffle product} (see Hoffman \cite{Ho00}), also called sticky shuffle or quasi-shuffle by various authors. Notice that for the $\lambda$-shuffles to be defined, the alphabet $\Omega$ must be a semigroup, which we assume to be commutative.
     
     MZVs can be seen as linear maps from some words to $\R$. They come in two \ty{flavours}: shuffle and stuffle. The shuffle MZVs are defined on words written in an alphabet 
     with two letters; I take $\{x,y\}$. The stuffle MZVs are defined on words written in the alphabet $\N^*$. Shuffle MZVs are iterated integrals while stuffle MZVs are iterated sums. Both of these maps are given in Definition \ref{def:MZVs}. They are also both defined on some words only, named convergent words (Definition \ref{def:conv_words}). The names ``shuffle'' and ``stuffle'' MZVs are justified by Theorem \ref{theo:alg_prop_MZVs_stuffle_shuffle} which precisely states that shuffle (resp.~stuffle) MZVs are algebra morphisms for the shuffle (resp. stuffle) product of words. These two versions of MZVs are linked by the branched binarisation map (Definition \ref{defn:bin_map_words}) according to Kontsevitch's relation (Theorem \ref{thm:shuffle_stuffle_words}). The last property of MZVs we are interested into is the famous \emph{Hoffman's regularisation relation}, stated in Theorem \ref{theo:Hoffman_reg_relations}.
     
     Our goal is to find a generalisation to rooted forests of MZVs that have all these properties: exist in two versions (iterated sums and iterated integrals), are algebra morphisms for generalisation to rooted forests of the shuffle and stuffle products of words, are related by a branched version of the binarisation map, and obey a generalisation of Hoffman's regularisation relation. My original motivation to study this problem was the observation that trees are sometimes easier to deal with than words.
     
     In order to tackle this task, we first introduce various classical elements of the combinatorics of rooted trees and forests. I will not quote them all here, but for clarity let me say that rooted forests are \cy{obtained by  concatenating} rooted trees, which are connected directed graphs. Our trees and forests are always rooted (thus we will sometimes drop this adjective) and non-planar. In other words, they are commutative. They are often decorated, in the sense that we have a decoration map from their vertices to some decoration set $\Omega$. An important structure is the \emph{branching map} $B_+$ which to a forest $F=T_1\cdots T_n$ associates the tree obtained by adding a new vertex (eventually decorated) to $F$ and linking it to the roots of all of the $T_i$s. Therefore this new vertex becomes the root of the tree $B_+(F)$.
     
     This grafting operation is important because it endows rooted forests with an operated structure (see Definition \ref{def:op_structures}, a notion introduced in \cite{Guo07}). Theorem \ref{thm:univ_prop_tree} is then a classical result, namely that rooted forests, together with the \cy{grafting} map, are the initial object in the category of operated algebras. Roughly speaking, an operated algebra is an algebra $A$ together with a family of maps from $A$ to $A$ indexed by a set (see \cite{Guo07}). This universal result is due to \cite{KP13,Guo07} and is the result we will use to build our generalisations of MZVs. \cy{These generalisations will be built using the branching of maps (Definition \ref{defn:phi_hat}) as well as the} lifting of maps (Definition \ref{defn:phi_sharp}).
     
     The next combinatorial structure that we introduce is less classical: it is the shuffle and $\lambda$-shuffle of rooted forests (Definition \ref{defn:shuffle_tree}). We observe (Proposition \ref{prop:shuffles_trees}) that these products are unital, commutative but not associative. 
     
     Words can be endowed with various structures of operated algebras, one for each choice of a $\lambda$-shuffle product, thus one for each choice of $\lambda\in\R$. The universal property of the rooted forests then implies the existence of morphisms of operated algebras from forests to words, one for each weight $\lambda\in\R$. We call these morphisms \emph{flattening maps} (Definition \ref{defn:flattening_maps}). These flattening maps give us a way to go back to words and we take this opportunity to define the words versions of the lifted and branched maps (Definition \ref{def:branched_lifted_words}). The last structure we need to introduce is rather classical:  \emph{Rota-Baxter operators} (Definition \ref{def:RB_maps}). 
     
     We can then state and prove the first main result of this second chapter. It is from \cite[Theorem 2.13]{CGPZ1} and \cite[Theorems 2.20 and 5.8]{Cl20}. It can be stated as follows.
     \begin{theo*} (Theorem \ref{thm:flattening})
      Let $A$ be a commutative algebra and $P\in\End(A)$ a linear map. Then the following are equivalent:
       \begin{enumerate}
  \item $P$ is a Rota-Baxter map of weight $\lambda$.
  \item The branching of $P$ factorises through words via the flattening map of weight $\lambda$.
  \item The words version of the branching of $P$ is an algebra morphism for the $\lambda$-shuffle product.
  \item The forests version of the branching of $P$ is an algebra morphism for the $\lambda$-shuffle product of rooted forests.
 \end{enumerate}
     \end{theo*}
     This result will be the cornerstone of our construction to generalisations of MZVs and their study, which are the next topics of the second chapter.\\
     
     The previous theorem will be used in analytic spaces of functions, in particular on \emph{log-polyhomogeneous symbols} which we introduce in Definition
     \ref{def:log_poly_hom}. We also state without proofs some useful properties of these objects in Propositions \ref{prop:translation} and \ref{prop:bifiltration}. To control the behaviour of log-polyhomogeneous symbols we use the \emph{Euler-MacLaurin formula} with which the symbols behave in the way given by Proposition \ref{prop:prop_P}.
     
     At this point we can at last use our method of branching and lifting of maps to log-polyhomogeneous symbols to define our first generalisation of MZVs to forests. The crucial point is that we can control the asymptotic behaviour of these branched maps rather \ty{simply}: this is done in Proposition \ref{prop:order_calz}. In particular we can identify the forests on which these branched maps have a behaviour at infinity mild enough so that we can take the suitable limit. This motivates our definition of \emph{convergent forests} (Definition \ref{def:convergent_forests}) which generalise the \ty{aforementioned} convergent words. \emph{Arborified Zeta Values} (AZVs) are defined on convergent forests as these limits in Definition-Proposition \ref{defnprop:arborified_zeta}. From their definition we directly have that AZVs define an algebra morphism for the concatenation product of forests (Proposition \ref{prop:stuffle_alg_mor}).
     
     The same construction can be applied to words and allows us to give new proofs to classical results of the theory of MZVs. In particular, we can prove that stuffle MZVs are defined on convergent words (Definition-Proposition \ref{defnprop:zeta_stuffle}) and that they form algebra morphisms for the stuffle products (Proposition \ref{prop:zeta_stuffle_mor}). But our main result concerns AZVs: it is half of the second main result of this second chapter. We state this result in full here although the shuffle aspects will only make sense in the next paragraph.
     \begin{theo*} (Theorems \ref{thm:main_result_stuffle} and \ref{thm:main_result_shuffle})
      The stuffle (resp. starred stuffle, shuffle) AZVs can be written as finite linear combinations of stuffle (resp. starred stuffle, shuffle) MZVs with rational coefficients. These expressions are given by the 1-flattening (resp. -1-flattening, 0-flattening) of forests.
      
      \ty{Furthermore} stuffle (resp. starred stuffle, shuffle) AZVs form algebra morphisms for the stuffle (resp. anti-stuffle, shuffle) products of rooted forests.
     \end{theo*}
     The next section of this second chapter is dedicated to building shuffle AZVs and proving the shuffle part of this theorem. Shuffle AZVs are constructed in a slightly different way: through Chen \ty{integrals} (Proposition-Definition \ref{defnprop:Chen_int}). We define them as a map which can then be arborified (Definition \ref{defn:arborified_chen_int}). Then our first main result gives almost for free that arborified Chen integrals are linear combination of usual Chen integrals with \ty{integer} coefficients (Proposition \ref{prop:chen_int_arbo_words}).
     
     Using known properties of single variable multiple polylogarithms (Definition \ref{defn:multiple_polylogs}) we are able to define single variable arborified multiple polylogarithms in Definition-Proposition \ref{defnprop:arbo_polylogs}. Their properties are proved in Theorem \ref{thm:arborified_polylogs} and we are also able to give a new proof of a classical property of the usual multiple polylogarithms (Proposition \ref{prop:polylogs_shuffle_mor}).
     
     In the same way that MZVs can be \ty{defined} as a limit of multiple polylogarithms, we define shuffle AZVs as a limit of arborified multiple polylogarithms in Definition \ref{defn:shuffle_AZVs}. The shuffle part of the second main result of this chapter, which has already been written, is then a direct consequence of the \ty{aforementioned} properties of arborified multiple polylogarithms. As before, we also obtain a new proof that shuffle MZVs form an algebra morphism for the shuffle product of words (Proposition \ref{prop:shuffle_zeta_alg_mor_shuffle}).
     
     At this point we have gone quite far in our quest of generalising MZVs to rooted forests. We have built
     shuffle at stuffle AZVs which each generalise \ty{their respective} version of MZVs. We have shown that these AZVs build algebra morphisms for non-associative generalisations of the shuffle and stuffle products. We have also been able to show in a rather constructive way that AZVs are linear combinations of MZVs with rational coefficients. The only things missing are Kontsevitch's relation and Hoffman's regularisation relations.
     
     To tackle this next point we endow the algebra of forests decorated by the set $\{x,y\}$ with a new structure of operated algebra. This gives a natural and purely algebraic generalisation of the binarisation map, constructed in Definition \ref{defn:branched_bin_map}. However we show in Theorem \ref{thm:relation_shuffle_stuffle} that the generalisation of Kontsevitch's relation does not hold with this branched binarisation map. However we still try to look for a generalisation of Hoffman's regularisation relation. It turns out that we have two legitimate candidates for this, but Proposition  \ref{prop:branched_Hoffman} invalidates one, and while Proposition \ref{prop:conv_Hoffman_trees} gives us hope for the second one, one quickly finds a counter-example (e.g. Equation \eqref{eq:counter_example_Hoffman}). The situation of our \ty{attempt} is summarised in Table \ref{table:MZV_AZV}.\\
     
     This situation, admittedly rather bleak, is improved with our next result. It essentially states that shuffle AZVs do admit a series representation given by the branched binarisation map, just not the series defining stuffle AZVs. This is the third important result of this second chapter upon which most of the rest of the chapter is built.
     \begin{theo*} (Theorem \ref{thm:integral_sum})
      Shuffle AZVs can be written as multiple series given by the branched binarisation map.
     \end{theo*}
     This result suggests \ty{we} first study these multiple series. They are another generalisation to rooted forests of stuffle MZVs and we call them \emph{Tree Zeta Values} (TZVs) in Definition \ref{defn:tree_zeta_values}. From the previous theorem as well as from what we already proved for AZVs we \ty{straightforwardly} derive some properties of TZVs. These are stated in Proposition \ref{prop:crucial} and Theorem \ref{thm:tree_zeta_MZVs}.
     
     To derive further properties of TZVs, we define (Definition \ref{defn:yew}) a new product, the upsilon product, which is commutative and unital but not associative (Proposition \ref{prop:properties_yew}). Theorem \ref{thm:yew_formula} gives an inductive formula for the upsilon product. We further show that it stabilises convergent forests (Proposition \ref{prop:yew_stabilises_conv_forests}) and that TZVs form an algebra morphism for this product (Theorem \ref{thm:yew_TZVs}). We also show that the upsilon product is associative on ladder trees, a fact that allows us to define a flattening map associated to the upsilon product (Definition \ref{defn:flattening_yew}). We relate this flattening with a previous flattening and our binarisation maps in Proposition \ref{thm:comm_fl_fl_yew}. Putting all these facts together we obtain in Corollary \ref{coro:TZVs_stuffle_MZVs} another explicit way to express TZVs in terms of MZVs through the upsilon-flattening.
     
     The next section is dedicated to applications of our results to other generalisations of MZVs. The first of these applications is to \emph{Mordell-Tornheim zetas} (Definition \ref{defn:MT}). We obtain with Proposition \ref{prop:MT_s1_0} a formula to express some Mordell-Tornheim zetas but our main result regarding these object is Theorem \ref{thm:MT_tree} which gives new proofs for their convergence domain and number-theoretical contents. We also manage to obtain a formula for Mordell-Tornheim zetas that seems to be new; namely Equation \eqref{eq:expression_MT}.
     
     However our main applications of our results on AZVs and TZVs regard \emph{Conical Zeta Values} (CZVs). These appear in particular in the amplitude of some superstring theories and are introduced in this text in Definition \ref{defn:useful_stuff_cone}. We are interested in these objects because TZVs are CZVs. So characterising which CZVs are TZVs allows us to obtain general results for these objects, such as formula to express them in terms of MZVs. 
     
     For this purpose we define tree-like cones and their convergent counterparts in Definition \ref{prop:forest_cones}. We then directly have a relation for CZVs on tree-like cones and TZVs, given in Proposition \ref{prop:forest_cones}. This in turn provides a formula (Theorem \ref{thm:tree_CZVs_MZVs}) for CZVs on tree-like cones. Thus it is important for us to be able to determine which cones are tree-like. This is done in Theorem \ref{thm:carac_tree_cone} which requires a few definitions to be stated and proved, definitions that we skip here. The last subsection, Subsection \ref{subsec:computations_CZVs} contains various examples of computations of CZVs using the methods \ty{developed} through this chapter.\\
     
     As in the first one, the second chapter ends with a section where open questions are listed. There are natural questions of enumerative combinatorics (for example, the dimension of the associator of the shuffle of rooted forests) which have not yet been looked at. Another line of research that I am currently exploring with my PhD student Douglas Modesto is to look for algebraic structures (e.g. non associative Hopf algebras) related to the shuffle and stuffle products of rooted forests. I have also started to look with collaborators Lucile Sautot and Ludovic Journaux at possible applications of these products to data science.
     
     However, the question I would most like to tackle is the initial one I asked. We are still missing a fully satisfactory generalisation of MZVs to rooted forests. At this moment, what seems to be missing is a generalisation of the stuffle product to rooted forests such that TZVs are an algebra morphism for this product. I list various \ty{unsuccessful} attempts I made at this question, and \ty{offers} two lines of research that still seem rather promising to me. The first one is to exploit the fact that TZVs are CZVs and that these latter objects are coalgebra morphisms for cone decomposition and algebra morphisms for the Minkowski product of cones. The goal would be to find a way to decompose (non-trivially) any Minkowski product of tree-like cones in tree-like cones. The second line of research (which we started to look at with Douglas Modesto and Pierre Catoire) is to work in the category of tridendriform algebras.
     
     \subsection*{Chapter 3: Locality Structures and multivariate renormalisation}
     
     The third chapter is based on \cite{CGPZ1} and \cite{CGPZ2} written with Li Guo, Sylvie Paycha and Bin Zhang. Some ideas from \cite{CFLP22} (written with Lo\"ic Foissy, Diego Lopez Valencia and Sylvie Paycha) are also present in this chapter. Our motivation to introduce locality structures was to have a mathematical framework suited to build a multivariate renormalisation scheme. This was fully achieved and we have applied the scheme to various objects, one of which is presented below. However, there are also open questions regarding locality structures.\\
     
     Locality is of course a concept stemming from Physics. It can be summed up as the idea that events that are far apart in some sense cannot influence each other in short time. With coauthors Li Guo, Sylvie Paycha and Bin Zhang we \ty{developed} a mathematical framework that is suited to encode this concept: the framework of \emph{locality structures}.
     
     The simplest of these structures are locality sets (Definition \ref{defn:independence}). A locality set is a set $X$ together with a symmetric relation $\top\subseteq X\times X$ dubbed \emph{independence relation}. If $(x,y)\in\top$, we say that $x$ and $y$ are \emph{independent}. \cy{Then for $Y\subseteq X$, $Y$ inherits a locality relation from $X$ and the \emph{polar set} $Y^\top$ of $Y$ is defined as the set of elements of $X$ independent to every elements of $Y$:
%      \begin{equation*}
      $Y^T:=\{x\in X|\,\forall y\in Y,\,x\top y\}.$
%      \end{equation*}
} 
     
     We then define the notion of locality vector space (Definition \ref{defn:lvs}) for which we require the independence relation to respect the linear structure in \cy{the following sens: for any subset $X\subseteq V$, its polar set $X^\top$ must be a vector subspace of $V$}. Together with the notion of locality linear map (Definition \ref{defn:locallmap}) we obtain the category of locality vector spaces in which we can already state some results such as Proposition \ref{lem:locallinearmap}.
     
     Some of the most interesting locality structures are the ones with a product. Locality semigroups, monoids and groups are defined in Definition \ref{defn:lsg}. The basic idea is that we have only partial products, which are defined only on pairs of independent elements. Locality algebras (Definition \ref{defn:localisedalgebra}) are one of the structures we will use the most. We also introduce locality Rota-Baxter operators (Definition \ref{defn:lrba}) on locality algebras and show in Proposition \ref{prop:multpi} that some classical results on Rota-Baxter operators stay true in the locality setup. We give an important analytical example of locality vector space which is also a locality algebra with important locality subalgebras in Propositions \ref{prop:mult_germs_loc_alg} and \ref{pp:merodecomp}. This example is on germs of multivariate meromorphic functions with linear poles at the origin and rational coefficients which are introduced in Definition \ref{defn:mero_germs_lin_rational}.
     
     Our next goal is to define the locality versions of tensor products, a task that turns out to be surprisingly subtle. We first define the final locality relation on a set (Definition \ref{defn:final_relation}), a terminology we borrow from topology. We provide in Proposition \ref{prop:description_final_relation} a description of this final locality relation in a simple case. A special case of a final locality relation is the quotient locality on quotients of locality vector spaces which we introduce in Definition \ref{defn:quotientlocality}. A curious but important observation is that the quotient of two locality vector spaces is not always a locality vector space for the quotient locality relation.
     
     In any case we define locality tensor products as quotients of locality vector spaces in Definition \ref{defn:loctensprod1}. Higher locality cartesian and tensor products are also defined (Definition  \ref{defn:loctensalgebra}). These have many useful properties that were shown in \cite{CFLP22} but these are beyond the scope of this chapter.\\
     
     Since our original motivation is to build a multivariate renormalisation scheme, for which locality structures are adapted, we need a locality version of the Birkhoff-Hopf factorisation. Thus we need the locality versions of coalgebraic structures. This is done first in Definition \ref{defn:colocalcoproduct} where locality coproducts are introduced. We then define successively locality bialgebras (Definition \ref{defn:locbialgebra}), the convolution product of locality linear maps (Definition \ref{def:conv_prod_loc}) and locality Hopf algebras (Definition \ref{defn:LHopf}). Then, as in the usual non-locality case, the antipode of a locality Hopf algebra can be written in terms of its unit and counit (Proposition \ref{prop:localisedantipode}). Our next step is Theorem \ref{thm:loc_conv_prod} which states that the usual properties of convolution of linear maps stay true in the locality framework (stability, group property and existence of an inverse).
     
     With these tools at our disposal we can prove the main result of this section, namely the locality version of the Birkhoff-Hopf factorisation.
     \begin{theo*} (Theorem \ref{thm:abflhopf})
      Let $(H,\top_H)$ be a locality Hopf algebra, $(A,\top_A)$ be a locality algebra and $\phi:(H,\top_H)\longrightarrow(A,\top_A)$ be a locality algebra morphism. Then $\phi$ admits a unique locality Birkhoff-Hopf factorisation $\phi= \phi_1^{\star (-1)} \star \phi_2$ with $\phi_i: H\to \K+A_i$ locality algebra homomorphisms.
      
      Furthermore, if $A_1$ is a locality subalgebra and $A_2$ is a locality ideal, then we have $\phi_1^{\star(-1)}=\pi_1 \phi$, the projection of $\phi$ onto $A_1$ along $A_2$.
     \end{theo*}
     The \ty{striking} point of this result is the last sentence. Under some technical assumptions that hold true in our case of interest (multivariate meromorphic germs with linear poles), the renormalised values of $\phi$ are simply given the minimal \ty{subtraction} scheme! This can be understood as a payback for using multivariate renormalisation. At the price of working with more technical objects (namely multivariate instead of single variable germs) we keep track of the locality properties we wish to preserve through renormalisation. Since locality is preserved in some sense at each step of the process, we do not need to re-\ty{enforce} it at the very end. Thus the Birkhoff-Hopf factorisation reduces to a minimal \ty{subtraction}.
     
     The multivariate renormalisation scheme is a \ty{straightforwardly} consequence of the above theorem. The next section is dedicated to one application of this scheme. We first need to introduce the locality versions of the operated structures that were already discussed in the second chapter (Definitions \ref{defn:loc_op_set}, \ref{defn:basedlocsg} and \ref{defn:morphoplocstr}). We then recall the definition of the Connes-Kreimer coproduct of rooted forests (Definition \ref{def:CK_coproducts}) and finally define a natural locality relation on forests decorated by a locality set (Definition \ref{def:loc_rooted_forests}). Theorem \ref{thm:loc_Hopf_prop_forests} then states that the so-called properly decorated forests (which are introduced in Definition \ref{def:prop_dec_forests}) have the structure of a locality Hopf algebra.
     
     We then move on to show in Proposition \ref{prop:operatedalgebra} that these properly decorated forests, together with a restriction of the grafting operator, have the structure of a locality operated algebra. This structure is not random: we show in Theorem \ref{thm:univ_prop_trees_loc} that the universal property of rooted forests in the category of operated algebras still holds in the locality framework, when one replaces forests by properly decorated forests. At this point we have all the algebraic tools we need to define and characterise the multivariate renormalisation of Kreimer's toy model. After some analytical work, this is achieved in Definition-Theorem \ref{defthm:reg_ren_maps}.\\
     
     As in the previous chapters, this one ends with some open questions. The main one is a conjecture that I have hinted at in this summary. It has to do with the fact that the quotient of locality vector space is not always a locality vector space. But the locality tensor product is defined as a quotient of locality space, so is it a locality vector space? The conjecture states that it is.
     \begin{conj*} 
        Let $(E,\top)$ be any locality vector space, $V$ and $W$ any two locality \ty{vector subspaces} of $E$. Then $V\otimes_\top W$ is a locality vector space for the quotient locality relation.
     \end{conj*}
     We have given evidence for this conjecture in \cite{CFLP22}. First it holds in many cases of interests, namely locality vector spaces  with a locality basis and on Hilbert spaces where the locality is given by orthogonality. Furthemore, our work to find a counterexample was unsuccessful and actually suggested more precise conjectures. A possible approach to this conjecture would be to build a (locality) lattice from a locality vector space and show that building a counterexample would be equivalent to building a path with absurd properties in this lattice.
     
     There are many other open questions in the theory of locality structures. One is whether one can endow the category of locality vector spaces with a monoidal structure? This seems far more ambitious than the above conjecture but might nonetheless be relevant. It also suggests that locality categories are not yet very well understood. For example, we also have no notion of a locality dual. Another open question is the classification of locality Lie groups. We have shown that it would be stricly richer than usual non-locality Lie groups.
     
     \subsection*{Chapter 4: Resurgence in Quantum Field Theory}
     
     This last chapter is based on \cite{Cl21} and \cite{BC19}, the latter with Marc Bellon. It uses results from previous papers \cite{BC14,BC15,BC18} (with Marc Bellon) but also \cite{BS13,Be10} that I did not authored. The purpose of this chapter is to prove the summability à la \'Ecalle of the solution of a (truncated) Schwinger-Dyson Equation (SDE) and a Renormalisation Group Equation (RGE) and to argue that some related technics could be of use for asymptotically free QFTs and/or to provide a non-perturbative mass generation mechanism.\\
     
     After a short discussion of why resurgence theory is  relevant for Physics and in particular for QFT, this chapter introduces some basics concepts needed for \'Ecalle resurgence theory. Notice that we present only a small piece of resurgence theory \cite{Ecalle81,Ecalle81b,Ecalle81c}. First we present the (formal) Borel transform in Definition \ref{def:formal_Borel_tsfm} and state some of its properties in Proposition \ref{prop:properties_formal_Borel_tsfm}. We then give the notion of 1-Gevrey series in Definition \ref{def:Gevrey}, which are precisely the formal series whose Borel transform is convergent, as stated in Theorem \ref{thm:Gevrey_give_conv_Borel}. The Laplace transform is introduced in Definition \ref{def:Laplace_tsfm}. It is the last piece of the famous \emph{Borel-Laplace resummation method} (Definition \ref{defn:Borel_Laplace}). We state in Theorem \ref{thm:analyticity_domain_Borel_resum} the analyticity domain of Borel-Laplace resummed functions, which nicely avoid Dyson's argument for the summability of perturbative series in QFT. However, one can check that in practice Borel-Ecalle resummation method is not enough to resum most divergent series of QFT, due to the presence of singularities of the Borel transform in the direction where one wishes to perform the Laplace integrals. \'Ecalle's resurgence theory offers a (literal) workaround for this issue.
     
     We then turn our attention to resurgence theory itself, and more specifically to its aspects that allow to generalise the Borel-Laplace resummation method. Resurgent functions are introduced in Definition \ref{defn:resum_function} and their stability under the convolution product is stated in Theorem \ref{thm:stability_resu_fct}. We state one more result of resurgent analysis, Theorem \ref{thm:bound_resu_Sauzin}, which provides a bound on convolutions of resurgent functions.
     
     Another crucial piece of \'Ecalle's generalisation of the Borel-Laplace resummation method is the notion of well-behaved averages. For this we first need to work in the ramified plane (Definition \ref{def:ramified_plane}). We then define averages (Definition \ref{def:averages}) and the most important ones for our purpose: those that are well-behaved (Definition \ref{defn:well_behaved_averages}). With this last piece we are able to define the \emph{Borel-\'Ecalle resummation method} which we give in the form of Theorem \ref{thm:Borel_Ecalle_resummation}. \\
     
     Our goal is to give an example of application of this method to a QFT model. Thus the next section is dedicated to presenting this model and the equations we want to tackle. The model is an exactly supersymmetric one, called Wess-Zumino. We aim at building an analytic solution to the system formed of a truncation of the Schwinger-Dyson Equation (SDE) (Equation \eqref{eq:SDnlin}) and the Renormalisation Group Equation (Equation \eqref{eq:RGE}). We first rewrite equations in more tractable forms: Equations \eqref{eq:recursion_gamma} and \eqref{SDE}. These two equations are then Borel-transformed. 
     
     The starting point of our analysis is the asymptotic behavior of the coefficients of the anomalous dimension of the theory. This asymptotic behaviour is given in Equation \eqref{eq:asymp_behavior_cn}. We also take for granted a fact that was ``proven'' in the physicists's sense of the term, namely that the Borel transformed anomalous dimension of the Wess-Zumino model is resurgent (Claim \ref{thm:resurgence_gamma}). We do not attempt to rigorously prove this as it would probably require some sophisticated methods (e.g. resurgence monomials) which would in turn offer a completely different approach to the problem we are tackling.
     
     We then turn our attention to the analysis of the RGE. We first give in Proposition \ref{prop:solution_RGE} a solution of this equation. This solution is then shown in Proposition \ref{prop:G_one_Gevrey} to be 1-Gevrey. This implies that the Borel transform of the two-point function is convergent near the origin. It is more difficult, but still doable, to show that it is actually resurgent. This is done in Theorem \ref{thm:resurgence_two_points_function} after some preliminary work.
     
     The last step is to show that this Borel transform admits suitable analytical bounds at infinity. We first observe in Subsection \ref{subsect:need_SDE} that this seems to be doable only when one uses the SDE on top of the RGE. This turns out to be rather delicate and requires a series of bounds that we cannot easily summarise. Eventually, Theorem \ref{thm:bound_two_point_infinity} gives us the required bound. This, together with previous result, directly implies the main result of this section.
\begin{theo*} (Corollary \ref{cor:main_result})
 The solution of the renormalisation group equation and the Schwinger-Dyson equation is Borel-\'Ecalle resummable. For any $L$ in $\R^*_+$, the resummed function $a\to G^{\rm res}(a,L)$ is analytic in the open subset of $\C$ defined by
 \begin{equation*}
  \left|a-\frac{1}{20L}\right| < \frac{1}{20L}.
 \end{equation*}
%  for any $L$ in $\R^*_+$.
\end{theo*}
We finish this section by a short discussion where we observe that the analyticity domain of the resummed function escapes Dyson's argument. Furthermore, let us assume that the asymptotic bound for the Borel transform of the two-point function is saturated; meaning that the bound we find is an equivalence of functions. Then this provides a nice non-perturbative mass generation mechanism, and that the mass gap of the theory is given by the asymptotic behaviour of the two-point function. \\

Of course one of the holy \ty{grails} of theoretical Physics is to explain the mass gap of Quantum Chromodynamics (QCD). The last observation suggests that resurgence theory is a contender for partially achieving this goal. Thus the last section of this thesis aims at presenting how (generalisations of) the methods presented above could work for asymptotically free theories. This section is rather speculative since we do not work out one specific theory.

First we recall that 't Hooft argued that the two-point function of an asymptotically free QFT can only be analytic in a horn-shaped domain, which we illustrate in Figure \ref{fig:an_dom_tHooft}. We then present another aspect of \'Ecalle's theory, namely acceleration. The idea is that a Borel transform might have an asymptotic behaviour too divergent to allow the Laplace transform. \'Ecalle then argues that one can perform an \emph{acceleration}, in practice changing variable in the Borel plane in a specific way (See Equation \eqref{def_acc}). We then show that the simplest of such acceleration gives rise to a resummed function that still solves the equations we started from (this is \'Ecalle's result) and is analytic in a domain illustrated in Figure \ref{fig:an_dom_acc} which coincides with the analyliticity domain predicted by 't Hooft!

This allows us to make a reasonable conjecture, namely that the two-point function of \ty{asymptotically} free QFT should not be summable à la Borel-\'Ecalle, but rather accelero-summable:
\begin{conj*} (Conjecture  \ref{conj:asympt_bound})
 For an asymptotically free QFT it exists
 for each values of the \ty{kinematic} parameter $\Lambda$ a number $\sigma_\Lambda$ such that the formal series $G(\Lambda,a)$ is accelero-summable with acceleratrix $F(y)=\frac{1}{\sigma_\Lambda}\log(y)$.
\end{conj*}
Furthermore, the numbers $\sigma_\Lambda$ should obey bounds depending on the first coefficient of the beta function of the theory.

We end this chapter with some other open questions that are not formulated as conjectures. The first one concerns the non-\ty{perturbative} mass generation mechanism that was \ty{mentioned} above. Does this mechanism fits within the accelero-summation framework? Another question, which is of importance to the physicists's approach of resurgence theory is whether or not one could find a Sokal-Watson's theorem for Borel-\'Ecalle (accelero-)summation. Finally, \'Ecalle's resummation method demands to choose a well-behaved average. It is still unclear how the resummed function depends on this choice.

% \section*{Directions de recherche}
% 
% \addcontentsline{toc}{section}{Directions de recherche}
% 
% Potentiellement titre à changer. Basé sur le texte à la con que l'HDR m'avait demandé de faire, un peu nettoyé et etoffé.
    
    \section*{Notations and list of symbols}
    
    \addcontentsline{toc}{section}{Notations and list of symbols}
    
    I use $\N:=\Z_{\geq0}$ and $\N^*:=\Z_{\geq1}$. Furthermore, for any $n\in\N$, I set $[n]:=\{1,\cdots,n\}$. I also set $[0]:=\emptyset$. \\
    ~\\
    I also allow myself to use the ugly but practical notation $\R_{\geq x}:=\{y\in\R|y\geq x\}$ (where of course $\R$ can be replaced by other sets so that the notation makes sense, and $\geq$ can be replaced by $\leq$, $>$ or $<$). I also use $\R_+:=\R_{\geq0}$. Furthermore, as for $\Z$, I use the star $*$ to exclude the zero, for example $\R^*:=\R\setminus\{0\}$, $\R^*_+:=\R_+\setminus\{0\}$. \\ 
    ~\\
    This star should not be confused with the other star $\star$ that I use for the convolution products (in bialgebras or of resurgent functions). These two stars are used very differently and I expect no ambiguity from them. \\
    ~\\
    In order to enhance readability, I sometimes (improperly) write $k,l\in\N$ instead of $(k,l)\in\N^2$. In particular I use this short hand notation for the cases where $k$ and $l$ are subscripts. \\
    ~\\
    Below is a list of symbols used for various spaces and maps through the thesis.\\
    ~\\
%     \vspace{5cm}
%     ~\\
%     XXXXXX add resrgent spaces!!!\\
    \begin{tabular}{|c|c|}
\hline
    Symbol & Space or map \\ \hline \hline
    $\PGr$ & TRAP of corolla oriented generalised graphs. \\ \hline 
    $\PGr(X)$ & TRAP of corolla oriented generalised graphs decorated by $X$. \\ \hline
    $CS^\alpha$ & Classical (or polyhomogeneous) symbols of weight $\alpha\in\R$. \\ \hline
    $\C\dsetminus\Omega$ & $\Omega$-ramified plane (homotopy classes of rectifiable paths avoiding $\Omega$). \\ \hline
    $\calF$ & $\R$-vector space freely generated by rooted forests. \\ \hline
    $\calF_\Omega$ & $\R$-vector space freely generated by rooted forests decorated by the set $\Omega$. \\ \hline
    $\Gr$ & PROP of generalised graphs. \\ \hline
    $\Gri$ & Set of indecomposable generalised graphs. \\ \hline
    $\Hom_V$ & TRAP or PROP of linear maps between tensor powers of \\
    & a finite dimensional vector space $V$. \\ \hline
    $\Hom_V^c$ & TRAP or PROP of continuous linear maps between tensor powers of \\
    & a nuclear \ty{Fr\'echet} space $V$. \\ \hline
    $\Hom_V^{fr}$ & TRAP of linear maps of finite rank between tensor powers of a vector space $V$. \\ \hline
    $\mathcal{K}_{M}^\infty$ & TRAP of generalised smooth kernels on \\
    & the smooth finite dimensional orientable closed manifold $M$. \\ \hline
    $\calM_\Q$ & Vector space of these meromorphic germs with linear poles at zero \\
    & and rational coefficients. \\ \hline
    $\calM_+$ & Vector space of holomorphic germs at zeros (with rational coefficients). \\ \hline 
    $\calM_-^Q$ & Vector space of polar germs w.r.t. the system $Q$ of scalar products.\\ \hline
    $MT(s_1,\cdots,s_n|s)$ & Mordell-Tornheim zeta values. \\ \hline
    $\calP^{\alpha,k}$ & Log-polyhomogeneous symbols of order $(\alpha,k)\in\R\times\N$. \\ \hline
    $\widehat{\mathcal{R}}_\Omega$ & Set of $\Omega$-resurgent functions. \\ \hline
    $\calS^r$ & Symbols of weight $r\in\R$. \\ \hline
    $\rPGr(X)$ & TRAP of solar corolla oriented generalised graphs decorated by $X$. \\ \hline 
    $\rGr(X)$ & TRAP of solar generalised graphs decorated by $X$. \\ \hline
    $\calT$ & $\R$-vector space freely generated by rooted trees. \\ \hline
    $\calT_\Omega$ & $\R$-vector space freely generated by rooted trees decorated by the set $\Omega$. \\ \hline
    $\mathcal{U}_\Omega$ & Complex domain obtained from  radial cut starting from the first singularity of $\Omega$. \\ \hline
    $\widehat{\mathcal{U}}_\Omega$ & Set of uniform functions on $\C\dsetminus\Omega$. \\ \hline
    $W_\Omega$ & Set of words written in the alphabet $\Omega$. \\ \hline
    $\calW_\Omega$ & $\R$-vector space freely generated by words written in the alphabet $\Omega$. \\ \hline
    $z^{-1}\C[[z^{-1}]]_1$ & Set of 1-Gevrey formal series. \\ \hline
    $\zeta$ & Conical zeta values (multiple series). \\ \hline
    $\zeta_\shuffle$ & Shuffle MZVs (iterated integrals). \\ \hline
    $\zeta_\stuffle$ & Stuffle MZVs (iterated series). \\ \hline
%     $\zeta_\shuffle$ & Shuffle MZVs (iterated integrals). \\ \hline
    $\zeta^\star_\stuffle$ & Starred stuffle MZVs (iterated series with non-strict inequalities). \\ \hline
    $\zeta_\shuffle^T$ & Shuffle AZVs (iterated integrals). \\ \hline
    $\zeta_\stuffle^T$ & Stuffle AZVs (iterated series). \\ \hline
    $\zeta^{T,\star}_\stuffle$ & Starred stuffle AZVs (iterated series with non-strict inequalities). \\ \hline
    $\zeta^t$ & Tree zeta values (multiple series). \\ \hline
\end{tabular}

% % % % %  COMMENTER CE QUI EST CI-DESSOUS.
% % 
%  \bibliographystyle{unsrt}
%  \addcontentsline{toc}{section}{Bibliography}
% \bibliography{HDR_biblio}
%  
% \end{document}

%% file: PROP.tex
% 
% \documentclass[11pt,twoside,a4paper]{book}
% % \documentclass[11pt,twoside,a4paper]{article}
% 
% \usepackage{HDR}
% 
% %%% SI ON LAISSE LES DEUX LIGNES SUIVANTES DANS LE STY, ELLES NE MARCHENT PAS... (probablement à cause de \input).
% \input{xy}
% \xyoption{all}
% 
% 
% 
% \begin{document}
% % % % 
% % % % % % COMMENTER CE QUI EST AVANT CA ET ENDDOCUMENT POUR COMPILER HDR.tex.

\chapter{PROPs and TRAPs for QFT} 
\label{chap:PROP}

I introduce structures that might help rigorously define Feynman rules for Quantum Field Theories.
Most of the results presented in this chapter were proven in \cite{ClFoPa20} and \cite{ClFoPa21}.

\section*{Introduction}

\addcontentsline{toc}{section}{Introduction}

\subsection*{Statement of the problem}

\addcontentsline{toc}{subsection}{Statement of the problem}

Feynman graphs, introduced in \cite{Fe49}, have somehow secured themselves into popular culture as well as into many many textbooks \cite{Peskin95,Weinberg95}. One can easily argue that they are the most well-known sight in the broad field of QFT. This is probably due to their unique graphical representation, apt to strike the imagination as carrying arcane meanings.

Is is not the point of this work to say that Feynman graphs are not mysterious. Quite the contrary: when, as a master student, I stubbled upon some Feynman graphs that needed to be computed, I was so unable to grasp how one is supposed to perform the computations that I decided to never compute any Feynman graphs. It is now more than ten years later and I have betrailed all I once hold dear and holy: I now wish to deal with them, although I still do not know how to evaluate them.

What I do know is that I don't want to simply apply Feynman rules. This is not only out of lazyness and bad memories. It is also because at some point in the last ten years, I have become a mathematician and I now do care if a computation has a meaning when I performe it. Therefore I first and foremost wish to understand Feynman rules, and prove that they are well-defined.

Before giving precisions about the question at hand, let me recall that for a QFT $\tau$, its lagrangian tells us which type of interactions can appear in its \ty{perturbative} series. To compute the terms of the perturbative series, Feynman devised in 1948 his \ty{aforementioned} {\bf Feynman graphs} (or Feynman diagrams). The {\bf Feynman rules} dictate which diagrams appear in a given theory and are a 
set of prescriptions that tell how to associate to each graph an analytical expression (typically an integral) that has to be evaluated. So formally speaking the Feynman rules of a QFT $\tau$ can be seen as a map 
\begin{equation*}
 F_\tau:\calG_\tau\longrightarrow X
\end{equation*}
with $\calG_\tau$ the set of Feynman graphs of the theory and $X$ some analytical space to be determined. In Section \ref{section:QFT} below we investigate what $X$ should be. It turns out that $X$ is a space of distribution valued meromorphic germs.

Phycisists notoriously do not care so much of the well-definess of their computations. This is not a judgement: they just have better things to do. So they do not try to rigourously define the map $F_\tau$ and actually do not necessarily think of it as a map. They instead apply it (or rather apply the ``rules'') to some graphs and observe that they end up in some analytical space.

This approach is extremely successful but quite frustrating for the mathematicians. In order to deal with this frustration I would like to \emph{define} the map $F_\tau$ and show that $F_\tau(G)$ lies in the right space $X$ for any Feynman \ty{graph} $G$. So this is a very analytical problem. It was actually solved in some cases and in particular recently for a QFT on a compact Riemannian manifold in \cite{DZ17}. In this paper the authors use a plethora of analytical methods to deal with the many challenges they encounter. The approach I wish to present here is quite different.

Indeed, I like to deal with analytical problem by finding a structure and a universal property that will solve the problem. So I wish to look for a category $\calC$ such that 
\begin{itemize}
 \item the space $X$ is an object in $\calC$,
 \item $\calC$ admits a free object that can be described in term of graphs.
\end{itemize}
I then wish to see the Feynman graphs as the free object of the category $\calC$. Their universal property will then give us the existence of essential properties of the map $F_\tau$, i.e. of the Feynman rules.

\subsection*{Content and main results}

\addcontentsline{toc}{subsection}{Content and main results}

In order to make the first steps of the program presented above, my collaborators Sylvie Paycha from Potsdam Universität and Loïc Foissy from Université du Littoral and myself started to peruse \ty{literature} looking for a category with the desired properties. We found out the PROPs (PROducts and Permutations) were good candidates. I will present PROPs in details below so won't say more about their definition here. The key point was a ``folklore theorem'', namely that graphs are a free PROP. This statement is usually stated in a very abstract categorical framework. For our purpose we needed to make it more pedestrian. We then set to prove it in the pedestrian approach. This is presented in the first three sections of this chapter.

In Section \ref{sec:PROP_def} I introduce PROPs, first in their categorical setting (Definition \ref{defn:PROP_cat}). I then work down this abstract definition to obtain a pedestrian definition of PROPs (Definition \ref{def:prop}).

This pedestrian approach allows me to give examples of PROPs in Section \ref{sec:PROP_ex}. In particular, I show in Theorem \ref{thm:Hom_V_generalised} that morphisms of tensor products of a nuclear Fréchet spaces can be endowed with a PROP structure. Since nuclear Fréchet spaces have similarities with the space $X$ where Feynman rules evaluate Feynman graphs, this is a step in the right direction. Another important example of a PROP is Theorem 
\ref{theo:ProP_graph} that (isoclasses of generalised) graphs have a PROP structure. These objects are introduced in Definitions \ref{def:graph} and \ref{def:morph_graphs}.

The first main result of this chapter in in Section \ref{sec:free_prop}. It is Theorem \ref{thm:freeness_Gr}, stating that the PROP of graphs introduced earlier is freely generated by indecomposable graphs. This result is a precise statement of the \ty{aforementioned} ``folklore theorem'' that graphs are a free PROP. However, as it is discussed after the sketch of proof in Subsection \ref{subsec:freePROP}, this result does not provide enough help to rigorously build Feynman rules.

\vspace{0.4cm}

This issues come from the existence of ``big loops'' in graphs that cannot be tamed by the PROP structure of graphs. Sylvie Paycha, Loïc Foissy and myself then started to look for another category that would be able to deal with these big loops. We came up with the notion of TRAPs (Traces and Permutations)\footnote{to our dismay, it was pointed out to us that a version of TRAPs had already \ty{been} introduced under the name of ``wheeled PROPs'', see Remark \ref{rk:wheeledPROP}.} which are the subject of the next four sections of this chapter. 

In Section \ref{sec:TRAP} I \ty{define} the category of TRAPs (Definitions \ref{def:Trap} and \ref{defn:trap_morphism}). Notice that only the pedestrian version of the definition is given. A categorical definition exists for unitary TRAPs but not, as far as I am aware, for non-unitary TRAPs.

Examples of TRAPs are presented in Section \ref{sec:TRAPs_ex}. In particular, Theorem \ref{thm:Hom_V_generalised} states that morphisms of tensor products of a nuclear Fréchet spaces admit a TRAP structure; and in Theorem \ref{theo:Kinfty} I obtain the same result for smoothing pseudo-differential operators.

The most important example of TRAP is build in Section \ref{sec:freeTRAP}: it is the TRAP of corolla oriented graphs (Definitions \ref{def:solar_graphs} and \ref{def:isoclassesTRAPgraphs}). Theorem \ref{thm:TRAPgraphs} is an important result: graphs also have a TRAP structure. But Section \ref{sec:freeTRAP} also contains the second main result of this chapter, namely Theorem \ref{thm:freetraps}. It states that the TRAP of graphs is actually freely generated by the decorations of the vertices of the graphs.

Section \ref{sec:applicationsTRAPs} is then devoted to the applications of these results. First I show in Proposition-Definition \ref{propverticalconcatenation} how TRAPs are related to PROPs. I then introduce in Definition \ref{defi:generalised_Traces} the notion of generalised traces for TRAPs and give some of their properties in Proposition \ref{prop:generalised_traces}. The most important of these applications for QFT is probably the notion of $P$-amplitude (Definition \ref{def:generalised-P-convolution}). Properties of $P$-amplitudes are given in Proposition \ref{prop:gen_conv_vert_conc}. These various constructions (vertical concatenation in TRAPs, generalised traces and amplitudes) are then applied in Theorem \ref{thm:generalised_convolution} to the TRAP of smoothing pseudo-differential operators already mentioned. 

\vspace{0.4cm}

In the final Section \ref{section:QFT} I \ty{attempt} to draw a plan on how to use the universal property of graphs in the category of TRAPs to define Feynman rules. I start by expanding the above discussion on Feynman rules to try to interpret some \ty{formulae} that can be found in excellent \ty{Physics} textbook \cite{Ri91}. This leads me to introduce various relevant analytical spaces, and in particular the space of distribution-valued meromorphic germs in Definition \ref{def:mero_germs_distrib}. I can then precisely formulate our goal of defining Feynman rules in Conjecture \ref{conjecture1}. I then conjecture that a relevant space of distributon-valued meromorphic germs carries a TRAP structure (Conjecture \ref{conj:TRAP_reg}). Then I precisely explain in Definition \ref{def:TRAP_Feynman_rules} how this Conjecture, if true, allows to rigorously define Feynamn rules and how it would solve our initial problem (Proposition \ref{prop:rapport_conj}). Along the way I discuss the various challenges one would have to overcome to fulfill this program and sometimes what methods could be successful to do so.

\section{The category of PROPs} \label{sec:PROP_def}

\subsection{Elements of category theory}

Category theory is not the main point of this thesis and as such I will not write down a full introduction to category theory. Not only it would be needlessly long and cumbersome to read, but also not essential since I will mostly use a pedestrian definition for PROPs. We refer the author to the many great introductions to category theory available in the \ty{literature}; for example \cite{ML71} from which the following definitions are borrowed.
\begin{defn}
 A \textbf{strict monoidal category} is a triple $(\calC,\otimes,e)$ with
 \begin{enumerate}
  \item $\calC$ a category.
  \item $\otimes$ a bifunctor $\otimes:\calC\times\calC\longrightarrow\calC$ which is associative:
  \begin{equation*} %\label{eq:asso_functor}
   \otimes\circ(\otimes\times \Id_\calC)=\otimes\circ(\Id_\calC\times\otimes)
  \end{equation*}
  (where, as in \cite{ML71}, we identify $\calC\times(\calC\times \calC)$ and $(\calC\times \calC)\times\calC$).
  \item $e\in Obj(\calC)$ is a unit for $\otimes$:
  \begin{equation*}
   \otimes(e\times A) = A = \otimes(A\times e)
  \end{equation*}
  for any $A\in Obj(\calC)$.
  \end{enumerate}
\end{defn}
\begin{rk}
%  This definition is actually the definition of \emph{strict} monoidal categories. Since I will always consider strict monoidal categories, I did not fill the need to write strict. This is in no way an original choice but is justified by the fact that general (i.e. non strict) monoidal categories require the definition of natural morphisms.
I define only \emph{strict} monoidal category since I will not need the more general notion of monoidal category, which would require the definition of natural morphisms. 

To be more precise, one should specify that we will actually work with non-strict monoidal categories. Very roughly speaking, for these objects the equal signs of the definition above should be replaced by natural isomorphisms as the associator and the unitors satisfying the well-known pentagon and triangle diagrams. However, MacLane's coherence theorem \cite{mac1963natural} for monoidal categories essentially guarantees that one can always “strictify” monoidal
categories which is a common practice among specialists. \cy{We will follow this practice here and refer the reader to \cite[XI.5]{kassel2012quantum} for a classical introduction to this topic, or \cite{becerra2023strictification} for a recent review.}
% think of the categories as being strict.
% 
% {\color{red} Peut-être à enlever pour la def de ``symmetric''.}
\end{rk}
Notice that the fact that $\otimes$ is a bifunctor implies in particular for any objects $A,B$ of $\calC$ we have
\begin{equation*}
 \Id_A\otimes\Id_B=\Id_{A\otimes B}.
\end{equation*}
Furthermore, for any objects $A$, $B$, $C$, $A'$, $B'$ and $C'$ of $\calC$ and any morphisms $f:A\longrightarrow B$, $g:B\longrightarrow C$, $\ty{f':}A\longrightarrow B'$ and $g':B'\longrightarrow C'$ of $\calC$ we have
\begin{equation} \label{eq:monoidal_prod_function} 
 \cy{(g\otimes g')\circ(f\otimes f')=(g\circ f)\otimes(g'\circ f').}
%  (f\otimes g)\circ(f'\otimes g') = (f\circ f')\otimes (g\circ g').
\end{equation}
\begin{example}
 The archetypal example of a (non-strict) monoidal is the category $\Vect_\K$ of vector spaces over a field $\K$, with linear maps as morphisms and the usual tensor product as monoidal product (hence the notation for general monoidal products). Of course, in practice, one does not think of $\Vect_\K$ as not being strict. For example, we know how with which isomorphisms we should identity $(A\otimes B)\otimes C$ and $A\otimes(B\otimes C)$.
\end{example}
We will need monoidal categories where $A\otimes B\simeq B\otimes A$.
\begin{defn} \label{defn:symm_mon_cat} (\cite{ML71})
 A \textbf{symmetric monoidal category} is a strict monoidal category with isomorphisms $\gamma_{A,B}:A\otimes B\longrightarrow B\otimes A$ such that:
 \begin{itemize}
  \item the $\gamma_{A,B}$ are natural, i.e. for any objects $A,A'$ and $B,B'$ and for any $f:A\longrightarrow A'$ and $g:B\longrightarrow B'$ we have
  \begin{figure}[h!] 
  		\begin{center}
  			\begin{tikzpicture}[->,>=stealth',shorten >=1pt,auto,node distance=3cm,thick]
  			\tikzstyle{arrow}=[->]
  			
  			\node (1) {$A\otimes B$};
  			\node (2) [right of=1] {$B\otimes A$};
  			\node (3) [below of=1] {$A'\otimes B'$};
  			\node (4) [right of=3] {$B'\otimes A'$};

  			\path
  			(1) edge node [above] {$\gamma_{A,B}$} (2)
  			(1) edge node [left] {$f\otimes g$} (3)
  			(3) edge node [below] {$\gamma_{A',B'}$} (4)
  			(2) edge node [right] {$g\otimes f$} (4);
  			\end{tikzpicture}
  		\end{center}
  	\end{figure}
%   	(i.e. the $\gamma$s are natural).
  \item for any objects $A$, $B$ and $C$, 
%   writing collectively all the isomorphisms $(A\otimes B)\otimes C\longrightarrow A\otimes(B\otimes C)$ as $\alpha$ 
  we have
  \begin{figure}[h!] 
  		\begin{center}
  			\begin{tikzpicture}[->,>=stealth',shorten >=1pt,auto,node distance=3cm,thick]
  			\tikzstyle{arrow}=[->]
  			
  			\node (1) {$A\otimes B\otimes C$};
%   			\node (2) [right of=1] {$A\otimes (B\otimes C)$};
  			\node (3) [right of=2] {$B\otimes C\otimes A$};
  			\node (4) [below of=1] {$B\otimes A\otimes C$};
%   			\node (5) [right of=4] {$B\otimes (A\otimes C)$};
%   			\node (6) [right of=5] {$B\otimes (C\otimes A)$};

  			\path
%   			(1) edge node [above] {$\alpha$} (2)
  			(1) edge node [above] {$\gamma_{A,B\otimes C}$} (3)
%   			(3) edge node [right] {$\alpha$} (6)
  			(1) edge node [left] {$\gamma_{A,B}\otimes \Id_C$} (4)
  			(3) edge node [right] {$~~\Id_B\otimes\gamma_{C,A}$} (4);
%   			(4) edge node [below] {$\alpha$} (5)
%   			(5) edge node [below] {$\Id_B\otimes\gamma_{A,C}$} (6);
  			\end{tikzpicture}
  		\end{center}
  	\end{figure}
  \item for any objects $A$ and $B$
  \begin{equation*}
   \cy{\gamma_{B,A}\circ}\gamma_{A,B}=\Id_{A\otimes B}.
  \end{equation*}
 \end{itemize}
\end{defn}
\begin{rk}
  For non-strict monoidal categories, the diagram above has to be enlarged with associators. Writing these various isomorphisms as $\alpha$ we have to require the extended diagram \ref{diag:non_strict_sym} below to commute.
    \begin{figure}[h!] \label{diag:non_strict_sym}
  		\begin{center}
  			\begin{tikzpicture}[->,>=stealth',shorten >=1pt,auto,node distance=3cm,thick]
  			\tikzstyle{arrow}=[->]
  			
  			\node (1) {$(A\otimes B)\otimes C$};
  			\node (2) [right of=1] {$A\otimes (B\otimes C)$};
  			\node (3) [right of=2] {$(B\otimes C)\otimes A$};
  			\node (4) [below of=1] {$(B\otimes A)\otimes C$};
  			\node (5) [right of=4] {$B\otimes (A\otimes C)$};
  			\node (6) [right of=5] {$B\otimes (C\otimes A)$};

  			\path
  			(1) edge node [above] {$\alpha$} (2)
  			(2) edge node [above] {$\gamma_{A,B\otimes C}$} (3)
  			(3) edge node [right] {$\alpha$} (6)
  			(1) edge node [left] {$\gamma_{A,B}\otimes \Id_C$} (4)
  			(4) edge node [below] {$\alpha$} (5)
  			(5) edge node [below] {$\Id_B\otimes\gamma_{A,C}$} (6);
  			\end{tikzpicture}
  		\end{center}
  	\end{figure} 
\end{rk}
Recall that a category is {\bf locally small} if for any objects $A$ and $B$ the morphisms between $A$ and $B$ form a set and not a proper class. 
% Then a locally small category $\calC$ is called enriched over a monoidal category $\calD$ if, roughly speaking, for any pair of objects $A$ $B$ of $\calC$, the sets of morphisms of $\calC$ between $A$ and $B$ is a object in $\calD$. Let us state a more rigorous definition of enriched categories.
\begin{defn} \label{defn:enriched_cat}
 let $\calD$ be a (strict) monoidal category. A locally small category $\calC$ is \textbf{enriched over $\calD$} if 
%  \begin{itemize}
%   \item 
  for any $(A,B)\in Obj(\calC)^2$, $\Hom(A,B)$ (the set of morphisms from $A$ to $B$) is a object of  $\calD$.
%   \item for any $A\in Obj(\calC)$ there is a morphism $I_A:e\longrightarrow \Hom(A,A)$ with $e$ the monoidal unit of $\calD$,
%   \item for any $(A,B,C)\in Obj(\calC)^3$ there is a map 
%   \begin{equation*}
%    \circ_{A,B,C}:\Hom(A,B)\otimes\Hom(B,C)\longrightarrow\Hom(A,C)
%   \end{equation*}
%   which is associative and unital.
%  \end{itemize}
\end{defn}
As for (strict) monoidal categories, the archetypal example of an enriched category is $\Vect_\K$, which is enriched over itself. Indeed the set of linear morphisms between any two vector spaces has the structure of a vector space.

Enriched categories seem to have been considered for the first time in \cite{ML65}, where PROPs were also introduced. Indeed enriched categories enter the categorical definition of PROPs, which we are now able to give. In the meantime, the definitive text on enriched categories seems to be \cite{kelly1982basic}.

\subsection{PROPs: categorical definition}

PROPs (PROducts and Permutations) seem to have appeared first in \cite{ML65}. They have been since a very active field of research and it is not within the scope of this thesis to review the full \ty{literature} on the topic. Let us just quote \cite{Markl} for a (fairly) recent review of the topic.

To relate them to a structure that is perhaps more familiar to the reader, one could present PROPs as a generalisation of operads and co-operads. While operads can be seen as dealing with operations with several inputs and one outputs, and co-operads with operations with several outputs and one inputs, PROPs allow to treat operations with several inputs \emph{and} outputs. On top of \cite{Markl}, a classical reference to operads is \cite{loday2012algebraic}. \cite{bremner2016algebraic} is another, more recent, introduction on operads which emphazises computational aspects.

Let us now dive in and give a compact definition of PROPs, taken from \cite{AT-L19} and also from \cite{Markl}.
\begin{defn} \label{defn:PROP_cat}
 A \textbf{PROP} (PROducts and Permutations) is a strict symmetric monoidal  category enriched over $\Vect_\K$ whose objects are indexed by natural number $\{[n]\}$ and whose monoidal product is given by the addition:
 \begin{equation} \label{eq:monoidal_prod_add}
  [n]\otimes[m]=[n+m].
 \end{equation}
\end{defn}
We can also state a categorical definition of morphisms of PROP, which we quote from \cite{AT-L19}.
\begin{defn} 
 Let $P$ and $Q$ be two PROPs. A \textbf{morphism of PROPs} between $P$ and $Q$ is a strict monoidal functor $F:P\longrightarrow Q$, i.e. a functor such that $F([n]_P)\otimes_Q F([m]_P)=F([n]_P\otimes_P[m]_P)$.
 
 We write $\PROP$ the category of PROPs.
\end{defn}
Let us now discuss this discussion in order to make it clearer to the layman mathematician. In this discussion we will not aim at mathematical rigour. The complete axioms will be given in the next Subsection. A PROP consists of a category whose objects can be identified with the integers, therefore:
\begin{itemize}
 \item A PROP is only \ty{characterised} by its morphisms, and the sets of morphisms\footnote{these are indeed sets since the category, \cy{being enriched, is locally small (Definition \ref{defn:enriched_cat})}} between any two objects $[l]$ and $[k]$ are vector spaces since the category is enriched over $\Vect_\K$. A PROP can thus be written as a family $P=(P(k,l))_{k,l\in  \N}$ of vector spaces \cy{with $P(k,l):=\Hom([k],[l])$}.
 \item Since the elements of the vector spaces $P(k,l)$ are actually morphisms from $[l]$ to $[k]$, we can compose elements     
 if their domain and image coincide. Thus we have associative maps 
 \begin{equation*}
  \circ:P(l,m)\otimes P(k,l)\longrightarrow P(k,m).
 \end{equation*}
 This is item \ref{item:3a} of Definition \ref{def:prop} below.
 \item Since any objects $[n]$ admits a unit $\Id_n$, the maps $\circ$, seen as products, are unital. This is item \ref{item:3b} of Definition \ref{def:prop} below.
 \item The monoidal structure of the category gives another product in a PROP. By \cy{the previous point,} equations \eqref{eq:monoidal_prod_function} and \eqref{eq:monoidal_prod_add} the monoidal product gives a family of products
 \begin{equation*}
  \otimes:P(k,l)\otimes P(k',l')\longrightarrow P(k+k',l+l')
 \end{equation*}
 which are associative (by definition of a strict monoidal category) and commutative (since the category is symmetric). These requirements are items \ref{item:2a} and \ref{item:2c} of Definition \ref{def:prop} below.
 \item The \ty{monoidal} unit implies that the products $\otimes$ are unital. Furthermore, Equation \eqref{eq:monoidal_prod_add} implies that the monoidal unit is the object $[0]$ and that the unit for the products $\otimes$ is a element of $P(0,0)$. This is item \ref{item:2b} of Definition \ref{def:prop} below.
 
 \item Equation \eqref{eq:monoidal_prod_function} gives a compatibility condition between the products $\circ$ and $\otimes$. It is item \ref{item:4} of Definition \ref{def:prop} below.
 \item For $m\geq1$, \ty{identifying} the object $m$ with $[m]=\{1,\cdots,m\}$\cy{, Equation \eqref{eq:monoidal_prod_add} implies $[m]=[1]^{\otimes m}$. Then the symmetry of the monoidal product (Definition \ref{defn:symm_mon_cat}) together with the definition of the $P(k,l)$ implies that the ``flip'' operation $\tau=\gamma_{[1],[1]}$ defined by $\tau(x\otimes y)=y\otimes x$ lies in $P(2,2)$. Then the functoriality of the tensor product gives maps $\Id_{[i]}\otimes\tau\otimes\Id_{[m-i-2]}:[1]^{\otimes[m]}\longrightarrow[1]^{\otimes[m]}$. These maps are transpositions acting on $[m]$, and composing them we that the symmetry group with $m$ elements
 } $\sym_m$ is a subgroup of \ty{$\Hom([m],[m])$}. More precisely, there is a map $f_\sigma\in\ty{\Hom([m],[m])}$ for any $\sigma\in\sym_m$. This induces a left action of $\sym_m$ on \ty{$\Hom([n],[m])$} by setting $\sigma.p:=f_\sigma\circ p$. \ty{With} the same argument we also obtain a right action of the symmetry group. Thus a PROP is a $\sym\times\sym^{op}$-module (item \ref{item:1} of Definition \ref{def:prop} below).
 
 \item Finally, since these actions of the symmetry groups are given by the categorical structures, we have compatibility axioms between these actions and the product $\otimes$ and $\circ$. These are items \ref{item:5} and \ref{item:6} of Definition \ref{def:prop} below.
\end{itemize}

% COMPLETE FROM THE PRESENTATION FROM THE OPERAD READING GROUP!

% UNE QUESTION: QU4EST6CE QU'UN MORPHISME DE PROP DANS CE CONTEXTE?

\subsection{PROPs: pedestrian definition}

From the discussion above we can reformulate Definition \ref{defn:PROP_cat} to obtain a \cy{definition of PROPs less compact but more tractable for examples}.
\begin{defn} \label{def:prop}
A \textbf{PROP} is a family $P=(P(k,l))_{k,l\in  \N}$ of vector spaces such that:
\begin{enumerate}
\item \label{item:1} $P$ is a $\sym\times\sym^{op}$-module, that is to say, for any $(k,l)\in  \N^2$, $P(k,l)$
is a $\sym_l\times \sym_k^{op}$-module. In other words, there exist maps
\begin{align*}
&\left\{\begin{array}{rcl}
\sym_l\times P(k,l)&\longrightarrow&P(k,l)\\
(\sigma,p)&\longmapsto&\sigma\cdot p,
\end{array}\right.&
&\left\{\begin{array}{rcl}
P(k,l)\times \sym_k&\longrightarrow&P(k,l)\\
(p,\tau)&\longmapsto&p\cdot \tau,
\end{array}\right.
\end{align*}
such that for any $(k,l)\in  \N^2$, for any $(\sigma,\sigma',\tau,\tau')\in \sym_l^2\times \sym_k^2$, for any $p\in P(k,l)$,
\begin{align*}
&&\mathrm{Id}_{[l]}\cdot p&=p\cdot \mathrm{Id}_{[k]}=p,\\
\sigma\cdot (\sigma'\cdot p)&=(\sigma\sigma')\cdot p,&
\sigma \cdot (p\cdot \tau)&=(\sigma\cdot p)\cdot \tau,&
(p\cdot\tau)\cdot \tau'&=p\cdot(\tau\tau').
\end{align*}
\item \label{item:2} For any $(k,l,k',l')\in  \N^4$, there exists a product $*$ from $P(k,l)\otimes P(k',l')$ to $P(k+k',l+l')$
such that:
\begin{enumerate}
\item \label{item:2a} For any $(k,l,k',l',k'',l'')\in  \N^6$, for any $(p,p',p'')\in P(k,l)\times  P(k',l') \times P(k'',l'')$,
\[p*(p'*p'')=(p*p')*p''.\]
\item \label{item:2b} There exists $I_0\in P(0,0)$, such that for any $(k,l)\in  \N^2$, for any $p\in P(k,l)$,
\[p*I_0=I_0*p=p.\]
\item \label{item:2c} The product $*$ is commutative in the following sense: for any $(k,k',l,l')\in \N^4$, for any $(p,p')\in P(k,l)\times P(k',l')$,
\begin{equation}\label{eqpstarpprime}c_{l,l'}\cdot (p*p')=(p'*p)\cdot c_{k,k'},\end{equation}
where for any $(m,n)\in  \N^2$, $c_{m,n}\in \sym_{m+n}$ is defined by:
\begin{align} 
\label{defcmn} c_{m,n}(i)&=\begin{cases}
i+n\mbox{ if }i\leq m,\\
i-m\mbox{ if }i>m.
\end{cases}
\end{align}
\end{enumerate}
This product $*$ is called the  \textbf{horizontal concatenation}.
\item For any $(k,l,m)\in  \N^3$, there exists a product $\circ$ from $P(l,m)\otimes P(k,l)$ to $P(k,m)$ such that:
\begin{enumerate}
\item \label{item:3a} For any $(k,l,m,n)\in  \N^4$, for any $(p,q,r)\in P(m,n)\times P(l,m)\times P(k,l)$,
\[p\circ (q\circ r)=(p\circ q)\circ r.\]
\item \label{item:3b} There exists $I_1\in P(1,1)$, such that for any $(k,l)\in  \N^2$, for any $p\in P(k,l)$,
\[p\circ I_k=I_l\circ p=p,\]
where we put $I_n=I_1^{*n}$ for any $n\in  \N$, with the convention $I_1^{*0}=I_0$.
\end{enumerate}
This product $\circ$ is called the  \textbf{vertical concatenation}.
\item \label{item:4} The vertical and horizontal concatenations are compatible: for any $(k,k',l,l',m,m') \in  \N^6$, 
for any $(p,p',q,q')\in P(l,m)\times P(l',m') \times P(k,l) \times P(k',l')$,
\[(p*p')\circ(q*q')=(p\circ q)*(p'{\circ}q').\]
\item \label{item:5} The vertical concatenation and the action of $\sym\times \sym^{op}$ are compatible:
for any $(k,l,m)\in  \N^3$, for any $(p,q)\in P(l,m)\times P(k,l)$, for any $(\sigma,\tau,\nu) \in \sym_m\times\sym_l\times\sym_k$,
\begin{align*}
\sigma\cdot(p\circ q)&=(\sigma\cdot p)\circ q,&
(p\circ q)\cdot \nu&=p\circ (q\cdot \nu),&
(p\cdot \tau)\circ q&=p\circ (\tau\cdot q).
\end{align*}
\item \label{item:6} The horizontal concatenation and the action of $\sym\times \sym^{op}$ are compatible:
% \begin{enumerate}
% \item 
for any $(k,k',l,l')\in  \N^4$, for any $(p,p')\in P(k,l)\times P(k',l')$, 
for any $(\sigma,\sigma',\tau,\tau') \in \sym_l\times\sym_{l'}\times\sym_k\times\sym_{k'}$,
\begin{align*}
(\sigma\cdot p)*(\sigma'\cdot p')&=(\sigma\otimes \sigma')\cdot(p*p'),&
(p\cdot \tau)*(p'\cdot \tau')&=(p*p')\cdot (\tau \otimes \tau'),
\end{align*}
where for any $\alpha \in \sym_m$, $\beta \in \sym_n$, $\alpha \otimes \beta \in \sym_{m+n}$ is defined by:
\[\alpha \otimes \beta(i)=\begin{cases}
\alpha(i)\mbox{ if }i\leq m,\\
\beta(i-m)+m\mbox{ if }i>m.
\end{cases}\]
% \item \label{item:6b}Commutativity of the horizontal concatenation). For any $(k,k',l,l')\in \N^4$, for any $(p,p')\in P(k,l)\times P(k',l')$,
% \begin{equation}\label{eqpstarpprime}c_{l,l'}\cdot (p*p')=(p'*p)\cdot c_{k,k'},\end{equation}
% where for any $(m,n)\in  \N^2$, $c_{m,n}\in \sym_{m+n}$ is defined by:
% \begin{align} 
% \label{defcmn} c_{m,n}(i)&=\begin{cases}
% i+n\mbox{ if }i\leq m,\\
% i-m\mbox{ if }i>m.
% \end{cases}
% \end{align}
% \end{enumerate}
\end{enumerate}\end{defn}
We can also write the pedestrian definition of morphisms of PROPs.
\begin{defn} \label{def:PROP_morph}
      Let $P=(P(k,l))_{k,l\geq0}$ and $Q=(Q(k,l))_{k,l\geq0}$ be two PROPs. A \textbf{morphism of PROPs} is a family 
     $\phi=(\phi_{k,l})_{k,l\geq0}$ of linear  maps $\phi_{k,l}:P(k,l)\mapsto Q(k,l)$ which form a morphism for the horizontal 
     concatenation, the vertical concatenation and the actions of the symmetric groups. More precisely, for any $(k,l,m,n)\in \N^4$:
     \begin{itemize}
      \item $\forall (p,q)\in P(l,m)\times P(k,l),~\phi_{k,m}(p\circ q) = \phi_{l,m}(p)\circ \phi_{k,l}(q)$,
      \item $\forall (p,q)\in P(k,l)\times P(n,m),~\phi_{k+n,l+m}(p* q) = \phi_{k,l}(p)* \phi_{n,m}(q)$,
      \item $\forall (\sigma,p)\in\sym_l\times P(k,l),~\phi_{k,l}(\sigma.p)=\sigma.\phi_{k,l}(p)$,
      \item $\forall (p,\tau)\in P(k,l)\times\sym_k,~\phi_{k,l}(p.\tau)=\phi_{k,l}(p).\tau$.
     \end{itemize}
     We will write $\PROP$ the category of PROPs. 
     
 By abuse of  notation, we  shall write $\phi(p)$ instead of $\phi_{k,l}(p)$ for $p\in P(k,l)$. 
\end{defn}

\section{Examples of PROPs} \label{sec:PROP_ex}

\subsection{The PROP of linear morphisms} \label{subsec:PROP_lin_morph}

We start with the most classical example of PROP (see for example \cite{Vallette1} and \cite{Markl}) namely the PROP of linear morphisms between tensor products of a finite dimensional vector space.
\begin{defiprop} \label{defi:Hom_V} 
Given a finite dimensional $\K$-vector space $V$, the PROP 
$\Hom_V$ is defined in the following way:
\begin{enumerate}
\item For any $k,l\in  \N$, \[\Hom_V(k,l):=\Hom(V^{\otimes k}, V^{\otimes l}).\]
\item For any $\sigma \in \sym_n$, let  $\theta_\sigma$ be the endomorphism of $V^{\otimes n}$ defined by
\[\theta_\sigma(v_1\otimes \ldots \otimes v_n):=v_{\sigma^{-1}(1)}\otimes \ldots \otimes v_{\sigma^{-1}(n)}.\] 
This defines a left action of $\sym_n$ on $V^{\otimes n}$. 
For any $(k,l)\in  \N^2$, for any $f\in \Hom_V (k,l)$, for any  $(\sigma,\tau) \in \sym_l\times \sym_k$, we set:
\begin{align*}
\sigma\cdot f&:=\theta_\sigma \circ f ,&f\cdot \tau&:=f\circ \theta_\tau.
\end{align*} 
\item The horizontal concatenation is the tensor product of maps and $I_0:\K\longrightarrow\K$ is the identity 
map $I_0:=\mathrm{Id}_\K$.
\item The vertical concatenation is the usual composition of maps and $I_1:V\longrightarrow V$ is the identity map $I_1:=\mathrm{Id}_V$.
\end{enumerate}
 \end{defiprop} 

 \begin{rk}
  Following the convention that for a  PROP  $P=(P(k,l))_{k,l\in \N}$, an element in $  P(k,l)$ has  ``$k$ entries and $l$ exits'', for the PROP $\Hom_V$,  an element $f\in\Hom_V(k,l)$ has  ``$k$ entries and $l$ exits''.
 \end{rk}

\begin{proof}
  \begin{enumerate}
   \item The maps $\theta_\sigma$ turns $\Hom_V$ into a $\sym_l\times \sym_k^{op}$-module by associativity of the composition product.
   \item The horizontal concatenation is associative as a result of the associativity of the tensor product $\otimes$, and we trivially have that $\otimes$ maps $\Hom_V(k,l)\otimes\Hom_V(k',l')$ to $\Hom_V(k+k',l+l')$. Furthermore, if $(k,l)\in  \N^2$ and  $f\in\Hom_V(k,l)$, for any $v\in V^{\otimes k}$, we have
   \begin{equation*}
    (I_0\otimes f)(v) =(I_0\otimes f)(1.v):= I_0(1)\otimes f(v) = 1_\K\otimes f(v) = f(v)
   \end{equation*}
   \item The vertical concatenation is associative as the consequence of the associativity of the composition product. We furthermore have 
   $I_n:=I_1^{\otimes n}=\mathrm{Id}_V^{\otimes n} = \mathrm{Id}_{V^{\otimes n}}$ where the last identity follows
   from the definition of the tensor product of maps.
   \item For any $f\in\Hom_V(l,m)$, $f'\in\Hom_V(l',m')$, $g\in\Hom_V(k,l)$, $g'\in\Hom_V(k',l')$, $v\in V^{\otimes k}$ and 
   $v'\in V^{\otimes k'}$ we have
   \begin{align*}
    (f\otimes f')\circ (g\otimes g')(v\otimes v') & = (f\otimes f')(g(v)\otimes g'(v')) \\
    & = (f\circ g)(v)\otimes (f'\circ g')(v') \\
    & = [(f\circ g)\otimes (f'\circ g')](v\otimes v').
   \end{align*}
   Thus, the horizontal and vertical concatenation are compatible.
   \item The vertical concatenation and the action of $\sym\times \sym^{op}$ are compatible by associativity of the composition product.
   \item For any $f\in \Hom_V(k,l)$, $f'\in \Hom_V(k',l')$, $\sigma \in \sym_l$, $\sigma'\in \sym_{l'}$, $v\in V^{\otimes k}$, 
   $v'\in V^{\otimes k'}$ we have 
   \begin{align*}
    (\sigma.f)\otimes(\sigma'.f')(v\otimes v') & = ({\theta_\sigma}\circ f)\otimes ({\theta_{\sigma'}}\circ f')(v\otimes v') \\
					       & = {\theta_\sigma}(f(v))\otimes ({\theta_{\sigma'}}f'(v') \\
					       & = ({\theta_\sigma}\otimes {\theta_{\sigma'}})[f(v)\otimes f'(v')] \\
					       & = (\sigma\otimes \sigma').(f\otimes f')(v\otimes v').
   \end{align*}
   Similarly, we have $(f.\tau)\otimes(f'.\tau') = (f\otimes f').(\tau\otimes\tau')$ and $c_{l,l'}\cdot (f*f')=(f'*f)\cdot c_{k,k'}$, therefore the horizontal action of $\sym\times \sym^{op}$ are compatible. 
  \end{enumerate}
 \end{proof}
 Since in QFT one has typically to deal with infinite dimensional space, we have in \cite{ClFoPa20} generalised the PROP $\Hom_V$ to a PROP of morphisms of tensor products of infinite dimensional vector spaces.
 
 \subsection{A PROP for Fréchet nuclear spaces}
 
 It is well-known that when dealing with tensor products of infinite dimensional vector spaces on has to be much more careful than in the finite dimensional case. In his seminal work \cite{Gr54}, Grothendieck showed that the notion of ``nuclear spaces'' is the right framework to deal with \ty{these} difficulties. We will therefore work with nuclear spaces which, for the sake of simplicity, we will assume to be Fréchet. One could also work in some more general frameworks, for example with barreled nuclear spaces. We found out we do not need this level of generality for our task at hand.
 
 \subsubsection{Topologies on tensor products} 
 
The first challenge when dealing with tensor products of infinite dimensional \ty{(topological) vector spaces} is that one can define many different non-equivalent topologies on the tensor product. We recall the main ones here.

A first possibility is the so-called \textbf{$\epsilon$-topology}; \cite[Definition 43.1]{Treves67}. For two   topological vector spaces $E$ and $F$, one 
can show (\cite[Proposition 42.4]{Treves67}) the isomorphism  of vector spaces
$E\otimes F\simeq \mathcal{B}^c(E'_\sigma\times F'_\sigma,\K)$ where 
$\mathcal{B}^c(E'_\sigma\times F'_\sigma,\K)$ denotes the space of continuous bilinear maps from $E'_\sigma\times F'_\sigma$ to $\K$ and $E'_\sigma$ (resp. 
$F'_\sigma$)  the topological dual of $E$ (resp. $F$) for $\sigma$, the weak topology. 

Recall that a bilinear map $f:E\times F\longrightarrow K$ is called separately continuous if, 
for any pair $(x,y)\in E\times F$, the maps $z\mapsto f(x,z)$ and $z'\mapsto f(z',y)$ are continuous. We then clearly have that continuous bilinear maps 
build a linear subspace of the space $\mathcal{B}^{sc}(E\times F,\K)$ of separately continuous bilinear maps.

The space  $\mathcal{B}^{sc}(E\times F,\K)$ can be equipped with the topology of uniform convergence on products of equicontinuous subsets of $E'_\sigma$ with 
equicontinuous subsets of $F'_\sigma$. Recall that, for a topological space $X$ and a topological vector space $G$, a set $S$ of maps from $X$ to $G$ is 
said to be equicontinuous at $x_0\in X$ if, for any $V\subseteq G$ neighbourhood of zero, there is some neighbourhood  $V(x_0)\subseteq X$   of $x_0$, such 
that 
\begin{equation*}
 \forall f\in S,~x\in V(x_0) \Rightarrow f(x)-f(x_0)\in V.
\end{equation*}
In our case, $G$ is $\K$ and $X$ is $E_\sigma$ (resp. $F_\sigma$). This topology induces a topology on the subspace 
$\mathcal{B}^c(E'_\sigma\times F'_\sigma,\K)$ and thus on $E\otimes F$. We denote by $E\otimes_\epsilon F$ the topological vector space obtained by 
endowing $E\otimes F$ with this topology.

There is another  topology on $E\otimes F$   called the \textbf{projective topology}; 
\cite[Definition 43.2]{Treves67}.  The projective topology is defined as the strongest 
locally convex topology on $E\otimes F$ such that the canonical map $\phi:E\times F\longrightarrow E\otimes F$ is continuous. 
We write $E\otimes_\pi F$ the topological vector space obtained by 
endowing $E\otimes F$ with this topology.

The neighbourhoods of zero of the projective topology can be simply described in terms of neighbourhoods of zero in $E$ and $V$. A convex subset $S$ of 
$E\otimes F$ containing zero is a neighbourhood of zero if it exist a neighbourhood $U$ (resp. V) of zero in $E$ (resp. $F$) such that 
$U\otimes V:=\{u\otimes v|u\in U\wedge v\in V\}\subseteq S$.
 
 \subsubsection{Nuclear Fréchet spaces} \label{subsec:PROP_nuc}
 
 Most of the results stated here can  be found in \cite{Gr52,Gr54}. We also 
refer the readerto the more recent presentation \cite{Treves67} which has notations closer to our own.

We recall that 
\begin{itemize}
 \item A topological vector space is \textbf{Fr\'echet} if it is Hausdorff, has its topology induced by a countable family of semi-norms and is complete with respect to this family of semi-norms. 
 \item A topological vector space is called \textbf{reflexive} if $E''=(E')'=E$, where $E'$ is the topological dual of $E$.
\end{itemize}
In the following $E$ and $F$ are two topological vector spaces and $\Hom^c(E,F)$ is the set of continuous  linear maps from $E$ to $F$.
\begin{rk}
 When $E$ and $F$ are finite dimensional, we have $\Hom^c(E,F)$=$\Hom(E,F)$.
\end{rk}
In order to build the PROP $\Hom^c_V$ in the infinite dimensional case, we need Grothendieck's   completion of the tensor product, a notion we recall 
here in the setup of locally convex topological $\K$-vector spaces.

Let $E$ and $F$ be two vector spaces. Recall that there exists a   vector space $E\otimes F$, and a bilinear map 
$\phi:E\times F\longrightarrow E\otimes F$ such that for any vector space $V$ and bilinear map $f:E\times F\longrightarrow V$, there is a unique 
linear map $\tilde f:E\otimes F\to V$ satisfying  $f=\tilde f\circ \phi$. The space $E\otimes F$ is unique modulo isomorphism and is called the \textbf{tensor product} of $E$ and $F$.

Given two  topological vector spaces, $E$ and $F$ one can a priori equip   $E\otimes F$ with several topologies, among which the 
\textbf{$\epsilon$-topology} and the \textbf{projective topology} presented above. We denote as before by $E\otimes_\epsilon F$ (resp. $E\otimes_\pi F$) the space $E\otimes F$ endowed with the $\epsilon$-topology (resp. the projective topology)  and by $E\widehat\otimes _\epsilon F$ (resp. $E\widehat\otimes _\epsilon F$) of $E\otimes_\epsilon F$ (resp. $E\otimes_\epsilon F$) their completion with respect to the 
$\epsilon$-topology (resp. projective topology). These two spaces differ in general but coincide for nuclear spaces.
\begin{defn} \cite{Gr54}
 A locally convex topological vector space $E$ is \textbf{nuclear} if, and only if, for any locally convex topological vector space $F$,
 \begin{equation*}
 E\widehat\otimes _\epsilon F = E\widehat\otimes _\pi F =: E\widehat\otimes  F
\end{equation*}
holds, in which case $E\widehat\otimes  F $ is called the \textbf{completed tensor product} of $E$ and $F$.
\end{defn}
There are other equivalent definitions of nuclearity, see for example \cite{gelfand1964,hida2008}.

Given a locally convex topological vector space $E$, its topological dual $E'$ can be endowed with various topologies. An important one for our applications will be 
the \textbf{strong topology}, generated by the family of semi-norms of $E'$ defined, on any $f\in E'$: 
\begin{equation} \label{eq:strong_dual_topo}
 \text{for any bounded set }B\text{ of }E,~||f||_B:=\sup_{x\in B}|f(x)|.
\end{equation}
% for any bounded set $B$ of $E$. 
The topological dual $E'$ endowed with this topology is called the \textbf{strong dual}.

For Fr\'echet spaces, nuclearity is preserved under strong duality.
\begin{prop} \begin{itemize}
              \item 
\cite[Proposition 50.6]{Treves67} \label{prop:Frechet_nuclear}
 A Fr\'echet space is nuclear if and only if its strong dual is nuclear.
 \item 
 \cite[Proposition 36.5]{Treves67} A Fr\'echet nuclear space is reflexive.
             \end{itemize}
\end{prop}
Many spaces relevant to renormalisation issues are Fr\'echet and nuclear. We list here some examples.
\begin{example}\label{ex:findimtensor1}  
 Any finite dimensional vector space  can be equipped with a norm and for any of these norms, they are trivially  Banach, hence Fr\'echet and nuclear. 
 If $E$ and $F$ are finite dimensional vector spaces  we have
 $\Hom^c (E, F)=  \Hom (E, F)\simeq  E^* \otimes F$, where $  \Hom (E, F)$ stands for the space of $F$-valued linear maps on 
 $E$ and where the dual $E^*$ is the { \bf algebraic dual}.
\end{example} 
\begin{example}\label{ex:infindimtensor1} 
Let $U$  be an open subset of $\R^n$.
Take $E= C^\infty(U)=:{\mathcal E}(U)$. The topological dual 
 is the space ${E'}={\mathcal E}^\prime(U)$ of distributions 
 on $U$ with compact support. 
 
 Then $E$ is Fr\'echet (\cite{Treves67}, pp. 86-89), and $E'$ is nuclear (\cite{Treves67}, Corollary p. 530). By Proposition 
 \ref{prop:Frechet_nuclear}, $E$ is also nuclear. 
\end{example}
\begin{rk}\label{rk:dualnotfrechet}  Note that the dual $E' $ of a Fr\'echet space  $E$ is  never a Fr\'echet space (for any of the natural topologies on $E'$), unless $E$  is actually a Banach space (see for example 
\cite{kothe1969}).   In particular, ${\mathcal E}^\prime(U)$ is generally not Fr\'echet.
\end{rk} 
We now sum up various results of \cite{Treves67} of importance for  later purposes.
\begin{theo}\cite[Equations (50.17)--(50.19)]{Treves67}
 Let $E$ and $F$ be two Fr\'echet spaces, with $E$ nuclear. The following isomorphisms of topological vector spaces hold.
 \begin{align}
  & E'\widehat\otimes  F \simeq \Hom^c(E,F) \label{eq:E_prime_otimes_F} \\
  & E\widehat\otimes  F \simeq \Hom^c(E',F) \label{eq:E_otimes_F} \\
  & E'\widehat\otimes  F' \simeq (E\widehat\otimes  F)' \simeq {\mathcal B}^c(E\times F, \K). \label{eq:E_prime_otimes_F_prime}
 \end{align}
 with ${\mathcal B}^c(E\times F, \K)$ the set of continuous bilinear maps 
 $K:E\times F\longrightarrow\K$. Here the duals are endowed with the strong dual topology, 
 $\Hom^c(E,F)$ with the strong topology and ${\mathcal B}^c(E\times F, \K)$ with the topology of uniform convergence on products of bounded sets.
\end{theo}
We  also need the stability of Fr\'echet nuclear spaces under completed tensor products.
\begin{prop}
 Let $V$ be a Fr\'echet nuclear space. Then 
 \begin{equation} \label{eq:echange_dual_prod}
 \left(V^{ \widehat\otimes  k}\right)' \simeq\left(V'\right)^{ \widehat\otimes  k}
 \end{equation} 
 holds for any $k\geq1$, where the duals are endowed with their strong topologies.
\end{prop}
\begin{proof}
 Let $V$ be a Fr\'echet nuclear space. The case $k=1$ is trivial. Then Equation \eqref{eq:echange_dual_prod} with $k=2$ holds by Equation \eqref{eq:E_prime_otimes_F_prime} 
 with $E=F=V$. The cases $k\geq2$ are proved by induction, using $E=V^{\widehat\otimes  k-1}$ and $F=V$. Now, if $E$ and $F$ are two nuclear spaces then $E\widehat\otimes  F$ is a nuclear space (\cite[Equation (50.9)]{Treves67}). It is moreover complete since 
 $E\widehat\otimes  F$ is obtained by completion. Thus the completed tensor product  $E\widehat\otimes  F$ of two Fr\'echet nuclear spaces is a Fr\'echet nuclear space and the induction holds.
\end{proof}
 
 \subsubsection{A PROP for Fr\'echet nuclear spaces} \label{subsection:infinite_dim_prop}

We start by recalling the definition of distributions over a    finite dimensional smooth manifold ${X}$. We quote \cite[Definition 6.3.3]{Ho89}.
 \begin{defn} 
 To every coordinate system 
  $\kappa:U_k\subset {X}\longrightarrow V_k\subset\R^n$ we associate 
  a distribution $u_k\in\mathcal{D}'(V _k)$ such that 
  \begin{equation*}
   u_{k'}=(\kappa\circ\kappa'^{-1})^*u_k
  \end{equation*}
  in $\kappa'(U_k\cap U_{k'})$; with $(\kappa\circ\kappa'^{-1})^*u_k$ the pullback of 
  $u_k$ by $\kappa\circ\kappa'^{-1}$ whose existence and uniqueness
  is given by 
  \cite[Theorem 6.1.2]{Ho89}. Then the system $u_k$ of distributions is called a distribution on ${X}$. The set of distributions on ${X}$ is written 
  $\mathcal{D}'({X})$. Similarly we define $\mathcal{E}'({X})$, the set of distributions with 
  compact support.
 \end{defn}
 We can now state the properties that will allow us to obtain a PROP for infinite dimensional vector spaces. We state it without a rigorous proof but give precise reference to the existing proofs of the statements.
 \begin{prop} \label{prop:fction_manifold_Frechet_nuclear}
 $\mathcal{E}({X})$ is a Fr\'echet nuclear space and $\mathcal{E}'({X})$ is a nuclear, but non Fréchet space.
\end{prop}
The fact that $\mathcal{E}({X})$ is Fr\'echet a classical result of functional analysis that the space of functions over a smooth manifold is Fr\'echet (see for example \cite[Exercise 2.3.2]{BaCr13}). The fact that it is nuclear is a folklore result, often stated without proof nor references and the only proof known to the author is in \cite[p. 4]{BrDaLGRe17}. It then follows from Proposition \ref{prop:Frechet_nuclear}, that the space $\mathcal{E}'(X)$ is also nuclear. From Remark \ref{rk:dualnotfrechet} the 
space $\mathcal{E}'(X)$ is \emph{not} Fr\'echet  since the dual of a Fr\'echet space $F$ is Fr\'echet if and only if $F$ is Banach (see for example \cite{kothe1969}) which is not the case of $\mathcal{E}({X})$.

One further useful result is
\begin{prop} \label{prop:prod_function}
 Let ${X}$ and ${Y}$ be two finite dimensional smooth manifolds. Then 
 \begin{equation*}
 \Hom^c(\mathcal{E}'({X}),\mathcal{E}({Y}))\simeq \,  \mathcal{E}({X})\,\widehat\otimes\, \mathcal{E}({Y}) \simeq \mathcal{E}({X}\times {Y})
 \end{equation*}
 holds.
\end{prop}
The second isomorphism   \cite[Chap. 5, p. 105]{Gr52} can be proved using a version of the Schwartz kernel theorem for smoothing operators 
\cite[Theorem 2.4.5]{BaCr13} by means of the identification $\Hom^c(\mathcal{E}'({X}),\mathcal{E}({Y}))\simeq \mathcal{E}({X}\times {Y})$. The result then follows from 
\eqref{eq:E_otimes_F} applied to $\mathcal{E}(X)$ and $\mathcal{E}(Y)$ which are Fr\'echet nuclear spaces.

We can now introduce the spaces that will carry a PROP structure generalising the PROP $\Hom_V$ of Subsection \ref{subsec:PROP_lin_morph} to the infinite dimensional case.
\begin{defn} \label{defi:Hom_V_generalised}
 Let $V$ be a Fr\'echet nuclear space. For any $k,l\in  \N$, we set
\[\Hom_V^c(k,l)=\Hom^c(V^{\hat \otimes k}, V^{\hat \otimes l})\simeq(V')^{\widehat\otimes  k}\widehat\otimes  V^{\widehat\otimes  l},\]
where, as before $V'$ stands for the strong topological dual. Furthermore we set $\Hom^c_V:=(\Hom_V^c(k,l))_{k,l\geq0}$.

For any $\sigma \in \sym_n$, let  $\theta_\sigma$ be the endomorphism of $V^{\otimes n}$ defined by
\[\theta_\sigma(v_1\otimes \ldots \otimes v_n)=v_{\sigma^{-1}(1)}\otimes \ldots \otimes v_{\sigma^{-1}(n)}.\] It extends to a continuous linear map 
$\overline{\theta_\sigma}$ on the closure 
$V^{\widehat\otimes  n}$.   
For any $f\in \Hom_V^c (k,l)$, $\sigma \in \sym_l$, $\tau\in \sym_k$, we set:
\begin{align*}
\sigma\cdot f&=\overline{\theta_\sigma} \circ f ,&f\cdot \tau&=f\circ \overline{\theta_\tau}.
\end{align*} 
\end{defn}
In the above definition, the superscript ``c'' stands for continuous. As advertised, the family $\Hom_V^c$ carries a PROP structure.
\begin{theo} \label{thm:Hom_V_generalised}
 Let $V$ be a Fr\'echet nuclear space. $\Hom_V^c$, with the action of $\sym\times\sym^\mathrm{op}$ described above, is a PROP. Its horizontal 
 concatenation is the usual (topological) tensor product of maps with $I_0:\K\longrightarrow\K$ is the constant map $I_0(x):=1_\K$  and 
 its vertical concatenation is the usual composition of maps and $I_1:V\longrightarrow V$ is the identity map.
\end{theo}
\begin{proof}
  The proof is exactly the same as the proof of Definition-Proposition \ref{defi:Hom_V}.
 \end{proof}
  \begin{example}
      For a  finite dimensional vector space $V$ the classical PROP $\Hom_V$ of Proposition-Definition 
      \ref{defi:Hom_V} coincides with the the PROP $\Hom_V^c$.
     \end{example}
     \begin{example}
      Let $U$ be an open of $\R^n$. From Example \ref{ex:infindimtensor1} and Equation 
      \eqref{eq:echange_dual_prod}
the family  $({\mathcal K}_U(k, l))_{k,l\geq 0}$, with 
${\mathcal K}_U(k, l)= \left({\mathcal E}^\prime(U)\right)^{\widehat\otimes  k}\, \widehat\otimes  \,  
\left({\mathcal E}(U)\right)^{\widehat\otimes  l}$  
defines a PROP.
     \end{example}
     
     \begin{example}
     Let $X$ be a smooth finite dimensional manifold.
      From Proposition \ref{prop:fction_manifold_Frechet_nuclear} and Equation \eqref{eq:echange_dual_prod}
the family  $({\mathcal K}_X(k, l))_{k,l\geq 0}$,
with ${\mathcal K}_X(k, l)= \left({\mathcal E}^\prime(X)\right)^{\widehat\otimes  k}\, \widehat\otimes  \,  {\mathcal E}(X)^{\widehat\otimes  l}$  
defines a PROP.
     \end{example}

\subsection{The PROP of graphs} \label{subsec:PROP_graphs}

A folklore result is that graphs have a PROP structure and are the free PROP. To rigorously prove the latter statement we need to precisely define the PROP structure that exists on a family of (generalised) graphs. 

\subsubsection{Generalised graphs}

\begin{defn} \label{def:graph}
A \textbf{graph} is a family $G=(V(G),E(G),I(G),O(G),IO(G),s,t,\alpha,\beta)$, where:
\begin{enumerate}
\item $V(G)$ (set of vertices), $E(G)$ (set of internal edges), $I(G)$ (set of input edges),
$O(G)$ (set of output edges) and $IO(G)$ (set of input-output edges) are finite (maybe empty) sets.
\item $s:E(G)\sqcup O(G)\longrightarrow V(G)$ is a map (source map).
\item $t:E(G)\sqcup I(G)\longrightarrow V(G)$ is a map (target map).
\item $\alpha:I(G)\sqcup IO(G)\longrightarrow [i(G)]$ is a bijection, with $i(G)=|I(G)|+|IO(G)|$
(indexation of the input edges).
\item $\beta:O(G)\sqcup IO(G)\longrightarrow [o(G)]$ is a bijection, with $o(G)=|O(G)|+|IO(G)|$
(indexation of the output edges).
\end{enumerate}
\end{defn}
Note that this definition differs from  \cite[Definition 1.3.1]{ClFoPa20} since here the loops of graphs will play no role, neither for PROPs nor for TRAPs.
\begin{example}\label{ex4}
Here is a graph $G$ : 
\begin{align*}
V(G)&=\{x,y\},&E(G)&=\{a,b\},&I(G)&=\{c,d\},&O(G)&=\{e,f\},&IO(G)&=\{g\},
\end{align*}
and:
\begin{align*}
s&:\left\{\begin{array}{rcl}
a&\mapsto&y\\
b&\mapsto&x\\
e&\mapsto&y\\
f&\mapsto&y,
\end{array}\right.&
t&:\left\{\begin{array}{rcl}
a&\mapsto&x\\
b&\mapsto&y\\
c&\mapsto&x\\
d&\mapsto&x,
\end{array}\right.&
\alpha&:\left\{\begin{array}{rcl}
c&\mapsto&1\\
d&\mapsto&2\\
g&\mapsto&3,
\end{array}\right.&
\beta&:\left\{\begin{array}{rcl}
e&\mapsto&3\\
f&\mapsto&1\\
g&\mapsto&2.
\end{array}\right.&
\end{align*}
This is graphically represented as follows:
\[\xymatrix{1&&3&2\\
&\rond{y}\ar[ru]_e \ar[lu]^f \ar@/_1pc/[d]_a&&\\
&\rond{x}\ar@/_1pc/[u]_b&&\\
1\ar[ru]^c&&2\ar[lu]_d&3\ar[uuu]_g}\]
\end{example}

\begin{defn} \label{def:morph_graphs}
Let $G$ and $G'$ be two graphs. An \textbf{(resp. iso-)morphism} of graphs from $G$ to $G'$ is a family of (resp. bijections) maps $f=(f_V,f_E,f_I,f_O,f_{IO})$ with:
\begin{align*}
f_V:V(G)\longrightarrow & V(G'),\qquad \qquad f_E:E(G)\longrightarrow E(G'),\qquad \qquad f_I:I(G)\longrightarrow I(G'),\\
f_O\,:~ & O(G)\longrightarrow O(G'),\qquad \qquad f_{IO}:IO(G)\longrightarrow IO(G'),
\end{align*}
such that:
\begin{align*}
s'\circ f_E&=f_V\circ s_{\mid E(G)},&s'\circ f_O&=f_V\circ s_{\mid O(G)},\\
t'\circ f_E&=f_V\circ t_{\mid E(G)},&t'\circ f_I&=f_V\circ t_{\mid I(G)},\\
\alpha'\circ f_I&=\alpha_{\mid I(G)},&\alpha'\circ f_{IO}&=\alpha_{\mid IO(G)},\\
\beta'\circ f_O&=\beta_{\mid O(G)},&\beta'\circ f_{IO}&=\beta_{\mid IO(G)}.
\end{align*}
For any $k,l\in  \N$, we denote by $\Gr(k,l)$ the space generated by the isoclasses of graphs $G$ such that
$i(G)=k$ and $o(G)=l$, i.e. $\Gr(k,l)$ is the quotient space of graphs with $k$ input edges and $l$ output edges by the equivalence relation given by isomorphism.
\end{defn}
In what follows, we shall write \emph{graphs} for \emph{isoclasses of graphs}.
\begin{example}
The isomorphism class of the graph of Example \ref{ex4} is represented by:
\[\xymatrix{1&&3&2\\
&\rond{}\ar[ru] \ar[lu] \ar@/_1pc/[d]&&\\
&\rond{}\ar@/_1pc/[u]&&\\
1\ar[ru]&&2\ar[lu]&3\ar[uuu]}\]
\end{example}

\subsubsection{The PROP structure} \label{subsec:prop_graph}

We now want  to equip the set $\Gr$ of isoclasses of graphs  with a PROP structure.
\begin{itemize}
	\item  Let us first define an action of $\sym\times \sym^{op}$ on graphs.
Let $G=(V(G),E(G),I(G),O(G),IO(G),s,t,\alpha,\beta)\in \Gr(k,l)$, $\sigma \in \sym_k$ and $\tau\in \sym_l$.
Then:
\begin{equation} \label{eq:sym_sym_graph}
 \tau\cdot G\cdot \sigma=\cy{
(V(G),E(G),I(G),O(G),IO(G),s,t,\sigma^{-1}\circ \alpha,\tau \circ \beta)}.
\end{equation}
    \item We now define the \textbf{horizontal concatenation}. If $G$ and $G^\prime$ are two disjoint graphs,
	we define a graph $G*G'$ in the following way:
	\begin{align*}
	\qquad &V(G*G')=V(G)\sqcup V(G'),\qquad E(G*G')=E(G)\sqcup E(G'),\\
	I(G*G')=I(G)\sqcup &I(G'),\qquad O(G*G')=O(G)\sqcup O(G'),\qquad IO(G*G')=IO(G)\sqcup IO(G').
	\end{align*}
	The source and target maps are given by:
	\begin{align*}
	s''_{\mid E(G)\sqcup O(G)}&=s,&s''_{\mid E(G')\sqcup O(G')}&=s',\\
	t''_{\mid E(G)\sqcup I(G)}&=t,&t''_{\mid E(G')\sqcup I(G')}&=t'.
	\end{align*}
	The indexations of the input and output edges are given by:
	\begin{align*}
	\alpha''_{\mid I(G)\sqcup IO(G)}&=\alpha,&\alpha''_{\mid I(G')\sqcup IO(G')}&=i(G)+\alpha',\\
	\beta''_{\mid O(G)\sqcup IO(G)}&=\beta,&\beta''_{\mid O(G')\sqcup IO(G')}&=o(G)+\beta'
	\end{align*}
	with an obvious abuse of notation in the definition of the second column.
	Notice that this product is not commutative in the usual sense for $G*G'$ and $G'*G$ might differ by the indexation of their input and output 
	edges. However, it is commutative in the sense of Axiom \ref{item:2c} of PROPs.
	Roughly speaking, $G*G'$ is the disjoint union of $G$ and $G'$, the input and output edges of $G'$ being indexed
	after the input and output edges of $G$. 
	\begin{center}
		\begin{tikzpicture}[line cap=round,line join=round,>=triangle 45,x=0.5cm,y=0.5cm]
		\clip(-2.5,-4.) rectangle (1.,4.);
		\draw [line width=0.4pt] (-2.,1.)-- (0.5,1.);
		\draw [line width=0.4pt] (0.5,1.)-- (0.5,-1.);
		\draw [line width=0.4pt] (0.5,-1.)-- (-2.,-1.);
		\draw [line width=0.4pt] (-2.,-1.)-- (-2.,1.);
		\draw [->,line width=0.4pt] (-1.5,1.) -- (-1.5,3.);
		\draw [->,line width=0.4pt] (0.,1.) -- (0.,3.);
		\draw [->,line width=0.4pt] (-1.5,-3.) -- (-1.5,-1.);
		\draw [->,line width=0.4pt] (0.,-3.) -- (0.,-1.);
		\draw (-1.25,0.5) node[anchor=north west] {$G$};
		\draw (-1.8,-3) node[anchor=north west] {$1$};
		\draw (-0.3,-3) node[anchor=north west] {$k$};
		\draw (-1.4,-2.2) node[anchor=north west] {$\ldots$};
		\draw (-1.8,4.2) node[anchor=north west] {$1$};
		\draw (-0.3,4.2) node[anchor=north west] {$l$};
		\draw (-1.4,2.) node[anchor=north west] {$\ldots$};
		\end{tikzpicture}
		$\substack{\displaystyle *\\ \vspace{3cm}}$
		\begin{tikzpicture}[line cap=round,line join=round,>=triangle 45,x=0.5cm,y=0.5cm]
		\clip(-2.5,-4.) rectangle (0.7,4.);
		\draw [line width=0.4pt] (-2.,1.)-- (0.5,1.);
		\draw [line width=0.4pt] (0.5,1.)-- (0.5,-1.);
		\draw [line width=0.4pt] (0.5,-1.)-- (-2.,-1.);
		\draw [line width=0.4pt] (-2.,-1.)-- (-2.,1.);
		\draw [->,line width=0.4pt] (-1.5,1.) -- (-1.5,3.);
		\draw [->,line width=0.4pt] (0.,1.) -- (0.,3.);
		\draw [->,line width=0.4pt] (-1.5,-3.) -- (-1.5,-1.);
		\draw [->,line width=0.4pt] (0.,-3.) -- (0.,-1.);
		\draw (-1.25,0.5) node[anchor=north west] {$G'$};
		\draw (-1.8,-3) node[anchor=north west] {$1$};
		\draw (-0.3,-3) node[anchor=north west] {$k'$};
		\draw (-1.4,-2.2) node[anchor=north west] {$\ldots$};
		\draw (-1.8,4.2) node[anchor=north west] {$1$};
		\draw (-0.3,4.2) node[anchor=north west] {$l'$};
		\draw (-1.4,2.) node[anchor=north west] {$\ldots$};
		\end{tikzpicture}
		$\substack{\displaystyle =\\ \vspace{3cm}}$
		\begin{tikzpicture}[line cap=round,line join=round,>=triangle 45,x=0.5cm,y=0.5cm]
		\clip(-2.5,-4.) rectangle (0.5,4.);
		\draw [line width=0.4pt] (-2.,1.)-- (0.5,1.);
		\draw [line width=0.4pt] (0.5,1.)-- (0.5,-1.);
		\draw [line width=0.4pt] (0.5,-1.)-- (-2.,-1.);
		\draw [line width=0.4pt] (-2.,-1.)-- (-2.,1.);
		\draw [->,line width=0.4pt] (-1.5,1.) -- (-1.5,3.);
		\draw [->,line width=0.4pt] (0.,1.) -- (0.,3.);
		\draw [->,line width=0.4pt] (-1.5,-3.) -- (-1.5,-1.);
		\draw [->,line width=0.4pt] (0.,-3.) -- (0.,-1.);
		\draw (-1.25,0.5) node[anchor=north west] {$G$};
		\draw (-1.8,-3) node[anchor=north west] {$1$};
		\draw (-0.3,-3) node[anchor=north west] {$k$};
		\draw (-1.4,-2.2) node[anchor=north west] {$\ldots$};
		\draw (-1.8,4.2) node[anchor=north west] {$1$};
		\draw (-0.3,4.2) node[anchor=north west] {$l$};
		\draw (-1.4,2.) node[anchor=north west] {$\ldots$};
		\end{tikzpicture}
		\begin{tikzpicture}[line cap=round,line join=round,>=triangle 45,x=0.5cm,y=0.5cm]
		\clip(-2.5,-4.) rectangle (2.,4.);
		\draw [line width=0.4pt] (-2.,1.)-- (0.5,1.);
		\draw [line width=0.4pt] (0.5,1.)-- (0.5,-1.);
		\draw [line width=0.4pt] (0.5,-1.)-- (-2.,-1.);
		\draw [line width=0.4pt] (-2.,-1.)-- (-2.,1.);
		\draw [->,line width=0.4pt] (-1.5,1.) -- (-1.5,3.);
		\draw [->,line width=0.4pt] (0.,1.) -- (0.,3.);
		\draw [->,line width=0.4pt] (-1.5,-3.) -- (-1.5,-1.);
		\draw [->,line width=0.4pt] (0.,-3.) -- (0.,-1.);
		\draw (-1.25,0.5) node[anchor=north west] {$G'$};
		\draw (-2.3,-3) node[anchor=north west] {$k+1$};
		\draw (-0.3,-3) node[anchor=north west] {$k+k'$};
		\draw (-1.4,-2.2) node[anchor=north west] {$\ldots$};
		\draw (-2.3,4.2) node[anchor=north west] {$l+1$};
		\draw (-0.3,4.2) node[anchor=north west] {$l+l'$};
		\draw (-1.4,2.) node[anchor=north west] {$\ldots$};
		\end{tikzpicture}
		
		\vspace{-1.5cm}
	\end{center}
	
	\begin{example} Here is an example of horizontal concatenation :\\
	
	\vspace{-2cm}
	
			\[\xymatrix{1&3&2\\
			&\rond{}\ar[lu] \ar[u]\ar[d]&\\
			&\rond{}\ar[ruu]&\\
			1\ar[ru]&&2\ar[lu]} 
			\substack{\vspace{3cm}\\\displaystyle \mbox{$*$}}
		\xymatrix{1&&2\\
			&\rond{}\ar[ru] \ar[lu]&\\
			&&\\
			&1\ar[uu]&}
					\substack{\vspace{3cm}\\\displaystyle =}
					\xymatrix{1&3&2&4&&5\\
			&\rond{}\ar[lu] \ar[u]\ar[d]&&&\rond{}\ar[ru] \ar[lu]&\\
			&\rond{}\ar[ruu]&&&&\\
			1\ar[ru]&&2\ar[lu]&&3\ar[uu]&}\]
	\end{example}
	
	This product of graphs induces a product $*:\Gr(k,l)\otimes \Gr(k',l')\longrightarrow
	\Gr(k+k',l+l')$. If $G$, $G'$ and $G''$ are three graphs, clearly
	\[G*(G'*G'')=(G*G')*G''.\]
	Hence, the product $*$ is associative. Its unit $I_0$ is the unique graph such that
	$V(I_0)=E(I_0)=I(I_0)=O(I_0)=IO(I_0)=\emptyset$. 
	
	\item We now define the \textbf{ vertical concatenation}. Let $G$ and $G'$ be {disjoint} graphs such that $o(G)=i(G')$. We define a graph 
	$G''=G'\circ G$
	in the following way:
	\begin{align*}
	V(G'')&=V(G)\sqcup V(G'),\\
	E(G'')&=E(G)\sqcup E(G')\sqcup \{(f,e')\in O(G)\times I(G'):\beta(f)=\alpha'(e')\},\\
	I(G'')&=I(G)\sqcup \{(f,e')\in IO(G)\times I(G'):\beta(f)=\alpha'(e')\},\\
	O(G'')&=O(G)\sqcup \{(f,e')\in O(G)\times IO(G') : \beta(f)=\alpha'(e')\},\\
	IO(G'')&=\{(f,e')\in IO(G)\times IO(G'): \beta(f)=\alpha'(e')\}.
	\end{align*}	
	Its \textbf{source} and \textbf{target} maps  are given by:
	\begin{align*}
	s''_{\mid E(G)}&=s_{\mid E(G)},&s''_{\mid E(G')}&=s'_{\mid E(G')},&
	s''_{\mid O(G')}&=s'_{\mid O(G')},&s''((f,e'))&=s(f),\\
	t''_{\mid E(G)}&=t_{\mid E(G)},&t''_{\mid E(G')}&=s'_{\mid E(G')},&
	t''_{\mid I(G)}&=s_{\mid I(G)},&s''((f,e'))&=t'(e').
	\end{align*}
	The indexations of its input and output edges are given by:
	\begin{align*}
	\alpha''_{\mid I(G)}&=\alpha_{\mid I(G)},& \alpha''((f,e))&=\alpha(f),\\
	\beta''_{\mid O(G')}&=\beta'_{\mid O(G')},&\beta''((f,e))&=\beta'(e).
	\end{align*}
	Roughly speaking, $G'\circ G$ is obtained by gluing together the outgoing edges of $G$ and the incoming
	edges of $G'$ according to their indexation. 
		\begin{center}
		\begin{tikzpicture}[line cap=round,line join=round,>=triangle 45,x=0.5cm,y=0.5cm]
		\clip(-2.5,-4.) rectangle (1.,4.);
		\draw [line width=0.4pt] (-2.,1.)-- (0.5,1.);
		\draw [line width=0.4pt] (0.5,1.)-- (0.5,-1.);
		\draw [line width=0.4pt] (0.5,-1.)-- (-2.,-1.);
		\draw [line width=0.4pt] (-2.,-1.)-- (-2.,1.);
		\draw [->,line width=0.4pt] (-1.5,1.) -- (-1.5,3.);
		\draw [->,line width=0.4pt] (0.,1.) -- (0.,3.);
		\draw [->,line width=0.4pt] (-1.5,-3.) -- (-1.5,-1.);
		\draw [->,line width=0.4pt] (0.,-3.) -- (0.,-1.);
		\draw (-1.25,0.5) node[anchor=north west] {$G'$};
		\draw (-1.8,-3) node[anchor=north west] {$1$};
		\draw (-0.3,-3) node[anchor=north west] {$l$};
		\draw (-1.4,-2.2) node[anchor=north west] {$\ldots$};
		\draw (-1.8,4.2) node[anchor=north west] {$1$};
		\draw (-0.3,4.2) node[anchor=north west] {$m$};
		\draw (-1.4,2.) node[anchor=north west] {$\ldots$};
		\end{tikzpicture}
		$\substack{\displaystyle \circ\\ \vspace{3cm}}$		
		\begin{tikzpicture}[line cap=round,line join=round,>=triangle 45,x=0.5cm,y=0.5cm]
		\clip(-2.5,-4.) rectangle (0.7,4.);
		\draw [line width=0.4pt] (-2.,1.)-- (0.5,1.);
		\draw [line width=0.4pt] (0.5,1.)-- (0.5,-1.);
		\draw [line width=0.4pt] (0.5,-1.)-- (-2.,-1.);
		\draw [line width=0.4pt] (-2.,-1.)-- (-2.,1.);
		\draw [->,line width=0.4pt] (-1.5,1.) -- (-1.5,3.);
		\draw [->,line width=0.4pt] (0.,1.) -- (0.,3.);
		\draw [->,line width=0.4pt] (-1.5,-3.) -- (-1.5,-1.);
		\draw [->,line width=0.4pt] (0.,-3.) -- (0.,-1.);
		\draw (-1.25,0.5) node[anchor=north west] {$G$};
		\draw (-1.8,-3) node[anchor=north west] {$1$};
		\draw (-0.3,-3) node[anchor=north west] {$k$};
		\draw (-1.4,-2.2) node[anchor=north west] {$\ldots$};
		\draw (-1.8,4.2) node[anchor=north west] {$1$};
		\draw (-0.4,4.1) node[anchor=north west] {$l$};
		\draw (-1.4,2.) node[anchor=north west] {$\ldots$};
		\end{tikzpicture}
		$\substack{\displaystyle =\\ \vspace{3cm}}$
		\begin{tikzpicture}[line cap=round,line join=round,>=triangle 45,x=0.5cm,y=0.5cm]
		\clip(-2.5,-4.) rectangle (0.5,8.);
		\draw [line width=0.4pt] (-2.,1.)-- (0.5,1.);
		\draw [line width=0.4pt] (0.5,1.)-- (0.5,-1.);
		\draw [line width=0.4pt] (0.5,-1.)-- (-2.,-1.);
		\draw [line width=0.4pt] (-2.,-1.)-- (-2.,1.);
		\draw [->,line width=0.4pt] (-1.5,1.) -- (-1.5,3.);
		\draw [->,line width=0.4pt] (0.,1.) -- (0.,3.);
		\draw [->,line width=0.4pt] (-1.5,-3.) -- (-1.5,-1.);
		\draw [->,line width=0.4pt] (0.,-3.) -- (0.,-1.);
		\draw (-1.25,0.5) node[anchor=north west] {$G$};
		\draw (-1.8,-3) node[anchor=north west] {$1$};
		\draw (-0.3,-3) node[anchor=north west] {$k$};
		\draw (-1.4,-2.2) node[anchor=north west] {$\ldots$};
		\draw [line width=0.4pt] (-2.,5.)-- (0.5,5.);
		\draw [line width=0.4pt] (0.5,5.)-- (0.5,3.);
		\draw [line width=0.4pt] (0.5,3.)-- (-2.,3.);
		\draw [line width=0.4pt] (-2.,3.)-- (-2.,5.);
		\draw [->,line width=0.4pt] (-1.5,5.) -- (-1.5,7.);
		\draw [->,line width=0.4pt] (0.,5.) -- (0.,7.);
		\draw (-1.25,4.5) node[anchor=north west] {$G'$};
		\draw (-1.8,8.2) node[anchor=north west] {$1$};
		\draw (-0.4,8.1) node[anchor=north west] {$m$};
		\draw (-1.4,6.) node[anchor=north west] {$\ldots$};
		\end{tikzpicture}
		
		\vspace{-1.5cm}
	\end{center}
	\begin{example} Here is an example of vertical concatenation :\\
	
	\vspace{-1.8cm}
		\[
		\xymatrix{&2&1 \\ 
			&\rond{}\ar[u]&\rond{}\ar[l] \ar[u] \\
			1\ar[ru]&2\ar[u]&3\ar[u]&}\hspace{5mm}
		\substack{\vspace{2.5cm}\\ \displaystyle \circ}
		\xymatrix{&2&1&3\\
			&\rond{}\ar[u]\ar@/_1pc/[r]&\rond{}\ar[u]\ar[ru]\ar@/_1pc/[l]&\\
			1\ar[ru]&2\ar[u]&3\ar[u]&4\ar[lu]} 
						\substack{\vspace{2.5cm}\\ \hspace{.3cm}\displaystyle=}
		\xymatrix{&2&1&\\
			&\rond{}\ar[u]&\rond{}\ar[u]\ar[l]\\
			&\rond{}\ar[u]\ar@/_1pc/[r]&\rond{}\ar@/_.5pc/[lu]\ar[u]\ar@/_1pc/[l]&\\
			1\ar[ru]&2\ar[u]&3\ar[u]&4\ar[lu]}\]
	\end{example}
	\end{itemize}
We can now state and prove the main result of this section.
\begin{theo} \label{theo:ProP_graph}
	The family $\Gr=(\Gr_{k,l})_{k,l\in  \N}$, equipped with this $\sym\times \sym^{op}$-action and these horizontal and vertical concatenations,
	is a PROP.
\end{theo}
\begin{proof}
\begin{itemize}
	\item We check the \textbf{associativity} of $\circ$. Let $G$, $G'$ and $G''$ be three graphs with $o(G)=i(G')$ and $o(G')=i(G'')$. 
	The graphs $(G''\circ G')\circ G$ and $G''\circ (G'\circ G)$ may be different, but both are isomorphic to the graph $H$
	defined by:
	\begin{align*}
	V(H)&=V(G)\sqcup V(G')\sqcup V(G''),\\
	E(H)&=E(G)\sqcup E(G')\sqcup E(G'')\\
	&\sqcup \{(f,e)\in O(G)\times I(G'): \beta(f)=\alpha'(e)\}
	\sqcup \{(f,e)\in O(G')\times I(G''): \beta'(f)=\alpha''(e)\}\\
	&\sqcup \{(f,f',e)\in O(G)\times IO(G')\times I(G''): \beta(f)=\alpha'(f'),\beta'(f')=\alpha''(e)\},\\
	I(H)&=I(G)\sqcup \{(f,e)\in IO(G)\times I(G'): \beta(f)=\alpha'(e)\}\\
	&\sqcup \{(f,f',e)\in IO(G)\times IO(G')\times I(G''): \beta(f)=\alpha'(f'),\beta'(f')=\alpha''(e)\},\\
	O(H)&=O(G'')\sqcup \{(f,e)\in O(G')\times IO(G''): \beta'(f)=\alpha''(e)\}\\
	&\sqcup \{(f,f',e)\in O(G)\times IO(G')\times IO(G''): \beta(f)=\alpha'(f'),\beta'(f')=\alpha''(e)\},\\
	IO(H)&=\{(f,f',e)\in IO(G)\times IO(G')\times IO(G''): \beta(f)=\alpha'(f'),\beta'(f')=\alpha''(e)\},
	\end{align*}
	with immediate  source, target  and indexation maps. So $\circ$ induces an associative product  $\circ: \Gr(l,m)\otimes\Gr(k,l)\longrightarrow\Gr(k,m)$.  
	
	\item Let $I_1$ be the graph such that
	\begin{align*}
	V(I_1)&=E(I_1)=I(I_1)=O(I_1)=\emptyset,&
	IO(I_1)&=[1].
	\end{align*}
	We show that $I_1$ is the \textbf{unit} for $\circ$:  
	The indexation maps are both the identity of $[1]$. 
	For any integer $n\in  \N_0$, $I_1^{*n}$ is isomorphic to the graph $I_n$ such that
	\begin{align*}
	V(I_n)&=E(I_n)=I(I_n)=O(I_n))=\emptyset,&
	IO(I_n)&=[n],
	\end{align*}
	the indexation maps being both the identity of $[n]$. If $G$ is a graph and $k=i(G)$, then $H=G\circ I_k$
	is the graph such that:
	\begin{align*}
	V(H)&=V(G),&I(H)&=\{(\alpha(e),e): e\in I(G)\},\\
	E(H)&=E(G),&IO(H)&=\{(\alpha(e),e):e\in IO(G)\},\\
	O(H)&=O(G),&
	\end{align*}
	with immediate source, target  and indexation maps. This graph $H$ is isomorphic to $G$, via the isomorphism given by:
	\begin{align*}
	f_V&=\mathrm{Id}_{V(G)},&f_I((\alpha(e),e))&=e,\\
	f_E&=\mathrm{Id}_{E(G)},&f_{IO}((\alpha(e),e))&=e,\\
	f_O&=\mathrm{Id}_{O(G)}.
	\end{align*}
	Similarly, $I_l\circ G$ and $G$ are isomorphic. Hence, $I_1$ is the unit of $\circ$ in $\Gr$. 
	\item We check the \textbf{compatibility} of the horizontal and vertical concatenations.
	Let $G$, $G'$, $H$ and $H'$ be graphs such that $o(G)=i(H)$ and $o(G')=i(H')$. The graphs
	$(H*H')\circ (G*G')$ and $(H\circ G)* (H'\circ G')$ are both \cy{isomorphic} to the graph $K$, such that:
	\begin{align*}
	V(K)&=V(G)\sqcup V(G')\sqcup V(H)\sqcup V(H'),\\
	E(K)&=E(G)\sqcup E(G')\sqcup E(H)\sqcup E(H')\\
	&\sqcup \{(f,e)\in O(G)\times I(H): \beta(f)=\alpha'(e)\}\\
	&\sqcup \{(f,e)\in O(G')\times I(H'); \beta(f)=\alpha'(e)\},\\
	I(K)&=I(G)\sqcup I(G')\sqcup \{(f,e)\in IO(G)\times I(H); \beta(f)=\alpha'(e)\}\\
	&\sqcup \{(f,e)\in IO(G')\times I(H'); \beta(f)=\alpha'(e)\},\\
	O(K)&=O(H)\sqcup O(H')\sqcup \{(f,e)\in O(G)\times IO(H); \beta(f)=\alpha'(e)\}\\
	&\sqcup \{(f,e)\in O(G')\times IO(H'); \beta(f)=\alpha'(e)\},\\
	IO(K)&=\sqcup \{(f,e)\in IO(G)\times IO(H); \beta(f)=\alpha'(e)\}\\
	&\sqcup \{(f,e)\in IO(G')\times IO(H'); \beta(f)=\alpha'(e)\},
	\end{align*}
	with obvious source, target  and indexation maps. Hence, the vertical and the horizontal concatenations are compatible. 
	\item We check the \textbf{module} structure of $\Gr$ over the symmetric group.
	Let $G$ be a graph, $\sigma\in \sym_{o(G)}$ and $\tau\in \sym_{i(G)}$. We set:
	\begin{align}\label{eqGrmod}
	\sigma\cdot G&=(V(G),E(G),I(G),O(G),IO(G),s,t,\alpha,\sigma \circ\beta),\nonumber\\
	G\cdot \tau&=(V(G),E(G),I(G),O(G),IO(G),s,t,\tau^{-1}\circ\alpha,\beta).
	\end{align} 
	This induces a structure of $\sym\times \sym^{op}$-module over $\Gr$.
	
	\item Let us prove the compatibility of this action with the vertical concatenation. Let $G$ and $G'$ be two graphs
	such that $o(G)=i(G')$, and let $\sigma \in \sym_{o(G')}$, $\tau\in \sym_{o(G)}$,
	$\nu\in \sym_{i(G)}$. Clearly, the graphs $\sigma\cdot(G'\circ G)$ and $(\sigma\cdot G')\circ G$ are equal;
	the graphs $(G'\circ G)\cdot \nu$ and $G'\circ (G\cdot \nu)$ are equal. Let us compare the graphs
	$H=(G'\cdot \tau)\circ G$ and $H'=G'\circ (\tau \cdot G)$. Their set of vertices coincide. Moreover:
	\begin{align*}
	E(H)&=E(G)\sqcup E(G')\sqcup\{(f,e)\in O(G)\times I(G'): \beta(f)=\tau^{-1}\circ \alpha'(e)\},\\
	E(H')&=E(G)\sqcup E(G')\sqcup\{(f,e)\in O(G)\times I(G'): \tau\circ \beta(f)=\alpha'(e)\},
	\end{align*}
	so $E(H)=E(H')$. Similarly, $I(H)=I(H')$, $O(H)=O(H')$,  $IO(H)=IO(H')$ and $L(H)=L(H')$.  
	Moreover, the source, target and indexation maps are the same for $H$ and $H'$, so $H=H'$.
	\item 
	\ty{We prove} the \textbf{compatibility} of the $\sym\times \sym^{op}$-action with the horizontal composition.
	Let $G$ and $G'$ be two graphs, $\sigma \in \sym_{o(G)}$ and $\sigma'\in \sym_{o(G')}$. We put
	$H=(\sigma\cdot G)*(\sigma'\cdot G')$ and $H'=(\sigma \otimes \sigma')\cdot (G*G')$. 
	They have the same set of vertices, whether internal, input, output and input-output edges, and the source,
	target and indexation of output edges maps   for $H$ and $H'$ coincide.
	Both indexations of the set of output edges are   given by:
	\[\sigma''(e)=\begin{cases}
	\sigma \circ \beta(e)\mbox{ if }e\in O(G)\sqcup IO(G),\\
	o(G)+\sigma' \circ \beta'(e)\mbox{ if }e\in O(G')\sqcup IO(G').
	\end{cases}\]
	So $H=H'$.
	
	\cy{\item Finally we prove the commutativity of $*$.} Let $G$ and $G'$ be graphs. We set
	$H=c_{o(G),o(G')}\cdot (G*G')$ and $H'=(G'*G)\cdot c_{i(G),i(G')}$, where $c_{m, n}\in \sym_{m+n}$  was defined in (\ref{defcmn}). They have the same sets of vertices, internal,
	input, output and input-output edges, and the same source and target maps. The indexations maps are given by:
	\begin{align*}
	\alpha_H(e)&=\begin{cases}
	\alpha(e)+i(G')\mbox{ if }e\in I(G)\sqcup IO(G),\\
	\alpha'(e)\mbox{ if }e\in I(G')\sqcup IO(G'),
	\end{cases}\\
	\beta_H(e)&=\begin{cases}
	\beta(e)\mbox{ if }e\in O(G)\sqcup IO(G),\\
	\beta'(e)+o(G)\mbox{ if }e\in O(G')\sqcup IO(G'),
	\end{cases}\\ \\
	\alpha_{H'}(e)&=\begin{cases}
	\alpha'(e)\mbox{ if }e\in I(G')\sqcup IO(G'),\\
	\alpha(e)+i(G')\mbox{ if }e\in I(G)\sqcup IO(G),
	\end{cases}\\
	\beta_{H'}(e)&=\begin{cases}
	\beta'(e)+o(G)\mbox{ if }e\in O(G')\sqcup IO(G'),\\
	\beta(e)\mbox{ if }e\in O(G)\sqcup IO(G),
	\end{cases}
	\end{align*}
	so $H=H'$.
\end{itemize}
 
\end{proof}

\section{Freeness of the PROP of graphs} \label{sec:free_prop}

In this Section we prove the folklore result that a free PROP can be described in terms of graphs. Notice that a free PROP was already build in \cite{HaRo12}. However, this constrction of Hackney and Robertson is in the category of megagraphs. It is categorical in nature and thus not very well adapted to applications that are more down to earth, such as the ones we have in mind. Indeed, we will need a more explicit description of the free PROP.  

\subsection{Indecomposable graphs}

It is well-known that free operads can be described by trees (see \cite{ginzburg1994koszul} or \cite{bremner2016algebraic}). Thus it should be no suprise that free PROPs can be described by graphs. We now introduce the appropriate graphs.
\begin{defn}\label{def:indecomposable}
We call a  graph $G$ \textbf{indecomposable} if the \cy{four} following conditions hold:
\begin{enumerate}
\item $V(G)\neq \emptyset$.
\item $IO(G)=\emptyset$.
\item If $G'$ and $G''$ are two graphs such that $G=G'\circ G''$, then $V(G')=\emptyset$ or $V(G'')=\emptyset$.
\item If $G'$ and $G''$ are two graphs and $\sigma$, $\tau$ are two permutations such that
$G=\sigma \cdot (G'*G'')\cdot \tau$, then $V(G')=\emptyset$ or $V(G'')=\emptyset$.
\end{enumerate}
For any $k,l\in  \N$, the subspace of $\Gr(k,l)$ generated by isoclasses of indecomposable graphs $G$ with $i(G)=k$
and $o(G)=l$ is denoted by $\Gri(k,l)$.
\end{defn}
\begin{rk}
 The permutations in the fourth item of the definition of indecomposable graphs play an 
 important role: without them, one would allow for non connected graphs to be indecomposable, which can well happen when  
 the indexations of the inputs and outputs of the various connected components do not match. For 
 example, the graph 
\[\xymatrix{1&2&3&4\\
\rond{}\ar[u]\ar[ru]&&\rond{}\ar[u]\ar[ru]\\
1\ar[u]&3\ar[lu]&2\ar[u]}\]
would be indecomposable. Permuting  inputs we obtain 
\[\xymatrix{1&2&3&4\\
\rond{}\ar[u]\ar[ru]&&\rond{}\ar[u]\ar[ru]\\
1\ar[u]&2\ar[lu]&3\ar[u]}\]
which is decomposable. The same requirement does not arise for the vertical concatenation since 
one can write $\sigma.(P\circ Q).\tau=(\sigma.P)\circ (Q.\tau)=P'\circ Q'$.
\end{rk}
Clearly, $\Gri$ is a $\sym\times\sym^{op}$-submodule. Let us make this statement rigourous.
\begin{prop} \label{prop:ind_modules}
Let $G$ be a graph, $\sigma \in \sym_{o(G)}$ and $\tau\in \sym_{i(G)}$. Then
$G$ is indecomposable if, and only if, $\sigma\cdot G\cdot \tau$ is indecomposable. 
\end{prop}
\begin{proof}
Let us assume that $H=\sigma\cdot G\cdot \tau$ is indecomposable. Then $V(G)=V(H)\neq \emptyset$
and $IO(G)=IO(H)=\emptyset$. We are left to show that if $G=G'\circ G''$ or $G=\sigma\cdot(G'*G'')\cdot\tau$, then $V(G')=\emptyset$ or $V(G'')=\emptyset$.

Let us assume that $G=G'\circ G''$. Then:
\[H=\sigma\cdot(G'\circ G'')\cdot \tau=(\sigma\cdot G')\circ( G''\cdot \tau).\]
As $H$ is indecomposable, $V(G')=V(\sigma \cdot G')=\emptyset$ or $V(G'')=V(G''\cdot \tau)=\emptyset$. 
Let us now assume that $G=\sigma'\cdot (G'*G'')\cdot \tau'$. Then:
\[H=\sigma\cdot(\sigma'\cdot (G'*G'')\cdot \tau')\cdot \tau=((\sigma\sigma')\cdot G')*(G''\cdot (\tau\tau')).\]

As  $H$ is indecomposable, $V(G')= V \left((\sigma\sigma')\cdot G'\right)=\emptyset$ or $V(G'')=V(G''\cdot (\tau\tau'))=\emptyset$.\\

Conversely, if $G$ is indecomposable, then $G=\sigma^{-1}\cdot H\cdot \tau^{-1}$ is indecomposable,
so by the first point $H$ is indecomposable.
\end{proof}
In order to \ty{characterise} indecomposable graphs, we need to introduce some notations for restrictions of graphs.
\begin{notation} 
Let $G$ be a graph.
\begin{enumerate}
\item Let $J\subseteq V(G)$. We define (non uniquely due to the non uniqueness of the maps $\alpha'$ and $\beta'$) the graph $G_{\mid J}$ by:
\begin{align*}
V(G_{\mid J})&=J,\\
E(G_{\mid J})&=\{e\in E(G):  s(e)\in J, t(e)\in J\},\\
I(G_{\mid J})&=\{e\in I(G):  t(e)\in J\}\sqcup  \{e\in E(G): s(e)\notin J, t(e)\in J\},\\
O(G_{\mid J})&=\{e\in O(G):  s(e)\in J\}\sqcup  \{e\in E(G):  s(e)\in J, t(e)\notin J\},\\
IO(G_{\mid J})&=\cy{\emptyset}.
\end{align*}
The source and target maps are defined by:
\begin{align*} 
&\forall e\in E(G_{\mid J})\sqcup O(G_{\mid J}),&s_{G_{\mid J}}(e)&=s(e),\\
&\forall e\in E(G_{\mid J})\sqcup I(G_{\mid J}),&t_{G_{\mid J}}(e)&=t(e),
\end{align*}
The indexation of the input edges is any indexation map $\alpha'$ such that:
\begin{align*}
&\forall e,e'\in \left(I(G)\sqcup IO(G)\right)\cap \left(I(G_{\mid J})\sqcup IO(G_{\mid J})\right),&
\alpha'(e)<\alpha'(e')&\Longleftrightarrow \alpha(e)<\alpha(e').
\end{align*}
The indexation of the output edges is any indexation map $\beta'$ such that:
\begin{align*}
&\forall f,f'\in \left(O(G)\sqcup IO(G)\right)\cap \left(O(G_{\mid J})\sqcup IO(G_{\mid J})\right),&
\beta'(f)<\beta'(f')&\Longleftrightarrow \beta(f)<\beta(f').
\end{align*}
\item We denote by $\tilde{G}$ the graph obtained from $G$ by deleting all its input-output edge. Rigorously, $\tilde G$ is defined by:
% \begin{align*}
% &V(\tilde{G})=V(G),\qquad E(\tilde{G})=E(G),\\
% I(\tilde{G})=I(G)&,\qquad\qquad  O(\tilde{G})=O(G),\qquad\qquad  IO(\tilde{G})=\emptyset,\\
% &\tilde{s}=s,\qquad \qquad \qquad \tilde{t}=t.
% \end{align*}
\begin{align*}
I(\tilde{G})&=I(G),&O(\tilde{G})&=O(G),&IO(\tilde{G})&=\emptyset,\\
V(\tilde{G})&=V(G),&E(\tilde{G})&=E(G),\\
\tilde{s}&=s,&\tilde{t}&=t.
\end{align*}
The indexation of the input edges is the unique indexation map $\tilde{\alpha}$ such that:
\begin{align*}
&\forall e,e'\in I(G),&\tilde{\alpha}(e)<\tilde{\alpha}(e')&\Longleftrightarrow \alpha(e)<\alpha(e').
\end{align*}
The indexation of the output edges is the unique indexation map $\tilde{\beta}$ such that:
\begin{align*}
&\forall f,f''\in O(G),&\tilde{\beta}(f)<\tilde{\beta}(f')&\Longleftrightarrow \beta(f)<\beta(f').
\end{align*}
\end{enumerate}
\end{notation}
We also need to introduce the notions of paths and cycles of graphs.
\begin{defn}
Let $G$ be a graph.
\begin{enumerate}
\item A \textbf{path} in $G$ is a sequence $p=(e_1,\ldots,e_k)$ of internal edges of $G$ such that for any $i\in [k-1]$,
$t(e_i)=s(e_{i+1})$. The source of $p$ is $s(e_1)$ and its target is $t(e_k)$, and we shall say that 
$p$ is a path from $s(e_1)$ to $t(e_k)$ of length $k$. By convention, for any $x\in V(G)$, 
there exists a unique path from $x$ to $x$ of length $0$.
\item We shall say that a path $p$ is a \textbf{cycle} if its source and its target are equal 
and if its length is nonzero.
\end{enumerate}
\end{defn}
We consider oriented-pathwise connected components of graphs.
\begin{lemma} \label{lem:paths_subset}
Let $G$ be a graph such that $V(G)\neq \emptyset$. We denote by $\calO(G)$ the set of nonempty
subsets $I$ of $V(G)$ such that for any $x\in I$, for any $y\in V(G)$, if there exists a path in $G$ from $x$ to $y$,
then $y\in I$. Then:
\begin{enumerate}
\item If $I,J\in \calO(G)$, either $I\cap J=\emptyset$ or $I\cap J\in \calO(G)$.
\item For any $x \in V(G)$, there exists a unique element $\langle x\rangle \in \calO(G)$ which contains $x$
and is minimal for the inclusion. Moreover:
\[\langle x\rangle=\{y\in V(G): \mbox{ there exists a path in $G$ from $x$ to $y$}\}.\]
\end{enumerate}
\end{lemma}
Notice that, for any $x\in V(G)$, if $G_x$ is the connected component of $G$ that contains $x$ then  $\langle x\rangle\subseteq G_x$, but we do not necessarily have an equality, as the edges are \emph{oriented}.
\begin{proof}
\begin{enumerate}
 \item If $I\cap J\neq \emptyset$, let $x\in I\cap J$ and $y\in V(G)$ such that there exists a path in $G$ from $x$ to $y$. As $I,J\in \calO(G)$, $y\in I\cap J$, so $I\cap J\in \calO(G)$.
 \item Note that $V(G)\in \calO(G)$. Let $x\in V(G)$; by the first item, the following element of $\calO(G)$ is the minimal (for the inclusion) element of $\calO(G)$ that contains $x$:
\[\langle x\rangle=\bigcap_{I\in \calO(G), \: x\in I}I.\]
On the one hand,   a set $I$ in $\calO(G)$ contains   $x$ if and only if any path emanating from $x$ ends at an element of $I$. So it contains all the ending vertices of such paths and hence the set
\[I_x:=\{y\in V(G): \mbox{ there exists a path in $G$ from $x$ to $y$}\}.\] 
Thus, $I_x\subseteq \langle x\rangle$ \cy{since $\langle x\rangle\in\calO(G)$}. 
On the other hand, let $y\in I$ and $z\in V(G)$, such that there exists a path from $y$ to $z$ in $G$.
As there exists a path from $x$ to $y$ in $G$, there exists a path from $x$ to $z$, so $z\in I_x$.
Hence, $I_x$ lies in  $\calO(G)$ \ty{and} in turn contains $x$, so $\langle x\rangle\subseteq I_x$. 
\end{enumerate}
\end{proof}
As it turns out, every graph can be decomposed into indecomposable graphs.
\begin{prop}\label{prop:mindec}
Let $G$ be a graph such that $V(G)\neq \emptyset$. We denote by $J_1,\ldots,J_k$ the minimal elements
(for the inclusion) of the set  $\calO(G)$ of nonempty
subsets $I$ of $V(G)$ stable under paths as in Lemma \ref{lem:paths_subset}, and we set $G_i=\tilde{G}_{\mid J_i}$ for any $i\in [k]$. 
Then $G_1,\ldots,G_k$ are indecomposable graphs and there exists a graph $G_0$, an integer $p$ and a permutation $\gamma$ such that:
\begin{equation} \label{eq:min-dec}
 G=(\gamma \cdot (G_1*\ldots *G_k*I_p)\circ G_0).
\end{equation}
Such a decomposition will be called  \textbf{minimal}.
\end{prop}
\begin{rk}
 The minimal decomposition of a graph is not unique. It depends on the indexation of the minimal elements of $\calO(G)$
and of the choice of the indexation of their input and output edges.  Importantly, it only depends on that.
\end{rk}

\begin{proof}
By definition, $V(G_i)=J_i\neq \emptyset$ and $IO(G_i)=\emptyset$ for any $i$. 
Let us assume that $G_i=G'\circ G''$. If $V(G')\neq \emptyset$, then clearly $V(G')\in \calO(G_i)$ 
and, as $J_i\in \calO(G)$, we deduce that $V(G')\in \calO(G)$. As $J_i$ is minimal in $\calO(G)$,
$V(G')=J_i=V(G_i)$, so $V(G'')=\emptyset$. Similarly, if $G_i=\sigma\cdot (G'*G'')\cdot \tau$,
then $V(G')=\emptyset$ or $V(G'')=\emptyset$: we proved that $G_i$ is indecomposable.

Let us assume that $I=V(G_i)\cap V(G_j)\neq \emptyset$. Then \cy{by the first point of Lemma \ref{lem:paths_subset}} $I\in \calO(G)$ and, by minimality of $J_i$ and $J_j$,
$J_i=J_j=I$, so the $J_i$ are disjoint.

Let us set $K:=V(G)\setminus (J_1\cup\ldots \cup J_k)$ and $G':=G_{\mid K}$. As $J_1,\ldots, J_k$ lie in $\calO(G)$,
there is no internal edge of $G$ from a vertex of $G_i$ to a vertex of $G'$, and any outgoing edge of $G'$  is either
glued in $G$ to an incoming edge of $G_i$ or is an outgoing edge of $G$. Hence, 
there exist permutations
$\gamma$, $\sigma$ and $\tau$, and 
two integers $p:=|IO(G)|$ and $q:=|\{e\in I(G):t(e)\in J_1\cup\ldots \cup J_k\}|$ such that:
\[G=\gamma \cdot(G_1*\ldots*G_k*I_p)\circ (\sigma\cdot (I_q*G')\cdot \tau).\]
We conclude in taking $G_0=\sigma\cdot (I_q*G')\cdot \tau$. 
\end{proof}
We can finally fully \ty{characterise} indecomposable graphs.
\begin{prop}\label{propindecomposable}
Let $G$ be a graph such that $V(G)\neq \emptyset$ and $IO(G)=\emptyset$. The graph $G$ is indecomposable if, and only if, for any $x,y\in V(G)$, there exists a path from $x$ to $y$ in $G$.
\end{prop}
\begin{proof}
First notice that if $|V(G)|=1$ the result trivially holds. In the following, we  therefore assume that $|V(G)|\geq 2$.

Let $G=\gamma \cdot (G_1*\ldots *G_k*I_p)\circ G_0$ be a minimal decomposition of $G$. We prove the two directions of the implication separately.

$\Longrightarrow$:  Note that $V(G_1)\neq \emptyset$. As $G$ is indecomposable, 
necessarily  $V(G_0)=\emptyset$, and there exists a permutation
$\tau\in \sym_\cy{q}$ such that $G_0=I_q\cdot \tau$. Therefore, 
$G=\gamma \cdot (G_1*\ldots *G_k*I_p)\cdot \tau$. As $G$ is indecomposable, $k=1$ and $V(G)=V(G_1)=J_1$. Thus Proposition \ref{prop:mindec} implies that $V(G_1)$ is the unique minimal element of $\calO(G)$ for the inclusion. Furthermore, since any element of $\calO(G)$ is a subset of $V(G)=V(G_1)$, $V(G_1)$ is also the maximal element for the inclusion of $\calO(G)$. Consequently $\calO(G)$ is reduced to the singleton $\{V(G_1)\}=\{V(G)\}$. 
 
Therefore, by the second point of Lemma \ref{lem:paths_subset},
for any $x\in V(G)$, $\langle x\rangle=V(G)$, so for any $y\in V(G)$,
there exists a path from $x$ to $y$ in $G$.

$\Longleftarrow$: 
 If $k\geq 2$, there is no path in $G$ from any vertex of $G_1$ to any vertex of $G_2$, so $k=1$. 
Thus, $V(G_0)=\emptyset$
and there exists a permutation $\tau$ such that $G_0=I_p\cdot \tau$. We obtain that
\[G=\gamma\cdot (G_1*I_p)\cdot \tau.\]
As $IO(G)=\emptyset$, we obtain that  
$p=0$, so $G=\gamma\cdot G_\cy{1}\cdot \tau$ is indecomposable \cy{by Proposition \ref{prop:ind_modules}}.
\end{proof}
\begin{rk}
     Another way to formulate the above Proposition is to say that a graph $G$  is indecomposable if, and only if,  \cy{for any pair $(x,y)$ of its vertices,
     a cycle of  strictly positive length goes through $x$ and $y$}.
    \end{rk}
    
  \subsection{Free PROP} \label{subsec:freePROP}
  
  We now state and give a sketch of the proof of one of the main results of this section, namely the freeness of the PROP $\Gr$. 
\begin{theo} \label{thm:freeness_Gr}
Let $P$ be a PROP and $\phi:\Gri\longrightarrow P$ be a morphism of $\sym\times \sym^{op}$-modules.
There exists a unique PROP morphism $\Phi:\Gr\longrightarrow P$ such that $\Phi_{\mid \Gri}=\phi$. 
In other words, $\Gr$ is the free PROP generated by $\Gri$. 
\end{theo}

\begin{proof} We provide here a sketch of the proof,  and refer the reader to Subsection \ref{subsec:proof_free_PROP} for a full proof.
We define $\Phi(G)$ for any graph $G$ by induction on its number $n$ of vertices. If $n=0$,  there exists a permutation $\sigma\in \sym_k$ such that $G=\sigma\cdot I_k$. We set
\[\Phi(G)=\sigma\cdot I_k.\]
If $n>0$ and $G$ is indecomposable, we  set $\Phi(G)=\phi(G)$. Otherwise, 
let \[G=\gamma \cdot (G_1*\ldots *G_k*I_p)\circ G_0\]
be a  minimal decomposition of $G$.  As  $V(G_1)\neq \emptyset$, $|V(G_0)|<n$, we set:
\[\Phi(G)=\gamma \cdot (\phi(G_1)*\ldots *\phi(G_k)*I_p)\circ \Phi(G_0).\]
One can prove that this does not depend on the choice of the minimal decomposition of $G$ with the help of the PROP axioms
applied to $P$. Using minimal decompositions of vertical or horizontal concatenations of graphs,
one can show that $\Phi$ is compatible with both concatenations. 
\end{proof}
\cy{The most technical parts of the proof described above is to find the minimal decompositions of $G*G'$ and $G'\circ G$ in terms of the minimal decompositions of $G$ and $G'$.}

Now we see that we have made a step toward to our initial goal. Embedding Feynman graphs into our graphs given by Definition \ref{def:graph}, one obtains a Feynman rule $F:\Gr\mapsto\calA$ if one can show:
\begin{enumerate}
 \item That the target algebra $\calA$ carries a PROP structure,
 \item That it exists a morphism of $\sym\times \sym^{op}$-modules $f:\Gri\longrightarrow \calA$.
\end{enumerate}
While the first point can be conjectured for reasons that will be clarified below, the real issue lies with the second point. Indeed, an indecomposable graph can be arbitrarily complicated. Therefore if one were to use Theorem  \ref{thm:freeness_Gr} to build a Feynman rule, one would still has to prove that the building block of the Feynman rule, i.e.~the map $f$ is well-defined and well-behaved on an infinity of complicated diagrams. We see that PROPs do not simplify enough our work for practical purposes.

There is a workaround, namely to work only with graphs without closed loops. In \cite[Proposition 3.3.3]{ClFoPa20} it was shown that such graphs form a sub-PROP of $\Gr$ which is the free PROP generated by a infinite family of simple graphs with only one vertices. One could then use the fact that the internal edges of any graph can have their orientation changed such that the obtained graph has no closed loops. It would still be necessary to show that the obtained evaluation of the Feynman rule on a graph with oriented loops does not depend on the choice made to remove the oriented loops.

While it seems like a tractable task, we have chosen not to follow this direction for at least three reasons:

Firstly, closed loops in Feynman graphs have a tendency to create singularities that should then be removed by renormalisation. It is now a commonly accepted stance that these singularities actually carry relevant information about the structure of the Feynman graph being regularised. Permuting the internal edges would darken the meaning of the singularities.

Secondly, proving that the target algebra $\calA$ has a PROP structure seems to be an even more formidable task than proving that it has a TRAP (see Section \ref{sec:TRAP} below) structure. 

Thirdly, closed loops are so normal in QFT, so ingrained in the QFT culture that it seem a huge \emph{faux pas} to not try to deal with them. And this is not mentioning the fact that this way to removing closed loops severely lacks elegance.

Therefore, in \cite{ClFoPa20} and \cite{ClFoPa21} we have instead chosen to introduce a new category of structure more suited for our goals. Nonetheless, before introducing and studying TRAPs, let us give a \ty{detailed} proof of Theorem \ref{thm:freeness_Gr}.

%I will discuss in the next Section \ref{sec:deco_PROP} how one can enhance Theorem \ref{thm:freeness_Gr} in the case of decorated graphs. Not only is this important to understand how one can hope to use freeness results like \ref{thm:freeness_Gr} to build Feynman rules, but it is also a possible direction for generalisations of Multi Zeta Values to directed acyclic graphs.
\begin{rk}
 Equivalent freeness results (for graphs and cycleless graphs) also hold for decorated and planar graphs: see \cite[Theorems 3.4.3 and 4.1.4]{ClFoPa20}. These results are particularly important in the case of graphs decorated by a $\sym\times\sym^{\rm op}$-modules, where they allow to build an endofunctor of the category $\PROP$. Applied to the PROP $\Hom^c_V$ for a Fréchet nuclear space $V$, this construction gives rise to a generalised composition of operators (\cite[Corollaries 4.4.2 and 4.43]{ClFoPa20}). However, since according to the discussion above, PROPs are not the right structure for perturbative QFT, I have chosen not to include these results in this thesis.
 
%  Since planar graphs will not be used in the rest of this thesis, I have chosen not to include them here.
\end{rk}
% IL SEMBLE QUE J'EN AI BESOIN POUR LES ENDOFONCTEURS DE PROPs.

\subsection{Proof of freeness of $\Gr$} \label{subsec:proof_free_PROP}

Here is a full proof of Theorem \ref{thm:freeness_Gr}.
\begin{proof}
Let us define $\Phi(G)$ for any graph $G$ by induction on $n=|V(G)|$, such that for any permutation
$\sigma\in \sym_{i(G)}$, $\tau \in \sym_{o(G)}$,
\[\Phi(\sigma\cdot G\cdot \tau)=\sigma\cdot \Phi(G)\cdot \tau.\]
If $n=0$, there exists a unique permutation $\gamma\in \sym_k$ such that $G=\gamma\cdot I_k$. We put
\[\Phi(G)=\gamma\cdot I_k,\]
where we used the same notation $I_k$ for the units of \ty{$\Gr$} and $P$.

If $\sigma,\tau\in \sym_k$:
\begin{align*}
\Phi(\sigma \cdot G\cdot \tau)&=\Phi((\sigma\gamma)\cdot I_k\cdot \tau)\\
&=\Phi((\sigma\gamma \tau)\cdot I_k)\\
&=(\sigma\gamma \tau)\cdot I_k\\
&=\sigma\cdot(\gamma\cdot I_k)\cdot \tau\\
&=\sigma \cdot \Phi(G)\cdot \tau
\end{align*}
\cy{where we used the facts that $\Gr$ is an $\sym\times\sym^{op}$-module and that, for any $k\geq1$ and $\tau\in\sym_k$ we have $I_k\cdot\tau=(I_k\cdot\tau)\circ I_k=I_k\circ(\tau\cdot I_k)=\tau\cdot I_k$.}
Let us assume that $\Phi(G')$ is defined for any graph $G'$ such that $|V(G')|<n$.
Let \[G=\gamma \cdot (G_1*\ldots *G_k*I_p)\circ G_0\]
be a  minimal decomposition of $G$.  {If $G$ is indecomposable, we set $\Phi(G)=\phi(G)$. Otherwise,} as $V(G_1)\neq \emptyset$, $|V(G_0)|<n$. We put:
\[\Phi(G)=\gamma \cdot  (\phi(G_1)*\ldots *\phi(G_k)*I_p)\circ\Phi(G_0).\]
Let us first prove that this does not depend on the choice of the minimal decomposition of $G$.
Starting from a minimal decomposition of $G$, one obtains all possible minimal decompositions
of $G$ by a finite sequence of operations of type A and B:
\begin{itemize}
\item Type A: changing the indexations of the input and output edges of the graphs $G_i$.
We obtain a minimal decomposition $G=\gamma' \cdot (G'_1*\ldots*G'_k*I_p)\circ G'_0$,
such that there exists permutations $\alpha_i$, $\beta_i$, with:
\begin{align*}
G'_i&=\alpha_i\cdot G_i\cdot \beta_i,\\
G'_0&=(\beta_1^{-1}\otimes \ldots \otimes \beta_k^{-1}\otimes \mathrm{Id}_p)\cdot G_0,\\
\cy{\gamma}'&=\cy{\gamma}(\alpha_1^{-1}\otimes \ldots \otimes \alpha_k^{-1}\otimes \mathrm{Id}_p).
\end{align*}
\item Type B: permuting $G_l$ and $G_{l+1}$  for  $l\in [k-1]$. 
We obtain another minimal decomposition $G=\gamma' \cdot (G'_1*\ldots*G'_k*I_p)\circ G'_0$, with:
\begin{align*}
G'_i&=\begin{cases}
G_{l+1}\mbox{ if }i=l,\\
G_l\mbox{ if }\ty{i}=l+1,\\
G_i\mbox{ otherwise};
\end{cases}\\
G'_0&= (\mathrm{Id}_{i(G_1)+\ldots+i(G_{l-1})}\otimes c_{i(G_{l+1}),i(G_l)}\otimes \mathrm{Id}_{i(G_{l+2})+\ldots+i(G_k)+p})
\cy{\cdot} G_0,\\
\gamma'&=\gamma(\mathrm{Id}_{o(G_1)+\ldots+o(G_{l-1})}\otimes c_{o(G_l),o(G{l+1})}\otimes
 \mathrm{Id}_{o(G_{l+2})+\ldots+o(G_k)+p}).
\end{align*}
\end{itemize}
Let $G=\gamma \cdot (G'_1*\ldots*G'_k*I_p)\circ G'_0$ be another minimal decomposition of $G$.
Then it is enough to prove that
\[\gamma \cdot  (\phi(G_1)*\ldots *\phi(G_k)*I_p)\circ\Phi(G_0)
=\gamma' \cdot  (\phi(G_1')*\ldots *\phi(G_k')*I_p)\circ\Phi(G_0').\]
We can assume that \cy{this new minimal decomposition is obtained from the initial one} by a single operation of type A or of type B.
If it is of type A:
\begin{align*}
&\gamma' \cdot (\phi(G'_1)*\ldots*\phi(G'_k)*I_p)\circ \Phi(G'_0)\\
&=\gamma\cdot(\alpha_1^{-1}\otimes \ldots \otimes \alpha_k^{-1}\otimes \mathrm{Id}_p)\cdot
(\phi(\alpha_1\cdot G_1\cdot \beta_1)\otimes \ldots \otimes \phi(\alpha_k\cdot G_k\cdot \beta_k)*I_p)\\
&\circ \Phi((\beta_1^{-1}\otimes \ldots \otimes \beta_k^{-1}\otimes \mathrm{Id}_p)\cdot G_0)\\
&=\gamma\cdot(\alpha_1^{-1}\otimes \ldots \otimes \alpha_k^{-1}\otimes \mathrm{Id}_p)\cdot
(\alpha_1\cdot \phi(G_1)\cdot \beta_1\otimes \ldots \otimes\alpha_k\cdot  \phi(G_k)\cdot \beta_k*I_p)\\
&\circ((\beta_1^{-1}\otimes \ldots \otimes \beta_k^{-1}\otimes \mathrm{Id}_p)\cdot \Phi( G_0))\quad\cy{\text{by the induction hypothesis}}\\
&=\gamma(\alpha_1^{-1}\otimes \ldots \otimes \alpha_k^{-1}\otimes \mathrm{Id}_p)
(\alpha_1\otimes \ldots \otimes \alpha_k\otimes \mathrm{Id}_p)\cdot
(\phi(G_1)\otimes \ldots \otimes\phi(G_k)*I_p)\\
&\circ  (\beta_1\otimes \ldots \otimes \beta_k\otimes \mathrm{Id}_p)(\beta_1^{-1}\otimes \ldots \otimes \beta_k^{-1}\otimes \mathrm{Id}_p)\cdot \Phi(G_0)\quad\cy{\text{since $\phi$ is a morphism of $\sym\times\sym^{op}$-modules}}\\
&=\gamma\cdot(\phi(G_1)\otimes \ldots \otimes\phi(G_k)*I_p)\circ \Phi(G_0).
\end{align*}
If it is of type B:
\begin{align*}
&\gamma' \cdot (\phi(G'_1)*\ldots*\phi(G'_k)*I_p)\circ \Phi(G'_0)\\
&=\gamma (\mathrm{Id}_{o(G_1)+\ldots+o(G_{l-1})}\otimes c_{o(G_l),o(G_{l+1})}
\otimes \mathrm{Id}_{o(G_{l+2})+\ldots+o(G_k)+p})\\
& \cdot(\phi(G_1)*\ldots*\phi(G_{l+1})*\phi(G_l)*\ldots*\phi(G_k)*I_p)\\
&\circ\Phi((\mathrm{Id}_{i(G_1)+\ldots+i(G_{l-1})}\otimes c_{i(G_{l+1}),i(G_l)}\otimes \mathrm{Id}_{i(G_{l+2})+\ldots+i(G_k)+p})
\cdot G_0)\\
% &\gamma' \cdot (\phi(G'_1)*\ldots*\phi(G'_k)*I_p)\circ \Phi(G'_0)\\
&\cy{=}\gamma (\mathrm{Id}_{o(G_1)+\ldots+o(G_{l-1})}\otimes c_{o(G_l),o(G_{l+1})}
\otimes \mathrm{Id}_{o(G_{l+2})+\ldots+o(G_k)+p})\\
& \cdot(\phi(G_1)*\ldots*\phi(G_{l+1})*\phi(G_l)*\ldots*\phi(G_k)*I_p)\quad\cy{\text{by the induction hypothesis}}\\
&\circ(\mathrm{Id}_{i(G_1)+\ldots+i(G_{l-1})}\otimes c_{i(G_{l+1}),i(G_l)}\otimes \mathrm{Id}_{i(G_{l+2})+\ldots+i(G_k)+p})
\cdot\Phi(G_0)\\
&=\gamma \cdot (\phi(G_1)*\ldots*c_{o(G_l),o(G_{l+1})}\cdot(\phi(G_{l+1})*\phi(G_l))\cdot c_{i(G_{l+1}),i(G_l)}
*\ldots*\phi(G_k)*I_p)\circ\Phi(G_0)\\
&=\gamma \cdot (\phi(G_1)*\ldots*\phi(G_l)*\phi(G_{l+1}))*\ldots*\phi(G_k)*I_p)\circ\Phi(G_0)\quad\cy{\text{by commutativity of }*}\\
&=\gamma\cdot(\phi(G_1)\otimes \ldots \otimes\phi(G_k)*I_p)\circ \Phi(G_0).
\end{align*}
So $\Phi(G)$ is well-defined. Let $\sigma \in \sym_{o(G)}$ and $\tau\in \sym_{i(G)}$. We put $H=\sigma\cdot G\cdot \tau$. \cy{Then $H$ is indecomposable if, and only if, $G$ is (Proposition \ref{prop:ind_modules}). For these graphs, $\Phi$ coincides with $\phi$ which is a morphism of $\sym\times\sym^{op}$-modules by assumption. Otherwise, a}
 minimal decomposition of $H$ is given by:
\begin{align*}
H_0&=G_0\cdot \tau,&H_i&=G_i\mbox{ if }i\in [k],&\gamma'&=\sigma\gamma.
\end{align*}
Hence:
\begin{align*}
\Phi(H)&=\sigma\gamma\cdot (\phi(G_1)*\ldots*\phi(G_k)*I_p)\circ \Phi(G_0\cdot \tau)\\
&=\sigma\gamma\cdot (\phi(G_1)*\ldots*\phi(G_k)*I_p)\circ \Phi(G_0)\cdot \tau\quad\cy{\text{by the induction hypothesis}}\\
&=\sigma\cdot(\gamma\cdot (\phi(G_1)*\ldots*\phi(G_k)*I_p)\circ \Phi(G_0))\cdot \tau\\
&=\sigma\cdot \Phi(G)\cdot \tau.
\end{align*}
Consequently, we have defined a \cy{morphism of $\sym\times\sym^{op}$-modules} $\Phi:\Gr\longrightarrow P$, extending the morphism  $\phi$ of $\sym\times \sym^{op}$-modules.
Let us prove that it is compatible with both concatenations.

Let $G$ and $G'$ be two graphs. Let us prove that $\Phi(G*G')=\Phi(G)*\Phi(G')$ by induction on $n'=|V(G')|$.
 If $n'=0$, there exists $q\in\N$ and $\sigma' \in \sym_q$ such that $G'=\sigma'\cdot I_q$. 
We proceed by induction on $n=|V(G)|$. If $n=0$, there exists $p\in\N$ and $\sigma\in \sym_p$, such that 
$G=\sigma\cdot I_p$. Then $G*G'=(\sigma\otimes \sigma')\cdot I_{p+q}$, and:
\begin{align*}
\Phi(G*G')&=(\sigma\otimes \sigma')\cdot I_{p+q}\\
&=(\sigma\otimes \sigma')\cdot (I_p*I_q)\\
&=(\sigma \cdot I_p)*(\sigma'\cdot I_q)\quad\cy{\text{by compatibility of $\cdot$ and $*$}}\\
&=\Phi(G)*\Phi(G').
\end{align*} 
Otherwise (i.e. if $n>0$, still in the case $n'=0$), let $G=\gamma\cdot (G_1*\ldots*G_k*I_p)\circ G_0$ be a minimal decomposition of $G$.
A minimal decomposition of $G*G'$ is:
\[G*G'=(\gamma\otimes \sigma')\cdot (G_1*\ldots* G_k*I_{p+q}) \circ (G_0*I_q).\]
We can use the induction hypothesis on $G_0$: we have $|V(G_0)|<|V(G)|$ since otherwise $G$ the graphs $G_i$ would be of the form $\sigma_i.I_{p_i}$, which is in contradiction with the fact that the $G_i$ are indecomposable (Proposition \ref{prop:mindec}) and the definition of indecomposable graphs (Definition \ref{def:indecomposable}). So, using the induction hypothesis on $G_0$:
\begin{align*}
\Phi(G*G')&=(\gamma\otimes \sigma')\cdot (\phi(G_1)*\ldots*\phi(G_k)*I_{p+q})
\circ \Phi(G_0*I_q)\\
&=(\gamma\otimes \sigma')\cdot\phi(G_1)*\ldots*\phi(G_k)*I_p*I_q)
\circ (\Phi(G_0)*I_q)\\
&\cy{=(\gamma\otimes \sigma')\cdot\phi(G_1)*\ldots*\phi(G_k)\circ \Phi(G_0)*(I_q\circ I_q) \quad\text{by compatibility of $*$ and $\circ$}} \\
&=(\gamma\cdot (\phi(G_1)*\ldots*\phi(G_k)*I_p)\circ \Phi(G_0))*(\sigma'\cdot I_q)\quad\cy{\text{by compatibility of $*$ and $\cdot$}}\\
&=\Phi(G)*\Phi(G').
\end{align*}
So the result holds at rank $n'=0$.

Let us assume that the results hold at any rank  $\cy{|V(G')|}<n'$. Let us consider minimal decompositions of $G$ and $G'$:
\begin{align*}
G&=\gamma \cdot (G_1*\ldots*G_k*I_p)\circ G_0,&G'&=\gamma' \cdot  (G'_1*\ldots*G'_l*I_q)\circ G'_0,
\end{align*}
with the convention $k=0$ if $V(G)=\emptyset$. \cy{To} obtain a minimal decomposition of $G*G'$, \cy{let us set
\begin{equation*}
 a:=G_1*\ldots*G_k,\qquad b:=G'_1*\ldots*G'_l.
\end{equation*}
Then we have }
\begin{align*}
G*G' & \cy{=\big(\gamma \cdot (a*I_p)\circ G_0\big)*\big(\gamma' \cdot  (b*I_q)\circ G'_0\big)} \\
 & \cy{= \big[\big(\gamma \cdot (a*I_p)\big)*\big(\gamma' \cdot  (b*I_q)\big)\big]\circ(G_0* G_0') \quad\text{by compatibility of $*$ and $\circ$}} \\
 & \cy{=\big[(\gamma\otimes\gamma')\cdot(a*I_p*b*I_q)\big]\circ(G_0* G_0') \quad\text{by compatibility of $\cdot$ and $*$}} \\
 & \cy{=\big[(\gamma\otimes\gamma')\cdot\big(a*(c_{p,o(b)}\cdot (b*I_p)\cdot c_{i(b),p})*I_q\big)\big]\circ(G_0* G_0')\quad\text{by commutativity of $*$}} \\
 & \cy{=\big[(\gamma\otimes\gamma')(\Id_{o(a)}\otimes c_{p,o(b)}\otimes\Id_{q})\cdot (a*b*I_{p+q})\cdot(\Id_{i(a)}\otimes c_{i(b),p}\otimes\Id_q)\big]\circ(G_0* G_0')} 
\end{align*}
\cy{by compatibility of $\cdot$ and $*$. Using the compatibility of $\cdot$ and $\circ$ and noticing that the numbers of inputs and outputs of $a$ (resp. $b$) are given by the sum of the numbers of inputs and outputs of the $G_i$s (resp. of the $G'_j$s) we obtain the following minimal decomposition of $G*G'$:}
\begin{align*}
G*G'  
&=(\gamma\otimes \gamma')(\mathrm{Id}_{o(G_1)+\ldots+o(G_k)}\otimes c_{p,o(G'_1)+\ldots+o(G'_l)}\otimes \mathrm{Id}_q)\\
&\cdot (G_1*\ldots*G_k*G'_1*\ldots*G'_l*I_{p+q})\\
&\circ ((\mathrm{Id}_{i(G_1)+\ldots+i(G_k)}\otimes c_{i(G'_1)+\ldots+i(G'_l),p}\otimes \mathrm{Id}_q)) \cdot (G_0*G'_0)).
\end{align*}
We apply the induction assumption $\Phi(G*G')=\Phi(G)*\Phi(G')$ for  $|V(G')|<n'$  to  $G'_0$ whose number of vertices is smaller than that of 
$G'$ and hence smaller than $n'$ for the same reason than the one given in the case $n'=0$.
\begin{align*}
\Phi(G*G')
&=(\gamma\otimes \gamma')(\mathrm{Id}_{o(G_1)+\ldots+o(G_k)}\otimes c_{p,o(G'_1)+\ldots+o(G'_l)}\otimes \mathrm{Id}_q)\\
&\cdot (\phi(G_1)*\ldots*\phi(G_k)*\phi(G'_1)*\ldots*\phi(G'_l)*I_{p+q})\\
&\circ \cy{\Phi\big((\mathrm{Id}_{i(G_1)+\ldots+i(G_k)}\otimes c_{i(G'_1)+\ldots+i(G'_l),p}\otimes \mathrm{Id}_q) \cdot (G_0*G'_0)\big)}\\
&=(\gamma\otimes \gamma')(\mathrm{Id}_{o(G_1)+\ldots+o(G_k)}\otimes c_{p,o(G'_1)+\ldots+o(G'_l)}\otimes \mathrm{Id}_q)\\
&\cdot (\phi(G_1)*\ldots*\phi(G_k)*\phi(G'_1)*\ldots*\phi(G'_l)*I_{p+q})\\
&\circ (\mathrm{Id}_{i(G_1)+\ldots+i(G_k)}\otimes c_{i(G'_1)+\ldots+i(G'_l),p}\otimes \mathrm{Id}_q)) \cdot\Phi(G_0*G'_0)\\
&=(\gamma\otimes \gamma')(\mathrm{Id}_{o(G_1)+\ldots+o(G_k)}\otimes c_{p,o(G'_1)+\ldots+o(G'_l)}\otimes \mathrm{Id}_q)\\
&\cdot (\phi(G_1)*\ldots*\phi(G_k)*\phi(G'_1)*\ldots*\phi(G'_l)*I_p*I_q)\\
&\circ (\mathrm{Id}_{i(G_1)+\ldots+i(G_k)}\otimes c_{i(G'_1)+\ldots+i(G'_l),p}\otimes \mathrm{Id}_q)) \cdot(\Phi(G_0)*\Phi(G'_0))\\
&=(\gamma\otimes \gamma')\cdot
(\phi(G_1)*\ldots*\phi(G_k)*I_p*\phi(G'_1)*\ldots*\phi(G'_l)*I_q)\circ(\Phi(G_0)*\Phi(G'_0))\\
&=(\gamma\cdot(\phi(G_1)*\ldots*\phi(G_k)*I_p)\circ \Phi(G_0))
*(\gamma' \cdot(\phi(G'_1)*\ldots*\phi(G'_l)*I_q)\circ(\Phi(G'_0))\\
&=\Phi(G)*\Phi(G').
\end{align*}
So for any graphs $G$ and $G'$ we have $\Phi(G*G')=\Phi(G)*\Phi(G')$. \\

Now once again, let $G$, $G'$ be two graphs. Let us prove that $\Phi(G'\circ G)=\Phi(G')\circ \Phi(G)$. 
We proceed by induction on $n=|V(G)|+|V(G')|$. If $V(G')=\emptyset$,
there exists \cy{$p\in\N$ and} a permutation $\sigma\in \sym_p$ such that $G'=\sigma\cdot I_\cy{p}$. Then:
\begin{align*}
\Phi(G'\circ G)&=\Phi(\sigma\cdot G)=\sigma\cdot \Phi(G)=\sigma \cdot (I_p\circ \Phi(G))=(\sigma\cdot I_p)\circ \Phi(G)
=\Phi(G')\circ \Phi(G).
\end{align*}
Similarly, if $V(G)=\emptyset$, $\Phi(G'\circ G)=\Phi(G')\circ \Phi(G)$. Thus we have proved the \ty{case} $n=0$.

Let us assume that it holds up to rank $N$, and take $G$ and $G'$ such that $n=N+1$. By the previous argument, 
$\Phi(G\circ G')=\Phi(G)\circ \Phi(G')$ if $V(G)=\emptyset$ or $V(G')=\emptyset$.
We now assume that $V(G)$ and $V(G')$ are nonempty. \cy{Let us work out the case where both $G$ and $G'$ are indecomposable. Then $\Phi(G)=\phi(G)$ and $\Phi(G')=\phi(G')$. Set $H:=G'\circ G$. This is a minimal decomposition (Equation \eqref{eq:min-dec}) of $H$ with $\gamma$ being the identity, $k=1$, $p=0$, $H_1=G'$ and $H_0=G$. Then we have by definition of $\Phi$
\begin{equation*}
 \Phi(H) = \phi(G')\circ\Phi(G)=\Phi(G')\circ\Phi(G).
\end{equation*}
For the general case, let}
us consider minimal decompositions of $G$ and $G'$:
\begin{align*}
G&=\gamma \cdot(G_1*\ldots*G_k*I_p)\circ G_0&
G'&=\gamma' \cdot  (G'_1*\ldots*G'_l*I_q)\circ G'_0.
\end{align*}
In $G'\circ G$, the output edges of $G$ are glued with an input or an input-output edge of $G'$.
In particular, for any $i$,  output edges of $G_i$ are glued with input edges or input-output edges of $G'$.
Up to a change of indexation \cy{of the $G_i$s} we assume that there is some $r$ such that:
\begin{itemize}
\item For all $i\leq r$, at least one output edge of $G_i$ is glued with an input edge of $G'$.
\item If $i>r$, all output edges of $G_i$ are glued with input-output edges of $G'$. 
\end{itemize}
\cy{Let us compute the minimal decomposition of $H:=G'\circ G$. 
% We set to simplify notations: 
% \begin{equation*}
%  a=G_1*\cdots *G_r,\qquad b=G_{r+1}*\cdots *G_k,\qquad c=G'_1*\cdots *G'_l.
% \end{equation*}
% Then using the commutativity of $*$ we can write $G=\gamma\cdot(a*I_p*b)\circ G_0$
}

\textit{A particular sub-case}. We assume that the input-output edges of $G'$ glued with an output of one of the 
$G_i$ are the input edges of $G'$ with the greatest indices. 
\cy{Then $G'_0=G''_0*I_{o(G_{r+1})+\cdots+o(G_k)}$ (and $q\geq o(G_{r+1})+\cdots+o(G_k)$) and $\gamma$ decomposes into two permutations
\begin{equation*}
 \gamma=\gamma_1\otimes\gamma_2, \qquad \gamma_1\in\sym_{o(G_{1})+\ldots+o(G_r)+p},\qquad\gamma_2\in\sym_{o(G_{r+1})+\cdots+o(G_k)}.
\end{equation*}
Let us write a minimal decomposition of $H$
\begin{equation} \label{eq:min_dec_composition}
 H=\gamma''\cdot(H_1*\cdots *H_m*I_N)\circ H_0.
\end{equation}
Then by minimality of the $H_i$, the hypothesis of the sub-case directly gives
\begin{equation*}
 (H_1,\cdots,H_m)=(G'_1,\cdots,G'_l,G_{r+1},\cdots,G_k).
\end{equation*}
Indeed, for $i\leq r$, there exists a path from $v\in V(G_i)$ to $v'\in V(G'_0)\cup V(G'_j)$, so $G_i$ is not minimal (Lemma \ref{lem:paths_subset}). This implies $N=q-o(G_{r+1})-\cdots-o(G_k)$ which is indeed positive.

Now, by the hypothesis of the subcase, the outputs indices of $\gamma_2\cdot(G_{r+1}*\cdots *G_k)$ are the highest. Therefore this term can go through $G''_0*I_{o(G_{r+1})+\cdots+o(G_k)}$ and we obtain
\begin{align*}
 H= \gamma'\circ(\Id_{o(G'_1)+\cdots+o(G'_l)+N} & \otimes\gamma_2) \cdot(G'_1*\cdots *G'_l*I_N*G_{r+1}*\cdots *G_k)\circ(G''_0*I_{i(G_{r+1}+\cdots+i(G_k)}) \\
 & \circ(\gamma_1\otimes\Id_{i(G_{r+1}+\cdots+i(G_k)})\cdot(G_1*\cdots *G_r*I_p*I_{i(G_{r+1}+\cdots+i(G_k)})\circ G_0.
\end{align*}
One can straightforwardly check that this expression is well-defined. We have almost a minimal decomposition for $H$, we simply need to exchange $I_N$ and $G_{r+1}*\cdots *G_k$. Using the commutativity of $*$ we have
\begin{equation*}
 I_N*G_{r+1}*\cdots *G_k=c_{o(G_{r+1})+\cdots+o(G_k),N}\cdot(G_{r+1}*\cdots *G_k*I_N)\cdot c_{N,G_{r+1}+\cdots+i(G_k)}.
\end{equation*}
Putting the $c_{o(G_{r+1})+\cdots+o(G_k),N}$ out on the left and $c_{N,i(G_{r+1})+\cdots+i(G_k)}$ out on the right we obtain that the minimal decomposition \eqref{eq:min_dec_composition} is given by
\begin{align*}
 \gamma'' = \gamma'\circ(\Id_{o(G'_1)+\cdots+o(G'_l)+N} \otimes\gamma_2) & \circ(\Id_{o(G'_1)+\cdots+o(G'_l)}\otimes c_{o(G_{r+1})+\cdots+o(G_k),N}), \\
 (H_1,\cdots,H_m) & =(G'_1,\cdots,G'_l,G_{r+1},\cdots,G_k), \\
 N =q & -o(G_{r+1})-\cdots-o(G_k), \\
 H_0 = (\Id_{i(G'_1)+\cdots+i(G'_l)} & \otimes c_{N,i(G_{r+1})+\cdots+i(G_k)})\cdot(G''_0*I_{i(G_{r+1}+\cdots+i(G_k)}) \\
 \circ(\gamma_1\otimes & \Id_{i(G_{r+1}+\cdots+i(G_k)})\cdot(G_1*\cdots *G_r*I_{p+i(G_{r+1}+\cdots+i(G_k)})\circ G_0.
\end{align*}
}
% \cy{Since the $G_{r+1},\cdots,G_k$ are plugged only on in-out edges of $G'$ and since those are the inputs of $G'$ with the highest indices, we can write $G'_0=G_0''*I_{o(b)}$.
% 
% By the hypothesis of the sub-case, it exist $\gamma_1\in\sym_{o(a)+p}$ and $\gamma_2\in\sym_{o(b)}$ such that $\gamma=\gamma_1\otimes\gamma_2$. Then we can write
% \begin{equation*}
%  H=\big[\gamma'\cdot(c*I_q)\circ(G_0''*I_{o(b)})\big]\circ\big[(\gamma_1\otimes\gamma_2)\cdot(a*b*I_p)\circ G_0\big]
% \end{equation*}
% }
% Then $G'_0=G''_0*I_{s+o(G_{r+1})+\ldots+o(G_k)}$
% for a certain $s$. Moreover, $\gamma$ can be written as $\gamma=\gamma_1\otimes \gamma_2$, such that a minimal
% decomposition of $H=G'\circ G$ is given by:
% \begin{align*}
% H_0&=(\mathrm{Id}_{i(G'_1)+\ldots+i(G'_l)}\otimes c_{i(G_{r+1})+\ldots+i(G_k)+p,s})\cdot G'_0\\
% &\circ \ty{\gamma_1}\cdot (G_1*\ldots*G_r *I_{i(G_{r+1})+\ldots+i(G_k)+p})\cy{\circ} G_0,\\
% (H_1,\ldots,H_m)&=(G'_1,\ldots,G'_l,G_{r+1},\ldots,G_k),\\
% \gamma''&=\gamma'(\mathrm{Id}_{o(G'_1)+\ldots+o(G'_l)}\otimes c_{s,o(G_{r+1})+\ldots+o(G_k)+p})
% (\mathrm{Id}_{o(G'_1)+\ldots+o(G'_l)+s}\otimes \gamma_2).
% \end{align*}
\cy{Since $l\geq1$ and the $G_i'$s are indecomposable, we have $|V(H_0)|=|V(G_1)|+\ldots+|V(G_r)|+|V(G_0)|+|V(G_0')|<n$, we can apply} the induction hypothesis on \cy{$H_0$. Then we can undo all each of the manipulations we performed to obtain the minimal decomposition of $H_0$.}
\begin{align*}
 \Phi(H) & = \cy{\gamma'\circ(\Id_{o(G'_1)+\cdots+o(G'_l)+N} \otimes\gamma_2) \circ(\Id_{o(G'_1)+\cdots+o(G'_l)}\otimes c_{o(G_{r+1})+\cdots+o(G_k),N})} \\
% \Phi(H)&=\gamma'(\mathrm{Id}_{o(G'_1)+\ldots+o(G'_l)}\otimes c_{s,o(G_{r+1})+\ldots+o(G_k)+p})
% (\mathrm{Id}_{o(G'_1)+\ldots+o(G'_l)+s}\otimes \gamma_2)\\
&\cdot((\phi(G'_1)*\ldots *\phi(G'_l)*\phi(G_{r+1}\ty{)}*\ldots*\phi(G_k))\\
& \cy{\circ\Phi\Big((\Id_{i(G'_1)+\cdots+i(G'_l)}  \otimes c_{N,i(G_{r+1})+\cdots+i(G_k)})\cdot(G''_0*I_{i(G_{r+1}+\cdots+i(G_k)})} \\
 & \cy{\circ(\gamma_1\otimes  \Id_{i(G_{r+1}+\cdots+i(G_k)})\cdot(G_1*\cdots *G_r*I_{p+i(G_{r+1}+\cdots+i(G_k)})\circ G_0\Big)} \quad\cy{\text{by definition of }\Phi}\\
% &\circ \Phi((\mathrm{Id}_{i(G'_1)+\ldots+i(G'_l)}\otimes c_{i(G_{r+1})+\ldots+i(G_k)+p,s})\cdot G'_0\\
% &\circ (\gamma_1\cdot (\cy{G_1*\ldots*G_r }*I_{i(G_{r+1})+\ldots+i(G_k)+p})\cdot G_0)
% \quad\cy{\text{by definition of }\Phi}\\
& \cy{=\gamma'\circ(\Id_{o(G'_1)+\cdots+o(G'_l)+N} \otimes\gamma_2) \circ(\Id_{o(G'_1)+\cdots+o(G'_l)}\otimes c_{o(G_{r+1})+\cdots+o(G_k),N})} \\
% &=\gamma'(\mathrm{Id}_{o(G'_1)+\ldots+o(G'_l)}\otimes c_{s,o(G_{r+1})+\ldots+o(G_k)+p})
% (\mathrm{Id}_{o(G'_1)+\ldots+o(G'_l)+s}\otimes \gamma_2)\\
&\cdot((\phi(G'_1)*\ldots *\phi(G'_l)*\phi(G_{r+1}\ty{)}*\ldots*\phi(G_k))\\
&\cy{\circ(\Id_{i(G'_1)+\cdots+i(G'_l)}  \otimes c_{N,i(G_{r+1})+\cdots+i(G_k)})\cdot(\Phi(G''_0)*I_{i(G_{r+1}+\cdots+i(G_k)})} \\
& \cy{\circ(\gamma_1\otimes  \Id_{i(G_{r+1}+\cdots+i(G_k)})}\cdot (\phi(G_1)*\cdots*\phi(G_r)*I_{p+i(G_{r+1})+\cdots+i(G_k)}\cy{\circ}\Phi(G_0) \quad\cy{\text{by the ind. hyp.}}\\
% 
% 
% 
% &\circ (\mathrm{Id}_{i(G'_1)+\ldots+i(G'_l)}\otimes c_{i(G_{r+1})+\ldots+i(G_k)+p,s})\cdot \Phi(G'_0)\\
% &\circ (\gamma_1\cdot (\phi(G_1)*\ldots*\phi(G_r) *I_{i(G_{r+1})+\ldots+i(G_k)+p})\cdot \Phi(G_0)) \quad\cy{\text{by the induction hypothesis }}\\
&=(\gamma\cdot (\phi(G_1)*\ldots*\phi(G_k)*I_p)\circ\Phi(G_0))
\circ(\gamma' \cdot (\phi(G'_1)*\ldots*\phi(G'_l)*I_q)\circ\Phi(G'_0))\\
&=\Phi(G)\circ \Phi(G').
\end{align*}
\cy{For the second equality, we have also used that the $G_i$s are indecomposable and thus that $\Phi(G_i)=\phi(G_i)$ by definition of $\Phi$. For the third equality we make the same steps than those giving the minimal decomposition of $G'\circ G$.} \\

\textit{General case}. There exists a permutation $\sigma$, such that if $H'=G'\cdot \sigma^{-1}$
and $H=\sigma\cdot G$, then the condition of the particular sub-case holds for $(H,H')$. Then:
\begin{align*}
\Phi(G'\circ G)&=\Phi((G'\cdot \sigma^{-1}\sigma)\circ G)\\
&=\Phi((G'\cdot \sigma^{-1})\circ (\sigma\cdot G))\\
&=\Phi(G'\cdot \sigma^{-1})\circ\Phi(\sigma\cdot G)\quad\text{since the subcase holds}\\
&=(\Phi(G')\cdot \sigma^{-1})\circ(\sigma\cdot\Phi(G))\\
&=(\Phi(G')\cdot \cy{(\sigma^{-1}\circ\sigma))\circ}\Phi(G)\\
&=\Phi(G')\cy{\circ} \Phi(G).
\end{align*}
Finally, for any graphs $G$ and $G'$, $\Phi(G\circ G')=\Phi(G)\circ \Phi(G')$ and $\Phi$ is a morphism of PROPs.
\end{proof}

\section{The category of TRAPs} \label{sec:TRAP}

TRAPs (TRAces and Permutations) are objects introduced in \cite{ClFoPa20} in order to sove the issues with PROPs discussed. Their key feature is that they contain ``trace maps'' that allow to build closed loops. We start by describing the category of TRAPs.

\subsection{Definition}

There exists a categorical definition of TRAPs, but this definition requires the introduction of more sophisticated concept. Indeed, TRAPs cannot (to the best of the author's knowledge) be described as a category, but instead as an algebra over a monad. Notice that this is true for PROPs as well. However, this definition is not very practical for applications. It was shown in \cite{ClFoPa21} that the category of unital TRAPs is isomorphic to the category of wheeled PROPs as introduced in \cite{Merkulov2006} (see also \cite{JY15}). 

I will not include this result here and focus instead on non unitary TRAPs. Therefore, I will state only the pedestrian definition of TRAPs. This pedestrian definition largely overlaps with the Definition \ref{def:prop} of PROPs. The main difference is that the vertical concatenations (or compositions) are replaced with partial \ty{trace} maps. We will show later (in Definition-Proposition \ref{propverticalconcatenation}) that one can build these compositions from the partial traces maps.
\begin{defn} \label{def:Trap}
A \textbf{TRAP} (TRAces and Permutations) is a family $P=(P(k,l))_{k,l\geqslant 0}$ of \cy{vector spaces}, equipped with the following structures:
\begin{enumerate}
\item $P$ is a $\sym\times \sym^{op}$-module.

\item $P$ admits an {\bf horizontal concatenation}, that is to say that for any $(k,l,k',l')$ in $ \N^4$
, there is a map
\begin{align*}
*:&\left\{\begin{array}{rcl}
P(k,l)\times P(k',l')&\longrightarrow&P(k+k',l+l')\\
(p, p')&\longmapsto&p*p',
\end{array}\right.
\end{align*}
which is associative, commutative, unital and equivariant under the action of $\sym\times\sym$.

\item For any $k,l\geqslant 1$, for any $i\in [k]$, $j\in [l]$, there is a map
\begin{equation}\label{eq:partialtrace}
t_{i,j}:\left\{\begin{array}{rcl}
P(k,l)&\longrightarrow&P(k-1,l-1) \\
p&\longmapsto&t_{i,j}(p),
\end{array}\right.
\end{equation}
called the \textbf{partial trace map}, such that:
\begin{enumerate}
\item (Commutativity). For any $k,l\geqslant 2$, for any $i\in [k]$, $j\in [l]$, $i'\in [k-1]$, $j'\in [l-1]$,
\begin{align*}
t_{i',j'}\circ t_{i,j}&=\begin{cases}
t_{i-1,j-1}\circ t_{i',j'}\mbox{ if }i'<i,\: j'<j,\\
t_{i,j-1}\circ t_{i'+1,j'}\mbox{ if }i'\geqslant i,\: j'<j,\\
t_{i-1,j}\circ t_{i',j'+1}\mbox{ if }i'<i,\: j'\geqslant j,\\
t_{i,j}\circ t_{i'+1,j'+1}\mbox{ if }i'\geqslant i,\: j'\geqslant j.
\end{cases}
\end{align*}
\item (Equivariance). For any $k,l\geqslant 1$, for any $i\in [k]$, $j\in [l]$,
$\sigma \in \sym_l$, $\tau\in \sym_k$, for any $p\in P(k,l)$,
\[t_{i,j}(\sigma\cdot p\cdot \tau)=l_j(\sigma)\cdot (t_{\tau(i),\sigma^{-1}(j)}(p))\cdot r_i(\tau),\]
with the following notation: if $\alpha\in \sym_n$ and $q\in [n]$, then 
$\left(l_q(\alpha),r_q(\alpha)\right)\in \sym^2_{n-1}$ are defined by
\begin{align*}
\label{defalphak}  
l_q(\alpha)(s)&=\begin{cases}
\alpha(s)\mbox{ if }s<\alpha^{-1}(q)\mbox{ and }\alpha(s)<q,\\
\alpha(s)-1\mbox{ if }s<\alpha^{-1}(q) \mbox{ and }\alpha(s)>q,\\
\alpha(s+1)\mbox{ if }s\geq \alpha^{-1}(q) \mbox{ and }\alpha(s+1)<q,\\
\alpha(s+1)-1\mbox{ if }s\geq \alpha^{-1}(q) \mbox{ and }\alpha(s+1)>q,
\end{cases}\\ \\
\nonumber  r_q(\alpha)(s)&=\begin{cases}
\alpha(s)\mbox{ if }s<q\mbox{ and }\alpha(s)<\alpha(q),\\
\alpha(s)-1\mbox{ if }s<q \mbox{ and }\alpha(s)>\alpha(q),\\
\alpha(s+1)\mbox{ if } s\geq q \mbox{ and }\alpha(s+1)<\alpha(q),\\
\alpha(s+1)-1\mbox{ if }s\geq q \mbox{ and }\alpha(s+1)>\alpha(q).
\end{cases}
\end{align*}
In other words, if we represent $\alpha$ by the word $\alpha(1)\ldots \alpha(n)$, then 
$l_q(\alpha)$ is represented by the word obtained by deleting the letter $q$ and subtracting $1$ to the letters
$>q$, whereas $r_q(\alpha)$ is represented by the word obtained by deleting the letter $\alpha(q)$ 
and substracting $1$ to the letters $>\alpha(q)$. Note that $r_q(\alpha)=l_{\alpha(q)}(\alpha)$.

\item (Compatibility with the horizontal concatenation). 
For any $k,l,k',l'\geq 1$, for any $i\in [k+l]$, $j\in [k'+l']$, for any $p\in P(k,l)$, $p'\in P(k',l')$:
\[t_{i,j}(p*p')=\begin{cases}
t_{i,j }(p)*p'\mbox{ if }i\leqslant k,\: j\leqslant l,\\
p* t_{i-k,j-l}(p')\mbox{ if }i>k,\: j>l.
\end{cases}\]
\end{enumerate}\end{enumerate}
\end{defn}
The first two points of this definition are precisely the points 1, 2 and 6 of Definition \ref{def:prop}. Some remarks regarding this definition are in order.
\begin{rk}
 The abuse of notation $t_{i,j}$ is legitimate since a full notation such as  $t_{i,j}^{k,l}$ is not necessary in practice. Indeed the indices $k$ and $l$ in $t_{i,j}(p)$ are entirely determined by  the element $p$ to which  $t_{i,j}$ is applied.
\end{rk} 
Quite a few can be said about the lack of unity in the above definition.
\begin{rk}  \label{rk:wheeledPROP}
 In \cite{ClFoPa21}, unital TRAPs were defined. Let $P$ be a TRAP. A unit of $P$ is an element $I\in P(1,1)$ such that for any $p\in P$, the image under the partial trace maps of $I*p$ and $p*I$ is $p$ up to the action of a simple \ty{permutation}. These axioms are actually equivalent to 
 \begin{equation} \label{eq:unit_TRAP}
  t_{1,2}(I*p)=p\qquad\cy{\text{and}\quad t_{2,1}(I*p)=p},
 \end{equation}
 for any $k,l\geqslant 1$ and any $p\in P(k,l)$.
 
 Here I do not introduce this axiom in more details for two reasons. Firstly, it would require later on more general graphs that the ones if Subsection \ref{subsec:PROP_graphs}. Secondly, the TRAPs appearing when dealing with infinite dimensional vector spaces are typically non unital. Actually, unital TRAPs are useful to relate TRAPs and other structure: PROPs and wheeled PROPs. I will state the relation between TRAPs and PROPs since PROPs were introduced in this thesis before, but for the link between unital TRAPs and wheeled PROPs I will simply refer the reader to \cite[Theorem 5.3.1]{ClFoPa21}.
 
 Notice that one could in principle add a formal unit $I$ to any TRAP, but then one also needs to describe $t_{1,1}(I)$ and concatenation with this element. This make the TRAP much more challenging to describe, as can be seen in \cite{ClFoPa21} where the unital version of the TRAP of graphs to be described below requires more complicated objects than the generalised graphs described in Subsection \ref{subsec:PROP_graphs}. Another issue with adding a formal unit is that in some spaces relevant to the theory and application of TRAPs, we do \emph{have} units, just that the partial traces are not defined on them. This happens for example for spaces of operators. The identity is the usual identity operator, but the trace of the identity is not defined in general. So in this case, the identity cannot be an element of the TRAP. Adding another, formal identity seems rather artficial.
%  {subsection:Frechet_nuc} XXXXX
\end{rk}
\begin{rk}
 One could also point out that we gave only a pedestrian definition of TRAPs, and that a concise definition in the fashion of \ref{defn:PROP_cat} is still lacking. However, it was recently pointed out to me that wheeled PROPs are a special case of traced monoidal categories \cite{joyal1996traced}. Since TRAPs are essentially wheeled PROPs, they can probably be consisely defined via traced monoidal categories, at least in their unital version. Since should be explored further in future work. I am thankful to Vladimir Dotsenko for pointing out traced monoidal categories to me.
\end{rk}
Finally, one should notice that TRAPs and PROPs are related concepts.
\begin{rk}
 It should be clear from this definition that one can define a vertical concatenation, or composition, by iterating the partial traces. This is postponed to Subsection \ref{subsec:vert_conc_TRAP}, where we will show that one can build a PROP from a unital (see above) TRAP: see Proposition \ref{propverticalconcatenation} and Remark \ref{rk:TRAP_PROP}.
\end{rk}
Before introducing morphisms of TRAPs, here is one simple but useful Lemma that allows to simplify the axioms of TRAPs.

\begin{lemma}\label{lemmeaxiomessimples}
Let $P=(P(k,l))_{k,l\geq0}$ be a $\sym\times \sym^{op}$-module, equipped with a horizontal concatenation $*$
which is associative, commutative, unital and equivariant under the action of $\sym\times\sym$, and with maps $t_{i,j}$ satisfying axioms 3. (a)-(b) of Definition \ref{def:Trap}. 

Assuming that for any $k,l,k',l'\geqslant 1$, for any $p\in P(k,l)$, $p'\in P(k',l')$,
\[t_{1,1}(p*p')=t_{1,1}(p)*p',\]
then Axiom 3.(c) of Definition \ref{def:Trap} is satisfied.
\end{lemma}

\begin{proof} 
  Let $p\in P(k,l)$ and $p'\in P(k',l')$. We take $i$ in $ [k+l]$, $j$ in $[k'+l']$ and define the transpositions $\sigma=(1,j)$, $\tau=(1,i)$, with the convention $(1,1)=\mathrm{Id}$. We use the notation $\sigma_j:=l_j(\sigma)$ and $\tau_i:=r_i(\tau)$. Let us  consider several cases.
  \begin{itemize}
  	\item  
If $i\leqslant k$ and $j\leqslant l$,  then:
\begin{align*}
t_{i,j}(p*p')&=t_{i,j}(\sigma^2\cdot( p*p')\cdot \tau^2)\\
&=\sigma_j\cdot t_{1,1}(\sigma\cdot (p*p')\cdot \tau)\cdot \tau_i\quad\text{by equivariance of $t_{i,j}$} \\
&=\sigma_j\cdot(t_{1,1}((\sigma\cdot p\cdot \tau)*p'))\cdot \tau_i \quad\text{since $i\leqslant k$, $j\leqslant l$ and by equivariance of $*$} \\
&=\sigma_j \cdot (t_{1,1}(\sigma\cdot p\cdot \tau)*p')\cdot \tau_i\quad\text{by the assumption of the Lemma,}\\
&=(\sigma_j\cdot (t_{1,1}(\sigma\cdot p\cdot \tau)\cdot \tau_i))*p'\quad\text{since $i\leqslant k$, $j\leqslant l$ and by equivariance of $*$} \\
&=t_{i,j}(p)*p'\quad\text{by equivariance of $t_{i,j}$.}
\end{align*}
\item If $i>k$ and $j>l$, using $c_{m,n}^{-1}=c_{n,m}$, and as before writing  $(c_{l',l})_j:=l_j(c_{l',l})$ and $(c_{k,k'})_i:=r_i(c_{k,k'})$ we have
\begin{align*}
t_{i,j}(p*p')&=t_{i,j}(c_{l',l}\cdot (p'*p)\cdot c_{k,k'})\quad\text{by commutativity of $*$} \\
&=(c_{l',l})_j\cdot t_{i-k,j-l}(p'*p)\cdot (c_{k,k'})_i\quad\text{by equivariance of $t_{i,j}$} \\
&=c_{l'-1,l}\cdot (t_{i-k,j-l}(p')*p)\cdot c_{k,k'-1}\quad\text{By the first point,}\\
&=p*t_{i-k,j-l}(p')\quad\text{by commutativity of $*$}.
\end{align*}
\end{itemize}
Thus Axiom $3.(c)$ is satisfied.
\end{proof}

\subsection{Morphisms of TRAPs}

The definition of morphisms of TRAPs, though somewhat complicated, is quite natural.

\begin{defn}       \label{defn:trap_morphism}
     Let $P=(P(k,l))_{k,l\geq0}$ and $Q=(Q(k,l))_{k,l\geq0}$ be two TRAPs with partial trace maps $(t_{i,j}^P)_{i,j\geq0}$ and $(t_{i,j}^Q)_{i,j\geq0}$ 
     respectively. A \textbf{morphism of TRAPs} is a family $\phi=(\phi(k,l))_{k,l\geq0}$ 
     of morphisms of $\sym\times \sym^{op}$-modules  $\phi(k,l):P(k,l)\longrightarrow Q(k,l)$ which are compatible with the horizontal concatenation, 
     and the partial trace maps. More precisely, for any $(k,l,m,n)\in \N^4$:
     \begin{enumerate}
     \item For any $(\sigma,p,\tau)$ in $\sym_l\times P(k,l)\times\sym_k$, $\phi(k,l)(\sigma\cdot p\cdot \tau)
     =\sigma\cdot \phi(k,l)(p)\cdot\tau$.
     \item $\phi(0,0)(I_0)=I_0$.
      \item $\forall (p,q)\in P(k,l)\times P(n,m),~\phi(k+n,l+m)(p* q) = \phi(k,l)(p)* \phi(n,m)(q)$,
      \item $\forall (p,i,j)\in P(k,l)\times [k]\times [l]$, $\phi(k-1,l-1)\circ t^P_{i,j}(p)=t^Q_{i,j}\circ \phi(k,l)(p)$.
     \end{enumerate} 
     With a slight  abuse of notation, we write $\phi(p)$ instead of $\phi(k,l)(p)$ for $p\in P(k,l)$.
     We denote by $\Trap$ the category of TRAPs and TRAPs morphisms.
  \end{defn}
In the same way that Lemma \ref{lemmeaxiomessimples} simplifies the axioms of a TRAP, we can also simplify the axioms for morphisms of TRAPs.
\begin{lemma}\label{lemmemorphismes}
Let $P=(P(k,l))_{k,l\geq0}$ and $Q=(Q(k,l))_{k,l\geq0}$ be two TRAPs and $\phi=( \phi(k,l))_{k,l\geq0}$ be a family of set maps $\phi(k,l):P(k,l)\longrightarrow Q(k,l)$ satisfying Points 1-3 of Definition \ref{defn:trap_morphism}. Suppose further that for any $k,l\geqslant 1$, for any $p\in P(k,l)$
\[t_{1,1}\circ \phi(p)=\phi\circ t_{1,1}(p).\]
Then $\phi$ is a map of TRAPs.
\end{lemma}

\begin{proof}If $i$,  $j$ and $p$ lies respectively in $[k]$, $[l]$, and $ P(k,l)$, then
\begin{align*}
\phi\circ t_{i,j}(p)&=\phi\circ t_{i,j}((1,j)^2\cdot p\cdot (1,i)^2)\\
&=\phi\left(\cy{l_j(1,j)}\cdot t_{1,1}((1,j)\cdot p\cdot (1,i))\cdot \cy{r_i(1,i)}\right)\quad\text{by equivariance of $t_{i,j}$ (Axiom $3b$\ty{)}} \\
&=\cy{l_j(1,j)}\cdot  \phi\circ t_{1,1}((1,j)\cdot p\cdot (1,i))\cdot \cy{r_i(1,i)}\quad\text{by Point 1 of Definition \ref{defn:trap_morphism},}\\
&=\cy{l_j(1,j)}\cdot  t_{1,1}\circ \phi ((1,j)\cdot p\cdot (1,i))\cdot \cy{r_i(1,i)}\quad\text{by the assumption of the Lemma,}\\
&=t_{i,j}((1,j)\cdot \phi((1,j)\cdot p\cdot (1,i))\cdot (1,i))\quad\text{by equivariance of $t_{i,j}$ } \\
&=t_{i,j}\circ \phi(p)\quad\text{by Point 1 of Definition \ref{defn:trap_morphism},},
\end{align*}
with the convention $(1,1)=\mathrm{Id}$. It follows that $\phi$ is a morphism of TRAPs. \end{proof}

In particular, to show that a collection of linear maps between two TRAPs preserving 
the horizontal concatenation and the actions of the symmetry group is a morphism of TRAPs, it is enough to check 
the properties of Lemma \ref{lemmemorphismes}.
  
Finally, let us point out that the category $\Trap$ admits a natural notion of subobjects (see \cite[Chapter V.7]{ML71}).
  \begin{defn} \label{defn:subtrap}
  We define a {\bf sub-TRAP} of a TRAP   $P=(P(k,l))_{k,l\geqslant 0}$ to be a  $\sym\times \sym^{op}$-submodule   $Q=(Q(k,l))_{k,l\geqslant 0}$ of $P$ which contains the unit $I_0\in P(0,0)$ and is stable under the partial trace map of $P$. 
\end{defn}

\section{Examples of TRAPs} \label{sec:TRAPs_ex}

We give here three examples of TRAPs that are relevant for analysis in infinite dimensions, as a first step toward QFT.

\subsection{The TRAP of linear morphisms} \label{subsec:Hom_TRAP}

The most basic example of TRAP is on the same space than the most basic example of PROP which was introduced in Subsection \ref{subsec:PROP_lin_morph}. Let us recall that for a finite dimensional vector space $V$ and its algebraic dual $V^*$ we consider the family \[ \Hom_V=\left(\Hom_V(k, l)\right)_{k,l\geq0}:= { (\Hom(V^{\otimes k},V^{\otimes l}))_{k,l\geq0}},\] where for any $(k,l)\in  \N^2$,  $\Hom(V^{\otimes k},V^{\otimes l})$ is  the vector space of linear maps from $V^{\otimes k}$ to $V^{\otimes l}$.

We shall identify $\Hom_V(k,l)$  and $  {V^{*\otimes k}\otimes V^{\otimes l}}$ through the isomorphism
	\[\theta_{k,l}:\left\{\begin{array}{rcl}
	{V^{*\otimes k}\otimes V^{\otimes l}}&\longrightarrow&\Hom_V(k,l)\\
	 { f_1\ldots f_k\otimes v_1\ldots v_l}&\longmapsto&\left\{\begin{array}{rcl}
	V^{\otimes k}&\longrightarrow&V^{\otimes l}\\
	x_1\ldots x_k&\longmapsto&f_1(x_1)\ldots f_k(x_k)\,v_1\ldots v_l,
	\end{array}\right.
	\end{array}\right.\]
 	where with  some abuse of notation, we have set $f_1\cdots f_k:= f_1\otimes\cdots \otimes f_k\in V^{*\otimes k}$ and 
	$v_1\cdots v_l:= v_1\otimes \cdots \otimes v_l\in V^{\otimes l}$.
	
	As in Subsection \ref{subsec:PROP_lin_morph} $\sym\times\sym^{\rm op}$-module structure is defined first from the fact that for any vector space $W$, the tensor power $W^{\otimes k}$ is a left $\sym_k$-module
	with the action defined by
	\[\sigma \cdot w_1\ldots w_k=w_{\sigma^{-1}(1)}\ldots w_{\sigma^{-1}(k)}.\]
	Furthermore, via the identification $\theta:=\left(\theta_{k, l}\right)_{k,l\geq0}$,  we can equip the family 
	\begin{equation*}
	 \Hom_V=(\Hom(V^{\otimes k},V^{\otimes l}))_{k,l\geq0}
	\end{equation*}
with the structure
	of a $\sym_l\times \sym_k^{op}$-module by putting, for an $f\in \Hom(V^{\otimes k},V^{\otimes l})$,
	for any $(\sigma,\tau)\in \sym_k\times \sym_l$:
	\begin{equation} \label{EQaction}
	\forall v_1\ldots v_k\in V^{\otimes l},
	\tau\cdot f\cdot \sigma(v_1\ldots v_k)=\tau\cdot f(\sigma\cdot v_1\ldots v_k). 
	\end{equation}
	
	The horizontal concatenation is the same as the one introduced in Subsection \ref{subsec:PROP_lin_morph}, i.e. the usual tensor product of linear maps:
	if $f\in \Hom_V(k,l)$ and $g\in \Hom_V(k',l')$, then
	\[f\otimes g:\left\{\begin{array}{rcl}
	V^{\otimes (k+k')}&\longrightarrow&V^{\otimes (l+l')}\\
	v_1\ldots v_{k+k'}&\longmapsto&f(v_1\ldots v_k)\otimes g(v_{k+1}\ldots v_{k+k'}).
	\end{array}\right.\]
	And as before its unit is the unit of the underlying field $I_0:=1$.
	
	Finally we define the partial trace maps as
	\begin{equation} \label{EQtrace} 
	 t_{i,j}(\theta_{k,l}(f_1\ldots f_k\otimes v_1\ldots v_l))=f_i(v_j)\,\theta_{k-1,l-1}(f_1\ldots f_{i-1}f_{i+1}\ldots f_k \otimes v_1\ldots v_{j-1}v_{j+1}\ldots v_l).
	\end{equation}
\begin{rk}
 Notice that for $p\in\Hom_V(1,1)$, $t_{1,1}(p)$ coincides with the usual trace of a linear map on $V$.  Proposition \ref{prop:generalised_traces} below generalises the notion of trace to any element of $\Hom_V$ and to any TRAP.
\end{rk}
\begin{prop}\label{prop:HomV} 
	For a finite dimensional $\K$-vector space $V$, the above construction equips $\Hom_V$ with  a TRAP structure.
\end{prop}
\begin{rk}
 $\Hom_V$ is actually a unitary TRAP, with unit given by the identity map $Id_V\in\Hom_V(1,1)$. This actually shows that $\Hom_V$ is a wheeled PROP, a fact that was used in the context of invariant theory (see \cite{DM19}) and algebraic topology (see \cite{KV21}). 
\end{rk}
\begin{proof}
	Properties 1. and 2. are trivially satisfied, with $I_0=1\in \K$.
% 	Property 3.(a) is direct.  Let us prove Property 3.(b). \cy{First, one checks through direct computations that $\sigma\cdot\theta_{k,l}(f_1\ldots f_k\otimes v_1\ldots v_l) \cdot \tau$. Then}
	\begin{align*}
	t_{i,j}(\sigma\cdot\theta_{k,l}(f_1\ldots f_k\otimes v_1\ldots v_l) \cdot \tau)
	&=t_{i,j}\circ \theta_{k,l}(f_{\tau(1)}\ldots f_{\tau(k)}\otimes v_{\sigma^{-1}(1)}\ldots v_{\sigma^{-1}(l)})\\
	&=f_{\tau(i)}(v_{\sigma^{-1}(j)})
	\theta_{k-1,l-1}(f_{\tau(1)}\ldots f_{\tau(i-1)} f_{\tau(i+1)}\ldots f_{\tau(k)}\\
	&\otimes v_{\sigma^{-1}(1)}\ldots  v_{\sigma^{-1}(j-1)} v_{\sigma^{-1}(j+1)}\ldots v_{\sigma^{-1}(l)})\\
	&=\sigma_j\cdot t_{\tau(i),\sigma^{-1}(j)}\theta_{k,l}(f_1\ldots f_k\otimes v_1\ldots v_l)\cdot \tau_i.
	\end{align*}
	Property 3.(c) is straightforward from the definition and Lemma \ref{lemmeaxiomessimples}. Thus $\Hom_V$ is a TRAP.
\end{proof}

When $V$ is not finite-dimensional, $\theta$ is an injective, non surjective map. Its range is the subspace $\Hom_V^{fr}$ of linear
maps from $V^{\otimes k}$ to $V^{\otimes l}$ of finite rank. We can equip $\Hom_V^{fr}$ with a similar TRAP structure:
\begin{prop}\label{prop:Homfr}
	With the $\sym\times \sym^{op}$ action defined  by (\ref{EQaction}), the usual tensor product of maps
	and the partial trace maps defined by (\ref{EQtrace}), $\Hom_V^{fr}$ is a TRAP.
\end{prop}
\begin{rk}
 The TRAP $\Hom_V^{fr}$ is unitary if, and only if, $V$ is finite-dimensional. This justifies the previous statement that non unitary TRAPs are the objects adapted to analysis on infinite dimensional spaces.
 
 Let us show that if $\Hom_V^{fr}$ has a unit then $V$ is finite dimensional\ty{. Assume $I$ is the} unit of $P$, i.e. Equation \eqref{eq:unit_TRAP} holds. Then $I$ has finite rank, let us fix a basis $(e_1,\ldots,e_k)$ of $\mathrm{Im}({I})$. There exist  $\lambda_1,\ldots,\lambda_k\in V^*$ such that for any $v\in V$,
	\[I(v)=\sum_{i=1}^k \lambda_i(v)e_i.\]
	In other words, 
		\[I=\theta_{1,1}\left({\sum_{i=1}^k e_i\otimes \lambda_i}\right).\]
	Let $v\in V$, nonzero, and let $\lambda \in V^*$ such that $\lambda(v)=1$. We consider $f=\theta_{1,1}(v\otimes \lambda)$.
	Then $f(v)=\lambda(v)v=v$. Moreover:
	\begin{align*}
	v&=f(v)=t_{1,2}(I*f)(v)\\
	&=t_{1,2}\circ \theta_{2,2}\left(\sum_{i=1}^k {e_i}\,  v\otimes \lambda_i \lambda\right)(v)\\
	&=\theta_{1,1}\left(\sum_{i=1}^k { \lambda_i}{ (v)\, e_i} \otimes {\lambda}\right)(v)\\
	&=\sum_{i=1}^k \lambda(v)\lambda_i(v)\, {e_i} \\
	&=\sum_{i=1}^k \lambda_i(v)\, {e_i} ,
	\end{align*}
\end{rk}
 We end this paragraph with an example of a TRAP similar to the TRAP $\Hom_V$ but of a more geometric nature.
\begin{example} [The TRAP of tensors]
		Given a finite dimensional smooth  manifold $M$ and a point $x\in M$, we build the $\Hom$-TRAP $\Hom_{T_xM}$ where $T_xM$  is the tangent space to $M$ at the point $x$.  
		Given a pair $(p, q)\in \N^2$ we have, using the musical isomorphisms (see for example \cite[Chapter 3]{Lee})
		\[ \Hom_{T_xM}(p, q)\simeq (T_x^*M)^{\otimes p}\otimes T_xM^{\otimes q},\]
		where we have set   $V^{\otimes 0}=\R$. The partial trace maps   are built by pairing cotangent  and tangent vectors. 
		We note that if $M$ is equipped with a Riemannian metric, thanks to  the  musical isomorphisms  $T_x^*M\ni \alpha  \longmapsto \alpha^\sharp \in T_x M$ and  $T_xM\ni v\longmapsto v^\flat \in  T_x^*M$ between  $T_x^*M$ and $ T_xM$, these dual pairings can be seen as contractions via the metric tensor.
		
		This yields a smooth fibration  $ \Hom_{TM}:=\{ \Hom_{T_xM}, x\in M\} $ of TRAPs parametrised by $M$. For any $(p, q)\in \N^2$,  a smooth section of $\Hom_{TM}(p,q)$ defines a smooth $(p,q)$ tensor on $M$.   
\end{example}
This example suggests an interesting open question: could one define TRAPs on rings and modules? Provided one would need to such structures, it seems that at least some of the properties of TRAPs would still be valid on this more general setting.

\subsection{The TRAP of continuous morphisms}\label{subsection:Frechet_nuc} 

Still following the same plan as the sections presenting PROPs, we proceed to generalise the TRAP $\Hom_V$ to infinite dimensional vector spaces by replacing $V^{\otimes k}$ in $ \Hom_V$ by Grothendieck's nuclear spaces \cite{Gr52,Gr54} already introduced in Subsection \ref{subsec:PROP_nuc}. 

Recall that the topological dual $E'$ of  a locally convex topological vector space $E$  can be endowed with various topologies, one of which is 
	the \textbf{strong topology}, namely the topology of uniform convergence on the bounded domains of $E$. It is  generated by the family of semi-norms of $E'$ defined on any $f\in E'$ by $||f||_B:=\sup_{x\in B}|f(x)|$
	for any bounded set $B$ of $E$ (Equation \eqref{eq:strong_dual_topo}). The topological dual $E'$ endowed with this topology is called the \textbf{strong dual}.

In the following $E$ and $F$ are two topological vector spaces and $\Hom^c(E,F)$ is the set of continuous  linear maps from $E$ to $F$.

Let us recall Definition \ref{defi:Hom_V_generalised} since we use it once again here. For a  Fr\'echet nuclear space $V$ we set 

\begin{equation} \label{eq:iso_TRAP_cont_morph}
\Hom_V^c(k,l)=\Hom^c(V^{\hat \otimes k}, V^{\hat \otimes l})\simeq(V')^{\widehat\otimes  k}\widehat\otimes  V^{\widehat\otimes  l},
\end{equation}
for any $(k,l)\in\N^2$ and where $V'$ stands for the strong topological dual. 

Furthermore, for any $\sigma \in \sym_n$, we define $\theta_\sigma$ the endomorphism of $V^{\otimes n}$ defined by
\[\theta_\sigma(v_1\otimes \ldots \otimes v_n)=v_{\sigma^{-1}(1)}\otimes \ldots \otimes v_{\sigma^{-1}(n)}\] 
which extends to a continuous linear map 
$\overline{\theta_\sigma}$ on the closure 
$V^{\widehat\otimes  n}$.   
For any $f\in \Hom_V^c (k,l)$, $\sigma \in \sym_l$, $\tau\in \sym_k$, we had then set:
\begin{align*}
\sigma\cdot f&=\overline{\theta_\sigma} \circ f ,&f\cdot \tau&=f\circ \overline{\theta_\tau}.
\end{align*} 
On top of its PROP structure build and proved in Subsection \ref{subsec:PROP_nuc}, the family $\Hom_V^c$ carries a TRAP structure.
\begin{theo} \label{thm:Hom_V_generalised}
	Let $V$ be a Fr\'echet nuclear space. $\Hom_V^c$, with the action of $\sym\times\sym^\mathrm{op}$ described above, is a TRAP. Its horizontal 
	concatenation is the usual topological tensor product of maps with $I_0:\K\longrightarrow\K$ given by the 
identity map of $\K$, its partial trace  maps  coincide with those of  the TRAP $\Hom_V$ on elements of $(V')^{\widehat\otimes  k}\widehat\otimes  V^{\widehat\otimes  l}$
\begin{equation*}
 t_{i,j}(f_1\cdots f_k\otimes v_1\cdots v_l) = f_i(v_j) f_1\cdots f_{i-1}f_{i+1}\cdots f_k\otimes v_1\cdots v_{j-1}v_{j+1}\cdots v_l
\end{equation*}
with the same notations as in   Subsection \ref{subsec:Hom_TRAP}.
\end{theo}
\begin{rk}
 The TRAP $\Hom_V^c$ is unitary if, and only if, $V$ is finite dimensional. This is because the unit would be the identity of $V$, but $t_{1,1}(Id_V)=\Tr(V)$ is only defined when $V$ is finite dimensional. Here the issue is with the completion of the tensor product. The identity of $V$ is mapped by the isomorphism of Equation \ref{propverticalconcatenation} to an infinite series of tensor powers, and the partial trace map $t_{1,1}$ is not defined on this object. Once again we see that non unitary TRAP is the right structure when dealing with infinite dimensional spaces.
\end{rk}
\begin{proof}
	The proof of the TRAP structure of $\Hom^c_V$ goes as in Proposition \ref{prop:HomV}. The proof of the unital case is the same as the proof of Proposition \ref{prop:Homfr}.
\end{proof}
\begin{example}
	For a  finite dimensional vector space $V$, the TRAP $\Hom_V^c$  coincides with the  TRAP  $\Hom_V$ of Subsection \ref{subsec:Hom_TRAP}.
\end{example} 
\begin{example}
	Let $M$ be a smooth finite dimensional manifold.
	From Proposition \ref{prop:prod_function} and Equation \eqref{eq:echange_dual_prod}, it follows that
	the family  $({\mathcal K}_M(k, l))_{k,l\geq0}$,
	with ${\mathcal K}_M(k, l)= \left({\mathcal E}^\prime(M)\right)^{\widehat\otimes  k}\, \widehat\otimes  \,  {\mathcal E}(M)^{\widehat\otimes  l}$  
	defines a TRAP.
\end{example}

\subsection{The TRAP of smoothing pseudo-differential operators}\label{subsec:smoothingop}

We apply our results on TRAPs to  tensor products 
of  Fr\'echet spaces ${\mathcal E}(M)$  of smooth functions on a given smooth finite dimensional  orientable closed manifold $M$ and  $ \mu(z)$ a volume form on $M$. From now on, we work with vector space over $\C$. Recall from Proposition 
\ref{prop:prod_function}   that such spaces are stable under tensor products and morphisms in  $\mathrm{Hom}^c( {\mathcal E}^\prime(M), {\mathcal E}(N))$ are determined by    smoothing kernels in ${\mathcal E}(M\times N)$.

  We consider smooth kernels
which stabilise ${\mathcal E}(M)$ and set, for $(k,l)\neq(0,0)$:
\begin{equation}\label{eq:Klm}{\mathcal K}_{M}^\infty(k, l):={\mathcal E}({M}^k\times {M}^l)\simeq{\mathcal E}({M})^{\widehat\otimes  k}\, \hat \otimes \,  {\mathcal E}({M})^{\widehat\otimes  l},\end{equation}
whose elements we refer to as smooth generalised kernels. We also set ${\mathcal K}_M^{\infty}(0,0)\simeq\C\otimes\C$ and ${\mathcal K}_M^{\infty}:=\left({\mathcal K}_{M}^\infty(k, l)\right)_{k,l\geq0}$. \cy{Notice that ${\mathcal K}_M^{\infty}$ has the structure of a $\sym\times\sym^{\rm op}$-module from the construction of the previous subsection.}
\begin{theo}\label{theo:Kinfty}
	Let $M$ be a smooth finite dimensional orientable closed manifold. The family of topological vector spaces $\left({\mathcal K}_{M}^\infty(k, l)\right)_{k,l\geq0}$ can be equipped with the partial trace  maps $t_{i,j}:{\mathcal K}_{M}^\infty(k, l)\longrightarrow{\mathcal K}_{M}^\infty(k-1,l-1)$
	with $t_{i,j}(K_1\otimes K_2)$ defined by
	\begin{align}\label{eq:trconvK}
	t_{i,j}(K_1\otimes K_2) & (x_1, \cdots, x_{k-1}, y_1, \cdots,  y_{l-1}):= \\
	& \int_M K_1(x_1 ,  \cdots,x_{i-1},z,x_{i}\cdots, x_{k-1})\, K_2(y_1,\cdots,y_{j-1},z,y_{j}\cdots, y_{\cy{l-1}}) \,d\mu(z),\nonumber
	\end{align}
	where $d\mu(z)$ is a volume form on $M$.

Together with the  horizontal concatenation given by   the tensor product of maps $(K_1\otimes K_2)*(K'_1\otimes K'_2)=K_a\otimes K_b$ with $K_a:=K_1\otimes K'_1$ and $K_b:=K_2\otimes K'_2$
this	defines a TRAP, written $\mathcal{K}_{M}^\infty$,  which we call the TRAP of {\bf  generalised smooth kernels} on $M$. 
\end{theo}  
\begin{rk} 
 Note that the partial trace    amounts to what one could call a partial convolution.
\end{rk}
\begin{proof}
 The unit $I_0\in\mathcal{K}_{M}^\infty(0,0)\simeq \C\otimes\C$ of the horizontal concatenation $*$ is the constant map $f:\C\to \C$ defined by $f(x)=1$. It is the unit of $*$ by bilinearity of the tensor product. The horizontal concatenation on the $\sym\times \sym^{op}$-module   $\left({\mathcal K}_{M}^\infty(k, l)\right)_{k,l\geq0}$    satisfies axiom 2. of  Definition \ref{def:Trap} by the properties of the tensor product. We want to check that the maps $t_{i,j}$ are well-defined and  satisfy axioms 3. (a)-(c).

	The existence of the integral follows from the smoothness of $K_1$ and $K_2$ and the 
	closedness of $M$. Therefore, by definition of ${\mathcal K}_M^{\infty}$, to show that $t_{i,j}(K_1\otimes K_2)\in{\mathcal K}_M^{\infty}(k-1,l-1)$ and thus that $t_{i,j}$ is a partial trace, is enough to show that the function 
	$t_{i,j}(K_1\otimes K_2):M^{k-1}\times M^{l-1}\longrightarrow \C$ is smooth.  Since
	$K_1$ and $K_2$ are smooth, the map 
	\begin{equation*}
	(x_1 ,\cdots, x_{k-1},y_1,\cdots, y_{\cy{l-1}}) \longmapsto K_1(x_1 ,  \cdots,x_{i-1},z,x_{i}\cdots, x_{k-1})\, K_2(y_1,\cdots,y_{j-1},z,y_{j}\cdots, y_{\cy{l-1}})
	\end{equation*}
	is infinitely differentiable for any $z\in M$. For $\vec x=(x_1,\cdots,x_k)\in M^k$ and $\vec \alpha=(\alpha_1,\cdots,\alpha_k)\in\N_0^k$ we use the short-hand notation
	\begin{equation*}
	 \partial_{\vec x}^{\vec\alpha}:=\frac{\partial^{\alpha_1}}{\partial x_1^{\alpha_1}}\cdots\frac{\partial^{\alpha_k}}{\partial x_k^{\alpha_k}}.
	\end{equation*}
    Then, since $M$ is compact, the partial derivatives
	\begin{equation*}
	\partial_{\vec x}^{\vec\alpha}\partial_{\vec y}^{\vec \beta}K_1(x_1 ,  \cdots,x_{i-1},z,x_{i}\cdots, x_{k-1})\, K_2(y_1,\cdots,y_{j-1},z,y_{j}\cdots, y_{\cy{l-1}})
	\end{equation*}
	are bounded uniformly in $z$ and hence
	\begin{align*}
	& \int_M \partial_{\vec x}^{\vec\alpha}\partial_{\vec y}^{\vec \beta}K_1(x_1 ,  \cdots,x_{i-1},z,x_{i}\cdots, x_{k-1})\, K_2(y_1,\cdots,y_{j-1},z,y_{j}\cdots, y_{\cy{l-1}}) \,d\mu(z) \\
	= & \partial_{\vec x}^{\vec\alpha}\partial_{\vec y}^{\vec \beta}\int_MK_1(x_1 ,  \cdots,x_{i-1},z,x_{i}\cdots, x_{k-1})\, K_2(y_1,\cdots,y_{j-1},z,y_{j}\cdots, y_{\cy{l-1}}) \,d\mu(z) \\
	= & \partial_{\vec x}^{\vec\alpha}\partial_{\vec y}^{\vec \beta} t_{i,j}(K_1\otimes K_2) (x_1, \cdots, x_{k-1}, y_1, \cdots,  y_{l-1}).
	\end{align*}
	Therefore the map $t_{i,j}(K_1\otimes K_2) (x_1, \cdots, x_{k-1}, y_1, \cdots,  y_{l-1})$ is smooth.
% 	\[t_{1,2}(I*p)=p.\]

\cy{Axiom 3.(a) follows from the Fubini theorem, which holds in this case since $M$ is compact and by definition of the spaces $\calK_M^\infty(k,l)$. Axiom 3.(b) is a consequence of Theorem \ref{thm:Hom_V_generalised}.}

Finally, to check Axiom 3.(c), by Lemma \ref{lemmeaxiomessimples} it is enough to check the compatibility of the horizontal concatenation with the partial trace  to show that $t_{1,1}(p*p')=t_{1,1}(p)*p'$ for any pair $(p,p')\in\mathcal{K}_{M}^\infty(k,l)\times\mathcal{K}_{M}^\infty(k',l')$ with $k,k',l,l'\geq1$. Setting $p=K_1\otimes K_2$ and $p'=K'_1\otimes K'_2$ we have, by definition of the partial trace  maps and the horizontal concatenation
	\begin{align*}
	 & t_{1,1}(p*p')(x_1,\cdots,x_{k+k'-1},y_1,\cdots,y_{l+l'-1}) \\
	  = & \int K_1(z,x_1,\cdots,x_{k-1})K'_1(x_{k},\cdots,x_{k+k'-1})K_2(z,y_1,\cdots,y_{l-1})K'_2(x_{l},\cdots,x_{l+l'-1})d\mu(z) \\
	  = & \left(\int K_1(z,x_1,\cdots,x_{k-1})K_2(z,y_1,\cdots,y_{l-1})d\mu(z)\right)K'_1(x_{k},\cdots,x_{k+k'-1})K'_2(x_{l},\cdots,x_{l+l'-1}) \\
	  = & (t_{1,1}(p)*p')(x_1,\cdots,x_{k+k'-1},y_1,\cdots,y_{l+l'-1}).\qedhere
	\end{align*}
\end{proof}
\begin{rk}
     Notice that $\mathcal{K}_{M}^\infty$ is non unitary, since the map $f:M\times M\longrightarrow\C$ which could play the role of a vertical unity is a $\delta$ distribution supported on the diagonal. This will further strengthen our claim that non-unitary TRAPs offer an appropriate  framework to host infinite dimensional spaces. We  expect  non-unitary TRAPs to host more general distributions. 
    \end{rk}

\section{Free TRAPs} \label{sec:freeTRAP}

\subsection{Corolla ordered graphs}

Recall from Subssection \ref{subsec:PROP_graphs} that we call a {\bf graph} a family 
\begin{equation*}
 G=(V(G),E(G),I(G),O(G),IO(G),s,t,\alpha,\beta)
\end{equation*}
with $V(G)$ the set of vertices, $E(G)$ the set of internal edges, $I(G)$ the set of incoming edge, $O(G)$ the set of outgoing edges and $IO(G)$ the set of incoming and outgoing edges. The sets $V(G)$, $ E(G)$ , $I(G)$, $O(G)$ and $IO(G)$ are finite. Furthermore $s$ and $t$ are respectively the source and target maps of the edges and $\alpha$ and $\beta$ are respectively the indexation of the inputs and outputs of the graph. We refer the reader to Definition \ref{def:graph} for the details.

To show that graphs are a free TRAP, we will need one more structure on graphs. Furthermore, since we are dealing here with non unitary TRAPs only, the set of input-output edges $IO(G)$ will play no role, and on the contrary shall be empty. This justifies the next definition.
\begin{defn} \label{def:solar_graphs}
 For a graph $G=(V(G),E(G),I(G),O(G),IO(G),s,t,\alpha,\beta)$,  for any $v\in V(G)$, $O(v)$ and $I(v)$ respectively refer to the set of {\bf outgoing and ingoing edges of the vertex $v$}. In other words, for any $v\in V(G)$,
 \begin{equation*}
  O(v):=\{e\in E(G)\sqcup O(G)|s(e)=v\},\qquad I(v):=\{e\in E(G)\sqcup I(G)|t(e)=v.\}
 \end{equation*}
 We denote the cardinals of the sets $O(v)$ and $I(v)$ as $o(G)$, $i(G)$, $o(v)$ and $i(v)$ respectively (recall the cardinals of the sets $O(G)$ and $I(G)$ are denoted as $o(G)$ and $i(G)$ respectively).
 
 A \textbf{{corolla ordered} graph} is a graph $G$ such that for any vertex $v$, the set of incoming edges $I(v)$ of $v$
and the set of outgoing edges $O(v)$ of $v$ are totally ordered and we shall denote both order relations by $\leq_v$.

A (resp. a corolla oriented) graph $G$ is \textbf{solar} if $IO(G)=\emptyset$. 
\end{defn}
For a solar graph (that is, such that $IO(G)=\emptyset$), the terminology 'solar' refers to its radiating aspect with rays around a central body. In \cite{JY15} such graphs are called 'ordinary'. 

Graphically, if $G$ is a corolla ordered graph, we shall represent the orders on the incoming and outgoing edges of a vertex by drawing box-shaped vertices, with the incoming and outgoing edges of any vertex  ordered from left to right. For example, we distinguish the following two  situations:
\begin{align*}
&\begin{tikzpicture}[line cap=round,line join=round,>=triangle 45,x=0.7cm,y=0.7cm]
\clip(0.8,-0.5) rectangle (2.2,2.5);
\draw [line width=.4pt] (0.8,0.)-- (2.2,0.);
\draw [line width=.4pt] (2.2,0.)-- (2.2,-0.5);
\draw [line width=.4pt] (2.2,-0.5)-- (0.8,-0.5);
\draw [line width=.4pt] (0.8,-0.5)-- (0.8,0.);
\draw [line width=.4pt] (0.8,2.)-- (2.2,2.);
\draw [line width=.4pt] (2.2,2.)-- (2.2,2.5);
\draw [line width=.4pt] (2.2,2.5)-- (0.8,2.5);
\draw [line width=.4pt] (0.8,2.5)-- (0.8,2.);
\draw [->,line width=.4pt] (1.,0.) -- (1.,2.);
\draw [->,line width=.4pt] (2.,0.) -- (2.,2.);
\end{tikzpicture}&
&\begin{tikzpicture}[line cap=round,line join=round,>=triangle 45,x=0.7cm,y=0.7cm]
\clip(0.8,-0.5) rectangle (2.2,2.5);
\draw [line width=.4pt] (0.8,0.)-- (2.2,0.);
\draw [line width=.4pt] (2.2,0.)-- (2.2,-0.5);
\draw [line width=.4pt] (2.2,-0.5)-- (0.8,-0.5);
\draw [line width=.4pt] (0.8,-0.5)-- (0.8,0.);
\draw [line width=.4pt] (0.8,2.)-- (2.2,2.);
\draw [line width=.4pt] (2.2,2.)-- (2.2,2.5);
\draw [line width=.4pt] (2.2,2.5)-- (0.8,2.5);
\draw [line width=.4pt] (0.8,2.5)-- (0.8,2.);
\draw [->,line width=.4pt] (1.,0.) -- (2.,2.);
\draw [->,line width=.4pt] (2.,0.) -- (1.,2.);
\end{tikzpicture} 
\end{align*}
We note that any graph $G$ can be made into a corolla ordered graph in $\prod_{v\in V(G)}i(v)!o(v)!$ ways, by chosing an order on the ingoing and outgoing vertices of the graph. For example the graph of Example \ref{ex4} can be made into a corolla ordered graph in $3!1!1!3!=36$ ways. Here are three of them:
\begin{align*}
\begin{tikzpicture}[line cap=round,line join=round,>=triangle 45,x=0.7cm,y=0.7cm]
\clip(0.9,-0.1) rectangle (9.5,4.7);
\draw [line width=0.4pt] (1.,2.)-- (3.,2.);
\draw [line width=0.4pt] (3.,2.)-- (3.,3.);
\draw [line width=0.4pt] (3.,3.)-- (1.,3.);
\draw [line width=0.4pt] (1.,3.)-- (1.,2.);
\draw [->,line width=0.4pt] (1.5,1.)-- (1.5,2.);
\draw [->,line width=0.4pt] (2.,1.)-- (2.,2.);
\draw [->,line width=0.4pt] (2.5,1.5)-- (2.5,2.);
\draw [shift={(3.75,1.5)},line width=0.4pt]  plot[domain=3.141592653589793:6.283185307179586,variable=\t]({1.*1.25*cos(\t r)+0.*1.25*sin(\t r)},{0.*1.25*cos(\t r)+1.*1.25*sin(\t r)});
\draw [line width=0.4pt] (5.,1.5)-- (5.,3.5);
\draw [line width=0.4pt] (2.,3.)-- (2.,3.5);
\draw [shift={(3.,3.5)},line width=0.4pt]  plot[domain=0.:3.141592653589793,variable=\t]({1.*1.*cos(\t r)+0.*1.*sin(\t r)},{0.*1.*cos(\t r)+1.*1.*sin(\t r)});
\draw [line width=0.4pt] (4.,3.5)-- (4.,1.5);
\draw [line width=0.4pt] (6.,2.)-- (8.,2.);
\draw [line width=0.4pt] (8.,2.)-- (8.,3.);
\draw [line width=0.4pt] (8.,3.)-- (6.,3.);
\draw [,line width=0.4pt] (6.,3.)-- (6.,2.);
\draw [line width=0.4pt] (4.,3.5)-- (4.,1.5);
\draw [->,line width=0.4pt] (7.,1.5)-- (7.,2.);
\draw [shift={(5.5,1.5)},line width=0.4pt]  plot[domain=3.141592653589793:6.283185307179586,variable=\t]({1.*1.5*cos(\t r)+0.*1.5*sin(\t r)},{0.*1.5*cos(\t r)+1.*1.5*sin(\t r)});
\draw [line width=0.4pt] (7.,3.)-- (7.,3.5);
\draw [shift={(6.,3.5)},line width=0.4pt]  plot[domain=0.:3.141592653589793,variable=\t]({1.*1.*cos(\t r)+0.*1.*sin(\t r)},{0.*1.*cos(\t r)+1.*1.*sin(\t r)});
\draw [->,line width=0.4pt] (6.5,3.)-- (6.5,3.5);
\draw [->,line width=0.4pt] (7.5,3.)-- (7.5,3.5);
\draw (1.8,2.8) node[anchor=north west] {$x$};
\draw (6.8,2.8) node[anchor=north west] {$y$};
\draw [->,line width=0.4pt] (9.,1.)-- (9.,4.);
\draw (1.1,1.1) node[anchor=north west] {$1$};
\draw (1.6,1.1) node[anchor=north west] {$2$};
\draw (8.6,1.1) node[anchor=north west] {$3$};
\draw (6.1,4.1) node[anchor=north west] {$1$};
\draw (7.1,4.1) node[anchor=north west] {$3$};
\draw (8.6,4.6) node[anchor=north west] {$2$};
\end{tikzpicture}
% 	\substack{\displaystyle \\ \vspace{3cm}	\xymatrix{&\ar@(ul,dl)[]}\xymatrix{&\ar@(ul,dl)[]}}
\end{align*}

\vspace{-1.5cm}

\begin{align*}
\begin{tikzpicture}[line cap=round,line join=round,>=triangle 45,x=0.7cm,y=0.7cm]
\clip(0.9,-0.1) rectangle (9.5,4.7);
\draw [line width=0.4pt] (1.,2.)-- (3.,2.);
\draw [line width=0.4pt] (3.,2.)-- (3.,3.);
\draw [line width=0.4pt] (3.,3.)-- (1.,3.);
\draw [line width=0.4pt] (1.,3.)-- (1.,2.);
\draw [->,line width=0.4pt] (1.5,1.)-- (2.,2.);
\draw [->,line width=0.4pt] (2.,1.)-- (1.5,2.);
\draw [->,line width=0.4pt] (2.5,1.5)-- (2.5,2.);
\draw [shift={(3.75,1.5)},line width=0.4pt]  plot[domain=3.141592653589793:6.283185307179586,variable=\t]({1.*1.25*cos(\t r)+0.*1.25*sin(\t r)},{0.*1.25*cos(\t r)+1.*1.25*sin(\t r)});
\draw [line width=0.4pt] (5.,1.5)-- (5.,3.5);
\draw [line width=0.4pt] (2.,3.)-- (2.,3.5);
\draw [shift={(3.,3.5)},line width=0.4pt]  plot[domain=0.:3.141592653589793,variable=\t]({1.*1.*cos(\t r)+0.*1.*sin(\t r)},{0.*1.*cos(\t r)+1.*1.*sin(\t r)});
\draw [line width=0.4pt] (4.,3.5)-- (4.,1.5);
\draw [line width=0.4pt] (6.,2.)-- (8.,2.);
\draw [line width=0.4pt] (8.,2.)-- (8.,3.);
\draw [line width=0.4pt] (8.,3.)-- (6.,3.);
\draw [,line width=0.4pt] (6.,3.)-- (6.,2.);
\draw [line width=0.4pt] (4.,3.5)-- (4.,1.5);
\draw [->,line width=0.4pt] (7.,1.5)-- (7.,2.);
\draw [shift={(5.5,1.5)},line width=0.4pt]  plot[domain=3.141592653589793:6.283185307179586,variable=\t]({1.*1.5*cos(\t r)+0.*1.5*sin(\t r)},{0.*1.5*cos(\t r)+1.*1.5*sin(\t r)});
\draw [->,line width=0.4pt] (7.,3.)-- (7.,3.5);
\draw [shift={(5.75,3.5)},line width=.4pt]  plot[domain=0.:3.141592653589793,variable=\t]({1.*0.75*cos(\t r)+0.*0.75*sin(\t r)},{0.*0.75*cos(\t r)+1.*0.75*sin(\t r)});
\draw [line width=0.4pt] (6.5,3.)-- (6.5,3.5);
\draw [->,line width=0.4pt] (7.5,3.)-- (7.5,3.5);
\draw (1.8,2.8) node[anchor=north west] {$x$};
\draw (6.8,2.8) node[anchor=north west] {$y$};
\draw [->,line width=0.4pt] (9.,1.)-- (9.,4.);
\draw (1.1,1.1) node[anchor=north west] {$1$};
\draw (1.6,1.1) node[anchor=north west] {$2$};
\draw (8.6,1.1) node[anchor=north west] {$3$};
\draw (6.6,4.1) node[anchor=north west] {$1$};
\draw (7.1,4.1) node[anchor=north west] {$3$};
\draw (8.6,4.6) node[anchor=north west] {$2$};
\end{tikzpicture}	
% \substack{\displaystyle \\ \vspace{3cm}	\xymatrix{&\ar@(ul,dl)[]}\xymatrix{&\ar@(ul,dl)[]}}
\end{align*}

\vspace{-1.5cm}
\begin{align*}
\begin{tikzpicture}[line cap=round,line join=round,>=triangle 45,x=0.7cm,y=0.7cm]
\clip(0.9,-0.1) rectangle (9.5,4.7);
\draw [line width=0.4pt] (1.,2.)-- (3.,2.);
\draw [line width=0.4pt] (3.,2.)-- (3.,3.);
\draw [line width=0.4pt] (3.,3.)-- (1.,3.);
\draw [line width=0.4pt] (1.,3.)-- (1.,2.);
\draw [->,line width=0.4pt] (1.5,1.5)-- (1.5,2.);
\draw [->,line width=0.4pt] (2.,1.5)-- (2.,2.);
\draw [->,line width=0.4pt] (2.5,1.5)-- (2.5,2.);
\draw [line width=0.4pt] (5.,1.5)-- (5.,3.5);
\draw [line width=0.4pt] (2.,3.)-- (2.,3.5);
\draw [shift={(3.,3.5)},line width=0.4pt]  plot[domain=0.:3.141592653589793,variable=\t]({1.*1.*cos(\t r)+0.*1.*sin(\t r)},{0.*1.*cos(\t r)+1.*1.*sin(\t r)});
\draw [line width=0.4pt] (4.,3.5)-- (4.,1.5);
\draw [line width=0.4pt] (6.,2.)-- (8.,2.);
\draw [line width=0.4pt] (8.,2.)-- (8.,3.);
\draw [line width=0.4pt] (8.,3.)-- (6.,3.);
\draw [,line width=0.4pt] (6.,3.)-- (6.,2.);
\draw [line width=0.4pt] (4.,3.5)-- (4.,1.5);
\draw [->,line width=0.4pt] (7.,1.5)-- (7.,2.);
\draw [shift={(5.5,1.5)},line width=0.4pt]  plot[domain=3.141592653589793:6.283185307179586,variable=\t]({1.*1.5*cos(\t r)+0.*1.5*sin(\t r)},{0.*1.5*cos(\t r)+1.*1.5*sin(\t r)});
\draw [->,line width=0.4pt] (7.,3.)-- (7.5,4.);
\draw [shift={(5.75,3.5)},line width=.4pt]  plot[domain=0.:3.141592653589793,variable=\t]({1.*0.75*cos(\t r)+0.*0.75*sin(\t r)},{0.*0.75*cos(\t r)+1.*0.75*sin(\t r)});
\draw [shift={(2.5,1.5)},line width=.4pt]  plot[domain=3.141592653589793:4.71238898038469,variable=\t]({1.*1.*cos(\t r)+0.*1.*sin(\t r)},{0.*1.*cos(\t r)+1.*1.*sin(\t r)});
\draw [shift={(4.,1.5)},line width=.4pt]  plot[domain=-1.5707963267948966:0.,variable=\t]({1.*1.*cos(\t r)+0.*1.*sin(\t r)},{0.*1.*cos(\t r)+1.*1.*sin(\t r)});
\draw [line width=.4pt] (2.5,0.5)-- (4.,0.5);
\draw [line width=0.4pt] (6.5,3.)-- (6.5,3.5);
\draw [->,line width=0.4pt] (7.5,3.)-- (7.,4.);
\draw (1.8,2.8) node[anchor=north west] {$x$};
\draw (6.8,2.8) node[anchor=north west] {$y$};
\draw [->,line width=0.4pt] (9.,1.)-- (9.,4.);
\draw (1.6,1.6) node[anchor=north west] {$1$};
\draw (2.1,1.6) node[anchor=north west] {$2$};
\draw (8.6,1.1) node[anchor=north west] {$3$};
\draw (6.6,4.6) node[anchor=north west] {$1$};
\draw (7.1,4.6) node[anchor=north west] {$3$};
\draw (8.6,4.6) node[anchor=north west] {$2$};
\end{tikzpicture}
% 	\substack{\displaystyle \\ \vspace{3cm}	\xymatrix{&\ar@(ul,dl)[]}\xymatrix{&\ar@(ul,dl)[]}}
\end{align*}

% \vspace{-1.5cm}
Recall that for two graphs $G$ and $G'$, a \textbf{morphism} of graphs from $G$ to $G'$ is a family of   maps $f=(f_V,f_E,f_I,f_O,f_{IO})$ where each map sends vertices to vertices, each type of edges of $G$ to the corresponding type of edges in $G'$ and furthermore respects the source, target and indexation maps. We refer the reader to Definition \ref{def:morph_graphs} for the details. Recall also that an \textbf{isomorphism} of graphs from $G$ to $G'$ is a morphism of graphs $f=(f_V,f_E,f_I,f_O,f_{IO},f_L)$ from $G$ to $G'$ such that $f_V,f_E,f_I,f_O,f_{IO}$ and $f_L$ are bijections.
  
We can now introduce the morphisms relevant to the graphs we will deal with.
\begin{defn} \label{def:isoclassesTRAPgraphs}
 Let $G$ and $G'$ be two corolla ordered graphs.
 \begin{enumerate}
  \item A  {\textbf{morphism}} of corolla ordered graphs from $G$ to $G'$ is \ty{a} {\textbf{morphism}}   of graphs
$f$ from $G$ to $G'$ which preserves the order of incoming and outgoing edges that is, for any vertex of $v$:
\begin{itemize}
\item For any incoming edges $e$, $e'$ of $v$, $e\leqslant_v e'$ in $G$ if, and only if, $f(e)\leqslant_{f(v)} f(e')$ in $G'$.
\item For any outgoing edges $e$, $e'$ of $v$, $e\leqslant_v e'$ in $G$ if, and only if, $f(e)\leqslant_{f(v)} f(e')$ in $G'$.
\end{itemize}
  \item An {\bf isomorphism} of corolla ordered graphs from $G$ to $G'$ is a morphism of corolla ordered graphs from $G$ to $G'$ that is also an isomorphism of graphs from $G$ to $G'$.
  \item Recall that for any $(k,l)$ in $\N^2$ $\Gr(k,l)$ is the quotient space of graphs with $k$ input edges and $l$ output edges by the equivalence relation given by isomorphism. Similarly, we denote by $\PGr(k,l)$ the set of isoclasses of corolla ordered 
graphs $G$ such that $i(G)=k$ and $o(G)=l$.
\item The subset of $\Gr(k,l)$ formed by isoclasses of solar graphs is denoted by $\rGr(k,l)$ and the subset of $\PGr(k,l)$ formed by isoclasses of solar {corolla ordered} graphs is denoted by $\rPGr(k,l)$. 
 \end{enumerate}
\end{defn}
 Throughout this thesis, $X=(X(k,l))_{k,l\geqslant 0}$ is  a family of sets.
 \begin{defn}\label{defidecorations}
  A graph decorated by $X=(X(k,l))_{k,l\geqslant 0}$ (or $X$-decorated graph, or simply decorated graph) is a couple $(G,d_G)$ with $G$ a graph as in Definition \ref{def:graph} and $\displaystyle d_G:V(G)\longrightarrow \bigsqcup_{k,l\geq0}X(k,l)$ an arbitrary map, such that for any vertex $v\in V(G)$, $d_G(v)\in X(i(v),o(v))$. We denote by $\Gr(X)$ the set of graphs decorated by $X$. Similarly, we define  $X$-decorated {corolla ordered} graphs which we denote by $\PGr(X)$.
  
We further write $\Gr(X)(k,l)$ (respectively, $\PGr(X)(k,l)$, $\rGr(X)(k,l)$ and\\ $\rPGr(X)(k,l)$) the set of graphs (respectively, of corolla ordered graphs, solar graphs, solar corolla ordered graphs) 
  decorated by $X$ with $k$ inputs (that is $|I(G)|=k$) and $l$ ouputs (that is $|O(G)|=l$).
 \end{defn}

\subsection{TRAPs of graphs} \label{subsectionTRAPSgraphs}

The horizontal concatenation and the structure of $\sym\times\sym^{\rm op}$-module on $\PGr$, $\rGr$ and $\rPGr$ is induced by the one of on $\Gr$. The structure of $\sym\times\sym^{\rm op}$-module is given by Equation \ref{eq:sym_sym_graph}. Let us recall it here.
\begin{equation*}
 \tau\cdot G\cdot \sigma=
(V(G),E(G),I(G),O(G),IO(G),s,t,\sigma^{-1}\circ \alpha,\tau \circ \beta).
\end{equation*}
Let us further recall that for two graphs $G$ and $G'$ the the horizontal concatenation $G*G'$ is given by the disjoint union of the edges and various types of edges of $G$ and $G'$. The source, target and indexation maps are induced by the same maps of $G$ and $G'$. We refer the reader to Subsection \ref{subsec:PROP_graphs} for the precise definitions. Similarly, the orders of the sets $I(v)$ and $O(v)$ for each $v\in V(G*G')=V(G)\sqcup V(G')$ are induced by the total orders on the sets of $V(G)$ and $V(G')$.

Let us finally define the partial trace maps. Let us first give an  outline of their definition. Let $G\in \rGr(k,l)$ (thus $IO(G)=\emptyset$), $1\leqslant i\leqslant k$ and $1\leqslant j\leqslant l$. We set $e_i=\alpha_G^{-1}(i)\in I(G)$,  $f_j=\beta_G^{-1}(j)\in O(G)$ and define $t_{i,j}(G)$ as the graph obtained by identifying the input of $e_i$ with the output $j$ of $f_j$. This creates an edge in $E(G)$. We then reindex in non-decreasing order,  the inputs and the outputs of the obtained graph.

Rigorously we define the graph $G'=t_{i,j}(G)$ in the following way:
	\begin{align*}
	V(G')&=V(G),&E(G')&=E(G)\sqcup \{(e_i,f_j)\},\\
	I(G')&=I(G)\setminus\{e_i\},& O(G')&=O(G)\setminus\{f_j\},\\
	IO(G')&=IO(G)=\emptyset,\\
	s_{G'}(e)&=\begin{cases}
	s_G(f_j)\mbox{ if }e=(e_i,f_j),\\
	s_G(e)\mbox{ otherwise},
	\end{cases}&
	t_{G'}(e)&=\begin{cases}
	t_G(e_i)\mbox{ if }e=(e_i,f_j),\\
	t_G(e)\mbox{ otherwise},
	\end{cases}\\
	\alpha_{G'}(e)&=\begin{cases}
	\alpha_G(e)\mbox{ if }\alpha_G(e)<i,\\
	\alpha_G(e)-1\mbox{  if }\alpha_G(e)\geqslant i,
	\end{cases}&
	\beta_{G'}(e)&=\begin{cases}
	\beta_G(e)\mbox{ if }\beta_G(e)<j,\\
	\beta_G(e)-1\mbox{ if }\beta_G(e)\geqslant j.
	\end{cases}
	\end{align*}

Graphically:
\begin{center}
\begin{tikzpicture}[line cap=round,line join=round,>=triangle 45,x=0.5cm,y=0.5cm]
\clip(-2.5,-4.5) rectangle (2.5,4.);
\draw [line width=0.4pt] (-2.,1.)-- (2.,1.);
\draw [line width=0.4pt] (2.,1.)-- (2.,-1.);
\draw [line width=0.4pt] (2.,-1.)-- (-2.,-1.);
\draw [line width=0.4pt] (-2.,-1.)-- (-2.,1.);
\draw [->,line width=0.4pt] (-1.5,1.) -- (-1.5,3.);
\draw [->,line width=0.4pt] (0.,1.) -- (0.,3.);
\draw [->,line width=0.4pt] (1.5,1.) -- (1.5,3.);
\draw [->,line width=0.4pt] (-1.5,-3.) -- (-1.5,-1.);
\draw [->,line width=0.4pt] (0.,-3.) -- (0.,-1.);
\draw [->,line width=0.4pt] (1.5,-3.) -- (1.5,-1.);
\draw (-0.5,0.5) node[anchor=north west] {$G$};
\draw (-1.8,-3) node[anchor=north west] {$1$};
\draw (-0.3,-3) node[anchor=north west] {$i$};
\draw (1.2,-3) node[anchor=north west] {$k$};
\draw (-1.4,-2.2) node[anchor=north west] {$\ldots$};
\draw (0.1,-2.2) node[anchor=north west] {$\ldots$};
\draw (-1.8,4.2) node[anchor=north west] {$1$};
\draw (-0.3,4.2) node[anchor=north west] {$j$};
\draw (1.2,4.2) node[anchor=north west] {$l$};
\draw (-1.4,2.) node[anchor=north west] {$\ldots$};
\draw (0.1,2.) node[anchor=north west] {$\ldots$};
\end{tikzpicture}
$\substack{\displaystyle \stackrel{t_{i,j}}{\longmapsto}\\ \vspace{4cm}}$
\begin{tikzpicture}[line cap=round,line join=round,>=triangle 45,x=0.5cm,y=0.5cm]
\clip(-3.5,-4.5) rectangle (3.5,5.);
\draw [line width=0.4pt] (-2.,1.)-- (2.,1.);
\draw [line width=0.4pt] (2.,1.)-- (2.,-1.);
\draw [line width=0.4pt] (2.,-1.)-- (-2.,-1.);
\draw [line width=0.4pt] (-2.,-1.)-- (-2.,1.);
\draw [->,line width=0.4pt] (-1.5,1.) -- (-1.5,3.);
\draw [line width=0.4pt] (0.,1.) -- (0.,3.);
\draw [->,line width=0.4pt] (1.5,1.) -- (1.5,3.);
\draw [->,line width=0.4pt] (-1.5,-3.) -- (-1.5,-1.);
\draw [->,line width=0.4pt] (0.,-3.) -- (0.,-1.);
\draw [->,line width=0.4pt] (1.5,-3.) -- (1.5,-1.);
\draw (-0.5,0.5) node[anchor=north west] {$G$};
\draw (-1.8,-3) node[anchor=north west] {$1$};
\draw (1.2,-3) node[anchor=north west] {$k-1$};
\draw (-1.4,-2.2) node[anchor=north west] {$\ldots$};
\draw (0.1,-2.2) node[anchor=north west] {$\ldots$};
\draw (-1.8,4.2) node[anchor=north west] {$1$};
\draw (1.2,4.2) node[anchor=north west] {$l-1$};
\draw (-1.4,2.) node[anchor=north west] {$\ldots$};
\draw (0.1,2.) node[anchor=north west] {$\ldots$};
\draw [shift={(-1.5,-3.)},line width=0.4pt]  plot[domain=3.141592653589793:6.283185307179586,variable=\t]({1.*1.5*cos(\t r)+0.*1.5*sin(\t r)},{0.*1.5*cos(\t r)+1.*1.5*sin(\t r)});
\draw [shift={(-1.5,3.)},line width=0.4pt]  plot[domain=0.:3.141592653589793,variable=\t]({1.*1.5*cos(\t r)+0.*1.5*sin(\t r)},{0.*1.5*cos(\t r)+1.*1.5*sin(\t r)});
\draw [line width=0.4pt] (-3.,3.) -- (-3.,-3.);
\end{tikzpicture}
\vspace{-2cm}
\end{center}
As before, if $G$ is corolla ordered, or if it is $X$-decorated, then $t_{i,j}(G)$ is corolla ordered, or $X$-decorated.

\begin{example} Let $G$ be the following graph:
% \[\xymatrix{1&\\
% \rond{}\ar[u]&\\
% 2\ar[u]&3\ar[lu]}\]
% Then:
% \begin{equation*}
% \substack{\vspace{1cm}\\ \displaystyle t_{2,1}(G)=t_{3,1}(G)=\hspace{.3cm}}
%  \xymatrix{
% \rond{}\ar@(ul,dl)[]\\
% 1\ar[u]}
% \end{equation*}
% \end{example}
% Let us further XXXXXX
% \begin{example}
\[\xymatrix{1&3&2\\
&\rond{}\ar[ru] \ar[u] &&\\
&\rond{}\ar[u]\ar[uul]&&\\
1\ar[ru]&3\ar[u]&2\ar[lu]}\]
Then:
\begin{equation*}
 \substack{\vspace{1cm}\\ \displaystyle t_{1,1}(G)=t_{2,1}(G)=t_{3,1}(G)=\hspace{.3cm}}
 \xymatrix{2&&1\\
&\rond{}\ar[ru] \ar[lu] &&\\
&\rond{}\ar@(ul,dl)[]\ar[u]&&\\
% 1\ar[ru]
&1\ar[u]&2\ar[lu]}
\end{equation*}
and
\[ \substack{\vspace{1cm}\\ \displaystyle t_{i,2}(G)=t_{i,3}(G)=\hspace{.3cm}}\xymatrix{1&&2\\
&\rond{}\ar[ru]  \ar@/_1pc/[d]&&\\
&\rond{}\ar@/_1pc/[u]\ar[uul]&&\\
1\ar[ru]&&2\ar[lu]}\]
for $i\in\{1,2,3\}$.
\end{example}
Let us know state and show the main result of this Subsection.
\begin{theo} \label{thm:TRAPgraphs}
 With this data, $\rGr$, $\rPGr$, $\rGr(X)$ and $\rPGr(X)$ are TRAPs. 
\end{theo}
\begin{proof}
 We provide the proof for $\PGr$. The proof is similar for the three other cases. Properties 1.~and 2., as in Subsection \ref{subsec:PROP_graphs}, follow directly from the symmetric group actions and the horizontal concatenation of graphs defined above.
Let us give a graphical interpretation of the proof of Property 3.(a), when $i'<i$ and $j'<j$.

\begin{center}
\begin{tikzpicture}[line cap=round,line join=round,>=triangle 45,x=0.5cm,y=0.5cm]
\clip(-2.5,-4.5) rectangle (4,4.);
\draw [line width=0.4pt] (-2.,1.)-- (3.5,1.);
\draw [line width=0.4pt] (3.5,1.)-- (3.5,-1.);
\draw [line width=0.4pt] (3.5,-1.)-- (-2.,-1.);
\draw [line width=0.4pt] (-2.,-1.)-- (-2.,1.);
\draw [->,line width=0.4pt] (-1.5,1.) -- (-1.5,3.);
\draw [->,line width=0.4pt] (0.,1.) -- (0.,3.);
\draw [->,line width=0.4pt] (1.5,1.) -- (1.5,3.);
\draw [->,line width=0.4pt] (3.,1.) -- (3.,3.);
\draw [->,line width=0.4pt] (-1.5,-3.) -- (-1.5,-1.);
\draw [->,line width=0.4pt] (0.,-3.) -- (0.,-1.);
\draw [->,line width=0.4pt] (1.5,-3.) -- (1.5,-1.);
\draw [->,line width=0.4pt] (3.,-3.) -- (3.,-1.);
\draw (0.25,0.5) node[anchor=north west] {$G$};
\draw (-1.8,-3) node[anchor=north west] {$1$};
\draw (-0.3,-3) node[anchor=north west] {$i'$};
\draw (1.2,-3) node[anchor=north west] {$i$};
\draw (2.7,-3) node[anchor=north west] {$k$};
\draw (-1.4,-2.2) node[anchor=north west] {$\ldots$};
\draw (0.1,-2.2) node[anchor=north west] {$\ldots$};
\draw (1.6,-2.2) node[anchor=north west] {$\ldots$};
\draw (-1.8,4.2) node[anchor=north west] {$1$};
\draw (-0.3,4.2) node[anchor=north west] {$j'$};
\draw (1.2,4.2) node[anchor=north west] {$j$};
\draw (2.7,4.2) node[anchor=north west] {$l$};
\draw (-1.4,2.) node[anchor=north west] {$\ldots$};
\draw (0.1,2.) node[anchor=north west] {$\ldots$};
\draw (1.6,2.) node[anchor=north west] {$\ldots$};
\end{tikzpicture}
$\substack{\displaystyle \stackrel{t_{i,j}}{\longmapsto}\\ \vspace{4cm}}$
\begin{tikzpicture}[line cap=round,line join=round,>=triangle 45,x=0.5cm,y=0.5cm]
\clip(-2.5,-4.5) rectangle (5.,5.);
\draw [line width=0.4pt] (-2.,1.)-- (3.5,1.);
\draw [line width=0.4pt] (3.5,1.)-- (3.5,-1.);
\draw [line width=0.4pt] (3.5,-1.)-- (-2.,-1.);
\draw [line width=0.4pt] (-2.,-1.)-- (-2.,1.);
\draw [->,line width=0.4pt] (-1.5,1.) -- (-1.5,3.);
\draw [->,line width=0.4pt] (0.,1.) -- (0.,3.);
\draw [line width=0.4pt] (1.5,1.) -- (1.5,3.);
\draw [->,line width=0.4pt] (3.,1.) -- (3.,3.);
\draw [->,line width=0.4pt] (-1.5,-3.) -- (-1.5,-1.);
\draw [->,line width=0.4pt] (0.,-3.) -- (0.,-1.);
\draw [->,line width=0.4pt] (1.5,-3.) -- (1.5,-1.);
\draw [->,line width=0.4pt] (3.,-3.) -- (3.,-1.);
\draw (0.25,0.5) node[anchor=north west] {$G$};
\draw (-1.8,-3) node[anchor=north west] {$1$};
\draw (-0.3,-3) node[anchor=north west] {$i'$};
\draw (2.1,-3) node[anchor=north west] {$k-1$};
\draw (-1.4,-2.2) node[anchor=north west] {$\ldots$};
\draw (0.1,-2.2) node[anchor=north west] {$\ldots$};
\draw (1.6,-2.2) node[anchor=north west] {$\ldots$};
\draw (-1.8,4.2) node[anchor=north west] {$1$};
\draw (-0.3,4.2) node[anchor=north west] {$j'$};
\draw (2.1,4.2) node[anchor=north west] {$l-1$};
\draw (-1.4,2.) node[anchor=north west] {$\ldots$};
\draw (0.1,2.) node[anchor=north west] {$\ldots$};
\draw (1.6,2.) node[anchor=north west] {$\ldots$};
\draw [shift={(3.,-3.)},line width=0.4pt]  plot[domain=3.141592653589793:6.283185307179586,variable=\t]({1.*1.5*cos(\t r)+0.*1.5*sin(\t r)},{0.*1.5*cos(\t r)+1.*1.5*sin(\t r)});
\draw [shift={(3.,3.)},line width=0.4pt]  plot[domain=0.:3.141592653589793,variable=\t]({1.*1.5*cos(\t r)+0.*1.5*sin(\t r)},{0.*1.5*cos(\t r)+1.*1.5*sin(\t r)});
\draw [line width=0.4pt] (4.5,3.) -- (4.5,-3.);
\end{tikzpicture}
$\substack{\displaystyle \stackrel{t_{i',j'}}{\longmapsto}\\ \vspace{4cm}}$
\begin{tikzpicture}[line cap=round,line join=round,>=triangle 45,x=0.5cm,y=0.5cm]
\clip(-3.5,-4.5) rectangle (5.,5.);
\draw [line width=0.4pt] (-2.,1.)-- (3.5,1.);
\draw [line width=0.4pt] (3.5,1.)-- (3.5,-1.);
\draw [line width=0.4pt] (3.5,-1.)-- (-2.,-1.);
\draw [line width=0.4pt] (-2.,-1.)-- (-2.,1.);
\draw [->,line width=0.4pt] (-1.5,1.) -- (-1.5,3.);
\draw [line width=0.4pt] (0.,1.) -- (0.,3.);
\draw [line width=0.4pt] (1.5,1.) -- (1.5,3.);
\draw [->,line width=0.4pt] (3.,1.) -- (3.,3.);
\draw [->,line width=0.4pt] (-1.5,-3.) -- (-1.5,-1.);
\draw [->,line width=0.4pt] (0.,-3.) -- (0.,-1.);
\draw [->,line width=0.4pt] (1.5,-3.) -- (1.5,-1.);
\draw [->,line width=0.4pt] (3.,-3.) -- (3.,-1.);
\draw (0.25,0.5) node[anchor=north west] {$G$};
\draw (-1.8,-3) node[anchor=north west] {$1$};
\draw (2.1,-3) node[anchor=north west] {$k-2$};
\draw (-1.4,-2.2) node[anchor=north west] {$\ldots$};
\draw (0.1,-2.2) node[anchor=north west] {$\ldots$};
\draw (1.6,-2.2) node[anchor=north west] {$\ldots$};
\draw (-1.8,4.2) node[anchor=north west] {$1$};
\draw (2.1,4.2) node[anchor=north west] {$l-2$};
\draw (-1.4,2.) node[anchor=north west] {$\ldots$};
\draw (0.1,2.) node[anchor=north west] {$\ldots$};
\draw (1.6,2.) node[anchor=north west] {$\ldots$};
\draw [shift={(-1.5,-3.)},line width=0.4pt]  plot[domain=3.141592653589793:6.283185307179586,variable=\t]({1.*1.5*cos(\t r)+0.*1.5*sin(\t r)},{0.*1.5*cos(\t r)+1.*1.5*sin(\t r)});
\draw [shift={(-1.5,3.)},line width=0.4pt]  plot[domain=0.:3.141592653589793,variable=\t]({1.*1.5*cos(\t r)+0.*1.5*sin(\t r)},{0.*1.5*cos(\t r)+1.*1.5*sin(\t r)});
\draw [line width=0.4pt] (-3.,3.) -- (-3.,-3.);
\draw [shift={(3.,-3.)},line width=0.4pt]  plot[domain=3.141592653589793:6.283185307179586,variable=\t]({1.*1.5*cos(\t r)+0.*1.5*sin(\t r)},{0.*1.5*cos(\t r)+1.*1.5*sin(\t r)});
\draw [shift={(3.,3.)},line width=0.4pt]  plot[domain=0.:3.141592653589793,variable=\t]({1.*1.5*cos(\t r)+0.*1.5*sin(\t r)},{0.*1.5*cos(\t r)+1.*1.5*sin(\t r)});
\draw [line width=0.4pt] (4.5,3.) -- (4.5,-3.);
\end{tikzpicture}

\vspace{-2cm}
\end{center}

\begin{center}
\begin{tikzpicture}[line cap=round,line join=round,>=triangle 45,x=0.5cm,y=0.5cm]
\clip(-2.5,-4.5) rectangle (4,4.);
\draw [line width=0.4pt] (-2.,1.)-- (3.5,1.);
\draw [line width=0.4pt] (3.5,1.)-- (3.5,-1.);
\draw [line width=0.4pt] (3.5,-1.)-- (-2.,-1.);
\draw [line width=0.4pt] (-2.,-1.)-- (-2.,1.);
\draw [->,line width=0.4pt] (-1.5,1.) -- (-1.5,3.);
\draw [->,line width=0.4pt] (0.,1.) -- (0.,3.);
\draw [->,line width=0.4pt] (1.5,1.) -- (1.5,3.);
\draw [->,line width=0.4pt] (3.,1.) -- (3.,3.);
\draw [->,line width=0.4pt] (-1.5,-3.) -- (-1.5,-1.);
\draw [->,line width=0.4pt] (0.,-3.) -- (0.,-1.);
\draw [->,line width=0.4pt] (1.5,-3.) -- (1.5,-1.);
\draw [->,line width=0.4pt] (3.,-3.) -- (3.,-1.);
\draw (0.25,0.5) node[anchor=north west] {$G$};
\draw (-1.8,-3) node[anchor=north west] {$1$};
\draw (-0.3,-3) node[anchor=north west] {$i'$};
\draw (1.2,-3) node[anchor=north west] {$i$};
\draw (2.7,-3) node[anchor=north west] {$k$};
\draw (-1.4,-2.2) node[anchor=north west] {$\ldots$};
\draw (0.1,-2.2) node[anchor=north west] {$\ldots$};
\draw (1.6,-2.2) node[anchor=north west] {$\ldots$};
\draw (-1.8,4.2) node[anchor=north west] {$1$};
\draw (-0.3,4.2) node[anchor=north west] {$j'$};
\draw (1.2,4.2) node[anchor=north west] {$j$};
\draw (2.7,4.2) node[anchor=north west] {$l$};
\draw (-1.4,2.) node[anchor=north west] {$\ldots$};
\draw (0.1,2.) node[anchor=north west] {$\ldots$};
\draw (1.6,2.) node[anchor=north west] {$\ldots$};
\end{tikzpicture}
$\substack{\displaystyle \stackrel{t_{i',j'}}{\longmapsto}\\ \vspace{4cm}}$
\begin{tikzpicture}[line cap=round,line join=round,>=triangle 45,x=0.5cm,y=0.5cm]
\clip(-3.5,-4.5) rectangle (4.,5.);
\draw [line width=0.4pt] (-2.,1.)-- (3.5,1.);
\draw [line width=0.4pt] (3.5,1.)-- (3.5,-1.);
\draw [line width=0.4pt] (3.5,-1.)-- (-2.,-1.);
\draw [line width=0.4pt] (-2.,-1.)-- (-2.,1.);
\draw [->,line width=0.4pt] (-1.5,1.) -- (-1.5,3.);
\draw [line width=0.4pt] (0.,1.) -- (0.,3.);
\draw [->,line width=0.4pt] (1.5,1.) -- (1.5,3.);
\draw [->,line width=0.4pt] (3.,1.) -- (3.,3.);
\draw [->,line width=0.4pt] (-1.5,-3.) -- (-1.5,-1.);
\draw [->,line width=0.4pt] (0.,-3.) -- (0.,-1.);
\draw [->,line width=0.4pt] (1.5,-3.) -- (1.5,-1.);
\draw [->,line width=0.4pt] (3.,-3.) -- (3.,-1.);
\draw (0.25,0.5) node[anchor=north west] {$G$};
\draw (-1.8,-3) node[anchor=north west] {$1$};
\draw (0.3,-3) node[anchor=north west] {$i-1$};
\draw (2.1,-3) node[anchor=north west] {$k-1$};
\draw (-1.4,-2.2) node[anchor=north west] {$\ldots$};
\draw (0.1,-2.2) node[anchor=north west] {$\ldots$};
\draw (1.6,-2.2) node[anchor=north west] {$\ldots$};
\draw (-1.8,4.2) node[anchor=north west] {$1$};
\draw (0.3,4.2) node[anchor=north west] {$j-1$};
\draw (2.1,4.2) node[anchor=north west] {$l-1$};
\draw (-1.4,2.) node[anchor=north west] {$\ldots$};
\draw (0.1,2.) node[anchor=north west] {$\ldots$};
\draw (1.6,2.) node[anchor=north west] {$\ldots$};
\draw [shift={(-1.5,-3.)},line width=0.4pt]  plot[domain=3.141592653589793:6.283185307179586,variable=\t]({1.*1.5*cos(\t r)+0.*1.5*sin(\t r)},{0.*1.5*cos(\t r)+1.*1.5*sin(\t r)});
\draw [shift={(-1.5,3.)},line width=0.4pt]  plot[domain=0.:3.141592653589793,variable=\t]({1.*1.5*cos(\t r)+0.*1.5*sin(\t r)},{0.*1.5*cos(\t r)+1.*1.5*sin(\t r)});
\draw [line width=0.4pt] (-3.,3.) -- (-3.,-3.);
\end{tikzpicture}
$\substack{\displaystyle \stackrel{t_{i-1,j-1}}{\longmapsto}\\ \vspace{4cm}}$
\begin{tikzpicture}[line cap=round,line join=round,>=triangle 45,x=0.5cm,y=0.5cm]
\clip(-3.5,-4.5) rectangle (5.,5.);
\draw [line width=0.4pt] (-2.,1.)-- (3.5,1.);
\draw [line width=0.4pt] (3.5,1.)-- (3.5,-1.);
\draw [line width=0.4pt] (3.5,-1.)-- (-2.,-1.);
\draw [line width=0.4pt] (-2.,-1.)-- (-2.,1.);
\draw [->,line width=0.4pt] (-1.5,1.) -- (-1.5,3.);
\draw [line width=0.4pt] (0.,1.) -- (0.,3.);
\draw [line width=0.4pt] (1.5,1.) -- (1.5,3.);
\draw [->,line width=0.4pt] (3.,1.) -- (3.,3.);
\draw [->,line width=0.4pt] (-1.5,-3.) -- (-1.5,-1.);
\draw [->,line width=0.4pt] (0.,-3.) -- (0.,-1.);
\draw [->,line width=0.4pt] (1.5,-3.) -- (1.5,-1.);
\draw [->,line width=0.4pt] (3.,-3.) -- (3.,-1.);
\draw (0.25,0.5) node[anchor=north west] {$G$};
\draw (-1.8,-3) node[anchor=north west] {$1$};
\draw (2.1,-3) node[anchor=north west] {$k-2$};
\draw (-1.4,-2.2) node[anchor=north west] {$\ldots$};
\draw (0.1,-2.2) node[anchor=north west] {$\ldots$};
\draw (1.6,-2.2) node[anchor=north west] {$\ldots$};
\draw (-1.8,4.2) node[anchor=north west] {$1$};
\draw (2.1,4.2) node[anchor=north west] {$l-2$};
\draw (-1.4,2.) node[anchor=north west] {$\ldots$};
\draw (0.1,2.) node[anchor=north west] {$\ldots$};
\draw (1.6,2.) node[anchor=north west] {$\ldots$};
\draw [shift={(-1.5,-3.)},line width=0.4pt]  plot[domain=3.141592653589793:6.283185307179586,variable=\t]({1.*1.5*cos(\t r)+0.*1.5*sin(\t r)},{0.*1.5*cos(\t r)+1.*1.5*sin(\t r)});
\draw [shift={(-1.5,3.)},line width=0.4pt]  plot[domain=0.:3.141592653589793,variable=\t]({1.*1.5*cos(\t r)+0.*1.5*sin(\t r)},{0.*1.5*cos(\t r)+1.*1.5*sin(\t r)});
\draw [line width=0.4pt] (-3.,3.) -- (-3.,-3.);
\draw [shift={(3.,-3.)},line width=0.4pt]  plot[domain=3.141592653589793:6.283185307179586,variable=\t]({1.*1.5*cos(\t r)+0.*1.5*sin(\t r)},{0.*1.5*cos(\t r)+1.*1.5*sin(\t r)});
\draw [shift={(3.,3.)},line width=0.4pt]  plot[domain=0.:3.141592653589793,variable=\t]({1.*1.5*cos(\t r)+0.*1.5*sin(\t r)},{0.*1.5*cos(\t r)+1.*1.5*sin(\t r)});
\draw [line width=0.4pt] (4.5,3.) -- (4.5,-3.);
\end{tikzpicture}

\vspace{-2cm}
\end{center}
For Property 3.(b), let us consider a graph $p=G$ .
As the input edge indexed by $i$ in $\sigma \cdot G\cdot \tau$ is the input edge of $G$ indexed by $\tau(i)$
and the output edge indexed by $j$ in $\sigma \cdot G\cdot \tau$ is the output edge of $G$ indexed by $\sigma^{-1}(j)$,
$G_1=t_{i,j}(\sigma\cdot G\cdot \tau)$ is the graph obtained by gluing together
the input indexed by $\tau(j)$ and the output indexed by $\sigma^{-1}(j)$, reindexing \cy{the input
% according to $\sigma_i$ 
and the output edges 
% by $\tau_j$, 
we obtain} $G_1=\cy{l_i(\sigma)}\cdot t_{\tau(i),\sigma^{-1}(j)}(G)\cdot \cy{r_j(\tau)}$.

Let us prove Property 3.(c). By Lemma  \ref{lemmeaxiomessimples}, it is enough to prove it for
$(p,p')=(G,G')$ a pair of graphs and $(i,j)=(1,1)$. In this case,
$e_i$ and $f_j$ are both edges of $G$, so $t_{1,1}(G*G')=t_{1,1}(G)*G'$.
\end{proof}
\begin{rk}
 Notice that our definition of graphs implies that $\Gr$, $\PGr$ and their decorated counterparts are \emph{not} TRAPs. Indeed take $I\in\Gr(1,1)$ the graph without vertices and that contain only a edge, an ingoing-outgoing edge. Then $tr_{1,1}(I)$ is not defined. One can add a loop component to the definition of graphs such that $tr_{1,1}(I)$ is a loop. Then $\Gr$, $\PGr$ and their decorated counterparts are \emph{unitary} TRAPs, with unit $I$. See \cite{ClFoPa20,ClFoPa21} for a presentation of these unitary TRAPs. 
\end{rk}
\begin{rk} \label{rk:subTRAPutiles}
     $\rGr$, $\rPGr$, $\rGr(X)$ and $\rPGr(X)$ admit sub-TRAPs, for example with vertices with only a prescribed number of ingoing and outgoing edges. These sub-TRAPs might be of importance in the question of renormalisability of QFTs, but this question is to be tackled in future work and we should not rigorously introduce these objects here.
    \end{rk}

\subsection{Morphisms of TRAPs and free TRAPs} \label{subsec:morphism_TRAPs_free}

As before, $X=(X(k,l))_{k,l\geqslant 0}$ is  a family of sets. It turns out that $\rPGr(X)$ is the free TRAP generated by $X$. For any $x \in X(k,l)$, we identify $x$ with the graph in $\rPGr(k,l)(X)$ with one vertex decorated by $x$, $k$ incoming edges, totally ordered according to their indices, and $l$ outgoing edges, totally ordered according to their indices. For example,   $x\in X(3,2)$  is identified with the corolla ordered graph
\begin{align} \label{eq:simple_decorated_graph}
\substack{\hspace{5mm} \:\begin{tikzpicture}[line cap=round,line join=round,>=triangle 45,x=0.7cm,y=0.7cm]
\clip(0.6,-2.1) rectangle (2.2,2.);
\draw [line width=.4pt] (0.8,0.)-- (2.2,0.);
\draw [line width=.4pt] (2.2,0.)-- (2.2,-0.5);
\draw [line width=.4pt] (2.2,-0.5)-- (0.8,-0.5);
\draw [line width=.4pt] (0.8,-0.5)-- (0.8,0.);
\draw (1.2,0.05) node[anchor=north west] {\scriptsize $x$};
\draw [->,line width=.4pt] (1.,0.) -- (1.,1.);
\draw [->,line width=.4pt] (2.,0.) -- (2.,1.);
\draw [->,line width=.4pt] (1.,-1.5) -- (1.,-0.5);
\draw [->,line width=.4pt] (1.5,-1.5) -- (1.5,-0.5);
\draw [->,line width=.4pt] (2.,-1.5) -- (2.,-0.5);
\draw (0.7,-1.4) node[anchor=north west] {\scriptsize $1$};
\draw (1.2,-1.4) node[anchor=north west] {\scriptsize $2$};
\draw (1.7,-1.4) node[anchor=north west] {\scriptsize $3$};
\draw (0.7,1.6) node[anchor=north west] {\scriptsize $1$};
\draw (1.7,1.6) node[anchor=north west] {\scriptsize $2$};
\end{tikzpicture}}
\end{align}
\begin{theo} \label{thm:freetraps}
Let $P$ be a  TRAP and $\phi=(\phi(k,l))_{k,l\geq0}$ be a map from $X$ to $P$ that is, for any $(k,l)\in \N^2$, $\phi(k,l):X(k,l)\longrightarrow P(k,l)$ is a map.
Then there exists a unique TRAP morphism $\Phi:\rPGr(X)\longrightarrow P$, sending $x$ to $\phi(x)$ for any $x\in X$. 

In other words, $\rPGr(X)$  is the free TRAP generated by $X$.
\end{theo}

\begin{rk}\label{rk:PhiPGrX}
 In practice we often have $P=X$ and $\phi=\mathrm{Id}$  which yields a map 
 \begin{equation} \label{eq:PhiPGrX}
  \Phi:\rPGr(X)\longrightarrow X
 \end{equation} 
 from decorated corolla ordered graphs to the space $X$ of decorations.
\end{rk}
\begin{proof}
We provide here a sketch of the proof,  and refer the reader to the Subsection \ref{subsec:proof_free_TRAP} for a full proof. We define $\Phi(G)$ for any graph $G\in \rPGr(k,l)(X)$ by induction on the number $N$ of internal edges of $G$. 

If $N=0$, then $G$ can be written as
\begin{equation*} 
 G=\sigma\cdot (x_1*\ldots *x_r)\cdot \tau.
\end{equation*}
 where $r$ lies in  $\N$, $(k_i,{l_i})$ lies in  $\N^2$ for any $i$, $x_i$ in $ X{({k_i,l_i})}$ and $\sigma$ in $ \sym_{k_1+\ldots+k_r}$, $\tau$ in $ \sym_{l_1+\ldots+l_r}$. We then set
\begin{equation} \label{eq:G_simple_solar}
 \Phi(G)=\sigma\cdot(\phi(x_1)*\ldots*\phi(x_k))\cdot \tau.
\end{equation}
We can prove that this does not depend on the choice of the decomposition of $G$, with the help of the TRAP axioms applied to $P$. 

Let us now assume that $\Phi(G')$ is defined for any graph with $N-1$ internal edges, for a given $N \geqslant 1$. Let $G$ be a graph with $N$ internal edges and let $e$ be one of these edges. Let $G_e$ be a graph obtained by cutting this edge in two, such that $G=t_{1,1}(G_e)$.  We then set:
\[\Phi(G)=t_{1,1}\circ \Phi(G_e).\]
We can prove that this does not depend on the choice of $e$. It can then be shown that $\Phi$ defined as above is compatible with the partial trace maps. 
\end{proof}
Since the ingoing and outgoing edges of each vertex  of a corolla ordered graph are totally ordered, each corolla ordered graph $\PGr$ is naturally acted upon by $\sym\times \sym^{op}$.
 \begin{defn}\label{defiactionsommets}
  For any corolla ordered graph $G\in\PGr$ and any vertex $v\in V(G)$, there is a natural action of $\sym_{o(v)}\times \sym_{i(v)}^{op}$ induced by the action on the  totally ordered edges in $O(v)$ and $I(v)$. The {corolla ordered} graph obtained from $G$ by the action of $(\sigma,\tau)$ on the vertex $v$ is denoted by
  \begin{equation*}
   \sigma \cdot_v G\cdot_v \tau.
  \end{equation*}
 A similar action can be built on a corolla ordered graph $G$   decorated by a family of sets $X$:
  \begin{equation*}
   \sigma \cdot_v (G,d_G)\cdot_v \tau:=(\sigma \cdot_v G\cdot_v \tau,d_G).
  \end{equation*}
 \end{defn}
Let us illustrate the above definition with a simple example.
\begin{example}
For a graph with two vertices $v$ and $w$ and only to internal edges, both from $v$ to $w$. Then
\[\substack{\displaystyle (12)\cdot_v\\ \vspace{1.5cm}}
	\begin{tikzpicture}[line cap=round,line join=round,>=triangle 45,x=0.7cm,y=0.7cm]
	\clip(0.8,-0.5) rectangle (2.2,2.5);
	\draw [line width=.4pt] (0.8,0.)-- (2.2,0.);
	\draw [line width=.4pt] (2.2,0.)-- (2.2,-0.5); %right line of lower rectangle
	\draw [line width=.4pt] (2.2,-0.5)-- (0.8,-0.5);
	\draw [line width=.4pt] (0.8,-0.5)-- (0.8,0.); %left line of lower rectangle
	\draw [line width=.4pt] (0.8,2.)-- (2.2,2.); %lower line of upper rectangle
	\draw [line width=.4pt] (2.2,2.)-- (2.2,2.5);
	\draw [line width=.4pt] (2.2,2.5)-- (0.8,2.5); %upper line of upper rectangle
	\draw [line width=.4pt] (0.8,2.5)-- (0.8,2.);
	\draw [->,line width=.4pt] (1.,0.) -- (1.,2.);
	\draw [->,line width=.4pt] (2.,0.) -- (2.,2.);
	\draw (1.2,0.) node[anchor=north west] {\scriptsize $v$};
	\draw (1.2,2.5) node[anchor=north west] {\scriptsize $w$};
	\end{tikzpicture}\,\substack{\displaystyle =\\ \vspace{1.5cm}}\,
	\begin{tikzpicture}[line cap=round,line join=round,>=triangle 45,x=0.7cm,y=0.7cm]
	\clip(0.8,-0.5) rectangle (2.2,2.5);
	\draw [line width=.4pt] (0.8,0.)-- (2.2,0.); %upper line of lower rectangle
	\draw [line width=.4pt] (2.2,0.)-- (2.2,-0.5);
	\draw [line width=.4pt] (2.2,-0.5)-- (0.8,-0.5);
	\draw [line width=.4pt] (0.8,-0.5)-- (0.8,0.);
	\draw [line width=.4pt] (0.8,2.)-- (2.2,2.);
	\draw [line width=.4pt] (2.2,2.)-- (2.2,2.5);
	\draw [line width=.4pt] (2.2,2.5)-- (0.8,2.5);
	\draw [line width=.4pt] (0.8,2.5)-- (0.8,2.);
	\draw [->,line width=.4pt] (1.,0.) -- (1.,2.);
	\draw [->,line width=.4pt] (2.,0.) -- (2.,2.);
	\draw (1.2,0.) node[anchor=north west] {\scriptsize $v$};
	\draw (1.2,2.5) node[anchor=north west] {\scriptsize $w$};
	\end{tikzpicture}
	\substack{\displaystyle \cdot_w (12)\\ \vspace{1.5cm}}
	\substack{\displaystyle =\\ \vspace{1.5cm}}
	\begin{tikzpicture}[line cap=round,line join=round,>=triangle 45,x=0.7cm,y=0.7cm]
	\clip(0.8,-0.5) rectangle (2.2,2.5);
	\draw [line width=.4pt] (0.8,0.)-- (2.2,0.);
	\draw [line width=.4pt] (2.2,0.)-- (2.2,-0.5);
	\draw [line width=.4pt] (2.2,-0.5)-- (0.8,-0.5);
	\draw [line width=.4pt] (0.8,-0.5)-- (0.8,0.);
	\draw [line width=.4pt] (0.8,2.)-- (2.2,2.);
	\draw [line width=.4pt] (2.2,2.)-- (2.2,2.5);
	\draw [line width=.4pt] (2.2,2.5)-- (0.8,2.5);
	\draw [line width=.4pt] (0.8,2.5)-- (0.8,2.);
	\draw [->,line width=.4pt] (1.,0.) -- (2.,2.);
	\draw [->,line width=.4pt] (2.,0.) -- (1.,2.);
	\draw (1.2,0.) node[anchor=north west] {\scriptsize $v$};
	\draw (1.2,2.5) node[anchor=north west] {\scriptsize $w$};
	\end{tikzpicture}\substack{\displaystyle .\\ \vspace{1.5cm}}\]
\vspace{-.5cm}
	 
	 	 \noindent In these pictures, the labelling of the edges outgoing (respectively, ingoing  {to})  {from} the vertex $v$ (respectively, $w$) are labelled from left to right.
 \end{example}
 \cy{\begin{rk}
      Corolla-ordered graphs were introduced for two reasons. First, they simplify the proof of freeness given next subsection. Second, they are a crucial part to show that the categories of TRAPs and the categories of wheeled PROPs are isomorphic. Since this result is rather distant from the application we have in mind for TRAPs, it was elected to not include it here. This construction is given in details in \cite{ClFoPa21}.
     \end{rk}
}
Note that $\Gr(X)$ is obtained from $\PGr(X)$ by forgetting the total orders on the edges, which in fact is equivalent to the trivialisation of this action of symmetric groups on incoming and outgoing edges of any vertex. Hence: 

\begin{cor} \label{cor:extension_identity}
Let $P$ be a TRAP and $\phi=(\phi(k,l))_{k,l\geq0}$ be a map from $X$ to $P$. We assume that for any $x\in X(k,l)$, for any $(\sigma,\tau)\in \sym_k\otimes \sym_l$,
\[\tau\cdot \phi(x)\cdot \sigma=\phi(x).\]
There exists a unique TRAP morphism $\Phi:\rGr(X)\longrightarrow P$, sending $x$ to $\phi(x)$ for any $x\in X$.
\end{cor}
We end this paragraph with the non {corolla ordered} counterpart of Remark   \ref{rk:PhiPGrX}:
\begin{rk}\label{rk:PhiGrX} 
In practice we often have $P=X$ and $\phi=\mathrm{Id}_P$  which yields a  map 
 \begin{equation}\label{eq:PhiGrX}
  \Phi:\rGr(X)\longrightarrow X
 \end{equation} 
 from decorated graphs to the space $X$ of decorations.
\end{rk}

\subsection{Proof of freeness of $\PGr(X)$} \label{subsec:proof_free_TRAP}

We end this Section on free TRAP by giving a detailed proof of Theorem \ref{thm:freetraps}. Let $P$ be a TRAP and let 
$\phi:X\longrightarrow P$ be a map. \\

Let us first prove the existence of $\Phi$. We define $\Phi:\rPGr(X)\longrightarrow P$ by assigning to any graph $G\in\rPGr(X)(k,l)$ an element $\Phi(G)\in P(k,l)$. We proceed by induction on the number $N$  of internal edges of $G$. If $N=0$, then $G$ can be written (non uniquely) as
\begin{align}
\label{eqdecompo} G=\sigma\cdot(G_1*\ldots *G_r)\cdot \tau,
\end{align}
where $r$ in $\N$ is unique, $(k_i,l_i)$ in $ \N^2$ for any $i$, unique up to a permutation, $\sigma $ in $\sym_{k_1+\ldots+k_r}$, $\tau\in \sym_{l_1+\ldots+l_r}$ and $G_i \in X(k_i,l_i)$ for any $i$. 
For any graph $H$ we set in this case $H^{*0}=I_0=\emptyset$, the empty graph. We then put:
\[\Phi(G)=\sigma\cdot(\phi_{k_1,l_1}(G_1)*\ldots*\phi_{k_r,l_r}(G_r))\cdot \tau,\]
where as before $x^{*0}=I_0$ (the unit for horizontal concatenation in the image $P$ of $\Phi$) for any $x\in P$.\\

Let us assume now that $\Phi(G')$ is defined for any graph with $N-1$ internal edges, for a given $N \geqslant 1$.
Let $G$ be a graph with $N$ internal edges and let $e$ be one of these edges. 
Let $G_e$ be a graph obtained by cutting this edge in two:
\begin{itemize}
\item $V(G_e)=V(G)$.
\item $E(G_e)=E(G)\setminus \{e\}$, $I(G_e)=I(G)\sqcup \{e\}$, $O(G_e)=O(G)\sqcup \{e\}$.
\item $s_{G_e}=s_G$ and $t_{G_e}=t_G$.
\item For any $e'\in I(G_e)$, for any $f'\in O(G_e)$:
\begin{align*}
\alpha_{G_e}(e')&=\begin{cases}
1\mbox{ if }e'=e,\\
\alpha_{G}(e')+1\mbox{ if }e'\neq e,
\end{cases}
&\beta_{G_e}(f')&=\begin{cases}
1\mbox{ if }f'=e,\\
\beta_{G}(f')+1\mbox{ if }f'\neq e.
\end{cases}
\end{align*}
\end{itemize}
Notice that since $G$ lies in $\rPGr(X)$ (that is $IO(G)=\emptyset$) then $G_e$ also lies in $\rPGr(X)$. We have $G=t_{1,1}(G_e)$ and $G_e$ has $N-1$ internal edges. We then put:
\begin{equation}
\label{eqdefiarete}\Phi(G):=t_{1,1}\circ \Phi(G_e).
\end{equation}

\begin{lemma}
The map $\Phi$ is well-defined. Moreover, for any $\in \rPGr(X)(k,l)$ with $(k,l)\in \N^2$,
for any $\sigma \in \sym_l$, for any $\tau \in \sym_k$,
\[\Phi(\sigma\cdot G\cdot \tau)=\sigma\cdot \Phi(G)\cdot \tau.\] 
\end{lemma}

\begin{proof}
We proceed by induction on the number $N$ of internal vertices. For $N=0$, we have to show that
$\Phi(G)$ does not depend on the choice of the decomposition \eqref{eqdecompo} of $G$. 
Such a decomposition is determined  modulo a permutation of the vertices
and of the choice of $\sigma$ and $\tau$. Thus,  we can go from one decomposition of $G$ to any other one by means of a finite number of steps among the following two types:
\begin{enumerate}
\item We consider two decompositions of $G$ of the form
\begin{align*}
G&=\sigma\cdot(G_1*\ldots* G_i*G_{i+1} *\ldots *G_r)\cdot \tau,\\
G&=\sigma'\cdot(G_1*\ldots*G_{i+1}* G_i *\ldots *G_r)\cdot \tau',
\end{align*}
with
\begin{align*}
\sigma'&=\sigma(\mathrm{Id}_{l_1+\ldots+l_{i-1}}\otimes c_{l_i,l_{i+1}}\otimes \mathrm{Id}_{l_{i+2}+\ldots+l_r}),\\
\tau'&=(\mathrm{Id}_{k_1+\ldots+k_{i-1}}\otimes c_{k_{i+1},k_i}\otimes \mathrm{Id}_{k_{i+2}+\ldots+k_r})\tau.
\end{align*}
Then, by commutativity of $*$\cy{, and using $c_{l,m}^{-1}=c_{m,l}$}:
\begin{align*}
&\sigma'\cdot(\phi(G_1)*\ldots*\phi(G_r))\cdot \tau'\\
&=\sigma\cdot (\phi(G_1)*\ldots * c_{l_i,l_{i+1}}\cdot(\phi(G_{i+1})*\phi(G_i))\cdot c_{k_{i+1},k_i}
*\ldots*\phi(G_r))\cdot \tau\\
&=\sigma\cdot (\phi(G_1)*\ldots *\phi(G_i)*\phi(G_{i+1})*\ldots *\phi(G_r))\cdot \tau.
\end{align*}
\item We consider two decompositions of $G$ of the form
\begin{align*}
G&=\sigma\cdot(G_1*\ldots*G_r)\cdot \tau,\\
G&=\sigma'\cdot(\cy{G'_1*\ldots*G'_r})\cdot \tau',
\end{align*}
with
\begin{align*}
\sigma'&=\sigma (\sigma_1\otimes \ldots \otimes \sigma_r),&
\tau'&=(\tau_1\otimes \ldots \otimes \tau_r)\ty{\tau},&\cy{G'_i=\sigma_i^{-1}\cdot G_i\cdot\tau^{-1}_i,}
\end{align*}
with $\sigma_i\in \sym_{k_i}$ and $\tau_i\in \sym_{l_i}$ if $i\geqslant 1$.
Using the commutativity of $*$ and the invariance of the $x_{k,l}$, we find
\begin{align*}
\sigma'\cdot (\cy{\phi(G'_1)*\ldots*\phi(G'_r)})\cdot \tau'
&=\sigma \cdot (\sigma_1\cdot \cy{\phi(G'_1)}\cdot \tau_1*\ldots*\sigma_r\cdot \cy{\phi(G'_r)}\cdot \tau_r)\cdot \tau\\
&=\sigma\cdot (\phi(G_1)*\ldots*\phi(G_r))\cdot \tau.
\end{align*}
\end{enumerate}
Hence, $\Phi(G)$ is well-defined \cy{for $G$ without internal vertices}. Moreover, for any  $\tau' \in \sym_k$, $\sigma'\in \sym_l$,  a decomposition  of $G$ of the form
\[G=\sigma\cdot(G_1*\ldots* G_r)\cdot \tau,\]
give rise to a decomposition of $G'=\sigma'\cdot G\cdot \tau'$ given by
\[\sigma'\sigma\cdot(G_1*\ldots*G_r)\cdot \tau\tau',\]
and, by definition of $\Phi(G')$:
\begin{align*}
\Phi(G')&=\sigma'\sigma\cdot(\phi(G_1)*\ldots*\phi(G_r))\cdot \tau\tau'\\
&=\sigma'\cdot(\sigma\cdot(\phi(G_1)*\ldots*\phi(G_r))\cdot \tau)\tau'\\
&=\sigma'\cdot \Phi(G)\cdot \tau'
\end{align*}
\cy{since $\phi$ is a morphism of $\sym\times\sym^{op}$-modules.}

Let us assume the result at rank $N-1$ and let $G$ be a graph with $N$ internal edges. Let us prove that 
$\Phi(G)$ defined in (\ref{eqdefiarete}) does not depend on the choice of $e$. If $e'$ is another internal edge of $G$,
then:
\[(G_e)_{e'}=(12)\cdot (G_{e'})_e\cdot (12),\]
which implies, by definition of $\Phi(G_e)$ and $\Phi(G_{e'})$:
\begin{align*}
t_{1,1}\circ \Phi(G_e)&=t_{1,1}\circ t_{1,1}\circ \Phi((G_e)_{e'})\\
&=t_{1,1}\circ t_{1,1} \circ ((12)\cdot \Phi((G_{e'})_e)\cdot (12))\quad \text{\cy{by the induction hypothesis}}\\
&=t_{1,1}\circ t_{2,2}\circ \Phi((G_{e'})_e)\quad\cy{\text{by equivariance of the }t_{i,j}}\\
&=t_{1,1}\circ t_{1,1}\circ \Phi((G_{e'})_e)\quad\cy{\text{by commutativity of the }t_{i,j}}\\
&=t_{1,1}\circ \Phi(G_{e'}).
\end{align*}
So $\Phi(G)$ is well-defined. Let $\sigma \in \sym_k$ and $\tau\in \sym_l$. Then:
\[(\sigma\cdot G\cdot \tau)_e=((1)\otimes \sigma)\cdot (G_e)\cdot ((1)\otimes \tau),\]
so:
\begin{align*}
\Phi(\sigma \cdot G\cdot \tau)&=t_{1,1}\circ \Phi((\sigma\cdot G\cdot \tau)_e)\\
&=t_{1,1}\left(((1)\otimes \sigma)\cdot \Phi(G_e)\cdot ((1)\otimes \tau)\right) \quad\cy{\text{by the induction hypothesis}}\\
&=\cy{l_1((1)\otimes \sigma)}\cdot t_{1,1}\circ \Phi(G_e)\cdot \cy{r_1((1)\otimes \tau)}\quad\cy{\text{by equivariance of the }t_{i,j}}\\
&=\sigma\cdot \Phi(G)\cdot \tau
\end{align*}
\cy{by definition of $l_i$ and $r_j$.}
% where, for $\sigma\in\sym_k$ we use $\sigma_i$ for the permutation in $\sym_{k-1}$ defined by
% \begin{equation*}
%   \sigma_i(j) = \begin{cases}
%                & \sigma(j) \quad\text{if }j\leq i-1, \\
%                & \sigma(j-1)\quad\text{if }j\geq i. 
%               \end{cases} \qedhere
% \end{equation*}
\end{proof}

We have therefore defined a map $\Phi:\rPGr(X)\longrightarrow P$ in the non-unitary case,
compatible with the action of the symmetric groups. It remains to prove that $\Phi$ is compatible with the horizontal concatenation
$*$ and with the partial trace maps.

\begin{lemma}
For any graphs $G$, $G'$,
\[\Phi(G*G')=\Phi(G)*\Phi(G').\]
\end{lemma}

\begin{proof}
We proceed by induction on the number $N$ of internal edges of $G*G'$. If $N=0$, we put:
\begin{align*}
G&=\sigma\cdot (G_1*\ldots*G_r)\cdot \tau,\\
G'&=\sigma'\ty{\cdot} (G'_1*\ldots*G'_{r'})\cdot \tau'.
\end{align*}
We obtain:
\begin {equation*}
 G*G'=(\sigma \otimes \sigma') 
\cdot (G_1*\ldots *G_r*G'_1*\ldots *G'_{r'})\cdot
(\tau\otimes \tau'),
\end {equation*}
which gives, by \cy{equivarience} of $*$:
\begin{align*}
\Phi(G*G')&=(\sigma \otimes \sigma')\cdot (\phi(G_1)*\ldots *\phi(G_r)*\phi(G'_1)*\ldots *\phi(G'_{r'}))\cdot(\tau\otimes \tau')\\
&=(\sigma\cdot(\phi(G_1)*\ldots*\phi(G_r))\cdot \tau)*\sigma'\cdot(\phi(G'_1)*\ldots*\phi(G'_{r'}))\cdot \tau'\\
&=\Phi(G)*\Phi(G').
\end{align*}
If $N\geqslant 1$, let us assume that the result holds at rank $N-1$ and take an internal edge $e$ of $G*G'$. If $e$ is an internal edge of $G$, then
$(G*G')_e=G_e*G'$, and:
\begin{align*}
 \Phi(G*G')&=t_{1,1}\circ \Phi((G*G')_e)\\
 &=t_{1,1}\circ \Phi(G_e*G')\\
 &=t_{1,1}(\Phi(G_e)*\Phi(G'))\quad\text{by the induction hypothesis}\\
 &=\left(t_{1,1}\circ \Phi(G_e)\right)*\Phi(G') \quad\text{by Axiom 3.(c) of Definition \ref{def:Trap}}\\
 &=\Phi(G)*\Phi(G').
\end{align*}
If $e$ is an internal edge of $G'$, we obtain similarly that $\Phi(G'*G)=\Phi(G')*\Phi(G)$.
The result then 
follows from the \cy{commutativity of $*$.}
%(Axiom $2.(d)$ of Definition \ref{def:Trap}).
\end{proof}
We still need to prove the compatibility of $\Phi$ with the partial trace maps. 
\begin{lemma}
Let $G\in \rPGr(k,l)$ with $(k,l)\in \N^2$, $i\in [k]$ and $j\in [l]$. Then:
\[t_{i,j}\circ \Phi(G)=\Phi\circ t_{i,j}(G).\]
\end{lemma}

\begin{proof}
By Lemma \ref{lemmemorphismes}, it is enough to prove that $\Phi$ is compatible with $t_{1,1}$. Let $G\in\rPGr(X)(k,l)$, $e_1=\alpha^{-1}(1)\in I(G)$ and $f_1=\beta^{-1}(1)\in O(G)$. We set $G'=t_{1,1}(G)$ and $e=\{e_1,f_1\}$ to be the edge of $G'$ created in the process. Notice that since $G\in\rPGr(X)$ we have $G'\in\rPGr(X)$. Then $e\in E(G')$ and $G'_e=G$. By \cy{definition of $G'$ and} construction of $\Phi(G')$:
\[\Phi\circ t_{1,1}(G)=\Phi(G')=t_{1,1}\circ \Phi(G'_e)=t_{1,1}\circ \Phi(G).\]
So $\Phi$ is compatible with the partial trace maps, both in the unitary and non-unitary cases.
\end{proof}

We have proved the existence of $\Phi$. It remains to prove the unicity. \cy{Let $\Psi:\rPGr(X)\longrightarrow P$ be a TRAP morphism extending $\phi$. We prove $\Psi=\Phi$ by induction on the number of internal edges. First, $\Psi$ and $\Phi$ must coincide on graphs with no internal vertices since they are both extentions of $\phi$ and morphisms of TRAPs.

Assume now that $\Psi$ and $\Phi$ coincide on solar corolla oriented graphs with $N$ internal edges or less. Let $G$ be such a graph with $N+1$ internal edges. Then, for $e\in E(G)$.
\begin{equation*}
 \Phi(G)=t_{1,1}\circ\Phi(G_e)=t_{1,1}\circ\Psi(G_e)=\Psi\circ t_{1,1}(G_e)=\Psi(G).
\end{equation*}
We have used the definition of $\Phi$, the induction hypothesis, the fact that $\Psi$ is a morphism of TRAPs and the definition of $G_e$. Thus $\Psi=\Phi$ and $\Phi$ is unique.
}
% Any solar graph can be obtained from graphs with only one vertex, with the help of
% the horizontal concatenation and the partial trace maps, which allow to create the missing internal edges. 
% Hence, $\rPGr(X)$ is generated by $X$ as a TRAP, which implies the unicity of $\Phi$. 

This ends the proof of Theorem \ref{thm:freetraps}.

\section{Compositions, generalised trace and convolution}
\label{sec:applicationsTRAPs}

We will now see how the definition of TRAPs and the universal property of the TRAP of graphs imply that they carry more structures.

\subsection{Vertical concatenation in a TRAP} \label{subsec:vert_conc_TRAP}

% Ca pourrait être dans la section 4, à la fin???? Le problème c'est la représentation graphique...

In Subsection \ref{subsec:prop_graph}, we showed that $\Gr$ admits a PROP structure. It is clear that $\rGr$ is a sub-PROP and that the corolla ordered graphs as well as the decorated graphs inherit this PROP structure, up to the lack of unity. We therefore have both PROP and TRAP structures for (unitary) TRAP.

The vertical concatenation on $\rGr$ can be described in terms of the horizontal concatenation and of the partial trace maps:
If $G$ is a solar graph with $k$ inputs and $l$ outputs, and $G'$ a solar graph with $l$ inputs and $m$ outputs, then:
\begin{align*}
t_{k+1,1}\circ \ldots \circ t_{k+l-1,l-1}\circ t_{k+l,l}(G*G')&=\cy{G'\circ G},
\end{align*}
or, graphically:
\begin{center}
\begin{tikzpicture}[line cap=round,line join=round,>=triangle 45,x=0.5cm,y=0.5cm]
\clip(-2.5,-4.) rectangle (7.,5.);
\draw [line width=0.4pt] (-2.,1.)-- (0.5,1.);
\draw [line width=0.4pt] (0.5,1.)-- (0.5,-1.);
\draw [line width=0.4pt] (0.5,-1.)-- (-2.,-1.);
\draw [line width=0.4pt] (-2.,-1.)-- (-2.,1.);
\draw [line width=0.4pt] (-1.5,1.) -- (-1.5,2.);
\draw [line width=0.4pt] (0.,1.) -- (0.,2.);
\draw [->,line width=0.4pt] (-1.5,-3.) -- (-1.5,-1.);
\draw [->,line width=0.4pt] (0.,-3.) -- (0.,-1.);
\draw (-1.25,0.5) node[anchor=north west] {$G$};
\draw (-1.8,-3) node[anchor=north west] {$1$};
\draw (-0.3,-3) node[anchor=north west] {$k$};
\draw (-1.4,-2.2) node[anchor=north west] {$\ldots$};
\draw (-1.4,2.) node[anchor=north west] {$\ldots$};
\draw [shift={(0.,2.)},line width=0.4pt]  plot[domain=0.:3.141592653589793,variable=\t]({1.*1.5*cos(\t r)+0.*1.5*sin(\t r)},{0.*1.5*cos(\t r)+1.*1.5*sin(\t r)});
\draw [shift={(3.,-2.)},line width=0.4pt]  plot[domain=3.141592653589793:6.283185307179586,variable=\t]({1.*1.5*cos(\t r)+0.*1.5*sin(\t r)},{0.*1.5*cos(\t r)+1.*1.5*sin(\t r)});
\draw [line width=0.4pt] (1.5,2.) -- (1.5,-2.);
\draw [shift={(1.5,2.)},line width=0.4pt]  plot[domain=0.:3.141592653589793,variable=\t]({1.*1.5*cos(\t r)+0.*1.5*sin(\t r)},{0.*1.5*cos(\t r)+1.*1.5*sin(\t r)});
\draw [shift={(4.5,-2.)},line width=0.4pt]  plot[domain=3.141592653589793:6.283185307179586,variable=\t]({1.*1.5*cos(\t r)+0.*1.5*sin(\t r)},{0.*1.5*cos(\t r)+1.*1.5*sin(\t r)});
\draw [line width=0.4pt] (3.,2.) -- (3.,-2.);
\draw [line width=0.4pt] (4.,1.)-- (6.5,1.);
\draw [line width=0.4pt] (6.5,1.)-- (6.5,-1.);
\draw [line width=0.4pt] (6.5,-1.)-- (4.,-1.);
\draw [line width=0.4pt] (4.,-1.)-- (4.,1.);
\draw [->,line width=0.4pt] (4.5,1.) -- (4.5,3.);
\draw [->,line width=0.4pt] (6.,1.) -- (6.,3.);
\draw [->,line width=0.4pt] (4.5,-2.) -- (4.5,-1.);
\draw [->,line width=0.4pt] (6.,-2.) -- (6.,-1.);
\draw (4.75,0.5) node[anchor=north west] {$G'$};
\draw (4.2,4.2) node[anchor=north west] {$1$};
\draw (5.7,4) node[anchor=north west] {$m$};
\draw (4.4,-2.2) node[anchor=north west] {$\ldots$};
\draw (4.6,2.) node[anchor=north west] {$\ldots$};
\end{tikzpicture}
$\substack{\displaystyle =\\ \vspace{3cm}}$
\begin{tikzpicture}[line cap=round,line join=round,>=triangle 45,x=0.5cm,y=0.5cm]
\clip(-2.5,-4.) rectangle (1.,8.);
\draw [line width=0.4pt] (-2.,1.)-- (0.5,1.);
\draw [line width=0.4pt] (0.5,1.)-- (0.5,-1.);
\draw [line width=0.4pt] (0.5,-1.)-- (-2.,-1.);
\draw [line width=0.4pt] (-2.,-1.)-- (-2.,1.);
\draw [->,line width=0.4pt] (-1.5,1.) -- (-1.5,3.);
\draw [->,line width=0.4pt] (0.,1.) -- (0.,3.);
\draw [->,line width=0.4pt] (-1.5,-3.) -- (-1.5,-1.);
\draw [->,line width=0.4pt] (0.,-3.) -- (0.,-1.);
\draw (-1.25,0.5) node[anchor=north west] {$G$};
\draw (-1.8,-3) node[anchor=north west] {$1$};
\draw (-0.3,-3) node[anchor=north west] {$k$};
\draw (-1.4,-2.2) node[anchor=north west] {$\ldots$};
\draw (-1.4,2.) node[anchor=north west] {$\ldots$};
\draw [line width=0.4pt] (-2.,5.)-- (0.5,5.);
\draw [line width=0.4pt] (0.5,5.)-- (0.5,3.);
\draw [line width=0.4pt] (0.5,3.)-- (-2.,3.);
\draw [line width=0.4pt] (-2.,3.)-- (-2.,5.);
\draw [->,line width=0.4pt] (-1.5,5.) -- (-1.5,7.);
\draw [->,line width=0.4pt] (0.,5.) -- (0.,7.);
\draw (-1.25,4.5) node[anchor=north west] {$G'$};
\draw (-1.8,8.2) node[anchor=north west] {$1$};
\draw (-0.3,8.) node[anchor=north west] {$m$};
\draw (-1.4,6.) node[anchor=north west] {$\ldots$};
\end{tikzpicture}
\end{center}

\vspace{-1.5cm}

This construction can be generalised from TRAPs of graphs to an arbitrary TRAP.
\begin{defiprop}\label{propverticalconcatenation}
Let $P$ be a TRAP. We define a vertical concatenation\footnote{
	When there is a risk of confusion, we will write $\circ_P$ for the vertical concatenation of a given TRAP $P$.}  $\circ$ in the following way:
\begin{align*}
&\forall (k,l,m)\in \N^3,\:\forall q\in P(l,m),\:\forall p\in P(k,l),&
p\circ q&:=t_{k+1,1}\circ \ldots \circ t_{k+l-1,l-1}\circ t_{k+l,l}(p*q).
\end{align*}
This operation is associative: for any $(k,l,m,n) $ in  $\N^4$, for any $(p,q,r)$ in $ P(k,l)\times P(l,m)\times P(n,k)$,
\begin{equation}\label{eq:assovertconcTRAP}
r\circ (q\circ p)=(r\circ q)\circ p.
\end{equation}
\end{defiprop}
\begin{rk} \label{rk:TRAP_PROP} 
 We therefore have that non-unitary TRAPs carry a non-unitary version of the PROP structure. In the unitary case we obtain a full-fledge PROP; see \cite[Definition-Proposition 4.1.1]{ClFoPa21}.
\end{rk}
\begin{proof} 
Recall that in Subsection \ref{subsec:morphism_TRAPs_free} we identified any element $p$ of the decorating set and the solar graph with one vertex decorated by $p$ (see Equation \eqref{eq:simple_decorated_graph}). 
Let $\alpha:\rPGr(P)\longrightarrow P$ be the unique TRAP morphism such that $\alpha(p)=p$ for any $p\in P$ whose existence follows from Theorem \ref{thm:freetraps} and more specifically from the case \ty{detailed} in Remark \ref{rk:PhiPGrX}.
This is therefore a surjective TRAP morphism. As $\alpha$ respects the horizontal concatenation and the partial trace maps, 
for any graphs $G,G' \in \rPGr(P)$ such that $G\circ G'$ is well-defined, $\alpha(G)\circ\alpha(G')$ is also well-defined and
\[\alpha(G)\circ \alpha(G')=\alpha(G\circ G').\]
Since the vertical concatenation is clearly associative in $\rPGr(P)$, the vertical concatenation is associative in $P$.
\end{proof}
\begin{rk}
 In \cite[Paragraph 3.3.3]{JY15}, the authors also define {\bf partial vertical concatenations}, which could also be built in this framework by gluing only a subset of the outputs to the inputs with the partial trace maps. I do not pursue further this course since such partial vertical concatenations will play no role here.
\end{rk}
Let us illustrate the link between TRAPs and (non-unitary) PROPs in a example built from the TRAP of Proposition \ref{prop:Homfr}.
\begin{example}
Let $V$ be a vector space and let $f=\theta(v_1\ldots v_l\otimes f_1\ldots f_k)\in \Hom_V^{fr}(k,l)$,
$g=\theta(w_1\ldots w_m\otimes g_1\ldots g_l) \in \Hom_V^{fr}(l,m)$. Then, denoting by $\bullet$ the vertical
concatenation of $\Hom_V^{fr}$:
\begin{align*}
g\bullet f&=g_1(v_1)\ldots g_l(v_l)\theta(w_1\ldots w_m\otimes f_1\ldots f_k)\\
&=\theta(w_1\ldots w_m\otimes g_1\ldots g_l) \circ \theta(v_1\ldots v_l\otimes f_1\ldots f_k)\\
&=g\circ f.
\end{align*}
Hence, the vertical concatenation induced by the TRAP structure is the usual composition of linear maps. 
If $V$ is not finite-dimensional, this composition does not have a unit, as $\mathrm{Id}_V$ is not of finite rank.
\end{example}

We end  this subsection with a simple yet important proposition.
\begin{prop} \label{prop:morphismPROPfromTRAP}
For any two   TRAPs $P=(P(k,l))_{k,l\geq0}$ and $Q$, any TRAP morphism $\phi=(\phi(k,l))_{k,l\geq0}:P\longrightarrow Q$ is compatible with the vertical concatenations of $P$ and $Q$.
\end{prop}

\begin{proof}
 We  need to show that for any TRAPs $P$ and $Q$ and any TRAP morphism $\phi:P\longrightarrow Q$ as in the statement of the proposition, for any $(k,l,m)$ in $\N^3$, $p_1$ in $P(k,l)$ and $p_2$ in $P(l,m)$ we have
 \begin{equation*}
  \phi(k,m)(p_2\circ_P p_1) = \phi(k,l)(p_1) \circ_Q\phi(l,m)(p_2).
 \end{equation*}
 Using the definition of the vertical concatenation in the TRAP $P$ and the \cy{fourth} property of the Definition \ref{defn:trap_morphism} of morphisms of TRAPs we have
 \begin{equation*}
  \phi(k,m)(p_2\circ_P p_1) = t_{k+1,1}^Q\circ\cdots\circ t_{k+l,l}^Q[\phi(k+l,m+l)(p_1*p_2)]
 \end{equation*}
 with $t_{i,j}^Q$ the  partial trace maps  of the TRAP $Q$. Then using the \cy{third} property of Definition \ref{defn:trap_morphism} 
 we obtain:
 \begin{equation*}
  \phi(k,m)(p_2\circ_P p_1) = t_{k+1,1}^Q\circ\cdots\circ t_{k+l,l}^Q[\phi(k,l)(p_1)*\phi(l,m)(p_2)] 
  = \phi(k,l)(p_1) \circ_Q\phi(l,m)(p_2).
  \end{equation*}
\end{proof}

\subsection{The generalised trace  on  a TRAP}

If $G$ is a solar graph with the same number of inputs and outputs, we define its generalised trace by, roughly speaking, grafting
any of its input to the output with the same index:
\begin{center}
\begin{tikzpicture}[line cap=round,line join=round,>=triangle 45,x=0.5cm,y=0.5cm]
\clip(-2.5,-4.) rectangle (1.,4.);
\draw [line width=0.4pt] (-2.,1.)-- (0.5,1.);
\draw [line width=0.4pt] (0.5,1.)-- (0.5,-1.);
\draw [line width=0.4pt] (0.5,-1.)-- (-2.,-1.);
\draw [line width=0.4pt] (-2.,-1.)-- (-2.,1.);
\draw [->,line width=0.4pt] (-1.5,1.) -- (-1.5,3.);
\draw [->,line width=0.4pt] (0.,1.) -- (0.,3.);
\draw [->,line width=0.4pt] (-1.5,-3.) -- (-1.5,-1.);
\draw [->,line width=0.4pt] (0.,-3.) -- (0.,-1.);
\draw (-1.25,0.5) node[anchor=north west] {$G$};
\draw (-1.9,-3) node[anchor=north west] {$1$};
\draw (-0.3,-3) node[anchor=north west] {$k$};
\draw (-1.4,-2.2) node[anchor=north west] {$\ldots$};
\draw (-1.9,4.2) node[anchor=north west] {$1$};
\draw (-0.3,4.2) node[anchor=north west] {$k$};
\draw (-1.4,2.) node[anchor=north west] {$\ldots$};
\end{tikzpicture}
$\substack{\displaystyle \longmapsto\\ \vspace{3cm}}$
\begin{tikzpicture}[line cap=round,line join=round,>=triangle 45,x=0.5cm,y=0.5cm]
\clip(-2.5,-4.) rectangle (3.5,4.);
\draw [line width=0.4pt] (-2.,1.)-- (0.5,1.);
\draw [line width=0.4pt] (0.5,1.)-- (0.5,-1.);
\draw [line width=0.4pt] (0.5,-1.)-- (-2.,-1.);
\draw [line width=0.4pt] (-2.,-1.)-- (-2.,1.);
\draw [line width=0.4pt] (-1.5,1.) -- (-1.5,2.);
\draw [line width=0.4pt] (0.,1.) -- (0.,2.);
\draw [line width=0.4pt] (-1.5,-2.) -- (-1.5,-1.);
\draw [line width=0.4pt] (0.,-2.) -- (0.,-1.);
\draw (-1.25,0.5) node[anchor=north west] {$G$};
\draw (-1.4,-2.2) node[anchor=north west] {$\ldots$};
\draw (-1.4,2.) node[anchor=north west] {$\ldots$};
\draw [shift={(1.5,-2.)},line width=0.4pt]  plot[domain=3.141592653589793:6.283185307179586,variable=\t]({1.*1.5*cos(\t r)+0.*1.5*sin(\t r)},{0.*1.5*cos(\t r)+1.*1.5*sin(\t r)});
\draw [shift={(0.,-2.)},line width=0.4pt]  plot[domain=3.141592653589793:6.283185307179586,variable=\t]({1.*1.5*cos(\t r)+0.*1.5*sin(\t r)},{0.*1.5*cos(\t r)+1.*1.5*sin(\t r)});
\draw [shift={(1.5,2.)},line width=0.4pt]  plot[domain=0.:3.141592653589793,variable=\t]({1.*1.5*cos(\t r)+0.*1.5*sin(\t r)},{0.*1.5*cos(\t r)+1.*1.5*sin(\t r)});
\draw [shift={(0.,2.)},line width=0.4pt]  plot[domain=0.:3.141592653589793,variable=\t]({1.*1.5*cos(\t r)+0.*1.5*sin(\t r)},{0.*1.5*cos(\t r)+1.*1.5*sin(\t r)});
\draw [->,line width=0.4pt] (1.5,2.) -- (1.5,-2.);
\draw [->,line width=0.4pt] (3.,2.) -- (3.,-2.);
\end{tikzpicture}

\vspace{-1.5cm}
\end{center}
This construction preserve corolla ordered graphs and $X$-decorated graphs.
Moreover, we can describe this construction in terms of the partial trace maps: if $G\in {\rPGr}(X)(k,k)$, then
its generalized trace is 
\[t_{1,1}\circ \ldots \circ t_{k,k}(G)=t_{1,1}\circ \ldots \circ t_{1,1}(G).\]
This formulas have a meaning for any TRAP:

\begin{defn} \label{defi:generalised_Traces}
Let $P$ be a TRAP. For any $p$ in $ P(k,k)$, with $k$ in $\N$,  the generalised trace on $P$ is defined as:
\[\mathrm{Tr}_P(p):=t_{1,1}\circ \ldots \circ t_{k,k}(p)\in P(0,0).\]
\end{defn}
In the case of the TRAPs $\rPGr(X)$, we shall simply write $\mathrm{Tr}$ instead of $\mathrm{Tr}_{\rPGr(X)}$.

Let us now state some properties of these generalised traces.
\begin{prop} \label{prop:generalised_traces}
Let $P$ be a TRAP.
\begin{enumerate}
\item For any $(k,l)$ in $ \N^2$, for any $(p,q)$ in $P(k,l)\times P(l,k)$,
\begin{align*}
\mathrm{Tr}_P(p\circ q)&=\mathrm{Tr}_P(q\circ p),
\end{align*}
which justifies the terminology "trace".
\item For any $(k,l)$ in $  \N^2$, for any $(p,q)$ in $ P(k,k)\times P(l,l)$,
\begin{align*}
\mathrm{Tr}_P(p*q)&=\mathrm{Tr}_P(p)*\mathrm{Tr}_P(q).
\end{align*}
\end{enumerate}
\end{prop}
\begin{proof}
Let $\alpha:{\rPGr}(P)\longrightarrow P$ be, as before in the proof of Definition-Proposition \ref{propverticalconcatenation}, the unique TRAP morphism which extends the identity map on $P$.
Since $\alpha$ respects the partial trace maps, for any graph $G\in {\rPGr}(P)(k,k)$,
\[\alpha \circ \mathrm{Tr}(G)=\mathrm{Tr}_P\circ \alpha(G).\]
Let $p,q\in P(k,k)$. In $\rGr(P)$, $\mathrm{Tr}(q\circ p)$ and  $\mathrm{Tr}(p\circ q)$ are represented respectively by the graphs
\begin{align*}
&\begin{tikzpicture}[line cap=round,line join=round,>=triangle 45,x=0.5cm,y=0.5cm]
\clip(-2.5,-4.) rectangle (3.5,8.);
\draw [line width=0.4pt] (-2.,1.)-- (0.5,1.);
\draw [line width=0.4pt] (0.5,1.)-- (0.5,-1.);
\draw [line width=0.4pt] (0.5,-1.)-- (-2.,-1.);
\draw [line width=0.4pt] (-2.,-1.)-- (-2.,1.);
\draw [->,line width=0.4pt] (-1.5,1.) -- (-1.5,3.);
\draw [->,line width=0.4pt] (0.,1.) -- (0.,3.);
\draw [line width=0.4pt] (-1.5,-2.) -- (-1.5,-1.);
\draw [line width=0.4pt] (0.,-2.) -- (0.,-1.);
\draw (-1.25,0.5) node[anchor=north west] {$p$};
\draw (-1.4,-2.2) node[anchor=north west] {$\ldots$};
\draw (-1.4,2.) node[anchor=north west] {$\ldots$};
\draw [line width=0.4pt] (-2.,5.)-- (0.5,5.);
\draw [line width=0.4pt] (0.5,5.)-- (0.5,3.);
\draw [line width=0.4pt] (0.5,3.)-- (-2.,3.);
\draw [line width=0.4pt] (-2.,3.)-- (-2.,5.);
\draw [->,line width=0.4pt] (-1.5,5.) -- (-1.5,6.);
\draw [->,line width=0.4pt] (0.,5.) -- (0.,6.);
\draw (-1.25,4.5) node[anchor=north west] {$q$};
\draw (-1.4,6.) node[anchor=north west] {$\ldots$};
\draw [shift={(1.5,-2.)},line width=0.4pt]  plot[domain=3.141592653589793:6.283185307179586,variable=\t]({1.*1.5*cos(\t r)+0.*1.5*sin(\t r)},{0.*1.5*cos(\t r)+1.*1.5*sin(\t r)});
\draw [shift={(0.,-2.)},line width=0.4pt]  plot[domain=3.141592653589793:6.283185307179586,variable=\t]({1.*1.5*cos(\t r)+0.*1.5*sin(\t r)},{0.*1.5*cos(\t r)+1.*1.5*sin(\t r)});
\draw [shift={(1.5,6.)},line width=0.4pt]  plot[domain=0.:3.141592653589793,variable=\t]({1.*1.5*cos(\t r)+0.*1.5*sin(\t r)},{0.*1.5*cos(\t r)+1.*1.5*sin(\t r)});
\draw [shift={(0.,6.)},line width=0.4pt]  plot[domain=0.:3.141592653589793,variable=\t]({1.*1.5*cos(\t r)+0.*1.5*sin(\t r)},{0.*1.5*cos(\t r)+1.*1.5*sin(\t r)});
\draw [->,line width=0.4pt] (1.5,6.) -- (1.5,-2.);
\draw [->,line width=0.4pt] (3.,6.) -- (3.,-2.);
\end{tikzpicture}& 
\begin{tikzpicture}[line cap=round,line join=round,>=triangle 45,x=0.5cm,y=0.5cm]
\clip(-2.5,-4.) rectangle (3.5,8.);
\draw [line width=0.4pt] (-2.,1.)-- (0.5,1.);
\draw [line width=0.4pt] (0.5,1.)-- (0.5,-1.);
\draw [line width=0.4pt] (0.5,-1.)-- (-2.,-1.);
\draw [line width=0.4pt] (-2.,-1.)-- (-2.,1.);
\draw [->,line width=0.4pt] (-1.5,1.) -- (-1.5,3.);
\draw [->,line width=0.4pt] (0.,1.) -- (0.,3.);
\draw [line width=0.4pt] (-1.5,-2.) -- (-1.5,-1.);
\draw [line width=0.4pt] (0.,-2.) -- (0.,-1.);
\draw (-1.25,0.5) node[anchor=north west] {$q$};
\draw (-1.4,-2.2) node[anchor=north west] {$\ldots$};
\draw (-1.4,2.) node[anchor=north west] {$\ldots$};
\draw [line width=0.4pt] (-2.,5.)-- (0.5,5.);
\draw [line width=0.4pt] (0.5,5.)-- (0.5,3.);
\draw [line width=0.4pt] (0.5,3.)-- (-2.,3.);
\draw [line width=0.4pt] (-2.,3.)-- (-2.,5.);
\draw [->,line width=0.4pt] (-1.5,5.) -- (-1.5,6.);
\draw [->,line width=0.4pt] (0.,5.) -- (0.,6.);
\draw (-1.25,4.5) node[anchor=north west] {$p$};
\draw (-1.4,6.) node[anchor=north west] {$\ldots$};
\draw [shift={(1.5,-2.)},line width=0.4pt]  plot[domain=3.141592653589793:6.283185307179586,variable=\t]({1.*1.5*cos(\t r)+0.*1.5*sin(\t r)},{0.*1.5*cos(\t r)+1.*1.5*sin(\t r)});
\draw [shift={(0.,-2.)},line width=0.4pt]  plot[domain=3.141592653589793:6.283185307179586,variable=\t]({1.*1.5*cos(\t r)+0.*1.5*sin(\t r)},{0.*1.5*cos(\t r)+1.*1.5*sin(\t r)});
\draw [shift={(1.5,6.)},line width=0.4pt]  plot[domain=0.:3.141592653589793,variable=\t]({1.*1.5*cos(\t r)+0.*1.5*sin(\t r)},{0.*1.5*cos(\t r)+1.*1.5*sin(\t r)});
\draw [shift={(0.,6.)},line width=0.4pt]  plot[domain=0.:3.141592653589793,variable=\t]({1.*1.5*cos(\t r)+0.*1.5*sin(\t r)},{0.*1.5*cos(\t r)+1.*1.5*sin(\t r)});
\draw [->,line width=0.4pt] (1.5,6.) -- (1.5,-2.);
\draw [->,line width=0.4pt] (3.,6.) -- (3.,-2.);
\end{tikzpicture}
\end{align*}
which are the same. Applying $\alpha$, we obtain $\mathrm{Tr}_P(p\circ q)=\mathrm{Tr}_P(q\circ p)$. 
Moreover, the graph $\mathrm{Tr}(p*q)$ is represented by 

\begin{center}
\begin{tikzpicture}[line cap=round,line join=round,>=triangle 45,x=0.5cm,y=0.5cm]
\clip(-2.5,-4.) rectangle (3.5,4.);
\draw [line width=0.4pt] (-2.,1.)-- (0.5,1.);
\draw [line width=0.4pt] (0.5,1.)-- (0.5,-1.);
\draw [line width=0.4pt] (0.5,-1.)-- (-2.,-1.);
\draw [line width=0.4pt] (-2.,-1.)-- (-2.,1.);
\draw [line width=0.4pt] (-1.5,1.) -- (-1.5,2.);
\draw [line width=0.4pt] (0.,1.) -- (0.,2.);
\draw [line width=0.4pt] (-1.5,-2.) -- (-1.5,-1.);
\draw [line width=0.4pt] (0.,-2.) -- (0.,-1.);
\draw (-1.25,0.5) node[anchor=north west] {$p$};
\draw (-1.4,-2.2) node[anchor=north west] {$\ldots$};
\draw (-1.4,2.) node[anchor=north west] {$\ldots$};
\draw [shift={(1.5,-2.)},line width=0.4pt]  plot[domain=3.141592653589793:6.283185307179586,variable=\t]({1.*1.5*cos(\t r)+0.*1.5*sin(\t r)},{0.*1.5*cos(\t r)+1.*1.5*sin(\t r)});
\draw [shift={(0.,-2.)},line width=0.4pt]  plot[domain=3.141592653589793:6.283185307179586,variable=\t]({1.*1.5*cos(\t r)+0.*1.5*sin(\t r)},{0.*1.5*cos(\t r)+1.*1.5*sin(\t r)});
\draw [shift={(1.5,2.)},line width=0.4pt]  plot[domain=0.:3.141592653589793,variable=\t]({1.*1.5*cos(\t r)+0.*1.5*sin(\t r)},{0.*1.5*cos(\t r)+1.*1.5*sin(\t r)});
\draw [shift={(0.,2.)},line width=0.4pt]  plot[domain=0.:3.141592653589793,variable=\t]({1.*1.5*cos(\t r)+0.*1.5*sin(\t r)},{0.*1.5*cos(\t r)+1.*1.5*sin(\t r)});
\draw [->,line width=0.4pt] (1.5,2.) -- (1.5,-2.);
\draw [->,line width=0.4pt] (3.,2.) -- (3.,-2.);
\end{tikzpicture}
\begin{tikzpicture}[line cap=round,line join=round,>=triangle 45,x=0.5cm,y=0.5cm]
\clip(-2.5,-4.) rectangle (6.,4.);
\draw [line width=0.4pt] (-2.,1.)-- (0.5,1.);
\draw [line width=0.4pt] (0.5,1.)-- (0.5,-1.);
\draw [line width=0.4pt] (0.5,-1.)-- (-2.,-1.);
\draw [line width=0.4pt] (-2.,-1.)-- (-2.,1.);
\draw [line width=0.4pt] (-1.5,1.) -- (-1.5,2.);
\draw [line width=0.4pt] (0.,1.) -- (0.,2.);
\draw [line width=0.4pt] (-1.5,-2.) -- (-1.5,-1.);
\draw [line width=0.4pt] (0.,-2.) -- (0.,-1.);
\draw (-1.25,0.5) node[anchor=north west] {$q$};
\draw (-1.4,-2.2) node[anchor=north west] {$\ldots$};
\draw (-1.4,2.) node[anchor=north west] {$\ldots$};
\draw [shift={(1.5,-2.)},line width=0.4pt]  plot[domain=3.141592653589793:6.283185307179586,variable=\t]({1.*1.5*cos(\t r)+0.*1.5*sin(\t r)},{0.*1.5*cos(\t r)+1.*1.5*sin(\t r)});
\draw [shift={(0.,-2.)},line width=0.4pt]  plot[domain=3.141592653589793:6.283185307179586,variable=\t]({1.*1.5*cos(\t r)+0.*1.5*sin(\t r)},{0.*1.5*cos(\t r)+1.*1.5*sin(\t r)});
\draw [shift={(1.5,2.)},line width=0.4pt]  plot[domain=0.:3.141592653589793,variable=\t]({1.*1.5*cos(\t r)+0.*1.5*sin(\t r)},{0.*1.5*cos(\t r)+1.*1.5*sin(\t r)});
\draw [shift={(0.,2.)},line width=0.4pt]  plot[domain=0.:3.141592653589793,variable=\t]({1.*1.5*cos(\t r)+0.*1.5*sin(\t r)},{0.*1.5*cos(\t r)+1.*1.5*sin(\t r)});
\draw [->,line width=0.4pt] (1.5,2.) -- (1.5,-2.);
\draw [->,line width=0.4pt] (3.,2.) -- (3.,-2.);
\end{tikzpicture}
\end{center}
which is also a graphical representation of $\mathrm{Tr}(p)*\mathrm{Tr}(q)$. Applying $\alpha$,
we obtain $\mathrm{Tr}_P(p*q)=\mathrm{Tr}_P(p)*\mathrm{Tr}_P(q)$. 
\end{proof}
Let us check that our generalised traces indeed generalise the usual traces of endomorphisms of finite dimensional vector spaces.
\begin{example}
Let $V$ be a finite dimensional vector space and $f=\theta(v_1\ldots v_k\otimes f_1\ldots f_k) \in \Hom_V^{fr}(k,k)$. 
Identifying $\Hom_V(0,0)$ with $\mathbb{R}$, we obtain that
\[\mathrm{Tr}_{\Hom_V}(f)=f_1(v_1)\ldots f_k(v_k),\]
which is the usual trace of linear endomorphisms of a finite-dimensional vector space.
If $V$ is not finite-dimensional, $\mathrm{Tr}_{\Hom_V^{fr}}$ is a direct generalisation of this trace for linear endomorphisms
of finite rank.  
\end{example}

\subsection{Amplitudes and generalised convolutions}

By Theorem \ref{thm:freetraps} applied to $\phi=\mathrm{Id}_P$,
 we know that for any TRAP $P$ there exists a canonical TRAP map $\Phi_P:\rPGr(P)\longrightarrow P$ (see Remark \ref{rk:PhiPGrX}).
\begin{defn} \label{def:generalised-P-convolution}
 Let $G$ be a graph decorated by a TRAP $P$. 
 The \textbf{$P$-amplitude} associated to $G$ is the  image  of $G$ under $\Phi_P$.
 
 When $P=\mathcal{K}_M^\infty$ is the TRAP of smooth generalised kernels over a smooth finite dimensional closed Riemannian manifold $M$ of Subsection \ref{subsec:smoothingop} (that is if $P(k,l)=\mathcal{K}_M^\infty(k,l)$ with the r.h.s defined in \eqref{eq:Klm}), 
 we simply write $\Phi$ for $\Phi_P$ and call $\Phi(G)$   the \textbf{smooth amplitudes} 
 associated to $G\in\rPGr(\mathcal{K}_M^\infty)$.
\end{defn}
\begin{rk}
The terminology $P$-amplitude is justified in so far as it associates to a graph an expression in $P$ depending on the ingoing and outgoing edges of the graph in a similar way  as an amplitude  associated to a Feynman diagram depends on the external ingoing and outgoing momenta. So these \ty{$P$-amplitudes} are a step toward to our final goal of building amplitudes of Feynman graphs.
\end{rk}
The case of a path graph relates amplitudes to convolutions.
\begin{rk} \label{rk:generalised_conv}
	Let $G$ be a path graph decorated by $X=(\mathcal{K}_X^\infty(k,l))_{k,l\geq0}$, that is to say a solar graph such that $I(G)=O(G)=[1]$, $V(G)=\{v_1,\cdots,v_n\}$, $E(G)=\{e_1,\cdots,e_{n-1}\}$ and the source and target maps defined by $s_G(1)=v_n$, $t_G(1)=v_1$ and 
	\begin{equation*}
	 \forall i\in [n-1],~s_G(e_i)=v_i,\qquad t_G(e_i)=v_{i+1}.
	\end{equation*}
	Here is a graphical representation of this graph:
	\[\xymatrix{1\ar[r]& \rond{v_1}\ar[r]&\ldots \ar[r]&\rond{v_n}\ar[r]&1}\]
	Let $P_i, i=1, \cdots, n$ be   smoothing pseudo-differential operators each of which is defined by the kernel $K_i$ that decorates  the vertex $v_i$. Then the generalised convolution associated to the graph $G$ is the convolution $K_1\star\cdots\star K_n$ of the kernels $K_1,\cdots,K_n$, which is the 
	kernel of the smoothing pseudo-differential operator $P_1\circ\cdots\circ P_n$. In this sense, $P$-amplitudes can be seen as  generalisations of the convolution of multiple smooth kernels.
\end{rk}
The generalised amplitude of a TRAP $P$ respects the horizontal concatenation of $P$ but also the vertical concatenation of $P$ built from the partial traces of $P$ in Subsection \ref{subsec:vert_conc_TRAP}.
 \begin{prop} \label{prop:gen_conv_vert_conc}
     For any TRAP $P$, the $P$-amplitude  associated to a horizontal concatenation of graphs is the horizontal concatenation of their $P$-amplitudes: for any $G_1,G_2\in\rPGr(P)$,
\begin{equation*}
\Phi_P(G_1*G_2)=\Phi_P(G_1)*\Phi_P(G_2),
\end{equation*}     
     and the same holds for the vertical concatenation: if $G_1\circ G_2$ exists, then 
     \begin{equation*}
      \Phi_P(G_1\circ G_2) = \Phi_P(G_1)\circ_P\Phi_P(G_2)
     \end{equation*}
     with $\circ_P$ the vertical concatenation of $P$.
    \end{prop}
    \begin{proof}
    % First, notice that since $G_2\circ G_2$ exists and since $\Phi_P$ is a TRAP morphism, $\Phi_P(G_1)\circ\Phi_P(G_2)$ exists. 
     This follows directly from the fact that $\Phi_P$ is a TRAP morphism and from Proposition \ref{prop:morphismPROPfromTRAP}.
    \end{proof}
For any TRAP $P$, let $\iota_P:P\hookrightarrow\rPGr(P)$ be the canonical embedding of $P$ into the  TRAP of $P$-decorated graphs that is,  $\iota_P(p)$ is the solar graph with only one vertex decorated by $p$. We  have the following simple \cy{Proposition, which basically states that every TRAP computations can be performed via graphs}.
\begin{lemma}
 For any TRAP $P$ the following diagram commutes:
  \begin{align*}
  & \xymatrix{P\times_\circ P~\ar@{^{(}->}[r]
  \ar[d]_{\circ_P} & \rPGr(P)\times\rPGr(P) \ar[d]^{\circ} \\
  P & \ar[l]^{\Phi_P} \rPGr(P)}
 \end{align*}
 with $\circ_p$ the vertical concatenation of the TRAP $P$, the top arrow given by $\iota_P\times\iota_P$ and $P\times_\circ P\subseteq P\times P$ is the domain of the vertical concatenation of the TRAP $P$, as defined in Definition \ref{defi:generalised_Traces}.
\end{lemma}
  In words,  the vertical concatenation of two elements $p_1$ and $p_2$ of $P$ is the $P$-amplitude associated with the graph given by the vertical concatenation of two graphs with exactly one vertex, each decorated by one $p_i$.
Graphically, if $p\in P(k,l)$ and $q\in P(l,m)$:
\[\Phi_P\left(\substack{\hspace{5mm}\\ 
\begin{tikzpicture}[line cap=round,line join=round,>=triangle 45,x=0.5cm,y=0.5cm]
		\clip(-2.5,-4.) rectangle (0.5,8.);
		\draw [line width=0.4pt] (-2.,1.)-- (0.5,1.);
		\draw [line width=0.4pt] (0.5,1.)-- (0.5,-1.);
		\draw [line width=0.4pt] (0.5,-1.)-- (-2.,-1.);
		\draw [line width=0.4pt] (-2.,-1.)-- (-2.,1.);
		\draw [->,line width=0.4pt] (-1.5,1.) -- (-1.5,3.);
		\draw [->,line width=0.4pt] (0.,1.) -- (0.,3.);
		\draw [->,line width=0.4pt] (-1.5,-3.) -- (-1.5,-1.);
		\draw [->,line width=0.4pt] (0.,-3.) -- (0.,-1.);
		\draw (-1.25,0.5) node[anchor=north west] {$p$};
		\draw (-1.8,-3) node[anchor=north west] {$1$};
		\draw (-0.3,-3) node[anchor=north west] {$k$};
		\draw (-1.4,-2.2) node[anchor=north west] {$\ldots$};
		\draw [line width=0.4pt] (-2.,5.)-- (0.5,5.);
		\draw [line width=0.4pt] (0.5,5.)-- (0.5,3.);
		\draw [line width=0.4pt] (0.5,3.)-- (-2.,3.);
		\draw [line width=0.4pt] (-2.,3.)-- (-2.,5.);
		\draw [->,line width=0.4pt] (-1.5,5.) -- (-1.5,7.);
		\draw [->,line width=0.4pt] (0.,5.) -- (0.,7.);
		\draw (-1.25,4.5) node[anchor=north west] {$q$};
		\draw (-1.8,8.2) node[anchor=north west] {$1$};
		\draw (-0.4,8.1) node[anchor=north west] {$m$};
		\draw (-1.4,6.) node[anchor=north west] {$\ldots$};
		\end{tikzpicture}}\hspace{3mm}\right)
=\Phi_P(p)\circ_P \Phi_P(q).\]
\begin{proof}
 Let $P$ be a TRAP. Then for any $p_1$, $p_2$ in $P$ such that $p_1\circ_P p_2$ is well defined, $\iota_P(p_1)\circ\iota_P(p_2)$ is well-defined since $\iota_P$ respects the gradings and we have
 \begin{align*}
  \Phi_P(\iota_P(p_1)\circ\iota_P(p_2)) & = \Phi_P(\iota_P(p_1))\circ_P\Phi(\iota_P(p_2)) \quad \text{by Proposition \ref{prop:gen_conv_vert_conc}} \\
  & = p_1\circ_P p_2
 \end{align*}
 since for any TRAP $P$, $\Phi_P\circ\iota_P = \mathrm{Id}_P$ by definition of $\Phi_P$ (Equation \eqref{eq:G_simple_solar} with $k=1$ and $\phi=\mathrm{Id}_P$).
\end{proof}

 \begin{rk} 
 Note that the vertical concatenation is \emph{not} the same as the $P$-amplitude: the latter has a much larger domain. 
\end{rk}

Applying the above constructions to the TRAP of smooth kernels  described in Theorem \ref{theo:Kinfty}, whose  partial traces  \eqref{eq:trconvK} are given by integrations on the underlying manifold,  easily yields the following statement.
We use the notations of Subsection \ref{subsec:smoothingop}: $M$ is a smooth, finite dimensional orientable closed manifold and $\mu(z)$ is a volume form on $M$.
\begin{theo} \label{thm:generalised_convolution}
For the TRAP $\left(\mathcal{K}_M^\infty(k,l)\right)_{k,l\geq0}$, the following statements hold:
  \begin{enumerate}
   \item The vertical concatenation  of two kernels corresponds to their  \textbf{generalised  convolution}: 
   \begin{align*}
&\forall (k,l,m)\in \N^3,\:\forall K_1\in {\mathcal K}_{M}^\infty(k,l),\:\forall K_2\in {\mathcal K}_{M}^\infty(l,m), \forall (x_1, \cdots, x_k, z_1, \cdots, z_m)\in M^{k+m} ,\\
&K_2\circ K_1(x_1,  \cdots, x_k, z_1, \cdots, z_m)\\
&=t_{k+1,1}\circ \ldots \circ t_{k+l-1,l-1}\circ t_{k+l,l}(K_1\otimes K_2)(x_1,  \cdots, x_k, z_1, \cdots, z_m)\\
& = \int_{M^l} K_1(x_1, \cdots, x_k, y_1, \cdots, y_l)\, K_2(y_1,  \cdots, y_l,  z_1, \cdots, z_m)\, {d\mu(y_1)\, \cdots d\mu(y_l)},
\end{align*} 
obtained  by integrating along the diagonal $\Delta_M^l:=\{(y_1,\cdots, y_l, y_1, \cdots, y_l), y_i\in M\}\subset M^{2l}$. 

   \item The associativity property $K_3\circ (K_2\circ K_1)= (K_3\circ K_2)\circ K_1$ (cfr. (\ref{eq:assovertconcTRAP})) for any \cy{$K_1\in {\mathcal K}_{M}^\infty(k,l)$, $K_2\in {\mathcal K}_{M}^\infty(l,m)$ and}  $K_3\in {\mathcal K}_{M}^\infty(m,n)$,
     amounts to  the Fubini property for the corresponding multiple integrals{:
     \begin{align} \label{eq:Fubini}
      &\int_{M^m} \left(\int_{M^l} K_1(\vec{x},\vec{y}_1)K_2(\vec{y}_1,\vec{y}_2)d\vec{\mu}(\vec{y}_1)\right)K_3(\vec{y}_2,\vec{z})d\vec{\mu}(\vec{y}_2) \\
      &= \int_{M^l} K_1(\vec{x},\vec{y}_1)\left(\int_{M^m} K_2(\vec{y}_1,\vec{y}_2)K_3(\vec{y}_2,\vec{z})d\vec{\mu}(\vec{y}_2)\right) d\vec{\mu}(\vec{y}_1) \nonumber
     \end{align}
     for any $\vec{x}\in M^k$ and $\vec{z}\in M^n$, where we use the short-hand notations $d\vec{\mu}(\vec{y}_i):=d\mu(y_1)\cdots d\mu(y_{l_i})$.}
     
     \item The generalised trace  of a generalised kernel \cy{$K$}%=K_1\otimes K_2\in \mathcal{K}_M^\infty(k,k)$ 
     is given by the integral along the small diagonal of $M^k$:
   \begin{equation} \label{eq:generalised_trace_smooth_kernel}
    \mathrm{Tr}_{{K}^\infty}(K)=\int_{M^k}K(x_1, \cdots, x_k, x_1, \cdots, x_k)\, d\mu(x_1)\cdots d\mu(x_k).
   \end{equation}
   \cy{It}
% where we have set $K(\vec x,\vec y):=K_1(\vec x) K_2(\vec y)$ and
obeys the following cyclicity property:
   \[\mathrm{Tr}_{{K}^\infty}(\tilde K\circ K)= \mathrm{Tr}_{{K}^\infty}(K\circ \tilde K)\]
   for $K\in\mathcal{K}_M^\infty(k,l)$ and $\tilde K\in\mathcal{K}_M^\infty(l,k)$.
   
   \item The $\mathcal{K}_M^\infty$-amplitude is compatible with the horizontal and vertical concatenations
   in $\mathcal{K}_M^\infty$.
    \end{enumerate}
\end{theo}
\begin{proof} We prove the assertions one by one.
 \begin{enumerate}
  \item The vertical concatenation $\circ$ of Definition-Proposition \ref{propverticalconcatenation} applied to the TRAP $\mathcal{K}_M^{\infty}$ of smooth kernels of Theorem \ref{theo:Kinfty} gives the generalised convolution..
  \item As proved in Definition-Proposition \ref{propverticalconcatenation}, the vertical concatenation $\circ$ of any TRAP is associative. Writing the explicit expression of each side of the equation \cy{$K_3\circ (K_2\circ K_1)= (K_3\circ K_2)\circ K_1$} for the vertical concatenation of the TRAP $\mathcal{K}_M^{\infty}$ shows that the identity amounts to the Fubini property for multiple integrals as given by Equation \eqref{eq:Fubini}.
  \item By Equation \eqref{eq:Klm}, for any $K$ in $\mathcal{K}^\infty_M(k,k)$, we can write $K=K_1\otimes K_2$ with $K_1$ and $K_2$ in $\mathcal{E}^{\widehat\otimes k}$. The generalised trace of Definition \ref{defi:generalised_Traces} for the TRAP  $\mathcal{K}_M^{\infty}$ of smooth kernels of Theorem \ref{theo:Kinfty} combined with the partial traces of this TRAP given by Equation \eqref{eq:trconvK} yields Equation \eqref{eq:generalised_trace_smooth_kernel}. The cyclicity property of $\mathrm{Tr}_{{K}^\infty}$ follows from the cyclicity property of generalised traces (Proposition \ref{prop:generalised_traces}, item 1).
  \item This follows from Proposition \ref{prop:gen_conv_vert_conc} applied to the generalised amplitude of Definition \ref{def:generalised-P-convolution} for the TRAP $\mathcal{K}_M^{\infty}$ of smooth kernels discussed in Theorem \ref{theo:Kinfty}.
 \end{enumerate}
\end{proof}

\section{Toward an application to QFT} \label{section:QFT}

We now aim at making precise our initial claim, namely that TRAPs could be of use to a rigorous approach of Quantum Field Theory. We will focus on the simple but non-trivial case of a scalar QFT on $\R^d$ with only one field. The case of multiple fields could be tackled by introducing ``colored TRAPs'' where the inputs, outputs and traces of these modified TRAPs would carry colors. The free object of this enhanced category can be described in terms of colored graphs. QFTs on lorentzian space are also expected to be \ty{tractable} with the strategy I will present below but it is expected than this will require more sophisticated analytical tools.

\subsection{Feynman amplitudes as distributions} \label{subsec:Feynman}

As explained in the introduction of this Chapter, the goal of (perturbative) QFT is to use the so-called ``Feynamn rules'' to attribute a value to some specific type of graphs. In other word, one is looking at evaluating a map 
\begin{equation*}
 \Phi_\tau:\calG_\tau\longrightarrow X
\end{equation*}
where $\calG_\tau$ is the set of Feynman graphs of the theory $\tau$ one is considering. Instead of looking at evaluating the Feynman rules $\Phi_\tau$ of the theory $\tau$, we will aim at rigorously defining this map.

Let us try to figure out the space $X$ the Feynman rules takes their values into. To simplify the discussion we will consider a scalar theory in $\R^d$. Let us quote \cite[formula I.4.8.]{Ri91} where the bare amplitude \emph{in position space} of a Feynman graph $G$ is written as a distribution in $N$ variables, with $N$ the number of external legs of the graph. This amplitude is (formally) written as:
 \begin{equation} \label{eq:amplitude_F_naive}
  A_G(z_1,\cdots,z_{|E(G)|}) = \int\prod_{i=1}^ndx_i\prod_{l\in E(G)}C(x_l,y_l)
 \end{equation}
 with $n$ the total number of vertices of the Feynman graph $G$, $E(G)$ its set of internal edges, $x_l$ the starting point of the edge $l$ and $y_l$ its ending point. This is an abuse of notation since the integral is over the positions of all the vertices of the graph $G$, thus over the $y_j$ as well as the $x_i$. Finally the $C(x,y)$ are the propagator, or Green functions, of the theory. For a scalar field theory on $\R^d$ it is given by
 \begin{equation*}
  C(x,y)=\frac{K}{||x^2-y^2||^{(d-2)/2}}
 \end{equation*}
 with $||.||$ the usual euclidian norm on $\R^d$ and $K$ a constant that depends on $d$.
 \begin{rk}
  A difficult task one has to tackle when computing Feynman graphs is to include the fact that multiple seemingly different physical processes give the same contribution to the total amplitude. These are taken into account by multiplying the amplitude $A_G$ of the graph $G$ with a combinatorial factor. Tracking down these factors is a painstaking task. Luckily for us we will not need to do so here since we are only interested in the \emph{existence} of the Feynman rules and not their application to specific graphs. The specific values of these combinatorial factors will therefore always be ignored in this text.
 \end{rk}

 To make the above Formula \eqref{eq:amplitude_F_naive} more precise, let us introduce some notations.
 
 Recall (Definition \ref{def:solar_graphs}) that a solar graph is a family of finite sets and maps of finite sets $G=(V(G),E(G),I(G),O(G),s,t,\alpha,\beta)$ with $V(G)$ the set of vertices, $E(G)$ the set of internal edges, $I(G)$ the set of incoming edge and $O(G)$ the set of outgoing edges.  Furthermore $s$ and $t$ are, respectively, the source and target maps of the edges, and $\alpha$ and $\beta$ are, respectively, the indexation of the inputs and outputs of the graph. Feynman graphs are \cy{such} graphs with one additional structure.
 \begin{defn}
  \cy{The {\bf Feynman graph} of a QFT $\tau$ are solar corolla oriented graphs whose topology respect the constraints of the theory $\tau$.}

%   A {\bf Feynman graph} $G$ with $N=k+l$ external legs is a solar graph with $k$ inputs and $l$ outputs (i.e. with $|I(G)|=k$ and $|O(G)|=l$) together with a position map 
%   \begin{equation*}
%    f:V(G)\longrightarrow \R^d
%   \end{equation*}
%   which to each vertex of the graph associate a position in the physical space in which the QFT lives\footnote{in more physical cases, $\R^d$ would be replaced by a $d$ dimensional lorentzian space.}. For $v\in V(G)$, we write $x_v:=f(v)$ the position of the vertex $v$.
  
  For a given theory $\tau$ we write $\calG_\tau^{k,l,E}\subset\calG_\tau$ the set of its Feynman graphs with $k$ inputs, $l$ ouputs and $E$ internal edges.
 \end{defn}
 \cy{The topology condition on Feynman graphs comes from the fact that}
%  Notice that since 
 in general a theory allows only some diagrams to be relevant. \cy{For example, a theory can contain only regular diagrams of a certain valency, i.e. diagrams such that $i(v)+o(v)$ is the same fixed $n$ for all vertices $v$. Then} we do not have that $\calG_\tau^{k,l,E}$ contains all solar graphs with the set of its Feynman graphs with $k$ inputs, $l$ ouputs and $E$ internal edges.
 
 With this framework, Equation \eqref{eq:amplitude_F_naive} should be understood as a distribution acting on $N=k+l$ test functions. We have (still formally)
 \begin{align} \label{eq:amplitude_distrib_bare}
  \langle A_G,\phi_1\cdots\phi_k\psi_1\cdots\psi_l\rangle & := \int\prod_{v\in V(G)}dx_v\prod_{e\in E(G)}C(x_{s(e)},x_{t(e)})\prod_{i=1}^{k}\phi_i(x_{t(\alpha(i))})\prod_{j=1}^l\psi_j(x_{s(\beta(j))}) \\
  & = \int\prod_{v\in V(G)}dx_v\prod_{e\in E(G)}C(x_{s(e)},x_{t(e)})\prod_{i\in I(G)}\phi_{\alpha^{-1}(i)}(x_{t(i)})\prod_{o\in O(G)}\psi_{\beta^{-1}(o)}(x_{s(o)}) \nonumber
 \end{align}
 where we use $\phi_1\cdots\phi_k\psi_1\cdots\psi_l$ as a short hand notation for $\phi_1\otimes\cdots\otimes\phi_k\otimes\psi_1\otimes\cdots\otimes\psi_l$. We have $\phi_1\cdots\phi_k\psi_1\cdots\psi_l\in(\calD(\R^d))^{\otimes(k+l)}$.
 % and, in an abuse of language, we say that $A_G$ is a distribution over $k+l$ copies of $\R^d$ and we write 
%  \begin{equation*}
%   A_G\in\calD'\left((\R^d)^{\otimes(k+l)}\right).
%  \end{equation*}
 \begin{rk}
  Equation \eqref{eq:amplitude_distrib_bare} is the \emph{truncated} amplitude of the graph $G$. The non-truncated amplitude (which is the one of \cite{Ri91}) is defined for a different class of Feynman graphs. They do not have external edges, but instead have external vertices. So in particular they have $k+l$ more integrations and $k+l$ more propagators. Then each of the $\phi_i$ and $\psi_j$ are acted upon by a different variable, which is not necessarily the case in Equation \eqref{eq:amplitude_distrib_bare}. We could do the same work for non-truncated Feynman amplitudes but these objects are less adapted to TRAPs structures.
 \end{rk}
 Now we have a better idea of what the amplitude of a graph should be. However a crucial observation is that the R.H.S. of Equation \eqref{eq:amplitude_distrib_bare} is in general not well-defined. Indeed the integrations over all the positions of the internal vertices induce divergences when $x_v=x_{w}$ when $v$ and $w$ are two vertices linked by an internal edge.
 
 This issue is solved by a \emph{regularisation procedure}. There are multiple regularisation schemes available in the \ty{literature} (zeta, analytic, dimensional, cut-off). Roughly seaking, they all consist by modifying each of the propagators $C(x_v,x_w)$ by a regularised propagator $C^z(x_v,x_w)$ (or $C^\varepsilon(x_v,x_w)$ according to the chosen regularisation scheme). Notice that in the regularised propagators $z$ and $\varepsilon$ are parameters and not powers.
 
 I will not specify which specific regularisation scheme I will use since I will only write down conjectures. However I will always have in the back of my mind the analytical regularisation scheme and use its notations. Furthermore I will insist on using a \emph{multivariate} regularisation scheme, where each of the propagator is regularised with a different complex variable $z_i$ (with $i\in\{1,\cdots,E=|E(G)|\}$) living in an open $\Omega\subseteq\C$ containing the origin. This choice is made with the framework of \emph{multivariate renormalisation} in mind. Recently a framework for tackling renormalisation with multiple regularisation variables was \ty{developed} and implemented in various papers e.g. \cite{CGPZ1}, \cite{CGPZ3}.
 
 Putting every elements of this discussion together we obtain that the (regularised) Feynman rules of a QFT have to associate to a Feynman graph $G$ with $k$ inputs and $l$ outputs its amplitude $A_G$ which has to be a map 
 \begin{equation*}
  A_G:\vec z=(z_1,\cdots,z_E)\longmapsto\left(A_G(\vec z):\calD(\R^d)^{\otimes(k+l)}\longrightarrow\C\right)
 \end{equation*}
 where, as before, $E=|E(G)|$ is the number of internal edges of the graph $G$.

\subsection{Distribution-valued meromorphic germs}

For a fixed $\vec z=(z_1,\cdots,z_E)$, one should expect $A_G(\vec z)$ to be continuous in some sense since the tests functions they act upon can be interpreted as representing the distributions of the ingoing and outgoing particles. 
%We must then choose a topology on the space $\calD(\R^d)^{\otimes(k+l)}$. Luckily there is a canonical choice, since $\calD(\R^d)$ is a nuclear space (see for example \cite[Corollary page 530]{Treves67}).
\begin{defn}
 We call a linear map $A:\calD(\R^d)^{\otimes(k+l)}\longrightarrow\C$ a {\bf distribution on $(\R^d)^{\otimes(k+l)}$} if for any $i\in[k]$ and $j\in[l]$ the maps 
 $\phi_i\mapsto\langle A,\phi_1\cdots\phi_k\psi_1\cdots\psi_l\rangle$ and $\psi_j\mapsto\langle A,\phi_1\cdots\phi_k\psi_1\cdots\psi_l\rangle$ are distributions\footnote{i.e. are continuous for the usual LF topology -- the Limit Fréchet topology on the space compactly supported functions.} on $\R^d$. We write $\calD'\left((\R^d)^{\otimes(k+l)}\right)$ the space of distributions on $(\R^d)^{\otimes(k+l)}$.
 
 A distribution $A$ on $(\R^d)^{\otimes(k+l)}$ is called {\bf regular} if there exists $M\in\N$ and $f_A:(\R^d)^M\longrightarrow\R$ smooth such that, for any test function $\phi_1\cdots\phi_k\psi_1\cdots\psi_l$ we have
 \begin{equation} \label{eq:action_A_fct_test}
  \langle A,\phi_1\cdots\phi_k\psi_1\cdots\psi_l\rangle = \int_{(\R^d)^M}\prod_{i=1}^Mdx_i f_A(\vec{x})\prod_{i=1}^k\phi_i(x_{\sigma(i)})\prod_{j=1}^l\psi_j(x_{\rho(j)})
 \end{equation}
 for some injective maps $\sigma:[k]\longrightarrow[M]$ and $\rho:[l]\longrightarrow[M]$. We write $\mathcal{D}_{\rm reg}'\left((\R^d)^{\otimes(k+l)}\right)$ is the space of regular distributions on $(\R^d)^{\otimes(k+l)}$.
\end{defn}
\begin{rk}
 \begin{itemize}
  \item There is a clear abuse of notation in this definition. Indeed, strictly speaking, elements of $\calD'\left((\R^d)^{\otimes(k+l)}\right)$ are not distributions on $(\R^d)^{\otimes(k+l)}$: they are not evaluated on the space of functions over $(\R^d)^{\otimes(k+l)}$ but on tensors of functions of $\R^d$. They are neither tensor products of distributions.
  \item In the definition above, one could take the space of Schwartz functions as the space of test functions. We choose against it since it seems to be the standard choice and QFT and also will simplify the coming analysis.
 \end{itemize}
\end{rk}
We now need to introduce functions of $\C^N$ with values in  $\distribreg$. A space of meromorphic functions with values in distributions which is adapted to the QFT context was introduced in \cite{DZ17}. It is based on work on meromorphic germs \cite{GPZ20} as well as the classical work \cite{Gro53}. See also \cite[paragraph 3.1]{BL19} for a recent and very readable presentation of the topic. The following definition is closely inspired by the definitions of \cite[Section 3.2]{DZ17}. This approach is close, at least in spirit, to Epstein-Glaser renormalisation \cite{epstein1973role} which seem to nowaday be the favored renormalisation scheme in the perturbative algebraic quantum field theory community (see for example  \cite{dutsch2014dimensional}).
\begin{defn}
 Let $V$ be a complex locally convex Hausdorff vector space and let $V'$ be the space of continuous linear forms on $V$.
 %and let $U\subset\C^n$ be an open domain of $\C^n$. 
 A map $f:U\subseteq\C^n\longrightarrow V$ is called {\bf analytic} if for any $\alpha\in V'$ the map 
 \begin{align*}
  \langle \alpha, ~ f(.)\rangle: U & \longrightarrow \C \\
  \vec z & \longmapsto \langle \alpha,f(\vec z)\rangle
 \end{align*}
 is analytic on $U$. 
 
 Let $L_1,\cdots,L_k$ be linear forms on $\C^n$ and let $f:\C^n\setminus\{L_1=\cdots=L_k=0\}$ be a map such that 
 \begin{align*}
  L_1\cdots L_k f:\C^n\setminus\{L_1=&\cdots=L_k=0\} \longrightarrow V \\
  \vec z & \longmapsto L_1(\vec z)\cdots L_k(\vec z)f(\vec z)
 \end{align*}
 is analytic on $\C^n\setminus\{L_1=\cdots=L_k=0\}$. Then $f$ is called a {\bf meromorphic} function (with linear poles\footnote{here linear poles will always mean linear poles at the origin and will therefore not be specified. The concept can be easily generalised to linear poles at arbitrary point of $\C^n$}).
\end{defn}
Notice that in this definition we use the notion of analytic functions of multiple variables. I chose not to introduce this topic with details and instead to simply refer the reader to the classical introduction \cite{Hor68}.

Now, in QFT, we will be interested only in the behaviour of these object around a point. We therefore do not want to distinguish two maps that differ around this point: we will need germs. I find analytical or dimensional regularisations to be to most convenient, so the germs will be taken around the origin. For zeta regularisation for example, one as to take germs around $1$. The \cy{following definition is also taken from \cite{DZ17}.}
\begin{defn} \label{def:mero_germs_distrib}
 \ty{An} {\bf analytic germ} (at the origin) of analytic $V$-valued maps is an equivalence class of these maps for the natural equivalence relation
 \begin{equation*}
  f\sim g:\Longleftrightarrow \exists U\subseteq\C^n{\rm open~neighbourhood~of~0 }:f|_U=g|_U.
 \end{equation*}
 This definition naturally extends to meromorphic functions.
 
 For two meromorphic $V$-valued maps with linear poles $f$ and $g$, we write $f\sim g$ if, and only if, it exists linear forms $L_1,\cdots,L_N$ on $\C^n$ and $U\subseteq\C^n$ an open neighbourhood of the origin such that $L_1\cdots L_N f$ and $L_1\cdots L_N g$ are analytic on $\C^n\setminus\{L_1=\cdots=L_N=0\}$ and coincide on $U\setminus(\{L_1=\cdots=L_N=0\}\cap U)$.
 
 Then {\bf meromorphic germs} with linear poles (at the origin) of meromorphic $V$-valued maps are equivalence classes of these maps for this equivalence relation.

 In the case $V={\mathcal{D}_{\rm reg}'\left((\R^d)^{\otimes N}\right)}$ we write $\calM_{\rm lin}\left(\C^{n},{\mathcal{D}_{\rm reg}'\left((\R^d)^{\otimes N}\right)}\right)$ the space of distribution-valued meromorphic germs \cy{with linear poles}.
\end{defn}
We have admitted without proof the reasonable statement that ${\mathcal{D}_{\rm reg}'\left((\R^d)^{\otimes N}\right)}$ is a locally convex Hausdorff space. I do not write this as a conjecture since the Definitions above still make sense if it is not, simply replacing $V$ everywhere with ${\mathcal{D}_{\rm reg}'\left((\R^d)^{\otimes N}\right)}$.

Now, the discussion of Section \ref{subsec:Feynman} can now be reformulated in the following Conjecture.
\begin{conj} \label{conjecture1}
 Let $\tau$ be a quantum field theory and $\Phi_\tau$ its associated regularised Feynman rules. The evaluation on a Feynman graph $G\in\calG_\tau$ of $\Phi_\tau$ is a distribution-valued meromorphic germ:
 \begin{equation*}
  \Phi_\tau(G) \in\calM_{\rm lin}\left(\C^{|E(G)|},\distribreg\right).
 \end{equation*} 
\end{conj}
\begin{rk}
 This conjecture might be too crude as one could also wish to control the \emph{order} of the distribution at play. This was the case in \cite{DZ17} (see definition 3.3 and the proof of Theorem 5.3 therein).
\end{rk}

\subsection{The Main Conjecture}

Now we want to build a TRAP such that its associated amplitude defined in Definition \ref{def:generalised-P-convolution} is exactly the regularised Feynman rules. We will call {\bf Feynman TRAP} this conjectural TRAP and write it $F=(F(k,l))_{k,l\in\N}$. From the discussion above we see that the space $F(k,l)$ should include the image under the regularised Feynman rules of \emph{all} Feynman graphs with $k$ inputs and $l$ outputs. This is done by taking an inductive limit.

For any non-empty sets $X$ and $Y$, recall that there are canonical embeddings
\begin{equation*}
 \iota_k:{\rm Maps}(X^k,Y)\longrightarrow{\rm Maps}(X^{k+1},Y)
\end{equation*}
where a map $f:X^k\longrightarrow Y$ become constant in its last variable:
\begin{equation*}
 \iota_k(f)(x_1,\cdots,x_k,x_{k+1}):=f(x_1,\cdots,x_k).
\end{equation*}
This allows to take the direct limit:
\begin{equation*}
 {\rm Maps}(X^\infty,Y):=\lim_{\substack{\longrightarrow \\ k}}{\rm Maps}(X^k,Y).
\end{equation*}
When $X=\C$ and $Y=\distribreg$, this construction still holds for germs and we set 
\begin{equation} \label{def:space_Feynman_TRAP}
 F(k,l):=\calM_{\rm lin}\left(\C^{\infty},\distribreg\right)=\lim_{\substack{\longrightarrow \\ N}}\calM_{\rm lin}\left(\C^{N},\distribreg\right).
\end{equation}
Now we want to endow $F=(F(k,l))_{k,l\in\N}$ with a TRAP structure.

The horizontal concatenation $*$ is an extension of the usual tensor product of maps. For
\begin{equation*}
  A_{1} \in \mathcal{M}_{\rm lin}\left(\C^{N},\mathcal{D}_{\rm reg}'\left((\R^d)^{k_1+l_1}\right)\right),\quad \text{and}\quad A_{2} \in \mathcal{M}_{\rm lin}\left(\C^M,\mathcal{D}_{\rm reg}'\left((\R^d)^{k_2+l_2}\right)\right)
 \end{equation*}
 then $A_{1}*A_{2}$, evaluated on $\overrightarrow{z}=(z_1,\cdots,z_{N+M})$, is a distribution defined by its action on test functions:
 \begin{align} \label{eq:def_hor_conc_TRAP_QFT}
  & \langle A_{1}*A_{2}(\overrightarrow{z}),\phi_1\cdots\phi_{k_1+k_2}\psi_1\cdots\psi_{l_1+l_2}\rangle := \\
  \langle A_{1}(\overrightarrow{z_1}),&\phi_1\cdots\phi_{k_1}\psi_1\cdots\psi_{l_1}\rangle \langle A_{2}(\overrightarrow{z_2}),\phi_{k_1+1}\cdots\phi_{k_1+k_2}\psi_{l_1+1}\cdots\psi_{l_1+l_2}\rangle \nonumber
 \end{align}
 with $\overrightarrow{z_1}=(z_1,\cdots,z_N)$ and $\overrightarrow{z_2}=(z_{N+1},\cdots,z_{N+M})$.
 \begin{rk} \label{rk:locality_Feynman_rules_TRAP}
  If $A_i=A_{G_i}$ for two graphs $G_1$ and $G_2$, then the universal property of the TRAP of graphs will give us $A_{G_1}*A_{G_2}=A_{G_1G_2}$. In this case, the fact that the action of $A_{G_1}*A_{G_2}(\overrightarrow{z})$ on test functions is a product of the amplitudes associated to each of the graphs $G_1$ and $G_2$ is a mathematical expression of the physical concept of locality.
 \end{rk}
 The product $*$ can then be extended to be defined on $F(k_1,l_1)\otimes F(k_2,l_2)$ and then to an horizontal product on $F$. \\
 
%  \vspace{1cm}
 Next, we need to define the candidates for partial traces. For $A\in\mathcal{M}_{\rm lin}\left(\C^{N},\mathcal{D}_{\rm reg}'\left((\R^d)^{k+l}\right)\right)$ we can write, with the notations of Equation \eqref{eq:action_A_fct_test}
  \begin{equation*}
  \langle A(\vec{z}),\phi_1\cdots\phi_k\psi_1\cdots\psi_l\rangle = \int\prod_{a=1}^Mdx_a f_A^{\vec{z}}(\vec{x})\prod_{b=1}^k\phi_b(x_{\sigma(b)})\prod_{c=1}^l\psi_c(x_{\rho(c)}).
 \end{equation*}
 Then for any $(i,j)\in[k]\times[l]$ we define $tr_{i,j}(A)\in\mathcal{M}_{\rm lin}\left(\C^{N+1},\mathcal{D}_{\rm reg}'\left((\R^d)^{k-1+l-1}\right)\right)$ (conjecturally) by its action of test functions once evaluated on $\vec{z'}=(\vec{z},z_{N+1})\in\C^{N+1}$, namely
 \begin{equation*}
  \langle tr_{i,j}(A)(\vec{z'}),\phi_1\cdots\phi_{k-1}\psi_1\cdots\psi_{l-1}\rangle = \int\prod_{a=1}^Mdx_a f_A^{\vec{z}}(\vec{x})C^{z_{N+1}}(x_j,x_i) \prod_{b=1}^k\phi_b(x_{\sigma_i(b)}) \prod_{c=1}^l\psi_c(x_{\rho_j(c)})
 \end{equation*}
 with $\sigma_i:[k-1]\longrightarrow[M]$ and $\rho_j:[l-1]\longrightarrow[M]$ injective maps defined by
 \begin{align*}
  \sigma_i(b) = \begin{cases}
                     & \sigma(b)\quad\text{if }b<i \\
                     & \sigma(b+1)\quad\text{if }b\geq i
                    \end{cases}, \quad 
  \rho_j(c) = \begin{cases}
                   & \rho(c) \quad\text{if }c<j \\
                   & \rho(c+1)\quad\text{if }c\geq j.
                  \end{cases}
 \end{align*}
 Furthermore, recall that $C^z$ is the regularised propagator of the scalar QFT we are considering, 
 We then extend $tr_{i,j}$ to the whole of $F$. The main conjecture we have to prove is then
 \begin{conj} \label{conj:TRAP_reg}
  The partial traces and horizontal concatenation defined above, together with the natural actions of the permutations groups on the inputs and outputs endow $F$ with a TRAP structure.
 \end{conj}
 This conjecture would have to be proven in three steps:
 \begin{enumerate}
  \item For $A_1\in F(k_1,l_1)$ and $A_2\in F(k_2,l_2)$ one has to show that $A_1*A_2\in F(k_1+k_2,l_1+l_2)$.
  \item For $A\in F(k,l)$ one has to show that $tr_{i,j}(A)\in F(k-1,l-1)$.
  \item Finally, one would have to check that these structures, together with the natural actions of the permutation group, respect the various axioms of TRAPs.
 \end{enumerate}
 I fully expect the first and third items to be fairly easily tractable. On the other hand, in the second item lies (probably) the main obstacles to proving Conjecture \ref{conj:TRAP_reg}. Indeed, the partial traces allow to take into account loops of Feynman graphs. These loops are notoriously the origin of divergences plaguing perturbative QFT. One would thus expect that $tr_{i,j}(A)$, as a germ, has a singular locus containing strictly the singular locus of $A$. \cy{Notice also that this second task contains two different challenges: first to show that the image of a trace is smooth outside of its singular locus, second that its poles are still linear.}
 There is nevertheless hope to tackle this challenge. Indeed, analytical methods to take into account and classify divergences of Feynman graphs were recently devised in \cite{DZ17}. I expect that the methods of this paper or a generalisation of them will allow us to prove that we indeed have TRAP structures behind reasonable perturbative quantum field theories. \cy{The recent paper \cite{DPS22} also brings hope regarding the question of smoothness.}

\subsection{Amplitude of Feynman graphs}

For a \ty{perturbation} parameter $\lambda\in\R^*_+$ and $(k,l)\in\N^2$ let us define $\lambda_{k,l}\in\distribreg$ by its action on test functions
\begin{equation*}
 \langle\lambda_{k,l},\phi_1\cdots\phi_k\psi_1\cdots\psi_l\rangle := \lambda\int_{(\R^d)^{k+l}}\prod_{i=1}^kdx_i\prod_{j=1}^ldy_j\prod_{i=1}^k\phi_i(x_i)\prod_{j=1}^l\psi_j(y_j).
\end{equation*}
Furthermore let us set 
\begin{equation*}
 \Lambda:=\{\lambda_{k,l}|(k,l)\in\N^2\}.
\end{equation*}
For any $(k,l)\in\N^2$, we can see $\lambda_{k,l}$ as a constant function of $\C^\infty$ with values in $\distribreg$, thus as an element of $F(k,l)$. We therefore have an embedding $\iota_\Lambda:\Lambda\hookrightarrow F$.

Moreover any graph $G\in\rPGr$ can be decorated by $\Lambda$ by setting $d(v)=\lambda_{k,l}$ for any vertex $v\in V(G)$ with $k$ inputs and $l$ outputs. From Theorem \ref{thm:freetraps}, and under the assumption that Conjecture \ref{conj:TRAP_reg} we then have a unique TRAP morphism $\Phi:\rPGr(\Lambda)\longrightarrow F$. Finally, by ordering the ingoing and outgoing edges of each vertex of the Feynman graphs of a QFT $\tau$ we have an embedding $\iota_\tau:\calG_\tau\hookrightarrow\rPGr(\Lambda)$ of these Feynman graphs into the TRAP of graphs since the vertices of Feynman graphs are all decorated with $\lambda_{k,l}$. This allows to \emph{define} the Feynman rules.
\begin{defn} \label{def:TRAP_Feynman_rules}
Provided that Conjecture \ref{conj:TRAP_reg} holds, for a scalar quantum fiel theory $\tau$, we define its {\bf regularised Feynman rules} to be 
\begin{equation*}
 \Phi_\tau:=\Phi\circ\iota_\tau.
\end{equation*}
Diagrammatically:
 \begin{figure}[h!] 
  		\begin{center}
  			\begin{tikzpicture}[->,>=stealth',shorten >=1pt,auto,node distance=3cm,thick]
  			\tikzstyle{arrow}=[->]
  			
  			\node (1) {$\calG_\tau$};
  			\node (2) [right of=1] {$\rPGr(\Lambda)$};
  			\node (3) [right of=2] {$F$};
  			\node (4) [below of=2] {$\Lambda$};

  			\path
  			(1) edge node [above] {$\iota_\tau$} (2)
  			(2) edge node [above] {$\Phi$} (3);
  			
  			\draw [right hook-latex] (4) -- node[left] {} (2);
  			\draw [right hook-latex] (4) -- node[right] {$\iota_{\Lambda}$} (3);
%   			(4) edge node [left] {} (2)
%   			(4) edge node [above] {$\iota_\Lambda$} (3);
  			\end{tikzpicture}
  			\caption{The definition of regularised Feynman rule.}\label{fig:def_feynman_rule}
  		\end{center}
  	\end{figure} 
\end{defn}
In other words, the amplitude of a Feynman graph $G\in\calG_\tau$ is the $F$-amplitude of $\iota_\tau(G)$ (see Definition \ref{def:generalised-P-convolution}). This justifies at last the name ``amplitude'' given to this object. Notice also that \cy{in practice this amplitude would only be needed for the subset of graphs that are relevant for the QFT we are studying. Typically these would be graphs whose vertices have a fixed degree} 
% actually be used for a small sub-TRAP of $F$, typically for graphs with only a fixed total number of ingoing and outgoing edges
(see Remark \ref{rk:subTRAPutiles}).
\begin{rk}
 In order for the above definition of regularised Feynman rules to be consistent, it should not depend on the choices made to order the ingoing and outgoing edges of the vertices of the Feynman graphs. In other words, it should not depend on the choice of the embedding $\iota_\tau$. I expect this to be easy to show since all ingoing and outgoing edges play the same role in the definition of $\lambda_{k,l}$.
 
 Notice that in the case of a QFT with more than one field, the situation might be more intricate. However I do not expect this new aspect to be more than a mere technicality.
\end{rk}
We finish this presentation of a research project with a simple but nice result. First, by construction, if a Feynman graph $G\in\calG_\tau$ has $k$ inputs and $l$ outputs, then $\Phi_\tau(G)\in F(k,l)$ since $\Phi$ is a morphism of TRAPs. Furthermore, since the $\lambda_{k,l}$s  are constants in the sense that they do not depend on the parameter $z_i\in\C$, we have that the dependences in the $z_i$ come solely from the partial traces maps. We then have that $\Phi_\tau(G)$ depends precisely on the number of trace maps one has to use to build $\iota_\tau(G)$ from the graphs $G_{k,l}$ with only one vertex, $k$ incoming edges, $l$ outgoing edges and no internal edges. This number is precisely the number of internal edges of the graph $G$. In other words, together with Remark \ref{rk:locality_Feynman_rules_TRAP} we have
\begin{prop} \label{prop:rapport_conj}
 If Conjecture \ref{conj:TRAP_reg} holds, then Conjecture \ref{conjecture1} holds, with the Feynman rule defined by Definition \ref{def:TRAP_Feynman_rules}. Furthermore these Feynman rules define an algebra morphism with the locality property discussed in Remark \ref{rk:locality_Feynman_rules_TRAP}.
\end{prop}
Therefore Conjecture \ref{conj:TRAP_reg} solves precisely the problem we had set for ourselves!

\vspace{1cm}

Notice that we have not discussed the \emph{renormalisation} issue here, but only regularisation. However, this discussion was performed in the setup that allows for \emph{multivariate renormalisation} to be used. This will be the topic of Chapter \ref{chap:loc}.

% % %  COMMENTER CE QUI EST CI-DESSOUS.
% 
% \bibliographystyle{alpha}
% %  \addcontentsline{toc}{section}{Bibliography}
% \bibliography{HDR_biblio}
%  
% \end{document}

%% file: MZV.tex
% % % 19/01: 125 pages
% % % 
% \documentclass[11pt,twoside,a4paper]{book}
% % \documentclass[11pt,twoside,a4paper]{article}
% 
% \usepackage{HDR}
% 
% %% SI ON LAISSE LES DEUX LIGNES SUIVANTES DANS LE STY, ELLES NE MARCHENT PAS... (probablement à cause de \input).
% \input{xy}
% \xyoption{all}
% 
% 
% \begin{document}
% % 
% % % % 
% % % % % COMMENTER CE QUI EST AVANT CA ET ENDDOCUMENT POUR COMPILER HDR.tex.

\chapter{Generalisations of {Multiple Zeta} Values} \label{chap:MZV}

\section*{Introduction}

\addcontentsline{toc}{section}{Introduction}

\subsection*{State of the art} \addcontentsline{toc}{subsection}{State of the art}

Riemann's zeta sums are the dread and the hope of many a undergrad student\footnote{for mine, it is simply dread, only mitigated by blissful ignorance.}. One could call \ty{these sums {\bf zeta values}}, for the special cases where their exponent is an integer, \ty{greater than} or equal to two:
\begin{equation*}
 \zeta(n):=\sum_{p=1}^\infty\frac{1}{p^n}.
\end{equation*}
Of course, the notation $\zeta$ comes from the observation that these numbers are nothing but evaluations of the famous \ty{Riemann} zeta function at integers $n\geq2$. Zeta values and their generalisations \ty{appeared} in various aspects of Mathematics and Physics. \ty{They} are important in particular in the evaluation of Feynman amplitudes (see for example \cite{petermann1957fourth,laporta1996analytical,Brown2012ProofOT,todorov2014polylogarithms}). I mention this since it was one of my original \ty{motivations} to look at these objects. Notice also that the appearance of these numbers in Feynman amplitudes was the original motivation for the introduction of the romantically named Cosmic Galois Group \cite{cartier1998folle,brown1512feynman}.

\ty{Multiple Zeta} values (MZVs for short) can be seen as a multivariable generalisation of zeta values. They have a long history, starting like most mathematics with Euler \cite{Eu1796}, and \ty{emerging} here and there in many places of mathematics and physics in the following two centuries.  In the 1980s, MZVs have arisen in \'Ecalle's work \cite{Ecalle81b} (see also \cite{schneps2015ari}). A systematic study of MZVs was later initiated by Hoffman \cite{Ho92} and Zagier \cite{Za94} and has since then \ty{been} the subject of \ty{an important literature}. They have also been generalised and applied in more ways that this introduction could possibly list. We refer the reader to one of the many excellent introductions to the topics for a more in-depth historical point of view. \ty{I} learnt with \cite{Waldschmidt} and \cite{Bo15}. Other introductions focusing on various aspects of MZVs are \cite{kaneko2019introduction} and \cite{deligne2012multizetas}. \ty{The latter} introduces Brown's motivic zetas \cite{brown2012mixed,brown2012decomposition}. \cite{cartier2002fonctions} is another introduction, which is extremely pedagogical. 

Let us point out that \ty{during} their fairly long and quite non-linear history, MZVs have been called by many names: ``multizeta numbers''  by \'Ecalle, ``multiple harmonic sums'' by Hoffman, ``Euler-Zagier numbers'' by the Borwein brothers, ``multiple zeta values'' by Zagier, ``polyzeta 
numbers'' (in order to respect Weil's principles of not mixing \ty{Greek and Latin} roots) by Cartier... We follow what is now the most widespread name of ``\ty{multiple zeta} values'' and is probably a contraction of the name chosen by Zagier.

MZVs will be rigorously defined in Subsection \ref{subsec:MZV} below. Without introducing them rigorously yet, let us say that there exists now a variety of known results concerning their algebraic and number-theoretic properties, and ambitious conjectures. One of the most important conjectures is that all the rational relations among MZVs are given by exactly three families of \ty{relations}:
\begin{enumerate}
 \item The stuffle relations (Equation \eqref{eq:stuffle_prod}),
 \item The shuffle relations (Equation \eqref{eq:shuffle_prod}),
 \item Hoffman's relations (Equation \eqref{eq:Hoffman_reg_rel}).
\end{enumerate}
To these relations, one should add Kontsevich's relation (Equation \eqref{eq:Kontsevich}) which relates the two classical representation of MZVs via a binarisation map (Equation \eqref{eq:binarisation_map}). This relation is more of a defining relation, making sure that we are actually talking about the same objects whether we are writing MZVs as interated integrals \ty{or} series.

In this text we will not directly tackle this difficult and important question, however it has motivated the current study. Indeed, a way one might hope \ty{of tackling} such a conjecture is to generalise the object under consideration (here MZVs) and to study the conjecture in this more general space. This was the original motivation for this work: generalise MZVs to rooted forests in the hope \ty{of being} able to solve the generalisation of the above conjecture for these generalised MZVs. As it often happens, the generalisation turns out to be more difficult and interesting than expected and \ty{has for long been} the author's main \ty{object of study}.

\emph{A propos} generalisation of MZVs, this is a place as good as any to point out that MZVs admit a plethora of \ty{generalisations}. Let us mention without any particular order Hurwitz MZVs \cite{Bo15}, elliptic MZVs \cite{SchnepsElliptic,En13}, cell zeta values \cite{brown2010algebra}, $q$-zeta values \cite{chapoton2020moments} and Witten's MZVs \cite{Wi91} 
among others. In this paper, we will chiefly be interested in Arborified Zeta Values\footnote{also called ''branched zeta values`` in \cite{CGPZ1}} (AZVs) and Conical Zeta Values (CZVs).
AZVs appeared in the work of Ecalle \cite{Ecalle81b} and much later in the work of Yamamoto \cite{Ya14}. Their extensive study started in \cite{Ma13} and was completed in \cite{Cl20}. The renormalisation of their divergent counterparts was performed in \cite{CGPZ3}\footnote{this study was seen as a toy model for multivariate renormalisation and started the work presented here, when it was realised that some of the tools devised to deal with divergent AZVs could also be applied to study convergent AZVs}. On the other hand CZVs have been defined in \cite{GPZ13} and their divergent counterparts have been renormalised in \cite{GPZ17}. An important open question which we will investigate here is to characterise which CZVs are linear combinations of MZVs with rational coefficients. Notice that this question may be of interest in physics since some CZVs have been shown to appear in the perturbative expansion of amplitudes of some string theories, see for example \cite{Ze16,Ze17,zagier2019genus}. 

\subsection*{Content and main results} \addcontentsline{toc}{subsection}{Main results}

The content of this chapter is mostly from \cite{Cl20} and \cite{ClPe23} although some methods and results this work is based on appeared in the earlier work \cite{CGPZ2}. It starts with an introduction to classical structures of combinatorics of words. In particular, the famous shuffle and stuffle products are introduced in Definitions \ref{def:shuffle} and \ref{def:stuffle} respectively. These structures are necessary to introduce MZVs (Definition \ref{def:MZVs}) and state their well-known properties.

In section \ref{sec:comb_forest} we introduce rooted forests (Definition \ref{def:forests}) and state an important theorem due to Panzer and Kreimer, namely that the algebra of rooted forests is the initial object in the category of operated algebras (Definition \ref{thm:univ_prop_tree}). We then use this property to define some of the basic constructions that we will investigate, in particular branching of linear maps (Definition \ref{defn:phi_hat}) and flattening maps (Definition \ref{defn:flattening_maps}). We also define the shuffle and stuffle products of rooted forests (Definition \ref{defn:shuffle_tree}); a natural non-associative generalisation of the shuffle and stuffle products of words. These various combinatorial objects are related in the first (non-classical) important theorem of this Chapter, namely Theorem \ref{thm:flattening}. The missing link between these construction is Rota-Baxter operators.

\medskip

Section \ref{section:stuffle} presents a first application of this result. It is essentially applied to \ty{a} space of rooted forests decorated by log-polyhomogeneous symbols (see Definition \ref{def:log_poly_hom}). For \ty{these} objects, the Euler-MacLaurin formula (Equation \eqref{eq:EML_sum}) \ty{gives an interpolation between discrete sums and integrals. This makes the use of standard tools of integration possible}. We can then build our desired generalisation of MZVs as iterated series to rooted forests: the Arborified Zeta Values (AZVs, see Definition-Proposition \ref{defnprop:arborified_zeta}). Notice that although this construction does use some analytical tools, it is algebraic at its core. It is this algebraic side that allows us to straightforwardly obtain important properties of AZVs (Theorem \ref{thm:main_result_stuffle}). Notice that along the way, we also obtain a new algebraic proof of a classical property of MZVs, usually proved in a much more analytical and tedious way.

 The next Section \ref{section:shuffle} is essentially another application of Theorem \ref{thm:flattening}. We apply it to an integration map, thus obtaining a branched generalisation of Chen's integrals. This gives us arborified polylogarithms (Definition-Proposition \ref{defnprop:arbo_polylogs}) for which we readily derive a few properties (Theorem \ref{thm:arborified_polylogs}) from  Theorem \ref{thm:flattening}. We then define a second version of AZVs, as interated integrals, as evaluation of branched polylogarithms (Definition \ref{defn:shuffle_AZVs}). Results on polylogs then directly imply results for this second AZV and are stated in Theorem \ref{thm:main_result_shuffle}. As in the previous section, we obtain along the way new proofs of algebraic properties of MZVs and multiple polylogarithms. Notice that the results of Sections \ref{section:stuffle} and \ref{section:shuffle} could be obtained in completely similar ways, but were not. This was done first for technical reasons (the space of log-polyhomogeneous symbols has the structure of a filtered space that is less clear in the space of functions relevant in section \ref{section:shuffle} -- at least for the author); and second to illustrate that the methods presented here do have some variability.
 
 At this point, two sets of relations among MZVs have been lifted to AZVs. We then aim in Section \ref{section:Hoffman} at a generalisation to rooted forests to Hoffman's relation \eqref{eq:Hoffman_reg_rel}. Using once again the universal properties of rooted forests in the category of operated algebras, we define the branched binarisation map in Definition \ref{defn:branched_bin_map}. A preliminary step for Hoffman's relation is to investigate the generalisation of Kontsevich's relation \eqref{eq:Kontsevich}. However, \ty{the latter} does not hold for AZVs, as shown in Theorem \ref{thm:relation_shuffle_stuffle}. It is then no surprise that Hoffman's relation does not hold either for AZVs, as \ty{explictly} checked in Equation \eqref{eq:counter_example_Hoffman}.
 
\medskip

Looking back, one can understand why Kontsevich's relations fail for AZVs. Essentially, it is because AZVs as iterated series are not ''the right`` generalisation to rooted forests of MZVs as iterated series. The right relation should rather generalise Equation \eqref{eq:MZV_stuffle2}. And indeed, we obtain this representation in Theorem \ref{thm:integral_sum}, which could be seen as the second crucial result of this chapter. 
Section  \ref{section:int_to_series} is devoted to proving this result. 

Theorem \ref{thm:integral_sum} suggests a new generalisation of MZVs to rooted forests. Section \ref{section:TZVs} \ty{introduces} and \ty{studies} this new generalisation, which we call Tree Zeta Values (or TZVs for short). They are introduced in Definition \ref{defn:tree_zeta_values} and their main properties are stated right away in Proposition \ref{prop:crucial} and Theorem \ref{thm:tree_zeta_MZVs}. In order to study TZVs we introduce a new product, the $\yew$-product\footnote{pronounced ''Upsilon-product``} Definition \ref{defn:yew}). A combinatorial description of this product is given in Theorem \ref{thm:yew_formula}. While it is clear by definition that TZVs form an algebra morphism for $\yew$-product, we also show a possibly more surprising result, namely that the $\yew$-product allows us to directly relate TZVs and MZVs, without going through AZVs (Corollary \ref{coro:TZVs_stuffle_MZVs}). TZVs and MZVs are related by a new flattening map built from the $\yew$-product (see Definition \ref{defn:flattening_yew}).

The second to last section of this chapter is Section \ref{section:application}, where we apply our results on TZVs to find properties of other generalisations of MZVs. We start with the Mordell-Tornheim zeta values (Definition \ref{defn:MT}). For this object we obtain in particular new proofs of some classical properties (Theorem \ref{thm:MT_tree}) as well as some decomposition formulas (Equations \eqref{eq:MT_decom} and \eqref{eq:expression_MT}). Our next application \ty{concerns} Conical Zeta Values (CZVs) (see Equation \eqref{eq:def_CZV} for the definition). It is clear that TZVs are CZVs. So, for a family of cones (see Definition \ref{defn:tree_like_cone}), we can prove that the associated CZVs are linear combinations of MZVs. This is done in Theorem \ref{thm:tree_CZVs_MZVs}. We then turn our attention to \ty{characterise} for which cones this result holds. This is fully achieved by Theorem \ref{thm:carac_tree_cone}. This section ends with examples of CZVs evaluated in terms of MZVs with the methods \ty{developed} in the chapter.

\medskip

The chapter ends with a list of open questions that will be pursued in further research. Some are simple questions of enumerative combinatorics. Such questions could be of importance for future work, but are not necessarily very interesting by themselves. More interesting are the questions of algebraic combinatorics, in particular the properties of the shuffle of rooted forests, and eventual related coalgebraic structures. These questions are currently under scrutiny by our PhD student Douglas Modesto. Finally, a very important question is essentially number-theoretic: how could one generalise the stuffle product of words to rooted forests such that TZVs form an algebra morphism \emph{and} such that we have a generalised Hoffman's relation. A list of approaches that have been tried to solve this later question is given, one of which \ty{seemingly} very promising.

% \subsection*{Properties of MZVs} \addcontentsline{toc}{subsection}{Properties of MZVs}

% Stuffle and shuffle MZVs are linked through what we call the {\bf binarisation map}
% \begin{align} \label{eq:binarisation_map}
%  \fraks :\calW_{\N^*}&\longrightarrow\calW_{\{x,y\}} \\
%  (n_1\cdots n_k) & \longrightarrow (\underbrace{x\cdots x}_{n_1-1}y\cdots\underbrace{x\cdots x}_{n_k-1}y). \nonumber
% \end{align}
% Kontsevich's relation
% \begin{equation} \label{eq:shuffle_stuffle_words}
%  \zeta_\shuffle(\fraks(w)) = \zeta_\stuffle(w).
% \end{equation}
% Hoffman:
% \begin{equation} \label{eq:Hoffman_reg_rel}
%  \fraks\left((1)\stuffle w)\right) - (y)\shuffle\fraks(w) \in \Ker(\zeta_\shuffle).
% \end{equation}
% stuffle MZVs:
% \begin{align} \label{eq:sum_stuffle_MZV}
%  \zeta_\stuffle :~ & \calW_{\N^*}^{\rm conv}\subseteq\calW_{\N^*} \longrightarrow \R \nonumber\\
% 		  & (p_1\cdots p_k) \mapsto \sum_{n_1>\cdots>n_k>0}\frac{1}{n_1^{p_1}\cdots n_k^{p_k}} 
% \end{align}

\section{{Multiple Zeta} Values}

We start by recalling classical notions of combinatorics of words and their application to \ty{Multiple Zeta} Values (written MZVs for short). We do not give proofs of these classical results and instead refer the reader to one of the many classical textbooks of the subject, e.g. \cite{lothaire_1997} \cite{Waldschmidt} for an introduction on MZVs. 

\subsection{Products of words}

Let us dive right in and define words.
\begin{defn}
 Let $\Omega$ be a non-empty set. We call {\bf word} (written in the alphabet $\Omega$) a string of elements of $\Omega$. We write $W_\Omega$ the set of words written in the alphabet $\Omega$ and $\calW_\Omega$ its linear span over $\R$. In other words, $\calW_\Omega$ is  therefore the algebra over $\R$ of non-commutative polynomials with variables in $\Omega$. 
 
 We also write $\emptyset$ for the empty word. 
\end{defn}
We also need to introduce various notations that will come in handy later in this chapter.
\begin{defn} \label{def:counting_words}
 \begin{itemize}
  \item Let $w=(\omega_1\cdots\omega_k)$ be a word written in the alphabet $\Omega$. We define its {\bf length} $|w|$ to be 
  $|\omega_1\cdots\omega_k|:=k$. We further set $|\emptyset|:=0$. We write $\calW_\Omega^n$ (resp. $\calW_\Omega^{\leq n}$) the \ty{vector subspace} of $\calW_\Omega$ generated by words of length $n$ and $\emptyset$ (resp. generated by words of length \ty{less than} or equal to $n$).
  \item For any $\omega\in\Omega$ and $w\in W_\Omega$, we write $\sharp_\omega w$ the number of times that the letter $\omega$ appears in a word $w$.
  \item Let $(\Omega,\bullet)$ be a commutative semigroup\footnote{i.e. $\bullet:\Omega\times\Omega\longrightarrow\Omega$ is an associative and commutative binary product on $\Omega$, a priori without a unit.}. We define the {\bf weight with respect to the product $\bullet$} $||w||_\bullet$ of a word $w\in\calW_\Omega$ to be $0$ \ty{if} $w=\emptyset$ and 
  $||\omega_1\cdots\omega_k||_\bullet:=\omega_1\bullet\cdots\bullet\omega_k$. We then extend the weight to all of $\calW_\Omega$ by linearity. If the product on $\Omega$ is clear from context, we will speak of the weight of $w$ and write $||w||$.
 \end{itemize}
\end{defn}
The vector space $\calW_\Omega$ can be endowed with various algebra structures.
\begin{defn} \label{defn:conc_prod}
 The {\bf concatenation product} $\sqcup:\calW_\Omega\times\calW_\Omega\longrightarrow\calW_\Omega$ is defined by
 \begin{equation*}
  \emptyset\sqcup w = w\sqcup\emptyset  := w
 \end{equation*}
 for any $w\in\calW_\Omega$ and, for any $(\omega_1\cdots\omega_k)\in W_\Omega$ and $(\omega_1'\cdots\omega_n')\in W_\Omega$
  \begin{equation*}
   (\omega_1\cdots\omega_k)\sqcup(\omega_1'\cdots\omega_n') := (\omega_1\cdots\omega_k\omega_1'\cdots\omega_n').
  \end{equation*}
  We then extend it by bilinearity to a product on $\calW_\Omega$.
\end{defn}
It is easy to show that the concatenation product is associative and we obtain
\begin{prop}
 For any non-empty set $\Omega$, the triple $(\calW_\Omega,\sqcup,\emptyset)$ is an associative but non commutative unital algebra graded by the length of the words.
\end{prop}
We will be chiefly interested by algebra structures on the vector space $\calW_\Omega$, namely the shuffle and stuffle (or quasi-shuffle) products.
\begin{defn} \cite{EML53} \label{def:shuffle}
Let $\Omega$ be a non-empty set. The {\bf shuffle product} $\shuffle$ on $\calW_\Omega$ is recursively defined by
  \begin{align*}
   \emptyset\shuffle w = w\shuffle\emptyset & = w, \\
   \left((\omega)\sqcup w\right) \shuffle \left((\omega')\sqcup w'\right) & = (\omega)\sqcup\left[w \shuffle \left((\omega')\sqcup w'\right) \right] + (\omega')\sqcup\left[\left((\omega)\sqcup w\right)\shuffle w'\right]
  \end{align*}
  for any $(w,w')\in\calW_\Omega^2$ and any $(\omega,\omega')\in\Omega^2$. We extend it by bilinearity to a product on $\calW_\Omega$.
\end{defn}
Before discussing this product and introducing some important generalisations, let us give some examples (without \ty{detailed} computations) of shuffle products.
\begin{example} \label{ex:shuffles}
 \begin{itemize}
  \item Let $\Omega=\{x,y\}$ be a set with two elements. Then
  \begin{equation*}
   (xx)\shuffle(xy)=3(xxxy)+2(xxyx)+(xyxx).
  \end{equation*}
  \item Let $\Omega=\N^*$ be the set of strictly positive integers. Then 
  \begin{equation*}
   (213)\shuffle (51)=(21351)+(21531)+(21513)+(25131)+2(25113)+(52131)+2(52113)+(51213).
  \end{equation*}
 \end{itemize}
\end{example}
Now, the shuffle product is often explained as a product that gives all the possible shuffles of two decks of cards, hence the name. Indeed, each word, since it is totally ordered, can be seen as a deck. The shuffle of two decks then gives all the possible decks built from the two initial decks with the constraint that the orders of the two initial decks are preserved. 

The stuffle is also a product obtained by a shuffling of decks, but where two cards can \ty{be} sticked together, with their values added. For this reason, the stuffle product is also sometimes called ``sticky shuffle''.

Recall that a semi-group is a pair $(\Omega,\bullet)$ with $ \bullet:\Omega\times\Omega\longrightarrow\Omega$ an associative product on $\Omega$.
\begin{defn} \label{def:stuffle} \cite{Ho00}
 Let $(\Omega,\bullet)$ be a commutative semi-group and $\lambda\in\R$. The 
  {\bf $\lambda$-shuffle product} is recursively defined by
  \begin{align*}
   \emptyset\shuffle_\lambda w = w\shuffle_\lambda\emptyset & = w, \\
   \left((\omega)\sqcup w\right) \shuffle_\lambda \left((\omega')\sqcup w'\right) & = (\omega)\sqcup\left[w \shuffle_\lambda \left((\omega')\sqcup w'\right) \right] + (\omega')\sqcup\left[\left((\omega)\sqcup w\right)\shuffle_\lambda w'\right]  \\
										  & + \lambda(\omega\bullet\omega')\sqcup\left[w\shuffle_\lambda w'\right] 
  \end{align*}
  for any $(w,w')\in\calW_\Omega^2$ and any $(\omega,\omega')\in\Omega^2$. We extend it by bilinearity to a product on $\calW_\Omega$.
  
  For $\lambda=1$ we write $\stuffle:=\shuffle_1$ the {\bf stuffle} (or quasi-shuffle, or sticky shuffle) product. For $\lambda=0$ we simply write $\shuffle=\shuffle_0$\footnote{Notice that is this case we find the previous shuffle on the set $\Omega$, hence the notations. This also shows that the $\lambda$-products are generalisations of the shuffle product, as claimed earlier.}. For $\lambda=-1$, the product $\shuffle_{-1}$ is called the {\bf anti-stuffle} product.
\end{defn}
The shuffle and $\lambda$-shuffle have been studied and generalised in many different ways, adapted to various product. I will not \ty{attempt} to quote all the recent \ty{developements}, but instead simply mention some that I have encountered in the last recent years. Generalisations of $\lambda$-shuffle that are adapted to $q$-MZVs were introduced in \cite{IKOO11} and systematically studied in \cite{HI16}. A quasi-shuffle of rooted trees tailored to study pattern Hopf algebras was \ty{developed} in \cite{PV22}. Let me finally mention that quasi-shuffles have also been applied to stochastic calculus in \cite{E-FP21}.

The archetypal example of commutative semi-group is $(\N^*,+)$, so let us write the sticky counterpart to Example \ref{ex:shuffles}, second example.
\begin{example}
 For $(\Omega,\bullet)=(\N^*,+)$ and $\lambda=1$ we have
 \begin{align*}
  (213)\stuffle (51)& =(21351)+(21531)+(21513)+(2181)+(2154)+
  (25131)+2(25113)+(2514) \\
  + & (2523)+(2631)+(2631)+(264)+
  (52131)+2(52113)+(51213)+(5214)+(5223) \\
  + & (5313)+
  +(7131)+2(7113)+(723)+(714).
 \end{align*}
\end{example}
From their definitions it is clear that the shuffle and $\lambda$-shuffle products are commutative. It is less clear, but still true that they are associative.
\begin{theo} \cite{Ho00,EML53}
 Let $\lambda\in\R$ and $(\Omega,\bullet)$ a commutative semigroups (resp. $\Omega$ a non-empty set). Then $(\calW_\Omega,\shuffle_\lambda,\emptyset)$ (resp. $(\calW_\Omega,\shuffle,\emptyset)$) is an associative and commutative unital algebra graded over $\Omega$ by the weight $||.||_\Omega$ (resp. graded over $\N$ by the length $|.|$) and filtered by the length.
\end{theo}
\begin{rk}
 This theorem is proven by induction on the length of words. The proof is cumbersome but not difficult and we skip it here. Notice that for $\lambda\neq0$, the associativity and commutativity of the semi-group product $\bullet$ is necessary to have the associativity and commutativity of $\shuffle_\lambda$.
\end{rk}
Let us finish this Subsection by pointing out that, while we will need here only the concatenation, shuffle and $\lambda$-shuffle products of words, they carry much richer algebraic structures. First, of course, the shuffle and stuffle products of words are \ty{examples of} dendriform and tridendriform structures (see \cite{ebrahimi2008rota} and \cite{catoire2023tridendriform}). Words also carry \ty{an  operad} structure (see \cite{dotsenko2020word}).

Furthermore, the vector space of words can be endowed with various coproducts (the deshuffle coproduct and the deconcatenation coproduct in particular) which endow them with bialgebra structures. Since these bialgebras are graded they turn out to be Hopf algebras. We refer the reader to \cite{Ma03} for a gentle introduction to bialgebras and Hopf algebras. These various structures can be related by a duality relation. Finally, notice that the quasi-shuffle also carries a double bialgebra structure \cite{Fo22}.

Finally, while we will here essentially deal with the shuffle and stuffle products on words, other structures are also relevant to the theory of MZVs. Examples are the aforementioned dendriform and tridendriform structures, but the structure of Zinbiel algebra of words also \ty{plays} a role (see \cite{chapoton2022zinbiel}). We do not discuss these structures further since they will \ty{play no} role in the rest of this chapter.

\subsection{{Multiple Zeta} values as algebra morphisms} \label{subsec:MZV}

MZVs can be defined as algebra morphisms from the algebras of words already defined to real numbers. They are not defined on all words, a fact that justifies the following definition.
\begin{defn} \label{def:conv_words}
\begin{itemize}
 \item For $\Omega=\{x,y\}$, a word $w\in\calW_{\{x,y\}}$ is called {\bf convergent} if it is the empty word or starts with $x$ and ends with $y$. We write $\calW^{\rm conv}_{\{x,y\}}$ the set of convergent words in $\calW_{\{x,y\}}$:
 \begin{equation*}
  \calW^{\rm conv}_{\{x,y\}}=(x)\sqcup\calW_{\{x,y\}}\sqcup (y)\bigcup\{\emptyset\}.
 \end{equation*}
 \item For $\Omega=\N^*$, a word $w\in\calW_{\N^*}$ is called {\bf convergent} if it is empty or if its first letter is not $1$. We write $\calW^{\rm conv}_{\N^*}$ the set of convergent words in $\calW_{\N^*}$:
 \begin{equation*}
  \calW^{\rm conv}_{\N^*}=\left(\bigcup_{n=2}^{+\infty}(n)\sqcup\calW_{\N^*}\right)\bigcup\{\emptyset\}.
 \end{equation*}
\end{itemize}
\end{defn}
The next Proposition is trivial to show but important for our future constructions.
\begin{prop}
 $(\calW^{\rm conv}_{\{x,y\}},\shuffle,\emptyset)$ is a subalgebra of $(\calW_{\{x,y\}},\shuffle,\emptyset)$ and for any $\lambda\in\R$, $(\calW^{\rm conv}_{\N^*},\shuffle_\lambda,\emptyset)$ is a subalgebra of $(\calW_{\N^*},\shuffle_\lambda,\emptyset)$.
\end{prop}
We can now define MZVs as iterated series and integrals. We will assume that they are convergent on their domain of definition \ty{and later give a proof} for the iterated series. For the iterated integrals, the same approach to show convergence would have also worked but we choose instead to use known results of polylogarithm theory to illustrate the versability of \ty{our} methods. A proof of existence using standard analysis methods exists in the \ty{literature}, and is actually a standard exercise; see for example the internship report \cite{Co22} (in \ty{French}).
\begin{defn} \label{def:MZVs}
 Set $\omega_x$ and $\omega_y$ to be two forms defined by 
 \begin{equation*}
  \omega_x(t)=\frac{dt}{t},\qquad\omega_y(t)=\frac{dt}{1-t}.
 \end{equation*}
 Then the {\bf shuffle \ty{multiple zeta} values} (shuffle MZVs for short) is the map $\zeta_\shuffle:\calW_{\{x,y\}}^{\rm conv}\longrightarrow \R$ \ty{defined by} 
%  \begin{equation*}
%   \zeta_\shuffle:\calW_{\{x,y\}}^{\rm conv}\longrightarrow \R
%  \end{equation*}
 $\zeta_\shuffle(\emptyset):=1$ and for all $k$ in $\N^*$
 \begin{equation*}
  \forall(\eps_1,\cdots,\eps_k)\in\{x,y\}^k,~
  \zeta_\shuffle(\epsilon_1\cdots\epsilon_k):=\int_{1>t_1>\cdots>t_k>0}\omega_{\epsilon_1}(t_1)\cdots\omega_{\epsilon_k}(t_k).
 \end{equation*}
 The {\bf stuffle \ty{multiple zeta} values} (stuffle MZVs for short) is the map $\zeta_\stuffle:\calW_{\N^*}^{\rm conv}\longrightarrow \R$ \ty{defined by} 
 $\zeta_\stuffle(\emptyset):=1$ and 
%  \begin{equation*}
%   \zeta_\shuffle(s_1\cdots s_k):=\sum_{n_1>\cdots>n_k>0}\frac{1}{n_1^{s_1}}\cdots\frac{1}{n_k^{s_k}}.
%  \end{equation*}
\begin{equation} \label{eq:MZV_stuffle1}
   \zeta_\stuffle(s_1\cdots s_k):=\sum_{n_1>\cdots>n_k>0}\frac{1}{n_1^{s_1}\cdots n_k^{s_k}}.
 \end{equation}
\end{defn}
\cy{\begin{rk}
     Notice that we define a multiple zeta values as a map rather than the more typical (and grammatically correct) definition of the images of this map. We make this choice since the properties of the image will be easier to state as property of the map.
    \end{rk}}
Our goal in this chapter will be to \ty{build} generalisations of these two MZVs that preserve their properties in some sense. So let us now look at their properties. The first of these properties is the most well-known. 
\begin{theo} \label{theo:alg_prop_MZVs_stuffle_shuffle}
 $\zeta_\shuffle$ and $\zeta_\stuffle$ are   respectively algebra morphisms for the shuffle product $\shuffle$ on $\calW_{\{x,y\}}$ and the stuffle product $\stuffle$ on $\calW_{\N^*}$:
 \begin{align} \label{eq:stuffle_prod}
 & \zeta_\stuffle(w\stuffle w') = \zeta_\stuffle(w)\zeta_\stuffle(w'), \\
 & \zeta_\shuffle(w\shuffle w') = \zeta_\shuffle(w)\zeta_\shuffle(w').\label{eq:shuffle_prod}
\end{align}
\end{theo}
\begin{rk}
 The stuffle version of this result seems to be due to Hoffman \cite{Ho97}, but MZVs were already introduced in \cite{Ecalle81,Ecalle81b}, and some properties stated without proof (see point (ii) page 135 of \cite{Ecalle81} and Equation (12e15) page 429 of \cite{Ecalle81b}). The shuffle version of the previous theorem is a direct consequence of properties of Chen integrals studied in \cite{Ch77}. We will later give algebraic proofs of these results.
\end{rk}
We now \ty{give} a pedestrian definition of the binarisation map.
\begin{defn} \label{defn:bin_map_words}
 The {\bf binarisation map} is a map $\fraks:\calW_{\N^*}\longrightarrow\calW_{\{x,y\}}$ defined by $\fraks(\emptyset):=\emptyset$, $\fraks(s):=(\underbrace{x\cdots x}_{s-1~ times}y)$ and 
 \begin{equation*}
  \fraks(s_1\cdots s_k):=\fraks(s_1)\sqcup\cdots\sqcup\fraks(s_k).
 \end{equation*}
 We then extend $\fraks$ by linearity to the whole of $\calW_{\N^*}$.
\end{defn}
In other words, we have 
\begin{equation} \label{eq:binarisation_map}
 \fraks(s_1\cdots s_k)=(\underbrace{x\cdots x}_{s_1-1~ times}y\underbrace{x\cdots x}_{s_2-1~ times}y\cdots\underbrace{x\cdots x}_{s_k-1~ times}y).
\end{equation}
In particular, $\fraks$ is a morphism of algebras for the concatenation products and maps words of weight $N\in\N$ to words of length $N\in\N$.

According to \cite[Section 9]{Za94} where it first appeared, the next result is based on an observation by Kontsevich (although \cite[Remark 4 page 431]{Ecalle81b} might be a reference to this fact, as pointed out in \cite{Ma13}).
\begin{theo}[Kontsevich's relation] \label{thm:shuffle_stuffle_words}
 $\fraks$ maps convergent words to convergent words and for any $w\in\calW_{\N^*}^{\rm conv}$ we have 
 \begin{equation} \label{eq:Kontsevich}
 \zeta_\shuffle(\fraks(w)) = \zeta_\stuffle(w).
\end{equation}
\end{theo}
In particular, ${\rm Im}(\zeta_\stuffle) = {\rm Im}(\zeta_\shuffle)$ 
and this justifies the name ``\ty{multiple zeta} \emph{values}''; i.e. that we identify these maps and the elements of their image.
\begin{rk}
 This second property of MZVs is obtained by expanding $1/(1-t)$ in a series and \ty{permuting} series and integrals. It \ty{is} rigorously justified using basic analytical tools as was done for example in \cite{Co22}. Notice that once this computation is made one obtains an alternative form for stuffle MZVs, namely
 \begin{equation} \label{eq:MZV_stuffle2}
  \zeta_\stuffle(s_1,\cdots,s_k)=\sum_{m_1,\cdots,m_k=1}^{+\infty}\frac{1}{(m_1+m_2+\cdots+m_k)^{s_1}(m_2+\cdots+m_k)^{s_2}\cdots(m_k)^{s_k}}
 \end{equation}
 which \ty{coincides} with the series \eqref{eq:MZV_stuffle1} as can be seen with the change of variables $n_i=m_i+\cdots+m_k$.
\end{rk}
We can now directly state the third and last property of MZVs that we will study.
\begin{theo}[Hoffman's regularisation relations \cite{Ho92,Ho97}] \label{theo:Hoffman_reg_relations}
 For any convergent word $w$, 
$\fraks\left((1)\stuffle w)\right) - (y)\shuffle\fraks(w)$ is a convergent word and
\begin{equation} \label{eq:Hoffman_reg_rel}
 \fraks\left((1)\stuffle w)\right) - (y)\shuffle\fraks(w) \in \Ker(\zeta_\shuffle).
\end{equation}
\end{theo} 
There are many open conjectures on MZVs, due to Zagier \cite{Za94}, Hoffman \cite{Ho97}, Brown \cite{brown2012mixed}, Broadhurst and Kreimer \cite{broadhurst1997association} (see also \cite{carr2015broadhurst}) and many others... These conjectures are typically on the dimension of the $\Q$-vector space generated by MZVs of a given weight. Important results regarding these conjectures are for example \cite{brown2012decomposition,zagier2012evaluation} and \cite{ngo2021zagier,im2024zagier}. 

They have far-reaching consequences regarding the theory of \ty{transcendental} numbers. This is a very active field of research but rather far from my work and we will not present it further. Let us simply point out the recent text \cite{dac2022valeurs}
which attempts to summarize the state of the art regarding these conjectures.

My long-term goal is to build an \emph{arborified} version of the MZVs, namely maps $\zeta_\shuffle^T:\calF_{\{x,y\}}^{\rm conv}\longrightarrow\R$ and $\zeta_\stuffle^T:\calF_{\N^*}^{\rm conv}\longrightarrow\R$ that have properties generalising Equations \eqref{eq:stuffle_prod}, \eqref{eq:shuffle_prod} ,~\eqref{eq:Kontsevich} and \eqref{eq:Hoffman_reg_rel}. In order to make this statement precise, I need to define rooted forests.

\section{Combinatorics of rooted forests} \label{sec:comb_forest}

Rooted trees and forests can be seen as a generalisation of words as we will see later. Thus we introduce the combinatorics of rooted forests that I will use later on.

\subsection{The algebra of rooted forests}

There are many good introductions to graph theory, for example \cite{Wi96}. I recall some notions that will be useful in the sequel.
\begin{defn} \label{def:graph_stuff}
 \begin{itemize}
  \item A {\bf graph} is a pair of finite sets $G:=(V(G),E(G))$ with $E(G)\subseteq V(G)\times V(G)$. $E(G)$ is the set of edges of the graph and $V(G)$ the set of vertices of the graph. The {\bf empty graph} $(\emptyset,\emptyset)$ is denoted by $\emptyset$.
  \item A {\bf path} in a graph $G$ is a finite sequence of elements of $V(G)$: $p=(v_1,\cdots,v_n)$ such that for all $i\in[n-1]$, $(v_i,v_{i+1})$ is an edge of $G$.  
  By convention, there is always a path between a vertex and itself. 
  \item A graph $G$ is {\bf connected} if for any pair of vertices $(v_1,v_2)\in V(G)^2$ there exists a finite sequence of vertices $(u_0=v_1,u_1,\cdots,u_{n-1},u_n=v_2)$ of vertices of $G$ such that for all $i\in[n]$, there is a path from $(u_{i-1},u_i)$ is an edge of $E$ \emph{or} $(u_i,u_{i-1})$ is an edge of $E$. 
  \item Let $G_1$ and $G_2$ be two graphs. Their {\bf concatenation} is the graph $$G_1G_2:=(V(G_1)\bigsqcup V(G_2),E(G_1)\bigsqcup E(G_2)).$$
  \item Any graph $G$ can be written as the concatenation of non-empty connected graphs. These connected graphs are the {\bf connected components} of $G$.
 \end{itemize}
\end{defn}
\begin{rk}
 Notice that these graphs are the usual ones, and not the generalised graphs of the previous Chapter (Definitions \ref{def:graph} and \ref{def:solar_graphs}). In the categories of PROPs and TRAPs we needed these more sophisticated structures in order to build to free objects. In this Chapter, we will use that some usual, not generalised graphs also have a universal property, in the somewhat simpler category of operated algebras.
\end{rk}

Here we have actually defined oriented graphs since only these will appear in this work. In particular, ``graph'' will always be used for ``oriented graph''. Let me now introduce rooted trees and forests.
\begin{defn} \label{def:forests}
 \begin{itemize}
  \item For a graph $G=(V(G);E(G))$, let $\preceq$ be the binary relation on $V(G)$ defined by: $v_1\preceq v_2$ if, and only if, there exists a path from $v_1$ to $v_2$. We also denote by $\succeq$ the inverse relation. A {\bf directed acyclic graph} (DAG for short) is a graph such that $(V(G),\preceq)$ is a partially ordered set (poset for short).
  
  \item A {\bf forest} is a DAG such that there is at most one path between two vertices. A {\bf rooted forest} is a forest whose connected components each have a unique minimal element. These elements are called {\bf roots}. A {\bf rooted tree} is a connected rooted forest.
  \item Let $F$ be a rooted forest and $v_1,v_2$ be two vertices of $F$. If $(v_1,v_2)\in E(F)$\footnote{which implies $v_1\preceq v_2$}, then $v_1$ is called the {\bf direct ancestor} of $v_2$ and $v_2$ a {\bf direct descendant} of $v_1$. We write $v_1=a(v_2)$
  \item If a vertex of a forest $F$ has more than one direct descendant it is called a {\bf branching vertex} of $F$. Furthermore, a vertex that is maximal for the partial order $\preceq$ is called a {\bf leaf}.
  \item Let $\Omega$ be a set. A {\bf $\Omega$-decorated rooted forest} is a rooted forest $F$ together with an arbitrary {\bf decoration map} 
  $d_F:V(F)\mapsto\Omega$. For a rooted forest $(F,d_F)$ decorated by $\Omega$ and $\omega\in\Omega$, we write $V_\omega(F)\subseteq V(F)$ the set of vertices of $F$ decorated by $\omega$.
  \item Two rooted forests $F$ and $F'$ (resp. decorated rooted forests $(F,d_F)$ and $(F',d_{F'})$) are {\bf isomorphic} if a poset isomorphism  $f_V:V(F)\longrightarrow V(F')$ (resp. and $d_F = d_{F'}\circ f_V$) exists.
 \end{itemize}
 We write $\calF$ (resp. $\calF_\Omega$) \ty{for} the commutative algebra freely generated by isomorphism classes of rooted forests (resp. by $\Omega$-decorated rooted) with the product given by the concatenation of graphs. We also use $\calT$ and $\calT_\Omega$ for the vector spaces of isomorphism classes of rooted trees and $\Omega$-decorated rooted trees respectively.
\end{defn}
As usual, I always consider isomorphism classes of rooted forests and therefore identify trees and forests with their classes. Furthermore, when there is no need to specify the decoration map I simply write $F$ for a decorated forest $(F,d)$.

We now define gradings of rooted trees and forests similarly to the grading of words we presented above in Definition \ref{def:counting_words}.
\begin{defn}
 \begin{itemize}
  \item Let $F=(V(F),E(F))$ be a rooted forest. We set $|F|:=|V(F)|$.
  \item For any $\omega\in\Omega$ and $F$ a forest decorated by $\Omega$, let $\sharp_\omega F$ the number of vertices of $F$ decorated by 
  $\omega$:
  \begin{equation*}
   \sharp_\omega F:= |V_{\omega}|=|\{v\in V(F):d(v)=\omega\}|.
  \end{equation*}
  \item Let $(\Omega,\bullet)$ be a commutative semigroup and $(F,d)=((V(F),E(F)),d)\in\calF_\Omega$ be a rooted forest. We set 
  $||F||_\bullet:=\sum_{v\in V(F)}^\bullet d(v)$; where the sum is for the product $\bullet$.
 \end{itemize}
\end{defn}
Notice that the concatenation of graphs defined above is an associative and commutative product that stabilises rooted forests. It can be trivially extended to decorated forests: if $(F_1,d_1)$ and $(F_2,d_2)$ are two $\Omega$-decorated forests then their concatenation is $(F_1F_2,d_3)$ with $d_3(v):=d_1(v)$ when $v\in V(F_1)\subseteq V(F_1F_2)$ and $d_3(v):=d_2(v)$ when $v\in V(F_2)\subseteq V(F_1F_2)$.

It is also trivial that the concatenation of rooted forests respects the $\N$-grading given by the number of vertices and, if $(\Omega,\bullet)$ is a commutative semi-group, the $\Omega$-grading given by $||.||_\Omega$. We then have
\begin{prop}
 Let $(\Omega,\bullet)$ be a commutative semi-group. $(\calF,.,\emptyset)$ (resp. $(\calF_\Omega,.,\emptyset)$) is a $\N$-graded associative commutative algebra for the concatenation of rooted forests (resp. decorated rooted forests), with the empty forest as unit and number of vertices as graduation (resp. and is also a $\Omega$-graded algebra).
\end{prop}
\begin{rk}
 Non commutative versions of rooted trees and forests also \ty{exist}, see for example \cite{Fo13}. I will not introduce them here since they will play no role in the rest of this chapter.
\end{rk}
The last important object to define in the context of rooted trees is the grafting operator.
\begin{defn} \label{def:branching_map}
 Let $\Omega$ be a set and $B_+:\Omega\times\calF_\Omega\longrightarrow\calT_\Omega$ be the operation defined through the grafting operator which, to 
 a pair $(\omega,F=T_1\cdots T_k)$, associates the decorated tree obtained from $F$ by adding a root decorated by $\omega$ linked to each 
 root of $T_i$ for $i$ going from $1$ to $k$.
 \end{defn}
 \begin{example}
 Here is a graphical depiction of the action of the grafting operator for some small trees\footnote{ {the code to generate these trees was written by Lo\"ic Foissy.}}:
 \begin{equation*}
  B_+^\omega(\emptyset) = \tdun{$\omega$} \qquad B_+^\omega(\tdun{$\omega'$}) = \tddeux{$\omega$}{$\omega'$} \qquad 
  B_+^\omega(\tdun{$\omega'$}\tdun{$\omega''$}~) = \tdtroisun{$\omega$}{$\omega''$}{$\omega'$}.
 \end{equation*}
\end{example}
\begin{rk} \label{rk:ladder_trees}
 The grafting operator gives rise to a canonical injection $\iota_\Omega:\calW_\Omega\hookrightarrow\calF_\Omega$ which sends the empty word to the empty tree and non empty 
 words to {\bf ladder trees}:
 \begin{equation*}
  \iota_\Omega(\omega_1\cdots\omega_k):=B_+^{\omega_1}\circ\cdots\circ B_+^{\omega_k}(\emptyset).
 \end{equation*}
 It is in this sense that rooted trees and forests generalise words. In particular, we will want our generalised versions of MZVs to give the usual MZVs when restricted to ladder trees.
 
 These ladder trees got their name from the pictural representation of rooted trees. For example
 \begin{equation*}
  \iota_{\Omega}(abcd)=B_+^a\circ B_+^b\circ B_+^c\circ B_+^d(\emptyset)=\tdquatre{a}{b}{c}{d}.
 \end{equation*}

\end{rk}
As for words, rooted trees and forests enjoy many more structures, which we will not need here and therefore will not introduce. Let me simply point out that a coproduct exists that turns rooted forests into a Hopf algebra. This coproduct is called the Connes-Kreimer coproduct and allowed mathematicians to fully unravel the combinatorics of renormalisation \cite{CK1,CK2}. More will be said about this structure in Chapter \ref{chap:loc}.

\subsection{Operated structures}

We start by recalling the definition of the categories of operated structures \cite{Guo07}, as written in \cite{CGPZ2}. 
\begin{defn} \label{def:op_structures}
 Let $\Omega$ be a set. An {\bf $\Omega$-operated set} (resp. {\bf semigroup, monoid, vector space, algebra}) is a set 
 (resp. semigroup, monoid, vector space, algebra) $U$ together with a map\footnote{$\beta$ is just a map, without any necessary algebraic properties.} $\beta:\Omega\times U\longrightarrow U$.

 Let $(U,\beta_U)$ and $(V,\beta_V)$ be two $\Omega$-operated sets (resp. semigroups, monoids, vector spaces, algebras). A 
 {\bf morphism of $\Omega$-operated sets} (resp. semigroups, monoids, vector spaces, algebras) between $U$ and $V$ is a map (resp. a semigroup 
 morphism, a monoid morphism, a linear map, an algebra morphism) $\phi:U\longrightarrow V$ such that, for any $\omega$ in $\Omega$ and $u$ in $U$
 \begin{equation*}
  \phi(\beta_U(\omega,u)) = \beta_V(\omega,\phi(u)).
 \end{equation*}
 In other words, $\phi$ is such that diagram \ref{fig:operated_morphism} commutes.
\end{defn}
 \begin{figure}[ht!] 
  		\begin{center}
  			\begin{tikzpicture}[->,>=stealth',shorten >=1pt,auto,node distance=3cm,thick]
  			\tikzstyle{arrow}=[->]
  			
  			\node (1) {$\Omega\times U$};
  			\node (2) [right of=1] {$U$};
  			\node (3) [below of=1] {$\Omega\times V$};
  			\node (4) [right of=3] {$V$};

  			\path
  			(1) edge node [above] {$\beta_U$} (2)
  			(1) edge node [left] {$ {I_\Omega\times}\phi$} (3)
  			(3) edge node [below] {$\beta_V$} (4)
  			(2) edge node [right] {$\phi$} (4);
  			
  			\end{tikzpicture}
  			\caption{morphism of operated structures. }\label{fig:operated_morphism}
  		\end{center}
  	\end{figure}
\begin{example}
 For any set $\Omega$, $(\calF_\Omega,B_+)$ is an $\Omega$-operated algebra, with the operation given by the grafting operator.
\end{example}
This example is far from random, as the previous result shows. As a matter of fact, a great part of this Chapter \ty{rests upon} the next result. It was originally shown in \cite{KP13}, formulated in the present form in \cite{Guo07} and an alternative proof of this result can be found in \cite{CGPZ2}.
\begin{theo} \cite{KP13,Guo07} \label{thm:univ_prop_tree}
 Let $\Omega$ be a set. \ty{$\calF_\Omega$} is an initial object in the category of commutative $\Omega$-operated algebras, i.e. for any 
 commutative $\Omega$-operated algebra $(A,\beta)$,  {\ty{it} exists} a unique algebra morphism 
 $\Phi:\calF_\Omega\longrightarrow A$ such that diagram \ref{fig:univ_prop_forests} commutes
 \begin{figure}[ht!] 
  		\begin{center}
  			\begin{tikzpicture}[->,>=stealth',shorten >=1pt,auto,node distance=3cm,thick]
  			\tikzstyle{arrow}=[->]
  			
  			\node (1) {$\calF_\Omega$};
  			\node (2) [right of=1] {$\calF_\Omega$};
  			\node (3) [below of=1] {$A$};
  			\node (4) [right of=3] {$A$};

  			\path
  			(1) edge node [above] {$B_+^\omega$} (2)
  			(1) edge node [left] {$\Phi$} (3)
  			(3) edge node [below] {$\beta^\omega$} (4)
  			(2) edge node [right] {$\Phi$} (4);
  			
  			\end{tikzpicture}
  			\caption{Universal property of forests.}\label{fig:univ_prop_forests}
  		\end{center}
  	\end{figure}
 for every $\omega$ in $\Omega$.
\end{theo}
\begin{rk}
 This Theorem is proved constructively, by recursively defining the map $\Phi$ as a morphism of operated algebras. Its uniqueness is then shown indirectly. In particular, this proof still holds in the non commutative case, but we will not need this level of generality here.
\end{rk}
We use these universal properties in a special case, where the decorating set $\Omega$ has an algebra structure to lift maps on the decorating sets to maps on forests. We do this by using the original map to define an operation. This can be carried out in various ways, and we introduce here two that will be used later on. Such branchings were  introduced in \cite{CGPZ1} and further used in \cite{CGPZ2}.
\begin{defn} \label{defn:phi_hat}
 Let $\Omega$ be a commutative algebra and $\phi:\Omega\longrightarrow\Omega$ be a map. Let $\beta_\phi:\Omega\times\Omega\longrightarrow\Omega$ be the 
 operation of $\Omega$ on itself defined by 
 $\beta_\phi(\omega_1,\omega_2) := \phi(\omega_1.\omega_2)$. The {\bf branched $\phi$-map} (or branching of $\phi$) is the morphism of 
 commutative $\Omega$-operated algebras 
 $\widehat{\phi}:(\calF_\Omega,B_+)\longrightarrow(\Omega,\beta_\phi)$ whose existence and uniqueness is given by Theorem \ref{thm:univ_prop_tree}.
\end{defn}
Notice that the map $\widehat\phi$ is entirely determined by the relations
\begin{align}
 \widehat\phi(\emptyset) & = 1_\Omega, \nonumber \\
 \widehat\phi(F_1F_2) & = \widehat\phi(F_1)\widehat\phi(F_2), \label{eq:prod_hat_phi} \\
 \widehat\phi(B_+^\omega(F)) & = \phi\left(\omega\widehat\phi(F)\right). \nonumber
\end{align}
Before moving on, and since branching maps are a crucial element of the rest of the paper, let me work out some of their actions on small trees.
\begin{example}
 Let $\Omega$ and $\phi$ be as in the above Definition. Then
 \begin{align*}
  \widehat{\phi}(\emptyset) = 1_\Omega; \qquad & \widehat{\phi}(\tdun{$\omega$}) = \phi(\omega); \qquad \widehat{\phi}\left(\tddeux{$\omega_1$}{$\omega_2$} \right) = \phi\left(\omega_1\phi(\omega_2)\right); \\
  \widehat{\phi}\left(\tdtroisun{$\omega_1$}{$\omega_2$}{~$\omega_3$}\right) & =  \phi\left(\omega_1\phi(\omega_2)\phi(\omega_3)\right).
 \end{align*}
\end{example}
As already stated, there are other possibilities to lift maps from an algebra to rooted forests decorated by this algebra. Here is an other one that will play a role later on.
\begin{defn} \label{defn:phi_sharp}
 Let $\Omega_1,\Omega_2$ be two commutative algebras and $\phi:\Omega_1\longrightarrow\Omega_2$ be a map. Let 
 $\tilde\beta_\phi:\Omega_1\times\calF_{\Omega_2}\longrightarrow\calF_{\Omega_2}$ be the operation of $\Omega_1$ on 
 $\calF_{\Omega_2}$ defined by  $\tilde \beta_{\phi}(\omega_1,F) := B_+(\phi(\omega_1),F)$. The {\bf lifted $\phi$-map} (or lifting of $\phi$) 
 is the morphism of commutative $\Omega_1$-operated algebras 
 $\phi^\sharp:(\calF_{\Omega_1},B_+)\longrightarrow(\calF_{\Omega_2},\tilde\beta_{\phi})$ whose existence and uniqueness is given by Theorem 
 \ref{thm:univ_prop_tree}.
\end{defn}
As before, let me give examples of the action of a lifted map on small trees.
\begin{example}
 Let $\Omega_1,\Omega_2$ and $\phi$ be as in the above Definition. Then:
 \begin{align*}
  \phi^\sharp(\emptyset) = \emptyset; \qquad & \phi^\sharp(\tdun{$\omega$}) = \tdun{$\phi(\omega)$}\hspace{0.4cm}; \qquad \phi^\sharp\left(\tddeux{$\omega_1$}{$\omega_2$} \right) = \tddeux{$\phi(\omega_1)$}{$\phi(\omega_2)$}\hspace{0.5cm}; \\ 
  \phi^\sharp\left(\tdtroisun{$a$}{$b$}{~$c$}\right) & = 
  \tdtroisun{$\phi(a)$}{$\phi(b)$}{$\phi(c)\quad$}\hspace{0.6cm}.
 \end{align*}
 \end{example}
 
\subsection{Shuffle of rooted forests}

We will see later that branchings are related through Rota-Baxter maps to flattenings and shuffles of rooted forests. Let us start by introducing the \ty{latter}. They are a non-associative generalisation of the shuffle products of words. As far as I am aware, these products were introduced in my work \cite{Cl20}, however a different product with a similar construction was introduced earlier in \cite{Guo}.
\begin{defn} [\cite{Cl20}] \label{defn:shuffle_tree}
 Let $\Omega$ be a set (resp. $(\Omega,\bullet)$ be a commutative semigroup and $\lambda\in\R$). The {\bf shuffle product on trees}, $\shuffle$, (resp. the {\bf $\lambda$-shuffle product on trees}, $\shuffle_\lambda$,) of two forests $F$ and 
 $F'$ is defined recursively on $|F|+|F'|$. 
 
 If $|F|+|F'|=0$ (and thus $F=F'=\emptyset$), we set $\emptyset\shuffle\emptyset=\emptyset\shuffle_\lambda\emptyset=\emptyset$.
 
 For $N\in\N$, assume the shuffle (resp. $\lambda$-shuffle) products of forests has been defined on every forests $f,f'$ such that $|f|+|f'|\leq N$. Then for any two forests $F,F'$ such that $|F|+|F'|=N+1$;
 \begin{itemize}
  \item (Unit) if $F'=\emptyset$, set $\emptyset\shuffle F=F\shuffle\emptyset=F$; and $F\shuffle_\lambda\emptyset=\emptyset\shuffle_\lambda F=F$.
  \item (Compatibility with the concatenation product) if $F$ or $F'$ is not a tree, then we can write $F$ and $f$ uniquely as a concatenation of trees: $F=T_1\cdots T_k$ and $F'=t_1\cdots t_n$ with the $T_i$s and $t_j$s 
  nonempty, $k+n\geq3$ and set 
  \begin{equation*}
   F\shuffle F' = \frac{1}{kn}\sum_{i=1}^k\sum_{j=1}^n\left((T_i\shuffle t_j)T_1\cdots\widehat{T_i}\cdots T_n t_1\cdots\widehat{t_j}\cdots t_k\right)
  \end{equation*}
  (resp. 
  \begin{equation*}
   F\shuffle_\lambda F' = \frac{1}{kn}\sum_{i=1}^k\sum_{j=1}^n\left((T_i\shuffle_\lambda t_j)T_1\cdots\widehat{T_i}\cdots T_n t_1\cdots\widehat{t_j}\cdots t_k\right)~),
  \end{equation*}
  where $T_1\cdots\widehat{T_i}\cdots T_n$ stands for the concatenation of the trees $T_1,\cdots,T_n$ without the tree $T_i$.
  \item (Compatibility with the grafting) if $F=T = B_+^a(f)$ and $F'=T'=B_+^{a'}(f')$ are two nonempty trees, we set 
  \begin{equation*}
   T\shuffle T' = B_+^a(f\shuffle T') + B_+^{a'}(T\shuffle f')
  \end{equation*}
  (resp. 
  \begin{equation*}
   T\shuffle_\lambda T' = B_+^a(f\shuffle_\lambda T') + B_+^{a'}(T\shuffle_\lambda f') + \lambda B_+^{a\bullet a'}(f\shuffle_\lambda f')~).
  \end{equation*}
 \end{itemize}
 We extend these products by linearity to obtain products on the vector space $\calF_\Omega$.
\end{defn}
\begin{rk}
  The well-definedness of the products $\shuffle$, $\shuffle_\lambda$ follows from the fact that any forest can be uniquely written (up to permutation) as an iteration of concatenations and graftings.
  
  Notice also that I use the same symbols for shuffles on trees and shuffles on words, as whether we are working with words or with trees will be clear from context. So, as for words, we write $\stuffle:=\shuffle_1$ the 
  {\bf stuffle product on trees} and $\shuffle_{-1}$ the {\bf anti-stuffle product on trees}.
\end{rk}
As for words, $\shuffle_0=\shuffle$. I nevertheless make a distinction between the cases $\lambda=0$ and $\lambda\neq0$ as in the former case, \ty{$\Omega$} is not 
  required to have a semigroup structure. I will however sometimes treat the $\shuffle$ product as a special case of $\shuffle_\lambda$, keeping in mind that when dealing with 
  $\lambda=0$ (so with the shuffle product), I will always implicitly allow the decoration set to have no semigroup structure.  
\begin{example}
 Here are examples of stuffle of trees with $(\Omega,.)=(\N,+)$:
 \begin{gather*}
  \tdun{n}\tdun{m}\shuffle_\lambda\tdun{p} = \frac{1}{2}\Big(\tdun{n}\left(\tddeux{~m}{~p} + \tddeux{~p}{~m} + \lambda\tdun{m+p}\quad\right) + \tdun{m}\left(\tddeux{~n}{~p} + \tddeux{~p}{~n} + \lambda\tdun{n+p}\quad\right)\Big) \\
  \tdtroisun{q}{~m}{n}\shuffle_\lambda\tdun{p} = \tdquatrequatre{p}{q}{n}{m} + \frac{1}{2}\Big(\tdquatretrois{q}{m}{p}{n} + \tdquatretrois{q}{p}{m}{n} + \lambda\tdtroisun{q}{m+p}{n} +\tdquatrequatre{q}{p}{n}{m} + \tdquatretrois{q}{n}{p}{m~} + \tdquatretrois{q}{p}{n}{m~} + \lambda\tdtroisun{q}{n+p}{m~}\Big) + \lambda\tdtroisun{q+p}{m}{n}. 
 \end{gather*}
 \ty{Based on} these intermediate computations, one can compute more involved shuffles of trees. However due to their length,   we will not write down the result here. For example the shuffle 
 $\tdtroisun{r}{m}{n}\shuffle_\lambda\tdtroisun{s}{q}{p}$ is a sum of forty-four trees, of which twenty have six vertices, twenty have five vertices, and four have four vertices.
\end{example}
We have a simple but important result about the structures these shuffles endow spaces of rooted forests with.
\begin{prop} \label{prop:shuffles_trees}
 Let $\lambda\in\R^*$ and $(\Omega,\bullet)$ be a commutative semigroup; or $\lambda=0$ and \ty{$\Omega$ be} a set. Then $(\calF,\shuffle_\lambda,\emptyset)$ is an nonassociative, commutative, unital algebra.
\end{prop}
\begin{proof}
 The case $\lambda=0$ is a consequence of the more general case as every undefined product in $\Omega$ disappears if $\lambda$ is set to $0$. Therefore we will only explicitly work out the case $\lambda\neq0$. 
 
 Let $(\Omega,.)$ be a commutative semigroup and $\lambda\in\R$.
 \begin{enumerate}
  \item By definition of $\shuffle_\lambda$, $\emptyset\shuffle_\lambda F=F\shuffle_\lambda\emptyset=F$ for any $F\in\calF_\Omega$. Therefore $\emptyset$ is the unit for $\shuffle_\lambda$ as claimed.
  \item We prove the commutativity of $\shuffle_\lambda$ by induction on $|F|+|F'|$. If $|F|+|F'|=0$ then $F=F'=\emptyset$ and $F\shuffle_\lambda F'=\emptyset=F'\shuffle_\lambda F$ by definition. Let $N\in\N$ and assume that, for any 
  pair of forest $f,f'$ such that $|f|+|f'|\leq N$ we have $f\shuffle_\lambda f'=f'\shuffle_\lambda f$. Let $F, F'$ be two forests such that $|F|+|F'|=N+1$. We then distinguish three cases:
  \begin{itemize}
   \item If $F=\emptyset$ or $F'=\emptyset$, then $F\shuffle_\lambda F'= F'\shuffle_\lambda F$ since $\emptyset$ is the unit of $\shuffle_\lambda$.
   \item If $F=T=B_+^{a}(f)$ and $F'=T'=B_+^{a'}(f')$ are two nonempty trees, then 
   \begin{equation*}
    T\shuffle_\lambda T' - T'\shuffle_\lambda T = \lambda B_+^{a\bullet a'}(f\shuffle_\lambda f') - \lambda B_+^{a'\bullet a}(f'\shuffle_\lambda f).
   \end{equation*}
   The R.H.S. vanishes by commutativity of $(\Omega,\bullet)$ and the induction hypothesis.
   \item If $F$ or $F'$ is not a tree, we write $F=T_1\cdots T_k$ and $F'=t_1\cdots t_n$ with $k+n\geq3$. Then 
   \begin{align*}
    F\shuffle_\lambda F' & = \frac{1}{kn}\sum_{i=1}^k\sum_{j=1}^n\left((T_i\shuffle_\lambda t_j)T_1\cdots\widehat{T_i}\cdots T_n t_1\cdots\widehat{t_j}\cdots t_k\right) \\
			 & = \frac{1}{kn}\sum_{i=1}^k\sum_{j=1}^n\left((t_j\shuffle_\lambda T_i)t_1\cdots\widehat{t_j}\cdots t_kT_1\cdots\widehat{T_i}\cdots T_n\right)
   \end{align*}
   by the induction hypothesis and the commutativity of the concatenation product of trees. Exchanging the order of the summations we indeed find $F\shuffle_\lambda F' = F'\shuffle_\lambda F$.
  \end{itemize}
  We have treated the three possible cases. Thus $\shuffle_\lambda$ is indeed commutative.
 \end{enumerate}
\end{proof}
\begin{rk} \label{rk:future_research}
 We will focus here on the applications of these new shuffle products of trees and their nonassociativity to the study of AZVs. Linked questions, such that the existence of a coproduct associated to these shuffles and the 
 eventual existence of a comodule-bialgebra structure \cite{E-FFM17} are interesting questions, but away from the scope of the present work. As such, they are left for further investigations.
\end{rk}
\begin{rk}
 One can see\footnote{I thank Dominique Manchon for his very useful comments on this point.} that the coefficient $1/kn$ arising in the compatibility with the concatenation product of trees equation of the 
 definition of the shuffle products on trees will prevent associativity. However, one could also define ``unweighted'' shuffle products of rooted forests without this factor. It turns out that these unweighted shuffle products are still not associative. Indeed, if one computes $((T_1\dots T_n)\shuffle_\lambda T)\shuffle_\lambda(t_1\cdots t_m)$, we will see 
 appear terms containing the forests $(T_i\shuffle_\lambda T)(T_j\shuffle_\lambda t_k)$. Such terms will not be present in products  $(T_1\dots T_n)\shuffle_\lambda (T\shuffle_\lambda(t_1\cdots t_m))$.
 \end{rk}

Let us end this subsection by an illustration of the nonassocativity of $\shuffle_\lambda$, with $\lambda$ set to $0$ in order to have simpler 
expressions to handle.
\begin{coex}
 Let $\Omega$ be a set. Then an easy computation gives, for any $(a,b,c,d)\in\Omega^4$
 \begin{equation*}
  \left(\tdun{a}\tdun{b}\shuffle~\tdun{c}\right)\shuffle~\tdun{d} = \tdun{a}\tdun{b}\shuffle\left(\tdun{c}\shuffle~\tdun{d}\right) + \frac{1}{2}\left(\tddeux{a}{d} + \tddeux{d}{a}\right)\left(\tddeux{b}{c} + \tddeux{c}{b}\right) + \frac{1}{2}\left(\tddeux{b}{d} + \tddeux{d}{b}\right)\left(\tddeux{a}{c} + \tddeux{c}{a}\right).
 \end{equation*}
\end{coex}

\subsection{Back to words}

I now introduce the flattening maps. They have many different names. In \cite{Ma13}, they are called the arborifications (simple and contracting), following Ecalle in \cite{Ec92} (see also \cite{EV04}). They also appear in \cite{Ya20}. I choose here to follow \cite{CGPZ3} and \cite{Cl20} and to define these flattening maps through the universal properties of rooted forests given by Theorem \ref{thm:univ_prop_tree}.
\begin{defn} \label{defn:flattening_maps}
 Let $\lambda\in\R$ and $(\Omega,\bullet)$ be a commutative semigroup (resp. $\lambda=0$ and $\Omega$ a set). We define an operation of $\Omega$ on the commutative algebra $(\calW,\emptyset,\shuffle_\lambda)$ as:
 \begin{align} \label{eq:defn_c_plus}
 C_+ : ~ \Omega & \times\calW_\Omega \longrightarrow\calW_\Omega \\
         (\omega & ,w) \mapsto C_+^\omega(w):=(\omega)\sqcup w. \nonumber
\end{align}
Then the {\bf flattening map of weight $\lambda$} is the morphism of 
 $\Omega$-operated algebras $fl_\lambda:(\calF_\Omega,B_+)\mapsto(\calW_\Omega,C_+)$ whose 
 existence and uniqueness is given by Theorem \ref{thm:univ_prop_tree}. 
\end{defn}
Notice that these flattening maps are entirely determined by the following relations:
\begin{align*}
 fl_\lambda(\emptyset) & = \emptyset \\
 fl_\lambda(F_1F_2) & = fl_\lambda(F_1)\shuffle_\lambda fl_\lambda(F_2) \\
 fl_\lambda(B_+^\omega(F)) & = (\omega)\sqcup fl_\lambda(F).
\end{align*}
The following Lemma will be useful to prove the main result of this Section.
\begin{lemma} \label{lem:fl_shuffle}
 Let $\lambda\in\R$ and $(\Omega,\bullet)$ be a commutative semi-group (resp. $\lambda=0$ and $\Omega$ a set). Then $fl_\lambda$ is a algebra morphism for the $\lambda$-shuffle products of rooted forests and words: for any rooted forests $F$ and $F'$ in $\calF_\Omega$
 \begin{equation*}
  fl_\lambda(F\shuffle_\lambda F')=fl_\lambda(F)\shuffle_\lambda fl_\lambda(F')
 \end{equation*}
 with $\shuffle_\lambda$ the $\lambda$-shuffle of rooted forests in the LHS and the $\lambda$-shuffle of words in the RHS. 
\end{lemma}
\begin{proof}
 We prove this result for $(\Omega,\bullet)$ a commutative semi-group since the case where $\Omega$ is a set can be deduced from it by setting $\lambda=0$. The proof is by induction on $N=|F|+|F'|$. If $N=0$ then $F=F'=\emptyset$ and the result trivially holds. For $N\in\N$, let us assume that it holds for any rooted forests $f$ and $f'$ such that $|f|+|f'|\leq N$. Let $F$ and $F'$ be two rooted forests such that $|F|+|F'|=N+1$. We have three cases to consider:
 \begin{itemize}
  \item If $F=\emptyset$ or $F'=\emptyset$ the result trivially holds.
  \item If $F=T=B_+^a(f)$ and $F'=T'=B_+^b(f')$ are two rooted trees since $fl_\lambda$ is a map of $\Omega$-operated algebras we have by definition of $\omega_\lambda$, the $\lambda$-shuffle of rooted forests we have
  \begin{align*}
   fl_\lambda(F\shuffle_\lambda F') & =fl_\lambda(B_+^a(f)\shuffle_\lambda B_+^b(f')) \\
   & = fl_\lambda(B_+^a(f\shuffle_\lambda T')) + fl_\lambda(B_+^b(T\shuffle_\lambda f'))+\lambda fl_\lambda(B_+^{a\bullet b}(f\shuffle_\lambda f')) \\
   & = (a)\sqcup fl_\lambda(f\shuffle_\lambda T') + (b)\sqcup fl_\lambda(T\shuffle_\lambda f')+\lambda (a\bullet b)\sqcup fl_\lambda(f\shuffle_\lambda f').
  \end{align*}
  Since $|f|+ |T'|=|T|+|f'|=N$ and $|f|+|f'|=N-1$ we can use the induction hypothesis and we obtain
  \begin{align*}
   fl_\lambda(F\shuffle_\lambda F') & = (a)\sqcup (fl_\lambda(f)\shuffle_\lambda fl_\lambda(T')) + (b)\sqcup (fl_\lambda(T\shuffle_\lambda fl_\lambda(f'))+\lambda (a\bullet b)\sqcup (fl_\lambda(f)\shuffle_\lambda fl_\lambda(f')) \\
   & = ((a)\sqcup fl_\lambda(f))\shuffle_\lambda ((b)\sqcup fl_\lambda(f'))\qquad\text{by definition \ref{def:stuffle}} \\
   & = fl_\lambda(T)\shuffle_\lambda fl_\lambda(T')
  \end{align*}
  by definition of $fl_\lambda$. In this case the result also holds.
  \item If $F=T_1\cdots T_n$ and $F'=t_1\cdots t_k$ with $n+k\geq3$ we have by definition of the $\lambda$-shuffle of rooted forests
  \begin{align*}
   fl_\lambda(F\shuffle_\lambda F') & = \frac{1}{kn}\sum_{i=1}^k\sum_{j=1}^nfl_\lambda\left((T_i\shuffle_\lambda t_j)T_1\cdots\widehat{T_i}\cdots T_n t_1\cdots\widehat{t_j}\cdots t_k\right) \\
   & = \frac{1}{kn}\sum_{i=1}^k\sum_{j=1}^nfl_\lambda\left((T_i\shuffle_\lambda t_j)\right)\shuffle_\lambda fl_\lambda\left(T_1\cdots\widehat{T_i}\cdots T_n t_1\cdots\widehat{t_j}\cdots t_k\right).
  \end{align*}
  Since $n+k\geq3$ we have $|T_i|+|t_j|\leq N$ and we can use the induction hypothesis to each of these pair of trees. Then using that $fl_\lambda(TT')=fl_\lambda(T)\shuffle_\lambda fl_\lambda(T')$ we obtain
  \begin{equation*}
   fl_\lambda(F\shuffle_\lambda F') = \frac{1}{kn}\sum_{i=1}^k\sum_{j=1}^n fl_\lambda(t_1\cdots t_k)\shuffle_\lambda fl_\lambda(t_1\cdots t_k)=fl_\lambda(F)\shuffle_\lambda fl_\lambda(F')
  \end{equation*}
  which concludes the proof.
 \end{itemize}
\end{proof}
Before we move on, I state a simple yet important property of  
flattening maps.
\begin{prop} \label{prop:finite_Q_sum_flattening}
 Let $(\Omega,\bullet)$ be a commutative semigroup and $\lambda$ a rational (resp. integer) number. For any finite forest $F$ in 
 $\calF_\Omega$, $fl_\lambda(F)$ is a 
 finite sum of finite words with rational (resp. integer) coefficients.
 
 If $\Omega$ is a set, then for any finite forest $F$ in 
 $\calF_\Omega$, $fl_0(F)$ is again a 
 finite sum of finite words with integer coefficients.
\end{prop}
\begin{proof}
 This result is easily shown by induction on the number of vertices of $F$. It follows from the fact that the flattening of the empty 
 forest is trivially a finite sum of finite words with integer coefficients; that (provided $\lambda$ is rational; resp. integer) the 
 $\lambda$-shuffle 
 of two finite sums of finite words with rational (resp. integer) coefficients is a finite sum of finite words with rational (resp. integer) 
 coefficients since 
 $\Q$ (resp. $\Z$) is stable under multiplication and addition; and 
 finally that the concatenation by one letter of a finite sum of finite words with rational (resp. integer) coefficients is a finite sum of 
 finite words 
 with rational (resp. integer) coefficients.
 
 The case where $\Omega$ is a set follows from the case $\lambda=0$.
\end{proof}
As in \cite{CGPZ2}, one can build counterparts for words of Definitions \ref{defn:phi_hat} and \ref{defn:phi_sharp}.
\begin{defn} \label{def:branched_lifted_words}
 Let $\Omega$ be a commutative algebra, $\Omega_1$ and $\Omega_2$ be two sets. Let $\iota_\Omega$ 
 (resp. $\iota_{\Omega_1}$, $\iota_{\Omega_2}$) be the canonical inclusion of words into trees, which to any word associate a 
 ladder tree. Let $\phi:\Omega\longrightarrow\Omega$ and $\psi:\Omega_1\longrightarrow\Omega_2$ be two maps. The {\bf W-branched $\phi$ map} 
 $\widehat\phi_\calW:\calW_\Omega\longrightarrow\Omega$ is defined by 
 $\widehat\phi_\calW:=\widehat\phi\circ\iota_\Omega$; and the {\bf W-lifted $\phi$ map} $\phi^\sharp_\calW:\calW_{\Omega_1}\longrightarrow\calW_{\Omega_2}$ by 
 $\phi^\sharp_\calW:=\iota_{\Omega_2}^{-1}\circ\phi^\sharp\circ\iota_{\Omega_1}$.
\end{defn}
\begin{rk}
 \begin{itemize}
  \item The W-lifted $\phi$ map is well defined since $\phi^\sharp$ maps ladder trees (introduced in Remark \ref{rk:ladder_trees}) to ladder trees, thus the image of 
  $\phi^\sharp\circ\iota_{\Omega_1}$ is included in the image of $\iota_{\Omega_2}$.
  \item The maps $\widehat\phi_\calW$ and $\phi^\sharp_\calW$ can be recursively defined by the relations:
  \begin{equation*}
  \widehat\phi_\calW(\emptyset) := 1_\Omega,\qquad
  \widehat\phi_\calW\left((\omega)\sqcup w\right) := \phi\left(\omega\widehat\phi_\calW(w)\right)
 \end{equation*}
 for $\widehat\phi_\calW$ and 
 \begin{equation*}
  \phi^\sharp_\calW(\emptyset) := \emptyset, \qquad
  \phi^\sharp_\calW\left((\omega)\sqcup w\right) := \phi(\omega)\sqcup\phi^\sharp_\calW(w)
 \end{equation*}
 for $\phi^\sharp_\calW$.
 \end{itemize}
\end{rk}

\subsection{Rota-Baxter maps and an important Theorem}

The various structures defined above (branching, flattening, shuffles of words and rooted forests) are related through the notion of Rota-Baxter maps. I recall their definition now but will not enter more into this rich and fascinating topic. I refer instead the reader to \cite{Guo} for a rich and gentle introduction to the subject.
\begin{defn} \label{def:RB_maps}
 Let $\Omega$ be a commutative algebra and $\lambda$ a real number. A map $P:\Omega\longrightarrow\Omega$ is a {\bf Rota-Baxter map of weight 
 $\lambda$} if
 \begin{equation} \label{eq:RBO}
  P(a)P(b) = P(aP(b)) + P(P(a)b) + \lambda P(ab)
 \end{equation}
 for any $a$ and $b$ in $\Omega$.
\end{defn}
Let me give standard examples of Rota-Baxter maps, which will be of importance in the next Sections.
\begin{example} \label{ex:RBmap}
 \begin{enumerate}
  \item Let $X\subseteq L^1_{\rm loc}(\R)$ be an algebra of locally integrable functions on $\R$ equipped with the  {pointwise} product of 
  functions stable under the integration map $f\mapsto\left(x\mapsto\int_a^xf(x)dx\right)$. Then for any 
  $a$ in $\R$, the integration map is a Rota-Baxter operator of weight $0$. \label{ex:RB_integration}
  \item Let $l(\R)$ be the algebra of sequences on $\R$ equipped with the  {pointwise} product of sequences. The summation operator 
  $\Sigma:l(\R)\longrightarrow l(\R)$ defined by $\Sigma(u_n)(N):=\sum_{n=0}^N u_n$ is a Rota-Baxter map of weight $-1$. \label{ex:RB_sum_star}
  \item Let $\frakt_{-1}:\N^*\longrightarrow\N$ be the translation map by $-1$ and $\frakt^*_{-1}$ be its pull-back to $l(\R)$. Then $\frakt^*_{-1}\Sigma$ is 
  a Rota-Baxter map of weight $+1$. \label{ex:RB_sum}
 \end{enumerate}
\end{example}
The next result is the main Theorem of this long Section. The equivalences of the first three items were first shown in \cite{CGPZ3} in the locality framework and a partially new proof was given in \cite{Cl20}. In the same paper, the proof of the equivalence of the fourth item was shown. Here, I will give the proof of \cite{Cl20}.
\begin{theo} \cite[Theorem 2.13]{CGPZ1}, \cite[Theorems 2.20 and 5.8]{Cl20}. \label{thm:flattening} \\
 Let $\Omega$ be a commutative algebra and $P:\Omega\longrightarrow\Omega$ a linear map. For any real number $\lambda$, the following statements are equivalent:
 \begin{enumerate}
  \item $P$ is a Rota-Baxter map of weight $\lambda$.
  \item The $P$-branched map factorises through words: $\widehat{P}=\widehat{P}_\calW\circ fl_\lambda$. \label{thm:ii}
  \item $\widehat P_\calW$ is an algebra morphism for the $\lambda$-shuffle product, namely 
  $\widehat P_\calW(w\shuffle_\lambda w')=\widehat P_\calW(w)\widehat P_\calW(w')$ for any $w$ and $w'$ in $\calW_\Omega$. \label{thm:iii}
  \item \label{thm:iv} $\widehat{P}$ is an algebra morphism for the $\lambda$-shuffle product of rooted forests, namely, for any $F$ and $F'$ in $\calF_\Omega$:
  \begin{equation*}
   \widehat{P}(F\shuffle_\lambda F') = \widehat{P}(F)\widehat{P}(F').
  \end{equation*}
%   $\widehat{P}(F\shuffle_\lambda F') = \widehat{P}(F)\widehat{P}(F')$ for any $F,F'$ in $\calF_\Omega$.
%   \begin{equation*}
%    \widehat{P}(F\shuffle_\lambda F') = \widehat{P}(F)\widehat{P}(F').
%   \end{equation*}
 \end{enumerate}
\end{theo}
\begin{proof}
We prove the implications $(1)\Rightarrow(3)\Rightarrow(2)\Rightarrow(4)\Rightarrow(1)$.
\begin{itemize}
 \item $(1)\Rightarrow(3)$: Assuming that $(1)$ holds, we prove this result by induction on $p=|w|+|w'|$. Let $P$ be a Rota-Baxter operator of weight $\lambda$.
 \begin{itemize}
  \item If $p=0$, then $w=w'=\emptyset$ and the result trivially holds as both sides are equal to $1$.
  \item For $p$ a natural number, let us assume that 
  $\widehat P_\calW(w\shuffle_\lambda w')=\widehat P_\calW(w)\widehat P_\calW(w')$ for any words $w$ and $w'$ such that $|w|+|w'|\leq p$. Let 
  $w$ and $w'$ be two words such that $|w|+|w'|=p+1$. If $w=\emptyset$ or $w'=\emptyset$ the result once again trivially holds. Otherwise 
  there
  exist two words $\tilde w$ and $\tilde w'$ (eventually empty) and $\omega,\omega'$ in $\Omega$ such that $w=(\omega)\sqcup\tilde w$ and 
  $w'=(\omega')\sqcup\tilde w'$. Then by definition of $\widehat P_\calW$ we have
  \begin{align*}
   \widehat P_\calW(w)\widehat P_\calW(w') & = P\left(\omega\wPw(\tilde w)\right)P\left(\omega'\wPw(\tilde w')\right) \\
					   & = P\left(\omega\wPw(\tilde w)P\left(\omega'\wPw(\tilde w')\right)\right) + P\left(P\left(\omega\wPw(\tilde w)\right)\omega'\wPw(\tilde w')\right) \\
					   & \hspace{8cm}\quad + \lambda P\left(\omega\wPw(\tilde w)\omega'\wPw(\tilde w')\right) \\
					   & = P\left(\omega\wPw(\tilde w)\wPw(w')\right) + P\left(\wPw(w)\omega'\wPw(\tilde w')\right) + \lambda P\left(\omega\wPw(\tilde w)\omega'\wPw(\tilde w')\right)   
  \end{align*}
  were we have used that $P$ is a Rota-Baxter map of weight $\lambda$.
  
  On the other hand we have by definition of $\shuffle_\lambda$ and from the definition of $\wPw$:
  \begin{align*}
   \wPw(w\shuffle_\lambda w') & = \wPw\left((\omega)\sqcup(\tilde w\shuffle_\lambda w')\right) + \wPw\left((\omega')\sqcup(w\shuffle_\lambda\tilde w')\right) + \lambda \wPw\left((\omega\omega')\sqcup(\tilde w\shuffle_\lambda\tilde w')\right) \\
			      & = P\left(\omega\wPw(\tilde w\shuffle_\lambda w')\right) + P\left(\omega'\wPw(w\shuffle_\lambda\tilde w')\right) + \lambda P\left(\omega\omega'\wPw(\tilde w\shuffle_\lambda\tilde w')\right) \\
			      & = P\left(\omega\wPw(\tilde w)\wPw(w')\right) + P\left(\omega'\wPw(w)\wPw(\tilde w')\right) + \lambda P\left(\omega\omega'\wPw(\tilde w)\wPw(\tilde w')\right) \\
			      & \hspace{8cm}\quad \text{by the induction hypothesis}\\
			      & = P\left(\omega\wPw(\tilde w)\wPw(w')\right) + P\left(\wPw(w)\omega'\wPw(\tilde w')\right) + \lambda P\left(\omega\wPw(\tilde w)\omega'\wPw(\tilde w')\right) 
  \end{align*}
  by commutativity of $\Omega$. This concludes this part of the proof.
 \end{itemize}
 \item $(3)\Rightarrow(2)$: Assuming that $(3)$ holds, we once again prove this result by induction, this time on $p=|F|$.
 \begin{itemize}
  \item If $p=0$ then $F=\emptyset$ and the result trivially holds.
  \item For any $p$ a natural number, let us assume that $\wP(F)=\wPw(fl_\lambda(F))$ for any forest $F$ such that $|F|\leq p$. Let $F$ be a 
  forest 
  with exactly $p+1$ vertices. Therefore $F$ is nonempty and we have either $F=T=B_+^\omega(\tilde F)$ for some (eventually empty) forest 
  $\tilde F$ or $F=F_1F_2$ with $F_1$ and $F_2$ nonempty forests. In the first case we have
  \begin{align*}
   \wP(F=T=B_+^\omega(\tilde F)) & = P\left(\omega\wP(\tilde F)\right) \quad \text{by definition of }\wP \\
				 & = P\left(\omega\wPw(fl_\lambda(\tilde F))\right) \quad \text{by the induction hypothesis} \\
				 & = \wPw\left((\omega)\sqcup fl_\lambda(\tilde F)\right) \quad \text{by definition of }\wPw \\
				 & = \wPw\left(fl_\lambda(F)\right) \quad \text{by definition of }fl_\lambda.
  \end{align*}
  In the second case we have
  \begin{align*}
   \wP(F=F_1F_2) & = \wP(F_1)\wP(F_2) \quad \text{by definition of }\wP \\
		 & = \wPw\left(fl_\lambda(F_1)\right)\wPw\left(fl_\lambda(F_2)\right) \quad \text{by the induction hypothesis} \\
		 & = \wPw\left(fl_\lambda(F_1)\shuffle_\lambda fl_\lambda(F_2)\right) \quad \text{since }(iii)\text{ holds by assumption} \\
		 & = \wPw\left(fl_\lambda(F_1F_2)\right) \quad \text{by definition of }fl_\lambda.
  \end{align*}
  This concludes the induction.
 \end{itemize}
 \item $(2)\Rightarrow(4)$: Assuming that $(2)$ holds we have for any rooted forests $F$ and $F'$:
 \begin{align*}
  \wP(F\shuffle_\lambda F') & = \wPw(fl_\lambda(F\shuffle_\lambda F')) \\
  & = \wPw\left(fl_\lambda(F)\shuffle_\lambda fl_\lambda F')\right) \qquad\text{by Lemma \ref{lem:fl_shuffle}} \\
  & = \wPw\left(fl_\lambda(F)\right)\wPw\left(fl_\lambda(F)\right)\qquad\text{by definition of }\wPw \\
  & = \wP(F)\wP(F')
 \end{align*}
 since $(2)$ holds.
 \item $(4)\Rightarrow(1)$: For any $a$ and $b$ in $\Omega$, writing $\wP(T\shuffle_\lambda T')=\wP(T)\wP(T')$ with $T=\tdun{a}$ and $T'=\tdun{b}$ gives
 \begin{equation*}
  \wP(T)\wP(T') = P(a)P(b)
 \end{equation*}
 one the one hand and 
 \begin{equation*}
  \wP(T\shuffle_\lambda T') = \wP(\tddeux{a}{b})+\wP(\tddeux{b}{a})+\lambda\wP(\tdun{a b}\hspace{0.2cm}) = P(aP(b))+P(bP(a))+\lambda P(ab)
 \end{equation*}
 thus $P$ is a Rota-Baxter operator of weight $\lambda$.
\end{itemize}
 Hence we have proven the four implications, thus the theorem.
\end{proof}

\section{Arborified zeta values: series} \label{section:stuffle}

\subsection{Log-polyhomogeneous symbols}

I start by recalling the basic definition of analytical objects that will play an essential role for the definition of arborified zeta values (AZVs) as iterated series. The presentation here follows closely \cite{MP10,Pa12}.
Our goal is to use the universal property of rooted forests (Theorem \ref{thm:univ_prop_tree}) for forests decorated by a suitable analytical space to define AZVs, and deduce their properties directly from their construction. Log-polyhomogeneous symbols are the suitable analytical space.
\begin{defn} \label{def:log_poly_hom}
 \begin{enumerate}
  \item Let $r\in\R$. A smooth function $\sigma:\R_{\geq 1}\to\C$ is a {\bf symbol} (with constant coefficients) of order $r$ on 
  $\R_{\geq 1}$ if
  \begin{equation}\label{eq:symbol}
   \forall k\in\N,~\exists C_k:~\forall x\in\R_{\geq 1}, \ |\partial_x^k\sigma(x)|\leq C_k x^{r-k}.
  \end{equation}
  The set of such symbols of order $r$ on $\R_{\geq 1}$ (with constant coefficients) is written $\mathcal{S}^r$%, which is a real 
%   vector space since $r\leq r'~\Longrightarrow S^r(\R_{\geq1})\subseteq S^{r'}(\R_{\geq1})$. 
  \item Let $\alpha\in\R$. A symbol $\sigma:\R_{\geq 1}\longrightarrow\R$ is a {\bf classical symbol} also called {\bf poly-homogeneous symbol} on 
  $\R_{\geq 1}$ of order $\alpha$ if it  lies in $\mathcal{S}^{\alpha}$ and if there exists a sequence 
  $(a_{\alpha-j})_{j\in\N}$ of real numbers such that for any $N\in \N$
  \begin{equation}\label{eq:sigmaN} 
   \sigma_{(N)}(x):=\sigma(x)-\sum_{j=0}^{N-1} a_{\alpha-j}x^{\alpha-j}
  \end{equation}
  lies in $\mathcal{S}^{\alpha-N}$. When this is the case, I will use the following classical short hand notation 
  \begin{equation*}
   \sigma  \sim\sum \sigma_{\alpha -j}
  \end{equation*}
  where $\sigma_{\alpha -j}$ is the function defined by $\sigma_{\alpha -j}(x)=a_{\alpha-j}x^{\alpha-j}$. 
  I denote by $CS^\alpha$ the set of classical symbols on $\R_{\geq 1}$ of order $\alpha\in\R$ and I set
%   and the set of symbols of order 
%   less of equal than $\alpha$ is a real vector space. Notice that one needs to be careful, as the difference of two symbols 
%   of order $\alpha$ can be of order $\alpha-1$, 
%   and we set
  \begin{equation*}
   CS :=\sum _{\alpha\in\R}CS^{\alpha},\qquad CS^{\leq\beta}:=\sum _{\alpha\leq\beta}CS^{\alpha}
  \end{equation*} 	
  the linear span of classical symbols of all orders resp. of order less or equal to $\beta$). 
  \item Let $k\in\N$ and $\alpha\in\R$. A {\bf log-polyhomogeneous symbol of order $(\alpha,k)$} is a function $f:\R_{\geq1}\longrightarrow\R$ 
  such that, for any $x\in\R_{\geq1}$
  \begin{equation} \label{eq:log_symb}
   f(x) = \sum_{l=0}^kf_l(x)\log^l(x)
  \end{equation}
  with $f_l\in CS^\alpha$. I denote by $\calP^{\alpha;k}$ the real vector space spanned by log-polyhomogeneous symbol of 
  order 
  $(\alpha,k)$ and also define
  \begin{equation*}
   \calP^{\alpha;*} := \bigcup_{k\in\N}\calP^{\alpha;k}; \quad \calP^{*;*} := \bigcup_{k\in\N}\calP^{*;k}.   
  \end{equation*}
  with $\calP^{*;k}$ the linear span over $\R$ of  all $\calP^{\alpha;k}$ for $\alpha\in\R$.
 \end{enumerate}
\end{defn}
\begin{rk}
 In general, linear combination of symbols (resp. classical symbols, log-polyhomogeneous symbols) of same orders can give a symbols (resp. classical symbol, log-polyhomogeneous symbol) of lower \ty{degree}. However, it is clear from the definition that $\calS^r$ is a vector space since $r\leq r'~\Longrightarrow S^r(\R_{\geq1})\subseteq S^{r'}(\R_{\geq1})$. The same is true for $CS^{\leq\alpha}$. 
\end{rk}
I choose to restrict our discussion to the case $\alpha\in\R$. Symbols (resp. classical, log-polyhomogeneous) of complex orders can be introduced in a similar fashion, however we will not require this level of generality for the construction and study of stuffle arborified zeta values. For the same reason I choose to define these objects on $\R_{\geq1}$ rather than on $\R_+$, $\R$ or $\R^d$. I did not indicate that these functions are defined on $\R_{\geq1}$ on the notations $\calS^r$, $CS^\alpha$ and $\calP^{\alpha,k}$ for the sake of readability.
\begin{rk}
 Besides these restrictions, \cy{the above definition of log-polyhomogeneous symbols differs slightly to the one of \cite{MP10}. I require the} functions $f_l$ to be classical symbols rather than symbols. This is because these objects appear more naturally in the subsequent developments.
\end{rk}
For any connected subset $I$ of $\R_{\geq1}$ let $\calP^{\alpha,k}(I)$ be the sets of log-polyhomogeneous symbols over $I$. These objects have the same definition than  $\mathcal{P}^{\alpha;k}(\R_{\geq1})$ with $x\in\R_{\geq1}$ replaced by $x\in I$. Notice that 
for two connected subsets $I$, $J$ of $\R_{\geq1}$, $I\subseteq J\Rightarrow \calP^{\alpha,k}(J)\subseteq\calP^{\alpha,k}(I)$.

The following Lemma is a direct consequence of the definition of log-polyhomogeneous symbols and standard results of real analysis:
\begin{lemma} \label{lem:existence_lim}
 Let $\sigma\in\mathcal{P}^{\alpha;k}(I)$ be a log-polyhomogeneous symbol over $I\subseteq\R_{\geq1}$ with $I$ unbounded. If $\alpha<0$, then $\sigma$ admits a finite
 limit in $+\infty$ for any $k$. 
 If $\alpha=0$ then $\sigma$ admits a finite limit in $+\infty$ if and only if $k=0$.
\end{lemma}
In the following Sections, the case of classical symbols with integer orders will play an important role. In particular, we will need to integrate them. This justifies the introduction of log-polyhomogeneous symbols since \ty{the latter}, contrarily to classical symbols with integer orders, are stable under integration. 

Before moving on to the the Euler-MacLaurin formula, I should recall some important properties of log-polyhomogeneous symbols. Although, as pointed out before, the definition used here of log-polyhomogeneous symbols differs slightly from the one of \cite{Pa12}, it is simple enough to check that the proofs of \cite{Pa12} still work.
\begin{prop} (\cite[Proposition 2.55]{Pa12}) \label{prop:translation}
 Let $\alpha\in\R$ and $k\in\N$. Then for any real number $a<0$, the pull-back translation map $\mathfrak{t}^*_a$ stabilises $\mathcal{P}^{\alpha;k}$ i.e.
 for any $\sigma\in\mathcal{P}^{\alpha;k}$
 \begin{equation*}
  \mathfrak{t}^*_a\sigma:= \left(x\longmapsto\sigma(x+a)\right)
 \end{equation*}
 is an element of $\mathcal{P}^{\alpha;k}$.
\end{prop}
The last useful property of log-polyhomogeneous symbols is very simple to prove, see for example \cite[Exercice 2.18]{Pa12}. It reads:
\begin{prop} \label{prop:bifiltration}
 The space $\calP^{*;*}$ is a bifiltrated, i.e. for any $\alpha,\beta\in\R$ and $k,l\in\N$ we have
\begin{align*}
  & \alpha \leq \beta~\wedge~\alpha-\beta\in\Z \quad \Longrightarrow \quad \calP^{\alpha;k} \subseteq \calP^{\beta;k};\\
  & k \leq l\quad \Longrightarrow \quad \calP^{\alpha;k} \subseteq \calP^{\alpha;l}; \\
   &\calP^{\alpha;k}.\calP^{\beta;l} \subseteq \calP^{\alpha+\beta;k+l}.
 \end{align*}
\end{prop}

\subsection{Euler-MacLaurin formula}

Following the ideas \ty{developed} in \cite{CGPZ3}, the goal is to define AZVs by iterating a summation operator on log-polyhomogeneous symbols:
\begin{equation*}
 S(\sigma)(N):=\sum_{n=1}^N\sigma(n).
\end{equation*}
Of course, the operator $S$ does not stabilise the spaces of classical symbols nor log-polyhomogeneous symbols. So we use the {\bf Euler-MacLaurin formula} to interpolate $S$ at non-integer values. Recall that the Euler-MacLaurin formula reads
\begin{align}\label{eq:EML_sum}
 S(\sigma)(N) & = \int_1^N\sigma(x)dx + \frac{1}{2}\left(\sigma(N)+\sigma(1)\right) \nonumber\\
              & + \sum_{k=2}^K\frac{B_k}{k!}\,\left(\sigma^{(k-1)}(N)- \sigma^{(k-1)}(1)\right) + \frac{(-1)^{K+1}}{K!}\int_1^N \overline{B_{K}}(x)\,\sigma^{(K)}(x)\, dx
\end{align}
where $\overline{B_k}(x)= B_k\left(x-\lfloor x\rfloor \right)$ with $\lfloor x\rfloor$ the integer part of the real number $x$, and $B_k(x)$ the $k$-th Bernoulli polynomial. Notice that \eqref{eq:EML_sum} holds for any $K\geq2$.

Following once more \cite{MP10} and \cite{CGPZ2} we define an operator $P$ acting on smooth functions by
\begin{align} \label{eq:EML}
 P(\sigma)(x) & = \int_1^x\sigma(t)dt + \frac{1}{2}\left(\sigma(t)+\sigma(1)\right) \nonumber\\
              & + \sum_{k=2}^K\frac{B_k}{k!}\,\left(\sigma^{(k-1)}(t)- \sigma^{(k-1)}(1)\right) + \frac{(-1)^{K+1}}{K!}\int_1^x \overline{B_{K}}(t)\,\sigma^{(K)}(t)\, dt.  
\end{align}
Notice that for any $N\in\N^*$ we have $P(\sigma)(N)=S(\sigma)(N)$: $P$ interpolates the discrete sum $S(\sigma)$.

\begin{lemma} \label{lem:int_log}
 For any $\alpha\in\R\setminus\Z_{\geq-1}$ and $l\in\N$ we have
 \begin{equation*}
  \int:\calP^{\alpha;l}\longrightarrow\calP^{\alpha+1;l} + \calP^{0;0}.
 \end{equation*}
 Furthermore 
 \begin{equation*}
  \int:\calP^{-1;l}\longrightarrow\calP^{0;l+1} + \calP^{0;0}.
 \end{equation*}
 In both cases, the operator $\int$ is defined by $(\int f)(x):=\int_1^x\sigma(t)dt$.
\end{lemma}
\begin{proof}
 For any $\alpha\in\R\setminus\Z_{\geq-1}$ and $l\in\N$ we have for $\sigma$ in $\calP^{\alpha;l}$
 \begin{equation*}
  \int_1^x\sigma(t)dt = \sum_{l=0}^k\left[\sigma_{j=0}^{N_l} a_{\alpha-j}^{(l)}\int_1^x t^{\alpha-j}\log^l (t)dt + \int_1^x\sigma_{(N_l),l}(t)\log^l (t)dt\right]
 \end{equation*}
 where we used \eqref{eq:sigmaN} with obvious notations.
 \begin{itemize}
  \item Using that, if $\alpha\in\R\setminus\Z_{\geq-1}$ then $\alpha-j\in\R\setminus\Z_{\geq-1}$ for any $j\in\N$. Then we prove by induction 
  on $l$ that 
  \begin{equation*}
   x\mapsto\int_1^x t^\alpha\log^l(t)dt\in\calP^{\alpha+1;l}+\calP^{0;0}.
  \end{equation*}
  The case $l=0$ is shown by direct integration.
  
  Assuming that this result holds for some $l\in\N$, we have, by integration by parts
  \begin{equation*}
   \int_1^x t^\alpha\log^{l+1}(t)dt = \frac{x^{\alpha+1}}{\alpha+1}\log^l(x) + \frac{l}{\alpha+1} \int_1^x t^\alpha\log^l(t)dt.
  \end{equation*}
  By the induction hypothesis, we have $\left(x\mapsto\int_1^x t^\alpha\log^l(t)dt\right)\in\calP^{\alpha+1;l}+\calP^{0;0}$; then Proposition 
  \ref{prop:bifiltration} allows to conclude this induction.
  \item To end this proof, it is now enough to show that, for any $\alpha\in\R\setminus\Z_{\geq-1}$ and $l\in\N$, if 
  $\tau\in\cal^\alpha(\R_{\geq1})$ then it exists $\rho\in\calS^{\alpha+1}(\R_{\geq1})$ and $K\in\R$ such that
  \begin{equation*}
   \int_1^x\tau(t)\log^l (t)dt = \rho(x)\log^l (x) + K.
  \end{equation*}
  Clearly, $\int$ maps smooth functions of $\R_{\geq1}$ to smooth functions of $\R_{\geq1}$. Moreover, using the 
  bound \eqref{eq:symbol} we have
  \begin{equation*}
   \left|\int_1^x\tau(t)\log^l (t)dt\right| \leq C\int_1^x t^{\alpha}\log^l (t)dt
  \end{equation*}
  For some $C\in\R$. Then the induction of the previous point allows us to conclude.
 \end{itemize}
 In the case of $\calP^{-1;l}$ the same arguments hold but the integration (resp. integration by parts) will add one logarithm.
\end{proof}

We further recall a classical Lemma of the theory of log-polyhomogeneous symbols
\begin{lemma} \label{lem:deriv_log}
 The differentiation operator $\partial:\sigma\mapsto\sigma'$ sends $\calP^{\alpha;l}$ to $\calP^{\alpha-1;l}$. We also have, for any 
 $\sigma\in\calS^\alpha(\R_{\geq1})$ that $\partial(\sigma\log^l) = \tau\log^l$ for some $\tau\in\calS^{\alpha-1}$.
\end{lemma}
We can now state and prove the main Proposition of this subsection. This result is a refinement of \cite[Proposition 8]{MP10}.
\begin{prop} \label{prop:prop_P}
 For any $\alpha\in\R\setminus\Z_{\geq-1}$ and $k\in\N$ we have for the operator $P$ of Equation \eqref{eq:EML}
 \begin{align*}
  P:\mathcal{P}^{\alpha;k} & \longrightarrow \mathcal{P}^{\alpha+1;k} + \calP^{0;0}; \\
  P:\mathcal{P}^{-1;k} & \longrightarrow \mathcal{P}^{0;k+1} + \calP^{0;0}. \\
 \end{align*}
\end{prop}

 \begin{proof}
  \begin{enumerate}
   \item 
   Let us start with the case $\alpha\notin\Z_{\geq1}$. Applying Lemmas \ref{lem:int_log} and \ref{lem:deriv_log} in the Euler-MacLaurin 
   formula
   \eqref{eq:EML} we see that its is enough to show that it exists $\tau\in\calS^{\alpha+1-N}$ so that, for a big enough $K$
   \begin{equation*}
    \int_1^x \overline{B_{K}}(t)\,\sigma^{(K)}(t)\, dt = \tau(x)\log^k(x).
   \end{equation*}
   Using that $\overline{B_{K}}(t)$ is bounded and twice Lemma \ref{lem:deriv_log} and Lemma \ref{lem:int_log} we obtain that such a 
   $\tau\in\calS^{\alpha-K+1}$ exists. Therefore we obtain the desired bound by taking $K\geq N$.
   \item The case $\alpha=-1$ is exactly similar, using the second part of Lemma \ref{lem:int_log}.
  \end{enumerate}
 \end{proof}
 
 \subsection{Arborified zeta values as series}

\begin{defn}
 Let $\calR:\N^*\longrightarrow\calP^{*;*}$ be the map defined by
 \begin{equation*}
  \calR(\alpha):=\left(x\mapsto x^{-\alpha}\right)
 \end{equation*}
 for any $\alpha\in\N^*$. 
\end{defn}
We lift $\calR$ as detailed in Definition \ref{defn:phi_sharp} to obtain the lifted $\calR$-map $\calR^\sharp:\calF_{\N^*}\longrightarrow\calF_{\calP^{*;*}}$.
\begin{defn}
 For $\lambda\in\{0,-1\}$, we define the operator $\frakS_\lambda$ as $\frakt^*_{\lambda}P$. We further define $\calZ_\lambda:=\widehat{\frakS}_\lambda\circ\calR^\sharp:\calF_{\N^*}\longrightarrow\calP^{*;*}$
\end{defn}
The simple, yet important subsequent Lemma is a consequence of Propositions \ref{prop:prop_P} and \ref{prop:translation}
\begin{lemma} \label{lem:ana_prop_frakS}
  For any $\alpha\in\R\setminus\Z_{\geq-1}$, $k\in\N$ and $\lambda\in\{0,-1\}$ we have
 \begin{align*}
  \frakS_\lambda:\mathcal{P}^{\alpha;k} & \longrightarrow \mathcal{P}^{\alpha+1;k}(\R_{\geq1-\lambda}) + \calP^{0;0}(\R_{\geq1-\lambda}) \\
  \frakS_\lambda:\mathcal{P}^{-1;k} & \longrightarrow \mathcal{P}^{0;k+1}(\R_{\geq1-\lambda}) + \calP^{0;0}(\R_{\geq1-\lambda}). \\
 \end{align*}
\end{lemma}
Before stating the main analytic property of $\calZ_\lambda$, let us recall some useful notations for forests decorated by $\N^*$:
for any $\N^*$-decorated forest $(F,d)$, we write
 \begin{enumerate}
  \item $|F|:=|V(F)|$ as usual,
  \item $\sharp_n F$ the number of vertices of $F$ decorated by $n\in\N^*$:
  \begin{equation*}
   \sharp_n F:= |\{v\in V(F):d(v)=n\}|.
  \end{equation*}
 \end{enumerate}
 
 \begin{prop} \label{prop:order_calz}
 For any nonempty $\N^*$-decorated tree $(T,d)$ of root $r$ we have
 \begin{equation*}
  \calZ_\lambda(T) \in \calP^{1-d(r),\sharp_1 T}(\R_{\geq1-\lambda|F|}) + \calP^{0;0}(\R_{\geq1-\lambda|F|}).
 \end{equation*}
 For any nonempty $\N^*$-decorated tree $(F,d)$ with $F=T_1\cdots T_k$, with $T_i$ nonempty and of root $r_i$ for any $i\in\{1,\cdots k\}$,
 we have
 \begin{equation*}
  \calZ_\lambda(F) \in \calP^{1-r_{\rm min}(F),\sharp_1 F}(\R_{\geq1-\lambda|F|}) + \calP^{0;0}(\R_{\geq1-\lambda|F|})
 \end{equation*}
 with $r_{\rm min}(F):=\min_{i=1,\cdots,k}d(r_i)$.
\end{prop}

\begin{proof}
 We prove this result by induction on the number $n=|F|$ of vertices of the forest.
 
 For the case $n=1$ we have, for any $\alpha\in\N^*$
 \begin{equation*}
  \calZ_{\lambda}(\tdun{$\alpha$}) = \frakS_\lambda(t\mapsto t^{-\alpha}).
 \end{equation*}
 Since $(t\mapsto t^{-\alpha})$ lies in $\calP^{-\alpha,0}$, the result follows from Lemma \ref{lem:ana_prop_frakS}.
 
 Let us assume this result to hold for any nonempty forests of $n$ or fewer vertices. Let $F$ be a forest of $n+1$ vertices. We distinguish 
 two cases.
 \begin{itemize}
  \item If $F=F_1F_2$ with $F_1$ and $F_2$ non empty then we can write
  \begin{equation*}
   \calZ_\lambda(F) = \calZ_\lambda(F_1)\calZ_\lambda(F_2).
  \end{equation*}
  Indeed $\calZ_\lambda$ being the composition of two algebra morphisms, it is an algebra morphism. Then, using the induction hypothesis, the filtration properties 
  of Proposition \ref{prop:bifiltration} (which can be applied since $\calZ_\lambda(F_i)$ has integer order) and the remark that $I\subseteq J\Rightarrow \calP^{\alpha,k}(J)\subseteq\calP^{\alpha,k}(I)$, 
  we have 
  \begin{align*}
   \calZ_\lambda(F) & \in \calP^{2-r_{\rm min}(F_1)-r_{\rm min}(F_2),\sharp_1 F}(\R_{\geq1-\lambda|F|}) + \calP^{1-r_{\rm min}(F_1),\sharp_1 F_1}(\R_{\geq1-\lambda|F|})  \\ 
    & + \calP^{1-r_{\rm min}(F_2),\sharp_1 F_2}(\R_{\geq1-\lambda|F|}) + \calP^{0;0}(\R_{\geq1-\lambda|F|})
  \end{align*}
  which, by Proposition \ref{prop:bifiltration} implies the result for $F$ since $1-r_{\rm min}(F_1)\leq0$; $1-r_{\rm min}(F_2)\leq0$ and 
  since each of the orders are integers.
  \item In the case $F=T$ a tree, there exists a nonempty forest $\tilde F$ and a positive natural number $\alpha$ such that 
  $T=B_+^\alpha(\tilde F)$ since $|T|=n+1\geq2$. We then have
  \begin{equation*}
   \calZ(T) = \frakS\left(t\mapsto t^{-\alpha}\calZ(\tilde F)(t)\right).
  \end{equation*}
  Using the induction hypothesis on $\tilde F$ and the bifiltration property we then obtain
  \begin{equation*}
   t\mapsto t^{-\alpha}\calZ(\tilde F)(t) \in \calP^{1-d_{\rm min}-\alpha,\sharp_1\tilde F}(\R_{1-\lambda n}) + \calP^{-\alpha,0}(\R_{1-\lambda n}).
  \end{equation*}
  In the case $\alpha\geq2$, from Lemma \ref{lem:ana_prop_frakS} and the observation that $1-d_{\rm min}\leq0$, we obtain the desired result, once again 
  using Proposition \ref{prop:bifiltration} (since in this case $\sharp_1 T=\sharp_1 \tilde F$), which can be used for the first and the second indices as well.
  
  In the case $\alpha=1$, the same argument holds if one notices that $\sharp_1 T=\sharp_1 \tilde F + 1$.
  
  These two cases allow to conclude the proof of the Theorem.
 \end{itemize}
\end{proof}
We now introduce the set of rooted forests on which arborified zeta values will be defined.
\begin{defn} \label{def:convergent_forests}
 A $\N^*$-decorated tree $(T,d)$ is called {\bf convergent} if it is empty or if it has a root $r$ and $d(r)\geq2$, i.e. if the decoration of its root is strictly greater than one. A $\N^*$-decorated forest 
 $(F=T_1\cdots T_k,d)$ is called {\bf convergent} if $T_i$ is convergent for each $i\in\{1,\cdots,k\}$. Let $\calF_{\N^*}^{\rm conv}$ be the 
 subalgebra of convergent forests.
 \end{defn}
 It is clear that $\calF_{\N^*}^{\rm conv}$ that is a subalgebra of $\calF_{\N^*}$ for the concatenation of rooted forests since by definition $\emptyset\in\calF_{\N^*}^{\rm conv}$ and $\calF_{\N^*}^{\rm conv}$ is stable by concatenation of forests.
 
 We can finally define arborified zeta values. This definition is illustrated by diagram \ref{fig:defn_BZVs} below.
 \begin{defiprop} \label{defnprop:arborified_zeta}
 For any convergent $\N^*$-decorated forest $F$ and $\lambda\in\{0,-1\}$, $\calZ_\lambda(F)(x)$ admits a finite limit as $x$ goes to infinity. We define the maps 
 $\zeta^T_\stuffle,\zeta^{T,\star}_\stuffle:\calF_{\N^*}^{\rm conv}\longrightarrow\R$ by
 \begin{equation*}
  \zeta^T_\stuffle(F):=\lim_{x\to\infty}\calZ_{-1}(F)(x), \qquad \zeta^{T,\star}_\stuffle(F):=\lim_{x\to\infty}\calZ_{0}(F)(x) 
 \end{equation*}
 for any $F\in\calF_{\N^*}^{\rm conv}$.
\end{defiprop}
\begin{proof}
  For any convergent forest $F$ we have $1-d_{\rm min}\leq-1$. Therefore by Lemma \ref{lem:existence_lim} and Proposition
  \ref{prop:order_calz} the limits at $+\infty$ of $\calZ_\lambda(F)$ are well-defined and 
  finite for $\lambda\in\{0,-1\}$ provided that $F$ is convergent.
\end{proof}
Let us notice that the arborified zeta values defined here coincide with the branched zeta values of \cite{CGPZ1} in the convergent case where the renormalisation scheme reduces to a simple evaluation at $0$ of the regularisation 
parameters. We could therefore have defined AZVs through these branched zeta values. We have opted for this chapter to not contain the divergent counterparts of AZVs since these were introduced as a testing ground for the multivariate renormalisation scheme that is the main topic of the next chapter.
% be a self-contained presentation of AZVs. 
%Furthermore, the question of 
%renormalisation brings more involved analytic objects, such as meromorphic families of classical symbols, which are unnecessary in the convergent case.
 \begin{figure}[h!] 
  		\begin{center}
  			\begin{tikzpicture}[->,>=stealth',shorten >=1pt,auto,node distance=3cm,thick]
  			\tikzstyle{arrow}=[->]
  			
  			\node (1) {$\calF_{\N^*}^{\rm conv}$};
  			\node (2) [right of=1] {$\R$};
  			\node (3) [below of=1] {$\calF_{\calP^{*;*}}$};
  			\node (4) [right of=3] {$\calP^{*;*}$};

  			\path
  			(1) edge node [above] {$\zeta^T_\stuffle,~\zeta^{T,\star}_\stuffle$} (2)
  			(1) edge node [left] {$\calR^\sharp$} (3)
  			(3) edge node [below] {$\widehat{\frakS}_\lambda$} (4)
  			(1) edge node [left] {$\calZ_\lambda$} (4)
  			(4) edge node [right] {${\rm ev}_{\infty}$} (2);
  			
  			\end{tikzpicture}
  			\caption{Definition of arborified zeta values. }\label{fig:defn_BZVs}
  		\end{center}
  	\end{figure}
  	Using the fact that $\calZ_\lambda$ and ${\rm ev}_{\infty}$ are both algebra morphisms, we obtain the simple, yet important subsequent Proposition
  	\begin{prop} \label{prop:stuffle_alg_mor}
  	 The maps $\zeta^T_\stuffle$ and $\zeta^{T,\star}_\stuffle$ are algebra morphisms for the concatenation product of trees.
  	\end{prop}
\begin{rk}
Let us notice that the techniques used above can also be used to define branched Euler sums which correspond (in the words case) to convergent MZVs with some negative arguments. These objects are an active area of 
 research, see for example \cite{WX}.
\end{rk}

\subsection{Applying the main theorem}

The construction of the previous subsection can be adapted to build  multiple zeta values instead of arborified zetas simply by replacing 
$\calR^\sharp$ by $\calR^\sharp_\calW$ and $\widehat\frakS_\lambda$ by $\widehat\frakS_{\lambda,\calW}$. This gives an alternative construction of MZVs which will allow us to derive algebraic proofs of classical results.
\begin{defn}
 For $\lambda\in\{0,-1\}$, let $\calZ_{\lambda,\calW}:\calW_{\N^*}\longrightarrow\calP^{*;*}$ be the operator defined by 
 $\calZ_{\lambda,\calW}:=\calR_\calW^\sharp\circ\widehat\frakS_\calW$.
\end{defn}
We have an important analytic result which is proven in exactly the same fashion as Proposition \ref{prop:order_calz}, with only the case $w=C_+^\omega(\tilde w)$ to be 
 considered.
\begin{prop} \label{prop:calz_words_order}
 For any nonempty word $w=(\omega_1\cdots\omega_n)$, we have
 \begin{equation*}
  \calZ_{\lambda,\calW}(w) \in \calP^{1-\omega_1,\sharp_1 w}(\R_{\geq1-\lambda|w|}) + \calP^{0;0}(\R_{\geq1-\lambda|w|}).
 \end{equation*}
\end{prop}
Recall (Definition \ref{def:conv_words}) that a word $w$ written in the alphabet $\N^*$ is called convergent if it is empty or if its first letter on the left is 
 greater or equal to $2$ and that $\calW_{\N^*}^{\rm conv}$ is the 
 linear span of convergent words.

% We can now give the counterpart for words of Definition \ref{def:convergent_forests}.
% \begin{defn}
%  A word $w$ written in the alphabet $\N^*$ is called {\bf convergent} if it is empty or if its first letter on the left is 
%  greater or equal to $2$.  Let $\calW_{\N^*}^{\rm conv}$ be the 
%  linear span of convergent words.
% \end{defn}
As the terminology suggests, convergent words are the convergence domain of MZVs, for which we now give an alternative definition.
 \begin{defiprop} \label{defnprop:zeta_stuffle}
 For any convergent word written in the alphabet $\N^*$ and $\lambda\in\{0,-1\}$, $\calZ_{\lambda,\calW}(w)(x)$ admits a finite limit as $x$ 
 goes to infinity. We define the maps 
 $\zeta_\stuffle,\zeta^{\star}_\stuffle:\calW_{\N^*}^{\rm conv}\longrightarrow\R$ by
 \begin{equation*}
  \zeta_\stuffle(w):=\lim_{x\to\infty}\calZ_{-1,\calW}(w)(x), \qquad \zeta^{\star}_\stuffle(w):=\lim_{x\to\infty}\calZ_{0,\calW}(w)(x) 
 \end{equation*}
 for any $w\in\calW_{\N^*}^{\rm conv}$.
\end{defiprop}
\begin{proof}
 The proof of the statement is similar to the proof of Definition-Proposition \ref{defnprop:arborified_zeta}.
\end{proof}
We can \ty{now} start to deduce the algebraic proprties of MZVs from this definition. Before this, we need a very classical result.
\begin{lemma} \label{lem:subalgebra_shuffle_lambda}
 For any $\lambda\in\R$, $(\calW_{\N^*}^{\rm conv},\emptyset,\shuffle_\lambda)$ is a subalgebra of $(\calW_{\N^*},\emptyset,\shuffle_\lambda)$ .
\end{lemma}
\begin{proof}
 By definition $\emptyset\in\calW_{\N^*}^{\rm conv}$. The rest of the proof is easily carried out by induction on $|w|+|w'|$ for two convergent words 
 $w$ and $w'$, using Definition \ref{def:stuffle} of the product $\shuffle_\lambda$.
\end{proof}
\begin{prop} \label{prop:zeta_stuffle_mor}
 The map $\zeta_\stuffle$ is a algebra morphism for the stuffle product $\stuffle=\shuffle_1$. Furthermore the map $\zeta^\star_\stuffle$ is 
 an algebra morphism for the anti-stuffle product $\shuffle_{-1}$.
\end{prop}
\begin{proof}
 By Examples \ref{ex:RBmap}, \ref{ex:RB_sum} (resp. \ref{ex:RBmap}, \ref{ex:RB_sum_star}) we know that $\frakS_{-1}$ (resp. $\frakS_0$), when restricted 
 to $\N^*$, is a Rota-Baxter map of weight $+1$ (resp $-1$). One can then apply Theorem \ref{thm:flattening} \ref{thm:iii}  to get that 
 $N\mapsto\widehat\frakS_{-1,\calW}(N)$ (resp. $N\mapsto\widehat\frakS_{0,\calW}(N)$), is an algebra morphism for the stuffle 
 (resp. anti-stuffle) product. Taking the limit $N\to\infty$ gives to the result.
\end{proof}
We can now look at properties of AZVs and in particular relate them to MZVs but \ty{before} we will need a simple Lemma.
\begin{lemma} \label{lem:flattening_conv}
 For any $\lambda\in\R$, the flatting map $fl_\lambda$ maps convergent forests to convergent words. Furthermore, $\calF_{\N^*}^{\rm conv}$ is a subalgebra for the $\lambda$-shuffle of rooted forests $\shuffle^T_\lambda$.
\end{lemma}
\begin{proof}
 One can perform an easy proof by induction on the number of vertices of the convergent forest. For the empty forest, we have 
 $fl_\lambda(\emptyset)\in\calW_{\N^*}^{\rm conv}$. For a nonempty convergent tree $T=B_+^{n\geq2}(F)$ for some forest $F$. Then 
 $fl_\lambda(T)=C_+^{n\geq2}(fl_\lambda(F))\in\calW_{\N^*}^{\rm conv}$. Finally for the case $F=F_1F_2$, we have that $F_1$ and $F_2$ are 
 convergent if $F$ is and therefore $fl_\lambda(F) = fl_\lambda(F_1)\shuffle_\lambda fl_\lambda(F_2)\in\calW_{\N^*}^{\rm conv}$ by the 
 induction hypothesis and Lemma \ref{lem:subalgebra_shuffle_lambda}.
 
 Similarly, if $T$ and $T'$ (reps. $F$ and $F'$) are convergent trees (resp. forest), it is clear from the definition of the shuffle of rooted forests that $T\shuffle^T_\lambda T'$ (resp. $F\shuffle^T_\lambda F'$) is a linear combination of convergent trees (resp. forests).
\end{proof}
We can now state our main result.
\begin{theo} \label{thm:main_result_stuffle}
 For any convergent forest $F$, the convergent  arborified zeta value $\zeta^T_\stuffle(F)$ (resp. $\zeta^{T,\star}_\stuffle(F)$) is a finite 
 linear combination of convergent  {multiple zeta} values $\zeta_\stuffle(w)$ (resp. $\zeta^{\star}_\stuffle(w)$) with rational coefficients and can be written as a finite 
 linear combination of  {multiple zeta} values with integer coefficients given by the $1$-flattening (resp. the $-1$-flattening):
 \begin{equation*}
  \zeta_\stuffle^T=\zeta_\stuffle\circ fl_1\qquad\text{(resp. }\zeta_\stuffle^{T,\star}=\zeta^\star_\stuffle\circ fl_{-1}\text{)}.
 \end{equation*}

 Futhermore, the map $\zeta^T_\stuffle$ (resp. $\zeta^{T,\star}_\stuffle$) is an algebra morphism for the stuffle product of rooted forest $\stuffle^T:=\shuffle^T_1$ (resp. for the product $\shuffle^T_{-1}$).
\end{theo}
\begin{proof}
 Let $F$ be any convergent forest.
 
 Once again, by Examples \ref{ex:RBmap}, \ref{ex:RB_sum} (resp. \ref{ex:RBmap}, \ref{ex:RB_sum_star}) we know that $\frakS_{-1}$ (resp. 
 $\frakS_0$), when restricted to $\N^*$, is a Rota-Baxter map of weight $+1$ (resp $-1$). Applying Theorem \ref{thm:flattening} \ref{thm:ii} we
 get that $N\mapsto\widehat\frakS_{-1}(F)(N)$ (resp. $N\mapsto\widehat\frakS_{0}(F)(N)$), when restricted to $\N^*$, factorises through words, i.e. 
 that $\widehat\frakS_\lambda(F)(N) = \widehat\frakS_{\lambda,\calW}\circ fl_\lambda(F)(N)$ for any $N$ in $\Z_{\geq|F|+1}$ and 
 $\lambda\in\{0,-1\}$.
 
 Then by Lemma \ref{lem:flattening_conv} we can take the limit $N\to\infty$ in both sides which gives the result. 
 We obtain a finite sum with integer coefficients thanks to Proposition \ref{prop:finite_Q_sum_flattening}.
 
 Finally, the fact that $\zeta_\stuffle$ (resp. $\zeta^{\star}_\stuffle$) is an algebra morphism for the stuffle product of rooted forest $\stuffle^T$ (resp. for the product $\shuffle^T_{-1}$) follows from the fact that it is a composition of algebra morphisms for $\stuffle^T$ (resp. $\shuffle^T_{-1}$). First, once again, by Examples \ref{ex:RBmap} and \ref{ex:RB_sum} we can use Theorem \ref{thm:flattening} \ref{thm:iv} on $\widehat\frakS_\lambda$. Furthermore, $\calR^\sharp$ is an algebra morphism of the shuffle products of rooted forests as well since for any algebras morphism $P:\Omega_1\longrightarrow\Omega_2$ between two commutative algebras, the lifted map 
 $P^\sharp:\calF_{\Omega_1}\longrightarrow\calF_{\Omega_2}$ is an algebra morphism for the $\lambda$-shuffles of trees, for any value of 
 $\lambda\in\R$. This clearly holds from the definition of $P^\sharp$ and the fact that $P$ is an algebra morphism. However, it is can also be easily proven, once again by induction on the number of vertices of forests, using the fact that 
 $P^\sharp$ is a morphism of operated algebras.
 
 Notice that the result makes sense by the stability property of Lemma \ref{lem:flattening_conv}. In any case, 
 composing these algebra morphisms we obtain another algebra morphism, which is no other than $\zeta^T_\stuffle$ (resp. $\zeta^{T,\star}_\stuffle$).
\end{proof}
The results on AZVs obtained so far, the definitions and relations between branched zetas and  multiple zetas are summarized as the 
commutativity of diagram \ref{fig:MZV_AZV} below. Notice that the left-most arrow should be $fl_1$ for the un-starred AZVs and MZVs, and $fl_{-1}$ for the starred AZVs and MZVs.

 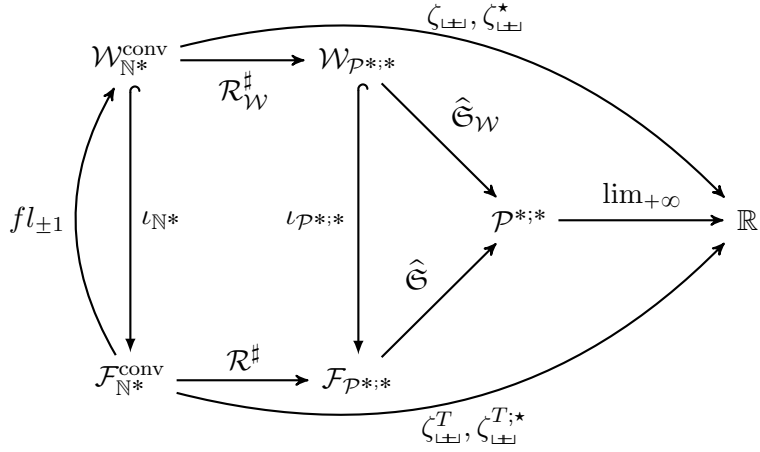
\begin{figure}[h!] 
  		\begin{center}
  			\begin{tikzpicture}[->,>=stealth',shorten >=1pt,auto,node distance=3cm,thick]
  			\tikzstyle{arrow}=[->]
  			
  			\node (1) {$\calW_{\N^*}^{\rm conv}$};
  			\node (2) [right of=1] {$\calW_{\calP^{*;*}}$};
  			\node (4) [below right of=2] {$\calP^{*;*}$};
  			\node (5) [below left of=4] {$\calF_{\calP^{*;*}}$};
  			\node (7) [left of=5] {$\calF_{\N^*}^{\rm conv}$};
  			\node (8) [right of=4] {$\R$};

  			\path
  			(1) edge node [below] {$\calR^\sharp_\calW$} (2)
  			(2) edge node [above right] {$\widehat\frakS_\calW$} (4)
  			(4) edge node [above] {$\lim_{+\infty}$} (8)
  			(7) edge node [above] {$\calR^\sharp$} (5)
  			(5) edge node [above left] {$\widehat\frakS$} (4);

  			\draw [right hook-latex] (1) -- node[right] {$\iota_{\N^*}$} (7);
  			\draw [right hook-latex] (2) -- node[left] {$\iota_{\calP^{*;*}}$} (5);
  			
  			\draw (1) to[bend left] node[above]{$\zeta_\stuffle,\zeta_\stuffle^\star$} (8);
  			\draw (7) to[bend right] node[below]{$\zeta^T_\stuffle,\zeta_\stuffle^{T;\star}$} (8);
  			\draw (7) to[bend left] node[left]{$fl_{\pm1}$} (1);
  			\end{tikzpicture}
  			\caption{ {Multiple zetas} and arborified zetas.}\label{fig:MZV_AZV}
  		\end{center}
  	\end{figure}
  	
% We will later apply the fouth point of Theorem \ref{thm:flattening} for the maps $\zeta^T_\stuffle$ and $\zeta^{T,\star}_\stuffle$. However, we will also apply it for the shuffle arborified zeta values. So in order to not repeat ourselves too much, let us now move on to these shuffle AZVs and come back after to the shuffle products of AZVs.

\section{Shuffle arborified zeta values} \label{section:shuffle}

\subsection{Chen integrals and arborification}

In \cite{Ch77} iterated integrals are recursively defined. One way to define them is as a map $Ch:\calW_X\longrightarrow \calI(I)$; where $I=[a,b]$ is a
closed interval, $\calI(I)$ is the set of continuous, integrable functions over $I$ and $X=\{f_1,\cdots,f_N\}$ is a finite subset of $\calI(I)$.
In \cite{Ch77} this recursive definition goes as follows:
\begin{defiprop} \label{defnprop:Chen_int}
 Let $I=[a,b]$ be a closed interval in $\R$ and $X=\{f_1,\cdots,f_N\}$ be a finite set with $f_i:I\longrightarrow\R$ smooth, continuous functions over 
 $I$. $Ch:\calW_X\longrightarrow \calI(I)$ is the linear map, whose action on the basis elements of $\calW_X$ is recursively defined by
 \begin{align*}
  Ch(\emptyset) & := \left(x\mapsto1~\forall x\in I\right)=:{\bf 1}, \\
  Ch((f_i)\sqcup w) & := \left(x\mapsto\int_a^x f_i(t)Ch(w)(t)dt~\forall x\in I\right)
 \end{align*} 
 for any $f_i$ in $X$ and $w$ in $\calW_X$.
\end{defiprop}
\begin{proof}
 In order to prove that this definition is consistent, one has to prove that $Ch(w)$ lies in $\calI(I)$ for any word $w$ in $\calW_X$. 
 We show by induction that it exists $M\in\R$ such that for any $x$ in $I$ and $w$ in $\calW_X$, one has
 \begin{equation*}
  |Ch(w)(x)| \leq \left(M|b-a|\right)^{|w|}.
 \end{equation*}
 This can easily be done by induction on the length of the word $w$. For a word of length $0$, it is trivially true. 
 Now, since each $f_i$ in $X$ is continuous over the closed interval $I$, $f_i$ is bounded by some constant $M_i$. Let 
 $M:=\max_{i=1,\cdots,n}M_i$. If $|Ch(w)(x)| \leq \left(M|b-a|\right)^{|w|}$ holds for every word $w$ of length $n\geq1$, then one has for any 
 $f_i$ in $X$
 \begin{equation*}
  |Ch((f_i)\sqcup w)(x)| \leq \int_a^x|f_i(t)Ch(w)(t)|dt \leq \left(M|b-a|\right)^{|w|+1}.
 \end{equation*}
 This shows that $|Ch(w)(x)| \leq \left(M|b-a|\right)^{|w|}$ for any $x$ in $I$ and $w$ in $\calW_X$. Continuity of $Ch(w)$ follows from the 
 standard theorems of integration theory.
\end{proof}
We denote by $Ch_X(I)$ the algebra (for the  {pointwise} product) freely generated by the image of $Ch$. It admits ${\bf 1}=Ch(\emptyset)$ as a unit.

This approach shows that Chen integrals are the elements of the image of a map $\widehat{\calI_a^b}_\calW$ going from $\calW_X$ to a subset of 
$\calI(I)$, where $\calI_a^b$ is the integration map defined by $\calI_a^b(f)(x):=\int_a^x f(t)dt$ for $x$ in $[a,b]$. Indeed, standard results of integration theory 
states that, since the interval $I$ is closed, the image of $\calI(I)$ under the map $\calI_a^b$ lies in $\calI(I)$. This suggest a natural 
generalisation of Chen integrals to arborified Chen integrals. Notice that this type of generalisation exists in the related context of rough paths, where one can build branched rough paths, see for example \cite{HK15}.
\begin{defn} \label{defn:arborified_chen_int}
 Let $I=[a,b]$ and $\calI(I)$ be as above. Let $\calI_a^b:\calI(I)\longrightarrow\calI(I)$ be 
 the integration map defined as above by $\calI_a^b(f)(x):=\int_a^x f(t)dt$. Let $X=\{f_1,\cdots,f_N\}$ is a finite subset of $\calI(I)$. 
 {\bf arborified Chen integrals} are elements of the image of $\widehat{\calI_a^b}:\calF_X\mapsto\calI(I)$.
\end{defn}
\begin{rk}
 One could recursively prove that $\widehat{\calI_a^b}$ is well-defined as in the case of words, that is that the image of $\widehat{\calI_a^b}$ 
 is indeed in $\calI(I)$. However by Definition \ref{defn:phi_hat} we already have this result since $\calI_a^b(\calI(I))\subseteq\calI(I)$.
\end{rk}
\begin{prop} \label{prop:chen_int_arbo_words}
 Any arborified Chen  {integral} is a finite sum of Chen iterated integrals with rational coefficients.
\end{prop} 
\begin{proof}
 This follows from the observation above that Chen iterated integrals are elements of the image of $\widehat{\calI_a^b}_\calW$. Then by Example 
 \ref{ex:RBmap} \ref{ex:RB_integration} we can apply Theorem \ref{thm:flattening} to $\widehat{\calI_a^b}$ which gives 
 $\widehat{\calI_a^b}=\widehat{\calI_a^b}_\calW\circ fl_0$. For any finite forest $F$, $fl_0(F)$ is a finite sum of words with rational coefficients by 
 Proposition \ref{prop:finite_Q_sum_flattening}, thus 
 we obtain the result.
\end{proof}

\subsection{Arborified polylogarithms} 

We specialise the construction of the previous subsection to the case $X=\{\sigma_x,\sigma_y\}$ with $\sigma_x(t):=1/t$ and 
$\sigma_y(t):=1/(1-t)$. Furthermore the above construction is carried out on $I=[\epsilon,z]$ with $0<\epsilon<z<1$.

However arborified polylogarithms (resp.  arborified shuffle zeta values) should be defined (up to convergence issues) as a map 
$Li^T:\calF_{\{x,y\}}\longrightarrow\Omega$ (resp. $\zeta^T_\shuffle:\calF_{\{x,y\}}\longrightarrow\R$) for  {a suitable space $\Omega$ of functions}. In order 
to build such a map, we follow the same strategy as for stuffle zeta values. Let $\calR_{\{x,y\}}:\{x,y\}\longrightarrow\{\sigma_x,\sigma_y\}$ defined 
by $\calR_{\{x,y\}}(\varepsilon)=\sigma_\varepsilon$ for $\varepsilon$ in $\{x,y\}$. 
\begin{defn}
 A forest $F$ in $\calF_{\{x,y\}}$ is called {\bf \ty{semi-convergent}} if each of its leaves and branching vertices are decorated by $y$ and {\bf convergent} if it is 
 \ty{semi-convergent} and each of its roots are decorated by $x$. The linear span of \ty{semi-convergent} forests is denoted by $\calF^{\rm semi}_{\{x,y\}}$ 
 and the linear span of convergent forests is denoted by $\calF^{\rm conv}_{\{x,y\}}$.
 
 Similarly, a word $w$ in $\calW_{\{x,y\}}$ is called {\bf \ty{semi-convergent}} if it is empty or if it ends by $y$ and {\bf convergent} if it is empty or if it ends by $y$ and starts by $x$. 
 We write $\calW^{\rm semi}_{\{x,y\}}$ and $\calW^{\rm conv}_{\{x,y\}}$ the linear span of \ty{semi-convergent} and convergent words respectively.
\end{defn}
This various notions are stable under natural product.
\begin{lemma}
 $\calF^{\rm semi}_{\{x,y\}}$ and $\calF^{\rm conv}_{\{x,y\}}$ are  {subalgebras} of $\calF_{\{x,y\}}$ for the concatenation of forests; $\calW^{\rm semi}_{\{x,y\}}$ and $\calW^{\rm conv}_{\{x,y\}}$
  are subalgebras of $\calW_{\{x,y\}}$ for the shuffle product $\shuffle$.
\end{lemma}
\begin{proof}
 As before, the result trivially holds for forests. For words the result follows from the fact that for any two words $w$ and $w'$, then we can write $w\shuffle w'=\sum_i w_i$ and the first (resp. last) letter 
 of each $w_i$ is the first (resp. last) letter of $w$ or $w'$.
\end{proof}
In order to state an important result of this Section, let us recall the definition of multiple polylogarithms.
\begin{defn} \label{defn:multiple_polylogs}
 Let $w\in\calW_{\{x,y\}}^{\rm semi}$ be a word either empty or whose last letter is $y$. The {\bf single variable multiple polylogarithm} (shortened in multiple polylogarithm in what follows) attached to $w$ is defined by
 \begin{align*}
   Li_w(z):=\begin{cases}
             & 0 \quad {\rm for } \quad z=0, \\
             & \lim_{\epsilon\to0}\widehat{\calI_\epsilon^z}_\calW\left(\calR_{\{x,y\}}^{\sharp,\calW}(w)\right) \quad {\rm for } \quad z\in]0,1[.
            \end{cases}
% \lim_{\epsilon\to0}\widehat{\cali_\epsilon^z}_\calw\left(\calr_{\{x,y\}}^{\sharp,\calw}(w)\right).
  \end{align*}
  We write $Li:\calW_{\{x,y\}}^{\rm semi}\longrightarrow\mathcal{C}^\infty([0,1[)$ the map which, to such a word, associates the map $z\mapsto Li_w(z)$.
\end{defn}
The existence of the limit for \ty{semi-convergent} words and the fact that $Li_w$ is a smooth map are well-known results of polylogarithms theory, see for example \cite{Br06}.
\begin{defiprop} \label{defnprop:arbo_polylogs}
 For any $z\in]0,1[$ and \ty{semi-convergent} $F$ the limit
 \begin{equation*}
  Li^T_F(z) := \lim_{\epsilon\to0}\widehat{\calI_\epsilon^z}\left(\calR_{\{x,y\}}^\sharp(F)\right)
 \end{equation*}
 exists. Setting $Li^T_F(0)=0$ we obtain a map 
 \begin{align*}
  Li^T_F:[0,1[ & \longrightarrow\R \\
           z & \longmapsto Li^T_F(z)
 \end{align*}
 \ty{which we call} the {\bf arborified polylogarithm} associated to the \ty{semi-convergent} forest $F$. The {\bf arborified polylogarithm map} is defined 
 by its action $Li^T:F\longrightarrow Li^T_F$ \ty{semi-convergent} forests.
\end{defiprop}
\begin{proof}
 One needs to prove the existence of the limit for any $z\in[0,1[$. It follows from Example \ref{ex:RBmap}, \ref{ex:RB_integration} that we can 
 apply Theorem \ref{thm:flattening} with $\lambda=0$. Furthermore, one easily shows by induction that the image of a \ty{semi-convergent} forest under 
 the map $fl_0$ is a finite sum of words, each ending with $y$. As stated above, it is a well-known fact (see for example \cite{Br06}) that for such a word $w$,
 the limit
 \begin{equation*}
  \lim_{\epsilon\to0}\widehat{\calI_\epsilon^z}_\calW\left(\calR_{\{x,y\}}^{\sharp,\calW}(w)\right)
 \end{equation*}
 exists. The result then follows.
\end{proof}
We have a simple but important compatibility Lemma whose proof is the same than the one of Lemma \ref{lem:flattening_conv}.
\begin{lemma}\label{lem:flattening_conv_xy}
 The flattening map of weight $0$ maps $\calF_{\{x,y\}}^{\rm semi}$ (resp. $\calF_{\{x,y\}}^{\rm conv}$) to $\calW_{\{x,y\}}^{\rm semi}$ (resp. $\calW_{\{x,y\}}^{\rm conv}$).
 Furthermore, $\calF_{\{x,y\}}^{\rm semi}$ and $\calF_{\{x,y\}}^{\rm conv}$ are subalgebras for the $0$-shuffle of rooted forests $\shuffle^T_0$.
\end{lemma}
We then have our main theorem on arborified polylogarithms:
\begin{theo} \label{thm:arborified_polylogs}
 For any \ty{semi-convergent} forest $F$, the arborified polylogarithm associated to $F$ enjoys the following properties:
 \begin{enumerate}
  \item It is a finite sum of multiple polylogarithms with rational coefficients that can be written as a finite 
 linear combination of multiple polylogarithms with integer coefficients;
  \item It is a smooth map on $[0,1[$;
  \item The arborified polylogarithm map $Li^T:F\longrightarrow Li^T_F$ is a algebra morphism for the concatenation of trees and the  {pointwise} product of functions.
  \item The arborified polylogarithm map $Li^T:F\longrightarrow Li^T_F$ is a algebra morphism for the shuffle product $\shuffle^T:=\shuffle^T_0$ of rooted forests and the  {pointwise} product of functions.
 \end{enumerate}
\end{theo}
 \begin{proof}
 \begin{enumerate}
  \item The proof of this result is the same as the proof of Theorem \ref{thm:main_result_stuffle}, using Example \ref{ex:RBmap}, \ref{ex:RB_integration}. The limits in the definition of $Li_w$ are well-defined 
  by Lemma \ref{lem:flattening_conv_xy}.
  \item The second point follows from the first, since multiple polylogarithms are smooth maps of $[0,1[$ and by the existance of the limit in the definition of $Li^T$.
  \item This follows from the fact that $Li^T$ is the composition of $\widehat{\calI_\epsilon^z}_\calW$ and $\calR_{\{x,y\}}^{\sharp,\calW}$, which are both algebra morphisms.
  Furthermore, if $FF'$ is a \ty{semi-convergent} forest, then $F$ and $F'$ are two \ty{semi-convergent} forests. Then $\lim_{\epsilon\to0}\widehat{\calI_\epsilon^z}\left(\calR_{\{x,y\}}^\sharp(F)\right)$
  and $\lim_{\epsilon\to0}\widehat{\calI_\epsilon^z}\left(\calR_{\{x,y\}}^\sharp(F')\right)$ exist and their product is equal to 
  \begin{equation*}
   \lim_{\epsilon\to0}\left[\widehat{\calI_\epsilon^z}\left(\calR_{\{x,y\}}^\sharp(F)\right)\widehat{\calI_\epsilon^z}\left(\calR_{\{x,y\}}^\sharp(F)\right)\right]
  \end{equation*}
  Thus the limit in the definition of $Li^T$ is also an algebra morphism and $Li^T$ is an algebra morphism as stated.
  \item Once again the proof of this result is as the proof of Theorem \ref{thm:main_result_stuffle}. It is an application of Theorem \ref{thm:flattening} \ref{thm:iv} together with Example \ref{ex:RBmap}. The result make sense by the stability property of Lemma \ref{lem:flattening_conv_xy}.
 \end{enumerate}

\end{proof}
As in the case of the stuffle branched zeta values, one can use this framework to provide a new proof that multiple polylogarithms are algebra morphisms for the shuffle product.
\begin{prop} \label{prop:polylogs_shuffle_mor}
 The map $Li:\calW_{\{x,y\}}^{\rm semi}\longrightarrow\mathcal{C}^{\infty}([0,1],\R)$ is an algebra morphism for the shuffle product.
\end{prop}
\begin{proof}
 The proof follows the same steps as the proof of Proposition \ref{prop:zeta_stuffle_mor}.
\end{proof}

\subsection{Arborified zeta values as integrals}

It is well-known (see for example \cite{Br06} or \cite{Waldschmidt}) that for a convergent word $w$ 
the limit $\lim_{z\to1}Li_w(z)$ exists. This allows the following definition, which clearly is another way to define shuffle MZVs.
\begin{defn} \label{defn:shuffle_MZVs}
 Let $w\in\calW_{\{x,y\}}$ be a word starting with $x$ and ending with $y$, then the {\bf shuffle multiple zeta value} associated to $w$ is the real number
 \begin{equation*}
  \zeta_\shuffle(w):=\lim_{z\to1}Li_w(z).
 \end{equation*}
 We write $\zeta_\shuffle$ the map defined by
 \begin{align*}
  \zeta_\shuffle:\{\emptyset\}\bigcup\left((x)\sqcup\calW_{\{x,y\}}\sqcup(y)\right) \longrightarrow \R & \\
  w \longmapsto \zeta_\shuffle(w & ).
 \end{align*}
\end{defn}
This definition can easily be generalised to trees, thanks to the following Lemma.
\begin{lemma} \label{lem:lim_z_1_forests}
 For any convergent forest $F\in\calF_{\{x,y\}}^{\rm conv}$, the limit $\lim_{z\to1}Li^T_F(z)$ exists.
\end{lemma}
\begin{proof}
 Let $F\in\calF_{\{x,y\}}^{\rm conv}$. If $F=\emptyset$, then $fl_0(F)=\emptyset$ and the result trivially holds. Otherwise $fl_0(F)\in(x)\sqcup\calW_{\{x,y\}}\sqcup(y)$. Indeed, we can write $fl_0(F)=\sum_{i\in I}w_i$ for some finite set $I$. Then each $w_i$ has the decoration of a root of $F$ (thus a $x$) 
 as its first letter; and the decoration of a leaf of $F$ (thus a $y$) as its last letter. This is shown ad absurdum: if we have $w_i=(y)\sqcup\tilde w$, then a vertex decorated by $y$ was not above all roots of $F$, since 
 every roots of $F$ is decorated by $x$. This is a contradiction. The same argument shows that no $w_i$ cannot end with an $x$. The result then follows from Theorem \ref{thm:arborified_polylogs} and the observation above that 
 $\lim_{z\to1}Li_w(z)$ exists for any 
 $w\in (x)\sqcup\calW_{\{x,y\}}\sqcup(y)$.
\end{proof}
This allows the following definition.
\begin{defn} \label{defn:shuffle_AZVs}
 For any convergent forest $F\in\calF_{\{x,y\}}^{\rm conv}$ the  {corresponding} {\bf shuffle arborified zeta value} is defined as
 \begin{equation*}
  \zeta_\shuffle^T(F) := \lim_{z\to1}Li_F^T(z).
 \end{equation*}
 We write $\zeta_\shuffle^T$ the map defined by
 \begin{align*}
  \zeta_\shuffle^T:\calF_{\{x,y\}}^{\rm conv} \longrightarrow & \R \\
  F \longmapsto \zeta_\shuffle^T&(F).
 \end{align*}
\end{defn}
Shuffle arborified zeta values enjoy the following properties.
\begin{theo} \label{thm:main_result_shuffle}
 For any convergent forest $F\in\calF_{\{x,y\}}^{\rm conv}$,  {the shuffle arborified zeta values} $\zeta_\shuffle^T(F)$ is a finite sum of  {multiple zeta} values with rational coefficients that can be written as a finite 
 sum of  {multiple zeta} values with integer coefficients given by the $0$-flattening:
 \begin{equation*}
  \zeta^T_\shuffle=\zeta_\shuffle\circ fl_0.
 \end{equation*}
Furthermore the map $\zeta_\shuffle^T:\calF_{\{x,y\}}^{\rm conv} \longrightarrow \R$
 is an algebra morphism for the concatenation product of trees as well as for the shuffle product $\shuffle^T$ of rooted forests.
\end{theo}
\begin{proof}
 This Theorem is a consequence of Theorem \ref{thm:arborified_polylogs} applied to convergent forests, for which one can take the limit $z\to1$ according to Lemma \ref{lem:lim_z_1_forests}.
\end{proof}
\begin{rk}
 One could prove this theorem along the steps of the proof of Theorem \ref{thm:main_result_stuffle}. This is true in many occurrences throughout this section. Hence Sections \ref{section:stuffle} and 
 \ref{section:shuffle} present two different (however equivalent) ways of building branched objects.
\end{rk}
We conclude this section by pointing out that, as a consequence of our previous results, we also have shown that  {multiple zetas} as iterated integrals are algebra morphism for the shuffle product.
\begin{prop} \label{prop:shuffle_zeta_alg_mor_shuffle}
 The map $\zeta_\shuffle:\{\emptyset\}\bigcup\left((x)\sqcup\calW_{\{x,y\}}\sqcup(y)\right) \longrightarrow \R$ is an algebra morphism for the shuffle product.
\end{prop}
\begin{proof}
 The proof follows from Proposition \ref{prop:polylogs_shuffle_mor} on words in $\{\emptyset\}\bigcup\left((x)\sqcup\calW_{\{x,y\}}\sqcup(y)\right)$ and taking the limit $z\to1$.
\end{proof}

\section{Attempting to relate AZVs} \label{section:Hoffman}

Now that we have two versions of AZVs\cy{: $\zeta_\stuffle^T$ and $\zeta_\shuffle^T$. The former generalises $\zeta_\stuffle$ since it is defined on rooted trees and coincides with MZVs on linear trees which are the image of words under the canonical embedding of words into trees. Similarly, $\zeta^T_\shuffle$ generalises $\zeta_\shuffle$. Therefore}  
%that each generalise one version of MZVs, 
we would like to relate \ty{$\zeta_\stuffle^T$ and $\zeta_\shuffle^T$} by a generalisation of Kontsevich's relation \eqref{eq:Kontsevich}.

\subsection{Branched binarisation map}

We aim to generalise the map \eqref{eq:binarisation_map} to trees, that is to build a map
$\fraks^T :\calF_{\N^*}\longrightarrow\calF_{\{x,y\}}$ which coincide with $\iota_{\{x,y\}}\circ\fraks$ when restricted to ladder trees.
In order to do so, once again we use  the universal property of trees given by Theorem \ref{thm:univ_prop_tree}.
\begin{defn} \label{defn:branched_bin_map}
 Let $\beta:\N^*\times\calF_{\{x,y\}}\longrightarrow\calF_{\{x,y\}}$ defined by
 \begin{align*}
  \beta(1,F) := B_+^y(F) \\
  \beta(n,F) := \left(B_+^x\right)^{n-1}\left(B_+^y(F)\right)
 \end{align*}
 for any $n\geq2$. The {\bf branched binarisation map} is the morphism of operated algebras 
 $\fraks^T :\calF_{\N^*}\longrightarrow\calF_{\{x,y\}}$ whose existence and uniqueness is given by Theorem \ref{thm:univ_prop_tree}.
\end{defn}
\begin{example} Here are some example of the action of the binarisation map:
 \begin{equation*}
  \fraks^T(\tdun{1}) = \tdun{y} \qquad  \fraks^T(\tdun{2}) = \tddeux{$x$}{$y$} \qquad \fraks^T\left(\tdtroisun{1}{1}{2}\right) = \tdquatredeux{$y$}{$y$}{$x$}{$y$} \qquad
  \fraks^T\left(\tdtroisun{$2$}{$1$}{$2$}\right) = \tcinqonze{$x$}{$y$}{$x$}{$y$}{$y$}.
 \end{equation*}
\end{example}
Now we state a simple lemma relating convergent forests in $\calF_{\N^*}$ and in $\calF_{\{x,y\}}$.
\begin{lemma} \label{lem:bin_map_properties}
The branched binarisation map maps convergent forests to convergent forests:
 \begin{equation*}
  \fraks^T\left(\calF_{\N^*}^{\rm conv}\right) = \calF_{\{x,y\}}^{\rm conv}.
 \end{equation*}
 Furthermore, $\fraks^T$ is a bijection.
\end{lemma}
\begin{proof}
 By definition of the operation $\beta$, if $F\in\calF_{\{x,y\}}$ is in the image of $\beta$, then $F$ is \ty{semi-convergent}. Therefore $\fraks^T\left(\calF_{\N^*}\right) \subseteq \calF_{\{x,y\}}^{\rm semi} \subseteq \calF_{\{x,y\}}$.
 Thus we only need to prove that the image of a convergent forest has its roots decorated by $x$s only.
 
 Let $T\in\calF_{\N^*}^{\rm conv}$ be a convergent tree. If $T=\emptyset$ then $\fraks^T(\emptyset)=\emptyset\in\calF_{\{x,y\}}^{\rm conv}$ by definition of $\calF_{\{x,y\}}^{\rm conv}$. 
 
 If $T\neq \emptyset$ then it exists a forest $F$ such that $T=B_+^p(F)$ with $p\geq2$. Then by definition of $\fraks^T$ we have
 \begin{equation*}
  \fraks^T(T) = \left(B_+^x\right)^{p-1}\left(B_+^y(\fraks^T(F))\right)
 \end{equation*}
 which lies in $\calF_{\{x,y\}}^{\rm conv}$ since $p-1\geq1$.
 
 Let $F\in\calF_{\N^*}^{\rm conv}$ be a convergent forest. Then we have $F=T_1\cdots T_k$ with $T_i\in\calF_{\N^*}^{\rm conv}$ by definition 
 of $\calF_{\N^*}^{\rm conv}$. Then 
 \begin{equation*}
  \fraks^T(F) = \fraks^T(T_1)\cdots\fraks^T(T_1)\in\calF_{\{x,y\}}^{\rm conv}
 \end{equation*}
 by definition of $\calF_{\{x,y\}}^{\rm conv}$. Therefore $\fraks^T\left(\calF_{\N^*}^{\rm conv}\right) \subseteq \calF_{\{x,y\}}^{\rm conv}$.
 
 The bijectivity of $\fraks^T$ is also shown by induction, using $|\fraks^T(F)| = ||F||$. The same argument on $(\fraks^T)^{-1}$ allows to show that 
 $(\fraks^T)^{-1}\left(\calF_{\{x,y\}}^{\rm conv}\right) \subseteq \calF_{\N^*}^{\rm conv}$; concluding the proof.
\end{proof}
Recall that a branching vertex is a vertex that has strictly more than one direct successor. This concept will be 
of importance when trying to relate shuffle and stuffle arborified zeta values through the branched binarisation map.

We now have a result which is a direct consequence of a result of the next Section (see Remark \ref{rk:AZV_TZV}). We give here only the proof for the ladder forests. We nonetheless state it here since it gives a negative answer to the question we asked.
\begin{theo} \label{thm:relation_shuffle_stuffle}
 For any convergent forest $F\in\calF_{\N^*}^{\rm conv}$ we have
 \begin{equation*}
  \zeta^T_{\shuffle}(\fraks^T(F)) \leq \zeta^T_{\stuffle}(F).
 \end{equation*}
 Furthermore, the inequality is an equality if, and only if, $F$ has no branching vertex (i.e. $F$ is the empty tree or $F=l_1\cdots l_k$ with 
 $l_i$ being ladder trees).
\end{theo}
\noindent\emph{Partial proof.}
 \begin{itemize}
  \item If $F=\emptyset$, then $\fraks^T(F)=\emptyset$ and the result holds by construction.
  \item If $F=l_1\cdots l_k$ with $l_i$ being ladder trees, then the result follows from the classical property \eqref{eq:Kontsevich} of the binarisation map together with the 
  facts that $\zeta^T_\shuffle$ and $\zeta^T_\stuffle$ are algebra morphisms for the concatenation product of trees, $\zeta_\stuffle(w) =  \zeta^T_\stuffle(\iota_{\N^*}(w))$ and 
  $\zeta_\shuffle(w) =  \zeta^T_\shuffle(\iota_{\{x,y\}}(w))$.
 \end{itemize}

 \subsection{Hoffman's relations for branched zetas}

In order to find the branched equivalent of Hoffman's regularisation relations, it is useful to recall that the concatenation product of trees can be seen as lift to trees of both the shuffle and stuffle products. 
The latter only differ in the way one goes back to words, that is to say by a choice of the 
flattening map. Therefore the most naive candidate for this relations is that $\fraks^T(\tdun{1} F) - \fraks^T(\tdun{1})\fraks^T(F)$ is convergent 
for any convergent forest $F$ and lies in the kernel of $\zeta_\shuffle^T$.

However, this statement is trivially true and does not impose any new relation on the algebra spanned by branched zetas. Indeed, since $\fraks^T$ is an 
algebra morphism we have $\fraks^T(\tdun{1} F) - \fraks^T(\tdun{1})\fraks^T(F)=0$\footnote{notice that $0$ and $\emptyset$ are distinct elements of our algebra}. 
This observation is a consequence of the fact that the distinction 
between the stuffle and shuffle products can only be made at the level of words, reached through flattening maps. Consequently, a more relevant quantity 
to study is 
\begin{equation*}
 \fraks\left(fl_1\left(\tdun{1} F\right)\right) - fl_0\left(\fraks^T\left(\tdun{1} F\right)\right)
\end{equation*}
for convergent forests $F$.  However, the following result states that this does not generalise Hoffman's relations.
 \begin{prop} \label{prop:branched_Hoffman}
 For any convergent forest $F\in\calF_{\N^*}^{\rm conv}$, 
 \begin{equation*}
  \fraks\left(fl_1\left(\tdun{1} F\right)\right) - fl_0\left(\fraks^T\left(\tdun{1} F\right)\right)
 \end{equation*}
 lies in the algebra of convergent words if, and only if, $F$ is empty or $F=l_1\cdots l_k$ with the $l_i$ ladder trees. In this case 
 \begin{equation*}
  \fraks\left(fl_1\left(\tdun{1} l_1\cdots l_k\right)\right) - fl_0\left(\fraks^T\left(\tdun{1} l_1\cdots l_k\right)\right) \in \Ker(\zeta_\shuffle).
 \end{equation*}
\end{prop}
 \begin{proof}
 \begin{itemize}
  \item If $F$ is empty, then $\fraks\left(fl_1\left(\tdun{1} F\right)\right) - fl_0\left(\fraks^T\left(\tdun{1} F\right)\right)=0$ and the 
  result trivially holds.
  \item If $F=l_1\cdots l_k$ with the $l_i$ ladder trees then, by definition of the flattening maps we have
  \begin{equation*}
  \fraks\left(fl_1\left(\tdun{1} F\right)\right) - fl_0\left(\fraks^T\left(\tdun{1} F\right)\right) = \fraks\left((y)\stuffle w_1\stuffle\cdots\stuffle w_k\right) - (y)\shuffle\fraks(w_1)\shuffle\cdots\shuffle\fraks(w_k)
 \end{equation*}
 with $w_i:=\iota_{\N^*}^{-1}(l_i)$. Then the result holds by Hoffman's regularisation relations \eqref{eq:Hoffman_reg_rel}.
 \item Let $T$ be a non ladder tree, and $v$ a branching vertex of $T$ with decoration $p_1$. Let $v'$ and $v''$ be two direct successors of $V$ with decorations $p_2$ and $p_3$ respectively. Then it exists two 
 eventually empty words $w$ and $w'$ such that
 \begin{equation*}
  fl_1(\tdun{1} T) = (1)\sqcup w\sqcup(p_1[p_2+p_3])\sqcup w' + X
 \end{equation*}
 with $X$ a finite linear combination of words written in the alphabet $\N^*$.
 
 By definition of $fl_0$, $\fraks\left((1)\sqcup w\sqcup(p_1[p_2+p_3])\sqcup w'\right)$ will not show up in $fl_0\left(\fraks^T\left(\tdun{1} F\right)\right)$. Thus we obtain
 \begin{equation*}
  \fraks\left(fl_1\left(\tdun{1} F\right)\right) - fl_0\left(\fraks^T\left(\tdun{1} F\right)\right) = \fraks\left((1)\sqcup w\sqcup(p_1[p_2+p_3])\sqcup w'\right) + Y
 \end{equation*}
 with $Y$ a finite linear combination of words written in the alphabet $\{x,y\}$ without the divergent tree 
 $\fraks\left((1)\sqcup w\sqcup(p_1[p_2+p_3])\sqcup w'\right)$. This implies the result for any non-ladder tree.
 
 For a forest with at least one branching point, the result follows from the fact that $\fraks^T$ and the flattening maps are algebras morphisms and the previous discussion on non ladder trees.
 \end{itemize}
\end{proof}
 This result was derived from the picture that the shuffle and stuffle products of words are lifted to trees to the concatenation of trees. We 
have also derived a generalisation to rooted forests of the shuffle and stuffle products of words. This leads us to an alternative  possible 
generalisation of Hoffman's relations  \eqref{eq:Hoffman_reg_rel}; namely that
 \begin{equation*}
 \tdun{1}\stuffle ~T - (\fraks^T)^{-1}\left(\tdun{y}\shuffle ~\fraks^T(T)\right) \quad \text{and} \quad \fraks^T\left(\tdun{1}\stuffle ~T\right) - \tdun{y}\shuffle ~\fraks^T(T)
\end{equation*}
lie in $\calF_{\N^*}^{\rm conv}$
for any convergent tree $T\in\calF_{\N^*}^{\rm conv}$. About this quantity one finds
\begin{prop} \label{prop:conv_Hoffman_trees}
 For any convergent forest $F\in\calF_{\N^*}^{\rm conv}$, one has
 \begin{align*}
  & \tdun{1}\stuffle ~F - (\fraks^T)^{-1}\left(\tdun{y}\shuffle ~\fraks^T(F)\right) \in \calF_{\N^*}^{\rm conv}; \\
  & \fraks^T\left(\tdun{1}\stuffle ~F\right) - \tdun{y}\shuffle ~\fraks^T(F) \in \calF_{\{x,y\}}^{\rm conv}.
 \end{align*}
\end{prop}
\begin{proof}
 First, notice that by Lemma \ref{lem:bin_map_properties} the two statements are equivalent. We therefore only prove the second one.
 
 For a convergent forest $F=T_1\cdots T_k$ , we have $\fraks^T(F) = \fraks^T(T_1)\cdots\fraks^T(T_k)$. Then we have, with obvious 
 notations:
 \begin{equation*}
  \tdun{1}\stuffle~F = \frac{1}{k}\sum_{i=1}^k\left(B_+^1(T_i) + X_{i,F}\right)F\setminus T_i
 \end{equation*}
 for some finite sum of trees $X_{i,F}\in\calF_{\N^*}{\rm conv}$. Therefore we have
 \begin{equation*}
  \fraks^T\left(\tdun{1}\stuffle~F\right) = \frac{1}{k}\sum_{i=1}^k\left(B_+^y(\fraks^T(T_i)) + \fraks^T(X_{i,F})\right)\fraks^T(F\setminus T_i).
 \end{equation*}
 On the other hand we have
 \begin{equation*}
  \tdun{y}\shuffle~\fraks^T(F) = \frac{1}{k}\sum_{i=1}^k\left(B_+^y(\fraks^T(T_i)) + Y_{i,F}\right)\fraks^T(F\setminus T_i)
 \end{equation*}
 for some $Y_{i,F}\in\calF_{\{x,y\}}^{\rm conv}$. Taking the difference of these two quantities, one obtains the result, since by 
 Lemma \ref{lem:bin_map_properties} the branched binarisation map $\fraks^T$ maps convergent forests to convergent forests.
\end{proof}
While this result might give us hope, one should expect Theorem \ref{thm:relation_shuffle_stuffle} to prevent the quantities 
$\fraks^T\left(\tdun{1}\stuffle ~F\right) - \tdun{y}\shuffle ~\fraks^T(F)$ from lying in the kernel of $\zeta^T_\shuffle$. And indeed, one 
finds, after a long yet straightforward computation, that
\begin{equation} \label{eq:counter_example_Hoffman}
 \zeta_\shuffle^T\left(\fraks^T\left(\tdun{1}\stuffle\tdtroisun{2}{1}{1}\right) - \tdun{$y$}\shuffle\tdquatrequatre{$x$}{$y$}{$y$}{$y$}\right) = \zeta(2,3)+\zeta(3,2) > 0.
\end{equation}
The precise characterisation of quantities of the form $\fraks^T\left(\tdun{1}\stuffle ~F\right) - \tdun{y}\shuffle ~\fraks^T(F)$ as well as their image under $\zeta_\shuffle^T$ could be of interest for future investigations. 

To sum up our work so far, let us compare the properties of MZVs and AZVs in Table \ref{table:MZV_AZV} below.
%  \begin{center}
% \begin{tabular}{|c||c|c|}
%  \hline
%  MZVs & ~AZVs~ & AZVs/TZVs \\
%  \hline %\vspace{4mm}
%  $\zeta_{\shuffle}$ & $\zeta^T_{\shuffle}$ & $\zeta^T_{\shuffle}$ \\
%  \hline 
%  $\zeta_{\stuffle}$ & $\zeta^T_{\stuffle}$ & $\zeta^t$ \\
%  \hline
%  Relation de Kontsevich & Non & $\zeta^T_{\shuffle}=\zeta^t\circ\fraks^T$ \\
%  \hline
%  $\shuffle$ & $\shuffle^T$ & $\shuffle^T$ \\
%  \hline
%  $\stuffle$ & $\stuffle^T$ & ? \\
%  \hline 
%  Régularisation de Hoffman & Non & ? \\
%  \hline
% \end{tabular}
% \end{center}
\begin{table}[h!] \label{table:MZV_AZV}
  \begin{center}
\begin{tabular}{|c||c|}
 \hline
 MZVs & ~AZVs~  \\
 \hline %\vspace{4mm}
 $\zeta_{\shuffle}$ & $\zeta^T_{\shuffle}$ \\
 \hline 
 $\zeta_{\stuffle}$ & $\zeta^T_{\stuffle}$ \\
 \hline
 Kontsevich's relation & No \\
 \hline
 $\shuffle$ & $\shuffle^T$  \\
 \hline
 $\stuffle$ & $\stuffle^T$ \\
 \hline 
 Hoffman's regularisation relation & No  \\
 \hline
\end{tabular}
\caption{Properties of MZVs and AZVs.}
\end{center}
\end{table}

The conclusion of this Section could be that branching vertices, at least with 
the current generalisation of MZVs, induce an important change when lifting their properties to rooted forests. 

\ty{Another possible point of view} is that some of the properties of MZVs are lost when one generalises them to rooted forests and that our initial goal is out of reach. In front of this observation, we have a choice: to despair or not to despair. As a young naive person I chose the latter. Instead, we will now see that shuffle AVZs \emph{can} be written as multiple sums generalising the MZVs. This \ty{motivates} the introduction of a new generalisation of MZVs to rooted forests.
 
\section{Series representation of AZVs} \label{section:int_to_series}

In order to prove the existence of a series representation for AZVs, we need to introduce a few more notions. Before, let us recall from Definition 
\ref{def:graph_stuff} that \ty{a {\bf path}} in a rooted forest $F$ is a finite sequence of elements of $V(F)$: $p=(v_1,\cdots,v_n)$ that is either a singleton (in which case $p=(v)$ is said to be a path between $v$ and itself) or such that for all $i\in[[1,n-1]]$, $(v_i,v_{i+1})$ is an edge of $F$. The same notion generalises to decorated rooted forests. 
\begin{defn} \label{defn:segment}
 A {\bf segment} of a rooted forest $(F,d_F)\in\calF_{\{x,y\}}$ is a non-empty path $s_v=(v_1,\cdots,v_n=v)$ such that $d_F(v_n)=y$, $d_F(v_i)=x$ for any $i$ in $\{1,\cdots,n-1\}$ and $d(a(v_1))=y$, with $a(v_1)$ the direct ancestor of $v_1$. We call the number $n$ the {\bf length} of the segment $s_v$ and write it $|s_v|$.
 
 We write $S(F):=\{s_v|v\in V_y(F)\}$ the set of segments of a rooted forest $F$.
\end{defn}
In words: a segment $s_v$ of a rooted forest is a path in this rooted forest, which ends at the vertex $v$ decorated by $y$ and starts just above the first ancestor of $v$ being also decorated by $y$.

The set $S(F)$ inherits a poset structure from the poset structure of $V(F)$. We also denote this partial order relation by $\preceq$: $s_v\preceq s_{v'}:\Longleftrightarrow v\preceq v'$. This allows us to define the depth of a segment.
\begin{defn}
 The {\bf depths} of the segments of a decorated rooted forest $F$ decorated by $\{x,y\}$ are recursively defined by:
 \begin{itemize}
  \item depth$(s_v)=0$ if $v$ is a leaf of $F$,
  \item depth$(s_v)=\max\{{\rm depth}(s_{v'})|v'\neq v\wedge s_v\preceq s_{v'}\}+1$.
 \end{itemize}
 We also set $N_F:=\max\{{\rm depth}(s_v)|s_v\in S(F)\}$ the maximal depth of a segment of $F$. For any $n\in\{0,\cdots,N_F\}$ we set 
 \begin{equation*}
  S_n(F) := \{s_v\in S(F)|{\rm depth}(s_v)=n\},
 \end{equation*}
  we further set:
 \begin{equation*}
  S^n(F):=\bigcup_{i=0}^n S_i(F)
 \end{equation*}
 for any $n$ in $\{0,\cdots,N_F\}$. For any such $n$ we also write
 \begin{equation*}
  ||S_n(F)||:=\sum_{s_v\in S_n(F)} |s_v| \quad\text{and}\quad ||S^n(F)||:=\sum_{s_v\in S^n(F)} |s_v|.
 \end{equation*}

\end{defn}
Notice that depth of segments of a rooted forest are well-defined because we have taken our forests to be finite.

We are now able to prove the following
\begin{theo} \label{thm:integral_sum}
 For any convergent forest $F$ the corresponding arborified zeta values admits the following series representation:
 \begin{equation*}
  \zeta_\shuffle^T(F) = \sum_{\substack{ n_v\geq1\\v\in V_y(F)}}\prod_{v\in V_y(F)}\left(\sum_{\substack{v'\in V_y(F)\\ v'\succeq v}}n_{v'}\right)^{-|s_v|}
 \end{equation*}
 where $s_v=(v_1,\cdots,v_n=v)$ is a the segment of $F$ ending at $v$ and $v\in V_y(F)$ in the first sum means that this series has a summation variable for each $v\in V_y(F)$.
\end{theo}
\begin{proof}
 Let $F$ be a convergent forest. Since the map $F\mapsto\zeta_\shuffle^T(F)$ is an algebra morphism for the concatenation product of trees, it is enough to show that the theorem holds for $F$ a rooted tree. Thus we can assume without loss of generality that $F$ is a rooted tree.
 
 Let $N_F$ be the maximal depth of the segments of this tree. If $N_F=0$, the theorem reduces to the usual series representation of a Riemann zeta. 
 
 If $N_F\geq1$, since we are working with convergent integrals we can use Fubini's theorem to regroup integrations of the segment. This mean that we can write the arborified zeta associated to $F$ as
 \begin{equation} \label{eq:step_zero}
  \zeta_\shuffle^T(F) = \int_{\Delta_F}\prod_{n=0}^{N_F}\prod_{s_v\in S_n(F)} d\omega_{s_v}
 \end{equation}
 where, for $s_v=(w_1,\cdots,w_{p})\in S_n(F)$ we have set 
 \begin{equation*}
  d\omega_{s_v}:=\frac{dz_{w_1}}{z_{w_1}}\cdots\frac{dz_{w_{p-1}}}{z_{w_{p-1}}}\frac{dz_{w_p}}{1-z_{w_p}}
 \end{equation*}
 (with an obvious abuse of notation if $p=1$). Now recall that for any $A\in\R_+$, $p\geq1$ and any $Z\in[0,1]$ ($Z\neq 1$ if $p=1$) we have
 \begin{equation} \label{eq:lemma_trad}
  \int_{0\leq z_p\leq\cdots\leq z_1\leq Z}\frac{dz_1}{z_1}\cdots\frac{dz_{p-1}}{z_{p-1}}\frac{dz_p}{1-z_p} (z_p)^A = \sum_{n=1}^{+\infty} \frac{Z^{n+A}}{(n+A)^p}
 \end{equation}
 (again with obvious abuses of notations when $p=1$). This classical result follows from the theorem of dominated convergence and the Taylor expansion of the function $x\mapsto (1-x)^{-1}$.
 
 We can now use \eqref{eq:lemma_trad} with $A=0$ in \eqref{eq:step_zero} to integrate all the variable attached to vertices belonging to a segment of depth zero. We obtain 
 \begin{equation} \label{eq:step:one}
  \zeta_\shuffle^T(F) = \sum_{\substack{n_v=1, \\ v\in V_y(F)| s_v\in S_0(F)}}^{+\infty} \prod_{v\in V_y(F)|s_v\in S_0(F)}(n_v)^{-|s_v|}\int_{\Delta_F\setminus S_0(F)}\prod_{n=1}^{N_F}\prod_{s_v\in S_n(F)} d\omega_{s_v}\prod_{v\in V_y(F)|s_v\in S_1(F)}(z_v)^{\sum_{\substack{v'\in V_y(F)\\v'\succ v}}n_{v'}}
 \end{equation}
 where $v'\succ v$ means that $v'$ is a descendant of $v$ that is distinct from $v$ 
 and where we have set 
 \begin{align*}
  & [0,1]^{|V(F)|-||S_0(F)||}\ni(z_{v_1},\cdots,z_{v_p})\in\Delta_F\setminus S_0(F) \\
  :\Longleftrightarrow & \left(\{v_1,\cdots,v_p\}=V(F)\setminus\{v'\in s_v|s_v \in S_0(F)\}~\wedge~(v_i\preceq v_j \Leftrightarrow z_{v_j}\leq z_{v_i})\right).
 \end{align*}
 Now for any $k\in\{0,\cdots,N_F-1\}$ we set 
 \begin{align*}
  & [0,1]^{|V(F)|-||S^k(F)||}\ni(z_{v_1},\cdots,z_{v_p})\in\Delta_F\setminus S^k(F) \\
  :\Longleftrightarrow & \left(\{v_1,\cdots,v_p\}=V(F)\setminus \{v'\in s_v|s_v \in S^k(F)\}~\wedge~(v_i\preceq v_j \Leftrightarrow z_{v_j}\leq z_{v_i})\right)
 \end{align*}
 (notice that we replaced $S_0(F)$ by $S^k(F)$). Let us prove by induction over $k$ that for any $k\in \{0,\cdots,N_F-1\}$ we have
 \begin{align} \label{eq:step_rec}
  \zeta_\shuffle^T(F) & = \sum_{\substack{n_v=1, \\ v\in V_y(F)| s_v\in S^k(F)}}^{+\infty} \prod_{v\in V_y(F)|s_v\in S^k(F)}\left(\sum_{\substack{v'\in S^k(F) \\ v'\succeq v}}n_{v'}\right)^{-|s_v|} \\ 
   & \times \underbrace{\int_{\Delta_F\setminus S^k(F)}\prod_{n=k+1}^{N_F}\prod_{s_v\in S_n(F)} d\omega_{s_v}\prod_{v\in V_y(F)|s_v\in S_{k+1}(F)}(z_v)^{\sum_{\substack{v'\in V_y(F) \\v'\succ v}}n_{v'}}}_{=:I_{F,k}}.  \nonumber
 \end{align}
 First, since $S_0(F)=S^0(F)$ and $\Delta_F\setminus S_0(F)=\Delta_F\setminus S^0(F)$, Equation \eqref{eq:step:one} is exactly Equation \eqref{eq:step_rec} for $k=0$. Then if $N_F=1$, we have proven Equation \eqref{eq:step_rec} in every case of interest. If $N_F\geq2$ then let us assume that Equation \eqref{eq:step_rec} holds for $k\in\{0,\cdots,N_F-2\}$. We then have, once again from the Taylor expansion of the function $x\mapsto(1-x)^{-1}$ and the dominated convergence theorem
 \begin{align*}
  I_{F,k} & = \int_{\Delta_F\setminus S^k(F)}\left(\prod_{n=k+2}^{N_F}\prod_{s_v\in S_n(F)} d\omega_{s_v}\right)\left(\prod_{s_v\in S_{k+1}(F)} d\omega_{s_v}\right)\prod_{v\in V_y(F)|s_v\in S_{k+1}(F)}(z_v)^{\sum_{\substack{v'\in V_y(F) \\v'\succ v}}n_{v'}} \\
  & = \sum_{\substack{n_v=0, \\ v\in V_y(F)| s_v\in S_{k+1}(F)}}^{+\infty} \int_{\Delta_F\setminus S^k(F)}\left(\prod_{n=k+2}^{N_F}\prod_{s_v\in S_n(F)} d\omega_{s_v}\right)  \\
  \times & \prod_{\substack{v\in V_y(F)| \\ s_v=(\vec{\alpha},v)\in S_{k+1}(F)}} \frac{d\vec{z}_{\vec{\alpha}}}{\vec{z}_{\vec\alpha}}dz_v(z_v)^{n_v} \prod_{v\in V_y(F)|s_v\in S_{k+1}(F)}(z_v)^{\sum_{\substack{v'\in V_y(F) \\v'\succ v}}n_{v'}}
 \end{align*}
 with the obvious notation that $\frac{d\vec{z}_{\vec{\alpha}}}{\vec{z}_{\vec\alpha}}$ is a product of $dz/z$. We can now merge the two last products to obtain
 \begin{equation*}
  I_{F,k} = \sum_{\substack{n_v=0, \\ v\in V_y(F)| s_v\in S_{k+1}(F)}}^{+\infty} \int_{\Delta_F\setminus S^k(F)}\left(\prod_{n=k+2}^{N_F}\prod_{s_v\in S_n(F)} d\omega_{s_v}\right) 
  \prod_{\substack{v\in V_y(F)| \\ s_v=(\vec{\alpha},v)\in S_{k+1}(F)}} \frac{d\vec{z}_{\vec{\alpha}}}{\vec{z}_{\vec\alpha}}dz_v(z_v)^{\sum_{\substack{v'\in V_y(F) \\v'\succeq v}}n_{v'}}.
 \end{equation*}
 Finally integrating the variable attached to vertices belonging to segments of depth $k$ and switching each of the summation variable by one we obtain
 \begin{equation*}
  I_{F,k} = \sum_{\substack{n_v=1, \\ v\in V_y(F)| s_v\in S_{k+1}(F)}}^{+\infty} \prod_{v\in V_y(F)|s_v\in S_{k+1}(F)}\left(\sum_{\substack{v'\in V_y(F) \\ v'\succeq v}}n_{v'}\right)^{-|s_v|}
  \underbrace{\int_{\Delta_F\setminus S^{k+1}(F)}\left(\prod_{n=k+2}^{N_F}\prod_{s_v\in S_n(F)} d\omega_{s_v}\right)}_{I_{F,k+1}}.
 \end{equation*}
 Notice that this whole computation consisted essentially of using Formula \eqref{eq:lemma_trad} for each of the segments of depth exactly $k$. In any case, plugging this expression for $I_{F,k}$ back into \eqref{eq:step_rec} we obtain exactly the same equation with $k$ replaced by $k+1$. So, by a finite induction we have proven Equation \eqref{eq:step_rec} for any $k\in\{0,\cdots,N_F\}$ for any value of $N_F\geq1$.
 
 Since we have assumed $F$ to be a connected rooted forest (i.e. a rooted tree), $F$ has exactly one segment $s_v$ of maximal depth $N_F$. Furthermore, since $F$ is convergent, we have $l:=|s_v|\geq2$. Thus, after a relabelling of $s_v$:
 \begin{equation*}
  s_v=(1,\cdots,l=v),
 \end{equation*}
 Equation \eqref{eq:step_rec} with $k=N_F-1$ reads
 \begin{align*}
  \zeta_\shuffle^T(F) & = \sum_{\substack{n_v=1, \\ v\in V_y(F)| s_v\in S^{N_F-1}(F)}}^{+\infty} \prod_{v\in s_v|s_v\in S^{N_F-1}(F)}\left(\sum_{\substack{v'\in S^{N_F-1}(F) \\ v'\succeq v}}n_{v'}\right)^{-|s_v|} \\ 
   & \times \int_{0\leq z_1<\cdots<z_l\leq1} \frac{dz_1}{z_1}\cdots\frac{dz_{l-1}}{z_{l-1}}\frac{dz_l}{1-z_l}\left(z_l\right)^{\sum_{\substack{v'\in V_y(F) \\ v'\neq v}}n_{v'}}
 \end{align*}
 Using once again Formula \eqref{eq:lemma_trad} with $Z=1$, $p=l\geq2$ and $A=\sum_{\substack{v'\in V_y(F) \\ v'\neq v}}n_{v'}$ we obtain the statement of the theorem once we write all the sums together.
\end{proof}

\begin{rk} \label{rk:AZV_TZV}
 This series representation of AZVs is \emph{not} the stuffle AZVs defined earlier.
 %Theorem A.4 of \cite{Cl20} implies that these two series applied to the same forest give in general different values. 
 Instead, this series representation of AZVs defined by iterated integrals should be seen as a new generalisation of MZVs defined as multiple series. 
 
 Notice however that this result implies Theorem \ref{thm:relation_shuffle_stuffle}. Indeed the series expression of $F=\fraks^T(f)$ given by Theorem \ref{thm:integral_sum} contains less terms that $\zeta^T_\stuffle(f)$, and strictly less if $f$ has at least one branching vertices. To see this, let us take $f=\tdtroisun{a}{c}{b}$. Then
 \begin{equation*}
  \zeta^T_\stuffle(f)=\sum_{\substack{n_1,n_2,n_3>0 \\ n_1>n_2,n_3}}\frac{1}{n_1^an_2^bn_3^b}, \qquad \zeta^T_\shuffle(\fraks(f))=\sum_{m_1,m_2,m_3>0}\frac{1}{(m_1+m_2+m_3)^am_2^bm_3^b}.
 \end{equation*}
 If we change the varibles in $\zeta^T_\shuffle(\fraks(f))$ to $n_3=m_3$, $n_2=m_2$ and $n_1=m_1+m_2+m_3$, the summand is the same that the one of $\zeta^T_\stuffle(f)$, but the condition is $n_1>n_2+n_3$. For example, terms like $(n_1,n_2,n_3)=(3,2,1)$ appears in $\zeta^T_\stuffle(f)$ but not in $\zeta^T_\shuffle(\fraks(f))$. This is true for any rooted forest with at least one branching vertex: we have precisely Theorem \ref{thm:relation_shuffle_stuffle}.

\end{rk}

\section{Tree zeta values} \label{section:TZVs}

Theorem \ref{thm:integral_sum} motivates a new generalisation of MZVs to trees. This section is devoted to studying this generalisation.

\subsection{Definition and first properties}

Let us first state the definition of these new iterated sums without taking care of their convergence.
\begin{defn} \label{defn:tree_zeta_values}
 For an $\N^*$-decorated rooted forest $F$, whenever it exists, let 
 \begin{equation} \label{eq:TZV}
  \zeta^t(F) := \sum_{\substack{ n_v\geq1\\v\in V(F)}}\prod_{v\in V(F)}\left(\sum_{\substack{v'\in V(F)\\ v'\succeq v}}n_{v'}\right)^{-\alpha_v}
 \end{equation}
 (with $\alpha_v=d_F(v)\in\N^*$ the decoration of the vertex $v$ and, as before, $v\in V(F)$ in the first sum means that this series has a summation variable for each $v\in V(F)$) be the {\bf tree zeta value} (TZV) associated to $F$. $\zeta^t$ is then extended by linearity to a linear map defined on a subset of $\calF_{\N^*}$.
\end{defn}
\begin{rk} 
In a private communication, F. Zerbini  suggested that tree zeta values could be of interest, and in particular that they were likely to be MZVs. I am thankful for this input.
\end{rk}
Here are some simple examples of TZVs, one of which was discussed in Remark \ref{rk:AZV_TZV}:
\begin{equation*}
 \zeta^t(\tdun{n})=\zeta_\stuffle(n),\qquad \zeta^t(\tdtroisun{$2$}{$1$}{$2$})=\sum_{n_1,n_2,n_3\geq1}\frac{1}{(n_1+n_2+n_3)^2n_2(n_3)^2}.
\end{equation*}
\begin{rk}
 As already pointed out, a simple change of summing variables in the interated series of Equation \eqref{eq:MZV_stuffle1} allows to rewrite stuffle MZVs as
 \begin{equation*}
  \zeta_\stuffle(n_1\cdots n_k)=\sum_{m_1,\cdots,m_k\geq1}\frac{1}{(m_1+\cdots+m_k)^{n_1}(m_2+\cdots+m_k)^{n_2}\cdots(m_k)^{n_k}}.
 \end{equation*}
 This is the expression of a tree zeta value associated to a ladder tree decorated by $n_1,\cdots,n_k$. Therefore, TZVs are indeed a new generalisation of MZVs to rooted forests.
\end{rk}
Recall from Definition \ref{defn:branched_bin_map} that the branched binarisation map $\fraks^T:\calF_{\N^*}\longrightarrow\calF_{\{x,y\}}$ is the the unique morphism of $\N^*$-operated algebras given by the universal property of $\calF_{\N^*}$. Roughly speaking, it can be understood as mapping a \ty{vertex} decorated by $n$ to a segment of length $n$ (in the sense of Definition \ref{defn:segment}).
We then have a simple but important result which is the justification of the definition of TZVs:
 \begin{prop} \label{prop:crucial}
  For any convergent $\N^*$-decorated rooted forest $F$, the tree zeta value $\zeta^t(F)$ associated to $F$ is convergent and is equal to the arborified zeta values associated to the convergent $\{x,y\}$-decorated rooted forest $\fraks^T(F)$:
  \begin{equation*}
   \zeta^t(F) = \zeta^T_\shuffle(\fraks^T(F)).
  \end{equation*}
 \end{prop}
 \begin{proof}
  The result follows from the simple observation that, for any convergent $\{x,y\}$-decorated rooted forest $f$ the series representation of $\zeta_\shuffle^T(f)$ given by Theorem \ref{thm:integral_sum} is precisely the tree zeta value $\zeta^t((\fraks^T)^{-1}(f))$. The convergence of $\zeta^t(F)$ for any convergent $\N^*$-decorated rooted forest $F$ then follows from the facts that $\fraks^T$ is a one to one map between the two sets of convergent rooted forests and that the arborified zeta value $\zeta_\shuffle^T(f)$ converges for any convergent $\{x,y\}$-decorated rooted forest $f$.
 \end{proof}
 
 From this Proposition, one easily derives important properties of tree zeta values from the properties of branched zeta values.
 \begin{theo} \label{thm:tree_zeta_MZVs}
  \begin{itemize}
   \item The map $\zeta^t:F\mapsto\zeta^t(F)$ is an algebra morphism from convergent $\N^*$-decorated rooted forests to $\R$ for the concatenation product of forests.
   \item For any convergent $\N^*$-decorated rooted forests $F$, the tree zeta value $\zeta^t(F)$ is a $\Q$-linear combination of MZVs, given by
   \begin{equation*}
    \zeta^t(F) = (\zeta_\shuffle\circ fl_0\circ\fraks^T)(F).
   \end{equation*}
  \end{itemize}
 \end{theo}
 \begin{proof}
  Both points follow directly from Proposition \ref{prop:crucial}
  \begin{itemize}
   \item The first point follows from the fact that $\zeta^t=\zeta_\shuffle^T\circ\fraks^T$ together with the fact that both $\zeta_\shuffle^T$ and $\fraks^T$ are algebra morphisms for the concatenation product of rooted forests.
   \item The second point follows from the same relation $\zeta^t=\zeta_\shuffle^T\circ\fraks^T$ together with Theorem \ref{thm:main_result_shuffle} which states that $\zeta_\shuffle^T=\zeta_\shuffle\circ fl_0$.
  \end{itemize}
 \end{proof}
 \begin{rk}
  Finite sums similar to TZVs were studied in \cite{On16} in the context of finite MZVs. In particular, \cite[Theorem 1.4]{On16} is the equivalent to Theorem \ref{thm:tree_zeta_MZVs} in the context of finite sums.
 \end{rk}
  Notice that $\zeta^t$ is \emph{not} an algebra morphism for the shuffle product of trees, since $\fraks^T$ is not. It can also be checked on computations that it neither is an algebra morphism for any of the stuffle products of rooted forests.  The rest of this Section is dedicated to the study of the algebraic properties of TZVs.
  
   \subsection{The Upsilon product}
 
 The $\yew$ product is an alternative product which will have interesting applications to TZVs as well as their applications to CZVs.
 \begin{defn} \label{defn:yew}
  The $\yew$ product (read ``Upsilon'')\footnote{which was the rune ``yew'' in the paper \cite{Cl20}, but the rune package seems to bring more issues than it is worth for long documents, and I have chosen to modify it here.} is a product $\yew:\calF_{\N^*}\otimes\calF_{N^*}\longrightarrow\calF_{\N^*}$ defined by
  \begin{equation*}
   \yew=(\s^T)^{-1}\circ\shuffle^T\circ(\s^T\otimes\s^T).
  \end{equation*}
 \end{defn}
 \begin{rk} \label{rk:yew_words}
  In some applications, we will use the $\yew$ product on words: for two words $w$ and $w'$ we will write $w\yew w'$ instead of $\iota^{-1}(\iota(w)\yew\iota(w'))$ (with $\iota$ the canonical injection of words into rooted forests). We allow ourselves to make this small abuse of notation as it should not give rise to any confusion and will allow to greatly lighten the notations.
 \end{rk}
 We start by a simple but important property of the $\yew$ product.
 \begin{prop} \label{prop:properties_yew}
  The $\yew$ product is commutative and unital, but not associative.
 \end{prop}
 \begin{proof}
  The commutativity directly follows from the commutativity of $\shuffle^T$. The non associativity as well and can be check on examples: take $F_1 = \tdun{2} \tdun{2}$, $F_2 = \tdun{2}$ and $F_3 = \tdun{2}$. After some computations we obtain:
	\begin{align*}
	  (F_1 \yew F_2) \yew F_3 &- F_1 \yew (F_2 \yew F_3) = 
	 2 \times \tddeux{2}{2} \tddeux{2}{2} + 8 \times \tddeux{2}{2} \tddeux{3}{1} + 8 \times \tddeux{3}{1} \tddeux{3}{1} \\
	  &- \left[ 3 \times \tdun{2} \tdtrois{2}{2}{2} + 6 \times \tdun{2}
\tdtrois{3}{1}{2} + 6 \times \tdun{2} \tdtrois{2}{3}{1}
+ 12 \times \tdun{2} \tdtrois{3}{2}{1} + 18 \times \tdun{2}
\tdtrois{4}{1}{1} \right].
	\end{align*}
	The fact that $\yew$ is unital follows directly from the facts that $\fraks^T(\emptyset)=\emptyset$ and that $\emptyset$ is the neutral element for $\shuffle^T$.
 \end{proof}
 We will now work out an explicit formula for the $\yew$ product. We start with two technical Lemmas about the shuffle product of rooted forests decorated by $\{x,y\}$.
  
 \begin{lemma} \label{lem:lem de lem pour yew}
  For any rooted forests $(f,g)\in\calF_{\{x,y\}}^2$ and any $n\in\N$ we have:
  \[ \big( (B_+^x)^{\circ n}\circ  B_+^y(f) \big) \shuffle^T B_+^y(g) 
	=  (B_+^x)^{\circ n}\circ  B_+^y \big[ f \shuffle^T B_+^y(g) \big] + \sum_{j=0}^n (B_+^x)^{\circ j}\circ  B_+^y \Big[ \big( (B_+^x)^{\circ(n-j)}  B_+^y(f) \big) \shuffle^T g \Big].  \]
 \end{lemma}
\begin{proof}
 We prove this result by induction on $n$. For $n=0$, the lemma holds by definition of the shuffle product $\shuffle^T$. Let us assume it holds for some $n\in\N$. We then have
 \begin{align*}
	& \big( (B_+^x)^{\circ(n+1)}\circ B_+^xy(f) \big) \shuffle^T B_+^y(g) \\
	= & B_+^x\circ \Big[ \big( (B_+^x)^{\circ n}\circ B_+^y (f) \big) \shuffle^T B_+^y(g) \Big] + B_+^y \Big[ \big( (B_+^x)^{\circ(n+1)}\circ  B_+^y (f) \big) \shuffle^T g \Big]  \\
		= & B_+^x \bigg( (B_+^x)^{\circ n}\circ B_+^y \big[ f \shuffle^T B_+^y(g) \big] + \sum_{j=0}^a (B_+^x)^{\circ j}\circ B_+^y \Big[ \big( (B_+^x)^{\circ(n-j)}\circ B_+^y(f) \big) \shuffle^T g \Big] \bigg)
         \\
        & \hspace{4cm}\qquad +  B_+^y \Big[ \big( (B_+^x)^{\circ(n+1)}\circ B_+^y (f) \big) \shuffle^T g \Big]  \text{ by the induction hypothesis} \\
% 		=&  (B_+^x)^{\circ(n+1)}\circ B_+^y \big[ f \shuffle^T B_+^y(g) \big] + \sum_{j=0}^n (B_+^x)^{\circ(j+1)}\circ B_+^y \Big[ \big( (B_+^x)^{\circ(n-j)}\circ B_+^y(f) \big) \shuffle^T g \Big] + B_+^y\circ \Big[ \big( (B_+^x)^{\circ(n+1)}\circ B_+^y (f) \big) \shuffle^T g \Big] \\
		&= (B_+^x)^{\circ(n+1)}\circ B_+^y \big[ f \shuffle^T B_+^y(g) \big] + \sum_{j=0}^{n+1} (B_+^x)^{\circ j}\circ B_+^y \Big[ \big( (B_+^x)^{\circ(n+1-j)} \circ B_+^y(f) \big) \shuffle^T g \Big]
	\end{align*}
	which conclude the induction and the proof.
\end{proof}
We generalise this result:
\begin{lemma} \label{lem:yew_inter}
 For any rooted forests $(f,g)\in\calf_{\{x,y\}}^2$ and any $(n,m)\in\N^2$ we have:
	\begin{align*}
	\big( (B_+^x)^{\circ n}\circ  B_+^y(f) \big) \shuffle^T \big( (B_+^x)^{\circ m}\circ  B_+^y(g) \big) &
	=  \sum_{i=0}^{m} \binom{n+i}{i} (B_+^x)^{\circ(n+i)}\circ B_+^y \left[ f \shuffle^T \big( (B_+^x)^{\circ(m-i)}\circ B_+^y(g) \big) \right] \\
	 +& \sum_{j=0}^{n} \binom{m+j}{j} (B_+^x)^{m+j}\circ B_+^y \left[ \big( (B_+^x)^{\circ(n-j)}\circ B_+^y(f) \big) \shuffle^T g \right].
	\end{align*}
\end{lemma}
\begin{proof}
 We prove this result by induction on $n+m$. If $n+m=0$, then it holds by definition of the shuffle product of rooted trees. Assume that the result hold for some $N=n+m$. Then for any $(n,m)\in\N^2$ such that $n+m=N+1$, if $n=0$ or $m=0$ the result reduces to the previous Lemma \ref{lem:lem de lem pour yew} (eventually using the commutativity of $\shuffle^T$). For $n\neq 0$ and $m\neq0$ we have
 \begin{align*}
  \big( (B_+^x)^{\circ n}\circ B_+^y(f) \big) \shuffle^T \big( (B_+^x)^{\circ m}\circ B_+^y(g) \big) & = B_+^x\circ \Big( \big( (B_+^x)^{\circ(n-1)}\circ B_+^y(f) \big) \shuffle^T \big( (B_+^x)^{\circ m}\circ B_+^y(g) \big) \Big) \\ 
	 + & B_+^x\circ \Big( \big( (B_+^x)^{\circ n}\circ B_+^y(f) \big) \shuffle^T \big( (B_+^x)^{\circ(m-1)}\circ B_+^y (g) \big) \Big).
 \end{align*}
 Using the induction hypothesis we obtain:
 \begin{align*}
  & B_+^x\circ \Big( \big( (B_+^x)^{\circ(n-1)}\circ B_+^y(f) \big) \shuffle^T \big( (B_+^x)^{\circ m}\circ B_+^y(g) \big) \Big) \\
  = & B_+^x \Bigg( 
  \sum_{i=0}^{m} \binom{n-1+i}{i} (B_+^x)^{\circ(n-1+i)}\circ B_+^y \left[ f \shuffle^T \big( (B_+^x)^{\circ(m-i)}\circ B_+^y(g) \big) \right] \\
  & \hspace{5cm}+ \sum_{j=0}^{n-1} \binom{m+j}{j} (B_+^x)^{m+i}\circ B_+^y \left[ \big( (B_+^x)^{\circ(n-1-j)}\circ B_+^y(f) \big) \shuffle^T g \right] \Bigg) \\
%   = & \sum_{i=0}^{b} \binom{a-1+i}{i} B^{a+i}_x B_y \left[ f \shuffle^T \big( B^{b-i} B_y(g) \big) \right] \\
% 	&\hspace*{5mm} + \sum_{j=0}^{a-1} \binom{b+j}{j} B^{b+1+j}_x B_y \left[ \big( B^{a-1-j}_x B_y(f) \big) \shuffle^T g \right] \\
  = & \sum_{i=0}^{m} \binom{n-1+i}{i} (B_+^x)^{\circ(n+i)}\circ B_+^y \left[ f \shuffle^T \big( (B_+^x)^{\circ(m-i)}\circ B_+^y(g) \big) \right] \\
	&\hspace{5cm}
	+ \sum_{j=1}^{n} \binom{m-1+j}{j-1} (B_+^x)^{\circ(m+j)}\circ B_+^y \left[ \big( (B_+^x)^{n-j}\circ B_+^y(f) \big) \shuffle^T g \right].
 \end{align*}
 Similarly we obtain:
 \begin{align*}
  & B_+^x\circ \Big( \big( (B_+^x)^{\circ n}\circ B_+^y(f) \big) \shuffle^T \big( (B_+^x)^{\circ(m-1)}\circ B_+^y (g) \big) \\
  = & \sum_{i=1}^{m} \binom{n-1+i}{i-1} (B_+^x)^{\circ(n+i)}\circ B_+^y \left[ f \shuffle^T \big( (B_+^x)^{\circ(m-i)} B_+^y(g) \big) \right] \\
	&\hspace{5cm} + \sum_{j=0}^{n} \binom{m-1+j}{j}  (B_+^x)^{\circ(m+j)}\circ B_+^y  \left[ \big( (B_+^x)^{\circ(n-j)}\circ B_+^y(f) \big) \shuffle^T g \right].
 \end{align*}
 Summing these two expressions and using Pascal's triangle (and the fact that $\binom{p-1}{0}=\binom{p}{0}$) gives the result, which concludes the induction and the proof.
\end{proof}
  \begin{rk}
 This result can be understood in a purely combinatorial way. Each tree in this shuffle has to have at least min$\{n,m\}$ $x$'s decorating its root and its descendants. These decorations can come either from the term $(B_+^x)^{\circ n}\circ B_+^y(f)$ or from the term $(B_+^x)^{\circ m}\circ  B_+^y(g)$. The binomial coefficients come from all the possibilities the $x$'s have to be chosen.
\end{rk}
We can now prove the main result of this Subsection, which gives an inductive formula to directly compute the $\yew$ product of two rooted forests, without references to the shuffle $\shuffle^T$ nor the binarisation map $\s^T$.
\begin{theo} \label{thm:yew_formula}
 The $\yew$ product admits the following inductive description:
 \begin{itemize}
	\item For any forest $F\in\calf_{\N^*}$,  $F \yew \emptyset = \emptyset \yew F = F$. 
	\item For any rooted trees $T_1 = B_+^n (f_1)$ and $T_2 = B_+^m(f_2)$:
	$$T_1 \yew T_2 = \sum_{i=0}^{m-1} \binom{n-1+i}{i} B_+^{n+i}\left[ f_1 \yew B_+^{m-i}(f_2) \right]  + \sum_{j=0}^{n-1} \binom{m-1+j}{j} B_+^{m+j} \left[ B_+^{n-j}(f_1) \yew f_2 \right].$$
	\item For any rooted forests $F_1 = T_1\cdots T_n$ and $F_2 = t_1\cdots t_k$:
	\[ F_1 \yew F_2 = \frac{1}{kn} \sum_{i=1}^k \sum_{j=1}^n \left(  (T_i \yew t_j) T_1 \cdots \widehat{T_i} \cdots T_k t_1 \cdots \widehat{t_j} \cdots t_n \right). \]
	\end{itemize}
\end{theo}
\begin{proof}
 \begin{itemize}
  \item The first point follows directly from the facts that $\s^T(\emptyset)=\emptyset$ and that $\emptyset$ is the unique neutral element for the shuffle product $\shuffle^T$.
  \item For the second point, let $T_1 = B_+^n (f_1)$ and $T_2 = B_+^m(f_2)$ be as in the Theorem. Then
  \begin{align*}
	& T_1 \yew T_2 = (\s^T)^{-1} ( \s^T(T_1) \shuffle^T \s^T(T_2))\qquad\text{by definition of }\yew  \\
    = & (\s^T)^{-1} \left( (B_+^x)^{\circ(n-1)}\circ B_+^y (\s^T(f_1)) \shuffle^T (B_+^x)^{\circ(n-1)}\circ B_+^y (\s^T(f_2)) \right) \qquad\text{by definition of }\s^T\\
    = & (\s^T)^{-1} \Bigg( \sum_{i=0}^{m-1} \binom{n-1+i}{i} (B_+^x)^{\circ(n-1+i)}\circ B_+^y \left[ \s^T(f_1) \shuffle^T (B_+^x)^{\circ(m-1-i)}\circ B_+^y (\s^T(f_2)) \right]  \\ 
		& + \sum_{j=0}^{n-1} \binom{m-1+j}{j} (B_+^x)^{\circ(m-1+j)}\circ B_+^y \left[ (B_+^x)^{\circ(n-1-j)}\circ B_+^y (\s^T(f_1)) \shuffle^T \s^T(f_2) \right] \Bigg)  ~ \text{by Lemma \ref{lem:yew_inter}} \\
    = & \sum_{i=0}^{m-1} \binom{n-1+i}{i} B_+^{n+i} \left( (\s^T)^{-1}\left[ \s^T(f_1) \shuffle^T \s^T(B_+^{m-i}f_2) \right] \right) \\ 
		& + \sum_{j=0}^{n-1} \binom{m-1+j}{j} B_+^{m+j} \left( (\s^T)^{-1} \left[\s^T( B_+^{n-j} f_1) \shuffle^T \s^T(f_2) \right] \right) \qquad\text{by definition of }(\s^T)^{-1} \\
    = & \sum_{i=0}^{m-1} \binom{n-1+i}{i} B_+^{n+i}\left[ f_1 \yew B_+^{m-i}(f_2) \right]  + \sum_{j=0}^{n-1} \binom{m-1+j}{j} B_+^{m+j} \left[ B_{n-j}(f_1) \yew f_2 \right]
	\end{align*}
	by definition of $\yew$.
  \item Finally, let $F_1 = T_1\cdots T_n$ and $F_2 = t_1\cdots t_k$ be as in the Theorem. Then 
  \begin{align*}
	& F_1 \yew F_2 = (\s^T)^{-1}( s^T(F_1) \shuffle^T \s^T(F_2)) \qquad\text{by definition of }\yew  \\
    = & (\s^T)^{-1} \Bigg( \frac{1}{kn} \sum_{i=1}^k \sum_{j=1}^n \Big(  (\s^T(T_i) \shuffle^T \s^T(t_j)) \s^T(T_1) \cdots \widehat{\s^T(T_i)} \cdots \s^T(T_k) \s^T(t_1) \cdots \widehat{\s^T(t_j)} \cdots \s^T(t_n) \Big) \Bigg) \\
    & \hspace{6cm}\llcorner\text{since }\s^T\text{ is an algebra morphism and by definition of }\shuffle^T \\
% 		&= \frac{1}{kn} \sum_{i=1}^k \sum_{j=1}^n \left(  (\s^T)^{-1}\left( (\s^T(T_i) \shuffle^T \s^T(t_j) \right) T_1 \cdots \widehat{T_i} \cdots T_k t_1 \cdots \widehat{t_j} \cdots t_n \right) \\
		&= \frac{1}{kn} \sum_{i=1}^k \sum_{j=1}^n \left(  (T_i \yew t_j) T_1 \cdots \widehat{T_i} \cdots T_k t_1 \cdots \widehat{t_j} \cdots t_n \right).
	\end{align*}
 \end{itemize}
\end{proof}
This result allows reasonably fast computations of $\yew$ products of rooted forests, especially in cases with few branchings. Let us illustrate this with some examples.
  \begin{example} \label{ex:yew_products}
We start with an explicit computation:
 \begin{align*}
		\tdun{2} \yew \tddeux{3}{1} 
		&= B_+^2( \tddeux{3}{1}) + 2. B_+^3( \tddeux{2}{1}) + 3. B_+^4( \tddeux{1}{1}) + B_+^3( \tdun{2} \yew \tdun{1}) + 3. B_+^4( \tdun{1} \yew \tdun{1}) \\
		&= \tdtrois{2}{3}{1} + 2. \tdtrois{3}{2}{1} + 3. \tdtrois{4}{1}{1} + B_+^3 \left[ B_+^2(\tdun{1}) + B_+^1( \tdun{2}) + B_+^2( \tdun{1}) \right] + 3. B_+^4 \left[ B_+^1( \tdun{1}) + B_+^1( \tdun{1}) \right] \\
% 		&= \tdtrois{2}{3}{1} + 2. \tdtrois{3}{2}{1} + 3. \tdtrois{4}{1}{1} + \tdtrois{3}{2}{1} + \tdtrois{3}{1}{2} + \tdtrois{3}{2}{1} + 6.\tdtrois{4}{1}{1} \\
		&= \tdtrois{2}{3}{1} + 4. \tdtrois{3}{2}{1} + 9. \tdtrois{4}{1}{1} + \tdtrois{3}{1}{2}
	\end{align*}
 The same computation can be carried out with arbitrary coefficients. We obtain
 \begin{align*}
  \tdun{m} \yew \tddeux{n}{k} = & \sum_{i=0}^{m-1}\binom{n-1+i}{i}\left[\sum_{i'=0}^{k-1}\binom{m-i-1+i'}{i'}\tdtrois{n+i}{m-i+i'}{k-i'} + \sum_{j'=0}^{m-i-1}\binom{k-1+j'}{j'}\tdtrois{n+i}{k+j'}{m-i-j'}\right] % \\ 
%   & 
+ \sum_{j=0}^{n-1}\binom{m-1+j}{j}\tdtrois{m+j}{n-j}{k}.
 \end{align*}
 We can carry further the computation. For example, one finds
 \begin{align*}
  \tddeux{n}{k}\yew\tddeux{m}{l} & =  \sum_{a=0}^{m-1}\binom{n-1+a}{a}
  \left\{ 
  \sum_{i=0}^{k-1}\binom{m-a-1+i}{i}
  \left[
  \sum_{i'=0}^{l-1} \binom{k-i-1+i'}{i'}
  \tdquatre{n+a}{m-a+i}{k-i+i'}{l-i'} + \sum_{j'=0}^{k-i-1}\binom{l-1+j'}{j'}\tdquatre{n+a}{m-a+i}{l+j'}{k-i-j'}~
  \right]\right. \\
  + & \left.\sum_{j=0}^{m-a-1}\binom{k-1+j}{j}\tdquatre{n+a}{k+j}{m-a-j}{l}\right\} + \sum_{b=0}^{n-1}\binom{m-1+b}{b}\left\{\sum_{i=0}^{l-1}\binom{n-b-1+i}{i}\left[\sum_{i'=0}^{k-1}\binom{l-i-1+i'}{i'}\tdquatre{m+b}{n-b+i}{l-i+i'}{k-i'}    \right.\right.  \\
  & +  \left.\left. \sum_{j'=0}^{l-i-1}\binom{k-1+j'}{j'}\tdquatre{m+b}{n-b+i}{k+j'}{l-i-j'}~\right] + \sum_{j=0}^{n-b-1}\binom{l-1+j}{j}\tdquatre{m+b}{l+j}{n-b-j}{k}\right\}.
 \end{align*}

\end{example}
 
 \subsection{The Upsilon product and TZVs}
 
  We will now show that TZVs form an algebra morphism for the $\yew$ product. We first need a simple property.
 \begin{prop} \label{prop:yew_stabilises_conv_forests}
  The $\yew$ product stabilises $\calf_{\N^*}^{\rm conv}$ i.e., for any $F_1,F_2\in\calf_{\N^*}^{\rm conv}$, $F_1\yew F_2\in \calf_{\N^*}^{\rm conv}$.
 \end{prop}
 \begin{proof}
  The result follows from the observation that $\s^T(\calf_{\N^*}^{\rm conv})=\calf_{\{x,y\}}^{\rm conv}$,  (\cite[Lemma A.3]{Cl20}), the fact that $\shuffle^T$ stabilises $\calf_{\{x,y\}}^{\rm conv}$ (\cite[Lemma 5.9]{Cl20}) and the obvious fact that $\s^T$ is a one-to-one map.
 \end{proof}
  Now our main result result makes sense.
 \begin{theo} \label{thm:yew_TZVs}
  The map $\zeta^t:\calf_{\N^*}^{\rm conv}\longrightarrow\R$ is an algebra morphism for the $\yew$ product:
  \begin{equation*}
   \forall F_1,F_2 \in \calf_{\N^*}^{\rm conv}, \quad \zeta^t \left( F_1 \yew F_2 \right) = \zeta^t(F_1) \zeta^t(F_2).
  \end{equation*}
 \end{theo}
 \begin{proof}
  For any $F_1,F_2 \in \calf_{\N^*}^{\rm conv}$ we have
  \begin{align*}
		\zeta^t(F_1 \yew F_2) &= \zeta_\shuffle^T( \s( F_1 \yew F_2) ) &\text{by Theorem \ref{thm:integral_sum}}\\
			&= \zeta_\shuffle^T( \s(F_1) \shuffle^T \s(F_2)) 	&\text{by definition of }\yew \\
			&= \zeta_\shuffle^T(\s(F_1)) \zeta_\shuffle^T(\s(F_2))	&\text{since }\zeta_\shuffle^T\text{ is an algebra morphism for }\shuffle^T\\
			&= \zeta^t(F_1) \zeta^t(F_2)	&\text{by Theorem \ref{thm:integral_sum}}
	\end{align*}
 \end{proof}
 As for the shuffle product $\shuffle^T$, the non-associativity of the $\yew$ product implies the existence of relations amongst TVZs that have no direct equivalent for MZVs. These relations follow directly from the previous theorem and the associativity of the product on $\R$. More precisely, we have
 \begin{cor}
  The image of the associator of $\yew$ restricted to $\calf_{\N^*}^{\rm conv}$ lies in $\ker(\zeta^t)$ i.e., for any $F_1,F_2$ and $F_3$ in $\calf_{\N^*}^{\rm conv}$ we have
  \begin{equation*}
   (F_1 \yew F_2) \yew F_3 - F_1 \yew (F_2 \yew F_3) \in \ker(\zeta^t).
  \end{equation*}
 \end{cor}
 Interestingly, these are not the only relations between TZVs and the $\yew$ product. Indeed, \ty{the latter} allows to relate $\zeta^t$ directly to the stuffle MZV map $\zeta_\stuffle$. We start by an important proposition which rely upon the associativity of the shuffle product of words $\shuffle$. 
 
 In order to state this result, recall that a {\bf ladder tree} decorated by a set $\Omega$ is a rooted tree of the form $l=B_+^{\omega_1}\circ\cdots\circ B_+^{\omega_n}(\emptyset)$ for some $n\in\N^*$ and $(\omega_1,\cdots,\omega_n)\in\Omega^n$. A {\bf ladder forest} is a rooted forest $f=l_1\cdots l_n$ where the trees $l_i$s are all ladder trees. We then have a simple lemma:
 \begin{lemma}
  The product $\yew$ is associative on ladder trees. In particular, for any ladder trees $l_1,\cdots,l_n$, the quantity $l_1\yew\cdots\yew l_n$ is well-defined and we furthermore have
  \begin{equation} \label{eq:yew_ladders_n}
   l_1 \yew \cdots \yew l_n = (\s^T)^{-1} ( \s^T(l_1) \shuffle^T \cdots \shuffle^T \s^T(l_n)).
  \end{equation}
 \end{lemma}
  \begin{proof}
  For any forests $f$, $g$ and $h$, using the definition of $\yew$ we obtain
  \begin{equation*}
   (f\yew g)\yew h = (\s^T)^{-1}\left(\left(\s^T(f)\shuffle^T\s^T(g)\right)\shuffle^T\s^T(h)\right).
  \end{equation*}
  Furthermore, if $f$ and $g$ are ladder trees we have
  \begin{equation*}
   \s^T(f)\shuffle^T\s^T(g)=\iota\left(\s\left(\iota^{-1}(f)\right)\shuffle\s\left(\iota^{-1}(g)\right)\right)
  \end{equation*}
  with $\iota$ the canonical injection of words into rooted forests. The associativity of $\yew$ on ladders trees follows from this formula and the associativity of the usual shuffle product $\shuffle$ on words.
  
  For the second statement on the Lemma, we already have by the previous result that for any ladder trees $l_1,\cdots,l_n$, the quantity $l_1\yew\cdots\yew l_n$ is well-defined. Let us now prove Equation \eqref{eq:yew_ladders_n} by induction over $n\geq2$. 
  
  For $n=2$, Equation \eqref{eq:yew_ladders_n} is the definition of the \cy{upsilon} product $\yew$. Assuming that for some $n\geq2$, the result holds for any $k\in\{2,\cdots,n\}$ take $l_1,\cdots,l_{n+1}$ be $n+1$ \ty{ladder} trees. If $n+1=3$, the result holds from the computation we performed for the first point of the Lemma. If $n+1\geq4$ we have 
  \begin{align*}
	l_1 \yew \cdots\yew l_{n+1} &= ((\s^T)^{-1}(\s^T(l_1) \shuffle^T \s^T(l_2))) \yew l_3 \cdots \yew l_{n+1} \\
		&= ((\s^T)^{-1}(\s^T(l_1) \shuffle^T \s(l_2))) \yew (\s^T)^{-1}( \s^T(l_3) \shuffle^T \cdots \shuffle^T \s^T(l_{n+1})) \\
		&= (\s^T)^{-1} \left( \s^T(l_1) \shuffle^T \cdots \shuffle^T \s^T(l_{n+1})\right)
	\end{align*}
  by the associativity of $\yew$ on ladder trees, the induction hypothesis and the definition of $\yew$ respectively. This concludes our proof.
 \end{proof}
 In other words, and with the abuse of words we allowed ourselves to make in Remark \ref{rk:yew_words}, $(\calW_{\N^*},\yew)$ is an associative algebra. This allows us to define the $\yew$-flattening map in a similar fashion than the usual flattening maps (see Definition \ref{defn:flattening_maps}):
 \begin{defn} \label{defn:flattening_yew}
 The {\bf $\yew$-flattening map} $fl_\yew:\mathcal{F}_{\N^*}\longrightarrow\calW_{\N^*}$  from the algebra of rooted forests decorated by $\N^*$ and 
 the algebra of words $\mathcal{W}_{\N^*}$  written in the alphabet ${\N^*}$ is the unique map of operated algebra between $(\calF_{\N^*},B_+,.)$ and $(\calW_{\N^*},C_+,\yew)$ given by Theorem \ref{thm:univ_prop_tree}; where $.$ is the usual concatenation of rooted forests and $C_+$ is the operation on words given by the left concatenation (see Equation \eqref{eq:defn_c_plus}).
 \end{defn}
 The fact that this definition makes sense may be related to the fact that the flattening maps are defined through an operation of $\N^*$ on vector spaces of words which is essentially the concatenation of words, i.e. the free multiplication. I thank M. Bordemann to have pointed this to me.
  
 Notice that the $\yew$-flattening map is recursively given by
 \begin{equation*}
  fl_\yew(\emptyset)=\emptyset,\quad fl_\yew(F_1F_2) = fl_\yew(F_1) \yew fl_\yew(F_2), \quad fl_\yew(B_+^n(F))=(n)\sqcup fl_\yew(F).
 \end{equation*}
 extended by linearity to a map on $\calf_{\N^*}$ and $\iota$ the canonical injection of words into rooted forests.
 
 We can now prove another important result of this Section.
\begin{prop} \label{thm:comm_fl_fl_yew}
 The flattening and binarisation maps are related by $fl_0 \circ \s^T = \s \circ fl_\yew$.
\end{prop}
\begin{proof}
 We will show that this result holds for any rooted forests as usual, i.e. by induction on the number of vertices of the forest. For $F\in\calF_{\N^*}$, let $n$ be the number of vertices of $F$.
 
 If $n=0$, we have $F=\emptyset$ and $fl_0 \circ \s^T(\emptyset) = \emptyset = \s \circ fl_\yew(\emptyset)$, and the result holds for $n=0$. Now for some $n\in\N$ let us assume that the result holds for any rooted forest with $k\leq n$ vertices and let $F$ be a forest with $n+1$ vertices. We have two cases to consider:
 \begin{itemize}
  \item If $F$ is a rooted tree, then we can write $F=B_+^r(f)$ for some $r\in\N^*$ and rooted forests $f$ with $n$ vertices. Then we have
  \begin{align*}
		 \s \circ fl_\yew(F) &= \s \circ fl_\yew \circ B_+^r (f) \\
		 	&= \s \circ (r \sqcup fl_\yew (f)) &\text{by definition of $fl_\yew$}\\
% 		 	&= (\underbrace{x\cdots x}_{r-1~{\rm times}}y)\sqcup (\s \circ fl_\yew)(f) &\text{par définition de $\s$.}\\
&= (x\cdots xy)\sqcup (\s \circ fl_\yew)(f)  &\text{by definition of $\s$}\\
		 	&= (x\cdots xy)\sqcup ( fl_0 \circ \s^T)(f) &\text{by the induction hypothesis} \\
		 	&= fl_0\circ (B_+^x)^{\circ (r-1)} \circ B_+^y \circ \s^T(f) &\text{by definition of $fl_0$}\\
		 	&= fl_0 \circ \s^T \circ B_+^r (f) &\text{by definition of $\s^T$}\\
		 	&= fl_0 \circ \s^T(F).
  \end{align*}
  In this computation, $(x\cdots xy)$ was always $r-1$ $x$s and one $y$.
  \item If $F = t_1\cdots t_k$ with $t_1,\cdots, t_k$ rooted trees then:
		\begin{align*}
	  	 \s \circ fl_\yew(F) &= \s \circ( fl_\yew(t_1) \yew \cdots \yew fl_\yew(t_k)) &\text{by definition of $\yew$} \\
	  	 &= (\s \circ fl_\yew)(f_1) \shuffle \cdots \shuffle (\s \circ fl_\yew)(f_k) &\text{by Equation \eqref{eq:yew_ladders_n}} \\
	  	 &= (fl_0 \circ \s^T)(f_1) \shuffle \cdots \shuffle (fl_0 \circ \s^T)(f_k) &\text{by the induction hypothesis} \\
	  	 &= (fl_0 \circ \s^T) (f_1 \cdots f_k) &\text{by definition of $fl$}.
		\end{align*}
 \end{itemize}
 This concludes the induction and the proof.
\end{proof}
Let us apply this result to TZVs. We first need the following simple Lemma:
\begin{lemma}
 $fl_\yew$ maps convergent rooted forests to convergent words: $fl_\yew( \calF_{\N^*}^{\rm conv}) = \calW_{\N^*}^{\rm conv}$.
\end{lemma}
\begin{proof}
 Notice that any convergent rooted tree $T$ is mapped by $fl_\yew$ to a linear combination of words starting by the decoration of the root of $T$. Furthermore, any convergent rooted forest $F$ is mapped by $fl_\yew$ to a linear combination of words starting by the decoration of the one of the root of $F$. Thus we have 
 $fl_\yew( \calF_{\N^*}^{\rm conv}) \subseteq \calW_{\N^*}^{\rm conv}$.
 
 Recall that $\iota:\calF_{\N^*}\longrightarrow\calW_{\N^*}$ is the canonical injection of words into rooted forests, mapping any word to a ladder tree. Then from the definition of $fl_\yew$ we have, for any word $w\in\calW_{\N^*}$ (not necessarily convergent) $fl_\yew(\iota(w))=w$. This holds in particular for convergent words and we have therefore that $fl_\yew( \calF_{\N^*}^{\rm conv}) = \calW_{\N^*}^{\rm conv}$.
\end{proof}
  We then obtain from Theorem \ref{thm:comm_fl_fl_yew} that the $\yew$-flattening relates TZVs and stuffle MZVs. This gives in particular a new, more direct proof that TZVs are MZVs and a new way to compute them.
\begin{cor} \label{coro:TZVs_stuffle_MZVs}
 We have $\zeta^t = \zeta_\stuffle \circ fl_\yew$. In particular for any convergent forest $F\in\calf_{\N^*}^{\rm conv}$, $\zeta^t(F)$ is a linear combination with integer coefficients of stuffle MZVs given by the $\yew$-flattening. 
\end{cor}
\begin{proof}
 From the previous Lemma, we have that $\zeta^t$ and $\zeta_\stuffle \circ fl_\yew$ are defined on the same set. We need to prove that they are equal.
 \begin{align*}
	\zeta^t &= \zeta^T_\shuffle \circ \s^T &\text{by Proposition \ref{prop:crucial}}\\
	&= \zeta_\shuffle \circ fl_0 \circ \s^T &\text{by Theorem  \ref{thm:flattening}} \\
	&= \zeta_\shuffle \circ \s \circ fl_\yew  &\text{by Theorem \ref{thm:comm_fl_fl_yew}} \\
	&= \zeta_\stuffle \circ fl_\yew &\text{by Kontsevich's relation \eqref{eq:Kontsevich}}
	\end{align*} 
\end{proof}
We will later see that TZVs are conical zeta values (CZVs), and use this result to compute some CZVs. In the mean time, and to conclude this section, we can already use the computations performed in Example \ref{ex:yew_products} to express some TZVs in terms of MZVs.
\begin{example}
 By definition of the $\yew$-flattening we have
 \begin{equation*}
  fl_\yew\left(\tdquatredeux{p}{m}{n}{k} \right) = (p)\sqcup((nk)\yew(m))
 \end{equation*}
 (with the abuse of notation discussed after Definition \ref{defn:flattening_yew}). Then by Corollary \ref{coro:TZVs_stuffle_MZVs} and the second computation of Example \ref{ex:yew_products} we obtain
 \begin{align*}
  & \zeta^t\left(\tdquatredeux{p}{m}{n}{k} \right) =  \sum_{i=0}^{m-1}\binom{n-1+i}{i}\left[\sum_{i'=0}^{k-1}\binom{m-i-1+i'}{i'}\zeta_\stuffle(p;n+i;m-i+i';k-i') \right. \\
  + & \left.\sum_{j'=0}^{m-i-1}\binom{k-1+j'}{j'}\zeta_\stuffle(p;n+i;k+j';m-i-j')\right]
+ \sum_{j=0}^{n-1}\binom{m-1+j}{j}\zeta_\stuffle(p;m+j;n-j;k).
 \end{align*}
 for any $p\geq2$, $(m,n,k)\in(\N^*)^3$.
 
 Finally, we can also use Corollary \ref{coro:TZVs_stuffle_MZVs} on the tree $\tcinq{p}{n}{k}{m}{l}$ and the third computation of Example \ref{ex:yew_products} to obtain
 \begin{align*}
  & \zeta^t\left(\tcinq{p}{n}{k}{m}{l}\right)= \\
   & \sum_{a=0}^{m-1}\binom{n-1+a}{a}\left\{ \sum_{i=0}^{k-1}\binom{m-a-1+i}{i}\left[\sum_{i'=0}^{l-1}\binom{k-i-1+i'}{i'}
  \zeta_\stuffle(p;n+a;m-a+i;k-i+i';l-i') \right.\right. \\
  & \hspace{4cm}+  \left.\sum_{j'=0}^{k-i-1}\binom{l-1+j'}{j'}\zeta_\stuffle(p;n+a;m-a+i;l+j';k-i-j')\right] \\
  & \hspace{6cm}+  \left.\sum_{j=0}^{m-a-1}\binom{k-1+j}{j}\zeta_\stuffle(p;n+a;k+j;m-a-j;l)\right\} \\
 + &  \sum_{b=0}^{n-1}\binom{m-1+b}{b}\left\{\sum_{i=0}^{l-1}\binom{n-b-1+i}{i}\left[\sum_{i'=0}^{k-1}\binom{l-i-1+i'}{i'}\zeta_\stuffle(p;m+b;n-b+i;l-i+i';k-i')    \right.\right.  \\
   & \hspace{4cm}+\left. \sum_{j'=0}^{l-i-1}\binom{k-1+j'}{j'}\zeta_\stuffle(p;m+b;n-b+i;k+j';l-i-j')\right] \\
   & \hspace{6cm}+\left.\sum_{j=0}^{n-b-1}\binom{l-1+j}{j}\zeta_\stuffle(p;m+b;l+j;n-b-j;k)\right\}.
 \end{align*}
 Notice that Corollary \ref{coro:TZVs_stuffle_MZVs}, together with Theorem \ref{thm:yew_formula} allows to (relatively) efficiently compute the evalutation of TZVs in terms of MZVs. 
\end{example}

\section{Applications to other generalisations of MZVs} \label{section:application}

 \subsection{Mordell-Tornheim zeta values} \label{subsection:huit_un}

 A special class of Mordell-Tornheim zetas were introduced in \cite{To50} and studied (albeit not in full generality) in \cite{Mo58} and \cite{Ho92}. Their full version was introduced in  \cite{Mat03} and further investigated in \cite{Ts05} and \cite{BrZh10}. The study of the finite version of Mordell-Tornheim zeta values was carried on in \cite{Ka16}. We will illustrate how Theorem \ref{thm:integral_sum}, together with results described before, has far-reaching consequences for various generalisations of MZVs. In particular, it provides new and straightforward proofs for some results regarding Mordell-Tornheim zeta values.
  \begin{defn}[\cite{To50}] \label{defn:MT}
  Let $(s,s_1,\cdots,s_r)\in\N^{r+1}$ be sequence of non-negative integers. The {\bf Mordell-Tornheim zeta value} associated to this sequence is
  \begin{equation} \label{eq:MT_zeta}
   MT(s_1,\cdots,s_r|s) := \sum_{n_1,\cdots,n_r\geq1}\frac{1}{n_1^{s_1}\cdots n_r^{s_r}(n_1+\cdots+n_r)^s}
  \end{equation}
  whenever this series is convergent. The integer $r$ is called the {\bf depth} of this Mordell-Tornheim zeta values, $s+s_1+\cdots+s_r$ its {\bf weight} and $s_1+\cdots+s_r$ its {\bf partial depth}.
 \end{defn}
 Since the series in \eqref{eq:MT_zeta} is invariant under a permutation of the $s_i$, it is traditional to assume $s_1\leq\cdots\leq s_r$. We will follow this convention.
 
  Bradley and Zhou gave (\cite[Theorem 2.2]{BrZh10}) a condition for the convergence of the series \eqref{eq:MT_zeta} to hold in the more general case where the $s_i$ are complex numbers. With our convention, this condition reads: if for any $k\in\{1,\cdots,r\}$, the inequality
  \begin{equation} \label{eq:conv_crit_BrZh}
   s+\sum_{i=1}^k s_i > k
  \end{equation}
  holds, then the series \eqref{eq:MT_zeta} converges.
  
  The same authors also proved (\cite[Theorem 1.1]{BrZh10}) that any convergent Mordell-Tornheim zeta value of weight $w$ and depth $r$ can be written as a linear combination with rational coefficients of MZVs of weight $w$ and depth $r$. For finite Mordell-Tornheim zeta values, the same result was obtained in
  \cite[Theorem 1.2]{Ka16}.
  
    We will show here that our Theorem \ref{thm:tree_zeta_MZVs} gives elementary proofs of the results of Bradley and Zhou for the case $s_2>0$. We will also provide an explicit formula for Mordell-Tornheim zeta values in the case $s_1=0$, which is reminiscent of \cite[Theorem 1.2]{Ka16} for finite Mordell-Tornheim zeta values.

 \begin{prop} \label{prop:MT_s1_0}
   The Mordell-Tornheim zeta values associated to the sequence $(s,s_1=0,s_2>0,\cdots,s_r)$ is convergent whenever $s\geq2$. In this case $MT(s_1=0,\cdots,s_r|s)$ can be written as a linear combination with integer coefficients of MZVs of weight $s+s_2+\cdots+s_r$ and depth $r$ given by 
   \begin{equation} \label{eq:MT_s1_0}
    MT(s_1=0,s_2,\cdots,s_r|s) = \zeta_\shuffle\left((\underbrace{x\cdots x}_{s-1}y)\sqcup \left((\underbrace{x\cdots x}_{s_2-1}y)\shuffle\cdots\shuffle(\underbrace{x\cdots x}_{s_r-1}y)\right)   \right).
   \end{equation}
  \end{prop}
  Notice that the condition $s\geq2$ is equivalent to the convergence criterion \eqref{eq:conv_crit_BrZh} of Bradley and Zhou in the case $s_1=0$ and $s_2>0$.
 \begin{proof}
   Observe that in the case $s_1=0$ and $s_2>0$ the series \eqref{eq:MT_zeta} coincides with $\zeta^t(T)$ with $(T,d_T)$ the decorated tree with $r$ vertices and $r-1$ leaves. Its root is decorated by $s$ and its leaves by $s_2,\cdots,s_r$.
   Thus $T$ is a convergent tree whenever $s\geq2$. In this case we then have by Theorem \ref{thm:tree_zeta_MZVs}
   \begin{equation*}
    MT(s_1=0,s_2,\cdots,s_r|s) = (\zeta_\shuffle\circ fl_0 \circ\fraks^T)(T).
   \end{equation*}
   Now $\fraks^T(T)$ is a tree with only one branching vertex. The segment between the root and the branching vertex contains $s$ vertices, the first $s-1$ being decorated by $x$ and the last one by $y$. Each of the $r-1$ segments between the branching vertex and one leaf contains $s_i$ vertices (with $i\in\{2,\cdots,r\}$) whose first $s_i-1$ vertices are decorated by $x$ and the last one by $y$. Therefore its flattening is precisely
   \begin{equation*}
    fl(\fraks^T(T))=(\underbrace{x\cdots x}_{s-1}y)\sqcup \left((\underbrace{x\cdots x}_{s_2-1}y)\shuffle\cdots\shuffle(\underbrace{x\cdots x}_{s_r-1}y)\right).
   \end{equation*}
   This gives Equation \eqref{eq:MT_s1_0}. In particular $MT(s_1=0,\cdots,s_r|s)$ can be written as a linear combination with integer coefficients of MZVs of weight $s+s_2+\cdots+s_r$ and depth $r$ (the number of $y$s in each words appearing in the expression of $MT(s_1=0,\cdots,s_r|s)$).
  \end{proof}
  We are now ready to prove our main result regarding Mordell-Tornheim zetas.
  \begin{theo} \label{thm:MT_tree}
   Let $s_1\in\N^*$. The Mordell-Tornheim zeta value associated with $(s,s_1,\cdots,s_r)$ is convergent whenever $s\geq1$. In this case $MT(s_1,\cdots,s_r|s)$ can be written as a linear combination of MZVs of weight $s+\sum_{i=1}^rs_i$ and depth $r$ with integer coefficients.
  \end{theo}
  Notice that, as before, the condition $s\geq1$ is equivalent to the convergence criterion \eqref{eq:conv_crit_BrZh} of Bradley and Zhou in the case $s_1>0$.
  \begin{proof}
   We prove this result by induction on the partial depth $n:=s_1+\cdots+s_r$ of $MT(s_1,\cdots,s_r|s)$. If $n=1$, the conditions $s_1\geq1$ and $s_i\geq s_1$ for all $i\in\{1,\cdots,r\}$ implies $r=1$ and $s_1=1$. In this case the Mordell-Tornheim zeta value reduce to the usual zeta value $\zeta(s+1)$ which is convergent if $s>0$.
   
   Now, assume the result holds for all Mordell-Tornheim zeta values of partial weight $n$ with $s_1\geq1$. Let $s$ and $s_1$ be greater or equal to one and $(s_1,\cdots,s_r)$ be an increasing sequence of integers such that $s_1+\cdots+ s_r=n+1$. Using the partial fraction decomposition we obtain
   \begin{equation*}
    \frac{1}{n_1\cdots n_r(n_1+\cdots+n_r)} = \frac{1}{(n_1+\cdots+n_r)^2} \sum_{i=1}^r\prod_{\substack{j=1 \\j\neq i}}^r \frac{1}{n_j}
   \end{equation*}
   (which holds since $\prod_{\substack{j=1 \\j\neq i}}^r \frac{1}{n_j}=\frac{n_i}{n_1\cdots n_r}$) we obtain
   \begin{equation*}
    \frac{1}{n_1^{s_1}\cdots n_r^{s_r}(n_1+\cdots+n_r)^s} = \sum_{i=1}^r\frac{1}{n_1^{s_1}\cdots n_i^{s_i-1}\cdots n_r^{s_r}(n_1+\cdots+n_r)^{s+1}}.
   \end{equation*}
   Summing over the $n_i$ we obtain (up to an irrelevant permutation of the $s_i$)
   \begin{equation} \label{eq:sum_decrease_weight_MT}
    MT(s_1,\cdots,s_r|s) = \sum_{i=1}^r MT(s_1,\cdots,s_i-1,\cdots,s_r|s+1)
   \end{equation}
   whenever the series on one of the two sides converges. Examining each of the terms of the RHS we have then two cases to consider:
   \begin{enumerate}
    \item $s_i=1$. Then by Proposition \ref{prop:MT_s1_0} $MT(s_1,\cdots,s_i-1,\cdots,s_r|s+1)$ is convergent since $s+1\geq2$, and it is a linear combination with integer coefficients of MZVs of depth $r$ and weight 
    \begin{equation*}
     s+1+\sum_{\substack{k=1\\k\neq i}}^r s_k.
    \end{equation*}
    \item $s_i\geq2$. In this case by the induction hypothesis $MT(s_1,\cdots,s_i-1,\cdots,s_r|s+1)$ is convergent and can be written as a linear combination of MZVs of weight $s+\sum_{i=1}^rs_i$ and depth $r$ with integer coefficients. 
   \end{enumerate}
   In any case the RHS of Equation \eqref{eq:sum_decrease_weight_MT} is convergent whenever $s\geq1$ and is then a finite sum of linear combinations of MZVs of weight $s+\sum_{i=1}^rs_i$ and depth $r$ with integer coefficients.
   
   This concludes the induction and proves the Theorem.
  \end{proof}
  Along the way, we have obtained the decomposition formula
  \begin{equation} \label{eq:MT_decom}
    MT(s_1,\cdots,s_r|s) = \sum_{i=1}^r MT(s_1,\cdots,s_i-1,\cdots,s_r|s+1)
   \end{equation}
  which only existed for finite Mordell-Tornheim zetas (see the first equation of \cite[Corollary 5.2]{BTT21}). I am very thanksful to Masataka Ono for finding out this reference and kindly telling me about it.
  
  This formula, together with Proposition \ref{prop:MT_s1_0} allows us to derive expressions for the Mordell-Tornheim zeta values with $s_1>0$. Indeed, iterating \eqref{eq:MT_decom} until one of the $s_i$ is cancelled, we obtain 
   \begin{equation*}
    MT(s_1,\cdots,s_r|s) = \sum_{i=1}^r \sum_{p_1=0}^{s_1-1}\cdots\sum_{p_{i-1}=0}^{s_{i-1}-1}\sum_{p_{i+1}=0}^{s_{i+1}-1}\cdots\sum_{p_r=0}^{s_r-1} MT\left(q_1,\cdots,q_r|s+\sum_{j=1}^r p_j\right)
   \end{equation*}
   with in each of the terms in the RHS, $p_i:=s_i$ and $q_j:=s_j-p_j$ for $q\in\{1,\cdots,r\}$. Each of the Mordell-Tornheim zeta values in the RHS are of the type treated by Proposition \ref{prop:MT_s1_0} since $q_i=0$ for each $i$ in the leftmost sum. Thus we obtain
   \begin{align} \label{eq:expression_MT}
    & MT(s_1,\cdots,s_r|s)  =  \\
    \nonumber 
\sum_{i=1}^r \sum_{p_1=0}^{s_1-1}\cdots & \sum_{p_{i-1}=0}^{s_{i-1}-1}\sum_{p_{i+1}=0}^{s_{i+1}-1}\cdots\sum_{p_r=0}^{s_r-1} \zeta_\shuffle\left(\fraks\left(s+\sum_{\substack{j=1}}^r p_j\right)\sqcup \Big(\fraks(s_1-p_1)\shuffle\cdots\shuffle\fraks(s_{r}-p_{r})\Big)   \right)
   \end{align}
   where in each of the terms in the RHS we have set $p_i:=s_i$ and used the convention $\fraks(0):=\emptyset$.
  
  \subsection{Conical zeta values}

Let us start by recalling some classical definition of cones \cite{Fulton93,Ziegler94}
\begin{defn} \label{defn:cones}
\begin{itemize}
 \item Let $v_1,\cdots,v_n$ by $n$ linearly independant nonzero vectors in $\Z^k$. The {\bf cone} associated to these vectors is
 \begin{equation*}
  C=\langle v_1,\cdots,v_n\rangle := \R^*_+v_1+\cdots\R^*_+v_n.
 \end{equation*}
 If furthermore $k=n$, the cone is called {\bf maximal}. We write $\calc$ the set of maximal cones. By convention, the empty set is a maximal cone.
 \item A {\bf decorated cone} is a pair $(C,\vec s)$ with $C=\langle v_1,\cdots,v_n\rangle$ a cone and $\vec s\in\N^n$. We call $\calc_\N$ the set of decorated maximal cones.
 \item  For a cone $C$ (resp. a decorated cone $(C,\vec s)$), write the vectors $v_1,\cdots,v_n$ in the canonical basis $\mathcal{B}=\{e_1,\cdots,e_n\}$ of $\R^n$: $v_i=\sum_{j=1}^na_{ij}e_j$. Then $A_C:=(a_{ij})_{1\leq i,j\leq n}$ is the {\bf representing matrix} (in the basis $\mathcal B$) of the cone $C$ (resp. $(C,\vec s)$).
 \item A cone $C$ (resp. a decorated cone $(C,\vec s)$) is {\bf unimodular} if its representing matrix $A_C$ has only $0$s and $1$s in its entries.
\end{itemize}
\end{defn}
We will only consider here maximal cones, therefore we simply write cones instead of maximal cones. 
\begin{rk}
 The definition above only covers open simplicial rational smooth cones. More general cones, or closed ones, will play no role here thus we do not introduce them. Similarly, in the case of decorated cones, one could have $\vec s\in\C^n$. Notice that a cone is invariant under permutations of the $v_i$s, while a decorated cone in general is only invariant under the simultaneous permutations of the $v_i$s and the components of $\vec s$.
\end{rk}
Conical zeta values (CZVs) were introduced in \cite{GPZ13} as  weighted sums over \ty{integer} points on cones. A generalisation of these objects was introduced before in \cite{Te04}. In \cite{Ze17} a description of CZVs in terms of matrices was given. We adopt here an intermediate definition, where the cone is encoded by a matrix but not the weight. Our definition is rigorously equivalent to the ones in \cite{GPZ13,Te04,Ze17}, but is more suitable for our purpose. 
\begin{defn} \label{defn:useful_stuff_cone}
 \begin{itemize}
  \item Let $\mathcal{B}=\{e_1,\cdots,e_n\}$ be the canonical basis of $\R^n$ and $l:\R^n\longrightarrow\R$ be a linear map defined by $l(v)=\sum_{i=1}^n l_iv_i$, with $v_i$ the coordinates of the vector $v$ in the canonical base: $v=\sum_{i=1}^n v_ie_i$. Then we say that $l$ is (resp.~strictly) {\bf positive (with respect to $\mathcal{B}$)}, and we write $l\geq0$\footnote{the more rigorous notation $l\geq_\mathcal{B}0$ is not necessary since we always work with the canonical basis of $\R^n$.} (resp $l>0$) if $l_i\geq0$ (resp. $l>0$) for any $i\in[n]$.
  
  \item Let $C=\langle v_1,\cdots,v_n\rangle$ (resp. $(C=\langle v_1,\cdots,v_n\rangle,\vec{s}$)) be a maximal cone (resp.~decorated cone). Write the vectors $v_i$s in the canonical basis $\mathcal{B}$: $v_i=\sum_{j=1}^na_{ij}e_j$ and set $l_i:(\R)^n\longrightarrow\R$ the linear maps defined by $l_i(w)=\sum_{j=1}^na_{ij}w_j$. Then  
  $C$ (resp. $(C,\vec s)$) is (resp. strictly) {\bf positive (with respect to $\mathcal{B}$)} if $l_i\geq0$ (resp. $l_i>0$) for any $i\in[n]$.
  
  \item Let $(C=\langle v_1,\cdots,v_n\rangle,\vec{s}=(s_1,\cdots,s_n))$ be a decorated (maximal) strictly positive cone and $l_i:(\R)^n\longrightarrow\R$ the associated linear maps as before. Then the {\bf conical zeta value} associated to $(C,\vec s)$  is
 \begin{equation} \label{eq:def_CZV}
  \zeta(C,\vec s) := \sum_{\vec{n}\in(\N^*)^n}\frac{1}{l_1(\vec{n})^{s_1}\cdots l_n(\vec{n})^{s_n}}
 \end{equation}
 whenever the series converge.
 
 We also denote by $\zeta$ the map which to a cone $(C,\vec s)$ associates the CZV $\zeta(C,\vec s)$ when it exists.
 \end{itemize}
\end{defn}
  Notice that the object that we call CZVs here were named ``Shintani zeta values'' in \cite{GPZ13}, following \cite{Mat03}. We use here instead what seems to be now the standard name, namely CZVs.
  \begin{rk}
 Notice that any unimodular cone is strictly positive with respect to the canonical basis. This justifies that we will not require our cones to be strictly positive since they will be unimodular.
\end{rk}
There are many important open questions concerning CZVs. 
An important one is the existence of linear relations with rational coefficients between CZVs. It was shown in \cite{GPZ13} that they obey a family of relations given by double subdivisions of cones which conjecturally generate all linear relations with rational coefficients between by CZVs. Another question is the number-theoretic content of CZVs. It was shown in \cite{Te04} that CZVs are evaluation of polylogarithms at $N$-th roots of unity. A conjecture by Dupont, refined by Panzer, written in \cite[Conjecture 2]{Ze17} and sometimes called the Dupont-Panzer-Zerbini conjecture states that for a cone $C$, $N$ is the least common multiplier of the minors of $A_C$. In this paper, we answer this second question in the case of tree-like cones.
  
\subsection{From trees to cones}

One easily sees that TZVs are CZVs. Furthermore, one can relate the tree underlining a TZVs to a specific cone. For this, one needs \ty{one} new definition.
\begin{defn}
 Let $F\in\calf$ be a  rooted forest with vertices $V(F)=\{1,\cdots,N\}$. The {\bf path matrix of $F$} $A_F=(a_{ij})_{i,j=1,\cdots,N}$ is the $N\times N$ matrix defined by
\begin{align*}
 a_{ij} = \begin{cases}
           & 1 \quad\text{if }i\preceq j\\
           & 0 \quad\text{otherwise}.
          \end{cases}
\end{align*}
Furthermore, let us set 
\begin{equation*}
 v_i(F):=\sum_{j=1}^Na_{ij}e_j
\end{equation*}
(with $\{e_1,\cdots,e_N\}$ the canonical basis of $\R^N$) the {\bf path vectors of $F$}.
\end{defn}
The path vectors of a rooted forest define a (maximal) cone.
\begin{lemma}
 We have a map from rooted forest to cones:
\begin{align*}
 \Phi:\calf~&\longrightarrow~ \calc \\
 F~&\longmapsto \langle v_1(F),\cdots,v_N(F)\rangle.
\end{align*}
\end{lemma}
\begin{proof}
%  Recall that all our cones are maximal cones. 
 We need to prove that, for any $F\in\calf$, the vectors $v_i(F)$ are linearly independant. We prove this by induction on $N=|V(F)|$. If $N=0$, then $F=\emptyset$ and $\Phi(\emptyset)=\emptyset$ is a cone. The case $N=1$ also trivially holds.
 
 Assume the result holds for all forests with $k\leq N$ vertices and let $F$ be a forest with $N+1$ leaves. Without loss of generality we can assume $V(F)=[N+1]$. We have two cases to consider: First, if $F=F_1F_2$ with $F_1$ and $F_2$ non empty, the result holds by the induction hypothesis used on $F_1$ and $F_2$ and the fact that the $v_i(F_1)$ and the $v_j(F_2)$ belong to two orthogonal subspaces of $\R^{N+1}$. 
 
 Second, if $F=T=B_+(\tilde F)$, we can assume without loss of generality that $N+1$ is the root of $T$. Then we have $v_{N+1}(T)=\sum_{k=1}^{N+1}e_k$ and $v_i(T)=v_i(\tilde F)$ for any $i\in[N]$. Thus 
 \begin{equation*}
  \sum_{k=1}^{N+1}\lambda_k v_k(T) = 0 ~\Longleftrightarrow~\lambda_{N+1}=0~\wedge~\sum_{k=1}^{N}\lambda_k v_k(\tilde F)=0
 \end{equation*}
 since $v_{N+1}(T)$ is the only vector with a non-zero $e_{N+1}$ component. The result then holds from the induction hypothesis used on $\tilde F$, which concludes the proof.
\end{proof}
Notice that we do not claim this map from forests to cones to be unique, nor new. It simply turns out that it is the map that we can use to relate TZVs and CZVs. To achieve this, we lift the map $\Phi$ to a map $\overline{\Phi}$ acting on decorated forests. 

 Let $\overline{\Phi}:\calf_{\N}\longrightarrow\calc_\N$ be the map defined by
 $\overline{\Phi}(F,d_F):=(\Phi(F),\vec{s}_F)$
 for any decorated forest $(F,d_F)\in\calf_{\N}$ with vertices $V(F)=\{1,\cdots,N\}$ and where we have set $\vec s_F:=(d_F(1),\cdots,d_F(N))$.
 
 The maps $\Phi$ and $\overline{\Phi}$  are not surjective. This justifies the following definition:
\begin{defn} \label{defn:tree_like_cone}
 A cone (resp. a decorated cone) is said to be a {\bf tree-like cone} (resp. a {\bf decorated tree-like cone}) when it lies in the image of $\Phi:\calf\longrightarrow\calc$ (resp. $\overline{\Phi}:\calf_{\N}\longrightarrow\calc_\N$). We write $\calct$ the set of tree-like cones.
 
 Furthermore, if a decorated cone lies in $\overline{\Phi}(\calf_{\N^*}^{\rm conv})$ it is called a {\bf convergent decorated tree-like cone}. We write $\calct_\N$ the set of decorated tree-like cone and $\calct_\N^{\rm conv}$ the set of convergent decorated tree-like cone.
\end{defn}
Let us recall before (Definition \ref{def:convergent_forests}) that a $\N^*$-decorated forest is convergent if the decorations of each of its roots are greater or equal to 2.
The following simple properties of the map $\overline{\Phi}$ are a key result.
\begin{prop} \label{prop:forest_cones}
 For any non-empty convergent forest $(F,d_F)\in\calf_{\N^*}^{\rm conv}$, the conical zeta value $\zeta(\overline{\Phi}(F,d_F))$ is convergent and 
 \begin{equation*}
  \zeta(\overline{\Phi}(F,d_F)) = \zeta^t(F,d_F).
 \end{equation*}
\end{prop}
\begin{proof}
 Let $(F,d_F)$ be any convergent forest. Up to relabelling, we can identify its vertices set $V(F)$ with $[N]$ for some $N\in\N$: $V(F)=\{1,\cdots,N\}$. First let us  observe that, for any $\vec n=(n_1,\cdots,n_N)\in\N^N$ and $i\in V(F)$ we have 
 \begin{equation*}
  l_i(F)(\vec n) = \sum_{\substack{j\in[N]\\j\succeq i}}n_j
 \end{equation*}
 by definition of $l_i(F)$, and where $\succeq$ is the partial order of the set vertices $V(F)$ of $(F,d_F)$. Then Equation \eqref{eq:TZV} applied to the convergent rooted forest $(F,d_F)$ gives
 \begin{equation*}
  \zeta^t(F,d_F) = \sum_{\vec n\in\N^N}\prod_{i=1}^N\left(\sum_{\substack{j\in[N]\\j\succeq i}}n_j\right)^{-d_F(i)} = \sum_{\vec n\in\N^N} \frac{1}{l_1(F)(\vec n)^{d_F(1)}\cdots l_N(F)(\vec n)^{d_F(N)}} = \zeta(\overline{\Phi}(F),\vec s_F)
 \end{equation*}
 as claimed in the Proposition. The convergence of $\zeta(\overline{\Phi}(F),\vec s_F)$ then follows from the first point of Proposition \ref{prop:crucial}.
\end{proof}
\begin{rk}
 This implies in particular that shuffle AZVs are CZVs and thus also the more general Shintani zetas. As such, if we take the decorations of $F$ to be complex parameters, the function $\vec s\mapsto\zeta^t(F)$ admits a meromorphic continuation to $\C^{|V(F)|}$ \cite{Ma03,Lo21}.
\end{rk}
We now readily obtain our next important result.
\begin{theo} \label{thm:tree_CZVs_MZVs}
 For any convergent decorated tree-like cone $(C,\vec s)=\overline{\Phi}(F,d_F)$, the associated conical zeta values $\zeta(C,\vec s)$ is a linear combinations of MZVs with rational coefficients given by
 \begin{equation*}
  \zeta(C,\vec s) = (\zeta_\shuffle\circ fl_0\circ\fraks^T)(F,d_F).
 \end{equation*}
\end{theo}
\begin{proof}
 For any convergent decorated cone $(C,\vec s)=\overline{\Phi}(F,d_F)$ the result follows from Proposition \ref{prop:forest_cones} and Theorem \ref{thm:tree_zeta_MZVs} applied to $\zeta^t(F,d_F)$.
\end{proof}
Theorems \ref{thm:tree_zeta_MZVs} and \ref{thm:tree_CZVs_MZVs}, together with Proposition \ref{prop:forest_cones} and our previous results on AZVs can be summarised as the commutativity of Figure \ref{fig:MZV_AZV} below (where the CZVs, AZVs and MZVs maps are all written $\zeta$).
\begin{figure}[h!] 
  		\begin{center}
  			\begin{tikzpicture}[->,>=stealth',shorten >=1pt,auto,node distance=3cm,thick]
  			\tikzstyle{arrow}=[->]
  			
  			\node (1) {$\calF_{\N^*}^{\rm conv}$};
  			\node (2) [right of=1] {$\calF_{\{x,y\}}^{\rm conv}$};
  			\node (3) [right of=2] {$\calW_{\{x,y\}}^{\rm conv}$};
  			\node (4) [below of=1] {$\calct_{\N}^{\rm conv}$};
  			\node (5) [right of=4] {$\R$};
%   			\node (8) [right of=4] {$\R$};

  			\path
  			(1) edge node [above] {$\fraks^T$} (2)
  			(2) edge node [above] {$fl_0$} (3)
  			(1) edge node [left] {$\overline{\Phi}$} (4)
  			(4) edge node [above] {$\zeta$} (5)
  			(1) edge node [above right] {$\zeta^t$} (5)
  			(2) edge node [right] {$\zeta_\shuffle^T$} (5)
  			(3) edge node [below right] {$\zeta_\shuffle$} (5);
  			\end{tikzpicture}
  			\caption{CZVs, AZVs, MZVs and TZVs.}\label{fig:MZV_AZV}
  		\end{center}
  	\end{figure} \\
Theorem \ref{thm:tree_CZVs_MZVs} answers one of the important question about CZVs for the convergent tree-like cones; namely that they can written as rational sums of MZVs. Therefore it is useful to be able to characterise which cones are (convergent) tree-like. This is the subject of the next Subsection.
\begin{rk}
 In \cite[Lemma 1]{Ze16}, another sufficient condition was given for a unimodular CVZs to be a linear combinations of MZVs with rational coefficients. This condition, called C1 in \cite{Ze17}, the rows of the representing matrix $A_C$ can be permuted such that the 1s of $A_C$ are consecutive in each column. It is easy to see with a counter-example that this condition has no relation with being tree-like.
 
 The dual condition, let us call it C2, is that the columns of $A_C$ can be permuted so that the 1s of $A_C$ are consecutive in each column. To the best of the author's knowledge, it is still a conjecture (proposed by F. Zerbini in \cite{Ze17}) that C2 is another sufficient (but not necessary) condition for unimodular CZVs to be a linear combinations of MZVs with rational coefficients. It is easy to show that being tree-like implies the C2 conditon, but that the converse does not hold. Thus, Theorem \ref{thm:tree_CZVs_MZVs} is an indication that Zerbini's conjecture holds.
\end{rk}

\subsection{Characterisation of tree-like cones}

We want to relate the R.H.S. of Equation \eqref{eq:def_CZV} and \eqref{eq:TZV}. In particular, each factor in the denominator of a CZV correspond to a vector generating the underlining cone, while for a TZV, each such term come from a vertex of the underlining rooted forest. So we need to associate to each vertex of a rooted forest a vector in $\R^{|V(F)|}$. Furthermore, observe that in Equation \eqref{eq:TZV} if $v_1\preceq v_2$, the term in the series associated to $v_1$ is
\begin{equation*}
 l_{v_1}(\vec n) = \sum_{\substack{v'\succeq v_1\\v'\not\succeq v_2}}n_{v'} + \sum_{v'\succeq v_2}n_{v'} = \sum_{\substack{v'\succeq v_1\\v'\not\succeq v_2}}n_{v'} + l_{v_2}(\vec n).
\end{equation*}
Therefore the partial order of a rooted forest $F$ should be transmitted in some sense to a partial order on the vector generating the tree-like cone $\Phi(F)$. This justifies the following definition.
\begin{defn}
 For any $n$-dimensional cone $C$ (resp. decorated cone $(C,\vec s)$), let $\preceq_C$ be the relation on $[n]$ defined by
 \begin{equation*}
  i\preceq_C j ~:\Longleftrightarrow~ l_i-l_j\geq0
 \end{equation*}
 with the linear maps $l_i$ and the notion of positive linear maps of Definition \ref{defn:useful_stuff_cone}. As before, we write $\succeq_C$ for the inverse relation.
\end{defn}
As expected, this defines a partial order.
\begin{lemma} \label{lem:trivial_poset}
 For any $n$-dimensional cone $C$ (resp. decorated cone $(C,\vec s)$), $([n],\preceq_C)$ is a poset.
\end{lemma}
\begin{proof}
 Reflexivity and transitivity are trivial. Anti-symmetry follows from the fact that $C$ is a maximal cone, so two different lines of the representing matrix $A_C$ have to be different. Thus $l_i=l_j$ implies $i=j$ for any $i$ and $j$ in $[n]$.
\end{proof}
Thus we obtain a map 
\begin{align} \label{defn_Psi}
 \Psi:\mathcal{C} & \longrightarrow \mathcal{P}^{\rm fin} \\
 C & \longmapsto ([n],\preceq_C) \nonumber
\end{align}
where $\mathcal{P}^{\rm fin}$ is the set of finite posets. We lift $\Psi$ to a map $\overline\Psi :\calc_\N\longmapsto\mathcal{P}^{\rm fin}_\N$ from decorated cones to decorated posets\footnote{where a decorated poset is a pair $(\calP,d)$ with $d:\calP\longrightarrow\Omega$ a map, and $\Omega$ is the decoration set. We will only encounter posets decorated by positive integers in this work, and not after this discussion, therefore it did not seem important enough to write a formal definition for this natural object.} by setting $\overline\Psi(C,\vec s):=(\Psi(C),d_C)$, with $d_C:[n]\mapsto\N$ defined by $d_C(i):=s_i$.

Since not every cone is a tree-like cone, and more generally, not every conical sum is indexed by a poset, we need a compatibility condition on the cone to ensure that its associated conical sum respects the poset structure associated to the cone. One finds out the right condition by observing that in \eqref{eq:TZV}, if $v'\succeq v$, then the term coming from $v$ in Equation \eqref{eq:TZV} contains the summation variable $n_{v'}$ associated to the vertex $v'$.
\begin{defn}
 A $n$-dimensional cone $C$ (resp. decorated cone $(C,\vec s)$) with representing matrix $A_C=(a_{ij})_{i,j=1,\cdots,n}$ is 
 {\bf poset compatible} if, for any $i,j\in[n]$ we have
 \begin{equation*}
  a_{ij} \neq0~\Longrightarrow~i\preceq_C j.
 \end{equation*}
\end{defn}
So the CZV associated to any $n$-dimensional poset compatible cone is a conical zeta value in which the sum is given by the partial order of the poset $([n],\preceq_C)$. To check that this poset is a rooted forest, we need to introduce one more object. 
\begin{defn}
 For any $n$-dimensional cone $C$ (resp. decorated cone $(C,\vec s)$), its {\bf second representing matrix} $B_C:=(b_{ij})_{i,j=1,..,n}$ is the incidence matrix of the Hasse diagram of $([n],\preceq_C)$.
 
 In other words, $b_{ij}=1$ when $j$ is a direct successor of $i$ and $0$ otherwise:
 \begin{equation*}
  b_{ij}=1~\Longleftrightarrow~i\preceq_C j~\wedge~(\forall k\in[n],~i\preceq_Ck\preceq_Cj\Rightarrow k\in\{i,j\}).
 \end{equation*}
\end{defn}
We then have 
\begin{lemma} \label{lem:carac_tree_cone}
 For a $n$-dimensional cone $C$ (resp. decorated cone $(C,\vec s)$), the poset $\Psi(C)=([n],\preceq_C)$ is a rooted forest if, and only if, its second representing matrix has at most one non-zero component per column.
\end{lemma}
\begin{proof}
 An oriented graph is a rooted forest if, and only if it has no oriented cycle, no non-oriented cycle\footnote{i.e. no circuits} and each of its connected components has exactly one minimal element.
 
 The first point is guaranteed by Lemma \ref{lem:trivial_poset} since if there is an oriented cycle, then we have $i,j\in[n]$ such that $i\neq j$ while $i\succeq j$ and $j\succeq i$. Thus in this case, $([n],\preceq_C)$ would not be a poset.
 
 The second and third points are equivalent to asking that each vertex has at most one direct ancestor for the relation $\preceq_C$. Since $j$ is a direct ancestor of $i$ if, and only if, $b_{ij}=1$, we have that $j$ has at most one direct ancestor if, and only if, it exists at most one $i\in[n]$ such that $b_{ij}=1$, thus that the $j$-th column of $B_C$ has at most one non-zero entry. Since this must hold for all $j\in[n]$, we have the result
\end{proof}
Lemma \ref{lem:carac_tree_cone} actually gives a characterisation of tree-like cones. Stating all the results together we have
\begin{theo} \label{thm:carac_tree_cone}
 \begin{enumerate}
  \item A unimodular cone $C$ (resp. a decorated unimodular cone $(C,\vec s)$) is a tree-like cone if, and only if, it is poset compatible and its second representing matrix has at most one non-zero entry in each column.
  \item The map $\Psi$ (resp. $\overline{\Psi}$) restricted to $\calct$ (resp. $\calct_\N$) is the inverse of the map $\Phi$ (resp. $\overline{\Phi}$).
  \item For any decorated unimodular cone $\calc=(C,\vec s)$ such that its second representing matrix has at most one non-zero component per column and $\overline\Psi(\calc)$ is a convergent forest, $\zeta(C,\vec s)$ converges and is a linear combination of MZVs of weights $||\vec s||:=s_1\cdots+s_n$ with rational coefficients. They evaluate as:
 \begin{equation*}
  \zeta(C,\vec s) = (\zeta_\shuffle\circ fl_0\circ \fraks^T\circ\overline\Psi)(C,\vec s).
 \end{equation*}
 \end{enumerate}
\end{theo}
\begin{proof}
 \begin{itemize}
  \item We start by showing the second point of the Theorem. 
  
  First, notice that the maps $\Phi$ and $\overline{\Phi}$ are injective since $F$ is actually an isomorphism class of rooted forests.

 Now, let $C$ be any tree-like cone and $F$ its preimage under $\Phi$. Since $F$ is an isomorphism class, we can assume without loss of generality that $V(F)=[n]$. Then by construction of $\Phi$ and $\preceq_C$, the latter is the partial order relation on $[n]=V(F)$ that defines the forest structure of $F$. Thus $B_C$ is the adjacency matrix of $F$ and since $\Psi(C)$ is by definition the poset whose adjacency matrix is $B_C$ we obtain $\Psi|_{\calct}=\Phi^{-1}$ as claimed.
 
 The same argument carries over to the decorated cases by definition of the lift.
  \item We now prove the first point of the Theorem. Lemma \ref{lem:carac_tree_cone} directly implies that if $C$ is a tree-like cone, then $\Psi(C)$ is a rooted forest. Thus, again by Lemma \ref{lem:carac_tree_cone}, $C$ is poset compatible and $B_C$ has at most one non-zero entry in each column. We thus have that each tree-like cone satisfies the hypothesis of Lemma \ref{lem:carac_tree_cone}.
  
  The sole thing left to be shown is the fact that, if a cone obeys the hypothesis of Lemma \ref{lem:carac_tree_cone}, namely that its second representing matrix has at most one non-zero component per column, then it is a tree-like cone. By Lemma \ref{lem:carac_tree_cone} for such a cone $C$, $\Psi(C)\in\calf$. Using the same argument as in the first point of this proof, we have that $\Phi(\Psi(C))=C$. Thus $C\in\calct$ by definition of tree-like cones.
  
  The same arguments hold with the lifted maps $\overline\Phi$ and $\overline\Psi$ by definition of the lifts.
  
  \item The third point is a direct reformulation of Theorem \ref{thm:tree_CZVs_MZVs} together with the results of Lemma \ref{lem:carac_tree_cone} and the first two points of the Theorem.
 \end{itemize}

\end{proof}

\subsection{Computations of CZVs} \label{subsec:computations_CZVs}

We have now tools to compute CZVs whose underlining cones are tree-like cones; \emph{and} to detect which cones are tree-like. The algorithm we obtain to compute $\zeta(C,\vec s)$ is the following:
\begin{enumerate}
 \item check that $C$ is poset compatible,
 \item compute the second representing matrix of $C$ and to check that it has at most one $1$ per column,
 \item apply the branched binarisation map $\fraks^T$ to $\overline{\Psi}(C,\vec s)$,
 \item apply the flattening map $fl_0$ to $\fraks^T(\overline{\Psi}(C,\vec s))$,
 \item apply the shuffle MZV map $\zeta_\shuffle$ to the resulting linear combination of words.
\end{enumerate}
Notice that in some case, it is simpler to directly apply the $\yew$-flattening with Theorem \ref{thm:yew_formula} to make the computation. In practice, one can replace steps 3. and 4. with the application of the $\yew$-flattening map $fl_\yew$; and step 5. with the application of the stuffle MZV map $\zeta_\stuffle$ 

We illustrate this procedure by computing some CZVs. Before this, let us point out that this algorithm produces the same result than the procedure of \cite[Definition 3.16]{OSY21}. Indeed, this procedure is equivalent to our flattening $fl_0$ but built from leaves rather than from the root (and excluding forests). Furthermore the harvestable trees of \cite{OSY21} are essentially our convergent trees in $\calf_{\{x,y\}}$. Then \cite[Theorem 3.17]{OSY21} is the finite sum version of Theorem \ref{thm:main_result_shuffle}.
\begin{example} \label{ex:simple_cone}
 The simplest non-trivial example of CZV that we can compute with our method is:
 \begin{equation*}
  \zeta(C_1,\vec s_1):=\sum_{p,q,r\geq1}\frac{1}{(p+q+r)^2qr}.
 \end{equation*}
 The representing and second representing matrices to be respectively:
 \begin{equation*}
  A_{C_1}= \begin{pmatrix}
            1 & 1 & 1 \\
            0 & 1 & 0 \\
            0 & 0 & 1
           \end{pmatrix},\quad 
  B_{C_1}= \begin{pmatrix}
            0 & 1 & 1 \\
            0 & 0 & 0 \\
            0 & 0 & 0
           \end{pmatrix}.
 \end{equation*}
 One easily checks from $A_{C_1}$ that $C_1$ is poset compatible and clearly $B_{C_1}$ satisfies the hypothesis of Lemma \ref{lem:carac_tree_cone}. We further have
 \begin{equation*}
  \overline\Psi(C_1,(2,1,1)) = \tdtroisun{$2$}{$~1$}{$1$}
 \end{equation*}
 which is convergent. Applying the branched binarisation map $\fraks^T$ and the MZVs map $\zeta_\shuffle$ gives the result
 \begin{equation*}
  \sum_{p,q,r\geq1}\frac{1}{(p+q+r)^2qr} = 2\zeta_\stuffle(2,1,1).
 \end{equation*}
\end{example}
The exact same computation can be used, up to the flattening, for higher powers in the denominator. However, in this case, it is more efficient to use the $\yew$-flattening $fl_\yew$.
\begin{prop}
 For any $\vec s=(n,m,l)\in(\N^*)^3$ with $n\geq2$ we have that the CZV
 \begin{equation*}
  \zeta(C_1,\vec s)=\sum_{p,q,r\geq1}\frac{1}{(p+q+r)^nq^mr^l}
 \end{equation*}
  is a linear combinations of MZVs with rational coefficients given by binomial coefficients:
  \begin{equation*}
   \sum_{p,q,r\geq1}\frac{1}{(p+q+r)^nq^mr^l} = \sum_{i=0}^{l-1} \binom{m-1+i}{i} \zeta_\stuffle (n, m+i, l-i) +  \sum_{j=0}^{m-1} \binom{l-1+j}{j} \zeta_\stuffle (n, l+j, m-j).
  \end{equation*}
\end{prop}
\begin{proof}
 First observe that the cone underlying this conical sum is again $C_1$. So it is a linear combination of MZVs with rational coefficients from the same argument than the one of the previous Example \ref{ex:simple_cone}. Again as in the Example above we have
 \begin{equation*}
  \overline\Psi(C_1,(n,m,l)) = \tdtroisun{$n$}{$~l$}{$m$}.
 \end{equation*}
 For the computation, from Proposition \ref{prop:forest_cones} and Corollary \ref{coro:TZVs_stuffle_MZVs} we deduce
$$ \zeta(C_1,(n,m,l) ) = \zeta_\stuffle \circ fl_\yew \left( \tdtroisun{$n$}{$~l$}{$m$} \right). $$
Using Theorem \ref{thm:yew_formula} we obtain
$$ \tdun{$m$} \yew \tdun{$l$} = \sum_{i=0}^{l-1} \binom{m-1+i}{i} \tddeux{m+i}{l-i}
+  \sum_{j=0}^{m-1} \binom{l-1+j}{j}\tddeux{l+j}{m-j}
$$
The result then follows from the definition of $fl_\yew$ (Definition \ref{defn:flattening_yew}).
\end{proof}
\begin{rk}
 The two computations above are a good illustration of an important point concerning the number-theoretic content of TZVs. It depends only of the weight of $||\vec s||:=s_1+\cdots+s_n$. Providing the standard conjectures on MZVs hold, $\zeta^t(\overline\Psi(C,\vec s))$ is always a linear combination of MZVs of weight exactly $||\vec s||$ with rational coefficients, for any decorated cone $(C,\vec s)\in\mathcal{CT}_{N^*}$ such that $\overline\Psi(C,\vec s)$ is a convergent forest. This is not the case for general CZVs, where some cancellations might happen and some terms in the expression of the CZV might be of lower weight than expected (or, of course, not be rational combination of MZVs). I am very thankful to E. Panzer who pointed out this fact to me and gave me an example.
\end{rk}
The two previous computations were special cases of Mordell-Tornheim zetas, but we can also compute more general conical sums. We present two more computations without all the intermediate steps.
\begin{example} \label{ex:less_simple}
 \begin{equation*}
  \zeta(C_2,\vec s_2) := \sum_{p,q,r,s,t\geq1}\frac{1}{(p+q+r+s+t)^4(q+t)^2rst}
 \end{equation*}
 We have 
 \begin{equation*}
  A_{C_2} = \begin{pmatrix}
             1 & 1 & 1 & 1 & 1 \\
             0 & 1 & 0 & 0 & 1 \\
             0 & 0 & 1 & 0 & 0 \\
             0 & 0 & 0 & 1 & 0 \\
             0 & 0 & 0 & 0 & 1 
            \end{pmatrix}, \quad 
  B_{C_2} = \begin{pmatrix}
             0 & 1 & 1 & 1 & 0 \\
             0 & 0 & 0 & 0 & 1 \\
             0 & 0 & 0 & 0 & 0 \\
             0 & 0 & 0 & 0 & 0 \\
             0 & 0 & 0 & 0 & 0
            \end{pmatrix}.
 \end{equation*}
 One easily checks that $C_2$ satisfies the hypotheses of Theorem \ref{thm:carac_tree_cone}. Applying the algorithm above, one readily finds
 \begin{equation*}
  \sum_{p,q,r,s,t\geq1}\frac{1}{(p+q+r+s+t)^4(q+t)^2rst} = 2\zeta_\stuffle(4,1,1,2,1) + 6\zeta_\stuffle(4,1,2,1,1) + 12\zeta_\stuffle(4,2,1,1,1).
 \end{equation*}
\end{example}
\begin{rk}
 Examples \ref{ex:simple_cone} and \ref{ex:less_simple} were kindly checked by E. Panzer using his HyperInt package. This package uses an integral representation of CZVs that differs from the one of Theorem \ref{thm:integral_sum}. The HyperInt package and the mathematics it relies on is presented in \cite{panzer2015algorithms}.
\end{rk}
The next (and last) example required more complicated computations that we will still not detail. It was kindly checked by Masataka Ono to be coherent with their finite MZVs approach \ty{developed} in \cite{OSY21}.
\begin{example}
 \begin{equation*}
  \zeta(C_3,\vec s_3) := \sum_{n_1,\cdots,n_7\geq1}\frac{1}{(n_1+\cdots+n_7)^5(n_2+\cdots+n_7)^2n_3(n_4)^2(n_5+n_6+n_7)^2n_6n_7}.
 \end{equation*}
 We have
 \begin{equation*}
  A_{C_3}=\begin{pmatrix}
           1 & 1 & 1 & 1 & 1 & 1 & 1 \\
           0 & 1 & 1 & 1 & 1 & 1 & 1 \\
           0 & 0 & 1 & 0 & 0 & 0 & 0 \\
           0 & 0 & 0 & 1 & 0 & 0 & 0 \\
           0 & 0 & 0 & 0 & 1 & 1 & 1 \\
           0 & 0 & 0 & 0 & 0 & 1 & 0 \\
           0 & 0 & 0 & 0 & 0 & 0 & 1
          \end{pmatrix},\quad 
  B_{C_3}=\begin{pmatrix}
           0 & 1 & 0 & 0 & 0 & 0 & 0 \\
           0 & 0 & 1 & 1 & 1 & 0 & 0 \\
           0 & 0 & 0 & 0 & 0 & 0 & 0 \\
           0 & 0 & 0 & 0 & 0 & 0 & 0 \\
           0 & 0 & 0 & 0 & 0 & 1 & 1 \\
           0 & 0 & 0 & 0 & 0 & 0 & 0 \\
           0 & 0 & 0 & 0 & 0 & 0 & 0 
          \end{pmatrix}
 \end{equation*}
 and the hypothesis of the third point of Theorem \ref{thm:carac_tree_cone} are satisfied. After some more computations we obtain
 \begin{align*}
  & \sum_{n_1,\cdots,n_7\geq1}\frac{1}{(n_1+\cdots+n_7)^5(n_2+\cdots+n_7)^2n_3(n_4)^2(n_5+n_6+n_7)^2n_6n_7} \\
  =~ & 8\zeta_\stuffle(5,2,1,2,2,1,1) + 16\zeta_\stuffle(5,2,1,3,1,1,1) + 2\zeta_\stuffle(5,2,1,2,1,1,2) + 4\zeta_\stuffle(5,2,1,2,1,2,1) \\
  +~& 48\zeta_\stuffle(5,2,2,2,1,1,1) 
  + 28\zeta_\stuffle(5,2,2,1,2,1,1) + 8\zeta_\stuffle(5,2,2,1,1,1,2)   + 16\zeta_\stuffle(5,2,2,1,1,2,1) \\
  +~& 40\zeta_\stuffle(5,2,3,1,1,1,1) .
 \end{align*}
\end{example}

\section{Perspectives}

As said in the introduction of this chapter, there are still quite a few important conjectures about MZVs and their generalisations. In the discussions above, some \ty{progress has} been made towards some of these conjectures (e.g. the \ty{characterisation} of which CZVs can be written as MZVs with rational coefficients), but in general it was not our objective to solve them. However, honesty compels me to say that it was the hope of making some progress towards them that started this research program. To be more specific, one could expect that a generalisation of standard conjectures about MZVs could hold for TZVs/AZVs, and that they would be easier to prove than the MZVs counterparts. And being even more optimistic, one could also expect that proving such a conjecture on rooted forests would imply the standard conjectures on MZVs.

This is a very far-\ty{fetched} leap, and I do not claim that I expect to tackle it head on in the near future. In particular because much more amenable open questions were raised in this chapter, some of which \ty{needing} to be answered before we look at the most important ones. For example, some enumerative combinatorics questions (e.g., what is the image of the associator of the shuffle product of rooted forests?) may be of importance. Another different line of research is to explore more in depth the link between TZVs and AZVs. In particular, the evaluation of CZVs (either via AZVs or by using the $\yew$-product) could be worth exploring, and also to encode in some formal computing language.

Another line of research, closer to my current interests, is to further study the shuffle product of rooted forests. The shuffle product as well as rooted forests have many interesting universal properties, does the shuffle of rooted forests unify them in some sense? Another natural question is whether or not the shuffle of rooted forests admits a coproduct that endows the \ty{algebra} of rooted forests with a new Hopf algebra structure? Could one give a description of its dual coproduct? These questions are currently being investigated with Douglas Modesto as part of \ty{his} PhD project.

I would also like to point out that another version of this shuffle product of rooted forests adapted to binary trees seems to have applications in data sciences. More specifically, ascending hierachical classification algorithms order a data set into a binary tree. 
%Such algorithms are routinely used as a preliminary step to clusterisation of a data set. 
With L. Sautot and L. Journaux, we wrote and implemented in R an algorithm to add one data to such a tree without having to recompute the whole tree. A natural question is how to merge to data sets and not only add new data one by one. Such a merging could be implemented via a shuffle of binary trees. We hope to tackle this question in the near future.

Another \ty{interesting} question, which might be the most important one in this domain, is to find a new generalisation of the stuffle product such that TZVs form an algebra morphism for this product, and such that a suitable generalisation of Hoffman's relation is satisfied. Such a generalisation would be of the form
$$\fraks^T(\tdun{1}\star F)-\tdun{y}\shuffle^T\fraks^T(F)\in\Ker(\zeta_\shuffle^T)$$
for any convergent forest $F\in\calF^{\rm conv}_{\N^*}$, where $\star$ is this new generalisation of the stuffle. I have already looked at quite a few ways to tackle this question. Let us list some, in no particular order:
\begin{itemize}
 \item Combinatorial approach. Looking for combinatorial structures that could underline the construction of TZVs in the same way that Rota-Baxter operators underline the construction of AZVs have not \ty{led} to anything.
 \item Brute force approach. Trying to solve the desired equations by writing the product of rooted forests as a sum over all reasonable forests and solving to determine their coefficients. The proliferation of rooted forests made the number of \ty{solutions} too high to hope to find a pattern.
 \item Universal property of quasi-symmetric functions. Interestingly, one could also define MZVs (as iterated integrals) using a universal property of quasi-symmetric functions\footnote{I realised that after discussions with Yannic Vargas, who I warmly thank to introducing me to this fascinating subject.}. This approach also gives easily that MZVs form an algebra morphism for the stuffle product. However, this approach does not seem to be a good one to build TZVs.
 \item Mould calculus. There is a mould that Ecalle defines that looks somewhat similar to TZVs. However, no property of this mould seems to be known. Its arborification also does not seem to have been studied.
 \item Cones decomposition.  In \cite{GPZ13} the authors showed that a decomposition of cones leads to relations between CZVs that generalise the stuffle relations between MZVs. This TZVs are CZVs that sit above specific cones (tree-like cones), it would be enough to show that an tree-like cone admits a (non-trivial) decomposition in tree-like cones.
\end{itemize}
This last approach seems to be by far the most promising one. Various duties have prevented me from exploring it further, but I intend to offer a PhD position with this question as the main focus of the student's investigation.
\begin{rk}
 After the redaction of this chapter, Pierre Catoire, from the Université du Littoral in Calais, came to Mulhouse to discuss related topics. From his short stay, an interesting concept emerged. We realised we could use the universal property of the dendriform and tridenform algebras of (planar) rooted trees (see \cite{ronco2002eulerian,loday2007algebra,catoire2023tridendriform}) to define respectively new shuffle and stuffle arborified MZVs. These new objects obey other versions of the shuffle and stuffle products and are not both completely different to the ones presented above. However, it is at this moment still not clear what the (tri)dendriform version of the branched binarisation map \eqref{eq:binarisation_map} could be. We hope to be able to see in the coming months if these new objects are valid candidates for generalisation of MZVs to rooted forests; planar or not.
\end{rk}

% An intersting stuffle for TZVs, such that Hoffman's relation are satisfied. Another proposal was made ($\stuffle^t$), but the most natural version of Hoffman's relation don't hold for this product. Quite a few ways have been explored toward this goal (combinatorics, brute force, using the universal properties of quasi-symmetric functions, mould calculus). Most promising way: the decomposition of cones.

% XXXXXX: CHECK THAT IN DIDN'T INCLUDE $\stuffle^t$.

% Implementation of the $\yew$ product in a formal computing program?
% In particular, it seems tractable for computers. The implementation of these results in formal language computation is left for future research. 

% % % % %  COMMENTER CE QUI EST CI-DESSOUS.
% 
% \bibliographystyle{unsrt}
%  \addcontentsline{toc}{section}{Bibliography}
% \bibliography{HDR_biblio}
%  
% \end{document}

%% file: LOC.tex
% 
% % 
% \documentclass[11pt,twoside,a4paper]{book}
% % % % \documentclass[11pt,twoside,a4paper]{article}
% 
% \usepackage{HDR}
% 
% %% SI ON LAISSE LES DEUX LIGNES SUIVANTES DANS LE STY, ELLES NE MARCHENT PAS... (probablement à cause de \input).
% \input{xy}
% \xyoption{all}
% 
% 
% \begin{document}
% 
% % % 
% % % % % COMMENTER CE QUI EST AVANT CA ET ENDDOCUMENT POUR COMPILER HDR.tex.

% % % % \chapter{Locality structures}% and multivariate renormalisation}
\chapter{Locality structures and multivariable renormalisation}
\label{chap:loc}

\section*{Introduction}

\addcontentsline{toc}{section}{Introduction}

\subsection*{Locality from physics to mathematics}

\addcontentsline{toc}{subsection}{Locality from physics to mathematics}

Locality is one of the central concept of modern and not-so-modern physics. \emph{Very} broadly speaking, it can be expressed by the following idea
\begin{quote}
 \emph{If two events are far away from each other, the outcome of one cannot influence the outcome of the other.}
\end{quote}
Of course, nearly every word of this sentence can mean something different for the various domain of Physics (or even other sciences) one is working with. It is not the purpose of this introduction to list them.

Since we expect nature's language to be mathematics, one needs to encode this locality principle, however it might manifest itself, into the mathematics used to describe reality. In (perturbative) Quantum Field Theory, QFT for short, this is typically encoded into the requirement that the Feynman rules (already \ty{mentioned} in Chapter \ref{chap:PROP}) are an algebra morphism for the concatenation of graphs:
\begin{equation*}
 \forall(G_1,G_2)\in\calG_\tau,\quad\Phi_\tau(G_1G_2)=\Phi_\tau(G_1)\Phi_\tau(G_2).
\end{equation*}
This is a relevant locality requirement since if $G=G_1G_2$ is a non-connected Feynman graph of the theory $\tau$, then the processes described by the graphs $G_1$ and $G_2$ respectively take place in different parts of spacetime and should not influence each other.

In realistic QFT, Feynman rules need to be renormalised. Since the theoretical predictions that will be tested against experiments are actually obtained through this renormalised Feynman rules, we need renormalisation to preserve locality.
A systematic algorithm to perform such a renormalisation was proposed by Bogoliubov, Parasiuk, Hepp and Zimmermann (abbreviated BPHZ renormalisation, and based on the so-called forest formula) \cite{BP57,He66,Zi70}. More recently,  Connes and   Kreimer \cite{CK1} gave an interpretation of this algorithm by means of a coproduct which enables to build -- using  dimensional regularisation -- a renormalised map via its algebraic Birkhoff-Hopf factorisation, regarded as an algebra homomorphism from the Hopf algebra of Feynman graphs to the Rota-Baxter algebra of meromorphic functions in one variable.

The first goal of this chapter is to build a \emph{multivariate} renormalisation scheme, in which regularised Feynman amplitudes would take their values in spaces of meromorphic functions (or rather germs) in several variables. Attents towards such a renormalisation scheme had been made, in particular by Speer in \cite{Sp74}. We present the approach of \cite{CGPZ1}, which completely solved this issue by introducing {\bf locality structures}, taylored to make sure the various maps we build have the locality properties we desire. The theory of locality structures actually goes beyond the multivariate renormalisation scheme. These themes will be brushed over or just mentioned in passing as we wish to focus here on the construction and application of the multivariate renormalisation scheme.

Before summarising the content of this chapter and stating its main results, it is worth addressing a natural question the reader might have: why want to build a multivariate renormalisation scheme when univariate ones are working well-enough? There might be various answers but the one that has my personal preference is the following: regularising every \emph{possible} source of divergences in a Feynman graph with a different variable allows one to always keep track of where each divergent term comes from. Thus we deal with expressions that could be more complicated (since no cancellation can take place) but that carry more information. \cy{This could already be seen as a fair trade-off but this feature expresses itself in a somewhat surprising result.} In natural applications of multivariate renormalisation, the Birkhoff-Hopf factorisation \emph{reduces} to a minimal \ty{subtraction}. \cy{This is a renormalisation scheme} where one simply removes the divergent parts. Such renormalisation schemes typically do not preserve locality.

This \ty{striking} feature justifies a posteriori multivariate renormalisation as a natural renormalisation scheme. Let us also point out that there are occurrences where the most naive regularisation would be a multivariate one.

\subsection*{State of the art}

\addcontentsline{toc}{subsection}{State of the art}

As we already mentioned, attempts to introduce multivariate renormalisation are not new, see for example \cite{Sp74}. The approach we present below is only new in that it relies on (what we believe to be) the right algebraic structures to build this renormalisation scheme. These so called {\bf locality structures} were originally introduced in \cite{CGPZ1}. They offer a framework to rigorously build a multivariate renormalisation scheme, a task they completely fulfilled. It has been applied to various objects, most notably in \cite{CGPZ2} to Kreimer's toy model \cite{Kr97}, in \cite{CGPZ3} to a divergent version of Arborified Zeta Values studied in Chapter \ref{chap:MZV} and in \cite{CGPZ2019} to divergent Conical Zeta Values.

To perform these renormalisations, the locality versions of various structures of interest have been build. For example, the locality counterparts of operated structures have been introduced in \cite{CGPZ2}. Another typical case is that known analytical spaces have been shown to carry locality structures, for example meromorphic germs again in \cite{CGPZ2} or families of meromorphic germs in \cite{CGPZ2019}. 

These results were obtained as means towards applications of our multivariate renormatisation scheme. However, locality structures contain more than the multivariate renormalisation scheme and I believe that, as it is often the case in mathematics, they are worth studying for their own sake. Various extensions of the theory presented above were proposed. For example, a locality version of lattice theory was studied in \cite{CGPZ21}. Other locality structures and their relation with various partial structures were investigated in \cite{Zh18}.Locality versions of the Poincaré-Birkhoff-Witt and the Milnor-Moore Theorems were proven in \cite{Lo23}, see also \cite{CFLP22}.

The link between this algebraic approach of the physical question of locality and other approaches to locality has also been explored. Let us point out in particular \cite{Re19} which unifies locality with causality in Algebraic QFT (AQFT). We refer the readers to \cite{Re16} for a gentle introduction to AQFT and to \cite{dutsch2019classical} for a more complete one. Other very recent and exciting results aiming at closing the gap between mathematical and physical renormalisation were obtained in \cite{GPZ24},

\subsection*{Content and main results}

\addcontentsline{toc}{subsection}{Content and main results}

The results presented in this chapter are mostly from \cite{CGPZ1} and \cite{CGPZ2}. Section \ref{sec:loc_tensor_prod} on locality tensor products is from \cite{CFLP22}.

\medskip

We start by introducing locality categories. We start with the category of locality sets (Definitions \ref{defn:independence} and \ref{defn:localmap}). We move on to more sophisticated locality categories, in particular locality vector spaces (Definition \ref{defn:lvs}), locality algebras (Definition \ref{defn:localisedalgebra}) and locality Rota-Baxter algebras (Definition \ref{defn:lrba}). 

The second section \ref{subsec:multi_mero} is the direct continuation of the first. There, we detail an important important example, namely the (locality) algebra of multivariate meromorphic germs with linear poles and rational coefficients. This space is built as a direct limit in Equation \eqref{eq:def_MQ}. A result of \cite{GPZ4} is that this space admits splittings into holomorphic and polar germs. We show in Proposition \ref{pp:merodecomp} that this space of multivariate meromorphic germs is indeed a locality algebra and that the projections on holomorphic and polar germs are locality algebra morphisms and locality Rota-Baxter morphisms, respectively.

\medskip 

The next locality categories we wish to build to obtain a locality Birkhoff-Hopf decomposition are the locality counterparts of coalgebras, bialgebras and Hopf algebras. For this, we need locality tensor products, which are the main focus of Section \ref{sec:loc_tensor_prod}. Since locality tensors, mimicking usual tensors, will be build as quotients, we start by introducing locality relations on quotient spaces. We show that if $(V,\top)$ is a locality vector space and $W$ is a \ty{vector subspace} of $V$, then $V/W$ inherits a natural locality relation $\ttop$ from $\top$, see Definition \ref{defn:quotientlocality}. We here encounter a striking difference with the non-locality case: the quotient of a locality vector space with a \ty{vector subspace} might not be a locality vector space for the natural relation $\ttop$ as explicitly shown in Counterexample \ref{counterex:quotient_loc}.

Nonetheless, we can still define locality tensor products of two (or more) subspaces of some ambient space (Definition \ref{defn:loctensprod1}) and endow this product with a canonical locality structure. However, as we explicitly state, it is still an open question whether or not these locality tensor products are always a locality vector space. Based on  the work of \cite{CFLP22} it is conjectured to be true.

\medskip

In Section \ref{sec:coal} we finally introduce locality coalgebras (Definition \ref{defn:colocalcoproduct}), bialgebras (Definition \ref{defn:locbialgebra}) and Hopf algebras (Definition \ref{defn:LHopf}). These structures induce a locality version of the convolution product (Definition \ref{def:conv_prod_loc}) and we argue that the usual result that graded connected bialgebras are Hopf algebras hold in the locality framework. We also prove results on the convolution product that are specific to locality structures: Theorem
\ref{thm:loc_conv_prod}. 

Within these structures we obtain the main result of this Chapter: Theorem \ref{thm:abflhopf} and in particular Equation \eqref{eq:renom_map}. It is the locality version of the Birkhoff-Hopf factorisation. We show that the above-mentioned BPHZ formula actually reduces to a minimal \ty{subtraction} when one works in the framework of locality structures. This is because locality ``takes care'' of the inner working of the quantities we are renormalising. Moving on, Corollary  \ref{co:abflhopf} is then our solution to our original problem: it is an application of the preceding Theorem where the target locality algebra is the one of multivariate meromorphic germs introduced in Subsection \ref{subsec:multi_mero}. 

\medskip

The last section of the chapter, Section \ref{sec:Kreimer} is an application of this multivariate renormalisation scheme to Kreimer's toy model of iterated integrals. These integrals reflect structures of rooted forests. We treat them by using the locality version of operated structures (Definitions \ref{defn:loc_op_set} and \ref{defn:basedlocsg}). We then introduce with Definition \ref{def:loc_rooted_forests} a locality relation on rooted forests decorated by a locality set. We further introduce properly decorated forests (Definition \ref{def:prop_dec_forests}) and show with Theorem \ref{thm:loc_Hopf_prop_forests} that the usual Connes-Kreimer Hopf algebra of rooted forests carry on to the locality framework. Notice that this result does not seem to have been available in the literature before.

We also show (Proposition \ref{prop:operatedalgebra}) that properly decorated forests, together with the usual grafting map, have the structure of a locality operated locality algebra. But the most important result of this Section is Theorem \ref{thm:univ_prop_trees_loc} which is the locality version of the universal property of rooted forests. Notice that it requires to take \emph{properly} decorated forests. Putting our various ingredients together, we build the multivariate regularisation map for Kreimer's toy model and renormalise it in one go in Theorem-Definition  \ref{defthm:reg_ren_maps}.

\medskip

The paper ends with a list of some open questions in the domain of locality structures.

\section{Locality categories}

We present the basic categories of the theory of locality structures. Except if stated otherwise, these definitions are all from \cite{CGPZ1}. Some of them have been generalised and/or refined in subsequent work, but we will often not need this further degree of generality.

\subsection{The category of locality sets}

The most basic locality structure is the locality set, which we define now.
 \begin{defn} \label{defn:independence} 
 		\begin{enumerate}
\item A {\bf locality set} is a  couple $(X, \top)$ where $X$ is a set and $ \top \subseteq X\times X$ is a symmetric relation on $X$, referred to as the
{\bf locality relation} of the locality set. When the underlying set $X$ needs to be emphasised, we use the notation $X\times_\top X$ or $\top_X$ for $\top$.
 \item  For any subset $U$ of a locality set $(X,\top)$, let
 			\begin{equation*}
 			U^\top:=\{x\in X\,|\, U\times\{x\}\subseteq X\times_\top X\}
 \end{equation*}
be the {\bf  polar subset} of $U$.
\item A {\bf locality subset} of a locality set $(X,\top_X)$ is a pair $(Y,\top_Y)$ with $Y\subseteq X$ and $\top_Y=\top_X\bigcap(Y\times Y)$.
\end{enumerate}	
\end{defn}
For a locality set $(X,\top)$ the locality relation $\top$ is often called the {\bf independence relation}. Therefore, for $x_1, x_2\in X$, we denote $x_1\top x_2$ if $(x_1,x_2)\in \top$ and say that $x_1$ and $x_2$ are {\bf independent} (w.r.t. the locality relation $\top$). 

We will soon give a number of examples of locality sets. Before, let us give another definition so that we indeed have a category of locality sets.
\begin{defn} \label{defn:localmap}
% \begin{enumerate}
%  \item 
 Let $(X,\top_X)$ and $(Y,\top_Y)$ be two locality sets. Two maps $\phi,\psi:X\longrightarrow Y$ are called {\bf independent} (w.r.t. $\top_X$ and $\top_Y$) if 
 \begin{equation*}
  (\phi\times\psi)(\top_X)\subseteq\top_Y\Longleftrightarrow\forall(x_1,x_2)\in X^2,~x_1\top_X x_2\Rightarrow \phi(x_1)\top_Y\psi(x_2).
 \end{equation*}
 Then $\phi:X\longrightarrow Y$ is called a {\bf locality map} if it is independent of itself.
\end{defn}
In other words, a map $\phi:X\longrightarrow Y$ between two locality sets $(X,\top_X)$ and $(Y,\top_Y)$ is a locality map if, and only if
\begin{equation*}
 \forall(x_1,x_2)\in X^2,~x_1\top_X x_2\Rightarrow \phi(x_1)\top_Y\phi(x_2).
\end{equation*}
If two maps $\phi$ and $\psi$ are independent, we often write $\phi\top\psi$. This is of course an abuse of notation, in the same way that saying that $\phi$ and $\psi$ are independent since the independence depends on the locality relations $\top_X$ and$\top_Y$ on each sets. These abuses should not lead to any ambiguities when we will make them.

So far we have as advertised a category of locality sets whose objects are locality sets and morphisms are locality maps. We denote $\LSet$ this category. \cy{Admettedly, locality sets are very general. For example ``being equal'' is an independence relation, and so is ``being different''. Therefore, there are not much than we will do with only locality sets in practice we will require more structures on these objects. Locality sets should be seen as the stepping stone to more relevant structures.}

\medskip

Let us now present examples of locality sets. First, let us point out that it is extremely easy to lift a set to a locality set by endowing it with a symmetric binary relation. For example ``$=$'' and ``$\neq$'' define two such locality relation. It is a simple exercise to show that a set with $n\in\N$ elements can be endowed with exactly $2^{n(n+1)/2}$ locality relations. Indeed such a relation can be think of as a $n\times n$ symmetric matrix whose entry are 0 or 1. Since a symmetric matrix has $n(n+1)/2$ independent entries, and any entries has two possible values, we find the claimed number of locality relations.

Let us now look at other simple yet interesting examples of locality sets. \cy{These examples have locality relations that are closer, intuitively speaking, to the physical notions of locality we have in mind.}
\begin{example} \label{ex:power_set}
 Let $X$ be a set and $\calP(X)$ be the power set of $X$. We endow $\calP(X)$ with the relation $\top^\calP$ defined by the empty intersection:
 \begin{equation*}
  \forall A,B\subset X,~A\top^\calP B:\Longleftrightarrow A\cap B=\emptyset.
 \end{equation*}
 Then $(\calP(X),\top^\calP)$ is a locality set.
\end{example}
Notice that many variations of $\top^\calP$ can be defined (e.g. $A\cup B=X$) but $\top^\calP$ turns out to be the most interesting one for examples and applications.
\begin{example} \label{ex:loc_disj_supp}
 Let $X$ be a set, $\K$ to be $\R$ or $\C$ and $\calA(X,\K)$ the set of maps from $X$ to $\K$. The relation $\top_{\rm dis}$ given by disjointness of (set-theoretic) support is defined by
 \begin{equation*}
  f\top_{\rm dis}g:\Longleftrightarrow \forall x\in X, f(x)g(x)=0.
 \end{equation*}
 Then $(\calA(X,\K),\top_{\rm dis})$ is a locality set.
 
 In a topological setting, this can be extended to the disjointness of the topological supports, i.e. of the closures of $\{x\in X|f(x)\neq 0\}$ and $\{x\in X|g(x)\neq 0\}$.
\end{example}
This example can be extended in various ways to distributions, by requiring disjointness of their supports, singular supports  or  wavefront sets (see \cite{BDH14} and references therein for a clear presentation of these topics). This gives hope that the program of Epstein-Glaser renormalisation could be formulated in the language of locality structures.
 
 We have already seen an interesting phenomenon, namely that given a locality set $(X,\top)$ one can often derive locality relations from $\top$ on sets related in some sense from $X$. The example we have seen is the set of maps between two locality sets which can be endowed with an independence relation. Another example of this phenomenon gives a generalisation of Example \ref{ex:power_set}.
 \begin{example}
  Let $(X,\top_X)$ be a locality set and $\calP(X)$ be the power set of $X$. We endow $\calP(X)$ with another locality relation $\top^{\calP}_X$ defined by
  \begin{equation*}
   A\top^{\calP}_X B:\Longleftrightarrow \forall a\in A,\forall b\in B,a\top_Xb.
  \end{equation*}
  for any subsets $A$ and $B$ of $X$. Then $(\calP(X),\top^{\calP}_X)$ is a locality set.
 \end{example}
 The last example is for decorated forests. Recall (Definition \ref{def:forests}) that a {\bf rooted forest} $F=(V(F),E(F))$ is a directed acyclic graph such that each component has a unique minimal element. Let $\Omega$ be a set. A {\bf $\Omega$-decorated rooted forest} is a pair $(F,d_F)$ (often the decoration map $d_F$ is omitted) with $F$ a rooted forest and $d_F:V(F)\longrightarrow\Omega$ a map. $\calF_\Omega$ is the vector space freely generated by $\Omega$-decorated rooted forests.
 \begin{example}
  Let $(\Omega,\top_\Omega)$ be a locality set. Then we endow the set $\calF_\Omega$ with a locality relation which in an abuse of notation we also write $\top_\Omega$. This relation is defined as
  \begin{equation*}
   F\top_\Omega F':\Longleftrightarrow\forall (v,v')\in V(F)\times V(F'),~d_F(v)\top_\Omega d_{F'}(v')
  \end{equation*}
  for any $\Omega$-decorated forests $(F,d_F)$ and $(F',d_{F'})$.
 \end{example}
 A locality subset of $\calF_\Omega,\top_\Omega)$ will play an important role in the sequel.
 
 \subsection{Locality vector spaces and algebras}

 We define now locality vector spaces, which together with locality linear maps to be defined latter on, form probably the most important of locality categories. 
\begin{defn} \label{defn:lvs}
A {\bf locality vector space} is a vector space $V$ equipped with a locality relation $\top$ which is compatible with the linear structure on $V$ in the sense that, for any  subset $X$ of $V$, $X^\top$ is a linear subspace of $V$.
\end{defn}
\begin{rk} \label{rk:locvecsp}
 For a locality vector space $(V,\top)$, since $V^\top$ is a linear subspace of $V$, we have  $\{0\}\times V\subset \top,$ or equivalently  $0\in V^\top.$ Note that there are no locality restrictions for the vector space structure (addition and scalar multiplication) on $V$, that is, the addition and scalar multiplication are everywhere defined.
\end{rk}
% \begin{example} \label{ex:merosp}
% The vector space $\calM_\Q $ equipped with the relation $\perp^Q$ in Definition~\ref{de:meroindep} is a locality vector space $\left(\calM _\Q,\perp^Q\right)$.
% \end{example}
An important example is when the locality structure on a vector space comes from a locality on a basis. It is a special case of locality basis.
\begin{example}
 From a locality set $(X, \top)$ let  $(KX, \top)$ be the locality vector space whose defining relation (denoted by the same symbol $\top$) is  the  linear extension of that on $X$. More precisely for $u, v\in   K X$, $(u,v)\in \top$ if the basis elements from $X$ appearing in $u$ are related via $\top$ to the basis elements appearing in $v$.
Thus
$$  KX\times_{\top_{  KX}}  KX=\bigcup_{U,V\subseteq X,(U,V)\subseteq \top}  KU\times   KV.$$
\end{example}
In order to have the category of locality vector spaces, let us define their morphisms.
\begin{defn}  \label{defn:locallmap}
Let $\left(U,\top_U\right)$ and $ (V,\top_V)$ be locality vector spaces, a linear map $\phi:\left( U,\top_U\right)\to \left(V,\top_V\right)$ is called a {\bf locality linear map} if it is a locality map.

We write $\LVect$ the category of locality vector spaces with locality linear maps.
\end{defn}
There are so far no clear way to turn $\LVect$ into a monoidal category with a monoidal structure relevant for locality (see Conjecture \ref{conj:loc_tensor} at the end of this Chapter and the discussion after). However, the set of locality linear maps from a locality vector space to another can itself be endowed with the structure of a locality vector space. Below 
% We have a useful property of locality linear maps, namely that the category $\LVect$ is enriched over $\LVect$. W
we only prove the case that is of importance for us.
\begin{prop} \label{lem:locallinearmap}
Let $(U,\top_U), (V,\top_V)$ be locality vector spaces and   $\phi,\psi:(U,\top_U) \longrightarrow (V,\top_V)$ be independent locality linear maps. Any two linear combinations of $\phi$ and $\psi$ are also independent. In particular, any linear combination  $\lambda\phi+\mu \psi$ with $\lambda, \mu \in \K$ is a locality linear map.
\end{prop}
\begin{proof}
Let $u_1, u_2$ be in $\top_U$. Since $\phi$ and $\psi$ are independent locality linear maps, we have $\phi(u_1)\top\phi(u_2)$ and $\phi(u_1)\top\psi(u_2)$. Since $\phi(u_1)^\top$ is a \ty{vector subspace} of $V$ this implies 
\begin{equation*}
 \phi(u_1)\top\left(\lambda\phi(u_2)+\mu\psi(u_2)\right).
\end{equation*}
Similarly, we also have $\psi(u_1)\top\left(\lambda\phi(u_2)+\mu\psi(u_2)\right)$ and hence, using once more the axiom of locality vector spaces we obtain
$\left(\lambda\phi(u_1)+\mu \psi(u_1)\right)\top_V  \left(\lambda\phi(u_2)+\mu \psi(u_2)\right).$
	\end{proof}	
	Locality algebras are defined through locality semigroups, so let us introduce this notion now. First, for a locality set $(X,T)$ and an integer $k\geq 2$, denote
\begin{equation} \label{eq:lprod}
X^{_\top k}: = \underbrace{X \times_\top \cdots \times_\top X}_{k \text{ factors}}: = \left\{ (x_1,\cdots,x_k)\in X^k\,\left|\, (x_i,x_j)\in \top, 1\leq i\neq j\leq k\right.\right\}.
\end{equation}
\begin{defn} \label{defn:lsg}
\begin{enumerate}
\item
An {\bf locality semigroup}\footnote{\label{ft:part}
As a special case of partial algebras \cite{Gratzer}, the terminology ``partial semigroup" is used for a set equipped with a partial associative product defined only for certain pairs of elements in the set. The condition for a locality semigroup is more restrictive than that of a partial semigroup in that the former requires that the pairs {for which  the partial product is defined stem  from} a symmetric relation and that
the partial product should be compatible with the locality relation in the sense of Equation \eqref{eq:semigrouploc}.}
 is a locality set $(G,\top)$ together with a partial product law defined on $\top$:
$$ m_G: G\times_\top G\longrightarrow  G
$$
for which the product is compatible with the locality relation on $G$, namely
\begin{equation} \label{eq:semigrouploc}
\forall U\subseteq G, \quad  m_G((U^\top\times U^\top)\cap\top)\subset U^\top
\end{equation}
and such that the following {\bf locality associativity property} holds:
\begin{equation}\label{eq:asso}
(x\cdot y) \cdot z = x\cdot (y\cdot z) \text{ for all }(x,y,z)\in G\times_\top G\times_\top G. 	
\end{equation}
Note that, because of the condition \eqref{eq:semigrouploc}, both sides of Eq.~(\ref{eq:asso}) are well-defined for any triple in the given subset.
\label{it:lsg}
\item
\ty{A} locality semigroup is {\bf commutative} if $m_G(x,y)=m_G(y,x)$ for $(x,y)\in \top$, noting that both sides of the equations are defined since $\top$ is symmetric.
\item
\ty{A}  {\bf locality   monoid} is a locality   semigroup $(G,\top, m_G)$ together with a {\bf unit element} $1_G\in G$ given by the defining property
\[\{1_G\}^\top=G\quad \text{ and }\quad m_G(x, 1_G)= m_G(1_G,x)=x\quad \forall  x\in G.\]
We denote the locality  monoid by $(G,\top,m_G, 1_G)$.
\label{defn:partial monoid}
\item A {\bf locality group} is a locality monoid $(G,\top, m_G,1_G)$ equipped with a morphism $\iota: G\to G$ of locality sets, called the {\bf inverse map}, such that  $(\iota (g),g)\in \top$ and $m_G(\iota(g),g)=m_G(g,\iota(g))=1_G$ for any $g\in G$.
\item
A {\bf sub-locality semigroup} of a locality semigroup $(G,\top,m_G)$ is a locality semigroup $(G',\top',m_{G'})$ with $G'\subseteq G$, $\top'=(G'\times G')\cap \top$ and $m_{G'}=m_G|_{\top'}$, that is, for $x, y\in G'$ and $(x, y)\in \top$, $m_G(x,y)$ is in $G'$.
A {\bf sub-locality monoid} of a locality monoid is a sub-locality semigroup of the corresponding locality semigroup which share the same unit. A {\bf sub-locality group} of a locality group is a sub-locality monoid of the corresponding locality monoid which
is also a locality group.
\label{it:lssg}
\end{enumerate}
\end{defn}
For notational convenience, we usually abbreviate $m_G(x,y)$ by $x\cdot y$ or simply $xy$. 

\begin{rk}
One easily checks that on   a locality   monoid $(G, \top,m_G, 1_G)$ if $(x_1,x_2, y_1, y_2)$ is in $G^{_\top4}$, then $(x_1 x_2, y_1, y_2)$ and $(x_1,x_2, y_1 y_2)$ are in $G^{_\top3}$ and hence $(x_1x_2, y_1y_2)\in \top.$
\end{rk}
Locality semigroup is our first structure where the philosophy of locality is made clear. In this structure, one can only multiply objects that are independent. This is one of the basic ideas behind locality structures. It is quite easy to exhibit examples of locality semigroup. Instead, we present a simple counter example of locality semigroup.
\begin{coex}
The set $G$ of linear subspaces of $\R^2$ is a locality set with respect to the following relation $\top_G$ on linear subspaces of $G$: $U, V\subseteq \R^2$ are called transverse if they intersect trivially, namely if
$U\cap V=\{0\}$. The set $G$ equipped with linear sums $+$  is  a monoid. But the corresponding $(G,\top_G, +)$ is not a locality monoid. Indeed, for the standard basis $\{e_1,e_2\}$ of $\R^2$,  the subspaces $\R e_1$ and $\R e_2$  both  intersect $\R (e_1+e_2)$ trivially, but $\R e_1 +\R e_2$ does not.
\end{coex}

Before defining locality algebras, we need first need a preliminary notion.
\begin{defn} 
 Let $V$ and $W$ be vector spaces and let $\top:=V\times_\top W \subseteq V\times W$. A map $f: V\times_\top W \to U$ to a vector space $U$ is called {\bf locality bilinear} (with respect to $\top$) if
$$f(v_1+v_2,w_1)=f(v_1,w_1)+f(v_2,w_1), \quad f(v_1,w_1+w_2)=f(v_1,w_1)+f(v_1,w_2),$$
$$f(kv_1,w_1)=kf(v_1,w_1), \quad 
f(v_1,kw_1)=kf(v_1,w_1)$$
for all $v_1,v_2\in V$, $w_1,w_2\in W$ and $k\in  \K $ such that all the pairs arising in the above expressions are in $V\times_\top W$.
\end{defn}
In most of our applications, the set $V\times_\top W \subseteq V\times W$ will be clear from context, and therefore we will not write that a map is bilinear \emph{with respect to $V\times_\top W$} except is absolutely necessary.
\begin{defn} \label{defn:localisedalgebra}
\begin{enumerate}
\item A {\bf nonunitary locality  algebra} over $\K$ is a locality vector space $(A,\top)$ over $K$ together with a locality bilinear map
	$$ m_A: A\times_\top A \to A$$ such that
	$(A,\top, m_A)$ is a locality semi-group in the sense of Definition~\ref{defn:lsg}.(\ref{it:lsg}).
	\item A {\bf locality algebra} is a nonunitary locality algebra $(A,\top, m_A)$ together with a {\bf unit} $1_A:\K\to A$ in the sense that
	$(A,\top, m_A, 1_A)$ is a locality monoid defined in  Definition~\ref{defn:lsg}.~(\ref{defn:partial monoid}). We shall omit explicitly  mentioning the unit $1_A$ and the product $m_A$ unless this generates an ambiguity.
\item
		A  linear subspace $B$ of a locality algebra $\left(A,\top , m_A \right) $ is called a {\bf sub-locality algebra} of $A$ if
$B$ is a sub-locality semigroup of $A$ in the sense of Definition~\ref{defn:lsg}.(\ref{it:lssg}).
\item
A sub-locality algebra $I$ of a locality commutative algebra $\left(A,\top ,m_A \right) $ is called a {\bf locality ideal} of $A$ if for any $b\in I$ we have
$b^\top \cdot b\subseteq I ~\forall b^\top\in\{b\}^\top$.
\item
\ty{A} locality-linear map $f:(A,\top_A,\cdot_A)\to (B,\top_B,\cdot_B)$ between two (non necessarily unital) locality algebras is called a {\bf locality algebra homomorphism} if
\begin{equation} \label{eq:defn:localisedideal}
f(u\cdot_A v)=f(u)\cdot_B f(v)\ \  \forall (u,v)\in\top_A.
\end{equation}
\item A locality algebra $A$ with a linear grading $A=\oplus_{n\geq 0}A_n$ is called a locality graded algebra if $m_A((A_m\times A_n)\cap\top_A) \subseteq A_{m+n}$ for all $m, n\in \Z$.
    \label{it:gradalg}
\end{enumerate}
\end{defn}
It is easy to check that a locality linear map $f:(A,\top_A,\cdot_A)\to (B,\top_B,\cdot_B)$ between two locality algebras is a locality algebra homomorphism if and only if $\ker f$ is a locality ideal of $A$,
by the same argument as the one for the corresponding result on an algebra homomorphism.

\begin{rk}
		\begin{enumerate}
\item For a locality   algebra $(A,\top)$ we have $\{0, 1_A\}\subset A^\top$  since $0\in  A^\top$ by Remark~\ref{rk:locvecsp}.
\item	If $A\times_\top A$ is $A\times A$ in a locality monoid and locality algebra, we recover the usual notions of monoid and algebra.
		\end{enumerate}
\label{rk:algebraunit}
\end{rk}
% Let us finish this Subsection with a straightforward result which will play a role later on.
% \begin{lemma}
% Let $(G,m_G,\top_G)$ be a locality semigroup. Let $k\geq  2$ and $1\leq i\leq k$. For $(x_1,\cdots,x_k)\in G^{_\top k}$ we have
% \begin{enumerate}
% \item
% $(\Id_G^{i-1}\times m_G\times \Id_G^{k-i-1})(x_1,\cdots,x_k)\in G^{_\top (k-1)}.$
% \item
% $(x_1\cdot \ldots \cdot x_i, x_{i+1}\cdot \ldots \cdot x_k)\in G\times_\top G,$.
% \end{enumerate}
% \label{lem:locprod}
% \end{lemma}

\subsection{Locality Rota-Baxter algebras and projection maps}

 Rota-Baxter operators and algebras are a classical topic (see for example \cite{Guo}). Let us now give their locality version. 
\begin{defn}
A linear operator $P: A\to  A$ on a locality algebra $(A,\top)$ over a field $ K $ is   called
		 a {\bf locality Rota-Baxter operator of
		 weight} $\lambda\in  K $, or simply a {\bf Rota-Baxter operator}, if it is a locality map, independent of $\Id_A$, and satisfies the following {\bf locality Rota-Baxter relation}:
		 \begin{equation}
 P({ a})\, P({ b})= P(P({ a})\,{ b})+ P({ a}\, P({ b})) +\lambda\, P({ a}\,{ b}) \quad \forall (a,b)\in \top.
 \label{eq:rbo}
 \end{equation}
We call the triple $(A,\top, P)$ a {\bf locality Rota-Baxter algebra}.
\label{defn:lrba}
\end{defn}
  \begin{rk}\begin{enumerate}
\item The right hand side of  {Eq.~}\eqref{eq:rbo} is well defined due to the condition that $P$ is independent of the identity. Then for any $(a,b)\in \top$ we have $a\top P(b)$ and $P(a)\top b$. Since $P$ is a locality map\footnote{which is actually implied from being independent of the identity, so that this requirement was actually not necessary. We leave it for readability.}, we also have then $P(a)\top P(b)$ so that the left hand side is well-defined
\item As in the classical setup  \cite[Proposition 1.1.12]{Guo}, if  $P$ is a locality Rota-Baxter operator of weight  $\lambda$, then $-\lambda-P$ is also a locality Rota-Baxter of weight $\lambda$.
\end{enumerate}
\label{rk:LRB}
\end{rk}
An important class of locality Rota-Baxter algebras arises from idempotent operators, i.e. projections.
\begin{prop}	
Let $\left(A,\top ,m_A \right)$  be  a locality algebra. Let $P:A\longrightarrow A$ be a locality linear idempotent  operator in which case there is a linear decomposition $A=A_1\oplus A_2$ with $A_1=\ker\, (\Id-P)$ and $A_2=\ker\,(P)$ so that $P$ is the projection   onto $A_1$ along $A_2$.  The following statements are equivalent:
\begin{enumerate}
\item $P$ or $\Id-P$ is a locality Rota-Baxter operator of weight $-1$; \label{it:irbo1}
\item $A_1$ and $A_2$ are locality subalgebras of $A$, and $P$ and $\Id-P$ are independent locality maps.
\label{it:irbo2}
		\end{enumerate}
If one of the conditions holds, then $P$ is a locality multiplicative map if and only if $A_2$ is a locality ideal of $A$.
\label{prop:multpi}
\end{prop}
\begin{proof}
We write $\pi_1=P$ and $\pi_2=\Id -P$.
\begin{itemize}
 \item ((\ref{it:irbo1}) $\Longrightarrow$ (\ref{it:irbo2}))
It follows from the locality Rota-Baxter identity \eqref{eq:rbo} that $A_1=P(A)$ is
a sub-locality algebra of $A$. Since $\Id-P$ is again an idempotent locality Rota-Baxter operator, $A_2=(\Id-P)(A)$ is also a sub-locality algebra of $A$. Finally, $P$  {and} $\Id-P$ are independent locality maps as a consequence of Definition \ref{defn:lrba}: for any $(a,b)\in\top$ we have $P(a)\top P(b)$ and $P(a)\top b$ since $P$ is a locality map independent from the identity. Then, since $P(a)^\top$ is a \ty{vector subspace} of $A$ we have $P(a)\top(\Id-P(b)$.

\item ((\ref{it:irbo2}) $\Longrightarrow$ (\ref{it:irbo1}))
Since $\pi_1$ and $\pi_2=\Id-\pi_1$ are locality Rota-Baxter operators of weight $-1$ at the same time in view of Remark~\ref{rk:LRB}, we only need to verify that $\pi_1$ is a locality Rota-Baxter operator of weight $-1$:
\begin{equation}\label{eq:rbeminus}
\pi_1({ a})\,\pi_1({ b})+\pi_1(a\,b)= \pi_1(\pi_1({ a})\,{ b})+ \pi_1({ a}\, \pi_1({ b}))\quad \forall ({ a}, { b})\in \top.
\end{equation}

Write $a=a_1+a_2$ and $b=b_1+ b_2$. Since the projections $\pi_i, i=1,2,$ are independent locality maps, it follows that $\{a_1,a_2\} \top \{b_1,b_2\}$. Thus every term in
$$ ab = a_1b_1 + a_1b_2+a_2b_1 +a_2b_2$$
is well defined, with $a_1b_1\in A_1$ and $a_2b_2\in A_2$.
Then the left hand side of Eq.~(\ref{eq:rbeminus}) becomes
$$ a_1b_1+\pi_1(a_1b_1 +a_1b_2+a_2b_1 +a_2b_2)
=2a_1b_1+\pi_1(a_1b_2)+\pi_1(a_2b_1).$$
The right hand side of Eq.~(\ref{eq:rbeminus}) becomes
$$\pi_1(a_1b)+\pi_1(ab_1)=\pi_1(a_1b_1+a_1b_2)+\pi_1(a_1b_1+a_2b_1)
=\pi_1(a_1b_2)+\pi_1(a_2b_1)+2a_1b_1,$$
as needed and since $\pi_1$ is linear.
\item For the last statement, let us first assume that $P$ is multiplicative. Then for any $a\in A_2$ and $b\in A_2^\top$ we have $P(ab)=P(a)P(b)=0$ thus $ab\in\Ker(P)=A_2$ and $A_2$ is a locality ideal. On the other hand if $A_2$ is an ideal, then for any $(a,b)\in\top$, if we write as before their decomposition $a=a_1+a_2$, $b=b_1+b_2$ we have on the one hand $P(a)P(b)=a_1b_1$. On the other hand $P(ab)=P(a_1b_1+a_1b_2+a_2b_1+a_2b_2)$. We have $a_1b_1\in A_1$ as before but $(a_1b_2+a_2b_1+a_2b_2)\in A_2$ since $A_2$ is a locality ideal. Thus by linearity $P(ab)=P(a_1b_1)=a_1b_1$ as needed.
\end{itemize}
\end{proof}
The structures we have introduced so far, from locality sets to locality Rota-Baxter algebras, will be used to build a multivariate renormalisation scheme. Our target algebra will be a space of multivariate meromorphic germs, to which we now turn our attention and describe its relevant locatity structures.

\section{The example of multivariate meromorphic germs} \label{subsec:multi_mero}

The key ideas of this section are taken from \cite{GPZ2015,GPZ2,GPZ4} from other authors. These constructions are necessary for the rest of this chapter.

\subsection{Multivariate meromorphic germs}

We briefly recall the construction of the algebra of multivariate meromorphic germs and refer the readers to \cite{GPZ2015,GPZ2,GPZ4} for details. This algebra will be the target algebra for the multivariate renormalisation scheme we aim at building.

First recall (for example from \cite{Hor68}) that a function $f:\C^k\longrightarrow\C$ is {\bf holomorphic} if each of its partial functions are holomorphic. Equivalently, it obeys the Cauchy-Riemann equations
\begin{equation*}
 \overline{\partial }f:=\sum_{i=1}^k\frac{\partial f}{\partial \bar z_i}d\bar z_i=0.
\end{equation*}
A {\bf meromorphic function} is a quotient of two holomorphic functions.

As already stated, we are interested in building a multivariate renormalisation scheme. In general, a renormalisation scheme consists of identifying ''problematic`` terms in a regularised expression and removing them in a way that preserves that desired property of the object being renormalised. These problematic terms are typically determined via a Laurent expansion. But Laurent expansion are naturally associated to germs and not function, but this detail is not relevant in this preliminary discussion. It is enough to recall that germs are equivalence classes of functions.  
% thus let us briefly recall their definition here. Once again, a fuller definition is given in \cite{Ho68} (Definition 6.2.1 in the edition of 1973 in my possession).
% 
% For $f$ and $g$ two meromorphic functions on $\C^k$, define and the following equivalence relation:
% \begin{equation} \label{eq:def_germs_mero}
%  f\sim g:\Longleftrightarrow \exists U\text{ neighborhood of the origin s.t. }f|_U=g|_U.
% \end{equation}
% We are being a bit unprecise here: for meromorphic functions, $f|_U=g|_U$ means that their numerators and denominators coincide on $U$ up to a non-vaniching function. More details are given later.
% 
% Then a {\bf germ} of meromorphic function at 0 is an equivalence class of meromorphic functions for this equivalence relation.
In practice we will always write a germ as one of its representant. \\

In their seminal article \cite{GPZ2015}, Li Guo, Sylvie Paycha and Bin Zhang \ty{developed} a theory of Laurent expansions for meromorphic germs in several variable using tools from lattice spaces and cones. We recall now some of these tools that we will use.

Consider the filtered rational Euclidean lattice space
$$\Big(\R ^\infty=\bigcup_{\geq 1} \R ^k, \Z ^\infty =\bigcup_{\geq 1} \Z ^k, Q=(Q_k(., .))_{k\geq 1}\Big),$$
where
$$ Q_k(.,.): \R ^k\otimes \R ^k \to \R, \quad k\geq 1,$$
is an inner product in $\R ^k$ such that $Q_k(\Z ^k\otimes \Z ^k)\subset \Q$ and $Q_{k+1}\circ(j_k\otimes j_k)=Q_k$, where $j_k:\R^k\hookrightarrow\R^{k+1}$ is the canonical embedding. $Q$ induces a locality relation $\perp^Q$ on $\R^\infty$. For example, for $x\in\R^k$ and $y\in\R^{k+1}$ we set 
\begin{equation} \label{eq:def_perpQ}
 x\perp^Q y~:\Longleftrightarrow~ Q_{k+1}(j_k(x),y)=0.
\end{equation}
The generalisation to any pair $(x,y)\in(\R^\infty)^2$ is obvious but cumbersome and will not be written. Using the canonical isomorphism $\R^k\cong(\R^k)^*$ we also have a locality relation on $(\R^\infty)^*:=\bigcup_{\geq 1} (\R ^k)^*$. We denote this other relation by the same symbol $\perp^Q$.
\begin{rk}
 As the name suggests, the construction above is only an example of filtered rational Euclidean lattice space. We will not need other examples of such a structure here and have therefore chosen not to give its definition in full generality, and we refer instead the reader to \cite[Definition 2.1]{GPZ4}.
\end{rk}
In order to see how this framework can be used for meromorphic germs we need a somewhat unusual definition of these objects, which can be worth explaining. 

A meromorphic function or germ with $k$ variable can be seen as a meromorphic function of germs of linear forms over $\C^k$. For example the meromorphic function or germs over $\C^3$
\begin{equation*}
 f(z_1,z_2,z_3) = \frac{z_1(z_1+2z_3)^2}{z_2+z_3}
\end{equation*}
is built from the linear maps $(z_1,z_2,z_3)\mapsto z_1$, $(z_1,z_2,z_3)\mapsto(z_1+2z_3)$ and $(z_1,z_2,z_3)\mapsto(z_2+z_3)$. Using the canonical isomorphism $Lin(\C^k,\C)\simeq(\C^k)^*\otimes\C$ we obtain that meromorphic germs can be seen as germs over $(\C^k)^*\otimes\C$.

This space $(\C^k)^*\otimes\C$, where the tensor product is taken over $\C$, has complex dimension $k$, so real dimension $2k$. This is the same than the space $(\R^k)^*\otimes\C$, where this times the tensor product is taken over $\R$. Thus these spaces are isomorphic and we understand that meromorphic functions or germs with $k$ variables can be seen as meromorphic functions or germs over $(\R^k)^*\otimes\C$. \\

We can now understand why the following definition, taken from \cite{GPZ2015}, is adapted to our context. We slightly modified it to fit the framework of \cite{CGPZ1}.
\begin{defn} \label{defn:mero_germs_lin_rational}
 \begin{enumerate}
  \item A {\bf meromorphic function} on $(\R^k)^*\otimes\C$ is the quotient of two holomorphic functions on $(\R^k)^*\otimes\C$ with respect to the canonical complex structure on $(\R^k)^*\otimes\C$. 
  %Such a function on $(\R^k)^*\otimes\C$ is holomorphic if
% it is holomorphic in the complex coordinates in any dual basis of $(\R^k)^*$.
  \item A {\bf germ of meromorphic functions at 0} or {\bf meromorphic germ} in short on $(\R^k)^*\otimes\C$ is the equivalence class of meromorphic functions on $(\R^k)^*\otimes\C$ for the following equivalence relation: for $f=u_1/u_2$ and $g=v_1/v_2$ two meromorphic functions, $f\sim g$ if, and only if, it exists $U$ neighborhood of the origin of $(\R^k)^*\otimes\C$ and $u:U\longrightarrow\C$ holomorphic such that
  \begin{equation*}
%     f\sim g:\Longleftrightarrow \exists U\text{ neighborhood of the origin, }\exists u:U\longrightarrow\C\text{ holomorphic s. t. }
    u_1|_U=u v_1|_U\wedge u_2|_U=u v_2|_U.
  \end{equation*}
  (essentially, $f$ and $g$ coincide on some neighborhood of the origin).

%   \item A {\bf germ of meromorphic functions at 0} or {\bf meromorphic germ} in short on $(\R^k)^*\otimes\C$ is the quotient of two holomorphic functions in a neighborhood of the origin in $(\R^k)^*\otimes\C$ with respect to the canonical complex structure on $(\R^k)^*\otimes\C$. Such a function on $(\R^k)^*\otimes\C$ is holomorphic if
% it is holomorphic in the complex coordinates in any dual basis of $(\R^k)^*$.
  \item A germ of meromorphic functions $f$ on $(\R^k)^*\otimes\C$ is said to have {\bf linear poles at zero with rational
coefficients} if there exist vectors $L_1,\cdots,L_n\in \Z_k^*\otimes \Q$ (possibly with repetitions)
such that $f\prod L_i$ is a holomorphic germ at zero whose Taylor expansion for coordinates
 in the dual basis $\{e^*_1,\cdots,e^*_k\}$ of a given (and hence every) basis $\{e_1,\cdots,e_k\}$ of $\Z^k$ has
coefficients in $\Q$.
\item We write $\calM_\Q\left((\R^k)^*\otimes\C\right)$ the vector space of these meromorphic germs with linear poles at zero and rational coefficients.
 \end{enumerate}
\end{defn}

\subsection{Splitting of meromorphic germs}

Now, $Q_k:\R ^k\otimes \R ^k \to \R$ induces an isomorphism $Q_k^*:\R^k\longrightarrow(\R^k)^*$ defined by
\begin{equation*}
 Q_k^*:u\mapsto\Big(Q_k(u,.):v\mapsto Q_k(u,v)\Big).
\end{equation*}
This induces a projection $(\R^{k+1})^*\twoheadrightarrow(\R^k)^*$ which in turn induces an embedding $\calM_\Q((\R^k)^*\otimes\C)\hookrightarrow\calM_\Q((\R^{k+1})^*\otimes\C)$. Thus we have a directed system and we can set 
\begin{equation} \label{eq:def_MQ}
 \calM_\Q:=\varinjlim_k\calM_\Q\left((\R^k)^*\otimes\C\right).
\end{equation}
By \cite[Corollary 4.18]{GPZ4}, any element of $\calM_\Q$ can be written as a sum of a holomorphic germ and elements the form
\begin{equation}
\frac{h(\ell_1,\cdots,\ell_m)}{L_1^{s_1}\cdots L_n^{s_n}}, \quad s_1, \cdots, s_n\in \Z_{>0},
 \label{eq:polar}
 \end{equation}  
 where $h$ is a holomorphic germ with rational coefficients in linear forms $\ell_1,\cdots,\ell_m\in (\Q ^k)^*$,  and $L_1,\cdots, L_n$ are linearly independent linear forms in $(\Q ^k)^*$, $\ell_i\perp^QL_j$ for all $i\in \{1,\cdots, m\}$ and $j\in \{1, \cdots, n\}$. An element of the form \eqref{eq:polar} which is called a {\bf polar germ} (for the inner product $Q$). 
 
 In other words, writing $\calM_+$ the space of holomorphic germs and $\calM_-^Q$ the set of polar germs of the form \eqref{eq:polar} we have a splitting (\cite[Corollary 4.18]{GPZ4})
 \begin{equation} \label{eq:merodecomp}
  \calM_\Q=\calM_+\oplus\calM_-^Q. 
 \end{equation}
 \begin{rk}
  This splitting depends of the chosen family of inner products $Q$, so a more appropriate notation would be $\calM_\Q=\calM_+\oplus_Q\calM_-^Q$ which we elected not to use to improve readibility. Notice also that while the set $\calM_-^Q$ depends on the chosen $Q$, the set of holomorphic germs $\calM_+$ does not.
 \end{rk}
 Below are two examples of this decomposition on meromorphic germs.
 \begin{example}
  Let $f\in\calM_\Q((\R^2)^*\otimes\C)\ni g$ defined by
  \begin{equation*}
   f(z_1,z_2)=\frac{z_1}{z_1-z_2},\qquad g(z_1,z_2)=\frac{z_1z_2}{z_1-z_2}.
  \end{equation*}
  To find their holomorphic and singular parts with respect to the canonical scalar product on $\R^2$ we decompose them as
  \begin{equation*}
   f(z_1,z_2)=\frac{1}{2}\frac{z_1-z_2+z_2+z_1}{z_1-z_2}  = \frac{1}{2}\left(1+\frac{z_1+z_2}{z_1-z_2}\right).
  \end{equation*}
  Thus, we find that the holomorphic part of $f$ is the constant germ $(z_1,z_2)\mapsto1/2$ and its singular part is $(z_1,z_2)\mapsto  \frac{1}{2}\frac{z_1+z_2}{z_1-z_2}$. The \ty{latter} is indeed singular w.r.t. the canonical scalar product since $\langle e_1+e_2,e_1-e_2\rangle=0$.
  
  For $g$ we have
  \begin{align*}
   g(z_1,z_2)=\frac{z_2}{2}\left(1+\frac{z_1+z_2}{z_1-z_2}\right) & = \frac{z_2}{2} + \frac{1}{4}(z_2-z_1+z_1+z_2)\frac{z_1+z_2}{z_1-z_2} \\
   & = \frac{z_2}{2} - \frac{z_1+z_2}{4} + \frac{1}{4}\frac{(z_1+z_2)^2}{z_1-z_2}.
  \end{align*}
  Thus the holomorphic part of $g$ is $(z_1,z_2)\mapsto\frac{z_2}{2} - \frac{z_1+z_2}{4}$ and its singular part is $(z_1,z_2)\mapsto  \frac{1}{4}\frac{(z_1+z_2)^2}{z_1-z_2}$ (both w.r.t. the canonical scalar product).
 \end{example}
%  {\color{red} XXXXX This paragraph later on????} \\
%  In other word, writing $\calM_+$ the space of holomorphic germs and $\calM_-^Q$ the set of polar germs of the form \eqref{eq:polar} we have a splitting
%  \begin{equation*}
%   \calM_\Q=\calM_+\oplus\calM_-^Q.
%  \end{equation*}
%  \begin{rk}
%   This splitting depends of the chosen family of inner products $Q$, so a more appropriate $\calM_Q=\calM_+\oplus_Q\calM_-^Q$ which we elected not to use to improve readibility. Notice also that while the set $\calM_-^Q$ depends on the chosen $Q$, the set of holomorphic germs $\calM_+$ does not.
%  \end{rk}
% {\color{red} XXXXX fin du paragraph.}

\subsection{The locality structure}

We can now define an independence relation on $\calM_\Q$ which is induced by $\perp^Q$ and that we therefore also denote by this symbol.
\begin{defn} (\cite{GPZ20})
 For a meromorphic function $f$, the {\bf dependence subspace} ${\rm Dep}(f)$ is the smallest subspace for the inclusion of $(\R^n)^*$ on which it depends. For a meromorphic germ, the dependence subspace is the dependence subspace of any of its representing elements.

Two meromorphic germs with rational coefficients $f$ and $g$ are {\bf orthogonal} (with respect to the given $Q$) if {${\rm Dep}(f)\perp^Q  {\rm Dep}(g)$}. Then we denote $f\perp^Q g$. Let $(\calM_\Q, \perp^Q)$ denote the resulting locality set.
% The pointwise product gives rise to a \loc  algebra $\left(\calM, \perp\right)$, and $(\calL,\perp)$ is viewed as a \loc subspace of the \loc linear space  $\left(\calM , \perp\right)$. See~\cite{CGPZ3} for another \loc structure on $\calM$.
\end{defn}
Other locality structures on $\calM_\Q$ exist, see \cite{CGPZ1,CGPZ2019}. This one is a good one for our application in that it encodes and makes rigorous the simple idea that two germs are independent if they ``don't depend on the same set of variables'' \emph{or if they can be written as depending on orthogonal variables} (w.r.t. the chosen scalar product). We illustrate this concept in the follwing example.
\begin{example}
 Let $f,g:\C^3\longrightarrow\C$ be defined by 
 \begin{equation*}
  f(z_1,z_2,z_3)=z_1-z_3,\qquad g(z_1,z_2,z_3)=z_2(z_1+z_3).
 \end{equation*}
 Then Dep$(f)=\langle e^*_1-e^*_3\rangle$ and Dep$(g)=\langle e^*_2,e^*_1+e^*_3\rangle$ with $\{e_1^*,e_2^*,e_3^\}$ the canonical basis of $(\R^3)^*$. Then $f\perp^Q g$ for the canonical scalar product since Dep$(f)$ and Dep$(g)$ are orthogonal.
\end{example}
It is clear that $(\calM_\Q, \perp^Q)$ actually belongs to a more sophisticated locality category.
\begin{prop} \label{prop:mult_germs_loc_alg}
 $(\calM_\Q, \perp^Q)$ is a locality vector space. Furthermore, the restricted multiplication $m:\calM_\Q\times\perp^Q\calM_\Q\longrightarrow\calM_\Q$ endows $(\calM_\Q, \perp^Q)$ with a locality algebra structure.
\end{prop}
% To finish this example we need to come back to Equation \eqref{eq:polar}. It implies that,  writing $\calM_+$ the space of holomorphic germs and $\calM_-^Q$ the set of polar germs of the form \eqref{eq:polar} we have a splitting
%  \begin{equation} \label{eq:merodecomp}
%   \calM_\Q=\calM_+\oplus\calM_-^Q. 
%  \end{equation}
%  \begin{rk}
%   This splitting depends of the chosen family of inner products $Q$, so a more appropriate notation would be $\calM_\Q=\calM_+\oplus_Q\calM_-^Q$ which we elected not to use to improve readibility. Notice also that while the set $\calM_-^Q$ depends on the chosen $Q$, the set of holomorphic germs $\calM_+$ does not.
%  \end{rk}
 Thus, the splitting given by Equation \eqref{eq:merodecomp} together with Propositions \ref{prop:multpi} and \ref{prop:mult_germs_loc_alg} directly imply a result that will be crucial for our multivariate renormalisation scheme.
\begin{prop}
In the decomposition in Equation \eqref{eq:merodecomp}, the space  $\calM_{+}$ is a subalgebra and a locality subalgebra of $\calM_\Q$. The space $\calM_{-}^Q$ is not a subalgebra but a locality subalgebra, in fact a locality   ideal of $\calM_\Q$. Consequently, the projection $\pi_+^Q:\calM_\Q \to \calM_{+}$ is a locality algebra homomorphism and  $(\calM_{\Q},\pi_-^Q)$ is a locality Rota-Baxter algebra..
\label{pp:merodecomp}
\end{prop}

In contrast to the multivariate case, the space  ${\calM}_{\Q,-}^Q(\R ^*\otimes \C )=\eps^{-1}\C[\eps^{-1}]$ is a subalgebra in the space 	 $	{\mathcal M}_\Q(\R^*\otimes\C  )$  of meromorphic functions in one variable.
	 This is a major difference between our multivariate setup and the usual  single variable framework used for renormalisation purposes. We circumvent the difficulty in  relaxing ordinary  multiplicativity  to a  multiplicativity allowed only on independent elements. In fact, $\calM_{\Q,-}^Q(\R^*\otimes \C ) $ is a locality ideal of $\calM_\Q(\R^*\otimes \C  )$ under the restriction of independence relation since the locality relation $\perp^Q$ restricted to $\calM_\Q(\R^*\otimes \C )$ is simply $(\C\times \calM_\Q(\R^*\otimes \C )) \cup (\calM_\Q (\R^*\otimes \C ) \times \C)$. Thus, the   locality algebra homomorphism  $\pi_+^Q$ restricts  to  a mere linear map on $\calM_\Q(\R^*\otimes \C  )$ with no additional multiplicativity property.

\section{Locality tensor products}
\label{sec:loc_tensor_prod}

The next locality categories we are interested in are the pendant of coalgebraic structures. For this, we need locality tensors which we now present. This presentation is from \cite{CFLP22}.

Let us recall that the tensor product of two vector spaces $V$ and $W$ reads
\begin{equation} \label{eq:def_tensor}
 V\otimes W:=\K (V\times W)\big/{\rm I_{bil}},
\end{equation}
where $\K (V\times W)$ is the vector space freely spanned by $V\times W$ and ${\rm I_{bil}}$ defined as its \ty{vector subspace} generated by all elements of the form
\begin{align*}
 (a+b,x)-(a,x)-(b,x),\qquad & (a,x+y)-(a,x)-(a,y), \\
 (ka,x)-k(a,x),\qquad
& (a,kx)-k(a,x).
\end{align*}
We will build locality tensor products in the same fashion and want to endow them will a locality relation, so we turn our attention to the more general problem of endowing quotient of locality vector space with a locality relation.

\subsection{Quotient locality as a final locality relation} 

We define a final locality in much the same way as a final topology. 
Recall that given two topologies $\tau_1$, $\tau_2$ on some set $X$,   $\tau_1$ is said to be {\bf coarser  {(weaker or smaller)}}  than $\tau_2$, or equivalently $\tau_2$ {\bf finer    {(stronger or larger)}}
 than $\tau_1$ if, and only if $\tau_1\subset \tau_2$.  
Also, given  a set $X$ and  $(X_i,\tau_i)_{i\in I}$  a family of topological spaces together with a family of maps  $f_i:X_i\to X$, the {\bf final topology  { (or strong,  colimit, coinduced, or inductive topology)}  
 $\bar\tau$} is the finest topology on $X$ such that all maps $f_i$ are continuous.
  With a small abuse of language, one says that the topology $\bar\tau$ is final with respect to the maps $f_i$.
  
  Let us now transpose this terminology to the locality setup. 
\begin{defn} \label{defn:final_relation}
\begin{itemize}
 \item Let $\top_1$ and $\top_2$ be two locality relations over a set $A$. We say $\top_1$ is  {\bf coarser} than $\top_2$  or equivalently, that $\top_2$ is 
 {\bf finer} than $\top_1$ if, and only if $\top_1\subset \top_2$.
 \item  Let $X$ be a set, $(X_i,\top_i)_{i\in I}$ a family of locality sets, and $f_i:X_i\to X$ a family of maps. The {\bf final locality relation $\ttop$} 
on $X$ is the  {coarsest  locality relation among the locality relations $\top$ on $X$ for which \[  f_i: (X_i,\top_i)\longrightarrow (X, \top), \quad i\in I \] are locality maps}. 
\end{itemize}
As before, with a slight abuse of language, we shall say that $\ttop$ is a {\bf  final locality relation} on $X$  for the maps $f_i$.
\end{defn}
Let us see with an example how taking inspiration from topologies is relevant for locality.
\begin{example} 
Let $X$ be a set and ${\mathcal P}(X)$ its power set. Disjointness of sets:
		\[A\top B\Longleftrightarrow A\cap B=\emptyset\] defines a locality relation on any subset ${\mathcal O}$ of ${\mathcal P}(X)$. 
	If $(X, {\mathcal O})$ is a topological space with topology ${\mathcal O}  \subset {\mathcal P}(X)$, this disjointness relation gives rise to  
	another locality relation (which with some abuse of notation, we denote by the same notation) given by the separation of points:
	\[x\top y\Longleftrightarrow \exists U, V\in {\mathcal O}, \quad \left( U\, \top \, V\right)\, \wedge \,  \left(x\in U\wedge y\in V\right).\]
The finer (coarser) the  topology ${\mathcal O}$, the  {larger (smaller)} the graph $\{(x, y), \, x\top y\}$ of the locality relation, hence the terminology we have chosen.
 	\end{example}
We can characterise final locality relations.
\begin{prop} \label{prop:description_final_relation}
 Given a  surjective map $ {\phi:A\to B}$, the locality relation $\top$ on $A$ induces a locality relation $\ttop$ on $B$ defined by 
\[ b_1\ttop b_2\Longleftrightarrow (\exists (a_1,a_2)\in A\times A : \phi(a_i)=b_i\ \rm{and}\ a_1\top a_2),\]
which is the final locality relation for the map $\phi$.
\end{prop}
\begin{proof} 
 {It is clear from the definition  of $\ttop$,}  that $\phi:(A,\top)\longrightarrow(B,\ttop)$ is a locality map.

Let $\top_B$ be a locality relation on $B$ such that $ {\phi:(A,\top)\longrightarrow(B,\top_B)}$ is a locality map. For any $(b_1,b_2)\in B^2$ we have
\begin{align*}
 b_1\ttop b_2 ~&\Longrightarrow~\left(\exists(a_1,a_2)\in A^2|\phi(a_i)=b_i~\wedge~a_1\top a_2\right)\quad \text{for }i\in\{1,2\}\\
 &\Longrightarrow~\left(\exists(a_1,a_2)\in A^2|\phi(a_i)=b_i~\wedge~\phi(a_1)\top_B \phi(a_2)\right)\quad \text{since }\phi\text{ is a locality map}\\
 &\Longrightarrow~b_1\top_B b_2.
\end{align*}
Therefore $\ttop\subseteq\top_B$.
\end{proof}
Applying Proposition \ref{prop:description_final_relation}   to  the canonical   projection map $\pi:V\to V/ W$  of a  locality  vector space  $(V, \top)$ to its quotient $V/ W$ by a linear subspace $W$, we equip  the quotient with the quotient locality relation.
\begin{defn}\label{defn:quotientlocality}
 For a subspace $W$ of a  \cy{locality} vector space $(V, \top)$, we call {\bf quotient locality} on 
 the quotient $V/ W$, the final locality relation
  \begin{equation*}%\label{eq:quotientloc}
 \left([u]\ttop[v]~\Longleftrightarrow \exists(u',v')\in[u]\times[v]:~u'\top v'\right)\quad\quad \forall ([u],[v])\in (V/ W)^2 
 \end{equation*}
  for the canonical projection map    $\pi:V\to V/ W$. 
%   This way, the  pre-locality space $(V, \top)$ gives rise to a  {pre-locality} vector quotient space $(V/ W, \ttop)$ and the projection map $\pi:(V, \top)\to (V/ W, \ttop)$ is a morphism of pre-locality vector spaces.
\end{defn}
Notice that $(V/W,\ttop)$ is a vector space with a locality relation but it is a priori \emph{not} in general a locality vector space\footnote{it is however a ``pre-locality vector space'', a notion \ty{developed in} \cite{CFLP22} to tackle this type of issue.}. The general question ``when is a quotient of locality vector spaces a locality vector space for the quotient locality'' seems very interesting but deeply out of reach. We just give a couterexample to show that the answer to this question cannot always be positive.
\begin{coex} \label{counterex:quotient_loc}
 	We equip  the vector space $V$ of  {real valued maps on $\R$} with the locality relation $\top$ given by disjoint supports: $f\top g\Longleftrightarrow {\rm supp}(f)\cap {\rm supp}(g)=\emptyset$. Let $W$ denote the linear subspace of constant functions. Consider the three functions $u,v,w$ in $V$ defined by
 	\begin{align*}
 		v:&\left\{\begin{array}{rcl}
 			\R&\longrightarrow&\R\\
 			x&\longmapsto&\begin{cases}
 				1\mbox{ if }x>1,\\
 				0\mbox{ otherwise},
 			\end{cases}
 		\end{array}\right.&
 		u:&\left\{\begin{array}{rcl}
 			\R&\longrightarrow&\R\\
 			x&\longmapsto&\begin{cases}
 				1\mbox{ if }x>0,\\
 				0\mbox{ otherwise},
 			\end{cases}
 		\end{array}\right.&
 		w:&\left\{\begin{array}{rcl}
 			\R&\longrightarrow&\R\\
 			x&\longmapsto&\begin{cases}
 				1\mbox{ if }x>2,\\
 				0\mbox{ otherwise}.
 			\end{cases}
 		\end{array}\right.\\
 		&\definecolor{qqttcc}{rgb}{0.,0.2,0.8}
 		\begin{tikzpicture}[line cap=round,line join=round,>=triangle 45,x=1.0cm,y=1.0cm]
 			\begin{axis}[
 				x=0.7cm,y=0.7cm,
 				axis lines=middle,
 				xmin=-2.5,
 				xmax=2.5,
 				ymin=-0.5,
 				ymax=1.5,
 				xtick={-2.0,-1.0,...,2.0},
 				ytick={-0.0,1.0,...,1.0},]
 				\clip(-2.5,-0.5) rectangle (2.5,1.5);
 				\draw [line width=0.8pt,color=qqttcc] (-5.,0.)-- (1.,0.);
 				\draw [line width=0.8pt,dash pattern=on 1pt off 1pt,color=qqttcc] (1.,0.)-- (1.,1.);
 				\draw [line width=0.8pt,color=qqttcc] (1.,1.)-- (7.,1.);
 			\end{axis}
 		\end{tikzpicture}&
 		&\definecolor{qqttcc}{rgb}{0.,0.2,0.8}
 		\begin{tikzpicture}[line cap=round,line join=round,>=triangle 45,x=0.7cm,y=0.7cm]
 			\begin{axis}[
 				x=0.7cm,y=0.7cm,
 				axis lines=middle,
 				xmin=-2.5,
 				xmax=2.5,
 				ymin=-0.5,
 				ymax=1.5,
 				xtick={-2.0,-1.0,...,2.0},
 				ytick={-0.0,1.0,...,1.0},]
 				\clip(-2.5,-0.5) rectangle (2.5,1.5);
 				\draw [line width=0.8pt,color=qqttcc] (-5.,0.)-- (0.,0.);
 				\draw [line width=0.8pt,dash pattern=on 1pt off 1pt,color=qqttcc] (0.,0.)-- (0.,1.);
 				\draw [line width=0.8pt,color=qqttcc] (0.,1.)-- (7.,1.);
 			\end{axis}
 		\end{tikzpicture}&
 		&\definecolor{qqttcc}{rgb}{0.,0.2,0.8}
 		\begin{tikzpicture}[line cap=round,line join=round,>=triangle 45,x=0.7cm,y=0.7cm]
 			\begin{axis}[
 				x=0.7cm,y=0.7cm,
 				axis lines=middle,
 				xmin=-2.5,
 				xmax=2.5,
 				ymin=-0.5,
 				ymax=1.5,
 				xtick={-2.0,-1.0,...,2.0},
 				ytick={-0.0,1.0,...,1.0},]
 				\clip(-2.5,-0.5) rectangle (2.5,1.5);
 				\draw [line width=0.8pt,color=qqttcc] (-5.,0.)-- (2.,0.);
 				\draw [line width=0.8pt,dash pattern=on 1pt off 1pt,color=qqttcc] (2.,0.)-- (2.,1.);
 				\draw [line width=0.8pt,color=qqttcc] (2.,1.)-- (7.,1.);
 			\end{axis}
 		\end{tikzpicture}
 	\end{align*} 
 	Then, in $V/W$, we have $[u]\ttop[v]$ since $u\top(v-1)$ and $[u]\ttop[w]$ since $(u-1)\top w$. However, 
 	$[u]\cancel{\ttop}[v+w]$. Thus $V/W$ is not a locality vector space for $\ttop$. 
 \end{coex} 
 
 \subsection{Locality tensor products}
 
 For two \ty{vector subspaces} $V$ and $W$ of an ambient locality vector space $E$, we want the locality version of the tensor product to depend on the locality on $E$. We will define this locality tensor product in a similar fashion than the usual tensor (Equation \eqref{eq:def_tensor}). 
 \begin{defn}\label{defn:loctensprod1}			%%%%%%%%%%%%%%%%%			Def: locality tensor product 1
	Given  $V$ and  $W$ subspaces of a locality vector space $(E,\top)$, the  locality  tensor product is the  vector space
	\begin{equation} \label{eq:VotimestopW}
	 V\otimes_{\top}W:={\K(V\times_\top W)}\big/{I_{\rm bil}^{ \top_\times}}
	\end{equation}
	with $I_{\rm bil}^{\top_\times}:=\K(V\times_\top W) \cap I_{\rm bil}$.
\end{defn}
 \begin{rk}
	Since $ V\times_\top W\subset V\times W$ and  $I_{\rm bil}^{\top_\times}:=\K(V\times_\top W) \cap I_{\rm bil}$,  we have an inclusion of vector spaces $V\otimes_{\top}W\subset V\otimes W$. This inclusion of vector spaces is what motivated this choice of locality tensor product over other options. If $V\times_\top W= V\times W$, then $V\otimes_{\top}W= V\otimes W$. 
\end{rk}
Many usual properties of usual tensor products are conserved by this locality tensor product, and in particular various form of their well-known universal property. This is beyond the scope of this thesis and we refer the readers to \cite[Part I]{CFLP22} for details.

We will also need higher locality tensor products. These need further construction to be properly defined. First, for $(E,\top)$ a locality vector space over $\K$, and $V_1, \cdots, V_n$ linear subspaces of $E$, let  $I_{{\rm mult}} {(V_1, \cdots, V_n)}$  generated by all elements of the form
\begin{align*}
 (x_1,...,x_{i-1},a_i+b_i,x_{i+1},...,&x_n)-(x_1,...,x_{i-1},a_i,x_{i+1},...,x_n)-(x_1,...,x_{i-1},b_i,x_{i+1},...,x_n)% \label{eq:multlinform1} 
 \\
 &(x_1,...,kx_i,...,x_n)-k(x_1,...,x_i,...,x_n) %\label{eq:multlinform2}
\end{align*}
for every $i\in[n]$, $k\in\K$ and $a_i,b_i,x_i\in V_i$. If $V_1=\cdots =V_n=V$, we write  $I_{{\rm mult},n}(V) $.
 \begin{defn}\label{defn:loctensalgebra}
	We   define  
	\begin{itemize}
		\item the {\bf  locality cartesian product} 
		%of $(V,\top)$ with itself as
		\begin{equation*}%\label{eq:Vtopn} 
		 V_1\times_\top \cdots \times_\top V_n:=\{(x_1,...,x_n)\in  V_1\times  \cdots \times  V_n|\forall(i,j\in[n]): i\neq j\Rightarrow (x_i,x_j)\in V_i\times_\top V_j\};
		 \end{equation*}
If  $V_i=V$ for any $i\in [n]$, we set $V^{\times_{\top}^n}:=  V_1\times_\top \cdots \times_\top V_n=V\times_\top\dots\times_\top V$.
		In particular $V^{\times_{\top}2}=\top$, $V^{\times_{\top}1}=V$ and we set by convention  $V^{\times_{\top}0}=\K$.
% 		and $V^{\times_\top\infty}:=\bigcup_{n\geq0}V^{\times_\top n}.$ 
		Note that $V_1\times_\top \cdots \times_\top V_n\ni (0_E, 0_E, \cdots, 0_E)$ where $0_E$ is the zero element in $E$, since $0_E\top 0_E$ by definition of a locality vector space.
		\item the {\bf     locality tensor product} \begin{equation} \label{eq:Votimesn}
		 V_1\otimes_{\top} \cdots \otimes_{\top} V_n~:=~{\K(V_1\times_\top \cdots \times_\top V_n)}\big/I^{\top_\times}_{{\rm mult}}(V_1, \cdots, V_n)
		 \end{equation}
		 with $I^{\top_\times}_{{\rm mult}}(V_1, \cdots, V_n):=(I_{{\rm mult}} {(V_1, \cdots, V_n)}\cap\K(V_1\times_\top \cdots \times_\top V_n))$.
		 
 If $V_i=V$ for any $i\in [n]$, we set $V^{\otimes_{\top}^n }:=  V_1\otimes_\top \cdots \otimes_\top V_n$.
 
 \item
We endow the locality tensor product $V_1\otimes_\top\cdots\otimes_\top V_n$ with the locality relation $\top_{\otimes n}$ defined as the quotient locality (see Definition \ref{defn:quotientlocality}) for the quotient map $\pi_n:\K(V_1\times_\top \cdots \times_\top V_n)\longrightarrow V_1\otimes_{\top} \cdots \otimes_{\top} V_n$.
		\end{itemize}
	\end{defn}
This higher locality tensor products allow us to define more sophisticated objects such as the locality tensor algebra and the locality universal envelopping algebra of a (locality) Lie group. These objects in turn have universal properties in suitable locality categories\footnote{at least conjecturally. The proof only exists in a weaker case, for \emph{pre}-locality vector space, a structure defined in \cite{CFLP22}.}. Once again, this lies beyond the scope of this thesis and we refer the reader to \cite{CFLP22} for more elaborate constructions.

Let us finish this section with an intriguing remark. The quotient locality relation allow us to endow locality tensor products with a locality structure. These are easy enough to describe explicitly: let  $(E,\top)$ be a locality vector space. Then for $X, Y\in E\otimes_\top E$, we have $X\top_{\otimes 2} Y$ if we can write
\begin{equation*}
 X=\sum_i v_i\otimes w_i,\qquad Y=\sum_i v'_i\otimes w'_i
\end{equation*}
such that $(v_i,w_i,v'_j,w'_j)\in E^{\times_\top 4}$ for any $i$ and $j$. An obvious generalisation exists for higher tensors products.

However, it is still a conjecture that $(E\otimes_\top E,\top_{\otimes2})$ is a locality vector space. We call it conjecture rather than open question since it is a locality vector space for many cases of interest and numerous attempts to build counterexamples have all failed. To tackle this question might very well require tools that are beyond the theory of locality structures, or a rather large enlargement of that theory.

\section{Coalgebraic locality structures} \label{sec:coal}

\subsection{Locality coalgebras}
 
 We recall that
a coalgebra $(C,\Delta)$ over a field $\K   $ is  {\bf counital} if there is a map $\eps:C\to  \K  $ such that $(\eps\otimes \Id _C)\Delta =(\Id _C\otimes \eps)\Delta=\Id _C$. It is {\bf  ($\Z_{\geq 0}$-)graded}   if
$$C=\bigoplus_{n\in\Z_{\geq 0}}C_n \quad \text{and} \quad \Delta(C_n)\subseteq\bigoplus_{p+q=n} C_p\otimes C_q, \quad \bigoplus_{n\geq 1} C_n\subseteq \ker \eps.$$
Thus $C=C_0+\ker \eps$.
Moreover a graded coalgebra is
called {\bf connected}  if  $C=C_0\oplus \ker \eps$. Consequently, $\eps$ restricts to a linear bijection $\eps: C_0\cong  K $ and $\ker \eps = \oplus_{n\geq 1}C_n$.
% For the sake of simplicity, we shall drop the explicit mention of the grading and simply call such a coalgebra  a connected coalgebra.
 \begin{defn} \label{defn:colocalcoproduct}
Let $(C,\top)$ be a locality vector space and let $\Delta:C\to C\otimes C$ be a linear map. $(C,\top, \Delta)$ is a {\bf locality noncounital coalgebra} if  it satisfies the following two conditions
\begin{enumerate}
\item for any $U\subset C$ (compare with Eq.~\eqref{eq:semigrouploc})
\begin{equation}
\Delta (U^\top) \subset U^\top \otimes _\top U^\top.
\label{eq:comag}
\end{equation}
In particular, $\Delta(C)\subseteq C\otimes_\top C$;
\label{colocalDelta} 	
\item     the following {\bf coassociativity}  holds:
$$(\Id _C\otimes \Delta)\,\Delta= (\Delta\otimes \Id _C)\, \Delta.$$
\label{localDelta}
\end{enumerate}
\begin{itemize}
\item
If in addition, there is a {\bf counit}, namely a linear map $\eps: C\to   \K  $ such that
  $ ({\rm Id}_C\otimes   \eps)\, \Delta=  (\eps\otimes {\rm Id}_C)\, \Delta= {\rm Id}_C$, then $(C,\top,\Delta,\eps)$ is called a {\bf locality coalgebra.}
\item A {\bf connected locality coalgebra} is a locality coalgebra $(C,\top, \Delta)$ with a grading $C=\oplus_{n\geq 0} C_n$ such that, for any $U\subseteq C$,
\begin{equation}
\Delta(C_n\cap U^\top)\subseteq\bigoplus_{p+q=n} (C_p\cap U^\top)\otimes_\top (C_q\cap U^\top), \quad \bigoplus_{n\geq 1} C_n = \ker \eps.
\label{eq:lconn}
\end{equation}

We denote  by $J$   the unique element of $C_0$ with $\eps (J)=1_ \K $, giving $C_0= \K \, J$.
\end{itemize}
\end{defn}
 \begin{rk}
Notice that whereas the conditions for a locality algebra are weaker than those for an algebra, the conditions for a locality coalgebra are stronger than those for a coalgebra. In particular, a locality coalgebra is a coalgebra and a connected locality coalgebra is a connected coalgebra.
\label{rk:conn}
\end{rk}
 Let us list a few useful general properties of locality coalgebras.
\begin{lemma}
Let $(C,\top,\Delta)$ be a locality coalgebra.
\begin{enumerate}
\item For any $n\geq 2$ and $0\leq i \leq n$,
\begin{equation} \label{eq:gcoasso}
\Id _C^{\otimes i} \otimes \Delta \otimes \Id _C^{\otimes (n-i-1)}: C^{\otimes _\top n} \to C^{\otimes_\top (n+1)}
\end{equation}
\label{it:coalgId}
\item
We have
$(\Delta\otimes \Delta)(C\otimes_\top C)\subseteq C^{\otimes_\top 4}$;
\label{it:coalgmap1}
\item
$\Delta:(C,\top)\longrightarrow(C\otimes_\top C,\top_{\otimes 2})$ is a locality map, i.e.:
$(\Delta\times \Delta)(C\times_\top C)\subseteq (C\otimes_\top C)\times _\top (C\otimes_\top C)$;
\label{it:coalgmap3}
\item
For any locality linear map $\phi$ independent of $\Id$, $n\geq 2$ and $0\leq i \leq n$, we have
$\Id^{\otimes i}\otimes \phi\otimes \Id^{\otimes (n-i-1)}: C^{\otimes_\top n} \to C^{\otimes_\top n}$. \label{it:coalgmap2}
\end{enumerate}
 \label{lem:coalgmap}
 \end{lemma}
 \begin{proof}
 \begin{enumerate}
  \item Let $n\geq 2$ and $1\leq i\leq n$ be given. By the definition of $C^{\otimes_\top n}$, any of its elements is a finite sum of pure tensors $c_1\otimes \cdots \otimes c_n$ with $(c_1,\cdots,c_n)\in C^{_\top n}$. Let $U=\{c_j\,|\, j\not =i+1\}$. Then $c_{i+1}\in U^\top$,
so by Equation \eqref{eq:comag} there exist $(d_1, e_1), \cdots, (d_k,e_k)\in U^\top\times _\top U^\top$, such that
$$\Delta (c_{i+1})=\sum_{ \ell} d_\ell \otimes e_\ell.
$$
Now
$$(\Id _C^{\otimes i} \otimes \Delta \otimes \Id _C^{\otimes (n-i-1)})(c_1\otimes \cdots \otimes c_n)=\sum_{ \ell} c_1\otimes \cdots c_i \otimes d_\ell\otimes e_\ell\otimes c_{i+2}\otimes \cdots \otimes c_n,
$$
and $c_1\otimes \cdots c_i \otimes d_\ell\otimes e_\ell\otimes c_{i+2}\otimes \cdots \otimes c_n\in C^{\otimes_\top (n+1)}$.

\item Since $(\Delta\otimes \Delta)=(\Delta\otimes \Id)(\Id \otimes \Delta)$, from Equation \eqref{eq:gcoasso} we obtain
\begin{eqnarray*}
(\Delta\otimes \Delta)(C\otimes_\top C)&=& (\Delta\otimes \Id)(\Id \otimes \Delta) (C\otimes_\top C)\\
&\subseteq & (\Delta\otimes \Id)(C\otimes_\top C \otimes_\top C)\\
&\subseteq & C\otimes_\top C\otimes_\top C\otimes_\top C.
\end{eqnarray*}

\item Let $(c_1,c_2)\in C\times_\top C$. Then $c_2\in \{c_1\}^\top$. So by Equation \eqref{eq:comag}, $\Delta(c_2)=\sum_{(c_2)} c_{2,(1)}\otimes c_{2,(2)}$ with $c_{2,(1)}\top c_{2,(1)}$ and $\{c_{2,(1)}, c_{2,(2)}\}\subseteq \{c_1\}^\top.$ Thus $c_1 \in \{c_{2,(1)},c_{2,(2)}\}^\top$. By Equation \eqref{eq:comag} again, $\Delta(c_1)=\sum_{(c_1)} c_{1,(1)}\otimes c_{2,(2)}$ with $c_{1,(1)}\top c_{1,(2)}$ and $\{c_{1,(1)},c_{1,(2)}\} \subseteq \{c_{2,(1)},c_{2,(2)}\}^\top.$ This shows that $(c _{1,(1)},c _{1,(2)}, c _{2,(1)}, c _{2,(2)})$ is in $C^{_\top 4}$ and hence $ \big((c _{1,(1)}\otimes c _{1,(2)}), (c _{2,(1)}\otimes c _{2,(2)})\big)$ is in $\top_{\otimes2}$ since this relation is defined as a final locality relation.

\item Again any element of $C^{\otimes_\top n}$ is a sum of pure tensors $c_1\otimes \cdots \otimes c_n$ with $(c_1,\cdots,c_n)\in C^{_\top n}$. Thus $(c_1,\cdots,c_{i-1}, \phi(c_i),c_{i+1},\cdots,c_n)$ is in $C^{_\top n}$. This is what we want since $(\Id^{\otimes i}\otimes \phi\otimes \Id^{\otimes (n-i-1)})(c_1\otimes\cdots\otimes c_n)=
c_1\otimes\cdots\otimes c_{i-1}\otimes \phi(c_i)\otimes c_{i+1}\otimes\cdots\otimes c_n$.
 \end{enumerate}
\end{proof}
We now define and state properties of the reduced locality coproduct which will be of use in the sequel.
\begin{lemma} \label{lem:reduced_coproduct}
Let $(C=\oplus_{n\geq 0} C_n,\top,\Delta)$ be a connected locality coalgebra. Define the {\bf reduced coproduct}
$\tilde{\Delta}(c):= \Delta(c)- J\otimes c-c\otimes J$.
Recursively define
\begin{equation}
\tilde{\Delta}^{(1)}=\tilde{\Delta}, \quad \tilde{\Delta}^{(k)}:=\left(\Id\otimes \tilde{\Delta}^{(k-1)}\right)\tilde{\Delta}, \ \  k\geq 2.
\label{eq:redcoprod}
\end{equation}
\begin{enumerate}
\item
For $c\in \oplus_{n\geq 1}C_n$,
$\tilde \Delta (c)= \sum_{(c)} c'\otimes c''$ with $\deg(c'), \deg(c'')>0$ and $(c',c'')\in C\times_\top C$;
\label{it:conil-1}
\item
If in addition  $c\in U^\top $ for some $U\subset C$, then the above pairs
$(c', c'')$ are in $U^\top\times _\top U ^\top$;
\label{it:conil0}
\item
$\tilde\Delta^{(k)}(x)$ is in $C  ^{\otimes _\top (k+1)}$ for all $x\in  C, k\in\N$;
\label{it:conil1}
\item
$\tilde \Delta^{(k)}(   C  _n)=\{0\}$ for all $k\geq n.$
\label{it:conil2}
\end{enumerate}
\label{lem:conil}
\end{lemma}
\begin{proof}
We only need to prove the second point since the first one is the special case when $U=\{0\}$.

Let $c\in C_n\cap U^\top$. By {Eq. \eqref{eq:lconn}},
we can write
$$ \Delta(c)=y\otimes J + J\otimes z + \sum_{(c)}c'\otimes c''$$
with $y, z\in C_n$, $c', c''\in U^\top$ and each $c'\otimes c''\in C_p \otimes _ \top C_q, p+q=n, p, q\geq 1$. Then by the same argument for a connected coalgebra~\cite[Theorem~2.3.3]{Guo}, we obtain $y=z=x$. This proves (\ref{it:conil0}).

Then (\ref{it:conil1}) follows from an easy induction on $k$ by the locality property of $\Delta$; while the proof of (\ref{it:conil2}) is similar to the case without a locality structure~\cite[Proposition II.2.1]{Ma03}.
\end{proof}

\subsection{Locality bialgebras and locality Hopf algebras}
 
 As in the usual, non-locality case, we can now merge the locality algebra and locality coalgebra structures to obtain locality bialgebras.
 \begin{defn} \label{defn:locbialgebra}
\begin{enumerate}
\item
An {\bf locality bialgebra} is a sextuple $( B , \top, m, u, \Delta, \eps)$ consisting of a locality  algebra $( B , m, u, \top)$ and a locality coalgebra $\left( B , \Delta, \top, \eps\right)$
that are locality compatible in the sense that $\Delta$ and $\eps$ are locality algebra homomorphisms.
\item
{A} locality bialgebra $B$ is called {\bf connected} if there is a $\Z_{\geq 0}$-grading $B=\oplus_{n\geq 0} B_n$ with respect to which $B$ is both a locality graded algebra in the sense of Definition~\ref{defn:localisedalgebra} and a connected locality coalgebra in the sense of Definition~\ref{defn:colocalcoproduct}. {Then $J=1_B$.}
\end{enumerate}
\end{defn}
For cases in which the unit and counit play no role, we will write $( B , \top, m,  \Delta)$ for a locality bialgebra.

The last structure we need to introduce is the Hopf algebra. It is define through the convolution product which we now introduce.
\begin{defn} \label{def:conv_prod_loc}
Let $\left( C ,\top_C , \Delta \right)$   be a locality coalgebra and let $\left(A ,\top_A,m_A\right) $ be a locality algebra. Let ${\mathcal L}:=\Hom_{\rm loc}(C,A)$  be the space of locality linear maps.
  Define
    $$\top_{\mathcal L}:=\left\{ (\phi,\psi)\in \calL\times \calL\,|\, (\phi\times \psi)(C\times_\top C)\subseteq A\times_\top A\right\}.$$
 For $(\phi,\psi)\in\top_\calL$, define the {\bf convolution product} of $\phi$ and $\psi$ by
\begin{equation}
\phi\star \psi: C \stackrel{\Delta_C}{\longrightarrow} C\otimes_\top C \stackrel{\phi\otimes \psi}{\longrightarrow} A\otimes_\top A \stackrel{m_A}{\longrightarrow} A.
\label{eq:lconv}
\end{equation}
\end{defn}
We postpone to the next subsection the study of the convolution product. For now, let us just point out that for $(\phi, \psi)\in\top_{\mathcal L}$, since $(\phi\times \psi)(C\times_\top C)\subseteq A\times_\top A$, we have
 $$(\phi\otimes \psi)(C\otimes_\top C)\subseteq A\otimes_\top A.
 $$
Hence the composition in Equation \eqref{eq:lconv} is well-defined, giving a well-defined convolution product.
\begin{defn} \label{defn:LHopf}
A {\bf locality Hopf algebra} is a locality   bialgebra $\left( B ,\top, m, \Delta,u, \eps\right)$  with an antipode, defined to be a linear map $S:  B \to  B $
 such that $S$ and $\Id_ B  $ are mutually independent (in the sense of Definition \ref{defn:locallmap}) and
 \[S\star \Id=\Id\star S = u  \eps.\]
\end{defn}
The usual proof (see e.g.\cite{Guo,Ma03}) for the existence of the antipode on connected bialgebras extends to locality bialgebras as follows. For $k\geq 1$, denote $m_1=m$ and $m_k=m(\Id_B\otimes m_{k-1})$. We omit it here and simply recall its spirit. First prove.
\begin{lemma}
 Let $\left(   B ,\top, m, u,\Delta,\eps\right)$ be  a connected locality bialgebra, $\tilde\Delta^{\otimes k}$ as in Eq.~(\ref{eq:redcoprod}) and $\alpha:   B  \to   B $ a locality linear map with $\alpha (1_B)=0$. Then
\begin{enumerate}
\item
$\alpha^{\star k} = m_{k-1} \alpha^{\otimes k} \tilde\Delta^{(k-1)}$ for all $k\geq 2$;
\label{it:apower1}
\item
$\alpha^{\star k}(   B  _n)=\{0\}$ for all $k\geq n+1.$
\label{it:apower2}
\end{enumerate}
\label{lem:apower}
\end{lemma}
(The first point is easily proven by induction and the second is a direct consequence of the first one and Lemma~\ref{lem:conil}.(\ref{it:conil2}).

We then have the locality version of the Sweedler-Takeuchi formula \cite{takeuchi1971free}:
\begin{prop}
		 Let $\left(   B  ,\top , m, u,\Delta,\eps\right)$ be  a graded connected locality bialgebra.
There is  a linear map $S: {   B  }\to {   B }$  with the properties of the antipode stated above. It is given by
		  \begin{equation*}
		   S=\sum_{k=0}^\infty (u  \eps -\Id)^{\star k}.
		  \end{equation*}
\label{prop:localisedantipode}
\end{prop}
This is proven with the previous Lemma by setting $\alpha:   B  \to   B $  defined by $\alpha=\Id-u  \eps$, which is locality linear, and $\alpha(1_B)=0$. The geometric series   $S=\sum_{k=0}^{\infty}(-1)^k\alpha^{\star k}$  which is locally
  finite by Lemma~\ref{lem:apower}.(\ref{it:apower2}) and hence well-defined, gives  the inverse of the identity for the convolution product. We postpone to the next section the proof that $\alpha$ is a locality map and the $S$ is independent from the identity map.
% \end{proof}
 
 \subsection{Locality and the convolution product}
 
 We show that the locality (independence) of linear maps are preserved under the convolution product. We start with a Lemma that is of use to prove Proposition \ref{prop:localisedantipode}.
\begin{lemma}  Let $(C,\top_C,\Delta)$  be a locality coalgebra  with counit $\eps_C: C\longrightarrow  \K   $. Let $(A,\top_A )$  be a locality   algebra with unit $u_A:    \K   \longrightarrow A$.
The map  $e:=   u_A   \eps_{C}:C\longrightarrow A$ is independent to any linear map $\phi: C\longrightarrow A$. In particular, the map  $e$ is a locality linear map.
\label{lem:mutuallyindIA}
\end{lemma}
\begin{proof} This is because   $\mathrm{im}\, e=  \K  \cdot 1 _A \subset A ^{\top_A}$ as we can see from Remark \ref{rk:algebraunit}.
\end{proof}
We now prove the aforementioned results regarding convolution of locality maps. We state together the main points of \cite[Proposition 4.9 and Theorem 5.7]{CGPZ1}.
\begin{theo} \label{thm:loc_conv_prod}
 Let $\left( C,\top_C ,  \Delta\right)$ (resp. $( B , \top_B, m, \Delta)$)  be a locality coalgebra (resp. a locality bialgebra) and $\left(A,\top_A,\cdot\right) $ be a locality commutative  algebra. Let
 			\begin{equation*}
 			\phi,\psi: \left( C,\top_ C  \right)\longrightarrow   \left(A ,\top_A\right)
 			\end{equation*}
 			(resp. $\phi,\psi: \left( B,\top_B  \right)\longrightarrow   \left(A ,\top_A\right)$) be independent locality linear maps (resp. locality algebra homomorphisms).
 \begin{enumerate}
  \item $\phi\star \psi$ is a locality linear map (resp. a locality algebra homomorphism) and
the triple $({\mathcal L},\top_{\mathcal L}, \star)$ is a locality algebra.
\item If moreover    $C$ (resp. $B$) is connected  then
$$\calG \calL:=\{\phi\in \calL\ | \ \phi (J)=1_A\}$$ is  a locality  group for the convolution product.
\label{it:conv2}
% \item
% Under this same assumption, we have \[(\phi, \psi)\in\top_{\mathcal L}\cap\left(\calG \calL\times \calG \calL\right) \Longrightarrow (\phi^{\star k}, \psi^{\star l})\in\top_{\mathcal L}\cap \left(\calG \calL\times \calG \calL\right) \quad \forall k, l\in \Z.\]	
 			\item Resp., if $ B $ is connected and $\phi$ is a homomorphism of locality algebras, then  so is its convolution inverse $ \phi^{\star (-1)}$. So the set $\calG$ of homomorphisms of locality algebras from $(B,\top_B)$ to $(A,\top_A)$ is a locality group with respect to the independent relation of locality linear maps.
\label{it:conv3}
 \end{enumerate}
\end{theo}
\begin{proof}
 \begin{enumerate}
  \item We separate the three statements of this first point.
  \begin{itemize}
   \item We first verify that $\phi\star\psi$ is a locality linear map. For $ c  _1\top_C\,  c _2$,
 by Lemma~\ref{lem:coalgmap}.(\ref{it:coalgmap3}), there are finitely many $(d_i,e_i), (f_j,g_j)\in C\times_\top C$ with $(d_i, e_i, f_j, g_j)\in C^{_\top 4}$, such that
$$\Delta ( c  _1)=\sum_{i} d  _{i}\otimes  e  _{i} \quad \text{ and} \quad \Delta ( c  _2)=\sum_{j}   f_{j}\otimes  g  _{j}.
 $$
 Then
$$
 \phi\star\psi( c _1)=\sum_{i} \phi( d_i)\psi( e_i) \quad
 \text{and }\quad
 \phi\star\psi( c _2)=\sum_{j} \phi( f_j)\psi( g_j).
$$
From $(\phi, \psi)\in\top_{\mathcal L}$, we obtain
$\left(\phi( d_i),\psi( e_i),\phi( f_j),\psi( g_j)\right)\in A^{_\top 4}$. So
$$\big(\sum_i \phi( d_i)\psi( e_i),\sum_j \phi( f_j)\psi( g_j)\big)$$
is in $A\times_\top A$ and
thus $(\phi\star\psi( c _1))\top_A(\phi\star\psi( c _2))$.

\item Next we need to verify the axioms for a locality semigroup: the closeness of $U^{\top_{\mathcal L}}$ under the convolution product for every $U\subseteq {\mathcal L}$ and the associativity.

Let $\psi$ and $\chi$ be independent locality linear maps in $U^{\top_\calL}$ and let $\phi$ be in $U$. Then
$\phi,\psi,\chi$ are pairwise independent. Therefore
$$\phi\times \psi\times \chi: C^{_\top 3} \longrightarrow A^{_\top 3}$$
is well defined.
For $(c_1,c_2)\in C^{_\top 2}$, that is $c_1\in \{c_2\}^\top$, there exist
$(d_1, e_1), \cdots, (d_k,e_k)\in \{c_2\}^\top \times _\top \{c_2\}^\top $, such that
$$\Delta ( c  _1)=\sum_{i} d  _{i}\otimes  e  _{i},
 $$
 with
$(d_i, e_i, c_2)\in C^{_\top 3}$.
Then
$(\psi( d_i), \chi(e_i), \phi (c_2))\in A^{_\top 3}$ and hence $(\psi( d_i)\chi( e_i))\top_A \phi (c_2).$
So we have
$(\psi\star\chi)( c _1)\top_A \phi (c_2),
$
which means
$ \psi\star \chi$ is in $\phi^{\top_{\mathcal L}}.$
Thus $\psi \star \chi \in U^{\top_\calL}$. This verifies the first axiom. The associativity of $\star$ follows from the associativity of $m$ and coassociativity of $\Delta$ as in the classical case.

\item For $( c,  d)\in\top_ B $, by the proof of Lemma~\ref{lem:coalgmap}.(\ref{it:coalgmap3}), we can write
$$\Delta (c)=\sum_i c_{i1}\otimes c_{i2}, \quad \Delta (d)=\sum_j d_{j1}\otimes d_{j2}$$
with
$(c_{i1}, c_{i2}, d_{j1}, d_{j2})\in B^{_\top 4}.$
Then
$$\Delta (cd)=(m\otimes m)\tau _{23}(\Delta \otimes \Delta )(c\otimes d)=\sum_{i,j} c_{i1}d_{j1}\otimes c_{i2}d_{j2}.
$$
So
\begin{eqnarray*}
 			(\phi\star \psi) ( c\,  d)&=&
 			\sum_{i,j} \phi\left(c_{i1}d_{j1}\right)\,\psi\left( c_{i2}d_{j2}\right) \\
 			&=&\sum_{i,j} \phi\left(  c_{i1} \right)\, \phi\left( d_{j1}\right)\,\psi\left(  c_{i2}\right)\,\psi\left(  d_{j2}\right) \\
 			&=&\sum_{i} \phi\left(  c_{i1}\right)\,\psi\left(  c_{i2}\right)\, \sum_{j} \phi\left( d_{j1}\right)\,\psi\left(  d_{j2}\right) \\
 			&=& \left(\phi\star \psi( c)\right)\,\left(\phi\star \psi( d)\right).
 			\end{eqnarray*}
  \end{itemize}
  \item For the second point, the bialgebra structure plays no role, so we will omit the various ``(resp.'' in this item.
  
  We assume that $ C $  is a connected   locality coalgebra.
For a locality linear map $\phi : C \to A$, we now prove by induction on the degree of $ c  _1$ that the map
\begin{equation}
\phi ^{\star(-1)}( c  _1)=\left\{\begin{array}{ll}
1_A, & c_1=J, \\
-\phi ( c _1)-\sum_{( c  _1)} \phi ( c  _1^\prime) \phi ^{\star(-1)}(  c  _1^{\prime \prime}), & c_1\in \ker \eps, \end{array} \right .
\label{eq:phirec}
\end{equation}
is well defined, and that $c _1\top_C\,  c$ implies $\phi ^{\star(-1)}( c  _1)\top_A \phi ( c ).$

{This} is trivial for degree $0$ since $\phi^{\star (-1)}(J)=1_A$. Assume for any $ c  _1\in C$ of degree $\le n$, $\phi ^{\star(-1)}( c  _1)$ is well defined, and for $c$ with $ c  _1\top_C\,  c $,
$\phi ^{\star(-1)}( c _1)\top_A \phi ( c )$ holds.

Now for any $c_1$ of degree $n+1\ge 1$ with $ c  _1\top_C\,  c  $, by Lemma~\ref {lem:conil}.(\ref{it:conil0}), we have
 $$\Delta ( c  _1)= c  _1\otimes J+J\otimes  c  _1+\sum _{( c  _1)} c  _1^\prime \otimes  c  _1^{\prime \prime}
\ \text{ with }
 (c  _1',  c  _1^{\prime \prime}, c)\in C^{_\top 3}.
 $$
By the induction hypothesis, $\phi ^{\star(-1)}( c _1^{\prime\prime})$ is well defined, such that $\phi ^{\star(-1)}( c _1^{\prime\prime})\top_A \phi (c)$ and $\phi ^{\star(-1)}( c _1^{\prime\prime})\top_A \phi (c_1^\prime)$. Since $\phi$ is a locality linear map, we also have
$\phi ( c  _1)\top_A \phi ( c  )$ and $ \phi ( c  _1^\prime )\top_A \phi ( c  ).$
Thus
$\phi ( c  _1^\prime) \phi ^{\star(-1)}(  c  _1^{\prime \prime})$ is well defined and
$(\phi ( c  _1^\prime) \phi ^{\star(-1)}(  c  _1^{\prime \prime}))\top_A \phi ( c ).
$
So, $\phi ^{\star(-1)}(  c  _1)$ is well defined and
$\phi ^{\star(-1)}( c  _1)\top _ A \phi ( c )$, which means $\phi\top_\calL \phi ^{\star(-1)}$.

Again by induction on the degree of $c_1$, we now  prove that
$\phi ^{\star(-1)}$ is a locality linear map by checking
\begin{equation}
\phi ^{\star(-1)}( c  _1)\top_A \phi ^{\star(-1)}( c  _2)\quad \forall c_2\in C, c_1\top_C c_2,
\label{eq:philoc}
\end{equation}
a fact which is obvious at degree 0. Assume that, for a given $n\geq 0$ and any $c_1$ of degree $\le n$
the equation holds. Consider $ c  _1 $ of degree $n+1\ge 1$. Since $c_1\top_C c_2$, we can choose
 $$\Delta ( c  _1)= c  _1\otimes J+J\otimes  c  _1+\sum _{( c  _1)} c  _1^\prime \otimes  c  _1^{\prime \prime},
 $$
such that
$ \{c_1,c_1^\prime, c_1''\}\top_C\, c_2.$
From this we have $\{\phi ( c  _1), \phi ( c  _1'),
\phi ^{\star(-1)}( c  _1'')\}\top_A\, \phi ^{\star(-1)} ( c  _2).$
So Eq.~(\ref {eq:phirec}) gives
$\phi ^{\star(-1)}( c  _1)\top_A \phi ^{\star(-1)}( c  _2).$
Therefore, we conclude that $\calG \calL $ is a locality group with unit $u_A \varepsilon_C$ by Lemma \ref{lem:mutuallyindIA}.

\item For this last item, we use an induction on the sum of degrees of $c$ and $d$, $c\top d$ to prove
$$\phi^{\star(-1)}(c)\phi^{\star(-1)}(d)=\phi^{\star(-1)}(cd),
$$
which is true if the sum of degrees is 0.

In general, by Lemma~\ref{lem:conil}.(\ref{it:conil0}), we write
$$\Delta (c)=c\otimes J +J\otimes c+\sum _{(c)}c'\otimes c^{\prime \prime}
, \quad
\Delta (d)=d\otimes J +J\otimes d+\sum _{(d)}d'\otimes d^{\prime \prime}
$$
with
$(c', c'', d', d'')\in B^{_\top 4}.$
So by $\Delta(cd)=\Delta(c)\Delta(d)$, we obtain
\begin {eqnarray*}
\Delta (cd)&=&cd\otimes J+J\otimes cd +c\otimes d+d\otimes c+\sum _{(d)}cd'\otimes d^{\prime \prime}+\sum _{(d)}d'\otimes cd^{\prime \prime}\\
&&+\sum _{(c)}c'd\otimes c^{\prime \prime}+\sum _{(c)}c'\otimes c^{\prime \prime}d+
\sum _{(c)(d)}c'd'\otimes c^{\prime \prime}d''.
\end{eqnarray*}
By Eq.~(\ref {eq:phirec}) we obtain
\begin{eqnarray*}\phi ^{\star(-1)}( cd  )&=&-\phi (cd)-\phi (c)\phi ^{\star(-1)}(d)-\phi (d)\phi ^{\star(-1)}(c)\\
&&-\sum _{(d)}\phi (cd')\phi ^{\star (-1)}( d^{\prime \prime})-\sum _{(d)}\phi (d')\phi ^{\star (-1)}( cd^{\prime \prime})\\
&&-\sum _{(c)}\phi (c'd)\phi ^{\star (-1)}( c^{\prime \prime})-\sum _{(c)}\phi (c')\phi ^{\star (-1)}( c^{\prime \prime}d)\\
&&-\sum _{(c)(d)}\phi (c'd')\phi ^{\star (-1)}( c^{\prime \prime}d'').
\end{eqnarray*}

By Eq.~(\ref {eq:phirec}) applied to $c$ and $d$, the locality multiplicativity of $\phi $, the commutativity of $A$ and induction hypothesis, we have
\begin{eqnarray*}\phi ^{\star(-1)}( cd  )&=&\phi(c)\phi (d)+\sum _{( d )}\phi (c)\phi ( d  ^\prime) \phi ^{\star(-1)}(  d  ^{\prime \prime})+\sum _{( c )}\phi ( c  ^\prime)\phi (d) \phi ^{\star(-1)}(  c  ^{\prime \prime})\\
&&+\sum _{( c )(d)} \phi ( d  ^\prime) \phi ^{\star(-1)}(  d  ^{\prime \prime})\phi ( c  ^\prime) \phi ^{\star(-1)}(  c  ^{\prime \prime})
+\sum _{( c )(d)}\phi ( c  ^\prime) \phi ^{\star(-1)}(  c  ^{\prime \prime})\phi ( d  ^\prime) \phi ^{\star(-1)}(  d  ^{\prime \prime})\\
&&-\sum _{(c)(d)}\phi (c'd')\phi ^{\star (-1)}( c^{\prime \prime}d'')\\
&=&\phi(c)\phi (d)+\sum _{( d )}\phi (c)\phi ( d  ^\prime) \phi ^{\star(-1)}(  d  ^{\prime \prime})+\sum _{( c )}\phi ( c  ^\prime)\phi (d) \phi ^{\star(-1)}(  c  ^{\prime \prime})\\
&&+\sum _{( c )(d)}\phi ( c  ^\prime) \phi ^{\star(-1)}(  c  ^{\prime \prime})\phi ( d  ^\prime) \phi ^{\star(-1)}(  d  ^{\prime \prime})\\
&=& \big(\phi ( c )+\sum _{( c )}\phi ( c  ^\prime) \phi ^{\star(-1)}(  c  ^{\prime \prime})\big)\big(\phi ( d )+\sum _{( d )}\phi ( d  ^\prime) \phi ^{\star(-1)}(  d  ^{\prime \prime})\big)\\
&=&\phi ^{\star(-1)}( c  )\phi ^{\star(-1)}( d).
\end{eqnarray*}
This completes the induction.
 \end{enumerate}
\end{proof}
Notice that the proof of Proposition \ref{prop:localisedantipode} follows with an easy induction from the results of this Theorem.

 \subsection{The locality Birkhoff-Hopf factorisation}
 
 \begin{theo} \label{thm:abflhopf}
  {\bf (Algebraic Birkhoff  factorisation, locality Hopf algebra version)}
Let $\left( H ,\top_H \right)$   be a locality connected Hopf algebra, $H=\oplus_{n\geq 0} H _n$, $H _0= \K  e$.
Let $\left(A,\top_A,\cdot \right) $ be a commutative locality algebra with decomposition $A=A_1\oplus A_2$ as a vector space such that the linear projections $\pi_i$ onto $A_i$ along $A_{\hat{i}}, \{\hat{i}\}:=[2]\backslash \{i\}, i=1,2,$  are independent locality linear maps and $1_A$ is in $A_1$.
Let
\begin{equation*}
 \phi: \left(H ,\top_H \right)\longrightarrow   \left(A,\top_A\right)
\end{equation*}
be a locality algebra homomorphism.
Then there are unique independent locality algebra homomorphisms $\phi_i: H\to \K+A_i$ with $\phi_i (\ker \eps)\subseteq A_i, i=1,2, $ such that

\begin{equation}
\phi= \phi_1^{\star (-1)} \star \phi_2.
\label{eq:lhabf}
\end{equation}
The map $  \phi_1^{\star(-1)}$ is also a locality algebra homomorphism and $\phi_1\top_\calL \{\phi,\phi_2\}$, $\phi_1^{\star(-1)}\top_\calL\{\phi_1, \phi_2\}$.

If $\psi: \left(C ,\top_C \right)\longrightarrow   \left(A,\top_A\right)
$
is also a locality linear map independent of $\phi$ with $\psi(J)=1_A$, then $\phi_i$ and $\psi_j$ are independent for $i, j =1,2$.
\begin{enumerate}
 \item If in addition $A_1$ is a sub-locality algebra of $A$, then $\phi_1^{\star (-1)}:C\to  \K  + A_1$.
 \item If in addition $A_2$ is a locality ideal of $A$, then $\phi_1^{\star(-1)}=\pi_1 \phi$
and $\phi_2$ is recursively given by
\begin{equation} \label{eq:phi_2}
 \phi_2(1_H)=1_A, \quad
\phi_2(c)=(\pi_2  \phi)(c)-\sum_{(c)}(\pi_1  \phi)(c')\phi_2(c'') \quad\forall c\in \ker \eps.
\end{equation}
with $\tilde\Delta(c)=\sum_{(c)}c'\otimes c''$ the reduced coproduct defined in Lemma \ref{lem:reduced_coproduct}.
\end{enumerate}
 \end{theo}
 \begin{rk}
  All the statements that make sense if $H$ is only a graded connected locality coalgebra (i.e. all the statements except that $\phi_i$ and $\phi_1^{\star(-1)}$ are algebra homomorphisms) are true in the category of locality \ty{coalgebras}. We do not state this level of generality since it will not be so useful later. The interested reader is refered to \cite[Theorem 4.10]{CGPZ1} for this purely coalgebraic case.
 \end{rk}

 \begin{proof} We prove separately the various statements of the theorem.
 \begin{itemize}
  \item \underline{existence of the $\phi_i$:} Let $n\geq 1$ and $ c  \in H_n$. Since $H$ is a connected locality Hopf algebra, we can write
$$\Delta ( c  )=J\otimes  c  + c  \otimes J+\sum _{( c  )} c '\otimes c ''
$$
with $\deg(c'), \deg(c'')>0$ and $c'\top_H c''$.

We first prove by induction on the degree $n$ of $c$ that the map given by
\begin {equation}
\phi_1( c )=\left\{\begin{array}{ll}
1_A, & c=J, \\
-\pi _1\Big(\phi( c )+\sum_{( c )} \phi_1( c ')\phi( c  '')\Big),
& c\in H_n, n>0, \end{array} \right .
\label{eq:LPhi-}
\end{equation}
is well-defined, and for any $d\in H$ with $d\top _H  c $, there is
$$\phi _1( c )\top _A \phi (d),
$$
which clearly hold for $ c $ of degree $0$.

Assume that these hold true for $ c $ of degree less or equal to $n$. Then for $ c  $ of degree $n+1$, according to Lemma \ref{lem:reduced_coproduct} $ c  '$ is of degree less or equal to $n$, so $\phi _1(c')$ is defined and $\phi _1(c')\top_A \phi ( c  '')$. Therefore $\phi _1( c  ')\phi ( c  '')$ makes sense by the induction hypothesis, and $\phi _1( c )$ is well-defined.

Now for any $ c \top_H d$, we have $\phi (c)\top_A \phi_1(d)$. By a similar induction on the degree of $c$, we obtain
$ \phi_1( c )\top_A \phi _1(d),
$
so
$\phi _1$ is a locality linear map.
Therefore, the map
\begin{equation}
\phi_2(c):=\left\{\begin{array}{ll}
1_A, & c= J, \\
\pi_2\big(\phi(c)+\sum_{(c)} \phi_1(c')\phi(c'')\big), &
c\in H_n, n>0, \end{array} \right .
\label{eq:LPhi+}
\end{equation}
is well-defined.

Notice that for $c\in C_n, n>0$,  Equation \eqref{eq:LPhi+} means
\begin {equation}
\phi_2(c)=\phi (c)+ \phi _1(c)+\sum_{(c)} \phi_1(c')\phi(c'').
\label{eq:LPhi2}
\end{equation}
With the condition on $J$, this in turn reads
$\phi _2=\phi _1\star \phi$ and hence $\phi =  \phi_1^{\star(-1)}\star \phi _2.$

\item \underline{uniqueness of the $\phi_i$:} The proof of the uniqueness of the maps $\phi_i$ is the same as the proof \cite[Theorem 4.4]{GPZ2} for the case of a trivial locality relation i.e., when $\top= C\times C$.

\item \underline{Independence properties of $\phi_i$ and $\phi_1^{\star(-1)}$:} We perform the same induction than in the first point. For any $d \in H$ with $ c \top_H d$, we can take $\{c',c''\}\top_H d$. Since $\phi $ is a locality map, we obtain
$\phi ( c  )\top_A \phi (d)$ and $\phi ( c''  )\top_A \phi (d)$ . Also the induction hypothesis gives
$\phi _1( c  ')\top_A \phi (d).$
Thus
$(\phi _1( c  ')\phi ( c  ''))\top_A \phi (d)$ since $(A,\top_A,.)$ is a locality algebra.
Therefore,
$$\Big(\phi( c )+\sum_{( c )} \phi_1( c ')\phi( c  '')\Big)\top_A \phi (d).
$$
Now since $\pi _1$ is a locality map and $\pi _1 $ and $\pi _2$ are independent for any $c\top_H d$ we have $\pi_1(c)\top_A\pi_1(d)$ and $\pi_1(c)\top_A\pi_2(d)$. Therefore $\pi_1(c)\top_A(\pi_1(d)+\pi_2(d))=d$ since $\Id _A=\pi _1+\pi _2$. Thus
$\pi _1$ and $\Id _A$ are independent. Thus
$\phi_1( c )\top_A \phi (d).
$
Therefore we have proved that $\phi _1\top_\calL \phi$.

From Equation \eqref{eq:LPhi2}, we easily obtain
$\phi _1\top_\calL \phi _2.
$
By Equation \eqref{eq:phirec}, an easy induction on the degree of $c$ shows that
$\phi_1^{\star(-1)}\top_\calL\, \{\phi _1, \ \phi _2\}.$

A similar induction shows that if $\psi: \left(C ,\top_C \right)\longrightarrow   \left(A,\top_A\right)$
is also a locality map, independent of $\phi$ with $\phi(J)=1_A$, then $\phi_i$ and $\psi_j$ are independent for $i, j =1,2$, proving the last statement of the theorem.

\item \underline{$\phi_i$ and $\phi_1^{\star(-1)}$ are locality algebra homomorphisms:} For $c\top_H d$, by Lemma~\ref{lem:conil}.(\ref{it:conil0}), we can write
$$\Delta (c)=c\otimes J +J\otimes c+\sum _{(c)}c'\otimes c^{\prime \prime}, \quad
\Delta (d)=d\otimes J +J\otimes d+\sum _{(d)}d'\otimes d^{\prime \prime}
$$
with
$(c', c'', d', d'')\in H^{_\top 4}.
$
According to Proposition \ref{prop:multpi}, $\pi_1$ is a locality Rota-Baxter operator of weight $-1$. Using \eqref{eq:rbo}, by a similar argument as in the non-locality case~\cite[Theorem~2.4.3]{Guo}, we can prove  that $\phi _1$ and $\phi _2$ are homomorphisms of locality algebras. Then, by Theorem \ref{thm:loc_conv_prod}, \ref{it:conv3}, $ \phi_1^{\star(-1)}$ is also a homomorphism of locality algebras.

\item \underline{$\phi_1^{\star(-1)}$ takes values in $\K+A_1$:} Assume that $A_1$ is a sub-locality algebra.  Equation \eqref{eq:phirec}) and a simple induction on $n\geq 0$ show
that $\phi_1^{\star (-1)}(c) \in  \K  + A_1$ for any $c\in H_n$.

\item \underline{Formulas for $\phi_1^{\star(-1)}$ and $\phi_2$:} Assume further that $A_1$ is a sub-locality algebra and $A_2$ is a locality ideal. We prove by induction on $n\geq 0$ that
\begin{equation}
	\phi_1^{\star-1}({c})  = (\pi_1  \phi)({ c})\quad \forall c\in  H_n.
\label{eq:BHF3}
	\end{equation}
	
Notice  that $\phi(J)=\phi_1(J)=1_A$ implies $(\phi_1\,\star\, ( \pi_1  \phi))(J)=1_A$, so Equation \eqref{eq:BHF3} holds for $n=0$
since $H_0=  K   J$. Assuming  that Equation \eqref{eq:BHF3} holds for any $c$ in $H$ of degree $\leq n$, we prove that Equation \eqref{eq:BHF3} holds for any element $c\in  H_{n+1}$.

We write
\begin{equation*}
 \Delta (c) ={c}\otimes J+ J\otimes {c}+\sum_{({c})}{c}'\otimes {c}''
\end{equation*}
with ${c}'\top_C {c}''$ of degree $\le n$.
 By the definition of $\phi_1 $,
  \begin{equation*} \phi_1({c})  =  -\pi_1\left(\phi(c)+  \sum_{({c})}\phi_1({c}^\prime)\, \phi({c}^{\prime\prime})\right)\end{equation*}
and we have already proven that
$\phi_1({c}^\prime)\top_A \phi({c}^{\prime\prime})$. We have also already proven that $\pi_1$ is independent to $\Id_A$ and this \ty{straightforwardly} implies that $\pi_2$ is also independent to $\Id_A$. Thus
$\phi_1({c}^\prime)\top_A \{\pi _1\phi({c}^{\prime\prime}), (\pi _2\phi)({c}^{\prime\prime})\}.$
Then
\begin{eqnarray*}
 \phi_1({c}) &=&-\pi_1\Big(\phi(c)- \sum_{({c})} \phi_1({c}^\prime) \big((\pi _1\phi)({c}^{\prime\prime})+(\pi _2 \phi)({c}^{\prime\prime})\big)\Big)\\
 &=&-(\pi_1\phi)(c)- \sum_{({c})} \phi_1({c}^\prime) (\pi _1\phi)({c}^{\prime\prime})
\end{eqnarray*}
 where we have used the facts that $A_1$ is a locality subalgebra, so $\phi_1({c}^\prime) \pi _1\phi({c}^{\prime\prime}\in A_1$ and that $A_2$ is a locality ideal so $\phi_1({c}^\prime)\pi _2 \phi({c}^{\prime\prime})\in A_2$.
 
 Consequently,
\begin{equation*}
 \big(\phi_1 \star(\pi_1  \phi)\big)({c})=\phi_1({c})+\pi_1\phi(c)+ \sum_{({c})} \phi_1({c}^\prime) \,\pi_1  \phi({c}^{\prime\prime})=0.
\end{equation*}
 We conclude that $\phi_1\,\star\,\left(\pi_1 \phi\right)=e$, leading to $ \phi_1^{\star(-1)} = \pi_1  \phi$ since $\phi_1^{\star(-1)}$ has been shown to exist. The {locality algebraic Birkhoff factorization}  $\phi =\phi_1^{\star -1} \star \phi_2$  then yields  for $c{\in \ker}(\varepsilon)$:
 \begin{equation*}
  \pi_1(\phi(c)) + \phi_2(c) + \sum_{(c)}\pi_1(\phi(c'))\phi_2(c'') = \phi(c)
 \end{equation*}
  Using $\pi_2={\Id}-\pi_1$ gives   the recursive expression \eqref{eq:phi_2} for $\phi_2$ in terms of $\pi_1$ and $\pi_2$.
 \end{itemize}
 \end{proof}
 \begin{rk} 
 %{\color{red} XXXXX check no redites avec l'intro!!!}
 As in the usual (non locality) case, the decomposition $\phi=\phi_1^{\star(-1)}\star\phi_2$ is a {\bf Birkhoff-Hopf decomposition} of the map $\phi$. Let us repeat the introduction: one striking result of Connes and Kreimer \cite{CK1} is that the the combinatorics of renormalisation is actually encoded by such a decomposition. In particular, if $\phi=\Phi_\tau$ is the unrenormalised Feynman rules of some QFT $\tau$, then $\phi_1^{\star(-1)}$ is their renormalised counterpart, while $\phi_2$ gives the counter-terms. Since this decomposition still exists in the context of locality structure we can define a renormalisation scheme with multiple regularisation variables. The striking fact that $\phi_1^{\star(-1)}$ is a locality map is summed up by the motto
  \begin{quote}
   ``Renormalisation preserves locality''.
  \end{quote}
  One possibly even more surprising result is that this renormalised map takes a particularly simple form in this framework of locality structures:
  \begin{equation} \label{eq:renom_map}
   \phi_1^{\star(-1)}=\pi_1\phi.
  \end{equation}
  This is sharp contrast with the non locality case, where the renormalised map is given by a complicated BPHZ \cite{BP57, He66, Zi70} formula. This can be understood as what we obtain for the price of working with locality structures instead of usual total structures. 
  
  To be more specific, let us say that the need for a Birkhoff-Hopf factorisation, or BPHZ formula (as opposed to a minimal \ty{subtraction}) is because we need to take into account the inner parts of a Feynman graph to obtain a renormalised map that is an algebra morphism, i.e. that preserves locality. In the locality framework, this is done by construction, by requiring for example that the vertices of the Feynman graphs are decorated by elements of a locality set that are pairwise independent. \cy{In practice, one needs to implement this locality framework using a suitable locality algebra for the regularisation. The simplest example seems to be an algebra of multivariate meromorphic germs: the multivariate renormalisation scheme.}
 \end{rk}
 Applying Theorem \ref{thm:abflhopf} to $(A,\top_A, \cdot) =\left(\calM_\Q, \perp^Q,\cdot\right)$  yields the following result, which is the application of the tools defined above to build a multivariate renormalisation scheme. 
 
%  {\color{red} XXXX plus tard???? Une nouvelle sous-section??? Pas si je mets la preuve en annexe...}
\begin{cor} \label{co:abflhopf}
Let $(H,\top_H)$   be a connected locality Hopf algebra.
	Let
	\begin{equation*}
	\phi: \left( H  ,\top_ H  \right)\longrightarrow   \left(\calM_\Q, \perp^Q\right)
	\end{equation*}
	be a locality  linear map  such that $\phi(1_H)=1_{\mathcal M _\Q}$.
	Let $\phi=(\phi_1^Q)^{\star (-1)}\star \phi_2^Q$ be the algebraic Birkhoff  factorisation in Eq.~\eqref{eq:lhabf} with $\phi_1^Q(1_H)=\phi_2^Q(1_H)=1_{\mathcal M_\Q}$. Then
	\begin{enumerate}
\item
$\pi_1^Q \phi$ is a locality linear map;
\item $(\phi_1^Q)^{\star (-1)}=\pi_1^Q  \, \phi$  so that
\begin{equation} \phi=(\pi_1^Q\,  \phi)\star \phi_2^Q;
\label{eq:BHlocalpi}
\end{equation}
\item  the maps $\phi_1^Q$, $\phi_2^Q$ are locality linear maps  and $\phi_1^Q\top_\calL \phi$,  $\pi_1^Q\, \phi\top_\calL \phi_2^Q$;
\item assuming in addition that  $\phi$ is a locality algebra homomorphism, then the maps
    $\pi_1^Q\phi, \phi_1^Q$ and $\phi_2^Q$ are locality algebra homomorphisms.
		\end{enumerate}
\end{cor}
\begin{proof}  The proof is straightforward; let us nevertheless mention that $\pi_1^Q\, \phi\top_\calL \phi_2^Q$ follows from $\phi_1\top_\calL \phi_2$ in Equation \ref{eq:lhabf} combined with the fact that the convolution inverse preserves locality. The crucial fact that $(\phi_1^Q)^{\star (-1)}=\pi_1^Q  \, \phi$ comes from the fact that $\calM_{+}$ is a locality subalgebra of $\calM_\Q$ and $\calM_{-}^Q$ is a locality ideal of $\calM_\Q$ (Proposition \ref{pp:merodecomp}).
\end{proof}
This is our multivariate renormalisation scheme: given a connected graded locality Hopf algebra $(H,\top,\Delta,.)$ and a locality map $\phi:H\longrightarrow\calM_\Q$, the renormalised map is simply $\pi_+^Q\phi$.
 
\section{Application to Kreimer's toy model} \label{sec:Kreimer}

We apply our multivariate renormaisation scheme to Kreimer's toy model, introduced in \cite{Kr97} which has some of the essential  features of QFT (see \cite{KK13}). We first need the locality version of operated structures.

\subsection{Locality operated structures}

We now merge operated structures (Definition \ref{def:op_structures}) with locality. These definitions are from \cite{CGPZ2}.
\begin{defn} \label{defn:loc_op_set}
Let $(\Omega, \top)$ be a locality set. An {\bf $(\Omega,\top)$-operated locality set} or simply a {\bf locality operated set} is a locality set $(U,\top_U)$ together with a {\bf partial action} $\beta$ of $\Omega$ on $U$ on a subset
$\top_{\Omega,U}:=\Omega\times_\top U\subseteq \Omega \times U$

$$ \beta:  \Omega\times_\top U \longrightarrow U, \ (\omega, x)\mapsto \beta ^\omega (x), $$
satisfying the following compatibility conditions.
	   \begin{enumerate}
\item
For
$$\Omega \times_\top U  \times_\top U: = \{(\omega,u,u')\in \Omega\times U\times U\,|\,
 (u,u')\in \top_U, (\omega, u), (\omega , u')\in \Omega \times_\top U\},$$
the map $\beta \times \Id_U: (\Omega\times_\top U) \times U \longrightarrow U\times U$ restricts to
$$\beta \times \Id_U: \Omega \times_\top U  \times_\top U
\longrightarrow  U\times_\top U.$$
 In other words,  {if} $(\omega,u,u')$ lies in $\Omega \times_\top U  \times_\top U$, then  $(\beta^\omega(u),u')$ lies in $\top_U$.
  \item
For
$$\Omega\times_\top \Omega \times_\top U:=\{(\omega,\omega',u)\in \Omega \times \Omega \times U\,|\, (\omega,\omega')\in \top_\Omega, (\omega,u), (\omega',u)\in \Omega\times_\top U\},$$
the map $\Id_\Omega \times \beta: \Omega\times (U \times_\top U) \longrightarrow \Omega \times U$  {restricts to}
$$ \Id_\Omega \times \beta: \Omega\times_\top \Omega \times_\top U \longrightarrow \Omega \times_\top U. $$
In other words, if $(\omega,\omega',u')$ lies in $\Omega \times_\top \Omega  \times_\top U$, then  $(\omega,\beta^{\omega'}(u))$ lies in $\Omega\times_\top U$.
	\end{enumerate}
 \end{defn}
 We can now define locality operations on more sophisticated locality structures.
 \begin{defn} \label{defn:basedlocsg}
Let $(\Omega , \top )$ be a locality set.
 \begin{enumerate}
  \item A {\bf locality $(\Omega,\top)$-operated semigroup} is a  quadruple  $\left(U,\top _U, \beta ,m_U\right)$, where $(U,\top _U,m_U)$ is a locality semigroup and
  $\left(U,\top _U, \beta \right)$ is a $(\Omega,\top)$-operated locality set such that
  \begin{equation}
    (\omega, u, u')\in\Omega\times_\top U\times_\top U~\Longrightarrow~(\omega, uu')\in\Omega\times_\top U;
  \end{equation}
 \item A {\bf locality $(\Omega,\top)$-operated monoid} is a  quintuple  $\left(U,\top _U, \beta ,m_U,1_U\right)$, where $(U,\top _U,m_U,1_U)$ is \ty{a} locality monoid
  and $\left(U,\top _U, \beta, m_U \right)$ is a $(\Omega,\top)$-operated locality semigroup such that $\Omega \times 1_U\subset \Omega \times _\top U$.
    \item A {\bf $(\Omega,\top)$-operated locality nonunitary algebra } (resp. {\bf $(\Omega,\top)$-operated locality unitary algebra}) is a  quadruple  $\left(U,\top _U, \beta ,m_U\right)$ (resp. quintuple
    $(U,\top _U, \beta,$ $m_U, 1_U)$) which is a locality algebra (resp. unitary algebra) and a locality $(\Omega,\top)$-operated semigroup (resp. monoid), satisfying the additional condition that for any $\omega\in \Omega$, the set $\{\omega\}^{\top_{\Omega,U}}:=\{ u\in U\,|\, \omega\top_{\Omega,U} u \}$ is a subspace of $U$ on which the action of $\omega$ is linear. More precisely, the last condition means that for any $u_1, u_2\in \{\omega\}^{\top_{\Omega,U}}$ and for any $k_1, k_2\in \K$, we have $k_1u_1+k_2u_2\in \{\omega\}^{\top_{\Omega,U}}$ and $\beta^\omega(k_1u_1+k_2u_2)=k_1\beta^\omega(u_1)+k_2\beta^\omega(u_2)$
(resp. this condition and $\Omega \times 1_U\subset \Omega \times _\top U$).
\end{enumerate}
In each case, the $(\Omega,\top)$-operated structure is called {\bf commutative} if the corresponding locality structure is.
\end{defn}
Notice that \cy{the structures introduced above depend on the chosen set $\top_{\Omega,U}$. Therefore this set should appear in the names of 
 these various locality operated structures to be completely rigourous.} We drop this requirement as it would unnecessarily complicate notations since \ty{the set $\top_{\Omega,U}$} will always be clear from context.
\begin{example}
A locality \, semigroup $(G,\top,\cdot) $ is a locality $(G, \top)$-operated commutative semigroup for the action $\alpha: G\times_\top G\to G$ given by the product on $G$.
\end{example}
We still have to define categories of locality operated locality structures in order to obtain categories.
\begin{defn}\label{defn:morphoplocstr}
Given $(\Omega,\top _\Omega)$-operated locality structures  (sets, semigroups, monoids, nonunitary algebras, algebras) $(U_i, \top _{U_i}, \beta_i)$, $i=1,2,$ a
{\bf morphism of locality operated  locality structures}  is a locality morphism (of sets, semigroups, monoids, nonunitary algebras, algebras) $f: U_1\to U_2$ such that
\begin{itemize}
\item
$(\Id_\Omega\times f)(\Omega\times_\top U_1)\subseteq \Omega\times_\top U_2$
and
\item  $f\circ \beta_1^\omega = \beta_2^{\omega} \circ f$ for all $\omega\in \Omega$.
\end{itemize}
\label{defn:morphismoperatedloc}
\end{defn}
In the following part of this section, we will focus on a locality (Hopf) algebra which has a universal property in the category of locality operated locality algebras. It is described in terms of rooted forests.

\subsection{The locality Hopf algebra of properly decorated forests}

Recall (Definition \ref{def:forests}) that a rooted forest (or forest for short) is a directed acyclic graph such that there is at most one path between any two vertices and also such that each connected component has a minimal element, called a root. A connected forest is called a (rooted) tree.  Recall also that for a set $\Omega$, an $\Omega$-decorated (rooted) forest is a pair $(F,d_F)$ with $F$ a (rooted) forest and $d_F:V(F)\longrightarrow\Omega$ a map from the vertices of $F$ to the set $\Omega$. If the set $\Omega$ is clear from the context, we will call these objects decorated forests. In general, we will also omit the $d_F$ and called $F$ a decorated forest. We write $\calF_\Omega$ the vector space freely generated by $\Omega$-decorated forests and $\calT_\Omega$ the vector space freely generated by $\Omega$-decorated trees. The concatenation of forests turns $\calF_\Omega$ into a commutative and associative unital algebra with the empty tree as neutral element.

We recall a classical (and important) Hopf algebraic \ty{structure} on $\calF_\Omega$, namely the Connes-Kreimer Hopf algebra. This structure was introduced in \cite{CK99} to understand the combinatorics of renormalisation. We also refer the reader to \cite{Fo13} for an accessible introduction to this beautiful topic and \cite{CK1} for applications to QFT.
\begin{defn} \cite{CK99}
 Let $F$ be a rooted forest. Set $E'(F):=E(F)\cup\{e_r|r\text{ a root of }F\}$, with $e_r$ a half-edge only linked to the root $r$. A {\bf cut} of $F$ is a subset $c$ of $E'(F)$. An {\bf admissible cut} of $F$ is a cut $c\subseteq E'(F)$ such that any path (see Definition \ref{def:graph_stuff}) of $F$ intersects $c$ at most once. We write $\Adm(F)$ the set of admissible cuts of $F$.
 
 For any cut $c$ of $F$, we write $R_c(F)$ the {\bf stump} given by $c$, namely the subforest of $F$ whose vertices and edges are below all the elements of $c$ and $P_c(F)$ the {\bf pruning} given by $c$, namely the subforest of $F$ whose vertices and edges are above all the elements of $c$.
\end{defn}
The R in $R_c(F)$ comes from ``roots''. Indeed, if $F=T$ is a rooted tree with unique root $r$ then $R_c(T)$ contains $r$ for any admissible cut $c\neq\{e_r\}$. Our definition is succinct since admissible cuts are a standard object in the combinatorics of rooted forests and we did not want to spend more time than strictly necessary to introduce these well-known notions. These definitions also naturally extend to decorated forests.

The coproduct of the Connes-Kreimer Hopf algebra can be described in terms of admissible cuts.
\begin{defn} \label{def:CK_coproducts} \cite{CK99}
 For any set $\Omega$, the {\bf Connes-Kreimer coproduct} is a map $\calF_\Omega\Delta_{CK}:\calF_\Omega\longrightarrow\calF_\Omega\otimes\calF_\Omega$ defined by its action of any rooted forest
 \begin{equation*}
  \Delta_{CK}(F)=\sum_{c\in\Adm(F)}R_c(F)\otimes P_c(F).
 \end{equation*}
\end{defn}
Notice that $\emptyset\in\Adm(F)$ and $R_\emptyset(F)=F$, $P_\emptyset(F)=\emptyset$. Furthermore, $c_{\rm full}:=\{e_r|r\text{ a root of }F\}\in\Adm(F)$ and $R_{c_{\rm full}}(F)=\emptyset$, $P_{c_{\rm full}}(F)=F$. Thus for any $F\neq\emptyset$
\begin{equation*}
  \Delta_{CK}(F)=\emptyset\otimes F+F\otimes\emptyset+\tilde\Delta_{CK}(F)
 \end{equation*}
i.e. $\Delta_{CK}$ is counital. \\

Let us now turn our attention to the case where the decoration set $\Omega$ carries a locality relation. \cy{In this case $\calF_\Omega$ inherits a locality relation: two forests $F_1$ and $F_2$ are independent if, and only if, each decoration of $F_1$ is independent to every decorations of $F_2$, and vice-versa.}
\begin{defn} \label{def:loc_rooted_forests}%\cite{CGPZ2} 
 Let $(\Omega,\, \top _\Omega )$ be a locality set and equip the space $\calF _{\Omega}$ of ${\Omega}$-decorated rooted forests  with the following independence relation induced by that of $\calP (\Omega)$ given by Example \ref{ex:power_set}:
\begin{align} \label{eq:IndForestss}
 					(F_1, d _1)\, \top _{\calF_\Omega}\,(F_2, d _2)& \Longleftrightarrow d_1(V(F_1))\, \top _\Omega\, d _2(V(F_2)) \\
 					& \Longleftrightarrow\forall(v_1,v_2)\in V(F_1)\times V(F_2),~d_1(v_1)\top_\Omega\, d_2(v_2) \nonumber
 \end{align}
 and extended by linearity.
\end{defn}
Since $\top _{\calF_\Omega}$ is defined on the whole of $\calF_\Omega$ by a linear extention, it is clear that $(\calF_\Omega,\top _{\calF_\Omega})$ is a locality vector space. It actually carries further (locality) structures but we are more interested by a subspace.
\begin{defn} \label{def:prop_dec_forests}
Let $(\Omega,\top_\Omega)$ be a locality set. A {\bf properly} $\Omega$-decorated forest is a $\Omega$-decorated forest $(F,d_F)$ such that any disjoint pair of vertices of $F$ are decorated by independent elements of $\Omega$:
 \begin{equation*}
 \forall(v_1,v_2)\in V(F)^2,~v_1\neq v_2\Longrightarrow d_F(v_1)\top_\Omega d_F(v_2).
\end{equation*}
 We write $\mathcal{F}_\Omega^{\rm prop}$ the set of properly $\Omega$-decorated forests. When $\Omega$
is clear from context, we call them properly decorated forests.
\end{defn}
Since properly decorated rooted forests are decorated forests, the usual concatenation and Connes-Kreimer coproducts are defined on properly decorated rooted forests. We then have the following simple albeit important result.
\begin{theo} \label{thm:loc_Hopf_prop_forests}
 Let $(\Omega,\top_\Omega)$ be a locality set. Then $(\calF_\Omega^{\rm prop},\top _{\calF_\Omega},.,\Delta_{CK})$ (with $.$ the usual concatenation of rooted forests) is a locality bialgebra. It is graded by the number of vertices and connected, hence a locality Hopf algebra.
\end{theo}
\begin{proof}
We start by proving the locality properties.
 \begin{itemize}
  \item $\calF_\Omega^{\rm prop}$ is a \ty{vector subspace} of $\calF_\Omega$ carrying the same locality relation, therefore it is a loclaity vector space.
  \item For any pair of properly decorated forest $F_1$ and $F_2$ with $F_1\top _{\calF_\Omega} F_2$, let $(v_1,v_2)\in V(F_1F_2)\times V(F_1F_2)$. Let $d:V(F_1F_2)\longrightarrow \Omega$ be the decoration map of $F_1F_2$ inherited from the decorations maps of $F_1$ and $F_2$. If $v_1$ and $v_2$ are both elements of $V(F_i)$ for $i=1$ or $2$, then $d(v_1)\top_\Omega d(v_2)$ since $F_i$ is properly decorated. Otherwise, we still have $d(v_1)\top_\Omega\, d(v_2)$ since $F_1\top _{\calF_\Omega} F_2$. Therefore $F_1F_2$ is properly decorated and $\calF_\Omega^{\rm prop}.\calF_\Omega^{\rm prop}\subseteq\calF_\Omega^{\rm prop}$.
  \item Let $F_1$ and $F_2$ be two independent properly decorated rooted forests that are both independent to $F_3$, a third properly decorated forest. Since $V(F_1F_2)=V(F_1)\sqcup V(F_2)$, and since $F_1\top _{\calF_\Omega} F_3$ and $F_2\top _{\calF_\Omega} F_3$; the decoration of any vertex of $F_1F_2$ is independent to the decoration of any vertex of $F_3$. Thus, since $\top _{\calF_\Omega}$ is defined through a linear extention, Property \eqref{eq:semigrouploc} holds and $(\calF_\Omega^{\rm prop},\top _{\calF_\Omega},.)$ is a locality algebra.
  \item Let $F\in\calF_\Omega^{\rm prop}$ and let $c\in\Adm(F)$. Since $F$ is properly decorated, any vertex of $R_c(F)$ will be independent to any vertex of $P_c(F)$. For the same reason both $R_c(F)$ and $P_c(F)$ are properly decorated rooted forests. Thus, by linearity of the locality tensor product, $\Delta_{CK}(F)\in\calF_\Omega^{\rm prop}\otimes_\top\calF_\Omega^{\rm prop}$\footnote{we use $\otimes_\top$ as a shorthand notation for the cumbersome $\otimes_{\top_{\calF_\Omega}}$.} i.e. $\Delta_{CK}(\calF_\Omega^{\rm prop})\subseteq\calF_\Omega^{\rm prop}\otimes_\top\calF_\Omega^{\rm prop}$. Furthermore, let $F_1$ and $F_2$ be two independent properly decorated rooted forests. Then since $F_1\top _{\calF_\Omega} F_2$, and $V(R_c(F_1))\subseteq V(F_1)$, $V(P_c(F_1))\subseteq V(F_1)$, the decoration of any vertex of $R_c(F_1)$ and $P_c(F_1)$ is idenpendent to the decoration of any vertex of $F_2$. Thus, again by linearity Property \eqref{eq:comag} holds and $(\calF_\Omega^{\rm prop},\top _{\calF_\Omega},\Delta_{CK})$ is a locality coalgebra.
 \end{itemize}
 
 The properties of usual (non locality) Hopf algebra come from the fact that the Connes-Kreimer Hopf algebra is a Hopf algebra and from Proposition \ref{prop:localisedantipode}.
\end{proof}
\begin{rk}
 A non-commutative version of properly decorated rooted forests also exists but will not be needed in this context.
\end{rk}

\subsection{A universal property}

Recall that the {\bf grafting operator} is the map 
\begin{equation*}
 B_+:\Omega\times\calF_\Omega\longrightarrow\calF_\Omega
\end{equation*}
defined by, for any $\omega\in\Omega$ and $F\in\calF_\Omega$, $B_+(\omega,F)$ is the rooted tree obtained by adding a root decorated by $\omega$ to the forest $F$ and grafting each roots of $F$ to this new root. We write $B_+^\omega(F)$ instead of $B_+(\omega,F)$ and for $\omega\in\Omega$ fixed, we have $B_+^\omega:\calF_\Omega\longrightarrow\calF_\Omega$ the related map.

We have a simple yet important property of properly decorated forests.
\begin{prop} \label{prop:operatedalgebra}
 Let $(\Omega,\top _\Omega)$ be a locality set.
 Then
 $(\calF_{\Omega}^{\rm prop},\top _{\calF_{\Omega}}, B_+,\cdot,\emptyset)$
  is a locality $(\Omega,\top _\Omega)$-operated commutative algebra where $B_+$ is the grafting operator restrained to the set 
  \begin{equation*}
   \top_{\Omega,\calF_\Omega}:=\{(\omega,F)\in\Omega\times\calF_\Omega^{\rm prop}~|~\forall v\in V(F),\omega\top_\Omega d_F(v)\}.
  \end{equation*}
\end{prop}
\begin{proof}
 The grafting operator $B_+$ satisfies by definition of $\top _{\calF_{\Omega}}$ the conditions of a
locality $(\Omega,\top _\Omega)$-operated algebra. Let us check them one by one. First, by the definitions of $\top_{\Omega,\calF_\Omega}$ and properly decorated forests, if $(\omega,F)\in\top_{\Omega,\calF_\Omega}$, then $B_+^{\omega}(F)\in\calF_\Omega^{\rm prop}$.

Second, if $F_1\top_{\calF_\Omega} F_2$ and $(\omega,F_i)\in\top_{\Omega,\calF_\Omega}$ for $i\in\{1,2\}$ then from the definition of $\top_{\calF_\Omega}$ and $\top_{\Omega,\calF_\Omega}$ we have $B_+^\omega(F_1)\top_{\calF_\Omega} F_2$ as needed.

Third, if $\omega_1\top_\Omega\omega_2$ and $(\omega_i,F)\in\top_{\Omega,\calF_\Omega}$ for $i\in\{1,2\}$ then again from the definition of $\top_{\calF_\Omega}$ and $\top_{\Omega,\calF_\Omega}$ we have $(\omega_1,B_=^{\omega_2}(F))\in\top_{\Omega,\calF_\Omega}$ as needed. Therefore, $(\calF_{\Omega}^{\rm prop},\top _{\calF_{\Omega}},B_+)$ is a $(\Omega,\top_\Omega)$-operated locality set.

The fact that it is actually a $(\Omega,\top_\Omega)$-operated locality \emph{algebra} comes directly from the definition of the concatenation product of decorated forests and the definition of $\top_{\Omega,\calF_\Omega}$.
\end{proof}
As it should be expected, this example is not random: properly decorated forests are actually the initial object in the category of locality operated locality structures. First, here is a straightforward consequence of the  definition of locality operated structures.
\begin {lemma}
Let $(U,\top_U)$ be an $(\Omega,\top_\Omega)$-operated locality set. For $m, n\geq 1$, denote
$$\Omega^{_\top m} \times_\top U^{_\top n}: =\left\{ (\omega_1,\cdots,\omega_m,x_1,\cdots,x_n)\in \Omega^m\times U^n\,\left |\,
\begin{array}{l}
(\omega_1,\cdots,\omega_m)\in \Omega^{_\top m}\\
 (x_1,\cdots,x_n)\in U^{_\top n} \\
 \omega_i\top_{\Omega,U}x_j\quad\forall(i,j)\in[m]\times[n]
 \end{array} \right . \right\}.$$
$ (\Id_\Omega^{m-1} \times \beta^{\omega_m} \times \Id_U^{n-1})(\Omega^{_\top m} \times_\top U^{_\top n})$ is contained in $\Omega^{_\top (m-1)}\times_\top U^{_\top n}$.

With a similar notation, we have
${(\beta \times \beta)} (\Omega\times_\top U\times_\top \Omega\times_\top U)\subseteq U\times_\top U.$
\label{lem:itact}
\end{lemma}
Next, we have a lemma that will allow us to perform inductions on the number of vertices of properly decorated rooted forests.
\begin{lemma}
Let $(\Omega,\top _{\Omega})$ be a locality set. An ${\Omega}$-properly decorated rooted forest in $\calF_{\Omega}^{\rm prop}$ is either the empty tree $1$ or it can be written uniquely in one of the following forms:
\begin{enumerate}
\item
$(F_1,d_1)\cdots (F_n,d_n), \ n\ge 2,$ with rooted trees $F_i\neq=\emptyset$ such that
$$(F_i,d_i)\top _{\calF _{\Omega}} (F_j,d_j),\ 1\leq i\neq= j\leq n. $$
Furthermore $\deg(F,d)=\deg(F_1,d_1)+\cdots+\deg(F_n,d_n)$;
\item	
$ B_+^{\omega}(F,d)$ for some $(F,d) \in \calF_{\Omega}^{\rm prop}$ which is independent of  $\bullet_\omega.$ Furthermore,
$$\deg(B_+^{\omega}(F,d))=\deg (F,d)+1.$$
\end{enumerate}
\label{lem:fdecomp}
\end{lemma}
\begin{proof}
The statements hold for any decorated forest without any independence requirement~\cite{Guo}. If a decorated forest has independent decorations, then the independence conditions in the statements automatically hold from the definition of the various locality relation at play.
\end{proof}
 The following results are also easy to verify.
\begin{lemma}
\begin{enumerate}
\item \label{lem:welldefined_point_i}
For rooted forests $(F_1,d_1),\cdots, (F_n,d_n)\in \calF_{\Omega}^{\rm prop}$, set as before $d_{F_1\cdots F_n}:V(F_1)\bigsqcup\cdots\bigsqcup V(F_n)\longrightarrow\Omega$ to be the map whose restriction to $V(F_i)$ is $d_i$. Then
the product forest 
\begin{equation*}
 (F_1,d_1)\cdots(F_n,d_{n})=(F_1\cdots F_n,d_{F_1\cdots F_n})
\end{equation*}
lies in $\calF_{\Omega}^{\rm prop}$ if and only if
$(F_i,d_i)\top _{\calF _{\Omega}} (F_j,d_j)$ for $1\leq i\neq j\leq n$;
\item
For a properly decorated rooted forest $(F,d)\in \calF_{\Omega}^{\rm prop}$, the rooted tree $B_+^{\omega}(F,d)$ lies in $\calF_{\Omega}^{\rm prop}$ if and only if $(F,d)$ is independent of  $\bullet_\omega.$
\end{enumerate}
\label{lem:welldefined}
\end{lemma}
 The following characterisation of a morphism of operated commutative monoids will later be useful.
\begin{lemma}
Let $U$ be an $(\Omega,\top_\Omega)$-operated commutative monoid. A map $\phi:\calF_{\Omega}^{\rm prop} \to U$ is a morphism of $(\Omega,\top_\Omega)$-operated commutative monoids if and only if
\begin{enumerate}
\item
$\phi(1) = 1_{U}$;
\label{it:mid}
\item for any $(F,d),(F',d') \in \calF_{\Omega}^{\rm prop}$, if $(F,d) \top _{\calF_\Omega}(F',d')$, then $\phi(F,d) \top _U \phi(F',d')$;
\label{it:mloc}
\item
For any $(T_1,d_1)\cdots (T_n,d_n)\in \calF_{\Omega}^{\rm prop}$ where $(T_1,d_1), \cdots, (T_n,d_n)$ are decorated rooted trees, the equation
$$\phi ((T_1,d_1)\cdots (T_n,d_n)) = \phi(T_1,d_1)\cdots\phi(T_n,d_n)$$
holds;
\label{it:mprod}
\item If $(\omega,(F,d))$ is in $\Omega\times_\top \calF_{\Omega}^{\rm prop}$, then $(\omega, \phi (F,d))$ is in $\Omega\times _\top U$;
\label{it:mact}
\item
For any $B_+^\omega(F,d)\in \calF_{\Omega}^{\rm prop}$, the equation
$\phi \left( B_+^\omega(F,d)\right) =  \beta_{U,+}^{\omega}\left(\phi(F,d)\right)$ holds.   \label{it:mgraft}
 \end{enumerate}
 \label{lem:char}
\end{lemma}
\begin{proof}
($\Longrightarrow$). Suppose  that $\phi:\calF_{\Omega}^{\rm prop}\to U$ is a morphism of operated commutative monoids. Then conditions~(\ref{it:mloc}), (\ref{it:mact}) and (\ref{it:mgraft}) hold by the locality of the map
$\phi$ and its compatibility with the actions of $(\Omega,\top_\Omega)$. Condition (\ref{it:mid}) is the unitary condition and Condition (\ref{it:mprod}) follows since the concatenation is the product in
$\calF_{\Omega}^{\rm prop}$.

($\Longleftarrow$). Now suppose that all the conditions are satisfied, so that we only need to verify that $\phi$ is multiplicative
for any $(F_1,d_1), (F_2,d_2)\in \calF_{\Omega}^{\rm prop}$ with
$(F_1,d_1)\top_{\calF_\Omega} (F_2,d_2)$.

By Lemma~\ref{lem:fdecomp}, we have decompositions
$$ (F_i,d_i)=({T}_{i,1},d_{i,1})\cdots ({T}_{i,n_i},d_{i,n_i}), i=1,2,$$
of $(F_i,d_i)$ into rooted trees. Furthermore, by Lemma~\ref{lem:welldefined}, the concatenation $(F_1,d_1)(F_2,d_2)$ is well-defined in $\calF_{\Omega}^{\rm prop}$ and then the decomposition of $(F_1,d_1)(F_2,d_2)$ in Lemma~\ref{lem:fdecomp} is
$$(F_1,d_1)(F_2,d_2)=({T}_{1,1},d_{1,1})\cdots ({T}_{1,n_1},d_{1,n_1}) ({T}_{{2},1},d_{1,1})\cdots ({T}_{{2},n_{2}},d_{{2},n_{2}}).$$
This gives the multiplicativity of $\phi$:
\begin{eqnarray*}
\phi((F_1,d_1)(F_2,d_2))
&=&\phi(({T}_{1,1},d_{1,1})\cdots ({T}_{1,n_1},d_{1,n_1}) ({T}_{{2},1},d_{1,1})\cdots ({T}_{{2},n_{2}},d_{{2},n_{2}}))\\
&=& \phi({T}_{1,1},d_{1,1})\cdots \phi({T}_{1,n_1},d_{1,n_1}) \phi({T}_{{2},1},d_{1,1})\cdots \phi({T}_{{2},n_{2}},d_{{2},n_{2}}) \\
&=& \phi(F_1,d_1)\phi(F_2,d_2),
\end{eqnarray*}
as required.
\end{proof}
Finally, let us state a useful result that is \ty{straightforwardly} derived from the definition of locality semigroups (Definition \ref{defn:localisedalgebra}).
\begin{lemma}
Let $(G,m_G,\top_G)$ be a locality semigroup. Let $k\geq  2$ and $1\leq i\leq k$. For $(x_1,\cdots,x_k)\in G^{_\top k}$ we have
\begin{enumerate}
\item
$(\Id_G^{i-1}\times m_G\times \Id_G^{k-i-1})(x_1,\cdots,x_k)\in G^{_\top (k-1)}.$
\item
$(x_1\cdot \ldots \cdot x_i, x_{i+1}\cdot \ldots \cdot x_k)\in G\times_\top G,$.
\end{enumerate}
\label{lem:locprod}
\end{lemma}
We are finally ready to prove
\begin{theo} \label{thm:univ_prop_trees_loc}
 Let a locality set $(\Omega,\top _{\Omega})$ be given. The quintuple $(\calF_{\Omega}^{\rm prop},\top _{\calF _{\Omega}},  B_+, \cdot, 1 )$ is the initial object in the category of
$(\Omega,\top _{\Omega})$-operated commutative locality algebras. More precisely, for any $(\Omega,\top_\Omega)$-operated commutative locality \ty{algebra} $U:=(U,\top_U,\beta_U,m_U,1_U)$, there is a unique morphism $\phi:=\phi_U: \calF_{\Omega,\top_\Omega}\to U$ of $(\Omega,\top_\Omega)$-operated commutative locality monoids.
\end{theo}
\begin{proof}
Let an $(\Omega,\top_\Omega)$-operated commutative locality algebra $(U,\top_U,\beta_U,m_U,1_U)$ be given. We only need to prove that there is a unique morphism of
$(\Omega,\top _{\Omega})$-operated commutative locality algebras
$$\phi=\phi_U: (\calF_{\Omega}^{\rm prop},\top _{\calF _{\Omega}},  B_+, \cdot, 1 )\longrightarrow	 (U,\top_U,\beta_U,m_U,1_U).$$
By Lemma~\ref{lem:char}, we only need to prove that there is a unique map
$\phi:\calF_{\Omega}^{\rm prop} \to U$ satisfying the conditions (\ref{it:mid}) --- (\ref{it:mgraft}).
For $k\geq 0$, we set
$$\calF_k: =\{ (F,d)\in \calF_{\Omega}^{\rm prop}\, |\, |F|\leq k\}.$$
We will prove by induction on $k\geq 0$ that there is a unique map $\phi_k: \calF_k\to U$ fulfilling the five conditions in
Lemma~\ref{lem:char} when $\calF_{\Omega}^{\rm prop}$ is replaced by $\calF_k$, in which case, we call the corresponding conditions
Condition~$(j)_k$ for $j\in \{\ref{it:mid},~ \ref{it:mloc},~ \ref{it:mprod},~ \ref{it:mact},~ \ref{it:mgraft}\}$.

When $k=0$, we have $\calF_k=\{\emptyset\}$. Then only Condition~(\ref{it:mid}) and Condition~(\ref{it:mloc}) apply when $(F,d)=(F',d')=\emptyset$, giving the unique map
$$\phi_0:\calF_0\to U, \quad \emptyset\mapsto 1_U.$$
Since $(1_U,1_U)\in U\times_T U$, Condition~(\ref{it:mloc}) is satisfied.

Another  instructive example to study before the inductive step is the case   $k=1$, so
$\calF_1=\{1\}\cup \{\bullet_\omega\,|\, \omega\in \Omega\}$. Since $\bullet_\omega=B_+^\omega(\emptyset)$, the only map $\phi:\calF_1\to U$ satisfying Conditions (\ref{it:mid}) --- (\ref{it:mgraft}) is given by
$$ \phi(1)=1_U, \quad \phi(\bullet_\omega)=\phi(\beta^\omega(\emptyset)) =\beta^\omega(\phi(\emptyset))=\beta^\omega(1_U).$$
Now let $k\geq 0$ be given and assume that there is a unique map $\phi_k: \calF_k\to U$ satisfying Conditions~$(\ref{it:mid})_k$ --- $(\ref{it:mgraft})_k$. Consider $f=(F,d)\in \calF_{k+1}$. If $f$ is already in $\calF_k$, then $\phi(f)$ is uniquely defined by the induction hypothesis. If $f\in \calF_{k+1}$ is not in $\calF_k$, then the degree of $f$ is at least $1$. So Lemma~\ref{lem:fdecomp} shows that either there is a factorisation $f=f_1\cdots f_n, n\geq 2,$ into independent properly decorated rooted trees, or
$f=B_+^\omega(\overline{f})$ for $\overline{f}$ independent of $\bullet_\omega$ and necessarily in $\calF_k$. The assignment
\begin{equation}
\phi_{k+1}(f)=\left\{\begin{array}{ll} \phi_k(f_1)\cdots \phi_k(f_n), & \text{ if } f = f_1\cdots f_n, \\
\beta^\omega(\phi_k(\overline{f})), & \text{ if } f = B_+^\omega(\overline{f}). \end{array} \right.
\label{eq:mphik+1}
\end{equation}
is then well-defined  since $\phi_{k}$ satisfies Conditions $(\ref{it:mid})_{k}$ --- $(\ref{it:mgraft})_{k}$ by assumption.

Note that this is in fact the only way to define $\phi_{k+1}$ satisfying the conditions in Lemma~\ref{lem:char}, proving the uniqueness of $\phi_{k+1}(f)$.

Next we verify that the map $\phi_{k+1}$ obtained this way indeed satisfies Conditions $(\ref{it:mid})_{k+1}$ --- $(\ref{it:mgraft})_{k+1}$. By the above equation and the inductive hypothesis, $\phi_{k+1}$ satisfies Conditions~$(\ref{it:mid})_{k+1}$, $(\ref{it:mprod})_{k+1}$ and $(\ref{it:mgraft})_{k+1}$.

To verify Condition~$(\ref{it:mloc})_{k+1}$, consider $f, f'\in \calF_{k+1}$ with $f \top_{\calF_{\Omega}} f'$.
Depending on whether $f$ or $f'$ lies or not in $\calF_k$, there are four cases to consider.
In the case when both $f$ and $f'$ are in $\calF_k$, the condition is satisfied by the induction hypothesis. For the remaining three cases, the verifications are similar, the most complicated one being when neither $f$ nor $f'$ lies in $\calF_k$. So we will only verify this case. For this case, we further have four subcases depending on which of the two forms in Lemma~\ref{lem:fdecomp} that $f$ or $f'$ takes.

{\bf Subcase 1. $f=f_1\cdots f_n, n\geq 1,$ for independent properly decorated trees $f_1,\cdots, f_n$ and $f'=f_1'\cdots f'_{n'}$ for independent properly decorated trees $f'_1,\cdots, f'_{n'}$.} Then all the factor trees are pairwise independent. Since all the factor trees are in $\calF_k$, by the inductive assumption of Condition~$(\ref{it:mloc})_k$, their images
$$\phi_k(f_i), 1\leq i\leq n, \quad \phi_k(f'_j), 1\leq j\leq n',$$
are pairwise independent in $U$. Furthermore the products
$\phi_k(f_1)\cdots \phi_k(f_n)$ and $\phi_k(f'_1)\cdots \phi_k(f'_{n'})$ are independence thanks to  Lemma \ref{lem:locprod}
By Lemma \ref{lem:locprod}, the products
% $\phi_k(f_1)\cdots \phi_k(f_n)$ and $\phi_k(f'_1)\cdots \phi_k(f'_{n'})$ are independent.
But by the construction of $\phi_{k+1}$ in Equation \eqref{eq:mphik+1}, the last two products equal to $\phi_{k+1}(f_1\cdots f_n)$ and $\phi_{k+1}(f'_1\cdots f'_{n'})$. This gives Condition~$(\ref{it:mloc})_{k+1}$ in this subcase.

{\bf Subcase 2.  $f=B_+^\omega(\overline{f})$ for $(\omega,\overline{f})\in \top_{\Omega,\calF_{\Omega}}$ and $f'=f'_1\cdots f'_n, n\geq 2,$ for independent properly decorated trees $f'_1,\cdots, f'_n$.} Since $n\geq 2$, we have $f'=f'_A f'_B$ with independent $f'_A$ and $f'_B$, both in $\calF_k$. Since $f$ and $f'$ are independent, we have
that $\overline{f},~f'_A$ and $f'_B$ are pairwise independent and the $(\omega,f)\in\top_{\Omega,\calF_{\Omega}}$ for $f$ being any one of these three forests.
Then Conditions~$(\ref{it:mloc})_k$ and $(\ref{it:mact})_k$ lead to $\phi_k(\overline{f}),~\phi_k(f'_A)$ and $\phi_k(f'_B)$ being pairwise independent and $(\omega,\phi_k(f))\in\top_{\Omega,\calF_{\Omega}}$ again for $f\in\{\overline{f},f'_A,f'_B\}$.

Applying Lemma~\ref{lem:itact} gives $(\beta^\omega(\phi_k(\overline{f})), \phi_k(f'_A)\phi_k(f'_B)) \in U \times_\top U$. By Equation \eqref{eq:mphik+1}, this implies that $(\phi_k(\beta^\omega(\overline{f})), \phi_k(f'_Af'_B))$ is in $U \times_\top U$
% , that is,
% $(\phi_k(\beta^\omega(\overline{f})), \phi_k(f'_Af'_B))$ is in $U \times_\top U$, by Eq.~(\ref{eq:mphik+1}).
This gives Condition~$(\ref{it:mloc})_{k+1}$ in this subcase.
 
 {\bf Subcase 3. $f=f_1\cdots f_n, n\geq 2,$ for independent properly decorated trees $f_1,\cdots, f_n$ and $f'=B_+^\omega(\overline{f}')$ for $(\omega,\overline{f}')\in \Omega\times_\top \calF_{\Omega,\top_\Omega}$.} This subcase follows from the previous subcase by the commutativity of the concatenation of the forest product and the locality relation.
 
 {\bf Subcase 4. $f=B_+^{\omega_1}(\overline{f_1})$ and $f'=B_+^{\omega_2}(\overline{f_2})$ for $(\omega_i,\overline{f_i})\in \Omega\times_\top \calF_{\Omega,\top_\Omega}$ ($i\in\{1,2\}$).}
Since the two forests are independent, we have
$(\omega_i,\overline{f_j})\in\top_{\Omega,\calF_\Omega}$ for $i,j\in\{1,2\}$, $\omega_1\top_\Omega\omega_2$ and $\overline{f_1}\top_{\calF_\Omega}\overline{f_2}$.

Then the locality of $\phi_k$, in particular Condition~$(\ref{it:mact})_k$,  guaranteed by the induction hypothesis, gives $(\omega_i,\phi(\overline{f_j}))\in\top_{\Omega,\calF_\Omega}$ for $i,j\in\{1,2\}$ and $\phi(\overline{f_1})\top_{\calF_\Omega}\phi(\overline{f_2})$. This
yields, by Lemma~\ref{lem:itact} and Equation \eqref{eq:mphik+1},
$$(\phi(B_+^\omega(\overline{f})), \phi(B_+^{\omega'}(\overline{f}'))) =(\beta^\omega(\phi(\overline{f})), \beta^{\omega'}(\phi(\overline{f}')))\in
U\times_\top U.$$
This gives Condition~$(\ref{it:mloc})_{k+1}$ in this subcase.

We have therefore completed the verification of Condition~$(\ref{it:mloc})_{k+1}$.

\medskip 

Let us finally check Condition~$(\ref{it:mact})_{k+1}$ assuming the induction hypothesis, distinguishing two cases:  $f\in \calF_{k+1}$ is of the form $f=f_1\cdots f_n, n\geq 2,$ for independent properly decorated trees $f_1,\cdots, f_n\in \calF_k$ or $f=B_+^{\omega'}(\overline{f})$ for $(\omega',\overline{f})\in \Omega\times_\top \calF_k$.

In the first case, we write $f=f_A f_B$ with $f_A, f_B\in \calF_k$. Then from $(\omega,f)\in \top_{\Omega,\calF_\Omega}$ we have $(\omega,f_A)\in \top_{\Omega,\calF_\Omega}$, $(\omega,f_B)\in \top_{\Omega,\calF_\Omega}$ as well as $f_A\top_{\calF_\Omega}f_B$. This implies $(\omega,\phi(f_A))\in \top_{\Omega,U}$, $(\omega,\phi(f_B))\in \top_{\Omega,U}$ and $\phi(f_A)\top_{\calF_\Omega}\phi(f_B)$.
 This gives
$$ (\omega,\phi_{k+1}(f_Af_B)) =(\omega,\phi(f_A)\phi(f_B))\in \top_{\Omega,U}$$ as needed.

In the second case, similarly we have $\omega\top_\Omega\omega'$, $(\omega,\overline{f}')\in\top_{\Omega,\calF_\Omega}$ and $(\omega',\overline{f}')\in\top_{\Omega,\calF_\Omega}$. This implies $(\omega,\phi(\overline{f}'))\in\top_{\Omega,\calF_\Omega}$ and $(\omega',\phi(\overline{f}'))\in\top_{\Omega,\calF_\Omega}$.
Therefore,
$(\omega,\phi(B_+^{\omega'}(\overline{f}'))) =(\omega,\beta^{\omega'}(\phi(\overline{f}')))\in\top_{\Omega,U}$ by the assumption on $\beta$.

\medskip 

This completes the verification of Condition~$(\ref{it:mact})_{k+1}$. Together with the verification of the other conditions for the existence of $\phi_{k+1}$ above, as well as that of the uniqueness of $\phi_{k+1}$ after Equation \eqref{eq:mphik+1} and the inductive step is completed.
\end{proof}
\begin{rk}
 For planar (i.e. non-commutative) properly decorated rooted forests, the same result would hold for non commutative locality operated locality algebra. The same proof holds in this case. We omit it since we will not need this level of generality.
\end{rk}

\subsection{Multivariate renormalisation of Kreimer's toy model}

Recall (Subsection \ref{subsec:multi_mero}) that endowing $\R^\infty$ with an Euclidean structure $Q=(Q_k(., .))_{k\geq 1}$
where for $k\geq1$, 
$Q_k(.,.): \R ^k\otimes \R ^k \to \R$ induces isomorphisms $Q_k^*:\R^k\longrightarrow(\R^k)^*$ defined by
\begin{equation*}
 Q_k^*:u\mapsto\Big(Q_k(u,.):v\mapsto Q_k(u,v)\Big).
\end{equation*}
Composing the maps $Q_n^*$ and their inverses with the canonical injections $j_n:\R^n\hookrightarrow\R^{n+1}$ gives rise to a direct system $(\iota_n^*)_{n\geq0}$ defined by
\begin{equation*}
 \iota_n^*:=Q_{n+1}^*\circ j_n\circ(Q_n^*)^{-1}:(\R ^n)^*\longrightarrow (\R ^{n+1})^*.
\end{equation*}
Then we set
$$\calL: =\varinjlim_n (\R^n)^*.$$
Let us recall that $\calM_\Q$ is the space of germs with linear poles rational coefficients, see Equation \eqref{eq:def_MQ}. In the algebra of germs with complex variables in $\C^\infty$ and with one real variable $x$, we consider the $\calM_\Q$-module of linear combinations
$$\Big\{ \sum_{i=1}^k f_i x^{L_i} \,\Big|\, f_i \in \calM, L_i\in \calL, 1\leq i\leq k, k\geq 1\Big\}.$$
It is an $\calM_\Q$-subalgebra since $x^L, L\in \calL$ is closed under multiplication. It can further be checked that $x^{L}, L\in \calL,$ are linearly independent over $\calM_\Q$. Thus it is  isomorphic to the group ring $\calM_\Q[\calL]$ over $\calM_\Q$ generated by the additive monoid $\calL$. We will   henceforth make this identification.

Let us also recall that $Q=(Q_k(., .))_{k\geq 1}$ also induces locality relations (both written $\perp^Q$) on $\R^\infty$ and $\calL$. This in turns extend to a locality relation, also written $\perp^Q$, on $\calM_\Q$ and $\calM_\Q$.

We now endow $\calM_\Q[\calL]$ with the structure of a $(\calL,\perp^Q)$-operated locality algebra. We revisit a map defined  in  \cite {GPZ4}, viewed here as an operating map on $\calM [\calL]$:
\begin {lemma}\label{lem:opL} For any $L\in \calL$, the operator
$$ {\mathcal I}^L(f)(x):=\int_0^\infty  \frac {f(y)y^{-L}}{y+x} dy,$$
defines a linear map from $\calM_\Q[\calL]$ to $\calM_\Q[\calL]$.
With
$$
 {\mathcal I}: \calL\times \calM_\Q[\calL]{\longrightarrow} \calM_\Q[\calL],\quad
(L, f)\longmapsto {\mathcal I}^L(f),
$$
denoting the resulting action of $\calL$ on $\calM_\Q[\calL]$, the triple $(\calM_\Q[\calL], \perp^Q, {\mathcal I})$ is an  $(\calL,\perp^Q)$-operated locality algebra.
\end{lemma}
\begin {proof} We first prove that ${\mathcal I}^L$ defines a map from $\calM_\Q[\calL]$ to $\calM_\Q [\calL]$. By the $\calM_\Q$-linearity, we only need to prove $\calI ^L(x^\ell )\in \calM_\Q [\calL]$.

For fixed $x>0$ and $z\in \C$, the map $y\longmapsto \frac{y^{-z}}{(y+x)^k}$ is locally integrable on $ (0, +\infty)$. Since $\frac{y^{-z}}{(y+x)^k}\underset{y\to 0}{\sim}C\, y^{-z}$ and $\frac{y^{-z}}{(y+x)^k} \underset{y\to \infty}{\sim} y^{-z-k}$, the integral $\int_0^\infty \frac{y^{-z}}{(y+x)^k}\,dy$ converges absolutely on the strip $ \Re (z) \in (1-k,1)$ and  the map $x\longmapsto \int_0^\infty \frac{y^{-z}}{y+x}\, dy$ is smooth with $k$-th derivative $x\longmapsto (-1)^k\,k!\,\int_0^\infty \frac{y^{-z}}{(y+x)^{k+1}} \, dy. $

For $\alpha >0$, we have
$$\lim _{x\to 0+}x^\alpha \,\ln (x)=0 , \ \lim _{x\to \infty }x^{-\alpha} \,\ln (x)=0,
$$ so that for fixed $x>0$,  the map $z\longmapsto \int_0^\infty \frac{y^{-z}}{y+x}\,dy$ is
  holomorphic in $z$ on the  strip $\Re (z)\in (0,1)$ with derivative $z\longmapsto -z\,\int_0^\infty \frac{\ln y\, y^{-z}}{y+x}\,dy$.
For any real number $ a \in (0, 1)$,  an explicit computation~\cite[Lemma 4.5]{GPZ4} gives:
 $$\int_0^\infty \frac{x^{-a}\,d x}{x+1}
  = \frac{\pi}{\sin (\pi a)},$$
 and hence
$$\int_0^\infty \frac{y^{-a}}{y+x}\,dy=
 x^{-a}\, \int_0^\infty \frac{y^{-a}\,d y}{y+1}=\frac{\pi}{\sin(\pi a)} x^{-a}.$$  It follows that for any complex number in the strip $\Re (z)\in (0,1)$, we have
$$\int_0^\infty \frac{y^{-z}}{y+x}\,dy=\frac{\pi}{\sin(\pi z)} x^{-z}.
$$

As a consequence of  this explicit formula, the map $z\longmapsto  \int_0^\infty \frac{y^{-z}}{y+x}\,dy$ extends to a meromorphic family (in the  parameter $z$) of smooth functions in $x$. Thus, when restricted to a neighborhood of $0$ in the parameter space, ${\mathcal I}^L$ can be viewed as a $\calM_\Q$-linear map from $\calM_\Q [\calL] $ to
$\calM_\Q [\calL] $. More precisely,
\begin{equation} \label{eq:formula_IL}
 {\mathcal I}^L(f x^{-\ell})=\left(x\mapsto f(x) \frac{\pi}{\sin(\pi(L+\ell))} x^{-L-\ell}\right), \quad f\in \calM_\Q, L\in \calL
\end{equation}
(on the l.h.s, we write $x^{-\ell}$ for $x\to x^{-\ell}\in\calM_\Q[\calL]$).

The conditions for $(\calM_\Q [\calL],\perp^Q,{\mathcal I} )$ to be an $(\calL,\perp^Q)$-operated locality algebra are then easy to check. For example, if (with a small abuse of notations) $f_1x^{L_1}$ and $f_2x^{L_2}$ are two independent elements of $\calM_\Q[\calL]$ and $L\in\calL$ is such that $L\perp^Q L_i$ then it is clear from the definition $\perp^Q$ that $\calI^L(f_1x^{L_1})\perp^Q f_2x^{L_2}$. We omit to write down the other various properties as is it nothing more than bookkeeping.
\end{proof}
The universal property of the $(\calL,\perp^Q)$-operated locality algebra $(\calF_{\calL}^{\rm prop},\top_{\calF_\calL})$ discussed in Theorem~\ref{thm:univ_prop_trees_loc} yields a unique locality algebra homomorphism
$${\mathcal R} : (\calF_{\calL}^{\rm prop}, \top_{\calF_{\calL}},B_+)\to (\calM_\Q[\calL], \perp^Q,{\mathcal I}),
$$
(with $\mathcal{I}:=\{\mathcal{I}^L,L\in\mathcal{L}\}$)
which is characterized by the following conditions:
\begin{align*}
\calR(\bullet_ L) & = \int_0^\infty \frac{y^{-L}}{y+x}dy = \frac{\pi}{\sin (\pi L)} x^{-L},
% \label{eq:rinit}
\\
\calR((F_1,d_1)(F_2,d_2)) & = \calR(F_1,d_1) \calR(F_1,d_2) \quad \text{for all } (F_1,d_1)(F_2,d_2)\in \calF_{\calL}^{\rm prop},
% \label{eq:rprod}
\\
\calR\left( B_+^L((F, d))\right) & = {\mathcal I}^L\left(\calR((F,d))\right) \quad \text{for all } B_+^L((F, d))\in \calF_{\calL}^{\rm prop}.
% \label{eq:rgraft}
\end{align*}
 We are almost ready to define the multivariate regularised and renormalised maps for Kreimer's toy model. We just need some straightforward yet useful preliminary results which we state without proof.
 \begin{lemma} \label{lem:ev_maps}
  \begin{enumerate}
   \item The evaluation map at $x=1$:
\begin{equation*}
{\rm ev}_1:={\rm ev}_{x=1}: \left(\calM_\Q[\calL ], \perp\right)\longrightarrow \left(\calM_\Q, \perp\right), \quad f x^{L}\mapsto f, \quad f\in \calM_\Q, L\in \calL,
% \label{eq:ev1}
\end{equation*}
is a homomorphism of locality algebras.
\item Recall (Equation \eqref{eq:merodecomp} and Proposition \eqref{pp:merodecomp}) that $\calM_\Q$ admits a splitting $\calM_Q=\calM_+\oplus\calM_-^Q$ into holomorphic and polar germs. The evaluation map at $\vec z=\vec0$ of holomorphic germs 
$${\rm ev}_0:={\rm ev}_{z=0}:\calM_+\to \C$$
is a homomorphism of locality algebras.
  \end{enumerate}
 \end{lemma}
 We have \ty{all the} tools we need to define the regularisation and renormalisation maps of Kreimer's toy model in our multivariate framework and to prove that they are locality algebra morphisms.
 \begin{defitheo} \label{defthm:reg_ren_maps}
  \begin{enumerate}
   \item The {\bf multivariate regularisation map} of Kreimer's toy model
$${ \mathcal  R}_1 := {\rm ev}_1 \circ { \mathcal  R}: \calF _{\calL}^{\rm prop}\to \calM_\Q
$$
is  a homomorphism of locality algebras
\begin{equation*}
\label{eq:calr1}
 \calR_1  : (\calF _{\calL}^{\rm prop}, \perp _{\calF_{\calL}})\to \left(\calM_\Q , \perp^Q \right).
\end{equation*}
% In particular,  for independent properly decorated forests $(F, d)$ and $(F', d')$, we have
% $$ { \mathcal  R}_1 (F, d)\perp^Q { \mathcal  R}_1(F', d') \ \text{ and } \ { \mathcal  R}_1\left( (F, d)\bullet (F', d')\right)= { \mathcal  R}_1 (F, d)\, { \mathcal  R}_1(F', d'). $$
\item The {\bf multivariate renormalisation map} of Kreimer's toy model ${ \mathcal  R}^{\rm ren}$  on $\calF _{\calL}^{\rm prop}$   defined by
\begin{equation*}
% \label{eq:calrren} 
\mathcal  R^{\rm ren}  :={\rm ev}_0\circ\pi_+^Q\circ  { \mathcal  R}_1: \calF _{\calL}^{\rm prop}\to \C
\end{equation*}
   (with $\pi_+^Q:\calM_\Q\to\calM_+$ is the projection to $\calM_+$ parallely to $\calM_-^Q$) is a  locality character on the locality algebra $\calF _{\calL}^{\rm prop}$. It coincides with the positive part of the Birkhoff-Hopf decomposition of $\calR_1$ evaluated at $\vec0$.
  \end{enumerate}
 \end{defitheo}
 \begin{proof}
  \begin{enumerate}
   \item By Theorem \ref{thm:univ_prop_trees_loc} and the first point of Lemma \ref{lem:ev_maps}, $\calR$ and ev$_1$ are two locality algebra morphisms, thus so is $\calR_1$.
   \item By the second point of Lemma \ref{lem:ev_maps} and Proposition \ref{pp:merodecomp}, ev$_0$ and $\pi_+^Q$ are two locality algebra morphisms, and $\calR_1$ is as well according to the first point, so $\calR^{\rm ren}$ is a locality algebra morphism as well.
   
   It coincides with the positive part of the Birkhoff-Hopf decomposition of $\calR_1$ evaluated at $\vec0$ by unicity of the locality Birkhoff-Hopf decomposition (Theorem \ref{thm:abflhopf}), but also as a direct consequence of Corollary \ref{co:abflhopf}.
  \end{enumerate}
 \end{proof}
Notice that, thanks to the relative simplicity of Kreimer's toy model, we can actually derive explicit properties of $\calR_1$ and $\calR^{\rm ren}$. We omit these results here since they are very specific to this model, which we only use as a demonstrating ground for our multivariate renormalisation scheme. We refer the reader to \cite[Proposition 5.5 and Corollary 5.8]{CGPZ2} for these details.

\section{Some open questions regarding locality structures}

% Locality structures have been originally introduced to offer a framework to rigorously build a multivariate renormalisation scheme, a task they completely fullfilled. It has been applied to various objects, most notably in \cite{CGPZ2} to Kreimer's toy model \cite{Kr97}, in \cite{CGPZ3} to a divergent version of Arborified Zeta Values studied in Chapter \ref{chap:MZV} and in \cite{CGPZ2019} to divergent Conical Zeta Values.

I hope to have presented a rather self-containing introduction to the main ideas and results of the theory of locality structures. In particular, I hope that the last Section \ref{sec:Kreimer} shows that they can be successfully applied to some relevant models. In conclusion, let me list some questions that are yet to be answered in the domain of locality structures, should the interest of a potential reader be roused enough that she or he would like to tackle them. 

A first rather obvious point, is that the various application of our multivariate renormalisation scheme (most notably Kreimer's toy model, CZVs and AZVs) are rather firmly stepped in the mathematics domain. Physical applications, e.g. multivariate renormalisation of physical models is the next logical step for our multivariate renormalisation scheme. Notice that some steps have been made into this direction, in particular in\cite{Re19} which unifies locality with causality in AQFT
% We refer the readers to \cite{Re16} for a gentle introduction to AQFT. 
and in \cite{GPZ24}, 
where the dependence on $Q$ of the splitting \eqref{eq:merodecomp} of meromorphic germs was investigated and interpreted as the action of a renormalisation group akin to Cartier's universal Galois group. 

To conclude this plea for applications of multivariate renormalisation to Physics, I would like to emphasize that it fits rather naturally into the program of a TRAP approach to building Feynman rules that was presented at the end of Chapter \ref{chap:PROP}. \ty{It} might be the proof of a considerable hubris, but I do hope that somehow both these approaches will become unified into an accepted, if not standard, approach to perturbative QFT. 

\medskip 

% Locality structures contain more than the multivariate renormalisation scheme and I believe that, as it is often the case in mathematics, they are worth studying for their own sake. Various extensions of the theory presented above were proposed. For example, a locality version of lattice theory was studied in \cite{CGPZ21}. Other locality structures and their relation with various partial structures were investigated in \cite{Zh18}. A locality version of the Poincaré-Birkhoff-Witt and the Milnor-Moore Theorems were proven in \cite{Lo23}, see also \cite{CFLP22}.

The domain has therefore reached a certain maturity. However, there are some gaps that have been proven rather delicate to close. The most important of them, in my opinion, was thoroughly 
investigated in \cite{CFLP22}.  It has to do with the locality structures carried be locality tensor products. It can be stated as a conjecture:
\begin{conj} \label{conj:loc_tensor}
 Let $(E,\top)$ be any locality vector space, $V$ and $W$ any two locality \ty{vector subspaces} of $E$. Then $(V\otimes_\top W,\ttop)$ is a locality vector space.
\end{conj}
Evidences for this conjecture have been exhibited in \cite{CFLP22}. Let me briefly list them here:
\begin{itemize}
 \item This conjecture holds (see \cite{CFLP22}[Theorem 6.4]) for locality vector spaces with {\bf locality basis}, that is roughly to say for locality vector spaces freely generated by a locality set, such that the locality on the vector space is induced by the one on the basis. See \cite{CFLP22}[Definition 6.1] for a precise definition.
 \item It holds for Hilbert spaces where the locality is given by being orthogonal w.r.t. the inner product\footnote{However in this case, the quotient locality relation becomes trivial (see \cite{CFLP22}[Example 2.7]) so this argument is less strong than the first one.}.
 \item We could not find counterexamples, although we did spend time and effort to build them.
\end{itemize}
Looking for counterexamples, we actually came to stronger conjectures, in particular that $I_{\rm bil}^{\top_\times}$ from Definition \ref{defn:loctensprod1} always has a property insuring that the quotient locality relation turns $V\otimes\top W$ into a locality vector space. This conjecture and others, in particular regarding locality envelopping algebras, have not been published as of yet.

% Other related conjectures have been made in the same paper.
I would nonetheless point out that these studies suggest a strategy to tackle Conjecture \ref{conj:loc_tensor}: to build a (locality) lattice from $\top$, $V$ and $W$ and to use tools of (locality) lattices theory \ty{developed} in \cite{CGPZ21}. At the moment, this is the only possible approach to this conjecture that I see as being reasonable

Other open problems of locality theory are still open, some technical and some more conceptual. Let me name a few here in no particular order.
\begin{itemize}
 \item As build, the category $\LVect$ is not monoidal. Defining a ``natural'' monoidal structure on this category could possibly solve Conjecture \ref{conj:loc_tensor}, but this is --in my opinion -- a rather unreasonable (in the sense of very interesting but difficult) approach to the conjecture.
 \item A classification of locality Lie groups is still missing. It was shown in \cite{CFLP22} that such a classification would necessarily be richer than the one of usual (non-locality) Lie groups.
 \item Given a locality vector space $(V,\top)$, there are not yet natural notions of a dual locality vector space $(V^*,\top^*)$. This question might be anecdotic or lead to a deeper understanding of what locality categories should be.
 \item \cy{It was pointed out by one of the referee that, as presented here, the notion of locality might be be too large. With this point of view, one could require some properties for the relation of Definition \ref{defn:independence} to be called an independence relation. It is possible that the right axioms force these restricted independence relations to automatically solve conjecture \ref{conj:loc_tensor}.}
\end{itemize}
% Does someone read the source .tex? I would be most surprised, so if so, let me know!
With this open list, I hope to convince the motivated reader wishing to attack open problems in the theory of locality structures that they are plethora. There is more than enough room for all interested researchers.
% and I hope to meet you there some time.

% % % % %  COMMENTER CE QUI EST CI-DESSOUS.
% 
% \bibliographystyle{unsrt}
%  \addcontentsline{toc}{section}{Bibliography}
% \bibliography{HDR_biblio}
%  
% \end{document}

%% file: RES.tex
% % 
% \documentclass[11pt,twoside,a4paper]{book}
% 
% \usepackage{HDR}
% 
% %% SI ON LAISSE LES DEUX LIGNES SUIVANTES DANS LE STY, ELLES NE MARCHENT PAS... (probablement à cause de \input).
% \input{xy}
% \xyoption{all}
% 
% 
% \begin{document}
% % 
% % % % 
% % % % % % COMMENTER CE QUI EST AVANT CA ET ENDDOCUMENT POUR COMPILER HDR.tex.

\chapter{Resurgence in Quantum Field Theory}
\label{chap:res}

This chapter is based mostly on \cite{Cl21} and \cite{BC19}. It also uses results from previous papers \cite{BS13,Be10,BC14,BC15,BC18}.

\section*{Introduction}

\addcontentsline{toc}{section}{Introduction}

\subsection*{Resurgence in Physics}

\addcontentsline{toc}{subsection}{Resurgence in Physics}

As we have already pointed out in the first and third chapter of this thesis, the classical approach to QFT is a \ty{perturbative one}, where the perturbative terms are typically computed with the help of Feynman graphs. These are typically divergent, and thus need to be regularised and renormalised. Once this tremendous task is achieved, one ends up with a formal series in the \ty{perturbative} parameter, typically the fine structure constant of the theory. A natural question is then ``is this \ty{perturbative} (formal) series convergent?''

The answer to this good question is a resounding ``no''. This is due to the famous Dyson's argument\footnote{which is often quoted but seems to never have been written down by Dyson, so that it is transmitted in an ``oral'' way through QFT practitionners. If someone knows of a reference by Dyson, I would be very thankful for them to let me know of it.}. The argument is often very well explained in the \ty{literature}, see for example \cite{schweber2020} where the author puts it in historical perspective but for completeness I would like to explain it as well.

Assume the perturbative series converges for a value of the fine structure constant $a>0$. Then this series would be convergent in a little disk around the origin. Such a disk contains negative values of $a$, so the theory would also be defined for such values. But for these values, pairs of particles are repelled. So pairs of particles created from quantum fluctuations of the void would become real. The vacuum would continuously radiates particles i.e. be unstable.

This is from 1951, and probably very soon after (e.g. it seems to be already discussed as a well-known fact in \cite{Ho79}) it was realised that Borel-Laplace resummation allows to escape this issue. Then we can replace the original question by the more resonable one ``is this \ty{perturbative} (formal) series summable with the Borel-Laplace method?'' I believe that the answer to that question is still ``no'', at least in many physically relevant cases. We will see later that this is because often the Borel transform has singularities in the direction one wishes to perform the Laplace transform. In this chapter, I argue that \'Ecalle's resurgence theory \cite{Ecalle81,Ecalle81b,Ecalle81c} is the right framework to solve this issue, and that \ty{perturbation} theories for QFT should be Borel-\'Ecalle summable, or accelero-summable in the case of asymptotically free QFTs. 

For this program, some Physics jargon will be used. We will not need to precisely introduce all the relevant physical notions, but in order to increase readibility for an eventual reader that would not be accustomed I try to explain the main ideas in Subsection \ref{subsec:physics_concepts}. The hope is that the rest of the text will be made more intelligible if the main equations (the renormalisation group equation and the Schwinger-Dyson equation) carry a meaning for the reader and are not seen as just abstract nonsense.\\

For completeness, I also need to point out that the use of resurgence theory I will present here is not the most classical in physics \ty{literature}. For physicists, as far as I understand, the main interest of resurgence is that it offers tools to compute non perturbative contributions to the theory from the \ty{perturbative} series. Before I am more precise, I would like to specify that by ``non perturbative'' I mean objects or effect that cannot be seen at any order in the \ty{perturbation} series. For example, the funtion $f:\R\longrightarrow\R$ defined by $f(0)=0$ and $f(x)=\exp(-1/x^2)$ if $x\neq0$. This function is smooth and we have $f^{(n)}(0)=0$ for all $n\in\N$. Therefore, its Taylor series is the null series: this function would not be detected \ty{perturbatively}. In this sense, physicists want to compute such contributions in the form of \emph{transseries} and manage to do so by exploiting tools of resurgence theory. We refer the reader to \cite{ABS19} for a review of this topics (although recent articles in the same spirit such as \cite{BoDu20} are of course missing).

Before broadly presenting this approach to QFT, I need to specify a \ty{vocabulary} term: \emph{sector}. It is related to the usual notion of regimes, for example of a solution of an ordinary differential equation. The solution $f$ can then typically typically be approximated as a sum of two terms: $f(x)=f_1(x)+f_2(x)$, where $f_1$ is dominant for some values of $x$, and $f_2$ for some other. In physics, $f_1$ and $f_2$ can typically be seen as comming from two different aspects of the theory. So \emph{sector} is often used to describe one part of the theory, e.g. one term of the lagrangian. It can also describe one specific part of a solution of the theory, or at least of an approximation of a solution. Later, we will essentially use this second meaning.

The crucial point of the physicists's approach can be summed up in three steps:
\begin{itemize}
 \item Asymptotics of the perturbative series leads to singularities of its Borel transform;
 \item Singularities of the Borel transform encode non \ty{perturbative} sectors of the theory;
 \item Corrections to the asymptotics correspond to expansion in the non perturbative sectors.
\end{itemize}
This lead to replace the \ty{perturbative} series $\phi(a)\in\C[[a]]$ by a transseries. The simplest of these is of the form
\begin{equation*}
 \Phi(a)=\sum_{n\in\Z}\phi_n(a)e^{-nS/a}
\end{equation*}
with $S\in\R^*$ is the transseries action, and $\phi_0(a)=\phi(a)$, the initial \ty{perturbative} series\footnote{which is usually thought of as describing the $0$-instanton sector of the theory; while $\phi_n$ is the perturbative series in the $n$-th instantons sectors if $n\geq1$ or in the $|n|$-th anti-instantons sector if $n\leq-1$.}. Then $\phi_1$ is encoded into the asymptotics of $\phi_0=\phi$, $\phi_2$ in the asymptotics of $\phi_1$ and in the asymptotics of the asymptotics of $\phi_0$, and so one. Writing $\phi_n(a)=\sum_{k\geq0}\phi_k^{(n)} a^k$ we have at the first few orders:
\begin{align*}
 \phi_k^{(0)}\mathop{\sim}_{k\to\infty}k!\psi_0(1/k) & \rightsquigarrow \phi_k^{(1)}(a)=A_1\psi_0(a) \\
 \phi_k^{(1)}\mathop{\sim}_{k\to\infty}k!\psi_1(1/k)~\wedge~\psi_0{(k)}\mathop{\sim}_{n\to\infty}k!\psi_1^1(1/k) & \rightsquigarrow \phi_k^{(2)}(a)=A_2\psi_1(a)+A_3\psi_1^1(a)
\end{align*}
and so on... Here the $\psi$s are themselves formal series and the $A$s are invariant that have to be determined, typically numerically. The fact that the same numbers come back from the asymptotics to the non-\ty{perturbative} sectors is the origin to the name ``resurgence''. I cannot emphasize enough how remarkable this feature is! The perturbative series ``knows'' about non \ty{perturbative} aspects of the theory: it is rather unintuitive for the physicist that I used to be.

What is maybe even more remarkable: this program can actually be carried out, and has been for many models. This can be done thanks to another aspect of resurgence theory, namely \'Ecalle's \emph{alien calculus}. %I will not go into further details, but simply point 

\subsection*{Content and main results}

\addcontentsline{toc}{subsection}{Content and main results}

This chapter starts with a presentation of aspects of resurgence theory relevant for the task at hand. First (Subsection \ref{subsect:Bl_res}) of the Borel-Laplace resummation method. In Subsection \ref{subsec:res_fct} I present some aspects of resurgent functions and resurgent analysis, in particular the useful Theorem 
\ref{thm:bound_resu_Sauzin} due to Sauzin. The last Subsection \ref{subsec:BE_res} of this first section introduces the notions of average and well-behaved average. The main features of the Borel-\'Ecalle resummation method are then summed up in Theorem \ref{thm:Borel_Ecalle_resummation}.

Section \ref{sec:WZ} starts by introducing some physical concepts that are discussed in the rest of the text (Subsection \ref{subsec:physics_concepts}). Then we \ty{introduce} the model we are studying, namely the Wess-Zumino model. We are looking for a solution to the system of integro-differential equations consisting of the (truncated) Schwinger-Dyson equation (SDE) \eqref{eq:SDnlin} and the Renormalisation Group Equation (RGE) \eqref{eq:RGE}. The solution is the two-point function of the Wess-Zumino model. In the same section we present the results we will use (in particular Claim \ref{thm:resurgence_gamma}) and formulate the SDE and the RGE in the Borel plane.

The third section is devoted to the proof of Theorem \ref{thm:resurgence_two_points_function}, namely that the two-point function is resurgent. This result is obtained by reorganising first the two-point function (Proposition \ref{prop:solution_RGE}) which allows to show that it is 1-Gevrey (Proposition \ref{prop:G_one_Gevrey}). Then we apply the \ty{aforementioned} Theorem \ref{thm:bound_resu_Sauzin}) in what is the heart of our resurgent analysis.

The fourth section is the most technical of this chapter. Using  complex analysis we are able to show a bound on the asymptotics of the two-point function (Theorem \ref{thm:bound_two_point_infinity}) which, together with Theorem \ref{thm:resurgence_two_points_function}, implies our main result: Corollary \ref{cor:main_result}. This result is obtained through a number of technical lemmas I will not try to summarise here. We also discuss at the very end of Subsection \ref{subsec:G_res} how this procedure could give rise to a non-\ty{perturbative} mass generation mechanism.

The final Section \ref{sec:res_final} of this chapter serves as a justification that Borel-\'Ecalle resummation is probably not enough for the most interesting theories, namely asymptotically free ones. We start by recalling that for such theories, 't Hooft argued in \cite{Ho79} that the two-point function can be analytic only in a ``horn-shaped domain'' that we illustrate in Figure \ref{fig:an_dom_tHooft}. We then briefly introduce in Subsection \ref{subsec:accelero} \'Ecalle's (strong) accelero-summation procedure and show that it gives an analyticity domain (illustrated in Figure \ref{fig:an_dom_acc}) for the two-point function that agrees with 't Hooft's horned-shaped domain. The chapter ends in Subsection \ref{subsec:res_questions} with a list of open questions concerning the applications of resurgence theory to QFT.

\section{Elements of resurgence theory}

In this section we introduce the classical Borel-Laplace resurgence method and some notions of resurgence theory that will be important for our purposes. For a more complete \ty{literature} on resurgence theory we refer the readers to the original texts \cite{Ecalle81,Ecalle81b,Ecalle81c} or \cite{Ec93}. Other classical introductions are \cite{Sa14}, \cite{MS16} (the \ty{latter} being more recent) or \cite{Do19} for a physicist's point of view.

\subsection{The Borel-Laplace resummation method} \label{subsect:Bl_res}

Many excellent introductions of the classical theory of Borel-Laplace resummation (not necessarily including resurgent aspects) can be found in the literature. In particular, the PhD thesis 
\cite{Bo11} offers a well-written and short presentation of this topic (in French), while the article \cite{Sa14} contains a very thorough 
introduction. Consider also the recent work \cite{MS16}. Nonetheless, I shortly present the Borel-Laplace resummation method for the sake of self-containedness.

First recall that the space $\C[[X]]$ of formal series in the $X$ variable can be identified with the space of sequences with values in $\C$. It can also be built as a completion of the space of \ty{polynomials} $\C[X]$. $\C[[X]]$ has a ring structure given by the {\bf Cauchy product}, or {\bf convolution product}.
\begin{equation*}
 \left(\sum_{n=0}^{+\infty}a_nX^n\right)\star \left(\sum_{n=0}^{+\infty}a_nX^n\right):=\sum_{N=0}^{+\infty}\left(\sum_{k=0}^Na_{N-k}b_k\right)X^N.
\end{equation*}
Then $X\mathbb{C}[[X]]$ is a sub-ring of $\mathbb{C}[[X]]$.
\begin{defn} \label{def:formal_Borel_tsfm}
The {\bf formal Borel transform} is defined on formal series as
 \begin{eqnarray*}
             \mathcal{B}: (z^{-1}~\mathbb{C}[[z^{-1}]],.) & \longrightarrow & (\mathbb{C}[[\zeta]],\star) \\
  \tilde{f}(z) = \frac{1}{z}\sum_{n=0}^{+\infty}\frac{c_n}{z^n} & \longrightarrow & \hat{f}(\zeta) = \sum_{n=0}^{+\infty}\frac{c_n}{n!}\zeta^n
 \end{eqnarray*}
 extended by linearity.
\end{defn}
\begin{rk}
 The Borel transforms $\hat f$ are typically written as series of a formal variable $\zeta$ (or other letters of the greek alphabet). This is to clarify whether we are working in the usual or direct space (latin letters) or in the Borel plane (i.e. after a Borel transform). Also note that the Cauchy product is written a dot . in the direct space, and as $\star$ in the Borel plane, for the same reason.
\end{rk}
The formal Borel transform enjoys many useful properties, easy to prove by manipulation of formal series (see for example \cite{Sa14}, 
\S 4.3 and 5.1). 
\begin{prop}\label{prop:properties_formal_Borel_tsfm_conv}
 The Borel transform is morphism of rings, and in particular a morphism of algebras for the Cauchy products. Furthermore, let $\tilde{f}(z),\tilde{g}(z)\in z^{-1}~\mathbb{C}[[z^{-1}]]$ be two formal series and $\hat{f},\hat{g}\in\mathbb{C}[[\zeta]]$ be their Borel transforms. If $\hat{f}$ and $\hat{g}$ are convergent then $\mathcal{B}(\tilde f.\tilde g)$ is also convergent and 
 \begin{equation*}
  \mathcal{B}(\tilde f.\tilde g)(\zeta) = (\hat f\star \hat g)(\zeta) = \int_0^\zeta\hat f(\eta) \hat g(\zeta-\eta)d\eta
 \end{equation*}
 for $\zeta$ in the intersection of the convergence domains of $\hat f$ and $\hat g$.
\end{prop}
The integral representation above of $(\hat f\star \hat g)(\zeta)$ justifies that we call $\star$ the ``convolution product'' in the Borel plane.

The Borel transform also preserves the differential structure of the ring of formal series:
\begin{prop} \label{prop:properties_formal_Borel_tsfm}
 Let $\tilde{f}(z),\tilde{g}(z)\in z^{-1}~\mathbb{C}[[z^{-1}]]$ be two formal series and $\hat{f},\hat{g}\in\mathbb{C}[[\zeta]]$ be their Borel 
 transforms. Then the following hold
 \begin{itemize}
%   \item $\mathcal{B}(\tilde f.\tilde g) = \hat f\star \hat g$;
  \item $\mathcal{B}(\partial\tilde f) = -\zeta\hat f$;
  \item $\mathcal{B}(z^{-1}\tilde f) = \int\hat f$;
  \item if $\tilde{f}(z)\in z^{-2}~\mathbb{C}[[z^{-1}]]$, then $\mathcal{B}(z\tilde f)=\frac{d\hat f}{d\zeta}$;
 \end{itemize}
 where the derivatives and the integral are formal (i.e. defined term by terms).
 These properties stay true in the case where $\hat f$ and  $\hat g$ are convergent with the usual derivative and integration.
%  In this case, the first property becomes
%  \begin{equation*}
%   \mathcal{B}(\tilde f.\tilde g)(\zeta) = (\hat f\star \hat g)(\zeta) = \int_0^\zeta\hat f(\eta) \hat g(\zeta-\eta)d\eta
%  \end{equation*}
%  for $\zeta$ in the intersection of the convergence domains of $\hat f$ and $\hat g$.
\end{prop}
We will in fact study the case where the Borel transform is convergent. There exists a simple necessary and sufficient 
condition of the convergence of the Borel transform, but we need one more definition.
\begin{defn} \label{def:Gevrey}
 A formal series $\tilde f(z) = \frac{1}{z}\sum_{n=0}^{+\infty}\frac{a_n}{z^n}$ is {\bf 1-Gevrey} if
  \begin{equation*}
   \exists A,B>0:~|a_n|\leq AB^n n!~~\forall n\in\N.
  \end{equation*}
  In this case, we write $\tilde f(z)\in z^{-1}\C[[z^{-1}]]_1$.
\end{defn}
The following theorem is a standard exercise to prove (see for example \cite{Sa14}, \S 4.2) but turns out to be rather important for the applications we have in mind.
\begin{theo} \label{thm:Gevrey_give_conv_Borel}
 Let $\tilde{f}(z)\in z^{-1}~\mathbb{C}[[z^{-1}]]$ be a formal series. Its Borel transform has a {strictly positive, possibly infinite} radius of convergence (in this case 
 we write $\hat f\in\C\{\zeta\}$) if and only if $\tilde f$ is 1-Gevrey.
\end{theo}
The Borel transform provides a tool to extend the usual notion of convergent series thanks to the existence of its inverse: the Laplace transform.
\begin{defn} \label{def:Laplace_tsfm}
 Let $\theta\in[0,2\pi[$ and set $\Gamma_\theta:=\{Re^{i\theta},R\in[0,+\infty[\}$. Let $\hat f\in\C\{\zeta\}$ be a germ admitting an 
 analytic continuation in an open subset of $\C$ containing $\Gamma_\theta$ and such that 
 \begin{equation} \label{eq:exp_bound}
  \exists c\in\R,~K>0:|\hat f(\zeta)|\leq K e^{c|\zeta|}
 \end{equation}
 for any $\zeta$ in $\Gamma_\theta$. Then the {\bf Laplace transform } of $\hat f$ in the direction $\theta$ is defined as
 \begin{equation*}
  \mathcal{L}^{\theta}[\hat f](z) = \int_0^{e^{i\theta}\infty}\hat f(\zeta)e^{-\zeta z}d\zeta.
 \end{equation*}
\end{defn}
The Laplace integral of this definition is finite for $z$ in an open subset of $\C$ to be specified later on. For the time being, let us say that
the composition of the Laplace and the Borel transforms is the so-called Borel-Laplace resummation method.
\begin{defn} \label{defn:Borel_Laplace}
 Let $\theta\in[0,2\pi[$, and $\tilde f(z)\in z^{-1}\C[[z^{-1}]]_1$ such that the Laplace transform of its Borel transform exists in the 
 direction $\theta$. Then $\tilde f(z)$ is said to be {\bf Borel summable} in the direction $\theta$. 
 
 The {\bf Borel-Laplace resummation operator} in the direction $\theta$ is defined on the functions Borel summable in the direction $\theta$ 
 as 
 \begin{equation*}
  S_{\theta} = \mathcal{L}^{\theta}\circ\mathcal{B}.
 \end{equation*}
 For a Borel summable formal series $\tilde f$, the function $z\mapsto S_\theta[f](z)$ is called the {\bf Borel sum} of $\tilde f$.
\end{defn}
Varying the direction $\theta$ of the resummation leads to interesting concepts and phenomena such as sectorial resummation and the 
Stokes phenomenon, however we will not {deal with} them here. However, I would like to point out that understanding and computing (in some sense) the Stokes phenomenon is one of the original motivations for resurgence theory. In particular, the now relatively famous alien calculus has been devised for this very purpose in \cite{Ecalle81,Ecalle81b,Ecalle81c}.
\begin{rk}
 It is easy to see that a formal series $\tilde f(z)\in z^{-1}\C[[z^{-1}]]$ with a finite non-zero radius of convergence has a Borel transform admitting
 an exponential 
 bound \eqref{eq:exp_bound} for all $\theta\in[0,2\pi[$ and that its the Borel sum in any direction coincide with the usual summation of 
 series. Thus the Borel-Laplace resummation method is an extension of the usual summation of series.
\end{rk}
We claimed in Definition \ref{defn:Borel_Laplace} that the Borel sum of a Borel summable function is a function. This is a consequence of 
the following theorem, which is itself a consequence of classical results of the theory of the Laplace transformation.
\begin{theo} \label{thm:analyticity_domain_Borel_resum}
 Let $\tilde f$ be a formal series, Borel summable in the direction $\theta$ with the exponential bound \eqref{eq:exp_bound}:
 \begin{equation*}
  \exists c\in\R,~{K>0}:|\hat f(\zeta)|\leq K e^{c|\zeta|}.
 \end{equation*}
 Then its Borel sum is analytic as a function of $z$ in the half-plane $\Re(ze^{i\theta})>c$.
\end{theo}
We have seen that one can perform the Borel-Laplace resummation method in non-singular directions of the Borel transform only. However, in 
many problem of interest (in particular, of interest to physicists), the Borel transform will have singularities in the direction where we wish to perform 
the resummation. \'Ecalle's solution to this important problem is to extend further the Borel-Laplace resummation method to the so-called resurgent functions, which we now introduce.

\subsection{Resurgent functions} \label{subsec:res_fct}

Let $\Omega$ be a non-empty, discrete and closed subset of $\C^*$. Recall that a function (or a germ) $f$ is {\bf analytically extensible} along a path $\gamma$ in $\C$ starting within the disc of convergence of the function if there is a finite 
family $(D_i)_{i\in \{1,\cdots,n\}}$ of convex open subset of $\C$ covering $\gamma$ such that $f$ is analytically extensible to $D\cup D_1\cup\cdots\cup D_n$. 

Notice that being analytically extensible along a path is much less strict than being analytically extensible. In particular, a function analytically extensible along a family of paths can be seen as a function over an open subset of a cover of $\C$ rather than a function of $\C$. 
\begin{defn} \label{defn:resum_function}
 A germ $\phi\in\C\{\zeta\}$ is said to be an {\bf $\Omega$-resurgent function} if it is analytically extensible along 
  any rectifiable (i.e. of finite length) path in $\C\setminus\Omega$. We set
  \begin{equation*}
   \widehat{\mathcal{R}}_\Omega := \{\text{all $\Omega$-analytically extensible}\} \subset\C\{\zeta\}.
  \end{equation*}
\end{defn}
A few important remarks are in order regarding this definition.
\begin{rk}
\begin{itemize}
 \item Technically speaking, this is the definition of \emph{minors} of resurgent functions. A resurgent function, defined in its full analytical glory, is a class of functions analytically extensible along any rectifiable paths. Working with this more general notions of resurgent functions has some advantages. In particular, it allows us to deal with the case where $0\in\Omega$. For a \ty{detailed} presentation of this subtle subjet we refer the readers to \cite{Sa14} and \cite{MS16}. The text \cite{Fa20} also contains a concise but enlightning presentation of this topic.
 
 \item It will be clear in a moment that the above definition is actually stronger than what we need. For our purpose, it would be \ty{enough to} have germs analytically extensible along paths that are rectifiable to paths not circling back around a element of $\Omega$. These objects, which could be called \emph{forward resurgent} have not, to the best of my knowledge, been studied. I do not use it and only mention them in passing since I will use important analytical results about resurgent functions that do not necessarily exist for these forward resurgent functions.
\end{itemize} 
\end{rk}
Now, the convolution product of two $\Omega$-resurgent function is well-defined inside the intersection of their convergence discs. A difficult question 
is whether or not this convolution product defines an $\Omega$-resurgent function. 
The following theorem is a cornerstone of resurgence theory, as it states when this is indeed the case 
and thus {when} resurgent functions are stable under an extension of the convolution 
product {and are therefore suited} to the study of non-linear differential equations.
\begin{theo}({Ecalle}, Sauzin \cite{Sa14}[Theorem 21.1]) \label{thm:stability_resu_fct}\\
 Let $\Omega\subset\C$ be non-empty, discrete and closed. Then $\widehat{\mathcal{R}}_\Omega$ is stable under the convolution product if, and only if, 
 $\Omega$ is closed under addition.
\end{theo}
The {classical} example below is already enough to show that $\Omega$ being closed under addition is a necessary condition. The hard part of the Theorem is thus to 
show that it is sufficient.
\begin{example} \label{ex:comvo_resu}
Take $\omega_1,\omega_2\in\Omega$ {non zero and not proportional} and two meromorphic (and therefore resurgent) functions defined by
 $\hat{f}_1(\zeta) = \frac{1}{\zeta-\omega_1}$, $\hat{f}_2(\zeta) = \frac{1}{\zeta-\omega_2}$.
 Then a direct computation gives
 \begin{align*}
  (\hat{f}_1\star\hat{f}_2)(\zeta):= & \int_0^\zeta \hat{f}_1(\eta)\hat{f}_2(\zeta-\eta)d\eta \\
                                  = & \frac{1}{\zeta-\omega_1-\omega_2}\left[\int_0^\zeta\frac{d\eta}{\eta-\omega_1} + \int_0^\zeta\frac{d\eta}{\eta-\omega_2}\right] \\
 \end{align*}
 {Now, take $R$ the Riemann surface obtained from removing the lines $\omega_i[1,+\infty[$ from $\C$.}
 One can {then} check that the R.H.S. has indeed a pole in $\omega_1+\omega_2$ {on some sheets of the Riemann surface $R$}. Therefore if $\omega_1+\omega_2$ is not an element of $\Omega$, $\hat{f}_1\star\hat{f}_2$
 is not $\Omega$-resurgent. {For more details, see for example \cite[Section 20]{Sa14}.}
\end{example}
%  For $\Omega\subset\C$ non-empty, discrete and closed we set $\rho(\Omega):=\min\{|\omega|:\omega\in\Omega^*\}$, with $\Omega^*=\Omega$ if 
%  $0\notin\Omega$ (this will be the case we will work with) and $\Omega^*=\Omega\setminus\{0\}$ otherwise.
We will use bounds on convolution products of resurgent functions. First, recall that for any open set $U\subset\C$ containing the origin and 
star-shaped with respect to the origin, the following bound holds by direct computation:
\begin{equation} \label{eq:bound_star_shaped}
 |(\hat\phi_1\star\cdots\star\hat\phi_n)(\zeta)|\leq \frac{|\zeta|^{n-1}}{(n-1)!}\max_{[0,\zeta]}|\hat\phi_1|\cdots\max_{[0,\zeta]}|\hat\phi_n|
\end{equation}
for any $\hat\phi_1,\cdots,\hat\phi_n$ holomorphic on $U$ and $\zeta$ in $U$. We used $[0,\zeta]$ to denote the straight line between $0$ and $\zeta$.

This bound will be useful to show that the two-point function has the right type of bound at infinity on {$\mathcal{U}_\Omega$, the connected star-shaped domain in $\C$ obtained from placing a radial cut starting from the first singularities of $\Omega$,} and converges near the 
origin. However, it will not allow {us} to prove that it is resurgent. For this we will need to prove the normal convergence of a series of analytic 
continuations of functions along any paths avoiding $\Omega$. It will require the refined results of \cite{Sa12}, specific to resurgence 
theory. In order to state {the main} result, we need to introduce some notations, the same as in \cite{Sa12}.

First, for $\Omega\subseteq\C^*$ non-empty, discrete and closed as before, let us set
$\rho(\Omega):=\min\{|\omega|:\omega\in\Omega\}$.
Second, let $\mathscr{S}_\Omega$ be the set of homotopy classes with fixed endpoints of path $\gamma:[0,l]\longrightarrow\C\setminus\Omega^*$ 
such that $\gamma(0)=0$. Then, for $\delta,L\geq0$ we set 
\begin{equation*}
 \calK_{\delta,L}(\Omega) := \{\zeta\in\scrS_\Omega|\exists\gamma\in\scrS_\Omega:\gamma(l)=\zeta,~\gamma\text{ of length }\leq L, \text{ dist}(\gamma(t),\Omega^*)\geq\delta~\forall t\in[0,l]\}.
\end{equation*}
It was shown in \cite{KaSa16} that $\mathscr{S}_\Omega$ has the structure of a Riemann surface, which is a cover of $\C\setminus\Omega$. 
Then $\calK_{\delta,L}(\Omega)$ can be described as the set of point of this cover which can be reached by paths of length less than 
$L$\footnote{in order to avoid confusion between the kinematic parameter of the two-point function and the length of the path we will 
denote the former by the letter $\Lambda$.} and staying at a distance at least $\delta$ of $\Omega^*$. One 
can in particular see the set of $\Omega$-resurgent functions as the set of locally integrable maps $f:\mathscr{S}_\Omega\longrightarrow\C$. This 
observation will become important to define the notion of average. But what we will mostly use is the following analytical bound for the convolution of resurgent functions.
\begin{theo} \label{thm:bound_resu_Sauzin} \cite[Theorem 1]{Sa12}\\
 Let $\Omega\subset\C$ be discrete, closed and stable 
 under addition. Let $\delta,L>0$ with 
 $\delta<\rho(\Omega)$. 
 Set
 \begin{equation*}
  C:=\rho(\Omega)\exp\left(3+\frac{6L)}{\delta}\right), \qquad \delta':=\frac{1}{2}\rho(\Omega)\exp\left(-2-\frac{4L}{\delta}\right), \qquad L':=L+\frac{\delta}{2};
 \end{equation*}
  Then, for any any $n\geq1$ and $\hat\phi_1,\cdots,\hat\phi_n\in\widehat{\mathcal{R}}_\Omega$
 \begin{equation} \label{eq:bound_conv_resu}
  \max_{\calK_{\delta,L}(\Omega)}|\hat\phi_1\star\cdots\star\hat\phi_n| \leq \frac{2}{\delta}\frac{C^n}{n!} \max_{\calK_{\delta',L'}(\Omega)}|\hat\phi_1|\cdots\max_{\calK_{\delta',L'}(\Omega)}|\hat\phi_n|.
 \end{equation}
\end{theo}
In subsequent work \cite{KaSa16}, Sauzin and Kamimoto have generalised this result to the cases where $\Omega$ is not stable under addition. One could 
 in principle use the result of \cite{KaSa16} to prove resurgence of the two-point functions on $\Z^*/3$ rather than $\N^*/3$, 
 However this is not needed for the Ecalle-Borel resummation procedure along the positive real axis, and we will satisfy ourselves with {exploiting} the above 
 bound, which is simpler to use.
 
 \subsection{Borel-\'Ecalle resummation method} \label{subsec:BE_res}
 
 We have already pointed out that our goal is to apply an extention of the Borel-Laplace resummation method to a two-point function. We now present this method, which was described in particular in \cite{EcMe95} and \cite{Menous}.
 
 As already pointed out, we don't need in practice to consider path going backward to perform a Borel-\'Ecalle resummation. To simplify the statements we take from now on $\Omega$ to be 
a subset of $\R_+^*$.
\begin{defn} \label{def:ramified_plane}
 Let $\C\dsetminus\Omega$ be the {\bf$\Omega$-ramified plane}, namely the space of homotopy classes $[\gamma]$ of rectifiable paths 
 $\gamma:[0,1]\mapsto\C\setminus\Omega$ such that $\forall t,t'\in[0,1],~t<t'\Rightarrow \Re(\gamma(t))< \Re(\gamma(t'))$.
\end{defn}
One can show that $\C\dsetminus\Omega$ has the structure of a Riemann surface, see \cite{KaSa16}. $\C\dsetminus\Omega$ is a cover of 
$\C\setminus\Omega$. Let $\pi:\C\dsetminus\Omega\longrightarrow\C\setminus\Omega$ be the 
canonical local biholomorphism associated to this Riemann surface. The precise definition of this geometrical object can be found in \cite[Section 3]{KaSa16}. I omit these definitions as they will play only a minor role in the present work.

Let $\zeta\in\C\setminus\Omega$ and $\underline\zeta\in\C\dsetminus\Omega$ such that $\pi(\underline{\zeta})=\zeta$. 
If $\Omega=\{\omega_1,\omega_2,\cdots\}\subset\R_+^*$ with $\omega_0:=0<\omega_1<\omega_2<\cdots$, we write $\zeta^{\epsilon_1,\cdots,\epsilon_n}$ 
instead of $\underline\zeta$, with 
$\epsilon_i\in\{+,-\}$, $(\epsilon_1,\cdots,\epsilon_n)$ the signature of the branch of $\C\dsetminus\Omega$ on which 
$\zeta^{\epsilon_1,\cdots,\epsilon_n}$ stands and 
$|\pi(\zeta^{\epsilon_1,\cdots,\epsilon_n})|\in]\omega_n,\omega_{n+1}[$. If $\epsilon_i=+$ (resp $-$), we avoid $\omega_i$ by going to the right (resp. to the left).

We further write $\mathcal{U}_\Omega$ the connected star-shaped domain in $\C$ obtained from placing a radial cut starting from $\omega_1$, the first singularity of $\Omega$. We identify $\mathcal{U}_\Omega$ with a subset of $\C\dsetminus\Omega$: for $\zeta\in\mathcal{U}_\Omega$, if $\Re(\zeta)\leq\omega_1$, we identify it with the homotopy class of the straight path from the origin to $\zeta$. Otherwise we identify it with $\zeta^{+,\cdots,+}$ if $\Im(\zeta)>0$ and with $\zeta\leftrightarrow\zeta^{-,\cdots,}$ if $\Im(\zeta)<0$. Notice that in the more common case where $\Omega\subset\R^*$, $\mathcal{U}_\Omega$ is obtained from placing two radial cuts at the first singularities of $\Omega$.

From now one we will make the simplifying assumption that $\Omega=\omega\N^*$ 
for some $\omega\in\R_+^*$. 

While performing a Borel-\'Ecalle resummation, it will be useful to see $\Omega$-resurgent functions as locally integrable functions from 
$\C\dsetminus\Omega$ to $\C$.
We also denotes by 
$\widehat{\mathcal{U}}_\Omega$ the set of uniform functions on $\C\dsetminus\Omega$; i.e. the set of functions $\hat\phi$ whose value at $\zeta$ do not depend on 
the branch of $\C\dsetminus\Omega$ $\zeta$ sits on:
\begin{equation*}
 \forall (\zeta,\zeta')\in(\C\dsetminus\Omega)^2,~\pi(\zeta)=\pi(\zeta')\Longrightarrow\hat\phi(\zeta)=\hat\phi(\zeta').
\end{equation*}
The crucial object that will allow us to perform an Borel-Laplace-like resummation in the direction where the Borel transform has (isolated) singularities is the average.
\begin{defn} \label{def:averages}
 An {\bf average} is a linear map ${\bf m}:\resOm\longrightarrow\uniOm$ defined by its weights $\{{\bf m}^{\varepsilon_1,\cdots,\varepsilon_n}\}$ and its 
 action on resurgent functions: for any such $\phi$ and any $\zeta\in\C\setminus\Omega$ with $|\zeta|\in[n\omega,(n+1)\omega[$
  \begin{equation*}
   ({\bf m}\phi)(\zeta) = \sum_{\varepsilon_1,\cdots,\varepsilon_n=\pm}{\bf m}^{\varepsilon_1,\cdots,\varepsilon_n}\phi(\zeta^{\varepsilon_1,\cdots,\varepsilon_n});
  \end{equation*}
  with the coherence relations ${\bf m}^\emptyset=1$ and 
  \begin{align*}
   {\bf m}^{\varepsilon_1,\cdots,\varepsilon_{n-1},+} & + {\bf m}^{\varepsilon_1,\cdots,\varepsilon_{n-1},-} = {\bf m}^{\varepsilon_1,\cdots,\varepsilon_{n-1}}\\
   {\sum_{\varepsilon_i=\pm}{\bf m}^{\varepsilon_1,\cdots,\varepsilon_i,\cdots,\varepsilon_{n}}} & {= {\bf m}^{\varepsilon_1,\cdots,\varepsilon_{i-1},\varepsilon_{i+1},\cdots,\varepsilon_{n}} \quad\forall i\in\{1,\cdots,n-1\}}
  \end{align*}
  {with ${\bf m}^{\varepsilon_1,\cdots,\varepsilon_{i-1},\varepsilon_{i+1},\cdots,\varepsilon_{n}}$ in the last condition an average over $\Omega\setminus\{\omega_i\}$ if $\Omega=\{\omega_j|j\in\N^*\}$ with the convention $0<\omega_1<\cdots<\omega_i<\cdots$ as before.}
\end{defn}
It is a simple (but not uninteresting) exercise to check that the following are examples of averages.
\begin{example} \label{ex:averages}
  \begin{itemize}
  \item Left lateral average:
  \begin{align*}
   {\bf mul}^{\varepsilon_1\cdots\varepsilon_n} = \begin{cases}
                                                   & 1 \quad \text{ if }\varepsilon_1=\cdots=\varepsilon_n=+ \\
                                                   & 0 \quad \text{ otherwise.}
                                                  \end{cases}
  \end{align*}
  \item Median average:
  \begin{equation*}
   {\bf mun}^{\varepsilon_1\cdots\varepsilon_n} = \frac{(2p)!(2q)!}{4^{p+q}(p+q)!p!q!}
  \end{equation*}
  with $p$ (resp. $q$) the number of $+$ (resp. of $-$) in $\{\varepsilon_1,\cdots,\varepsilon_n\}$.
  \item  Catalan average: Let 
  $Ca_n$ be the $n$-th Catalan number, $Qa_n(x)$ the $n$-th Catalan polynomial, $\alpha,\beta\in\R$, $\alpha+\beta=1$. 
  
  Write $\pmb{\varepsilon}=\varepsilon_1\cdots\varepsilon_n =(\pm)^{n_1}(\mp)^{n_2}\cdots(\varepsilon_s)^{n_s}$, set 
  \begin{equation*}
   {\bf man}_{(\alpha,\beta)}^{\pmb \varepsilon} = (\alpha\beta)^nCa_{n_1}\cdots Ca_{n_{s-1}}Qa_{n_s}((\alpha/\beta)^{\varepsilon_n}).
  \end{equation*}
 \end{itemize}
\end{example}
\begin{rk}
 Other type of operators\cy{, with values in different spaces} exist and are useful. {\bf Alien operators} are linear maps ${\bf op}:\resOm\longrightarrow\resOm$. As the name suggest, the famous {\bf alien derivatives} are alien operators that are derivative for the natural derivation $\partial\hat\phi:=-\zeta\hat\phi$. We refer the reader to \cite{Menous} \S1.1.3 for a detailed presentation of these objects.
 
 \cy{We will not need them here, but it is perhaps worth \ty{mentioning} that they are used in Physics \ty{literature} to compute non \ty{perturbative} contributions to some physical processes, see for example the classical work \cite{Do19} or \cite{bellon2016alien} for an application to QFT.}
\end{rk}
The notion of average is too weak to be used as such. Indeed, we want the averaged function {\bf m}$\phi$ to have properties insuring that the resummed functions it builds exists in non-trivial cases and has the right properties. More precisely we want want it to 
\begin{itemize}
 \item solve the same equation as $\phi$;
 \item be a real function\footnote{$f:U\subset\C\longrightarrow\C$ is real if $f(\bar z) =\overline{f(z)}$ whenever both sides of the equation make sense. 
 We require this condition since we want the resummed function to represent a physical quantity.}
 \item admit a Laplace transform provided that $\phi$ had a reasonable behavior at infinity.
\end{itemize}
These requirements are formalized by the notion of well-behaved average.
\begin{defn}  \label{defn:well_behaved_averages}
 An average {\bf m} is called {\bf well-behaved} if
  \begin{itemize}
   \item {\bf (P1)} It preserves the convolution ${\bf m}(\phi\star\psi) = ({\bf m}\phi)\star({\bf m}\psi)$.
   \item {\bf (P2)} It preserves reality: ${\bf m}^{\varepsilon_1\cdots\varepsilon_n} = \overline{\bf m}^{\bar\varepsilon_1\cdots\bar\varepsilon_n}$, 
   with $\bar\pm=\mp$.
   \item {\bf (P3)} It preserves exponential growths: 
   $\forall\phi\in\resOm,\zeta\in\C\setminus\Omega$
   \begin{equation*}
   |\phi(\zeta^{\pm\cdots\pm})|\leq Ke^{c|\zeta|} ~ \Longrightarrow ~ |({\bf m}\phi)(\zeta)|\leq  Ke^{c|\zeta|}
   \end{equation*}
  \end{itemize}
\end{defn}
% \begin{rk}
     The $\zeta^{\pm\cdots\pm}$ appearing in condition {\bf (P3)} can be seen as an element of $\mathcal{U}_\Omega$. So condition {\bf (P3)} can be formulated as: the average preserves exponential growth along $\mathcal{U}_\Omega$.
%     \end{rk}
\begin{rk}
 In general the equation one is studying with the Borel-\'Ecalle resummation method is a differential equation. However, averages naturally preserve the 
 differential structure: since $\mathcal{B}(\partial_z\tilde f)(\zeta)=-\zeta\hat f(\zeta)$ and since $\zeta\mapsto-\zeta$ is in $\uniOm$, 
 $\mathcal{L}[{\bf m}\mathcal{B}(\partial_z\tilde f)](z)=\partial_z\mathcal{L}[{\bf m}\mathcal{B}(\tilde f)](z)$. We used the variable $z=1/a$ for the 
 Borel transform for simplicity.
\end{rk}
The following table lists the properties of the averages of Example \ref{ex:averages}.
\begin{center}
\begin{tabular}{|l||c|c|c|}
\hline
             & {\bf (P1)} & {\bf (P2)} & {\bf (P3)} \\ \hline
 ${\bf mul}$ & {\Checkmark} & N & \Checkmark  \\ \hline
 ${\bf mun}$ & \Checkmark & \Checkmark & N  \\ \hline 
 ${\bf man}$ & \Checkmark & \Checkmark & \Checkmark  \\ \hline
\end{tabular}
\end{center}
In particular, the fact that the Catalan average is a well-behaved average is a highly non-trivial result of \cite{Menous}. A finite number of other 
families of well-behaved averages are known. It is conjectured there are no more than the ones already known. Progresses toward a classification of 
well-behaved averages were recently made in \cite{VB19}, using methods from the theory of Rota-Baxter algebras.
\begin{rk}
     One can use an average that has only conditions {\bf (P1)} and {\bf (P2)} when one works on a problem with only one (or eventually finitely many) alien derivatives acting non-trivially on the Borel transform. This is the case in non-trivial problems, e.g. for non-linear systems of ODEs of rank one studied in \cite{Co06} or the Schwinger-Dyson equation of the Yukawa model studied in \cite{BoDu20} and subsequent work.
     
     As soon as the an infinite number of alien derivatives act non-trivially on the Borel transform, condition {\bf (P3)} is needed, as implied by Lemma 5 and Equation (40) of \cite{EcMe95}. We will see in the next Section that we are in this case.
    \end{rk}

Finally, the core of the Borel-\'Ecalle resummation method can be summed up in the following theorem, which is a summary of various statements that exist in the \ty{literature}.
\begin{theo} \label{thm:Borel_Ecalle_resummation}
 Let $(E)$ a differential equation admitting a solution $\tilde f\in\C[[a]]_1$ such that $\hat f\in\resOm$ for some $\Omega=\omega\N^*\subset\R^*_+$ 
 and $|\phi(\zeta^{\pm\cdots\pm})|\leq Ke^{c|\zeta|}$ for $|\zeta|$ big enough. Let {\bf m} be a well-behaved average. Then 
 \begin{equation*}
  f^{\rm res}:=\mathcal{L}\circ{\bf m}\circ\mathcal{B}\circ\tilde f
 \end{equation*}
 is a solution of $(E)$ analytic in the open set
 \begin{equation*}
  U =\{a\in \C: |a-c/2|<c/2\}.
 \end{equation*} 
\end{theo}

\section{The Wess-Zumino model} \label{sec:WZ}

We introduce the model we are going to study and state some known facts about it. Most of these results are well-known (in particular the derivation of 
the Schwinger-Dyson and renormalisation group equations) and they can all be found for instance in my PhD thesis \cite{Cl15}. But before, we need to have a short discussion about the Physics concepts that will play a role here. This discussion does not aim at rigor, but rather to give some intuitions for the readers that would never heard of some of the concepts that will be \ty{mentioned} later.

\subsection{Some concepts from Physics} \label{subsec:physics_concepts}

One of the most remarkable idea of modern Physics is \emph{renormalisation flow} \cite{wilson1974renormalization}. Very roughly, the idea is that a theory is modified by the scale (typically the energy) at which we investigate it. For example, quantum electrodynamics will not give the same prediction if we investigate it at the energy range of a standard lightbulb or at the energy range of a large electron-positron collider. This idea is encoded in the \emph{renormalisation group equation} (RGE for short) of the theory.

Two of the most important quantities that appear in RGEs are the \emph{beta function} and the \emph{anomalous dimension}. Both can be seen as formal series of a (formal) \ty{coupling} parameter $a$. The beta function $\beta$ encodes how the (effective, measured) coupling parameter varies when one changes the energy at which the theory is studied. The anomalous dimension $\gamma$ deals with the behaviour of the theory under scale changes $x\mapsto\lambda x$. More precisely, it encodes how these variations under rescalling change with the energy at which the theory is studied. \\

In what follows, we will use the RGE of a specific model, but we will mostly be interested in another equation: the Schwinger-Dyson equation (SDE for short) of the model. SDEs were introduced in \cite{dyson1949s} and \cite{schwinger1951}. They can be described as equations of motion for Green functions of the the theory, but in my opinion they are better explained graphically. 

Classically, if a particle goes from a point A to a point B, it follows a path joining these two points, represented by a line. But in the quantum world, the particle goes simultaneously through all paths between A and B. So if we were to ``zoom on'' the line describing the classical path, quantum effects would manifest themselves as branchings. These represent paths where the particle splits into two or more, that then recombine together to give back one particle. Such processes are allowed in QFTs with interactions. Taking into account these branchings and the loops they create require to regularise and renormalise the theory, topics that were discussed at length in Chapters \ref{chap:PROP} and \ref{chap:loc} of this thesis.

However, SDEs use another idea. Zooming on the classical path, one sees quantum paths appear. But zooming again, we see even more paths appear. Physically, this correspond to going to higher order of the \ty{perturbation} theory. But one can easily convince oneself that while zooming, the same structures will keep appearing. So the Green functions of the theory have some kind of self-similarity structure, somewhat like fractals. Expressing this allow to write equations between the Green functions of the theory. These equations are the SDEs. This point of view allow to formulate SDEs as a purely combinatorial problem. With this approach L. Foissy could build combinatorial solutions of these equations in \cite{foissy2014general}. Combinatorial studies of SDEs is still an active field of research, see for example the recent article \cite{balduf2024tubings}. \\

The last concept from Physics that will be \ty{mentioned} below is Supersymmetry. Again, very roughly speaking, supersymmetric models are theories where each type of particle has a supersymmetric partner with opposite statistics. Therefore fermionic particles have bosonic partners, and vice-versa. Supersymmetry solves some issues of modern Physics, but has the drawback that no supersymmetric partners could ever be observed so far. For mathematical physicists, supersymmetric theories can be easier to handle and are therefore a good testing ground for the tools we wish to apply. This is exactly the case of the model that we will consider now.

\subsection{The model}

The Wess-Zumino model is one of the simplest possible supersymmetric model: it is massless and exactly supersymmetric. It was first introduced and 
studied in the papers \cite{WZ74a,WZ74b}, seminal to supersymmetry. This model has two features that make it suitable as a first QFT to study 
within the framework of resurgence theory. 

First, its beta function $\beta$ and anomalous dimensions $\gamma$ functions are proportional: $\beta=3\gamma$. This can in particular be shown using Hopf-algebraic techniques. It also presents the 
striking feature that it needs no vertex renormalisation, due to its (exact) supersymmetry. In words, the contributions for divergent Feynman graphs with three external legs cancel out. This is because supersymmetry imposes that for any such graph $G$, there is another graph $G^{\rm super}$ where the particles in $G$ are replaced by their supersymmetric partners, and that their evaluations exactly cancel each other: $F_{\rm WZ}(G)+F_{\rm WZ}(G^{\rm super})=0$, with $F_{\rm WZ}$ the Feynman rules of the Wess-Zumino model (see the introduction of Chapter \ref{chap:PROP}).

This implies that the SDE for the 
two-point function $G(L,a)$, truncated to the first loop, actually decouples from the SDEs for higher point functions. It reads
\begin{equation}\label{eq:SDnlin}
\left(
\tikz \node[prop]{} child[grow=east] child[grow=west];
\right)^{-1} = 1 - a \;\;
\begin{tikzpicture}[level distance = 5mm, node distance= 10mm,baseline=(x.base)]
 \node (upnode) [style=prop]{};
 \node (downnode) [below of=upnode,style=prop]{}; 
 \draw (upnode) to[out=180,in=180]   
 	node[name=x,coordinate,midway] {} (downnode);
\draw	(x)	child[grow=west] ;
\draw (upnode) to[out=0,in=0] 
 	node[name=y,coordinate,midway] {} (downnode) ;
\draw	(y) child[grow=east]  ;
\end{tikzpicture}.
\end{equation}

The other equation we are going to study is the RGE briefly discussed above. 
It takes the particularly simple form
\begin{equation} \label{eq:RGE}
 \partial_LG(L,a) = \gamma(a)(1+3a\partial_a)G(L,a)
\end{equation}
with $\gamma(a):=\partial_L {G(L,a)}|_{L=0}$ the anomalous dimension of the theory. In this equation, $a$ is the coupling constant of the theory and $L=\ln(p^2/\mu^2)$ is a kinematic parameter. $p^2$ is the impulsion going through the graphs in Equation \eqref{eq:SDnlin} and $\mu^2$ is the reference impulsion: the one at which the experiments are performed. Notice that the RGE \eqref{eq:RGE} relates the variations of the two-point function $G(L,a)$ under changes of the kinematical parameter and of the coupling constant. \\

As already stated, the goal of this chapter is to rigorously prove the Borel-\'Ecalle summability of the solution of the system \eqref{eq:SDnlin}-\eqref{eq:RGE}. For this, we will start from results of articles \cite{BC18,BC14,BC15} and use in particular analytical tools for resurgent functions developed in \cite{Sa12,Sa14}. This should be seen as a proof of concept that resurgent analysis allows us to prove summability for QFTs. However, other methods also allow to analyse summability of PDEs, see for example \cite{Co08,CoTa07} and references therein. These other methods could also be used in future studies of other QFT models. Let us now state the known results about the Wess-Zumino model.

Writing $G$ as a formal series in $L$
\begin{equation} \label{eq:L_expansion_G}
 G(L,a) = \sum_{k=0}^{+\infty}\gamma_k(a)\frac{L^k}{k!},
\end{equation}
(with $\gamma_0(a)=1$ and $\gamma_1(a)=:\gamma(a)$)
we can easily write the RGE \eqref{eq:RGE} as an induction relation on the $\gamma_k$s:
\begin{equation} \label{eq:recursion_gamma}
 \gamma_{k+1}(a) = {\gamma(a)}(1+ 3a\partial_a)\gamma_k.
\end{equation}
This justifies that we look for an equation {for} $\gamma$ rather than an equation {for} $G$. Plugging the expansion \eqref{eq:L_expansion_G} into the 
Schwinger-Dyson equation \eqref{eq:SDnlin} and computing the Feynman integral we obtain
\begin{equation} \label{SDE}
 \gamma(a) = \left.a\left(1+\sum_{n=1}^\infty\frac{\gamma_n(a)}{n!}\frac{d^n}{dx^n}\right)\left(1+\sum_{m=1}^\infty\frac{\gamma_m(a)}{m!}\frac{d^m}{dx^m}\right)H(x,y)\right|_{x=y=0},
\end{equation}
with $H$ the one-loop Mellin transform:
\begin{align}
 H(x,y) & = \frac{\Gamma(1+x)\Gamma(1+y)\Gamma(1-x-y)}{\Gamma(1-x)\Gamma(1-y)\Gamma(2+x+y)} \label{eq:Mellin0} \\
	& = \frac{1}{1+x+y}\exp\Bigl(2\sum_{k=1}^{+\infty}\frac{\zeta(2k+1)}{2k+1}\left((x+y)^{2k+1}-x^{2k+1}-y^{2k+1}\right)\Bigr). \label{eq:Mellin}
\end{align}
We will study the Borel transform of {Equation \ref{SDE}}. It maps the usual product of formal series to a convolution product and the identity function to the 
constant function $\zeta\mapsto1$. This suggests to separate the $1$ in the equation above from the rest: 
\begin{equation} \label{eq:SDE_Borel}
 \hat\gamma(\zeta) = 1 + \left.2\sum_{n=1}^\infty \frac{(1\star\hat\gamma_n)(\zeta)}{n!}\frac{d^n}{dx^n}H(x,y)\right|_{x=y=0} + \left.\sum_{n,m=1}^\infty \frac{(1\star\hat\gamma_n\star\hat\gamma_m)(\zeta)}{n!m!}\frac{d^n}{dx^n}\frac{d^m}{dy^m}H(x,y)\right|_{x=y=0}.
\end{equation}
Similarly, taking the Borel transform of the renormalisation group equation \eqref{eq:recursion_gamma} one obtains
\begin{equation} \label{eq:RGE_borel}
 \hat\gamma_{n+1}(\zeta) = \hat\gamma\star(4+3\zeta\partial_\zeta)\hat\gamma_n(\zeta).
\end{equation}
Now, $\gamma(a)$ is a formal series with coefficients in $\C$, without constant term:
\begin{equation*}
 \gamma(a) = \sum_{n=1}^{+\infty} c_n a^n.
\end{equation*}
{The asymptotic behavior of the coefficients $c_n$ was found in \cite[Equation (18)]{Be10}. The asymptotic behavior of the coefficients $c_n$ is:
\begin{equation} \label{eq:asymp_behavior_cn}
 c_{n+1} = -(3n+2+\mathcal{O}(n^{-1}))c_n.
\end{equation}
Furthermore, one easily check that the first terms of this expansion are given by $c_1=1$ and $c_2=-2$. From this result we readily derive the following handier bound.
\begin{lemma} \label{lem:bounds_cn}
 For any $n\in\N^*$, the following bounds hold
 \begin{equation*}
  (3\delta)^{n-1}(n-1)!\leq |c_n| \leq (3K)^n n!
 \end{equation*}
 for some $K>1$ and $\delta\in]0,1]$.
\end{lemma}
\begin{proof}
 The proof is by induction. The case $n=1$ holds since $c_1=1$. Assuming that both inequalities hold for $n\in\N^*$, we first have
 \begin{equation*}
  |c_{n+1}| = |3n+2+\mathcal{O}(n^{-1})||c_n| \leq 3K(n+1)|c_n|
 \end{equation*}
 provided $K$ has been chosen large enough. The upper bound of $|c_{n+1}|$ then follows from the upper bound of $|c_n|$. For the lower bound, one 
 writes
 \begin{equation*}
  |c_{n+1}| = |3n+2+\mathcal{O}(n^{-1})||c_n| 
  \geq 3n\delta|c_n|
 \end{equation*}
 (provided $\delta$ has been chosen small enough) and the lower bound of $|c_{n+1}|$ then follows from the lower bound of $|c_n|$.
\end{proof}
So in particular, $\gamma$ is 1-Gevrey. We will also need to use that it is resurgent. This was argued in \cite{BC15} with a level of rigor that is acceptable for physicists's work. However, I cannot truly say that it was proven in a rigorous mathematical sense, and I do not want to spend time and effort to give a full proof, so I state this result as a claim rather than a theorem.
\begin{claim} \label{thm:resurgence_gamma}
 $\hat\gamma$ is $\Z^*/3$-resurgent.
\end{claim}
The argument goes as follows: show that the two-point function is $1$-Gevrey (we will prove this rigorously below), so that its Borel transform is analytic in a disc around the origin. Then make the ansatz that $\widehat{G}(\zeta,L)$ can be \ty{parametrised} as a loop integral as follows
\begin{equation*}
 \widehat{G}(\zeta,L)=\oint f(\zeta,\xi)e^{3\xi L}\frac{d\xi}{\xi}.
\end{equation*}
Plug this into the Schwinger-Dyson equation, and study how the integration contour can be deformed. It turns out, that at some points it cannot and this implies that at this points, $\hat\gamma$ has singularities.

This proof is not completely rigourous because it assumes that the function $f(\zeta,\xi)$ can be extended beyond the original convergence domaine of $\widehat{G}(\zeta,L)$ and no theorem (that I know of) allow to easily show that. I do believe that the proof could be made completely rigorous, either using general method of complex analysis, or the more specialised \cite{Co08,CoTa07}. Another, quite exciting approach would be to use resurgent monomials to obtain this result. One could look for solutions of the SDE and the RGE in terms of these functions, which are resurgent and suited for resurgent analysis. We refer the reader to the preprint \cite{Fa20} for a clear and rather complete introduction to this new and exciting topic.

\section{Resurgent analysis of the RGE}

We shall now deeply dive into resurgent analysis, and starting from the results stated above, we show that the two-point function is resurgent.

\subsection{Solution of the renormalisation group equation}

We want to study the two-point function $G(L,a)$ as a formal series in $a$. We first show that $G(L,a)$ is indeed such a formal series thanks 
to the following lemma.
\begin{lemma}
 For any $L\in\C$; the formula \eqref{eq:L_expansion_G} defines a formal series in $a$ with coefficients depending on $L$.
\end{lemma}
\begin{proof}
 Since $\gamma_0(a)=1$ by definition and $\gamma_1(a)=\gamma(a)$ lies in $a\C[[a]]$ as a result of \cite{BeLoSc07}, we obtain from
 \eqref{eq:recursion_gamma} with a trivial induction that, for any $k\in\N$, $\gamma_k(a)\in a^k\C[[a]]$.
 Then, for $n\geq1$, contributions to $a^n$ in $G(L,a)$ can only come from $\gamma_1(a),\cdots,\gamma_n(a)$ and their sum is therefore finite.
\end{proof}
The fact that we had to make this small manipulation indicates that the expansion \eqref{eq:L_expansion_G} is not suited to the study of $G(L,a)$ as a
formal series in $a$. Instead we will use the two-point function written as a formal series in $a$.
\begin{equation} \label{eq:G_a_exp}
 G(L,a) = \sum_{n=0}^{+\infty} g_n(L)a^n\in A[[a]]
\end{equation}
with $A$ some suitable algebra of smooth functions or formal series. The precise nature of $A$ is given by the following Proposition:
\begin{prop} \label{prop:solution_RGE}
 The renormalisation group equation \eqref{eq:RGE} admits a solution of the form \eqref{eq:G_a_exp}, with $A=\C[L]$, explicitly given by $g_0(L)=1$ and
 \begin{equation} \label{eq:ansatz_sol_RGE}
  g_n(L) = \sum_{q=1}^n\left(\sum_{\substack{i_1,\cdots,i_q>0\\i_1+\cdots+i_q=n}}c_{i_1}\cdots c_{i_q}K_{i_1\cdots i_q}\right)\frac{L^q}{q}
 \end{equation}
 with the $c_n$ the coefficients of $\gamma(a)$ and $K_{i_1\cdots i_q}$ real numbers inductively defined for any $n\in\N$ and $q\in\{2,\cdots,n+1\}$ by $K_{n}=1$ and
 \begin{equation*}
  K_{i_1\cdots i_q} = (1+3(n+1-i_q))K_{i_1\cdots i_{q-1}}
 \end{equation*}
 with $i_1+\cdots +i_q=n+1$.
\end{prop}
\begin{proof}
 First, observe that the SDE \eqref{eq:SDnlin} taken at $a=0$ gives $G(L,0)=g_0(L)=1$. Furthermore, the RGE \eqref{eq:RGE} implies the following 
 family of differential equations (with $n\geq1$) when one replaces $G(L,a)$ by its representation \eqref{eq:G_a_exp}
 \begin{equation*}
  g_n'(L) = \sum_{p=1}^n c_p(1+3(n-p))g_{n-p}(L).
 \end{equation*}
 Notice that at this stage the derivative can be the derivative of a function or the derivative of formal series.
 
 We now prove that these equations are solved as claimed by \eqref{eq:ansatz_sol_RGE} by induction. For the case $n=1$, the equation reduces to
 $g_1'(L)=1$ since $c_1=1=g_0(L)$. This is solved to $g_1(L)=L$ since by the expansion \eqref{eq:G_a_exp}, $G(L,a)$ has only $1=g_0(L)$ as a term 
 independent of $L$. We thus find $K_1=1$ as claimed. 
 
 It will be important for the induction step to have performed the case $n=2$. Observing that $g_1(L)=c_1L$ since $c_1=1$ we find for $g_2$ the 
 equation $g'_2(L) = c_1(1+3(2-1))c_1L + c_2$. This integrates to 
 \begin{equation*}
  g_2(L) = (c_1)^2(1+3(2-1))\frac{L^2}{2} + c_2L
 \end{equation*}
 without constant term for the same reason than the case $n=1$ treated above. We then find $K_2=1$ and $K_{11}=(1+3(2-1))\frac{K_1}{2}$ 
 as claimed.
 
 Let us now assume that the statement of the proposition holds for $n\geq2$. Writing aside the term $p=n+1$, integrating and switching the sum over 
 $q$ by one we find
 \begin{equation*}
  g_{n+1}(L) = c_{n+1}L + \sum_{p=1}^nc_p(1+3(n+1-p)\sum_{q=2}^{n+1-(p-1)}\frac{L^q}{q} \sum_{\substack{i_1,\cdots,i_{q-1}>0\\i_1+\cdots+i_{q-1}=n+1-p}}c_{i_1}\cdots c_{i_{q-1}}K_{i_1\cdots i_{q-1}}.
 \end{equation*}
 As before, we do not have a constant term thanks to the expansion \eqref{eq:G_a_exp}.
 
 Noticing that $\sum_{p=1}^n\sum_{q=2}^{n+1-(p-1)}=\sum_{q=2}^{n+1}\sum_{p=1}^{n+1-(q-1)}$ we can rewrite $g_{n+1}(L)$ as 
 \begin{equation*}
  c_{n+1}L + \sum_{q=2}^{n+1}\frac{L^q}{q}\sum_{p=1}^{n+1-(q-1)}(\cdots).
 \end{equation*}
 Now we can relabel the sum over $p$ as a sum over $i_q$. Thus the sums over $p$ and $i_1,\cdots,i_{q-1}$ can be merged. We obtain
 \begin{equation*}
  g_{n+1}(L) = c_{n+1}L +  \sum_{q=2}^{n+1}\frac{L^q}{q}\left(\sum_{\substack{i_1,\cdots,i_{q}>0\\i_1+\cdots+i_{q}=n+1}}c_{i_1}\cdots c_{i_q}\underbrace{(1+3(n+1-i_q))K_{i_1\cdots i_{q-1}}}_{=:K_{i_1\cdots i_q}}\right)
 \end{equation*}
 We therefore have the right form for $g_{n+1}(L)$, $K_{n+1}=1$ and the induction relation over the $K_{i_1\cdots i_q}$ claimed in the 
 Proposition.
\end{proof}

\subsection{The two-point function is 1-Gevrey}

The next step in the Borel-Ecalle resummation procedure exposed above is to show the formal series \eqref{eq:G_a_exp} is 1-Gevrey. One can easily show that 
\begin{equation*}
 \frac{1}{q} K_{i_1\cdots i_q} \leq \frac{1}{n}K_{\underbrace{1\cdots1}_{n\text{ times}}} = (3n-2)!!!
\end{equation*}
with $n=i_1+\cdots+i_q$ and $(3n-2)!!!=\prod_{i=0}^{n-1}(3n-2-i)$.
However this bound is too crude: we need a bound that is not uniform in $q$. Indeed, one obtain from the Lemma \ref{lem:bounds_cn} that the term 
$c_{i_1}\cdots c_{i_q}$ in the solution \eqref{eq:ansatz_sol_RGE} is dominated by the case $q=1$ while the term $K_{i_1\cdots i_q}$ is dominated by 
the term $q=n$. It is the fact that these two bounds cannot be reached together that will allow to prove that the solution \eqref{eq:ansatz_sol_RGE} is 
1-Gevrey.

We first need a finer bound on the $K_{i_1\cdots i_q}$.
Recall that for $n\in\N^*$, a {\bf composition} of $n$ is a finite sequence $(i_1,\cdots,i_q)$ of strictly positive integers such that ${i_1}+\cdots+i_q=n$. 
For any composition $(i_1,\cdots,i_q)$ of $n\in\N^*$ recall that the 
{\bf multinomial number} $\binom{n}{i_1,\cdots,i_q}$ is defined by
\begin{equation*}
 \binom{n}{i_1,\ty{\ldots},i_q} := \frac{n!}{i_1!\cdots i_q!}.
\end{equation*}
These numbers famously appear in the multinomial theorem and have many important combinatorial properties.
\begin{lemma} \label{lem:bound_K_w}
 For any $n$ in $\N^*$ and composition $(i_1,\cdots,i_q)$ of $n$, we have
 \begin{equation*}
  \frac{1}{q}K_{i_1\cdots i_q} \leq \frac{3^n}{n}\binom{n}{i_1,\ty{\ldots},i_q}.
 \end{equation*}
\end{lemma}
\begin{proof}
 First, observe that, for any $n\in\N^*$, the case $q=1$ trivially hold since $K_n=1=\binom{n}{n}$. We now prove that the result holds for every $n$ and 
 every $q$ by induction over $n$.
 
 For $n=1$, the inequality trivially holds (it is the equality case). Assume it holds for all $p\in\{1,\cdots,n\}$ for some $n\in\N^*$ and let 
 $(i_1,\cdots,i_q)$ be a composition of $n+1$. We have 
 already seen {that} if $q=1$ the result holds. If $q\geq2$ we then have
 \begin{equation*}
  \frac{1}{q}K_{i_1\cdots i_q} 
  \leq (1+3(n+1-i_q)\frac{K_{i_1\cdots i_{q-1}}}{q-1} \leq (1+3(n+1-i_q)\frac{3^{n+1-i_q}}{n+1-i_q}\binom{n+1-i_q}{i_1,\ty{\ldots},i_{q-1}}
 \end{equation*}
 by the induction hypothesis, which we can use since $q\geq2$ and thus $i_q\in\{1,\cdots,n\}$.
 
 From the definition of the multinomial numbers, we have
 \begin{equation*}
  \binom{n+1-i_q}{i_1,\ty{\ldots},i_{q-1}} = \binom{n+1}{i_q}^{-1}\binom{n+1}{i_1,\ty{\ldots},i_q}.
 \end{equation*}
 The result on rank $n+1$ then follows from the observation that 
 \begin{equation*}
  \left(3+\frac{1}{n+1-i_q}\right)\binom{n+1}{i_q}^{-1} \leq 3^{i_q}
 \end{equation*}
 for every $n\in\N^*$ and $i_q\in\{1,\cdots,n\}$.
\end{proof}
With the help of Lemma \ref{lem:bounds_cn} we are now ready to prove the result justifying the title of this subsection, namely that the two-point function is 1-Gevrey.
\begin{prop} \label{prop:G_one_Gevrey}
 The two-point function $G(L,a)$ is 1-Gevrey as a formal series in $a$: for any $L\in\R$ 
 \begin{equation*}
  |g_n(L)| \leq \frac{3}{2}(18K^2 \tilde L)^n n!
 \end{equation*}
 with $\tilde L:=\max\{L,1\}$ and $K$ the constant appearing in the upper bound of $|c_n|$ in Lemma \ref{lem:bounds_cn}.
\end{prop}
\begin{rk} \label{rk:res_relevant_NP_phys1}
 In practice, we are interested in the non perturbative regime which in the Wess-Zumino model appears for $p^2=\mu^2\exp(L)\to\infty$. In this regime, we see 
 that the locus of the first singularity of the two-point function could depend on $L$ and in particular could go to zero as $L\to\infty$. We will see later 
 that this is not the case. However the first singularities of $\hat{G}(\zeta,L)$ can move in an intermediate regime.
 This indicates that the singularities of the Borel transform\footnote{at least the first one, but since a singularities in $\omega\in\C^*$
 generally produces new singularities in $\omega\N^*$ (as in Example \ref{ex:comvo_resu}), we expect that all singularities will depend on $L$, at 
 least in some 
 non perturbative regime.} contain non perturbative information of the theory. This is of course now a rather well-understood fact (and has been known for some time: see for example \cite{Ho79}) but I find this example quite striking in its simplicity. It suggests also that resurgence theory has to be an important tool to unravel 
 non perturbative aspects of QFTs.
\end{rk}
\begin{proof}
 Using Lemma \ref{lem:bounds_cn} we have
\begin{equation*}
  \left|\frac{c_n}{c_{i_1}\cdots c_{i_q}}\right| \geq \frac{(3\delta)^{n-1}(n-1)!}{(3K)^{i_1}i_1!\cdots(3K)^{i_q}i_q!}  = \frac{1}{3n}\frac{\delta^{n-1}}{K^n}\binom{n}{i_1,\ty{\ldots},i_q} \geq \frac{1}{3n}\frac{1}{K^n}\binom{n}{i_1,\ty{\ldots},i_q}.
 \end{equation*}
 Using this as an upper bound for $|c_{i_1}\cdots c_{i_q}|$ together with the bound for $\frac{1}{q}K_{i_1\cdots i_q}$ of Lemma 
 \ref{lem:bound_K_w} we obtain
 \begin{equation} \label{eq:bound_gn}
  |g_n(L)| \leq 3\sum_{q=1}^n\left(\sum_{\substack{i_1,\cdots,i_q>0 \\ i_1+\cdots+i_q=n}}(3K)^n|c_n|\right)L^q  = 3 (3K)^n|c_n| \sum_{q=1}^n\binom{n-1}{q-1}L^q
 \end{equation}
 where we have used the simple combinatorial result that there are $\binom{n-1}{q-1}$ compositions of $n$ with 
 length $q$. Using that $L^q\leq \tilde L^n$ for any $q\in\{1,\cdots,n\}$ and once more the upper bound for $|c_n|$ of Lemma \ref{lem:bounds_cn} we find 
 the result of the Theorem since $\sum_{q=1}^n\binom{n-1}{q-1}=2^{n-1}$.
\end{proof}
\begin{rk} \label{rk:res_relevant_NP_phys2}
 One can use the bound \eqref{eq:bound_gn} more directly to find a more precise bound:
 \begin{equation*}
  |g_n(L)| \leq 3 (9K^2)^n L(L+1)^{n-1}n!
 \end{equation*}
 which holds for all $L$. This bound indicates that the first singularities of the Borel transform is rejected to infinity in the perturbative 
 limit $L\to0$ (but \emph{not} that $G(L,a)$ is analytic in this limit), and therefore that the non perturbative effects encoded in the singularities 
 of the Borel transform vanish as expected in the perturbative limit $L\to 0$. This states more explicitly that the singularities of the Borel transform and therefore the non perturbative physics are reachable through standard analysis.
\end{rk}

\subsection{Resurgence of the two-point function} \label{subsec:G_res}

We will show the result that the title of the subsection suggests by using the claimed resurgence of $\hat\gamma$, which has a direct consequence:
\begin{lemma} \label{lem:gamma_n_resu}
 The function $\hat\gamma_n$ is  $\Omega$-resurgent for all $n$ in $\N^*$.
\end{lemma}
\begin{proof}
 This result is a direct consequence of the fact that the space of $\Omega$-resurgent functions is stable under convolution, derivation and
 multiplication by an analytic function together with the fact that $\hat\gamma$ is resurgent ({Claim} \ref{thm:resurgence_gamma}). This Lemma 
 is {then easily} shown by induction using the renormalisation group equation \eqref{eq:RGE_borel}.
\end{proof}
The space of resurgent functions is stable by sums, but the above Lemma is not enough to prove that 
$\sum_{n\geq1}\hat\gamma_n(\zeta)\frac{\Lambda^n}{n!}=:\hat{G}(\zeta,\Lambda)$ is $\Omega$-resurgent. We now introduce objects that will simplify the combinatorics of the proof.
\begin{defn}
 For any $n\in\N^*$ define the set $W_n$ as the subset of words written in the alphabet $\{a,b\}$ such that 
 \begin{equation*}
  W_1:=\{\emptyset\},\quad W_{n+1}:=\{(a)\sqcup w|w\in W_n\}\bigcup\{(ab)\sqcup w|w\in W_n\}
 \end{equation*}
 with $\sqcup$ the concatenation product of words. We further set $W:=\bigcup_{n\in\N^*} W_n$.
\end{defn}
The first few $W_n$ are given by
\begin{align*}
 & W_2=\{(a);(ab)\},\quad W_3=\{(aa),(aab),(aba),(abab)\} \\
 W_4=\{(aaa), & (aaab),(aaba),(aabab),(abaa),(abaab),(ababa),(ababab)\}.
\end{align*}

\begin{lemma} \label{lem:W_n}
 For any $n\in\N^*$ we have $|W_n|=2^{n-1}$.
\end{lemma}
\begin{proof}
 For any $n\in\N^*$ write $W_{n+1}=A_n\bigcup B_n$ with $A_n:=\{(a)\sqcup w|w\in W_n\}$ and $B_n:=\{(ab)\sqcup w|w\in W_n\}$. Notice that $|A_n|=|B_n|=|W_n|$. Let us check 
 that $A_n\cap B_n=\emptyset$. Let $W_{n+1}\ni w\in A_n\cap B_n$. Then it exists $w_1$ and $w_2$ in $W_n$ such that 
 \begin{equation*}
  w = (a)\sqcup w_1=(ab)\sqcup w_2.
 \end{equation*}
 This implies that $w_1\neq\emptyset$ and since every nonempty word in $W$ starts with $a$ we can write $w_1=(a)\sqcup w_3$ for some word $w_3$ not 
 necessarily in $W$. We then have $w=(a)\sqcup w_3 = (ab)\sqcup w_2$ which a contradiction. Then $A_n\cap B_n=\emptyset$ and 
 $|W_{n+1}|=2|W_n|$. The result then follows from $|W_1|=1=2^0$.
\end{proof}
For any $\delta, L>0$ and $N\in\N^*$, we will deduce a bound on $\hat\gamma_{N+1}$ from a bound on $\hat\gamma$  in the domain $\calK_{\delta,L}(\Omega)$ which contain 
 the path $\gamma$. So, fix $N\in\N^*$ and for $n\in\{1,\cdots,N+1\}$, set 
 \begin{equation*}
  \delta_n:=\frac{\delta}{2} + (n-1)\frac{\delta}{2N},\quad L_n:=L+\frac{\delta}{2} - (n-1)\frac{\delta}{2N}.
 \end{equation*}
 We did not write the dependence on $N$ of $\delta_n$ and $L_n$ to lighten the notations. Notice however that $\delta_1=\delta/2$ and $L_1=L+\delta/2$ for 
 any $N\in\N^*$.
 
 We now define a map 
 \begin{align*}
  f:W & \longrightarrow \widehat{\mathcal{R}}_\Omega \\
    w & \longmapsto f_w
 \end{align*}
 recursively by
 \begin{equation*}
  f_\emptyset(\zeta) :=  |\hat\gamma(\zeta)| + S,\quad f_{(a)\sqcup w}(\zeta):=4(f_\emptyset\star f_w)(\zeta),\quad f_{(ab)\sqcup w}(\zeta):=\frac{6NK}{\delta}(f_\emptyset\star f_w)(\zeta)
 \end{equation*}
 where $\star$ is the convolution product and where we have set 
 \begin{equation*}
  S:=\max_{\zeta\in\calK_{\delta_1,L_1}(\Omega)}|\hat\gamma(\zeta)|\quad\text{and}\quad K:= \max_{\zeta\in\calK_{\delta_1,L_1}(\Omega)}|\zeta|.
 \end{equation*}
 The map $f$ is well-defined due to the proof above that the sets $A_n$ and $B_n$ introduced above in the proof of Lemma \ref{lem:W_n} do not intersect. Furthermore its image is a subset 
 of the $\Omega$-resurgent functions since these functions are stable by convolution and by multiplication by analytic functions.
 
  The analytical part of the work is now essentially contained is the next Lemma.
 \begin{lemma} \label{lem:analytical_bound_f_w}
  For any $N\in\N^*$ and $n\in\{1,\cdots,N+1\}$ we have 
  \begin{equation*}
   |\hat\gamma_n(\zeta)|\leq \sum_{w\in W_n}f_w(\eta)
  \end{equation*}
  for any $\zeta,\eta\in\calK_{\delta_n,L_n}(\Omega)$.
 \end{lemma}
 \begin{proof}
  We prove this result by induction on $n$. For $n=1$ we have $f_\emptyset(\zeta)\geq S=\max_{\zeta\in\calK_{\delta_1,L_1}(\Omega)}|\hat\gamma(\zeta)|$ 
  and therefore the 
  lemma holds. Assume it holds for $n\in\{1,\cdots,N\}$. We then have, for any $\zeta\in\calK_{\delta_{n+1},L_{n+1}}(\Omega)$
  \begin{equation*}
   |\hat\gamma_{n+1}(\zeta)| \leq 4|\hat\gamma|\star|\hat\gamma_n|(\zeta) + 3|\hat\gamma|\star|\zeta\partial_\zeta\hat\gamma_n|(\zeta).
  \end{equation*}
  Then using the induction hypothesis and the continuity of the convolution product we have
  \begin{equation*}
   4|\hat\gamma|\star|\hat\gamma_n|(\zeta) \leq \sum_{w\in W_n} 4(f_\emptyset\star f_w)(\eta) = \sum_{w\in W_n}f_{(\star)\sqcup w}(\eta)
  \end{equation*}
  for any $\eta\in\calK_{\delta_{n},L_{n}}(\Omega)\subset\calK_{\delta_{n+1},L_{n+1}}(\Omega)$.
  
  Now, by definition, for any $\zeta\in\calK_{\delta_{n+1},L_{n+1}}(\Omega)$, the disc of center $\zeta$ and radius $\frac{\delta}{2N}$ lies in 
  $\calK_{\delta_{n},L_{n}}(\Omega)$. Therefore, using the definition of $K$ and Cauchy's estimate (see for example \cite[Theorem 10.26]{rudin1986}) on the disc of center $\zeta$ and radius 
  $\frac{\delta}{2N}$ we find
  \begin{equation*}
   |\zeta\partial_\zeta\hat\gamma_n(\zeta)| \leq \frac{2NK}{\delta}\max_{\zeta\in D(\zeta,\delta/2N)}|\hat\gamma_n(\zeta)| \leq \sum_{w\in W_n}\frac{2NK}{\delta}f_w(\eta)
  \end{equation*}
  for any $\eta\in\calK_{\delta_{n},L_{n}}(\Omega)\subset\calK_{\delta_{n+1},L_{n+1}}(\Omega)$. Thus
  \begin{equation*}
   3|\hat\gamma|\star|\zeta\partial_\zeta\hat\gamma_n|(\zeta) \leq \sum_{w\in W_n}\frac{6NK}{\delta}(f_\emptyset\star f_w)(\eta) = f_{(\star.)\sqcup w}(\eta)
  \end{equation*}
  for any $\eta\in\calK_{\delta_{n},L_{n}}(\Omega)\subset\calK_{\delta_{n+1},L_{n+1}}(\Omega)$. Combining this bound with the one for $4|\hat\gamma|\star|\hat\gamma_n|(\zeta)$
  we obtain
  \begin{equation*}
   |\hat\gamma_{n+1}(\zeta)| \leq \sum_{w\in W_n}\left(f_{(\star)\sqcup w}(\eta) + f_{(\star.)\sqcup w}(\eta)\right) = \sum_{w\in W_{n+1}}f_{w}(\eta)
  \end{equation*}
  for any $\eta\in\calK_{\delta_{n+1},L_{n+1}}(\Omega)$.
 \end{proof}

Finally, we need a technical lemma about analytic continuation of series.
\begin{lemma} \label{lem:analytic_continuation_series}
 Let $U\subset V$ be two open subsets of $\C$. Let $f_n:U\mapsto \C$ be a sequence of holomorphic functions such that:
 \begin{enumerate}
  \item $f:=\sum_{n=0}^\infty f_n$ is holomorphic in $U$;
  \item $f_n$ admits an analytic continuation $\tilde f_n$ to $V$;
  \item $\tilde f_n$ is bounded on $V$ by an analytic function $F_n$: $|\tilde f_n|\leq F_n$;
  \item The series $F=\sum_{n=0}^\infty F_n$ converges in $V$.
 \end{enumerate}
 Then $f$ admits an analytic continuation $\tilde f$ to $V$ and $|\tilde f|\leq F$.
\end{lemma}
\begin{proof}
 For any $z\in V$, let us set 
 \begin{equation*}
  S_N(z) := \sum_{n=0}^N|\tilde f_n(z)| \leq \sum_{n=0}^N F_n(z) \longrightarrow F(z)
 \end{equation*}
 as ${N\to \infty}$. Then $S_N(z)$ is increasing and bounded and therefore convergent. {Hence} the series $\tilde f(z) := \sum_{n=0}^\infty\tilde f_n(z)$ is absolutely 
 convergent and thus convergent. This series {is by definition} an analytic continuation of $f$ to $V$ and is bounded by $F$.
\end{proof}
We can now prove the main result of this section.
\begin{theo} \label{thm:resurgence_two_points_function}
 For any $\Lambda\in\R$, the map $\zeta\mapsto\hat{G}(\zeta,\Lambda)$ is $\Omega$-resurgent. 
\end{theo}
\begin{proof}
 For $\hat\phi\in\C\{\zeta\}$ an $\Omega$-resurgent function and $\gamma$ a rectifiable path in $\C\setminus\Omega$, we denote by cont$_{\gamma}(\hat\phi)$ the analytic continuation of $\hat\phi$ along the path $\gamma$. This notation is standard in the \ty{literature} of resurgence theory.
 
 Let $\delta,L>0$ with $\delta<\rho(\Omega)/2$. Let $\gamma$ be a path in $\calK_{\delta,L}(\Omega)$.
 According to Lemma \ref{lem:analytic_continuation_series} we only need to prove that the series 
 \begin{equation*}
  \sum_{n\geq1}(\text{cont}_\gamma\hat\gamma_n)(\zeta)\frac{\Lambda^n}{n!}
 \end{equation*}
 converges normally. Indeed, in this case, this series will be equal to a continuation of $\zeta\mapsto\hat{G}(\zeta,\Lambda)$
 \begin{equation*}
  (\text{cont}_\gamma\hat G)(\zeta,\Lambda):=\left(\text{cont}_\gamma\sum_{n=1}^{\infty}\hat\gamma_n\right)(\zeta).
 \end{equation*}
 From Lemma \eqref{lem:analytical_bound_f_w} this can be done by bounding the $f_w(\zeta)$ on $\calK_{\delta_n,L_n}(\Omega)$.
 
 Let $||w||_b$ be the number of times the 
 letter $b$ is present in the word $w\in W$. Then for any $n\in\{1,\cdots,N+1\}$ and $w\in W_n$ we have
 \begin{equation*}
  f_w(\zeta) = \left(\frac{6NK}{\delta}\right)^{||w||_b}4^{n-||w||_b}f_\emptyset^{\star n}(\zeta).
 \end{equation*}
 We can now use Sauzin's bound \eqref{eq:bound_conv_resu} for $n=N+1$:
 \begin{equation*}
  \max_{\zeta\in\calK_{\delta,L}(\Omega)} f_w(\zeta) \leq \left(\frac{6NK}{\delta}\right)^N4^{N+1}\frac{C^{N+1}}{(N+1)!}\left(\max_{\zeta\in\calK_{\delta/2,L+\delta/2}(\Omega)}f_\emptyset(\zeta)\right)^{N+1}
 \end{equation*}
 where we have used that $||w||_b\in\{0,1,\cdots,N\}$. Now, using that $\delta/2=\delta_1$ and $L+\delta/2=L_1$ we find 
 $\max_{\zeta\in\calK_{\delta/2,L+\delta/2}(\Omega)}f_\emptyset(\zeta)=2S$. Using Lemmas \ref{lem:analytical_bound_f_w} and \ref{lem:W_n} we obtain
 \begin{equation*}
  \max_{\zeta\in\calK_{\delta,L}(\Omega)}|\hat\gamma_{N+1}(\zeta)| \leq \frac{\delta}{12K}\left(\frac{96}{\delta}SKC\right)^{N+1}\frac{N^N}{(N+1)!}.
 \end{equation*}
 Using Stirling's formula we then have the following bound, for $N$ big 
 \begin{equation*}
  \max_{\zeta\in\calK_{\delta,L}(\Omega)}|\hat\gamma_{N+1}(\zeta)| \leq \frac{\delta}{12Ke}\left(\frac{96}{\delta}SKCe\right)^{N+1} \frac{1}{N\sqrt{2\pi N}}\left(1+\mathcal{O}\left(\frac{1}{\sqrt{N}}\right)\right).
 \end{equation*}
 This implies the normal convergence of the series 
 $\sum_{n\geq1}(\text{cont}_\gamma\hat\gamma_n)(\zeta)\frac{\Lambda^n}{n!}=:(\text{cont}_\gamma\hat G)(\zeta,\Lambda)$ and concludes the proof.
 \end{proof}
 \begin{rk}
  Lemma \ref{lem:gamma_n_resu} and this result imply that, if one excludes  miraculous cancellation of singularities, an infinite number of alien derivatives act non-trivially on  $\zeta\mapsto\hat G(\zeta,\Lambda)$. This is corroborated by the computations of \cite[Section 4.1]{BC18} where the main contributions to the (lateral) alien derivatives applied to $\hat G$ were computed and shown to be non-zero. Therefore, the full theory or well-behaved averages is needed for the summation of the two-point function of the Wess-Zumino model.
\end{rk}

\section{Asymptotic bound of the two-point function}
 
 As explained above, we now need to prove that  $\hat G(\zeta,L)$ admits an exponential bound in the star-shaped domain $\mathcal{U}_\Omega$ of $\C\dsetminus\Omega$.
 
 \subsection{The need for Dyson-Schwinger} \label{subsect:need_SDE}
 
 So far we have only used the renormalisation group equation (except of course in the results we used as starting points). The next lemma implies that one actually needs to study the Schwinger-Dyson equation in order to find the right type of bound on the two-point function.
\begin{lemma} \label{lem:bound_gamma_n}
 Let $g:\mathcal{U}_\Omega\longrightarrow\R_+$ be an increasing analytic function such that, for any $\zeta\in\mathcal{U}_\Omega$
 \begin{equation*}
  \max\Big\{\max_{\eta\in[0,\zeta]}|\hat\gamma(\eta)|,\max_{\eta\in[0,\zeta]}|\hat\gamma'(\eta)|\Big\} \leq g(\zeta).
 \end{equation*}
 Then for any $n\in\N^*$ we have
 \begin{equation*}
  \max\Big\{\max_{\eta\in[0,\zeta]}|\hat\gamma_n(\eta)|,\max_{\eta\in[0,\zeta]}|\hat\gamma_n'(\eta)|\Big\} \leq \left[(4+3|\zeta|)(1 + g(\zeta)|\zeta|)\right]^{n-1} g(\zeta).
 \end{equation*}
\end{lemma}
\begin{rk}
 The function $g$ exists since $\hat\gamma$ and $\hat\gamma'$ are analytic (but not bounded) on $\mathcal{U}_\Omega$.
\end{rk}
\begin{proof}
 We prove this Lemma by induction. The case $n=1$ holds by definition of $g$. Assuming the Lemma holds for some $n\in\N^*$; we use the bound 
 \eqref{eq:bound_star_shaped} (which we can use on $\mathcal{U}_\Omega$ since it is star-shaped with respect to the origin) 
 on the renormalisation group equation \eqref{eq:RGE_borel} to obtain, for any $\zeta\in\mathcal{U}_\Omega$
 \begin{align*}
  |\hat\gamma_{n+1}(\zeta)| & \leq g(\zeta)|\zeta|(4\max_{\eta\in[0,\zeta]}|\hat\gamma_n(\eta)|+3|\zeta|\max_{\eta\in[0,\zeta]}|\hat\gamma'_n(\eta)|) \\
			    & \leq (4+3|\zeta|)g(\zeta)|\zeta|\max\Big\{\max_{\eta\in[0,\zeta]}|\hat\gamma_n(\eta)|,\max_{\eta\in[0,\zeta]}|\hat\gamma_n'(\eta)|\Big\} \\
			    & \leq (4+3|\zeta|)(1+g(\zeta)|\zeta|)\max\Big\{\max_{\eta\in[0,\zeta]}|\hat\gamma(\eta)|,\max_{\eta\in[0,\zeta]}|\hat\gamma'(\eta)|\Big\}.
 \end{align*}
 For any $\eta\in[0,\zeta]$ we further have
 \begin{align*}
  |\hat\gamma_{n+1}(\eta)| & \leq (4+3|\eta|)(1+g(\eta)|\eta|)\max\Big\{\max_{\sigma\in[0,\eta]}|\hat\gamma_n(\sigma)|,\max_{\sigma\in[0,\eta]}|\hat\gamma_n'(\sigma)|\Big\} \\
			    & \leq (4+3|\zeta|)(1+g(\zeta)|\zeta|)\max\Big\{\max_{\eta\in[0,\zeta]}|\hat\gamma(\eta)|,\max_{\eta\in[0,\zeta]}|\hat\gamma'(\eta)|\Big\}
 \end{align*}
 since we have assumed $g$ to be increasing.
 Therefore $\max_{\eta\in[0,\zeta]}|\hat\gamma_{n+1}(\zeta)|$ admits the bound of the Lemma.
 
 To obtain a bound on $|\hat\gamma_{n+1}'(\zeta)|$ we use Leibniz's formula
 \begin{equation} \label{eq:Leibniz}
  \frac{d}{dt}\int_{a(t)}^{b(t)} f(t,x)dx = b'(t)f(t,b(t)) - a'(t)f(t,a(t)) + \int_{a(t)}^{b(t)} \frac{\partial f}{\partial t}(t,x)dx;
 \end{equation}
 which holds provided $a$, $b$  and $f$ are differentiable with continuous derivatives.
 
 In our case this formula gives
 \begin{equation*}
  \partial_\zeta(f\star g)(\zeta) = f(0)g(\zeta) + \int_0^\zeta f'(\zeta-\eta)g(\eta) d\eta = f(\zeta)g(0) + \int_0^\zeta f(\zeta-\eta) g'(\eta) d\eta.
 \end{equation*}
 (one gets the second equality through an integration by part). Using $\hat\gamma(0)=1$ and again the bound \eqref{eq:bound_star_shaped} {on} the
 renormalisation group equation \eqref{eq:RGE_borel} {derived once one obtains}, for any $\zeta\in\mathcal{U}_\Omega$
 \begin{equation*}
  |\hat\gamma_{n+1}'(\zeta)|\leq \left[(4+3|\zeta|)(1 + g(\zeta)|\zeta|)\right]\max\Big\{\max_{\eta\in[0,\zeta]}|\hat\gamma_n(\eta)|,\max_{\eta\in[0,\zeta]}|\hat\gamma_n'(\eta)|\Big\}
 \end{equation*}
 The same bound holds for any $\eta\in[0,\zeta]$ from the same argument than the one used for $\hat\gamma_n$.
 
 From these bounds, the Lemma holds by induction.
 \end{proof}
 Summing these $\hat\gamma_n$ we end up with the following asymptotic bound for the two-point function:
\begin{equation*}
 |\hat G(\zeta,L)| \leq K\exp(c|\zeta|^2g(\zeta)L),
\end{equation*}
for some bound $g(\zeta)$ of $\hat\gamma$ and $\hat\gamma'$ at infinity. This is too weak a bound to apply the Borel-\'Ecalle resummation method. The 
square of $|\zeta|$ comes from the $\zeta$ in the renormalisation group 
equation \eqref{eq:RGE_borel} and the $\zeta^{n-1}$ in the Equation \eqref{eq:bound_star_shaped}, which we used with $n=2$. In order to apply 
Borel-\'Ecalle resummation without accelero-summation, we have two challenges to tackle:
 \begin{itemize}
 \item relate the bounds for $\hat\gamma_n$ and for $\hat\gamma'_n$ in order to get rid of one of the powers of $\zeta$;
 \item find a specific bound on the asymptotic behavior of $\hat\gamma$.
\end{itemize}
The second issue will be solved using the Schwinger-Dyson equation, and the solution of the first one will actually use inputs from the Schwinger-Dyson 
equation as well.

 \subsection{The Schwinger-Dyson equation revisited}
 
 Expanding the sum in the Schwinger-Dyson equation in the Borel plane (Equation \eqref{eq:SDE_Borel}), and using $\mathcal{B}(af(a))=1\star \hat f$ we find 
\begin{equation*}
  \hat\gamma(\zeta) = 1 +2\sum_{n=1}^{+\infty}X_{0n}(1\star\hat\gamma_n)(\zeta) + \sum_{n,m=1}^{+\infty}X_{nm}(1\star\hat\gamma_n\star\hat\gamma_m)(\zeta).
  \end{equation*}
with 
\begin{equation*}
 X_{nm}:=\frac{1}{n!m!}\frac{d^n}{dx^n}\frac{d^m}{dy^m}H(x,y)|_{x=y=0}.
\end{equation*}
Now, observe that the series 
$\sum_{k=1}^{+\infty}\frac{\zeta(2k+1)}{2k+1}\left((x+y)^{2k+1}-x^{2k+1}-y^{2k+1}\right)$ contains no terms of the form 
$x^Ny^0$ nor $x^0y^N$. Thus 
\begin{equation*}
 \partial_x^n \left.\exp\Bigl(2\sum_{k=1}^{+\infty}\frac{\zeta(2k+1)}{2k+1}\left((x+y)^{2k+1}-x^{2k+1}-y^{2k+1}\right)\Bigr)\right|_{x=y=0} = 0;
\end{equation*}
and the same holds for the derivatives with respect to $y$. Therefore, the representation \eqref{eq:Mellin} of the Mellin transform $H$ gives us $X_{0n}=X_{n0}=(-1)^n$. We thus find the Schwinger-Dyson equation in the Borel plane: 
\begin{equation} \label{eq:SDE_Borel_expanded}
 \hat\gamma(\zeta) = 1 +2\sum_{n=1}^{+\infty}(-1)^n(1\star\hat\gamma_n)(\zeta) + \sum_{n,m=1}^{+\infty}X_{nm}(1\star\hat\gamma_n\star\hat\gamma_m)(\zeta).
\end{equation}
\begin{rk} \label{rk:analytic_continuation_bound}
 It is crucial to the rest of this proof to realise that, while Equation \eqref{eq:SDE_Borel_expanded} holds for any $\zeta\in\C\dsetminus\Omega$, the 
 series on the R.H.S. only converge in a small open subset of $\C\dsetminus\Omega$ which is mapped to a neighborhood of the origin in $\C$. Indeed, taking a derivative of this equation we obtain 
 \begin{equation*}
  \hat\gamma'(\zeta) = 2\sum_{n=1}^{+\infty}(-1)^n\hat\gamma_n(\zeta) + \sum_{n,m=1}^{+\infty}X_{nm}(\hat\gamma_n\star\hat\gamma_m)(\zeta).
 \end{equation*}
 The renormalisation group equation \eqref{eq:RGE_borel} together with the result of \cite{BC15} that $\hat\gamma(\zeta)\sim A\ln(1/3-\zeta)$ when 
 $\zeta$ goes to $1/3$ implies that $\hat\gamma_n$ has the same behavior when $\zeta$ goes to $1/3$. Thus 
 $\sum_{n=1}^{+\infty}(-1)^n\hat\gamma_n(\zeta)$ trivially diverges in an open set close to $1/3$.
 
 Therefore, the series of the R.H.S. of \eqref{eq:SDE_Borel_expanded} should be read as the analytic continuation of these series when one is away from 
 their convergent domain. This will be important since we will use bounds on $\hat\gamma_n$ of the form of  the bounds of Lemma \ref{lem:bound_gamma_n}
 which hold for any $\zeta\in\mathcal{U}_\Omega$. Provided the series of these bounds will admit an analytic extension to the whole of 
 $\mathcal{U}_\Omega$, it will provide a bound for $\hat\gamma$ as needed.
\end{rk}
One can compute the numbers $X_{nm}$ using the same type of argument we used to find $X_{n0}$, or directly using the Fa\`a-di-Bruno formula. 
However the result of this computation is not particularly enlightening. It will be enough for us to find a bound for $|X_{nm}|$.
\begin{lemma} \label{lem:bound_X_nm}
 For any any $r\in]0,1/2[$ it exists a {real positive number}  $K_r>0$ such that, for any $n,m\in\N^*$ we have
 \begin{equation} \label{eq:X_nm}
  |X_{nm}| \leq \frac{K_r}{r^{n+m}}.
 \end{equation}
\end{lemma}
\begin{proof} 
 We use the multivariate Cauchy inequality (see for example \cite[Theorem 2.2.7]{Hormander66}); namely that if a function $f:\C^n\longrightarrow\C$ is analytic and bounded by $M$ in the polydisc $\{z:|z_i|\leq r_i,~i=1,\cdots,n\}$, then $|\partial^{\alpha}f(0)|\leq M\frac{\alpha!}{r^\alpha}$ for 
 any multi-index $\alpha\in\N^n$ and with obvious notations for factorials and powers. According to \eqref{eq:Mellin0}, the Mellin transform $H$ is 
 analytic in the polydisc $\{(z_1,z_2)\in\C^2:|z_1|\leq r~\wedge~|z_2|\leq r\}$ for any $r\in]0,1/2[$. For any such $r$, set 
 $K_r:=\sup_{|z_1\leq r,z_2\leq r}|H(z_1,z_2)|$. The bound \eqref{eq:X_nm} follows then directly from the multivariate Cauchy inequality.
\end{proof}

\subsection{Intermediate bounds} \label{intermediate}

We start with a common bound of $\hat\gamma$ and $\zeta\partial_\zeta\hat\gamma$ to find bounds on 
$\hat\gamma_n$ and $\hat\gamma_n'$ for any $n\in\N^*$.
\begin{lemma} \label{lem:un_autre_lemme}
 Let $g:\mathcal{U}_\Omega\setminus\{0\}\longrightarrow\R$ be a holomorphic function increasing with $|\zeta|$ such that, for any $\zeta\in\mathcal{U}_\Omega\setminus\{0\}$, 
\begin{equation*}
 \max_{\eta\in[0,\zeta]}|\hat\gamma(\eta)|\leq g(\zeta) \quad {\rm and} \quad \max_{\eta\in[0,\zeta]}|\hat\gamma'(\eta)|\leq \frac{g(\zeta)}{|\zeta|}.
\end{equation*}
Let $(g_n)_{n\in\N^*}$ and $(h_n)_{n\in\N^*}$ be two 
 sequences of functions from $\mathcal{U}_\Omega\setminus \{0\}$ 
 to 
 $\R$ inductively defined for any $\zeta\in\mathcal{U}_\Omega\setminus \{0\}$ by $g_1{(\zeta)}:=g(\zeta)$, $h_1{(\zeta)}:=g(\zeta)/|\zeta|$ and 
 \begin{align*}
  g_{n+1}(\zeta) := g(\zeta)|\zeta|\left[4g_n(\zeta)+3|\zeta| h_n(\zeta)\right],\\
  h_{n+1}(\zeta) := \frac{g_{n+1}(\zeta)}{|\zeta|} + 4g_n(\zeta)+3|\zeta|h_n(\zeta).
 \end{align*}
 Then, for any $n\in\N^*$ and $\zeta\in\mathcal{U}_\Omega\setminus \{0\}$
 \begin{equation*}
  \max_{\eta\in[0,\zeta]}|\hat\gamma_n(\eta)| \leq g_n(\zeta), \quad \max_{\eta\in[0,\zeta]}|\hat\gamma_n'(\eta)| \leq h_n(\zeta).
 \end{equation*}
\end{lemma}
\begin{rk}
 Such a function $g$ exists since $\hat\gamma$ and $\zeta\partial_\zeta\hat\gamma$ are analytic on $\mathcal{U}_\Omega$. From such a bound $g$ we will later derive the existence of the bounds we need.
\end{rk}
\begin{proof}
 We prove this by induction: the case $n=1$ holds by definition of $g$.
 
 Assuming the result holds for $n\in\N^*$, using the renormalisation group equation \eqref{eq:RGE_borel}, the 
 bound \eqref{eq:bound_star_shaped} and the induction hypothesis we obtain
 \begin{equation*}
  |\hat\gamma_{n+1}(\zeta)| \leq g(\zeta)|\zeta|\left[4g_n(\zeta)+3|\zeta| h_n(\zeta)\right] =: g_{n+1}(\zeta).
 \end{equation*}
 Taking once again the derivative of the renormalisation group equation \eqref{eq:RGE_borel} we obtain, using Leibniz's formula \eqref{eq:Leibniz}
 \begin{equation*}
  \hat\gamma_{n+1}'(\zeta) = 4[\hat\gamma_n(\zeta) + (\hat\gamma'\star\hat\gamma_n)(\zeta)] + 3 [\zeta\hat\gamma_n'(\zeta) + (\hat\gamma'\star(\zeta\hat\gamma_n'))(\zeta)].
 \end{equation*}
 Using the bound \eqref{eq:bound_star_shaped} and the induction hypothesis on this equation gives the result for $\zeta$. The case of 
 $\eta\in[0,\zeta]$ holds from the same argument than the one of Lemma \ref{lem:bound_gamma_n}, which still holds since we assume $g$ to be increasing.
\end{proof}
We can now express together the bounds of $\hat\gamma_n$ and $\hat\gamma_n'$. 
\begin{lemma} \label{lem:encore_un_lemme}
 For any $\zeta\in\mathcal{U}_\Omega\setminus \{0\}$, and $g$ a bound of $\hat\gamma$ and $\zeta\hat\gamma'$ as in Lemma \ref{lem:un_autre_lemme} set 
 \begin{equation*}
  \alpha_g(\zeta) : = \frac{g(\zeta)}{g(\zeta)+1}.
 \end{equation*}
 Then, for any 
 $n\in\N^*$ and any $\zeta\in\mathcal{U}_\Omega\setminus \{0\}$
 \begin{equation*}
  h_n(\zeta) \leq \frac{1}{\alpha_g(\zeta)}\frac{g_n(\zeta)}{|\zeta|}.
 \end{equation*}
\end{lemma}
\begin{proof}
 For $n=1$, the inequality to show is the case $n=1$ of Lemma \ref{lem:un_autre_lemme} since $1/\alpha(\zeta)>1$.
 
 For $n=2$, direct computations give 
 \begin{equation*}
  \frac{1}{\alpha(\zeta)}\frac{g_2(\zeta)}{|\zeta|} = 7g(\zeta)(g(\zeta)+1) \geq h_2(\zeta)=14g(\zeta)
 \end{equation*}
 since $g(\zeta)\geq\max_{\eta\in[0,\zeta]}|\hat\gamma(\zeta)| \geq 1=\hat\gamma(0)$. 
 
 For any $n\geq2$ we have 
 \begin{equation*}
  \frac{1}{\alpha(\zeta)}\frac{g_{n+1}(\zeta)}{|\zeta|} = (g(\zeta)+1)[4g_n(\zeta)+3|\zeta|h_n(\zeta)] = h_{n+1}(\zeta).
 \end{equation*}
 Therefore the result also holds for any $n\geq2$.
\end{proof}
We can now prove the main result of this subsection.
\begin{prop} \label{prop:main_bound}
 Let $g:\mathcal{U}_\Omega\longrightarrow\R$ be a bound of $\hat\gamma$ and $\zeta\hat\gamma'$ as in Lemma 
 \ref{lem:un_autre_lemme}. 
 Then, for any $n\in\N^*$ and $\zeta\in\mathcal{U}_\Omega\setminus \{0\}$ 
 \begin{equation*}
  \max_{\eta\in[0,\zeta]}|\hat\gamma_n(\eta)| \leq \left[(7g(\zeta)+3)|\zeta|\right]^{n-1}g(\zeta).
 \end{equation*}
\end{prop}
\begin{proof}
 By Lemma \ref{lem:un_autre_lemme} it is sufficient to prove $g_n(\zeta)\leq \left[(7g(\zeta)+3)|\zeta|\right]^{n-1}g(\zeta)$ for any $n\in\N^*$. We prove 
 this by induction: the case $n=1$ trivially holds. Assuming the result holds for $n\in\N^*$, we have according to Lemma \ref{lem:encore_un_lemme}
 \begin{equation*}
  g_{n+1}(\zeta) \leq g(\zeta)|\zeta|\left(4+\frac{3}{\alpha(\zeta)}\right)g_n(\zeta) = |\zeta|\left(7g(\zeta)+3)\right)g_n(\zeta)
 \end{equation*}
 by definition of $\alpha(\zeta)$.
\end{proof}

\subsection{Borel-\'Ecalle resummation of the two-point function} \label{subsect:final}

We now need to find the right bound $g$ used everywhere in our proofs to derive the correct asymptotic behavior for $\hat G$.
\begin{prop} \label{prop:bound_gamma_infinity}
 {On  $\mathcal{U}_\Omega$,} $|\hat\gamma(\zeta)|$ {and $|\hat\gamma'(\zeta)|$ are} bounded in a neighborhood of infinity by $1$ {and $1/|\zeta|$ respectively}.
\end{prop}
\begin{proof}
 As before
 let $g:\mathcal{U}_\Omega{\setminus\{0\}}\longrightarrow\R$ be a bound of $\hat\gamma$ and $\zeta\hat\gamma'$ as in Lemma 
 \ref{lem:un_autre_lemme}. 
 Using the bound \eqref{eq:bound_star_shaped} on the Schwinger-Dyson equation \eqref{eq:SDE_Borel_expanded}
 with the bounds of Proposition \ref{prop:main_bound} for {the} $\hat\gamma_n$ and the bounds of Lemma \ref{lem:bound_X_nm} for the coefficients $X_{nm}$ we 
 find that $|\hat\gamma|$ is bounded on $\mathcal{U}_\Omega\setminus\{0\}$ by two geometric series. {More properly, and in the spirit of Remark \ref{rk:analytic_continuation_bound},} $|\hat\gamma|$ is bounded in $\mathcal{U}_\Omega\setminus \{0\}$ by the analytic continuation of ({products} of) geometric series. To be more 
 precise, one has 
 \begin{align*}
  |\hat\gamma(\zeta)| & \leq 1 + 2|\zeta|\sum_{n=1}^{\infty} \max_{\eta\in[0,\zeta]}|\hat\gamma_n(\eta)| + \frac{K_r}{2}\max_{\eta\in[0,\zeta]}|\zeta|^2\sum_{n,m=1}^{\infty}\frac{1}{r^{n+m}}\max_{\eta\in[0,\zeta]}|\hat\gamma_n(\zeta)|\max_{\eta\in[0,\zeta]}|\hat\gamma_m(\zeta)| \\
		      & \leq 1 +\frac{2|\zeta| g(\zeta)}{1-(7g(\zeta)+3)|\zeta|} + K_r\left(\frac{|\zeta|g(\zeta)}{r-(7g(\zeta)+3)|\zeta|}\right)^2 =: G(\zeta,g(\zeta))
 \end{align*}
 for any $r\in]0,1/2[$, $\zeta\in\mathcal{U}_\Omega{\setminus\{0\}}$ and with $K_r:=\sup_{|z_1\leq r,z_2\leq r}|H(z_1,z_2)|$. 
 Notice that we removed the $1/2$ in {the third term of} the last 
 bound in order for $G$ to have the following property: for any $\zeta\in\mathcal{U}_\Omega\setminus{\{0\}}$
 \begin{equation} \label{eq:bound_gamma_prime}
  |\hat\gamma'(\zeta)| \leq \frac{G(\zeta,g(\zeta))}{|\zeta|}.
 \end{equation}
 To prove this, we take the derivative of the Schwinger-Dyson equation \eqref{eq:SDE_Borel_expanded}:
\begin{equation*}
 \hat\gamma'(\zeta) = 2\sum_{n=1}^{+\infty}(-1)^n\hat\gamma_n(\zeta) + \sum_{n,m=1}^{+\infty}X_{nm}(\hat\gamma_n\star\hat\gamma_m)(\zeta).
\end{equation*}
Therefore 
\begin{align*}
 |\hat\gamma'(\eta)| & \leq 2\sum_{n=1}^{+\infty}|\hat\gamma_n(\zeta)| + \sum_{n,m=1}^{+\infty}|X_{nm}(\hat\gamma_n\star\hat\gamma_m)(\zeta)| \\
		      & \leq 2\sum_{n=1}^{\infty} \max_{\eta\in[0,\zeta]}|\hat\gamma_n(\eta)| + K_r|\zeta|\sum_{n,m=1}^{\infty}\frac{1}{r^{n+m}}\max_{\eta\in[0,\zeta]}|\hat\gamma_n(\eta)|\max_{\eta\in[0,\eta]}|\hat\gamma_m(\eta)| \\
		      & \leq \frac{1}{|\zeta|} + \sum_{n=1}^{\infty}\left[(7g(\zeta)+3)|\zeta|\right]^{n-1}g(\zeta)+ K_r|\zeta|\sum_{n,m=1}^{\infty} \frac{1}{r^{n+m}}\left[(7g(\zeta)+3)|\zeta|\right]^{n+m-2}g(\zeta)^2 \\
		      &= \frac{G(\zeta,g(\zeta))}{|\zeta|}
\end{align*}
as claimed, and where we have used Proposition \ref{prop:main_bound} in the last inequality.

It is a cumbersome but simple exercise to study the variations of $G$. However it is enough for the task at hand to check that $G$ is bounded at infinity by $1$. For $\zeta$ {in 
$\mathcal{U}_\Omega$}, we have 
\begin{equation*}
 G(\zeta,X)\sim  \underbrace{1 - \frac{2 X}{7X+3} + K_r\left(\frac{X}{7X+3}\right)^2}_{=: f(X)} +\mathcal{O}(|\zeta|^{-1})
\end{equation*}
for $|\zeta|\to\infty$. We can still choose $r\in]0,1/2[$. Since $H(0,0)=1$ and since $H$ is holomorphic in a neighborhood of $(0,0)$, we can take $r$ small enough 
{for $K_r$ to be} arbitrarily close to $1=H(0,0)$. It then is a simple exercise of real analysis to show that, provided $K_r<7$, $f$ is continuous and decreases over 
 $\R_+^*$. Therefore
 \begin{equation*}
  |\hat\gamma(\zeta)| \lesssim f(0) = 1
 \end{equation*}
 in a neighborhood of infinity. The bound for $\hat\gamma'$ in the same neighborhood of infinity comes from the inequality \eqref{eq:bound_gamma_prime}.
\end{proof}
The bounds for $\hat\gamma$ and $\zeta\hat\gamma$ that we just proved allow us to straightforwardly prove the needed bound for the two-point function.
\begin{theo} \label{thm:bound_two_point_infinity}
 There exist real constants $K,~M>0$ such that, for any $L\in\R$, the Borel transform of the solution of the Schwinger-Dyson equation \eqref{eq:SDnlin} and the renormalisation group equation 
 \eqref{eq:RGE} admits the following bound in $\mathcal{U}_\Omega$ around infinity
 \begin{equation*}
  |\hat G(\zeta,L)| \leq \frac{K}{|\zeta|}\exp\left(M|\zeta|L\right).
 \end{equation*}
\end{theo}
\begin{proof}
 From Proposition \ref{prop:bound_gamma_infinity} we can find a bound of $g$ of $\hat\gamma$ and $\hat\gamma'$ which is increasing and bounded at infinity. Using such a bound in Proposition \ref{prop:main_bound} we obtain
 \begin{align*}
  |\hat G(\zeta,L)| & \leq \sum_{n=1}^\infty[(7g(\zeta)+3)|\zeta|]^{n-1} g(\zeta)\frac{L^n}{n!} \\
		    & = \frac{g(\zeta)}{(7g(\zeta)+3)|\zeta|}{\Big(\exp\left[(7g(\zeta)+3)|\zeta|L\right] - 1\Big)} \\
		    & \leq \frac{{K}}{|\zeta|}\exp(M|\zeta|L)
 \end{align*}
 for some $K>0$, and where we have set $M:=7{[\sup_{\zeta\in\mathcal{U}_\Omega}g(\zeta)]}+3{<\infty}$ {since we have assumed $g$ to be bounded at infinity}. 
\end{proof}
Therefore, \'Ecalle's resummation results give us the following results:
\begin{cor} \label{cor:main_result}
 The solution of the renormalisation group equation \eqref{eq:RGE} and the Schwinger-Dyson equation \eqref{eq:SDnlin} is Borel-\'Ecalle resummable. For any $L$ in $\R^*_+$, the resummed function $a\to G^{\rm res}(a,L)$ is analytic in the open subset of $\C$ defined by
 \begin{equation*}
  \left|a-\frac{1}{20L}\right| < \frac{1}{20L}.
 \end{equation*}
%  for any $L$ in $\R^*_+$.
\end{cor}
\begin{proof}
 Theorems \ref{thm:bound_two_point_infinity} and \ref{thm:resurgence_two_points_function} directly imply that the solution of the renormalisation group equation \eqref{eq:RGE} and the Schwinger-Dyson equation \eqref{eq:SDnlin} is Borel-\'Ecalle resummable.
 
 For the analyticity domain, one has simply to observe that it only depends on the asymptotic bound, therefore it is enough to bound $g$. This function was assumed to be increasing, and bounded by $1$. So this implies that the coefficient $M$ in the proof of Theorem \ref{thm:bound_two_point_infinity} is bounded by $10$. The result then follows directly from Theorem \ref{thm:Borel_Ecalle_resummation}.
\end{proof}

Let us finish this long section by pointing out that we have shown the analyticity of a solution of the Schwinger-Dyson equation in an open disc tangent to the origin. {Assuming that the bound of Theorem \ref{thm:bound_two_point_infinity} is optimal, standard results of the theory of Laplace transform and of Borel-\'Ecalle resummation theory indicate that the resummed function $G^{\rm res}(a,L)$ admits a logarithmic singularity at $a=(10L)^{-1}$. Notice that this logarithmic singularity was already pointed out in the conclusion of \cite{BC18}.

If one sees the resummed function $G^{\rm res}(a,L)$ as a function of $p^2=\mu^2\exp(L)$, its singularities at finite $p^2$ can be seen as masses that were not present in the lagrangian but can only be seen after a resurgent analysis. Further notice that if the Borel transform of the two-point function has an exponential behavior at infinity 
\begin{equation} \label{bound_G_asympt_free}
 \widehat{G}(\zeta,L)\sim K\exp(ML|\zeta|)
\end{equation}
then the associated resummed function admits a simple pole at $ML=1/a~\Longleftrightarrow p^2=\mu^2\exp((aM)^{-1})$. In other words: under the assumption of the bound \eqref{bound_G_asympt_free} we have generated a mass $\mu^2\exp((aM)^{-1})$ for our theory. 

Finally, let us point out two things. First, that a bound of the form \eqref{bound_G_asympt_free} is what one should expect to obtain after performing an acceleration of the Borel transform. Furthermore, according to \cite{BC19}
such an acceleration will likely take place in the context of asymptotically free QFTs. Therefore we are confident that the proposed mechanism could, at least in principle, be applied to some Yang-Mills theories. Second, if one improves the bound \eqref{bound_G_asympt_free}\footnote{this being of course an abuse of language: it is only possible if Equation \eqref{bound_G_asympt_free} is a bound not an equivalence. We are not more precise in order to not burden the text with too much technical details.} that is to say find an $M'<M$ then the induced mass $\mu^2\exp((aM')^{-1})$ will increase. In other words: improving the bound \eqref{bound_G_asympt_free} increases the mass gap of the theory.

\section{Toward asymptotically free theories} \label{sec:res_final}

The work presented above, while encouraging, is for a rather specific, exactly supersymmetric, quantum field theory. The holy grail would be to obtain the same type of results for physically relevant QFTs, and in particular for quantum chromodynamics (QCD). In this section, we argue that an extension of the Borel-\'Ecalle resummation method would be needed for this and explore its consequences.

\subsection{Analyticity domain}

The analyticity domain one obtains after a Borel-\'Ecalle resummation procedure in the positive real direction is a disc tangent to the origin, whose center is on the positive real line, and whose diameter is given by the asymptotic behavior of the Borel transform. This is rather nice since it avoids the negative real numbers and thus Dyson's argument: the theory is not defined in a domain where it gives absurd results. Of course this is not specific to the Borel-\'Ecalle resummation mechanism: it is also true for the simpler Borel-Laplace resummation. So what makes Borel-\'Ecalle a better candidate for physical resummation and in particular QFT?

We have already seen one argument: in physically relevant theories, the perturbative series tend to be divergent and their Borel transform to have singularities on the positive real axis. This prevents their resummation with the Borel-Laplace method thus justifying that one needs Borel-\'Ecalle.
Another argument in favor of \'Ecalle's resurgence theory is its alien calculus, that allow to compute transseries contributions from the alien derivatives. We already \ty{mentioned} this in the introduction and here is not the place to say more. 

In this section, we will see another, often overlooked, argument. In his remarkable paper \cite{Ho79} 't Hooft analyses the renormalisation group equation of an asymptotically free theory. Based on physical assumptions of the theory such as the existence of a mass, he concludes that the singuralities will lie on circles of radius $K_n = \left|2\beta_0(2n+1)\pi\right|^{-1}$ and center $\pm i K_n$; where $n\in\Z$ and $\beta_0>0$ is the (opposite of) the first coefficient of the $\beta$-function of the theory: $\beta(a)=-\beta_0a^2+O(a^3)$. The argument can also be found in \cite{BC19} which might be easier to obtain than \cite{Ho79}.

Let us point out first that $|K_n|$ decreases when $|n|$ increases. The maximum values are $K_0 = \left(2\beta_0\pi\right)^{-1} = -K_{-1}$.
Second, all the values of $a$ for which the two-point function is singular will lie inside the two disks $D_\pm$ of radius $K_0$ and of center $\pm iK_0$. 

Finally $K_n$ goes to 0 as $n$ goes to infinity, therefore the poles will concentrate around the origin. Thus we find that the domain of analyticity of $G$ arbitrarily close to the 
origin is delimited by the circles $D_+$ and $D_-$ tangent to the origin (see Figure \ref{fig:an_dom_tHooft}). This is in strong
contradiction with the analyticity domain of 
Borel-Laplace or Borel-\ty{\'Ecalle} summed functions. We thus have an argument that $G$ cannot be Borel summable. \cy{However we will see
that accelero-summation gives the right} analyticity domain.
 \begin{figure}
  \begin{tikzpicture}
  %semicircle: 
   \draw [rotate=270] (0,4)  arc (180:0:1);
   \draw [rotate=270] (-2,4)  arc (180:0:1);
   %forbidden regions:
   \draw[pattern=north west lines, pattern color=blue] (-4,-4.5) rectangle (4,4.5);
   \draw [rotate=270,pattern=north west lines, pattern color=blue] (0,4)  arc (180:0:1);
   \draw [rotate=270,pattern=north west lines, pattern color=blue] (-2,4)  arc (180:0:1);
% legends
  \draw[pattern=north west lines, pattern color=blue] (8.4,2) rectangle (9.3,2.5);
   \node (tag) at (11,2.3) {Forbidden domain};
      %axis:
   \node (tag) at (3.8,0.2) {0};
   \draw[->] (4,-3) -- (4,3) node[above right]{$\Im(a)$};
   \draw[->] (-1,0) -> (6.1,0) node[right]{$\Re(a)$};
  \end{tikzpicture}
  \caption{Maximal analyticity domain from 't Hooft argument.}
   \label{fig:an_dom_tHooft}
 \end{figure}
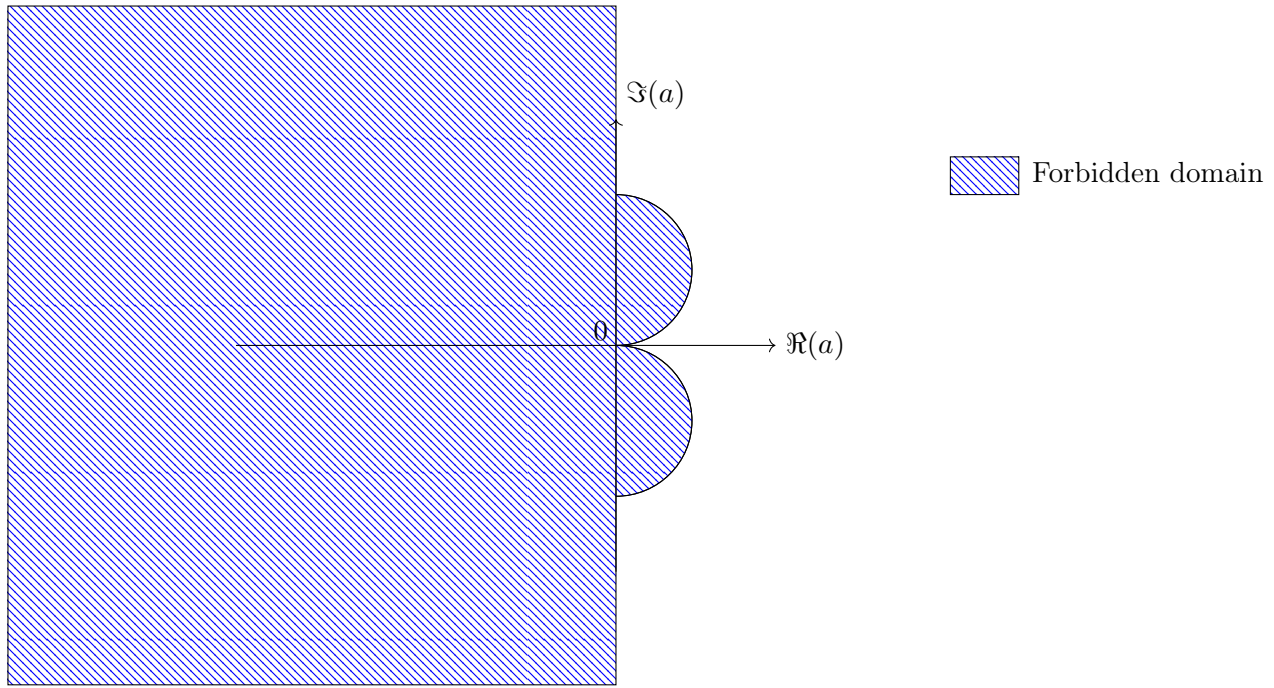
 
 \subsection{Accelero-summation in a nutshell} \label{subsec:accelero}
 
 This section intends to be a self-contained introduction to the basic concepts of accelero-summation.
We refer the reader to \cite{Ec93,Ec92} for a more detailed description.

Let $\tilde{f}(z) = \frac{1}{z}\sum_{n=0}^{+\infty}c_n\frac{1}{z^n}$ be a formal series in the variable $z:=a^{-1}$ and assume that $\widehat f$ is a $\Z^*$-resurgent function. Acceleration allows to generalise the Borel-Laplace and Borel-\'Ecalle procedures to more complex situations, where the behaviour at
infinity of the Borel transform is super-exponential and does not allow for a Laplace transform. Let $F:\C\mapsto\C$ be a function, which will be called the
acceleratrix, such that
\begin{itemize}
 \item $\overline{F(y)}=F(\bar y)$,
 \item $\lim_{y\to\infty}F(y)=+\infty$,
 \item $\lim_{y\to\infty}\frac{y}{F(y)}=+\infty$.
\end{itemize}
Let us specify that this last condition makes us deal with the so-called ``strong acceleration'' of
\cite{Ec92}, other forms exist but we will not need them here.  Acceleration allows to have a Borel transform with respect to the variable
$y$ starting from the Borel transform with respect to the variable $z=F(y)$.

Acceleration is performed by mapping the function $\widehat f(\zeta)$ to a germ $\widehat
f_{\rm acc}(\zeta)$ by
\begin{equation}  \label{def_acc}
 \widehat f_{\acc}(\zeta):=\int_0^{+\infty}C_F(\zeta,\xi)\widehat f(\xi)d\xi
\end{equation}
The acceleration kernel is the Borel transform (with respect to the variable $y$) of the function $\exp(-\xi F(y))$ and can be 
obtained through an inverse Laplace transform: 
\begin{equation*}
 C_F(\zeta,\xi):=\frac{1}{2i\pi}\int_{c-i\infty}^{c+i\infty}e^{-\xi F(y)+\zeta y}dy ,
\end{equation*}
where $c$ can be taken as any positive constant.

Similarly to the case of $\widehat f$ which was a priori only defined in a neighbourhood of 0, the preceding integral only defines $\widehat f_{\acc}(\zeta)$ in the
vicinity of the origin since the integral in~\eqref{def_acc} is only convergent for small enough $\zeta$ and we follow the same pattern. We analytically extend it, look for
singularities of this analytical continuation which are controlled by a new set of alien derivatives
and finally perform 
a Laplace (or typically an averaged Laplace) transform on it to obtain the resummed function\footnote{In principle,
other rounds of accelerations could be necessary before the final Laplace transform, but once again, we
do not aim at describing the most general procedure.}.  The resummed function is then given by the following integral
\begin{equation} \label{eq:Laplace}
 f^{\res}(y) = \int\widehat f_{\acc}(\xi)e^{-y \xi}d\xi.
\end{equation}
All accelerations have in common that they transform
convolution products to convolution products, ensuring that the whole procedure will give a sum which
satisfies the same equations as the formal series we start with. Since accelero-summation simply adds the acceleration step in the Borel-\'Ecalle resummation method, the resummed function obtained with an acceleration will satisfy the same equations than our initial formal series.
\begin{rk}
 The acceleration presented above was performed before averaging. It seems to be a folklore result, and a reasonable one, that \cy{accelerating} could be performed after averaging with the same result. Let me point out however that the Borel transform is more explicit before averaging and therefore that the asymptotics analysis should rather be performed before averaging.
\end{rk}

\subsection{Application to asymptotically free QFTs} \label{subsec:asymp_free_QFT}

It is now time to answer the main question: why should one expect that the two-point function of an asymptotically free QFT is accelero-summable, but not Borel-Laplace nor Borel-\'Ecalle summable? As hinted above the first reason is that the strong acceleration described above provides us with an explanation of the peculiar analyticity domain predicted by 't Hooft in \cite{Ho79}. Let us check that now.

As before, we assumed the Borel transform was performed in the variable $z=a^{-1}$. Let us assume that the Borel transform is such that we need an acceleration of the form 
\begin{equation*} 
 z = F(y) = \frac{1}{\sigma}\log(y).
\end{equation*}
After resummation, one obtains a function $y\to\tilde G^{\res}(y)$ analytic in a half-plane $H_R:=\{y\in\C:\Re(y)\geq R\}$ for some $R>0$. Writing
$\tilde G^{\res}(y)=\tilde G^{\res}(F^{-1}(z))=:G^{\res}(z)$, we obtain that $z\to G^{\res}(z)$ is analytic in the domain $F(H_R)$.

Using the principal branch of the logarithm $\log(x+it) = \log(|x+it|)+i\arctan(t/x)$ (which was already used to obtain 't Hooft's result) we see that the image under $F$
of a vertical line with real part $x\geq R$ is the curve parameterized by \(s=\arctan(t/x)\)
\begin{equation}
 S_x:=\left\{\left.\frac{1}{\sigma}\Bigl(\log(x)-\log(\cos(s))+is\Bigr)\right|s\in\left]-\frac{\pi}{2},\frac{\pi}{2}\right[\right\}.
\end{equation}
Since $\log(\cos(s))$ goes logarithmically to minus infinity when $s$ approaches $\pm\pi/2$, 
the analyticity domain of $z\to G^{\res}(z)$ is well approximated by the open rectangle
\begin{equation}
 F(H_R) = \left\{z\in\C\left|\Re(z)\geq\frac{1}{\sigma}\log(R)\wedge\Im(z)\in\left]-\frac{\pi}{2\sigma},\frac{\pi}{2\sigma}\right[\right.\right\}.
\end{equation}
Now, let us map this domain back into the $a$ plane.

The three lines approximating the boundary of $F(H_R)$ will be converted to circles including the
origin, as was recalled in section~1.  First, $F(H_R)$ is beyond the line with $\Re(z)=
\log(R)/\sigma$, so its transformation will be inside the circle with center $2\sigma/\log(R)$ and
radius the same expression.  The lines with $\Im(z) = \pm \pi/(2\sigma)$ will likewise be transformed
in circles with centers $\pm i \sigma/\pi$ and radius $\sigma/\pi$, but this time, the image of
$F(H_R)$ will be outside these circles.  All in all, we obtain in the usual $a$ plane a domain
squeezed between the two tangent circles centered at $\pm i \sigma /\pi$ near the origin, limited by
the circle centered at  $2\sigma/\log(R)$.  We must not forget that the three circle limits are but
approximations, since the real boundary must be smooth, because it is the image of line by a holomorphic
map. The situation is summed up in Figure \ref{fig:an_dom_acc}
 \begin{figure}
 \begin{center}
 \begin{tikzpicture} 
  %semicircles: 
   \draw [rotate=270] (0,-1)  arc (180:0:1);
   \draw [rotate=270] (-2,-1)  arc (180:0:1);
   \draw [rotate=270] (-2,-1)  arc (180:0:2);
%    \filldraw[thick,fill=green!50,even odd rule] (-0.25,0) circle (2) (0.25,0) circle (2);
\filldraw[rotate=270,thick,fill=green!50] (0,-1)  arc (180:0:1) (-2,-1)  arc (180:0:1) (-2,-1)  arc (180:0:2);
\filldraw[rotate=270,thick,fill=white] (0,-1)  arc (180:0:1);
\filldraw[rotate=270,thick,fill=white] (-2,-1)  arc (180:0:1);
%    \shade[top color=lightgray, bottom color=gray] (0,0) circle [radius=1cm];
   %tag:
%   \draw (1,1.1) rectangle (2,1.9) node[midway] {Ana. dom. of $G^{res}$};
   \node (tag) at (2.7,1.3) {Ana. dom. of $G^{\res}$};
   \draw[->] (1.5,1) .. controls (1.6,0.4) .. (0.6,0.5);
   %axis:
   \draw[->] (-1,-2.3) -- (-1,2.3) node[above]{$\Im(a)$};
   \node (tag) at (-1.2,0.2) {0};
   \draw[->] (-1.2,0) -> (2.1,0) node[right]{$\Re(a)$};
  \end{tikzpicture}
  \caption{The analyticity domain of a function resummed with an acceleration}  \label{fig:an_dom_acc}
  \end{center}
\end{figure}
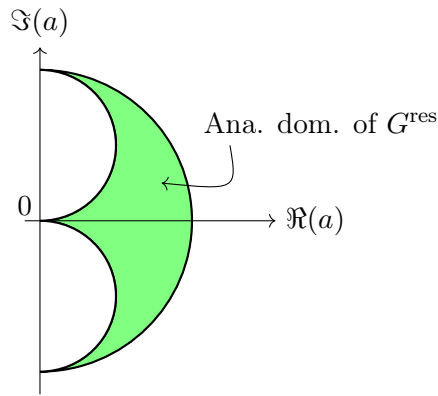

This fact alone suggests the conjecture, namely that two-point functions of a asymptotically free QFTs are accelero-summable (with the strong acceleration described above). 
% I instead offer a more precise conjecture:
\begin{conj} \label{conj:asympt_bound}
 For an asymptotically free QFT with $\beta$-function $\beta(a)=-\beta_0+\mathcal{O}(a^3)$, there exists
 for each value of the kinematical parameter $\Lambda$ a number $\sigma_\Lambda\geq\frac{1}{2\beta_0}$ such that the formal series $G(\Lambda,a)$ is accelero-summable with acceleratrix $F(y)=\frac{1}{\sigma_\Lambda}\log(y)$.
\end{conj}
The lower bound $\sigma_\Lambda\geq\frac{1}{2\beta_0}$ for $\sigma_\Lambda$ in the conjecture is so that the analyticity domain of the accelero-summed function does not intersect 't Hooft's forbidden domain of Figure \ref{fig:an_dom_tHooft}. Indeed it implies
\begin{equation*}
 \frac{\sigma_\Lambda}{\pi}\geq \frac{1}{2\pi\beta_0} = K_0.
\end{equation*}
Therefore the semicircles centered on the imaginary real axis in Figure \ref{fig:an_dom_acc} are bigger than the ones predicted by 't Hooft. This bound is necessary for the analyticity domain we find after an accelero-summation not to intersect the forbidden domain in Figure \ref{fig:an_dom_tHooft}.

\subsection{Other open questions} \label{subsec:res_questions}

Let us finish this chapter with a list of open questions that should be of interest for future research. In this list I only include question of resurgence theory, and not how some techniques of resurgence theory could be related to methods \ty{developed} in other fields, for example constructive field theory. \\

{\bf Question:} Could the non-perturbative mass generation mechanism presented in the end of Subsection \ref{subsect:final} be adapted to the framework of asymptotically free QFT with an accelero-summable two-point function? \\

Preliminary results suggest that the minimal acceleratrix $F(y)=\frac{1}{\sigma_\Lambda}\log(y)$ of Conjecture \ref{conj:asympt_bound} is not perfectly suited for this. Finding a different acceleratrix that would preserve 't Hooft's analyticity domain and offer a mass generation mechanism would help make this conjecture more precise: it would tell us what is the expected asymptotic behaviour of the two-point function.

However, such a modified Conjecture \ref{conj:asympt_bound} would be very ambitious. One can --through hard work-- find bounds for functions that prove their (accelero-)summability, but typically not show that these bounds are optimal. As a matter of fact we would need to have the precise asymptotic behaviour of the two-point function to prove that the proposed mass generation mechanism takes place (but numerical studies could offer enough evidence for physically acceptable results). This difficulty comes on top of the fact that the Schwinger-Dyson equation for Yang-Mills theories do not typically close. The results of \cite{BoDu20} suggests that the resurgence algebra of the solutions of various truncations of the two-point function could greatly differ.\\

The following question is related to the most standard use of resurgence in the physics \ty{literature}: its applications to compute transseries solutions of a given physical \ty{equation}. There is a very precise analytical link between a Borel 
summable series and the associated Borel resummed function, known as Watson's theorem \cite{Wa1912}. This classical result was generalised by Nevanlinna in \cite{Ne1918}\footnote{my most sincere thanks to V. Rivasseau that pointed out this reference to me and to M. Serone who quoted it precisely enough in his lecture notes \cite{Se24} for me to find the reference} and by Sokal in \cite{So80}. Sokal's generalisation of Watson's theorem, henceforth called Sokal-Watson theorem, is the one commonly used in QFT. To the best of my knowledge, the following question is still open: \\

{\bf Question:} Is there a Sokal-Watson's theorem for the Borel-\'Ecalle (accelero-)resummation method? \\

The most general case of transseries of any level (which are, in principle, obtainable from the most general Borel-\'Ecalle resummation method) is still beyond the reach of resurgence theory. Let us precise this question on the example of a level one transseries, as the one found in the Wess-Zumino model studied below, and most of physical applications to resurgence:
\begin{equation*}
 \Phi(a) = \sum_{k=0}^{+\infty}\Phi_n(a)e^{-Sn/a}.
\end{equation*}
Provided the Borel transforms $\widehat\Phi_n$ all satisfy simultaneously the assumptions of the Sokal-Watson theorem, we readily obtain a positive answer. These assumptions are that each of the Borel transform $\widehat\Phi_n$ admits an analytic extension to a common strip containing the real line, and a common exponential bond at infinity on this strip. Then we would obtain common analytical properties and bounds for each of the $\Phi_n$ which can be used to prove the analyticity domain and bounds on $\Phi$.

However in general the Borel transforms do not admit an analytic extension in a strip containing the positive real line. As we have seen they have singularities. They only admit analytic continuation along paths avoiding these singularities. It is more delicate but should, at least in principle, be solvable using Ecalle's resurgence relations which should allow us to reduce this case to the previous one.

A Sokal-Watson theorem for transseries would be of importance for the physical implications of the Borel-\'Ecalle resummation method. Indeed, for these applications only a transseries expansions of the Borel-\'Ecalle resummed function 
were computed. These transseries are not the full Borel-\'Ecalle resummed functions but rather a good approximation which can then be compared to experimental results. A Watson's theorem for Borel-\'Ecalle 
resummation which would be formulated with transseries would provide a more precise meaning to the word ``good'' in the previous sentence and allow to have estimates for error margins coming from the 
truncations of the transseries. \\

% Another reason why such a theorem would be of importance lies in the details of the physical applications of resurgence theory to physics. The coefficients of the transseries expansion are computed using the 
% so-called median average, which can be expressed in terms of the alien derivatives of the formal series to be resummed. The median average is one special average, a notion that will be introduced below. 
% However, it is not a ``well-behaved average'', which are the ones that should be used for the Borel-\'Ecalle resummation method. Nonetheless, one could expect the transseries expansion of a function to be 
% unique. Thus Watson's theorem for Borel-\'Ecalle resummation would give a better mathematical ground to physical computations.
The last question looks at an  point that has been left purposefully unclear in the discussion of the Borel-\'Ecalle resummation method. \\

{\bf Question:} How does the Borel-\'Ecalle resummed function depend on the choice of the well-behaved average? \\

A priori, different well-behaved averages being very different objects, one could assume the answer to be that different choices of averages build different solutions. However, a finer analysis of the problem at hand leads to more subtle conclusions. In particular, for physics-related problems studied with resurgence, the Riemann surface on which the resurgent functions have their domain is highly structured. This should induce that different average coincide up to one (or possibly finitely many) free parameters, as in the case of ODEs \cite{Co98}, \cite{Co06}.

Therefore, one could conjecture that it actually does not depend on the choices made and that changing averages amounts to a reparametrisation of the solution. This conjecture is motivated by an observation of 
\cite{Me97} that it indeed holds for a specific problem and from the fact that two averages are always related by a so-called passage automorphism. Even if the choice of the average changes the resummed 
function, one should expect stability of some physically relevant properties, for example the poles of the resummed function.\\

As a final word, I would like to repeat that I had presented here only one approach to resurgence. Among the other approaches that exist, the so-called \emph{modular resurgence} (see \cite{fantini2024modular}). The idea is to study some quantum modular form (see \cite{zagier2010quantum}) whose failure to modularity can be measured with a divergent resurgent series. These objects appear (at least conjecturally) in topological string theory, knot theory and quantum Chern-Simons theory.

% % % % %  COMMENTER CE QUI EST CI-DESSOUS.
% 
%  \bibliographystyle{unsrt}
%  \addcontentsline{toc}{section}{Bibliography}
% \bibliography{HDR_biblio}
%  
% \end{document}